\documentclass[12pt,polish,english]{mwrep}
\usepackage[T1]{fontenc}
\usepackage[cp1250]{inputenc}
\usepackage[a4paper]{geometry}
\usepackage{color}
\usepackage{amsthm}
\usepackage{amsmath}
\usepackage{amssymb}
\usepackage{graphicx}
\usepackage{setspace}
\usepackage[numbers]{natbib}
\PassOptionsToPackage{normalem}{ulem}
\usepackage{ulem}

\makeatletter

\providecommand{\tabularnewline}{\\}

  \theoremstyle{definition}
  \ifx\thechapter\undefined
    \newtheorem{defn}{\protect\definitionname}
  \else
    \newtheorem{defn}{\protect\definitionname}[chapter]
  \fi
  \theoremstyle{plain}
  \ifx\thechapter\undefined
    \newtheorem{fact}{\protect\factname}
  \else
    \newtheorem{fact}{\protect\factname}[chapter]
  \fi
  \theoremstyle{definition}
  \newtheorem*{problem*}{\protect\problemname}
 \theoremstyle{definition}
 \newtheorem*{defn*}{\protect\definitionname}
  \theoremstyle{remark}
  \newtheorem*{rem*}{\protect\remarkname}
  \theoremstyle{definition}
  \ifx\thechapter\undefined
    \newtheorem{problem}{\protect\problemname}
  \else
    \newtheorem{problem}{\protect\problemname}[chapter]
  \fi
  \theoremstyle{plain}
  \ifx\thechapter\undefined
    \newtheorem{prop}{\protect\propositionname}
  \else
    \newtheorem{prop}{\protect\propositionname}[chapter]
  \fi
  \theoremstyle{plain}
  \ifx\thechapter\undefined
    \newtheorem{lem}{\protect\lemmaname}
  \else
    \newtheorem{lem}{\protect\lemmaname}[chapter]
  \fi
  \theoremstyle{plain}
  \ifx\thechapter\undefined
    \newtheorem{thm}{\protect\theoremname}
  \else
    \newtheorem{thm}{\protect\theoremname}[chapter]
  \fi
  \theoremstyle{plain}
  \ifx\thechapter\undefined
    \newtheorem{cor}{\protect\corollaryname}
  \else
    \newtheorem{cor}{\protect\corollaryname}[chapter]
  \fi
  \theoremstyle{definition}
  \ifx\thechapter\undefined
    \newtheorem{example}{\protect\examplename}
  \else
    \newtheorem{example}{\protect\examplename}[chapter]
  \fi
\newenvironment{lyxlist}[1]
{\begin{list}{}
{\settowidth{\labelwidth}{#1}
 \setlength{\leftmargin}{\labelwidth}
 \addtolength{\leftmargin}{\labelsep}
 }}
{\end{list}}
  \theoremstyle{remark}
  \ifx\thechapter\undefined
    \newtheorem{rem}{\protect\remarkname}
  \else
    \newtheorem{rem}{\protect\remarkname}[chapter]
  \fi

\usepackage{ytableau}

\AtBeginDocument{
  
}

\makeatother

\usepackage{babel}
  \addto\captionsenglish{\renewcommand{\definitionname}{Definition}}
  \addto\captionsenglish{\renewcommand{\examplename}{Example}}
  \addto\captionsenglish{\renewcommand{\factname}{Fact}}
  \addto\captionsenglish{\renewcommand{\lemmaname}{Lemma}}
  \addto\captionsenglish{\renewcommand{\problemname}{Problem}}
  \addto\captionsenglish{\renewcommand{\propositionname}{Proposition}}
  \addto\captionsenglish{\renewcommand{\remarkname}{Remark}}
  \addto\captionspolish{\renewcommand{\definitionname}{Definicja}}
  \addto\captionspolish{\renewcommand{\examplename}{Przykład}}
  \addto\captionspolish{\renewcommand{\factname}{Fakt}}
  \addto\captionspolish{\renewcommand{\lemmaname}{Lemat}}
  \addto\captionspolish{\renewcommand{\problemname}{Problem}}
  \addto\captionspolish{\renewcommand{\propositionname}{Propozycja}}
  \addto\captionspolish{\renewcommand{\remarkname}{Uwaga}}
  \providecommand{\definitionname}{Definition}
  \providecommand{\examplename}{Example}
  \providecommand{\factname}{Fact}
  \providecommand{\lemmaname}{Lemma}
  \providecommand{\problemname}{Problem}
  \providecommand{\propositionname}{Proposition}
  \providecommand{\remarkname}{Remark}
\addto\captionsenglish{\renewcommand{\corollaryname}{Corollary}}
\addto\captionsenglish{\renewcommand{\theoremname}{Theorem}}
\addto\captionspolish{\renewcommand{\corollaryname}{Wniosek}}
\addto\captionspolish{\renewcommand{\theoremname}{Twierdzenie}}
\providecommand{\corollaryname}{Corollary}
\providecommand{\theoremname}{Theorem}

\begin{document}
\global\long\global\long\global\long\def\bra#1{\mbox{\ensuremath{\langle#1|}}}
\global\long\global\long\global\long\def\ket#1{\mbox{\ensuremath{|#1\rangle}}}
\global\long\global\long\global\long\def\bk#1#2{\mbox{\ensuremath{\ensuremath{\langle#1|#2\rangle}}}}
\global\long\global\long\global\long\def\kb#1#2{\mbox{\ensuremath{\ensuremath{\ensuremath{|#1\rangle\!\langle#2|}}}}}

\thispagestyle{empty} 

\noindent \begin{center}
\textsc{\large{}University of Warsaw}
\par\end{center}{\large \par}

\noindent \bigskip{}

\noindent \begin{center}
\textsc{Doctoral thesis}
\par\end{center}

\bigskip{}

\rule[0.5ex]{1\columnwidth}{1pt}

\noindent \begin{center}
\textbf{\textsc{\Large{}Applications of differential geometry and
representation theory to description of quantum correlations}}
\par\end{center}{\Large \par}

\rule[0.5ex]{1\columnwidth}{1pt}

\noindent \begin{center}
{\large{}}%
\begin{minipage}[t]{0.4\columnwidth}%
\begin{singlespace}
\noindent \textit{Author: }

\noindent \textsc{MSc Michał Oszmaniec}\end{singlespace}
\end{minipage}{\large{}\hfill{}}%
\begin{minipage}[t]{0.4\columnwidth}%
\begin{singlespace}
\noindent \textit{Supervisor: }

\noindent \textsc{Prof. Dr hab. Marek Kuś}\end{singlespace}
\end{minipage}\medskip{}

\par\end{center}

\bigskip{}

\bigskip{}

\begin{center}

\noindent A thesis submitted in fulfillment of the requirements for
the degree of Doctor of Physical Sciences in Physics 

\noindent in the

\noindent Faculty of Physics, University of Warsaw

\end{center}

\bigskip{}

\noindent \begin{center}
Prepared in The Center for Theoretical Physics of Polish Academy of
Sciences 
\par\end{center}

\bigskip{}

\bigskip{}
\bigskip{}

\noindent \begin{center}
\qquad{}\qquad{}\includegraphics[scale=0.12]{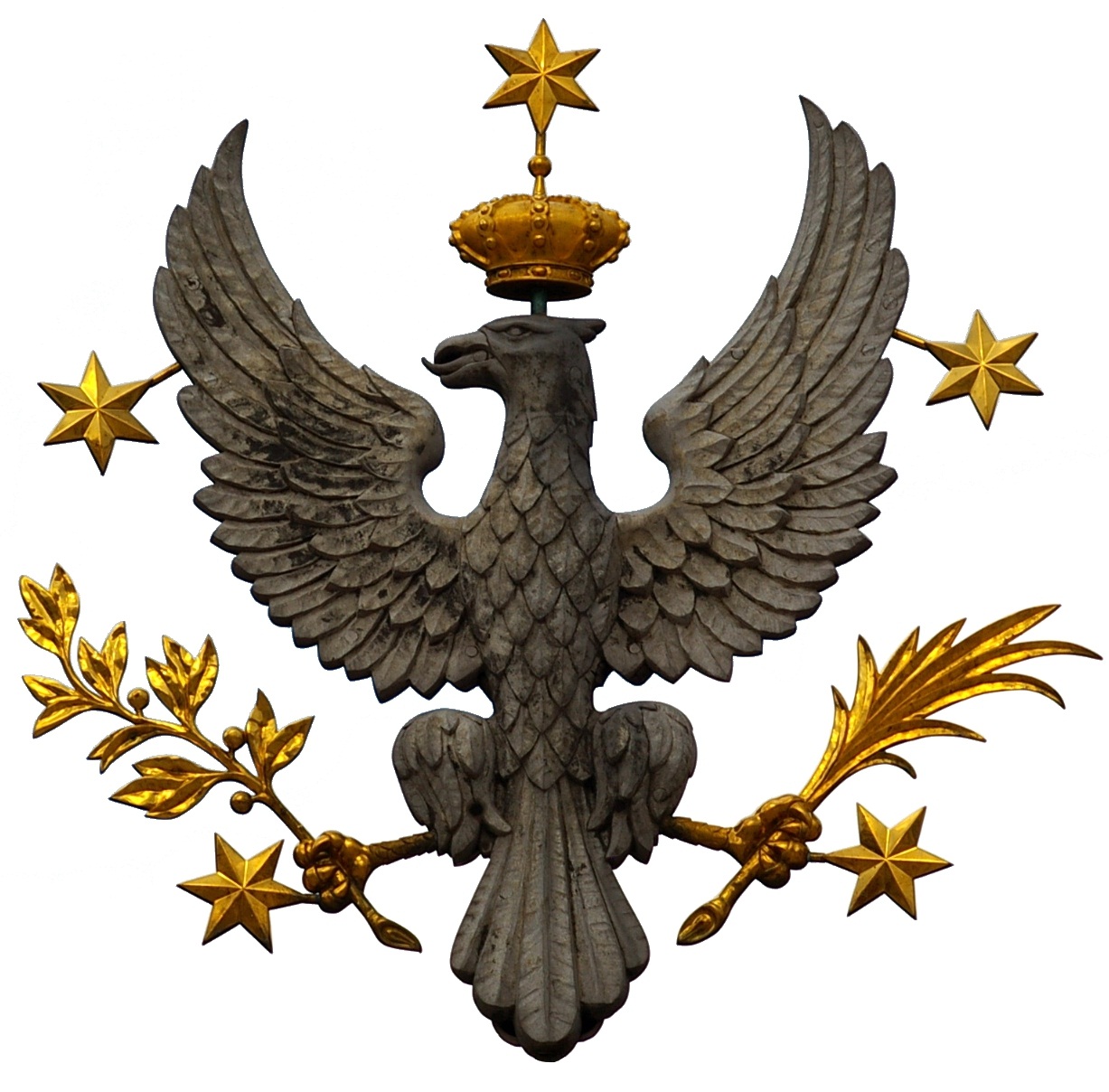}\qquad{}\qquad{}\qquad{}\includegraphics[scale=1.1]{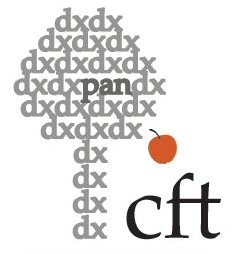}
\par\end{center}

\bigskip{}

\bigskip{}
\bigskip{}

\bigskip{}
\bigskip{}
\bigskip{}

\noindent \begin{center}
Warsaw, 2014
\par\end{center}

\newpage{}

~~\thispagestyle{empty}

\newpage{}

~~\thispagestyle{empty}

\bigskip{}
\bigskip{}
\bigskip{}
\bigskip{}
\bigskip{}
\bigskip{}
\bigskip{}
\bigskip{}
\bigskip{}
\bigskip{}
\bigskip{}
\bigskip{}
\bigskip{}
\bigskip{}
\bigskip{}
\bigskip{}
\bigskip{}

\bigskip{}
\bigskip{}
\bigskip{}
\bigskip{}
\bigskip{}
\bigskip{}
\bigskip{}
\bigskip{}

\bigskip{}
\bigskip{}
\bigskip{}
\bigskip{}
\bigskip{}
\bigskip{}
\bigskip{}
\bigskip{}
\bigskip{}
\bigskip{}
\bigskip{}
\bigskip{}
\bigskip{}
\bigskip{}
\bigskip{}
\bigskip{}

\large\begin{flushright}\textit{I dedicate this thesis to my Father}  \end{flushright}\normalsize

\newpage{}

~~\thispagestyle{empty}

\begin{center}

\Large\textbf{Acknowledgments}

\bigskip{}

\end{center}

\normalsize

I would like express my deepest gratitude to my supervisor Marek Kuś
and co-supervisor Adam Sawicki. This work would not be possible without
their constant support, encouragement and advice. 

Marek Kuś - for introducing me into the exiting field of quantum information,
for the trust put when three year ago he, despite my little experience
in this topic, accepted me to his group at the Center for Theoretical
Physics, and above all, for his contagious sense of humor and the
positive way of thinking. 

Adam Sawicki - for many inspiring discussions concerning connections
between quantum theory, differential geometry and representation theory. 

Part of this thesis is based on the work I did during my stay at the
Institute for Photonic Sciences in Barcelona in spring 2014. I am
grateful to my host Maciek Lewentein for his warm hospitality and
interest in my work. A special thanks go to Remik Augusiak, a member
of Maciek's group, for many hours of educating and inspiring discussions.
I also thank Janek Gutt and Szymon Charzyński, my colleagues from
the Center for Theoretical Physics.

I would like to thank the reviewers, Paweł Horodecki and Rafał Demkowicz-Dobrzański,
for pointing out numerous typos and misprints present in the original
version of the thesis.

I acknowledge the support of the NCN Grant No. DEC-2013/09/N/ST1/02772.

I am grateful to my family for their constant support during my Ph.D.
studies.

Last, but certainly not least, I would like to thank my girlfriend
Ola for her kindness, support and patience during the preparation
of this thesis.

\newpage{}

\thispagestyle{empty}
\begin{abstract}
One of the most important questions in quantum information theory
is the so-called separability problem. It involves characterizing
the set of separable (or, equivalently entangled) states among mixed
states of a multipartite quantum system. In this thesis we study the
generalization of this problem to types of quantum correlations that
are defined in a manner analogous to entanglement. We start with the
subset of set of all pure states of a given quantum system and call
states belonging to the convex hull of this subset ``non-correlated''
states. Consequently, the states laying outside the convex hull are
referred to as ``correlated''. 

In this work we focus on cases when there exist a symmetry group acting
on the relevant Hilbert space that preserves the class of ``non-correlated''
pure states. The presence of symmetries allows to obtain a unified
treatment of many types of seemingly unrelated types of correlations.
The symmetries also give the possibility to use the powerful methods
of differential geometry and representation theory to study various
properties of so-defined correlations: What is more, there exist many
physically-interesting classes of correlations that posses continuous
symmetries. In this work we apply our general results to particular
types of correlations: (i) entanglement of distinguishable particles,
(ii) particle entanglement of bosons, (iii) ``entanglement'' of
fermions, (iv) non-convex-Gaussian correlations in fermionic systems,
(v) genuine multiparty entanglement, and finally (vi) refined notions
of bipartite entanglement based on the concept of the Schmidt number.

The thesis is organized as follows. We first present the necessary
physical and mathematical background in Chapters \ref{chap:Introduction}
and \ref{chap:Mathematical-prelimenaries}. In Chapters \ref{chap:Multilinear-criteria-for-pure-states}-\ref{chap:Typical-properties-of},
which form the core of the thesis, we investigate the natural problems
and questions concerning the types of correlations defined above.
In Chapter \ref{chap:Multilinear-criteria-for-pure-states} we provide
exact polynomial criteria for characterization of various types of
correlations for pure quantum states. The Chapter \ref{chap:Complete-characterisation}
deals with cases in which it is possible to use the criteria from
Chapter \ref{chap:Multilinear-criteria-for-pure-states} to give a
complete analytical characterization of correlated mixed quantum states.
In Chapter \ref{chap:Polynomial-mixed states} we derive a variety
of polynomial criteria for detection of correlations in mixed states.
In Chapter \ref{chap:Typical-properties-of} we use the criteria derived
in the previous chapter and the technique of measure concentration
to study typical properties of correlations on sets of isospectral
density matrices. The thesis concludes in Chapter \ref{chap:Conclusions-and-outlook}
where we summarize the main contributions of the thesis and state
the open problems related to results presented in the thesis. 
\end{abstract}
\newpage{}

\thispagestyle{empty}

\noindent \begin{center}
\textbf{\large{}Streszczenie}
\par\end{center}{\large \par}

\selectlanguage{polish}%
Jednym z najważniejszych problemów w teorii informacji kwantowej jest
opisanie zbioru stanów splątanych (bądź równoważnie - separowalnych)
złożonego układu kwantowego. Splątanie kwantowe nie jest jednak jedynym
typem korelacji, który pojawia się w mechanice kwantowej. Tematem
przewodnim niniejszej rozprawy jest opis ogólniejszych typów korelacji,
które mogą być zdefiniowane w sposób analogiczny do splątania. Od
stany matematycznej problem jest następujący. Z pośród wszystkich
stanów czystych danego układu kwantowego wybieramy pewien podzbiór
stanów, zwanych stanami nieskolerowanymi. Konkretna postać stanów
nieskolerowanych zależy od danego problemu fizycznego. Mając dany
stan mieszany pytamy się, czy można go zapisać jako kombinację wypukłą
nieskolerowanych stanów czystych. Jeśli jest to możliwe, taki stan
nazywamy nieskorelowanym. W przeciwnym wypadku określamy go jako skorelowany.

W tej pracy będziemy rozpatrywać przypadki, w których na przestrzeni
Hilberta rozważanego układu fizycznego działa ciągła grupa symetrii
zachowująca zbiór nieskolerowanych stanów czystych. Obecność symetrii
pozwala traktować w ujednolicony sposób różne typy korelacji. Co więcej,
pozwala ona na stosowanie zmetod geometrii różniczkowej oraz teorii
reprezentacji grup i algebr Liego do opisu różnych własności zbioru
stanów skolerowanych. W niniejszej rozprawie zastosujemy te ogólne
metody do badania następujących typów korelacji: (i) splątanie cząstek
rozróżnialnych, (ii) splątanie cząstkowe bozonów, (iii) „splątanie”
fermionów, (iv) nie Gaussowskie korelacjei w układach fermionowych,
(v) właściwe splątanie kwantowe (genuine multiparty entanglement),
(vi) uogólnienie splątania dwucząstkowego bazujące na pojęciu rzędu
Schmidta. 

Praca podzielona została na osiem rozdziałów. Pierwsze dwa stanowią
wstęp i wprowadzone są w nich pojęcia fizyczne i matematyczne, które
używane są w dalszej cześć rozprawy. W następnych czterech rozdziałach,
stanowiących główną część rozprawy rozpatrujemy naturale zagadnienia
pojawiające się w kontekście próby opisu zbioru stanów skorelowanych.
W Rozdziale 3 wyprowadzamy wielomianową charakteryzację szerokiej
klasy nieskorelowanych stanów czystych. Rozdział 4 jest poświęcony
opisowi sytuacji, w których wielomianowa charakteryzacja stanów czystych
nieskorelowanych pozwala na pełny analityczny opis mieszanych stanów
nieskorelowanych. W Rozdziale 5 wyprowadzamy, korzystając z wyników
Rozdziału 3, rodzinę wielomianowych kryteriów wykrywających skorelowane
stany mieszane. Rozdział 6 poświęcony jest badaniu typowych własność
stanów skorelowanych na zbiorach izospektralnych macierzy gęstości.
W Rozdziale 7 znajduje się podsumowanie najważniejszych wyników pracy
oraz podana jest lista otwartych problemów związanych z jej tematyką.
Rozdział 8 zawiera dowody wyników, które zostały przedstawione bez
uzasadnienia w głównej części pracy.

\newpage{}

\selectlanguage{english}%
\thispagestyle{empty}

\noindent \begin{center}
\textbf{\large{}Structure of the thesis}
\par\end{center}{\large \par}

The thesis is organized as follows. The whole material is divided
into eight chapters which we list below together with a short description
of their content.
\begin{itemize}
\item Chapter \ref{chap:Introduction}: \textbf{Introduction}. We give here
the motivation and the context for the research presented in the thesis.
We also sketch the main problems that will be treated in the subsequent
chapters as well as ideas that will be used to solve them.
\item Chapter \ref{chap:Mathematical-prelimenaries}: \textbf{Mathematical
preliminaries}. We present here the methods of differential geometry
and representation theory of Lie groups that will be used in the thesis.
We also establish the mathematical notation that will be used latter.
\item Chapter \ref{chap:Multilinear-criteria-for-pure-states}: \textbf{Multilinear
criteria for detection of general correlations for pure states}. In
this chapter we provide exact multilinear criteria for characterization
of various types of correlations for pure quantum states.
\item Chapter \ref{chap:Complete-characterisation}: \textbf{Complete characterization
of correlations in mixed states}. We discuss here the cases in which
it is possible to use the criteria from Chapter \ref{chap:Multilinear-criteria-for-pure-states}
to give an analytical characterization of correlations in mixed quantum
states.
\item Chapter \ref{chap:Polynomial-mixed states}: \textbf{Polynomial criteria
for detection of correlations for mixed states}. We derive a variety
of polynomial criteria for detection of correlations in mixed states.
\item Chapter \ref{chap:Typical-properties-of}: \textbf{Typical properties
of correlation}s. We use here the polynomial criteria derived in the
previous chapter to study typical properties of correlations on sets
of isospectral density matrices.
\item Chapter \ref{chap:Conclusions-and-outlook}: \textbf{Conclusions and
outlook}. We summarize here the main contributions of this thesis.
We also state the most important open problems related to the results
presented in the thesis.
\item Chapter \ref{chap:Appendices}: \textbf{Appendix}. We give here the
proofs of the results presented in Chapters \ref{chap:Multilinear-criteria-for-pure-states}-\ref{chap:Typical-properties-of}
whose justification was omitted in the main text. 
\end{itemize}
The structure of the chapters that contain the original research (Chapters
\ref{chap:Multilinear-criteria-for-pure-states}-\ref{chap:Typical-properties-of})
is the following. Each of them begins with the introduction to the
topic that will be covered in a given Chapter. This introduction is
illustrated by an informally-stated \textbf{Problems} around which
the rest of the chapter concentrates. The chapter is further divided
into sections and subsections which correspond to specific aspects
of the given topic. The chapter concludes with the summary of obtained
results and a list of open problems related to the content of the
chapter.

Throughout the thesis the most important results will be stated in
the form of \textbf{Theorems}, \textbf{Lemmas}, \textbf{Propositions}
or \textbf{Corollaries}. These names correspond to the subjective
grading which the author attributes to the obtained results. Some
of results might have been published elsewhere but in a different
context or with the use of different methods. If we are aware of this
we state the appropriate reference. By \textbf{Facts} we will always
denote results that were known before. 

\newpage{}

\noindent \begin{center}
\textbf{\large{}List of publications }
\par\end{center}{\large \par}

\subsubsection*{Related to the content of the thesis }
\begin{enumerate}
\item Michał Oszmaniec and Marek Kuś, ``On detection of quasiclassical
states'', J. Phys. A: Math. Theor. 45 244034 (2012)
\item Michał Oszmaniec and Marek Kuś, \textquotedbl{}A universal framework
for entanglement detection\textquotedbl{}, Phys. Rev. A 88, 052328
(2013)
\item Michał Oszmaniec and Marek Kuś, ``Fraction of isospectral states
exhibiting quantum correlations'', Phys. Rev. A 90, 010302(R) (2014)
\item Michał Oszmaniec, Jan Gutt and Marek Kuś, ``Classical simulation
of fermionic linear optics augmented with noisy ancillas'', Phys.
Rev. A 90, 020302(R) (2014)
\end{enumerate}

\paragraph*{Not related to the content of the thesis}
\begin{enumerate}
\item Adam Sawicki, Michał Oszmaniec and Marek Kuś, “Critical sets of the
total variance can detect all stochastic local operations and classical
communication classes of multiparticle entanglement”, Phys. Rev. A
86, 040304(R) (2012)
\item Tomasz Maciążek, Michał Oszmaniec and Adam Sawicki, “How many invariant
polynomials are needed to decide local unitary equivalence of qubit
states?”, J. Math. Phys. 54, 092201 (2013)
\item Adam Sawicki, Michał Oszmaniec and Marek Kuś, \textquotedbl{}The convexity
of momentum map, Morse index and Quantum Entanglement\textquotedbl{},
Rev. Math. Phys. 26, 1450004 (2014)
\item Michał Oszmaniec, Piotr Suwara and Adam Sawicki, “Geometry and topology
of CC and CQ states”, J. Math. Phys. 55, 062204 (2014)
\item P. Migdał, J. Rodríguez-Laguna, M. Oszmaniec, and M. Lewenstein, ``Multiphoton
states related via linear optics'', Phys. Rev. A 89, 062329 (2014)
\end{enumerate}
\newpage{}

\tableofcontents{}

\newpage{}

~\thispagestyle{empty}

\noindent \begin{center}
\textbf{\large{}Selected symbols and abbreviations used in the thesis.}
\par\end{center}{\large \par}

\begin{center}
\begin{tabular}{|c|c|}
\hline 
Symbol & Explanation\tabularnewline
\hline 
$\mathcal{H}$ & Hilbert space associated to the considered physical system\tabularnewline
\hline 
$\mathcal{H}_{A}\otimes\mathcal{H}_{B}$ & Tensor product of Hilbert spaces $\mathcal{H}_{A}$ and $\mathcal{H}_{B}$\tabularnewline
\hline 
$\mathbb{C}^{d}$ & A standard $d$ dimensional complex Hilbert space\tabularnewline
\hline 
$\ket{\psi}$ & Vector from the considered Hilbert space\tabularnewline
\hline 
$\mathbb{I}$ & Identity operator on the relevant Hilbert space\tabularnewline
\hline 
$\mathrm{Sym}^{k}\left(\mathcal{H}\right)$ & $k$-fold symmetrization of the Hilbert space $\mathcal{H}$\tabularnewline
\hline 
$\mathbb{P}^{\mathrm{sym},k}$ & Orthogonal projector onto $\mathrm{Sym}^{k}\left(\mathcal{H}\right)$\tabularnewline
\hline 
$\mathbb{P}^{\mathrm{sym}}$ & Orthogonal projector onto $\mathrm{Sym}^{2}\left(\mathcal{H}\right)$\tabularnewline
\hline 
$\bigwedge^{k}\left(\mathcal{H}\right)$ & $k$-fold anti-symmetrization of the Hilbert space $\mathcal{H}$\tabularnewline
\hline 
$\mathbb{P}^{\mathrm{asym},k}$ & Orthogonal projector onto $\bigwedge^{k}\left(\mathcal{H}\right)$\tabularnewline
\hline 
$\mathbb{P}^{\mathrm{asym}}$ & Orthogonal projector onto $\bigwedge^{2}\left(\mathcal{H}\right)$\tabularnewline
\hline 
$\mathcal{D}\left(\mathcal{H}\right)$ & Set of mixed states on Hilbert space $\mathcal{H}$\tabularnewline
\hline 
$\mathcal{D}_{1}\left(\mathcal{H}\right)$ & Set of pure states on Hilbert space $\mathcal{H}$\tabularnewline
\hline 
$\mathcal{H}_{d}$ & The Hilbert space of $L$ distinguishable particles\tabularnewline
\hline 
$\mathcal{H}_{b}$ & The Hilbert space of $L$ bosons\tabularnewline
\hline 
$\mathcal{H}_{f}$ & The Hilbert space of $L$ fermions\tabularnewline
\hline 
$\mathcal{H}_{\mathrm{Fock}}\left(\mathbb{C}^{d}\right)$ & Fermionic $d$ mode Fock space\tabularnewline
\hline 
$\mathcal{M}$ & Non-correlated pure states\tabularnewline
\hline 
$\mathcal{M}^{c}$ & A convex hull of non-correlated pure states (non-correlated mixed
states)\tabularnewline
\hline 
$\mathcal{M}_{dist}$ & Pure separable states of $L$ particles\tabularnewline
\hline 
$\mathcal{M}_{b}$ & Pure separable states of $L$ bosons\tabularnewline
\hline 
$\mathcal{M}_{f}$ & Projectors onto $L$ partite Slater determinants \tabularnewline
\hline 
$\mathcal{M}_{g}$ & Pure fermionic Gaussian states\tabularnewline
\hline 
$A$ & Non-negative operator on $\mathrm{Sym}^{L}\left(\mathcal{H}\right)$
defining the set $\mathcal{M}$\tabularnewline
\hline 
$\mathrm{End}\left(\mathcal{H}\right)$ & The set of linear operators on $\mathcal{H}$\tabularnewline
\hline 
$\mathrm{Herm}\left(\mathcal{H}\right)$ & The set of Hermitian operators on $\mathcal{H}$\tabularnewline
\hline 
$\mathrm{tr}\left(X\right)$ & Trace of the operatator $X\in\mathrm{End}\left(\mathcal{H}\right)$\tabularnewline
\hline 
$\mathrm{SU}\left(\mathcal{H}\right)$ & Special unitary group on $\mathcal{H}$\tabularnewline
\hline 
$\Omega$ & The set of isospectral density matrices $\Omega$\tabularnewline
\hline 
$K$ & The symmetry group in question (usually compact and simply-connected)\tabularnewline
\hline 
$\mathfrak{k}$ & Lie algebra of a compact Lie group $K$\tabularnewline
\hline 
$\mathfrak{g}$ & Complex semisimple Lie algebra\tabularnewline
\hline 
$\mathfrak{h}$ & Cartan subalgebra of $\mathfrak{g}$\tabularnewline
\hline 
$\Pi,\pi$ & A representation of a Lie group, respectively Lie algebra\tabularnewline
\hline 
$\lambda$ & A weight of a \tabularnewline
\hline 
$\mu$ & Haar measure on $\mathrm{SU}\left(\mathcal{H}\right)$\tabularnewline
\hline 
TQC & Topological quantum computation \tabularnewline
\hline 
FLO & Fermionic linear optics\tabularnewline
\hline 
$\Omega$ & Set of isospectral density matrices in $\mathcal{H}$ (for the specified
ordered spectrum) \tabularnewline
\hline 
$\eta_{\Omega}^{\mathrm{corr}}$ & Fraction of correlated states on $\Omega$ (with respect to the measure
induced from $\mu$)\tabularnewline
\hline 
\end{tabular}
\par\end{center}

\newpage{}

\chapter{Introduction\label{chap:Introduction}}

\section{General motivation\label{sec:General-motivation}}

The notion of quantum correlations does not have a well defined meaning.
Usually one uses the term ``quantum correlations'' to refer to statistical
properties of quantum states, describing a particular physical system,
that do not have a classical analogue. Consequently, the concept of
quantum correlations depends both on the system as well as on the
particular physical property in question. Consider the example of
bipartite entanglement. In this case the physical system consists
of two distinguishable particles (located in spatially separated laboratories).
It turns out that quantum mechanics allow statistical correlations
among sub-systems of a compound quantum systems that forbid the attribution
of properties to the individual parties, even if these are far apart
from each other and the global state is well defined \citep{Schroedinger1935}.
The perception of the phenomenon of entanglement changed drastically
over the years. In the famous EPR paper \citep{Einstein1935} entanglement
was considered as an argument against the completeness of quantum
mechanics. Nowadays, entanglement is considered as an important physical
resource \citep{EntantHoro,Plenio2005} that can be used to perform
tasks that would be impossible in a completely classical world. It
is therefore natural to characterize the set entangled quantum states.
Although this task is relatively easy for pure states, the problem
becomes in general very difficult \citep{EntantHoro,Guehne2009} when
one wants to characterize entangled mixed states. 

This thesis will be concerned with the description of entanglement
and types of correlations that are defined in the analogous manner.
In our considerations will use extensively technical tools of differential
geometry and representation theory of Lie groups and Lie algebras
as many types of correlations studied (see Subsection \ref{sub:Generalized-entanglement-problem}
and Section \ref{sec:Original-contributions-to}) can be conveniently
treated within this formalism. 

In this section we give a motivation and the context for the work
presented in the rest of the Thesis. This section is organized as
follows. In Subsection \ref{sub:Entanglement-of-distingioshable}
we review the concept of entanglement, a prototypical example belonging
to the class of ``quantum correlations'' that will be studied in
this thesis. In Subsection \ref{sub:Generalized-entanglement-problem}
 motivate types of correlations (other than entanglement) that will
be studied in latter chapters. We also explain how differential geometry
and Lie groups naturally appear in studies of these kind of problems.
We conclude this section in Subsection \ref{sub:Quantum-nonlocality-and}
where we present briefly examples of correlations that do not fit
into the framework considered in this thesis: quantum nonlocality
and quantum discord. Throughout this section we will assume the standard
notation used in the field of quantum information (see Section \ref{sec:Standard-mathematical-structures}
for the introduction of the notation used in this thesis).

\subsection{Entanglement of distinguishable particles\label{sub:Entanglement-of-distingioshable}}

The phenomenon of quantum entanglement is one of the most intriguing
features of quantum mechanics, not present in the realm of classical
physics. It is a direct consequence of the superposition principle
applied to the composite quantum system. 

Let us introduce the concept of entanglement on the simplest example
of two spin $\frac{1}{2}$ particles (qbits). The Hilbert space corresponding
to this system is a tensor product of Hilbert spaces corresponding
to each particle,
\[
\mathcal{H}=\mathcal{H}_{A}\otimes\mathcal{H}_{B}\,,
\]
where $\mathcal{H}_{A,B}\approx\mathbb{C}^{2}$ and the subscripts
$A,B$ refer to different particles. In this space we can chose a
basis consisting of four states: $\ket 0\ket 0,\,\ket 0\ket 1,\,\ket 1\ket 0$
and $\ket 1\ket 1$, where vectors $\ket 0,\ket 1$ form a standard
basis of $\mathbb{C}^{2}$. We say that a pure quantum state $\ket{\psi}\in\mathbb{C}^{2}\otimes\mathbb{C}^{2}$
is separable if and only of it can be written in the form
\begin{equation}
\ket{\psi}=\ket{\psi_{A}}\ket{\psi_{B}}\,,\label{eq:separable first}
\end{equation}
where $\ket{\psi_{A,B}}\in\mathcal{H}_{A,B}$ are normalized single-particle
states. Pure states that do not have the form \eqref{eq:separable first}
are called entangled. An important example of an entangled state is
so-called Bell state 
\begin{equation}
\ket{\Psi}=\frac{1}{\sqrt{2}}\left(\ket 0\ket 1-\ket 1\ket 0\right)\,.\label{eq:bell state}
\end{equation}
Even the thought the state $\ket{\Psi}$ is a pure, it is impossible
to associate pure states to subsystem associated to each particle.
In order to make sense of the preceding statement we recall the notion
of the local observables acting on a composite quantum system. The
local observables associated to the particle $A$ have the form
\[
X_{A}=X\otimes\mathbb{I}_{B}\,,
\]
where $X$ is an arbitrary Hermitian operator on $\mathcal{H}_{A}$
and $\mathbb{I}_{B}$ is the identity operator on $\mathcal{H}_{B}$.
Analogously, local observables associated to particle $B$ are
\[
Y_{B}=\mathbb{I}_{A}\otimes Y\,,
\]
where $Y$ is an arbitrary Hermitian operator on $\mathcal{H}_{B}$
and $\mathbb{I}_{A}$ is the identity operator on $\mathcal{H}_{A}$.
The expectation values of local observables on the state $\ket{\Psi}$
are given by
\begin{equation}
\bra{\Psi}X\otimes\mathbb{I}_{B}\ket{\Psi}=\frac{1}{2}\mathrm{tr}\left(X\right)\,,\,\bra{\Psi}\mathbb{I}_{A}\otimes Y\ket{\Psi}=\frac{1}{2}\mathrm{tr}\left(Y\right)\,,\label{eq:local expectation}
\end{equation}
where $\mathrm{tr}\left(\cdot\right)$ is a trace of the operator.
By the virtue of Eq.\eqref{eq:local expectation} the expectation
values of local measurements performed on a system in state $\ket{\Psi}$
depend on local observables only through traces of $X$ and $Y$ respectively.
Consequently, it is impossible to associate pure states $\ket{\psi_{A}},\,\ket{\psi_{B}}$
that would describe the ``local states'' of particles $A$ and $B$
respectively. In fact, the only choice of local states consistent
with \eqref{eq:local expectation} is to associate maximally mixed
states to each particle (see Subsection \ref{sec:Standard-mathematical-structures}).

The state $\ket{\Psi}$ can be also used to demonstrate the essential
idea of the EPR paradox and is an example of a state exhibiting quantum
nonlocality (see Subsection \ref{sub:Quantum-nonlocality-and}). The
definition of entanglement for mixed quantum states is more involved.
A two qbit mixed state $\rho$ is said to be separable if and only
if it can expressed as a probabilistic mixture of pure separable states.
In other words
\begin{equation}
\rho=\sum_{i=1}^{N}p_{i}\kb{\psi_{i}}{\psi_{i}}\,,\label{eq:separable quantum state}
\end{equation}
where $p_{i}\geq0$, $\sum_{i=1}^{N}p_{i}=1$ and the states%
\footnote{Throughout this thesis we identify pure quantum states with orthonormal
projectors onto one dimensional subspaces in the Hilbert space of
a system. However, for the sake of simplicity, we will ignore in this
chapter the difference between states (denoted by $\kb{\psi}{\psi})$
and their vector representatives (denoted by $\ket{\psi)}$. %
} $\kb{\psi_{i}}{\psi_{i}}$ are separable (note that the presentation
\ref{eq:separable quantum state} is in general non-unique). A mixed
state $\rho$ is called entangled if and only if it is not separable.
We now sketch, following \citep{Werner1989}, the reasoning leading
to this definition of entanglement in general mixed states. Let us
notice that for the product pure qbit states (i.e. states of the form
$\rho=\rho_{A}\otimes\rho_{B}$ ) the joint measurements of quantum
observables corresponding to each particle factorize
\[
\mathrm{tr}\left(\left[\rho_{A}\otimes\rho_{B}\right]X\otimes Y\right)=\mathrm{tr}\left(\rho_{A}X\right)\cdot\mathrm{tr}\left(\rho_{B}Y\right)\,.
\]
Therefore the statistics measurements of observables associated to
each particle are described by independent random variables (for any
choice of single particle observables $X,Y$ ). Consequently, if the
state of the system is of the form $\rho_{A}\otimes\rho_{B}$, then
the simultaneous measurement of local observables correspond to performing
the separate experiments on two separate systems (described respectively
by the state $\rho_{A}$ and $\rho_{B}$ ). The correlations can be
introduced into the system if one has a generator of states that generates,
with the probability $p_{j}$, a product state $\rho_{A}^{j}\otimes\rho_{B}^{j}$
( $i=1,\ldots,M$ ). Now the quantum state describing the statistical
ensemble generated in this way is given by 
\begin{equation}
\rho=\sum_{j=1}^{M}p_{j}\rho_{A}^{j}\otimes\rho_{B}^{j}\,.\label{eq:separable alternative}
\end{equation}
The observation that every state $\rho$ of the form \eqref{eq:separable alternative}
can be casted in the form \eqref{eq:separable quantum state} and
vice versa motivates the definition of entanglement introduced above
\footnote{For the alternative motivation of such a definition, involving the
possible application of an entangled state to certain quantum information
processing tasks using LOCC protocols, see \citep{Masanes2008}.%
}. We have presented the concept of entanglement on the example of
two qbit system. The generalization to multipartite setting and general
single particle Hilbert spaces is straightforward.

The quantum entanglement is not only a peculiarity of quantum mechanics.
It is one of the cornerstones of the field of quantum information
and can be used as a resource to perform tasks that classically are
either impossible or very inefficient to implement. The examples include:
quantum teleportation \citep{Bennett1993}, superdense coding \citep{Bennett1992}
and quantum key distribution \citep{Bennett1984,Ekert1991} and precise
quantum metrology \citep{Giovannetti2011}. Although these applications
are the primary reason for the interest in entanglement, we limit
ourselves only to giving a short list given above. For an overview
of the role of entanglement in quantum information and quantum computing
see \citep{EntantHoro}. 

Let us conclude our considerations by a simple, yet profound observation.
The separability problem, i.e. a problem to decide whether a given
mixed quantum state is entangled or not, admits an important symmetry.
The set of pure separable states is invariant under the transformations
of the form%
\footnote{We focused on the case of two qbits but the analogous statements hold
also for the multipartite scenarios and for the arbitrary single particle
Hilbert spaces.%
}
\[
\kb{\psi}{\psi}\rightarrow\kb{\psi'}{\psi'}=U_{A}\otimes U_{B}\kb{\psi}{\psi}U_{A}^{\dagger}\otimes U_{B}^{\dagger}\,,
\]
where $U_{A}$ and $U_{B}$ are arbitrary unitary operators acting
on particle $A$ and $B$ respectively. Unitary operator of the form
$U_{A}\otimes U_{B}$ form a group called the local unitary group
$\mathrm{LU}$. This group encodes unitary local evolutions of a composite
system. Every two pure separable states can be transformed onto each
other by the conjugation via some element of a local unitary group.
Therefore the property of being entangled is invariant under the action
of the local unitary group. The local unitary group is the product
of Lie groups, $\mathrm{LU}=\mathrm{SU}\left(2\right)\times\mathrm{SU}\left(2\right)$,
and thus is itself a Lie group%
\footnote{The group $\mathrm{SU}\left(2\right)$ denotes the group of complex
unitary $2\times2$ matrices having a unit determinant.%
}. The above reasoning establishes a connection between the separability
problem and the theory of Lie groups and differential geometry. In
this part we give a list of concrete types of correlations that will
be investigated and discuss briefly their physical relevance.

\subsection{A generalized separability problem\label{sub:Generalized-entanglement-problem}}

Throughout this thesis we will be studying types of correlations that
can be defined formally in analogous manner to entanglement of distinguishable
particles. For this type of correlations one first specifies a class
of ``non-correlated'' pure states. The correlated states are then
defined as states that cannot be represented as probabilistic combinations
of pure ``non-correlated states''.

\subsubsection*{Correlations in systems with fixed number of non-distinguishable
particles. }

There exist different approaches to define entanglement for systems
consisting of fixed number of identical bosons or fermions. This is
caused by the fact that in these systems the underlying Hilbert space
is no longer a product of Hilbert spaces of individual subsystems
but rather a symmetric and antisymmetric part of it. Moreover, due
to non-distinguishablility of bosons and fermions it is impossible
to access operationally individual particles (via local operations
or measurements). The Hilbert spaces associated to system of $L$
bosons and respectively $L$ fermions are
\begin{equation}
\mathcal{H}_{b}=\mathrm{Sym}^{L}\left(\mathcal{H}\right)\,,\,\mathcal{H}_{f}=\mathrm{\bigwedge}^{L}\left(\mathcal{H}\right)\,,\label{eq:hilb spaces nondist}
\end{equation}
where $\mathcal{H}$ is a single-particle Hilbert space. By the virtue
of \eqref{eq:hilb spaces nondist} most pure bosonic and every pure
fermionic state can be treated as an entangled state of a system of
~$L$ identical distinguishable particles. But since the access to
the designated subsystems is restricted by indistinguishably, what
is the use of this kind of entanglement? For this reason the particle
entanglement described above is treated sometimes as a mathematical
artifact \citep{Esteve2008,Benatti2011}, not as a useful resource
(like the usual entanglement of distinguishable particles). In this
thesis we will not discuss further the problem of the definition of
the concept of entanglement for non-distinguishable particles which
is a subject of an ongoing debate (c.f. \citep{Benatti2011,Benatti2012,KilloranPlenio2014}).
Instead, using the analogy with the case of entanglement, we will
investigate the notions of correlations defined with the use of the
simplest pure states available in $\mathcal{H}_{b}$ and $\mathcal{H}_{f}$
respectively. 

In the case of $L$ identical bosons we will study particle entanglement.
In this scenario non-entangled pure states are of the form form
\begin{equation}
\ket{\psi}=\ket{\phi}^{\otimes L}\,,\label{eq:bos separable}
\end{equation}
where $\ket{\phi}$ is an arbitrary normalized vector from the single
particle Hilbert space $\mathcal{H}$ describing a single bosonic
particle. Just like in the case of distinguishable particles, non-correlated
(non-entangled) mixed states are defined as probabilistic mixtures
of states of the form \eqref{eq:bos separable}. Such a definition
is motivated by the fact that it is useful many physical contexts.
\begin{itemize}
\item In a recent paper \citep{KilloranPlenio2014} it was shown that particle
entanglement of bosons can be activated, via the procedure of mode
splitting, to the ``true'' entanglement of distinguishable particles.
\item Product bosonic states \eqref{eq:bos separable} appear as variational
states in the Hartree-Fock method applied to weakly interacting bosonic
systems \citep{Fetter2003}.
\item Entangled bosonic states (in a sense of the above definition) are
necessary to achieve ``Heisenberg scaling'' quantum metrology with
$L$ photonic states \citep{Giovannetti2011}.
\item States of the form \eqref{eq:bos separable}, for finite dimensional
Hilbert space $\mathcal{H}$, are precisely coherent states of the
group $\mathrm{SU\left(N\right)}$ in $\mathrm{Sym}^{L}\left(\mathcal{H}\right)$,
where $N=\mathrm{dim}\left(\mathcal{H}\right)$. For $N=2$ we recover
Bloch coherent states \citep{Wodkiewicz1985,Puri2001} for spin $s=\frac{N-1}{2}$.
\end{itemize}
In systems consisting of $L$ identical fermions we define non-correlated
pure states as states that can be expressed as a single Slater determinant,
\begin{equation}
\ket{\psi}=\ket{\phi_{1}}\wedge\ket{\phi_{2}}\wedge\ldots\wedge\ket{\phi_{L}}\,,\label{eq:slater determin}
\end{equation}
where vectors $\left\{ \ket{\phi_{i}}\right\} _{i=1}^{L}$ are arbitrary
normalized and pairwise orthogonal states form the single particle
Hilbert space $\mathcal{H}$. As before non-correlated (``separable'')
mixed fermionic states are defined as states being convex combinations
of Slater determinants. Such a definition is motivated by the following
observations:
\begin{itemize}
\item States of the form \eqref{eq:slater determin} are the simplest tensors
available in $\mathcal{H}_{f}$ \citep{SchliemannTwoFermions2001,EckertFermions2002}.
\item Slater determinants are used as variational class of states in the
Hartee-Fock method applied to fermionic (e.g., electron) systems \citep{Fetter2003}.
\item Slater determinants \eqref{eq:slater determin} are, for finite dimensional
Hilbert space $\mathcal{H}$, precisely coherent states of the group
$\mathrm{SU\left(N\right)}$ in $\bigwedge^{L}\left(\mathcal{H}\right)$,
where $N=\mathrm{dim}\left(\mathcal{H}\right)$. 
\end{itemize}

\subsubsection*{Non-convex-Gaussian correlations in fermionic systems}

Another interesting type of correlations that can be defined in an
analogous manner to entanglement are non-convex-Gaussian correlations
in fermionic systems \citep{powernoisy2013}. In this scenario one
considers a fermionic system in which fermions can occupy $d$ modes,
but the total number of particles can be arbitrary%
\footnote{Due to the Fermi exclusion principle the number of particles present
in such a system cannot exceed $d$.%
}. Consequently, the appropriate Hilbert space associated to this system
is $d$ mode Fermionic Fock space%
\footnote{For the precise definitions of this space and other objects that appear
in this part see Subsections \ref{sub:Spinor-represenations-of} and
\ref{sub:Fermionic-Gaussian-states}).%
} $\mathcal{H}_{\mathrm{Fock}}\left(\mathbb{C}^{d}\right)$. We fix
the class of non-correlated pure states to consist of pure fermionic
Gaussian states \citep{powernoisy2013,Lagr2004}. Consequently, the
set pf non-correlated mixed states consists of convex-Gaussian states,
i.e. states that are probabilistic (convex) combinations of pure Gaussian
states. This kind of correlations is relevant for models of quantum
computation based on fermionic Linear Optics \citep{powernoisy2013}
and Topological Quantum Computation with so-called Ising anyons \citep{universalfracBravyi}
(see Section \ref{sec:Classical-simulation-of-FLO} for details).
Moreover, pure fermionic Gaussian states are also used as variational
states in Bogolyubov-Hartee-Fock mean field theory \citep{Fetter2003,Derezinski2011}.
Lastly, pure Gaussian states are also called ``pure spinors'' and
are ``coherent sates'' of spinor representations of the Lie group
$\mathrm{Spin}\left(2d\right)$.

\subsubsection*{A unified picture}

The examples of correlations given above (entanglement, particle entanglement
of bosons, ``entanglement'' of fermions, non-convex-Gaussian fermionic
correlations) share an important common feature. In each of these
examples there exist a continuous symmetry group acting on the Hilbert
space of the system in question and the class of non-correlated pure
states consists of ``coherent states'' of the action of this group
\citep{GenCohPer}. For this reason the so-defined correlations are
invariant under the action of this group (in the case of entanglement
this group is the local unitary group). Exactly the same situation
occurs in the case of coherent optical sates%
\footnote{In this case the relevant Hilbert space is the Bosonic Fock space
and the coherent states are generated from the vacuum by the action
of displacement operators. %
} \citep{Puri2001}. In this scenario the non-correlated mixed states
are states possessing the positive $P$ representation (in the context
of quantum optics such states are also called classical states).

A unified perspective for studing such types of correlations was put
forward in \citep{Barnum2004}. In the cited article authors introduced
in the notion of ``generalized entanglement'' defined by specifying
``non-entangled'' states as convex combinations of general coherent
states of compact simply-connected Lie groups (see also \citep{Klyachko2008}).
In this work we will study mainly these types of correlations, focusing
on the generalized separability problem, i.e. a problem to characterize
(for a given scenario) the set of non-correlated mixed states. We
will investigate the following aspects of this problem (see Section
\ref{sec:Original-contributions-to} for details).
\begin{enumerate}
\item Characterization of the set of non-correlated pure states via zeros
of polynomial in state's density matrix.
\item An analytical characterization for the set of non-correlated mixed
states (in special cases).
\item Derivation of the family of polynomial criteria detecting correlations
in mixed states of the system in question.
\item How typical are correlated states among all mixed states of the considered
quantum system?
\end{enumerate}
In addition to the correlations defined via coherent states of a suitable
symmetry group, we will also study the generalized separability problem
for more general classes of states. In particular we will study the
``refined'' notions of entanglement: 
\begin{itemize}
\item Genuine multiparty entanglement for multipartite systems of distinguishable
particles \citep{EntantHoro,Guehne2009}
\item Refined classification of bipartite entanglement based on the notion
of the Schmidt number of a quantum states \citep{SchmidtNumHoro}. 
\end{itemize}
In the above cases one also defines the suitable classes non-correlated
pure states (see Subsections \ref{sub:GME} and \ref{sub: Schmid rank}
for details ) which are preserved by the action of the relevant symmetry
group. 

We would like to emphasize the importance of the role of symmetry
groups in the scheme presented above. The presence of the symmetries
allow us to use methods belonging to the realm of differential geometry
and representation theory of Lie groups and Lie algebras. These methods
allow to tackle problems that otherwise would have been very difficult
or impossible to solve. Such mathematical techniques have been used
before in entanglement theory. The following aspects of entanglement
were successfully treated via geometric and group-theoretic methods. 
\begin{itemize}
\item Local unitary equivalence of multipartite pure states%
\footnote{We say that two pure states $\kb{\psi}{\psi}$ and $\kb{\psi'}{\psi'}$
are locally unitary equivalent if and only if one can be transferred
onto another by some local unitary operation.%
}. This problem has been investigated with the usage of symplectic
geometry \citep{Sawicki2011,Maciazek2013} as well as the method of
group-invariant polynomials \citep{Grassl1998,Vrana2011,Vrana2011a}. 
\item Characterization of various classes of entanglement in pure states
for the multipartite setting \citep{Duer2000,Verstraete2002,Sawicki2013a,Sawicki2012,Sawicki2014}. 
\end{itemize}
However, to our knowledge, the geometric and group-theoretic methods
were not applied systematically to study the generalized separability
problem. Filling up this gap is the aim of this thesis. 

In the end of this part we would like to mention that the results
presented in work are of the independent interest in pure mathematics.
The class of the non-correlated mixed states in the generalized separability
problem is an example of an orbitope \citep{Sanyal2011}. An orbitope
is defined as a convex hull of an orbit of a linear action of a Lie
group acting in a vector space. Orbitopes have became recently a subject
of active mathematical studies \citep{Sanyal2011,Holmes1999,Biliotti2011,Grabowski2013}.

\subsection{Quantum nonlocality and quantum discord\label{sub:Quantum-nonlocality-and}}

Not all types correlations of appearing in the quantum information
theory have the form presented in the preceding paragraph. Here we
illustrate this fact by briefly discussing two examples of such correlations:
quantum nonlocality \citep{Brunner2014} and quantum discord \citep{Modi2012}.

\subsubsection*{Quantum nonlocality}

In order to introduce the concept of quantum nonlocality we first
define the so-called Bell scenario. Consider two observers (called
Alice and Bob) that are located in specially separated laboratories
and can perform independent experiments on a certain physical system.
We assume that during each experimental run each observer can independently
chose one of $2$ observables (we denote Alice's and Bob's observables
by $x_{1},x_{2}$ and $y_{1},y_{2}$ respectively) that can give $2$
different results%
\footnote{The generalization to arbitrary number of measurement setting and
measurement outcomes is straightforward.%
} (labeled by $a_{1},a_{2}$ and $b_{1},b_{2}$ for Alice and Bob respectively).
Under the assumption that each measurement setup is equally probable
the statistics of joint measurement of Alice and Bob are completely
described via the collection of conditional probabilities
\[
p\left(ab|xy\right)\,,\, a\in\left\{ a_{1},a_{2}\right\} ,\, b\in\left\{ b_{1},b_{2}\right\} ,\, x\in\left\{ x_{1},x_{2}\right\} ,\, y\in\left\{ y_{1},y_{2}\right\} \,.
\]

The collection of probabilities $p\left(ab|xy\right)$ is said to
be admit a local model if and only if for each measurement setting
$xy$ we have
\begin{equation}
p\left(ab|xy\right)=\int_{\Lambda}p\left(a|x,\lambda\right)p\left(b|y,\lambda\right)p\left(\lambda\right)d\lambda\,,\label{eq:local hiddend}
\end{equation}
where $\Lambda$ is the space of ``hidden parameters'' equipped
with the measure $p\left(\lambda\right)d\lambda$. Functions $p\left(a|x,\lambda\right)\,\text{and}\, p\left(b|y,\lambda\right)$
describe the results of measurements of Alice and Bob, conditioned
on the knowledge of the hidden variable $\lambda\in\Lambda$. 

Assume now that the conditional probabilities $p\left(ab|xy\right)$
come from quantum mechanics, i.e. there exist a two qbit state $\rho$
such that 
\[
p\left(ab|xy\right)=\mathrm{\mathrm{tr}\left(E_{a}^{x}\otimes E_{b}^{y}\rho\right)}\,,
\]
where operators $\left\{ E_{a}^{x}\right\} _{a=a_{1},a_{2}}$ and
$\left\{ E_{b}^{y}\right\} _{b=b_{1},b_{2}}$ describe local measurements
of Alice and respectively, Bob for each choice of the measurement
setting $xy$ (for fixed $xy$ and each pair of indices $a,b$ operators
$E_{a}^{x},E_{b}^{y}$ are rank one orthonormal projectors in $\mathbb{C}^{2}$).
Let us inverse the picture and take as a primary object the two qbit
quantum state $\rho$. A state $\rho$ posses a local hidden variable
model (LVM) %
\footnote{We do not describe here the philosophical justification for the name
``local hidden variable model'', focusing only on the technical
definition of this concept. For a complete justification see for example
\citep{Brunner2014}.%
} if and only if for every choice of local measurements $\left\{ E_{a}^{x}\right\} _{a=a_{1},a_{2}}$
and $\left\{ E_{b}^{y}\right\} _{b=b_{1},b_{2}}$ the corresponding
sets of conditional probabilities $p\left(ab|xy\right)$ admit are
local (in a sense of Eq.\eqref{eq:local hiddend}). 

States that do not admit a local hidden variable model are said to
exhibit quantum nonlocality. An example of such a state is already
the mentioned maximally entangled Bell state $\ket{\Psi}$ (see Eq.\eqref{eq:bell state}).
This fact was used in the famous Bell paper \citep{Bell1964} to show
that some predictions of quantum mechanics cannot be reproduced by
theories admitting local hidden variables (understood in a sense of
Eq.\eqref{eq:local hiddend}).

States that admit LMV form a convex set. However this set is not a
convex hull of some class of pure states. In particular in \citep{Werner1989}
it was showed that in the set of states admitting LMV is strictly
larger than the set of separable states.

\subsubsection*{Quantum discord}

Let us first recall the concept of the mutual information of two discrete
random variable $A,B$, characterized by the joint probability distribution
$p_{A,B}\left(a,b\right)$ (indices $a,b$ run over some implicit
discrete set of indices). We denote by $p_{A}\left(a\right)$ ($p_{B}\left(b\right)$)
the marginal distribution of the random variable $A$ (respectively
$B$). The mutual information, denoted by $I\left(A,B\right)$, is
defined by 
\begin{equation}
I\left(A,B\right)=H\left(A\right)+H\left(B\right)-H\left(A,B\right)\,,\label{eq:first formula}
\end{equation}
where $H\left(A\right)$ is the Shannon entropy of the random variable
$A$,
\[
H\left(A\right)=-\sum_{a}p_{A}\left(a\right)\mathrm{ln}\left(p_{A}\left(a\right)\right)\,.
\]
and $H\left(A,B\right)$ is the joint entropy of variables $A\,\text{and}\, B$,
\[
H\left(A,B\right)=-\sum_{a,b}p_{A,B}\left(a,b\right)\mathrm{ln}\left(p_{A,B}\left(a,b\right)\right)\,.
\]
The mutual information $I\left(A,B\right)$ can be considered as a
measure of how much more information is obtained during a singe measurement
of two variables $A,B$ jointly, as opposed to measuring variables
$A,B$ separately. The mutual information can be also considered as
an indicator of correlations between $X$ and $Y$ as it is non-negative
and vanish if and only the joint probability distribution has the
product form, $p_{A,B}\left(a,b\right)=p_{A}\left(a\right)\cdot p_{B}\left(b\right)$.
There exist another equivalent formula for the mutual information,
\begin{equation}
I\left(A,B\right)=H\left(A\right)-H\left(A|B\right)\,,\label{eq:secod formula matual}
\end{equation}
where $H\left(A|B\right)$ is the conditional entropy given by
\[
H\left(A|B\right)=\sum_{b}p_{B}\left(b\right)\left[-\sum_{a}p_{A}\left(a|b\right)\mathrm{ln}\left(p_{A}\left(a|b\right)\right)\right]\,,
\]
where $p_{a}\left(a|b\right)=\frac{p_{A,B}\left(a,b\right)}{p_{B}\left(b\right)}$
is the probability of the random variable $A$ conditioned on the
value of the random variable $B$. 

The concept of quantum discord, introduced independently in \citep{Henderson2001}
and \citep{Ollivier2001}, stems form the fact that two classically
equivalent formulas \eqref{eq:first formula} and \eqref{eq:secod formula matual}
for $I\left(A,B\right)$ do not give the same results in the quantum
domain. In quantum mechanics the immediate generalization of the $I\left(A,B\right)$
is quantum mutual information defined on the two party state $\rho^{AB}$
(we apply the usual convention: symbols $A$ and $B$ refer to different
subsystems of a composite system). It is given by 
\[
I\left(\rho^{AB}\right)=S\left(\rho^{A}\right)+S\left(\rho^{B}\right)-S\left(\rho^{AB}\right),\,
\]
where $S\left(\rho\right)=-\mathrm{\mathrm{tr}}\left(\rho\mathrm{ln}\left(\rho\right)\right)$
is the von-Neumann entropy of the density matrix $\rho$ and $\rho^{A},\rho^{B}$
denote one particle reduced density matrices of the state $\rho^{AB}$.
Another possible generalization of $I\left(A,B\right)$ involves Eq.\eqref{eq:secod formula matual}
and uses the notion of quantum conditional entropy. In difference
to the classical case, in quantum mechanics the conditional entropy
of a state $\rho^{XY}$ depends on the measurement scenario used in
the second subsystem
\begin{equation}
S\left(A|\left\{ E_{b}\right\} \right)=\sum_{b}p_{b}S\left(\rho_{A|b}\right)\,,\label{eq:alternative cond}
\end{equation}
where $\left\{ E_{b}\right\} $ is a collection of generalized measurements%
\footnote{The generalized measurements in quantum mechanics are given by a positive
operator-valued measure \citep{NielsenChaung2010}, i.e. a collection
of non-negative operators $\left\{ E_{\alpha}\right\} $ such that
$\sum_{\alpha}E_{\alpha}=\mathbb{I}$. The probability of the outcome
$\alpha$ if the system is in the state $\rho$ is given by $p_{\alpha}=\mathrm{tr}\left(E_{\alpha}\right)$.%
} performed in the subsystems $B$, $p_{b}=\mathrm{tr}\left(\mathbb{I}\otimes E_{b}\rho^{AB}\right)$
and $\rho_{A|b}=\frac{1}{p_{b}}\mathrm{tr}_{B}\left(\mathbb{I}\otimes E_{b}\rho^{AB}\right)$
(for the definition of partial trace $\mathrm{tr}_{B}\left(\cdot\right)$
see Subsection \ref{sec:Standard-mathematical-structures}). Using
\eqref{eq:alternative cond} and \eqref{eq:secod formula matual}
one obtains another generalization of $I\left(A,B\right)$,
\[
\mathcal{J}\left(A|\left\{ E_{b}\right\} \right)=S\left(\rho^{A}\right)-S\left(A|\left\{ E_{b}\right\} \right)\,.
\]

The quantum discord, denoted by $\mathcal{D}\left(A|B\right)$, is
now defined by \citep{Henderson2001,Ollivier2001}
\[
\mathcal{D}\left(A|B\right)=\max_{\left\{ E_{b}\right\} }\left[I\left(\rho^{AB}\right)-\mathcal{J}\left(A|\left\{ E_{b}\right\} \right)\right]\,,
\]
where the optimization is over all generalized local measurements
on the subsystem $B$. The quantum discord is a measure of correlations
more general than entanglement. It is non-negative and invariant under
the action of local unitary group \citep{Modi2012}. Moreover, non-correlated
states (states with the vanishing quantum discord) must be of the
form of so called quantum-classical states, i.e.
\begin{equation}
\rho^{AB}=\sum_{i}p_{i}\sigma_{i}\otimes\kb ii\,,\label{eq:QC}
\end{equation}
where $\left\{ p_{i}\right\} $ is a probability distribution, $\sigma_{i}$
are arbitrary quantum states on the first subsystem and $\left\{ \ket i\right\} $
form the orthonormal basis of the Hilbert space associated to the
second subsystem. States of the form \eqref{eq:QC} do not form a
convex set \citep{Ollivier2001}. Consequently, correlations defined
with the usage of quantum discord cannot be captured in the framework
presented in Subsection \ref{sub:Generalized-entanglement-problem}.
Interestingly enough, the class of quantum classical states can be
effectively described with the framework of symplectic geometry \citep{oszman2014c}.

\section{Outline of the thesis\label{sec:Original-contributions-to}}

Throughout this work we will be studing the notion of correlations
which comes from the generalization of the entanglement problem (see
Section \ref{sec:General-motivation}). We start with the subset $\mathcal{M}$
of the set of all pure states $\mathcal{D}_{1}\left(\mathcal{H}\right)$
on the Hilbert space $\mathcal{H}$ of interest,
\begin{equation}
\mathcal{M}\subset\mathcal{D}_{1}\left(\mathcal{H}\right)\,.\label{eq:choice of the class}
\end{equation}
In what follows we will refer to $\mathcal{M}$ as to the set of ``non-correlated''
pure states. Consequently, all states that do not belong to $\mathcal{M}$
are called correlated.  One can imagine that (just like the case of
the entanglement problem, where $\mathcal{M}$ consisted of pure separable
states) we have an easy access to ``free'' ore ``simple'' states
from the class $\mathcal{M}$. In analogy to the case of entanglement
we extend the notion of correlations to the realm of mixed states
$\mathcal{D}\left(\mathcal{H}\right)$ of the considered system. We
now define the class on ``non-correlated'' mixed states, denoted
by $\mathcal{M}^{c}$, as states that can be obtained as probabilistic
mixtures (convex combinations) of pure states from the set $\mathcal{M}$
(treated as a subset of $\mathcal{D}\left(\mathcal{H}\right)$),
\begin{equation}
\mathcal{M}^{c}=\left\{ \left.\sum_{i}p_{i}\kb{\psi_{i}}{\psi_{i}}\,\right|\,\sum_{i}p_{i}=1,\, p_{i}\geq0\,,\,\kb{\psi_{i}}{\psi_{i}}\in\mathcal{M}\right\} \,.\label{eq:convex hull def}
\end{equation}
In the language of convex geometry the resulting set is called the
\textit{convex hull} of the set $\mathcal{M}$ (which explain the
notation we used). Just like in the case of pure states we define
mixed correlated states as states that do not belong to the set $\mathcal{M}^{c}$.
The pictorial representation of the above construction is given in
Figure \ref{fig:Pictorial-representation-of}.

\begin{figure}[h]
\centering{}\includegraphics[width=9cm]{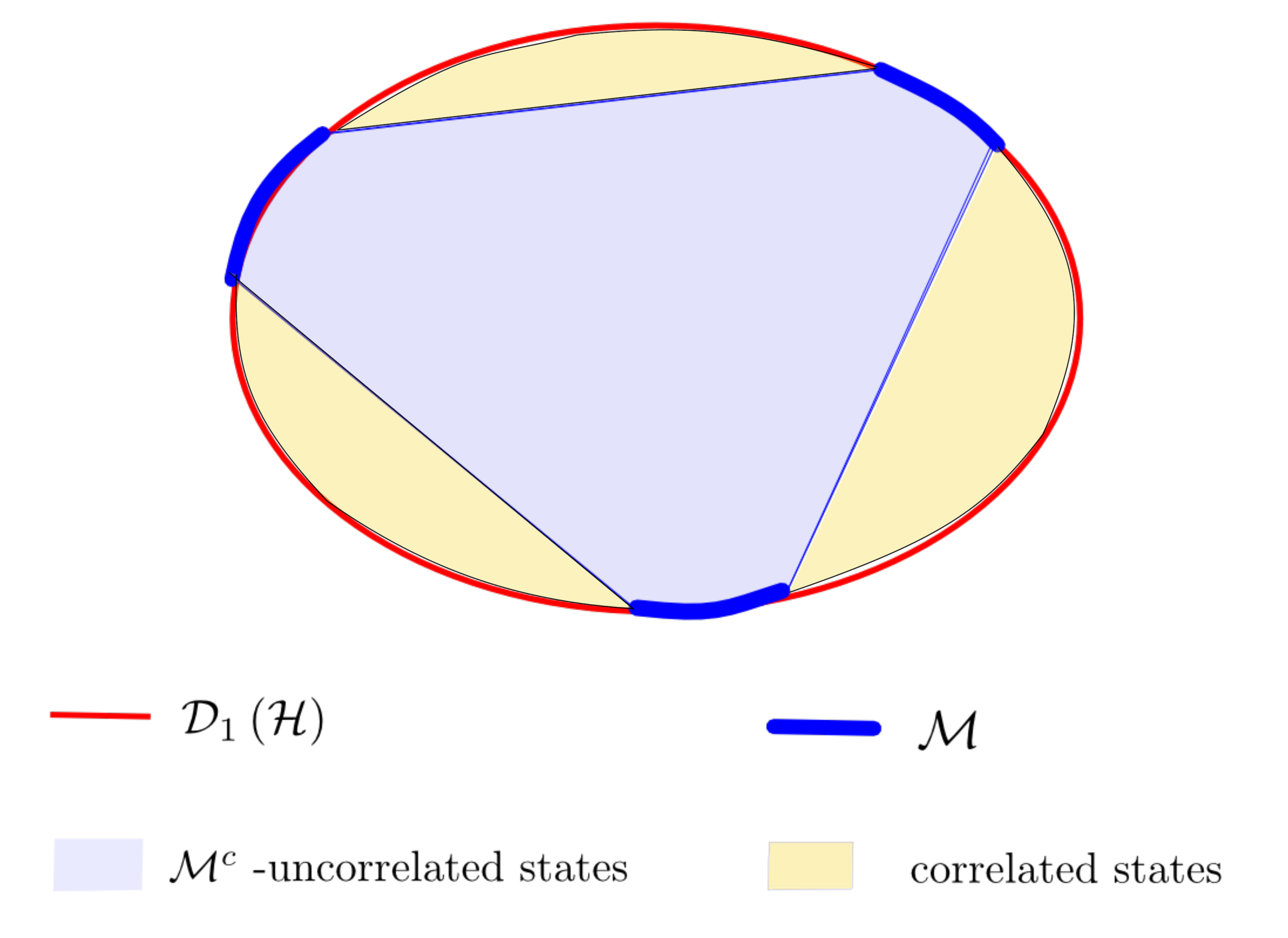}\protect\caption{\label{fig:Pictorial-representation-of}Pictorial representation of
the definition of correlated states (with respect to the choice of
$\mathcal{M}$) }
\end{figure}

In the above discussion we did not specified the structure of the
set $\mathcal{M}$, the class of ``correlations'' it represents,
nor the concrete physical system of interest. In fact, on the technical
level, we simply associated the ``lack of correlations'' to states
belonging to $\mathcal{M}$ and extended this notion to the problem
of entanglement. However we want to stress that the language introduced
above allows to study at the same footing many types of correlations
that bear a physical meaning (see Section \ref{sec:General-motivation}
for the list of examples). Throughout the thesis we will limit ourselves
to classes of states $\mathcal{M}$ that can be defined as the zero
set of a non-negative homogenous polynomial in the state’s density
matrix 
\begin{equation}
\kb{\psi}{\psi}\in\mathcal{M}\,\Longleftrightarrow\mathrm{tr}\left(\left[\kb{\psi}{\psi}^{\otimes k}\right]A\right)=0\,,\label{eq:definition of a class}
\end{equation}
where $A$ is the non-negative operator on the $k$-fold symmetrization
of the Hilbert space $\mathcal{H}$. In fact the expression $\mathrm{tr}\left(\left[\kb{\psi}{\psi}^{\otimes k}\right]A\right)$,
for a suitable choice of the operator $A$, can be used to define
arbitrary homogenous non-negative polynomial%
\footnote{By a homogenous polynomial $\mathcal{P}$ defined on a set of Hermitian
operator $X$ we understand a function which is a homogenous polynomial
in matrix elements $\bra iX\ket j$ of the operator $X$. For instance,
for $2\times2$ Hermitian matrix $X=\begin{pmatrix}a & b\\
b^{\ast} & c
\end{pmatrix}$ an exemplary homogenous real polynomial is $\mathcal{P}\left(X\right)=a^{2}+\left|b\right|^{2}+c^{2}$.%
} in pure state $\kb{\psi}{\psi}$. The equation \eqref{eq:definition of a class}
will be the leitmotif that will be used throughout the thesis. In
Chapters \ref{chap:Multilinear-criteria-for-pure-states}-\ref{chap:Typical-properties-of},
which form the core of the thesis, we will investigate the natural
problems and questions concerning the types of correlations defined
above. 

In \textbf{Chapter \ref{chap:Multilinear-criteria-for-pure-states}}
we will show that many classes of pure states naturally appearing
in quantum information and quantum many body physics can be defined
via the condition \eqref{eq:definition of a class}, for a suitable
choice of the polynomial $\mathcal{P}$. The examples include: 
\begin{enumerate}
\item Product states for $L$ distinguishable particles. 
\item Product states for $L$ bosons.
\item Slater determinants for $L$ fermions.
\item Pure fermionic Gaussian states.
\item Multipartite states that do not exhibit genuine multiparty entanglement.
\item Bipartite states with a bounded Schmidt rank.
\end{enumerate}
The notions of correlations corresponding to each choice of the class
of non-coherent states were discussed in Section \ref{sec:General-motivation}.
In fact this is not an accident that many classes of pure states relevant
to quantum physics can be characterized by the condition \eqref{eq:definition of a class}.
We show that this is a consequence of the existence of the symmetry
group that can be used to define the class $\mathcal{M}$. To be more
precise we show that it is possible to construct a polynomial $\mathcal{P}$
for classes of states which consist of ``the generalized Perelomov
coherent states'' of the symmetry group $K$ irreducibly represented
in the (usually finite-dimensional) Hilbert space $\mathcal{H}$ (see
Figure \ref{fig:group orbit}). 

\begin{figure}[h]
\centering{}\includegraphics[width=9cm]{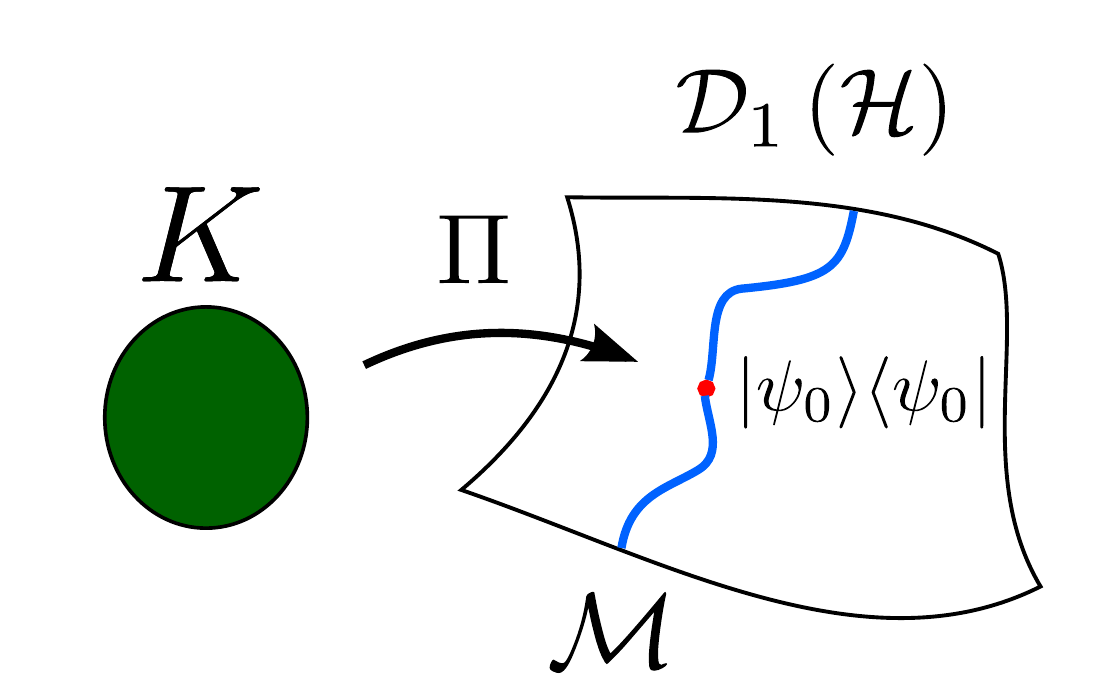}\protect\caption{\label{fig:group orbit}An graphical presentation of the set of ``generalized
Perelomov coherent states'' $\mathcal{M}\subset\mathcal{D}_{1}\left(\mathcal{H}\right)$
viewed as an orbit of the action of the symmetry group $K$ through
the state $\protect\kb{\psi_{0}}{\psi_{0}}\in\mathcal{D}_{1}\left(\mathcal{H}\right)$.
The group $K$ is irreducibly represented, via the representation
$\Pi$, in the Hilbert space $\mathcal{H}$.}
\end{figure}

We focus mostly%
\footnote{A part of Chapter \ref{chap:Multilinear-criteria-for-pure-states}
is also devoted to cases when $\mathcal{H}$ is infinite-dimensional
and $K$ is not even a Lie group. %
} on situations when the group $K$ is compact and simply-connected.
This class of groups is broad enough to cover cases 1-4 from the list
above. Moreover, this restriction allows us to use the representation
theory of semisimple Lie algebras%
\footnote{As explained in Chapter \ref{chap:Mathematical-prelimenaries} to
every compact simply-connected Lie group $K$ one can associate a
complex semisimple Lie algebra $\mathfrak{g}$.%
}. This technical tool allows us to write down explicitly the operator
$A$ that characterizes a given class of pure states.

In \textbf{Chapter \ref{chap:Complete-characterisation}} we ask in
witch cases the characterization \eqref{eq:definition of a class}
allows for a simple analytical characterization of the set of non-correlated
mixed states $\mathcal{M}^{c}$. We partially answer this question
by classifying all the cases when Perelomov coherent states of compact
simply-connected Lie groups can be characterized by a single anti-unitary
conjugation. This allows to characterize the set $\mathcal{M}^{c}$
analytically via the Uhlmann-Wotters construction (c.f. the discussion
of Wotters concurrence in Subsection \ref{sec:Methods-from-entanglement}).
We also apply our results to characterize analytically the set of
pure fermionic Gaussian states for four fermionic modes. As pointed
out in Section \ref{sec:General-motivation} this has a consequence
in the theory of quantum computation. 

In \textbf{Chapter \ref{chap:Polynomial-mixed states}} we use the
characterization \eqref{eq:definition of a class} of pure uncorrelated
states to derive polynomial criteria cor detection of correlation
in mixed states. In particular, starting just from the condition \eqref{eq:definition of a class}
we provide a universal way to construct ``nonlinear witnesses of
correlations'' in mixed states. This is a new method of constructing
witnesses of correlations different than entanglement. Moreover, for
classes 1-4 from the list above we describe completely the structure
of bilinear group-invariant criteria detecting correlations.

\textbf{Chapter \ref{chap:Typical-properties-of}} deals the following
question: are correlated states defined via Eq.\eqref{eq:definition of a class}
and the above construction typical among all possible mixed states
of the system? We answer this question partially by estimating from
below the fraction of correlated states on the manifold of isospectral
density matrices%
\footnote{The manifold of isospectral density matrices consists of states in
the considered Hilbert space that have a fixed spectrum. It is equipped
with the natural invariant measure coming form the Haar measure on
the unitary group of the Hilbert space.%
}. In order to achieve these bounds (they depend on the polynomial
$\mathcal{P}$ from Eq.\eqref{eq:definition of a class}) we use the
criteria introduced in Chapter \ref{chap:Polynomial-mixed states}
and the technique of concentration of measure.

The chapters we listed above are related to each other and shed a
light on the mathematical structure of correlations defined via the
choice of pure states $\mathcal{M}$ given by Eq.\eqref{eq:definition of a class}.
Many techniques that we use in this thesis comes form differential
geometry and representation theory of Lie groups. As these branches
of mathematics are not commonly used in quantum information theory
we present the necessary background in these fields in \textbf{Chapter
\ref{chap:Mathematical-prelimenaries}}. The relations between Chapters
\ref{chap:Multilinear-criteria-for-pure-states}-\ref{chap:Typical-properties-of},
with indicated aspects in which differential geometry and representation
theory are employed, are presented in Figure \ref{fig:skeme}.
\begin{figure}[h]
\centering{}\includegraphics[width=14cm]{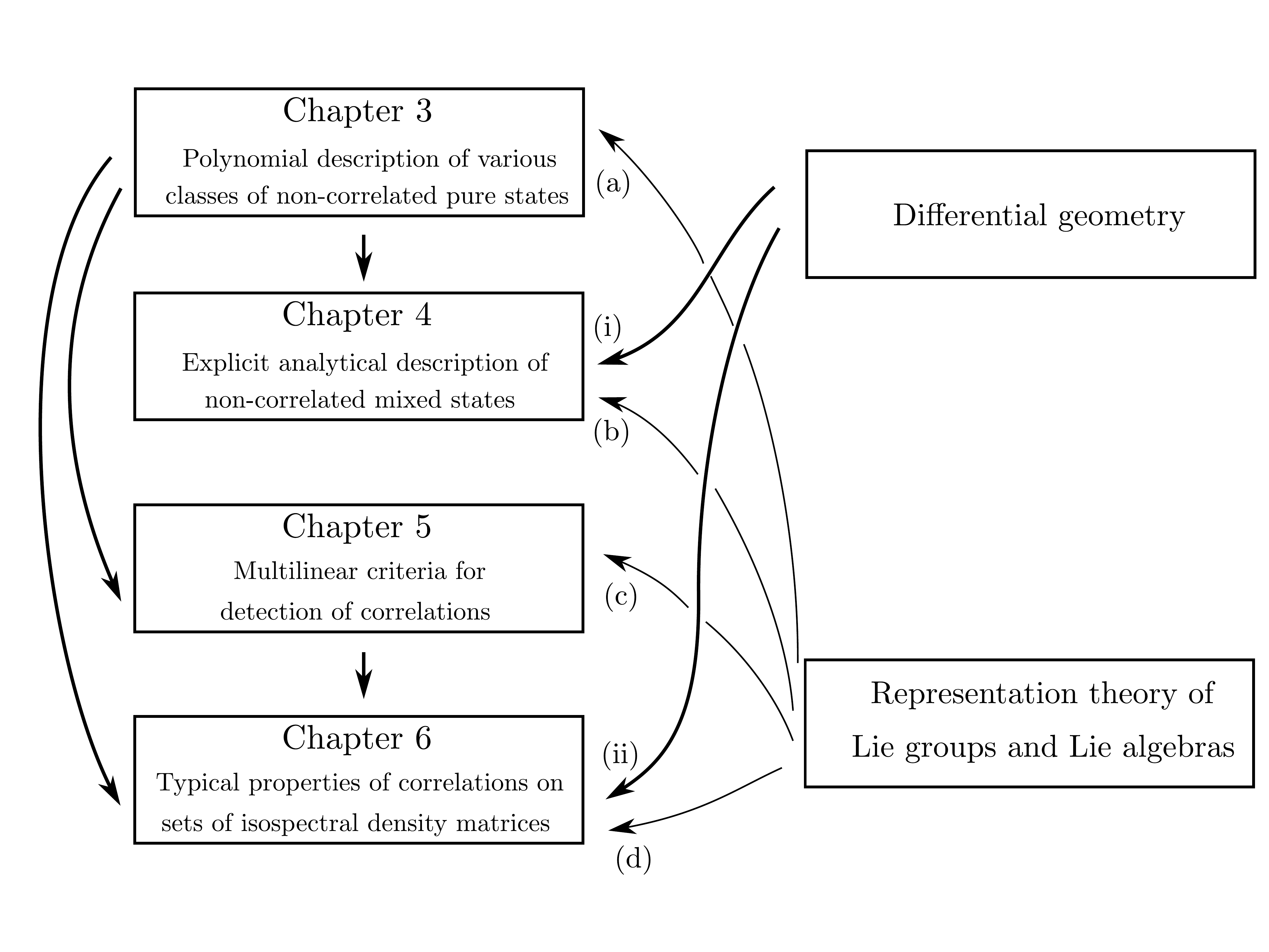}\protect\caption{\label{fig:skeme}A graphical presentation of the connections between
the core chapters of the thesis, with indicated applications of differential
geometry and representation theory of Lie groups and Lie algebras.
The indices labeling arrows have the following meaning: (i) Actions
of Lie groups on spheres, (ii) Concentration of measure on manifolds
of isospectral density matrices, (a) Perelomov coherent states, (b)
Representation theory of compact Lie groups, (c) Structure of irreducible
representations of semisimple Lie algebras, (d) dimensions of certain
irreducible representations of semisimple Lie algebras.}
\end{figure}

\chapter{Mathematical preliminaries\label{chap:Mathematical-prelimenaries}}

The purpose of this chapter is to present the mathematical concepts
and the terminology that will be used in the rest of the thesis. In
the first part of the chapter we sketch the standard algebraic formulation
of quantum mechanics. We use this opportunity to establish the notation
that will be used throughout the text. In the second part of the chapter
we give a survey of basic facts from differential geometry and theory
of Lie groups and Lie algebras. In the third section we review methods
of entanglement theory that are relevant for this thesis. The third
section contains some basic material from differential geometry. In
the last part of the chapter we briefly describe representation theory
of Lie groups and algebras, with emphasis on the representation theory
of compact simply-connected Lie groups and semisimple Lie algebras.
We will not be concerned here with functional-analytical details and
more advanced differential geometry. Whenever necessary, appropriate
mathematics will be presented in relevant parts of the thesis.

\section{Standard mathematical structures of quantum mechanics\label{sec:Standard-mathematical-structures}}

The basic object in quantum mechanics is a complex separable Hilbert
space $\mathcal{H}$ which is associated with every physical system. 
\begin{defn}
A complex Hilbert space $\mathcal{H}$ is a complex vector space equipped
with the $1\frac{1}{2}$ linear product $\bk{\cdot}{\cdot}$ such
that the space is complete with respect to the norm $\left\Vert \psi\right\Vert =\sqrt{\bk{\psi}{\psi}}$
induced by the inner product. A Hilbert space is called separable
if  it has a countable basis. 
\end{defn}
In the context of quantum information theory one usually deals with
finite dimensional Hilbert spaces and for this reason we do not need
to bother with the separability condition in most cases considered
in this thesis. Nevertheless, when necessary, functional-analytic
details related to infinite-dimensional Hilbert spaces will be provided.
Throughout the rest of this chapter we assume that the dimension of
$\mathcal{H}$ is finite. We adopt the usual physical convention that
the inner product $\bk{\cdot}{\cdot}$ is anti-linear in the first
argument and linear in the second argument, 
\[
\bk{\alpha\phi}{\psi}=\bar{\alpha}\bk{\phi}{\psi},\,\bk{\phi}{\alpha\psi}=\alpha\bk{\phi}{\psi},\,\phi,\psi\in\mathcal{H},\,\alpha\in\mathbb{C}\,.
\]
 We will also use the ``bra-ket'' physical notation which is in
agreement with the above convention. For any vector $\psi\in\mathcal{H}$
we identify $\psi$ with the ``ket'' $\ket{\psi}$. By the ``bra''
$\bra{\psi}$ we understand the linear form on $\mathcal{H}$ given
by the pairing $\mathcal{H}\ni\phi\rightarrow\bk{\psi}{\phi}\in\mathbb{C}$.
To every Hilbert space $\mathcal{H}$ we have associated the space
of linear operators acting on $\mathcal{H}$ which we denote by $\mathrm{End}\left(\mathcal{H}\right)$
and a space of invertible linear operators $\mathrm{GL}\left(\mathcal{H}\right)$.
For every operator $A\in\mathrm{End}\left(\mathcal{H}\right)$ we
define its Hermitian conjugate $A^{\dagger}$ via the relation, 
\[
\bk{\phi}{A\psi}=\bk{A^{\dagger}\phi}{\psi}\,,\,\mbox{\ensuremath{\ket{\phi}}},\,\ket{\psi}\in\mathcal{H}\,.
\]
Throughout the thesis we will also use the notation commonly used
in the folklore of quantum mechanics: $\bk{\phi}{A\psi}\equiv\bra{\phi}A\ket{\psi}$. 

The space $\mathrm{End}\left(\mathcal{H}\right)$ carries a natural
Hilbert space structure given by the Hilbert-Schmidt inner product,
\[
\bk AB_{\mathrm{HS}}=\mathrm{tr}\left(A^{\dagger}B\right),\, A,B\in\mathrm{End}\left(\mathcal{H}\right)\,.
\]
In the above expression $\mathrm{tr}\left(\cdot\right)$ denotes the
trace of the linear operator. In quantum mechanics the special role
is played by Hermitian operators, i.e., operators $A\in\mathrm{End}\left(\mathcal{H}\right)$
satisfying the condition $A^{\dagger}=A$. We denote the set of Hermitian
operators by $\mathrm{Herm}\left(\mathcal{H}\right)$. Hermitian operators
are diagonalizable and have real eigenvalues, i.e. they can be presented
as
\begin{equation}
A=\sum_{i}\lambda_{i}\mbox{\ensuremath{\kb{\psi_{i}}{\psi_{i}}}}\,,\label{eq: spectral decomp}
\end{equation}
where $\lambda_{i}$ denote real eigenvalues of $A$ and $\mbox{\ensuremath{\kb{\psi_{i}}{\psi_{i}}}}$
denote the orthonormal projectors onto one dimensional subspaces spanned
by the normalized eigenvectors $\ket{\psi_{i}}\in\mathcal{H}$ ($\bk{\psi_{i}}{\psi_{i}}=1$).
Vectors $\left\{ \ket{\psi_{i}}\right\} $ form the orthonormal basis
of $\mathcal{H}$. Note that we can have $\lambda_{i}=\lambda_{j}$
for $i\neq j$. In such a case the spectrum of $A$ is said to be
degenerated.
\begin{defn}
A set of non-negative operators on $\mathcal{H}$, $\mathrm{Herm}_{+}\left(\mathcal{H}\right)$,
consists of Hermitian operators satisfying $\bk{\psi}{A|\psi}\geq0$
for all $\ket{\psi}\in\mathcal{H}$. Equivalently, non-negative operators
can be characterized as Hermitian operators that have non-negative
eigenvalues. We will write $A\geq0$ for $A\in\mathrm{Herm}_{+}\left(\mathcal{H}\right)$.
\end{defn}
We will write $A\geq B$ if operators $A,B\in\mathrm{Herm}\left(\mathcal{H}\right)$
satisfy $A-B\geq0$. General states of a quantum system can be characterized
as subsets of $\mathrm{Herm}_{+}\left(\mathcal{H}\right)$.
\begin{defn}
Pure states on $\mathcal{H}$ are defined as rank-1 orthogonal projectors
acting on $\mathcal{H}$. General mixed states on $\mathcal{H}$ are
defined as operators $\rho$ satisfying $\rho\geq0$ and $\mathrm{Tr}\left(\rho\right)=1$.
We denote sets of pure and mixed states by $\mathcal{D}_{1}\left(\mathcal{H}\right)$
a $\mathcal{D}\left(\mathcal{H}\right)$ respectively. 
\end{defn}
Observables on a quantum system $\mathcal{H}$ (assuming no super-selection
rules are present) are identified with Hermitian operators $A\in\mathrm{Herm}\left(\mathcal{H}\right)$.
In general the value of the physical quantity associated to an operator
$A$ is not specified before the actual measurement takes place. Possible
values of measurements of $A$ are elements of its spectrum. Assume
that the system is in a mixed state $\rho$. The probability that
as a result of the measurement one obtains the value $\lambda$ is
given by $\mathrm{tr}\left(\rho\mathbb{P}_{\lambda}\right)$, where
$\mathbb{P}_{\lambda}$ is the projector onto the subspace spanned
by eigenvectors of $A$ for which the corresponding eigenvalue equals
$\lambda$ (see \eqref{eq: spectral decomp}). Consequently, the expectation
value of the observable described by the operator $A$ is given by
$\mathrm{tr}\left(\rho A\right)$. 

Another important class of operators on $\mathcal{H}$ are unitary
operators, i.e. operators $U\in\mathrm{End}\left(\mathcal{H}\right)$
which satisfy $U^{\dagger}=U^{-1}$. The set of unitary operators
is closed under composition and forms a group which we denote by $\mathrm{U}\left(\mathcal{H}\right)$.
One parameter family of unitary operators specify the time evolution
of an isolated quantum system. It is given by 
\begin{equation}
\rho_{t}=\mathrm{e}^{-itH}\rho_{0}\mathrm{e}^{itH}\,,\label{eq:evolution}
\end{equation}
where $\rho_{0}$ is the initial state of the system and $\rho_{t}$
is the state of the system at time $t$. Operator $H$ is a distinguished
Hermitian operator known as the Hamiltonian of the system%
\footnote{In \eqref{eq:evolution} we assumed that we work in unit system in
which the Planck constant equals one, $\hbar=1$ %
}. Hermiticity of $H$ ensures that the time evolution operator $U_{t}=\mathrm{e}^{-itH}$
is unitary, $U_{t}^{\dagger}=U_{t}^{-1}$. If we are allowed to perform
an arbitrary evolution on a quantum system system the set of possible
transformations are given by $\rho\rightarrow U\rho U^{\dagger}$,
where $U\in\mathrm{U}\mathcal{\left(H\right)}$.

So far we have discussed only isolated quantum systems. If a quantum
system consists of two subsystems A and B, the total Hilbert space
is given by the tensor product $\mathcal{H}_{A}\otimes\mathcal{H}_{B}$,
 where $\mathcal{H}_{A}$ and $\mathcal{H}_{B}$ are Hilbert spaces
describing subsystems $A$ and $B$ respectively. For the general
discussion and definition of the tensor product of Hilbert spaces
see \citep{NielsenChaung2010}. Here we present only basic properties
of $\mathcal{H}_{A}\otimes\mathcal{H}_{B}$ and omit their proofs.
The space $\mathcal{H}_{A}\otimes\mathcal{H}_{B}$ is spanned by vectors
of the form $\psi\otimes\phi$, where $\otimes:\mathcal{H}_{A}\times\mathcal{H}_{B}\rightarrow\mathcal{H}_{A}\otimes\mathcal{H}_{B}$
is a nontrivial bilinear mapping. Vectors of the form $\psi\otimes\phi$
are called separable vectors or simple tensors. The inner product
$\bk{\cdot}{\cdot}_{\mathcal{H}_{A}\otimes\mathcal{H}_{B}}$ on $\mathcal{H}_{A}\otimes\mathcal{H}_{B}$
is defined in terms of the inner products of Hilbert spaces describing
subsystem A and B: $\bk{\cdot}{\cdot}_{\mathcal{H}_{A}}$ and $\bk{\cdot}{\cdot}_{\mathcal{H}_{B}}$.
For simple tensors we have
\[
\bk{\psi_{1}\otimes\phi_{1}}{\psi_{2}\otimes\phi_{2}}_{\mathcal{H}_{A}\otimes\mathcal{H}_{B}}=\bk{\psi_{1}}{\psi_{2}}_{\mathcal{H}_{A}}\bk{\phi_{1}}{\phi_{2}}_{\mathcal{H}_{B}}\,.
\]
The inner product between arbitrary vectors from $\mathcal{H}_{A}\otimes\mathcal{H}_{B}$
is defined by demanding linearity in the second and anti-linearity
in the first argument of $\bk{\cdot}{\cdot}_{\mathcal{H}_{A}\otimes\mathcal{H}_{B}}$.
In what follows we will drop the subscripts in the inner products
corresponding to different subsystems. Also, in the Dirac ``bra-ket''
notation introduced earlier we will use interchangeably, unless it
causes a confusion, the following forms: 
\[
\bra{\psi\otimes\phi}\equiv\bra{\psi}\otimes\bra{\phi}\equiv\bra{\psi}\bra{\phi}\,,\,\ket{\psi\otimes\phi}\equiv\ket{\psi}\otimes\ket{\phi}\equiv\ket{\psi}\ket{\phi}\,.
\]
Assume now that the joint system is in the state $\rho\in\mathcal{D}\left(\mathcal{H}_{A}\otimes\mathcal{H}_{B}\right)$.
We want to compute the expectation value of an observable $A\in\mathrm{Herm}\left(\mathcal{H}_{A}\otimes\mathcal{H}_{B}\right)$
which is local, i.e. takes the form $A=\mathbb{I}_{A}\otimes\tilde{A}$,
where $\mathbb{I}_{A}$ is the identity on $\mathcal{H}_{A}$ and
$\tilde{A}\in\mathrm{Herm}\left(\mathcal{H}_{B}\right)$. It turns
out that the expectation value of $\rho$ on such $A$ can by expressed
via the ``local part'' of $\rho$. More precisely, for a given $\rho\in\mathcal{D}\left(\mathcal{H}_{A}\otimes\mathcal{H}_{B}\right)$
there exist a unique state $\rho_{B}\in\mathcal{D}\left(\mathcal{H}_{B}\right)$
such that for all $\tilde{A}\in\mathrm{Herm}\left(\mathcal{H}_{B}\right)$
we have,

\begin{equation}
\mathrm{tr}\left[\rho\left(\mathbb{I}_{A}\otimes\tilde{A}\right)\right]=\mathrm{tr}\left(\rho_{B}\tilde{A}\right)\,.\label{eq:partial trace}
\end{equation}
Operator $\rho_{B}$ is called the partial trace (over the subsystem
$A$) of $\rho$. In what follows we will use the notation $\rho_{B}=\mathrm{tr}_{A}\left(\rho\right)$.
The operation of partial trace is $\mathbb{R}$-linear in $\rho$
and can be extended to the $\mathbb{C}$-linear mapping $\mathrm{tr}_{A}:\mathrm{End}\left(\mathcal{H}_{A}\otimes\mathcal{H}_{B}\right)\rightarrow\mathrm{End}\left(\mathcal{H}_{B}\right)$
which preserves Hermiticity. In an analogous manner we define the
partial trace over subsystem $B$. The same construction can be introduced
for systems that consist of many subsystems. In such cases the total
Hilbert space has a structure of multiple tensor product $\mathcal{H}=\mathcal{H}_{1}\otimes\ldots\otimes\mathcal{H}_{L}$,
where $L$ is a number of subsystems. 

We now describe the most general dynamics that the quantum system,
described by the initial state $\rho_{0}\in\mathcal{D}\left(\mathcal{H}\right)$,
can undergo. Let $\rho$ denote the state after the evolution and
let $\mbox{\ensuremath{\Lambda}}$ be the map defining this evolution.
We demand that the evolution to be linear in $\rho$ and thus respects
the linear structure on $\mathrm{End}\left(\mathcal{H}\right)$. Also,
whatever happens during the evolution, the final state $\rho$ must
remain positive semi-definiteness if the evolved state $\tilde{\rho}$
has to be interpreted as some density matrix. The evolution must be
thus described by the so-called positive map. The following definition
takes into account the possibility that the output state of a given
process might not be defined on the same Hilbert space. 
\begin{defn}
The space of positive maps between $\mathrm{End}\left(\mathcal{H}_{1}\right)$
and $\mathrm{End}\left(\mathcal{H}_{2}\right)$, denoted by $\mathcal{P}\left(\mathcal{H}_{1},\mathcal{H}_{2}\right)$,
is defined as the set of linear mappings $\Lambda:\mathrm{End}\left(\mathcal{H}_{1}\right)\rightarrow\mathrm{End}\left(\mathcal{H}_{2}\right)$
such that for $A\ge0$ we have $\Lambda\left(A\right)\geq0$. We will
use the notation $\mathcal{P}\left(\mathcal{H},\mathcal{H}\right)=\mathcal{P}\left(\mathcal{H}\right)$.
\end{defn}
Another invariant during the hypothetical evolution must be the trace
$\rho$. Nevertheless, positivity and trace-preservation are just
necessary but not a sufficient conditions requested from a map representing
quantum evolution. The system in question can be always treated as
a part of a composite system described by the Hilbert space $\mathcal{H}\otimes\mathcal{H}_{env}$,
where $\mathcal{H}_{env}$ describe some environment. Under the evolution
 the state of a compound system, initially equal $\rho'=\rho\otimes\rho_{env}$,
must also evolve keeping the positive semi-definiteness intact, even
if ‘nothing happens’ to the environment itself. It means that the
map $\Lambda\otimes\mathbb{I}_{\mathcal{H}_{env}}$ (representing
lack of actual evolution of the environment) must also be a positive
map acting on density states of the system plus the environment. A
map $\Lambda$ fulfilling this condition for all auxiliary Hilbert
spaces $\mathcal{H}_{env}$ is called a completely positive map (in
short CP map). We denote the set of completely positive maps between
$\mathrm{End}\left(\mathcal{H}_{1}\right)$ and $\mathrm{End}\left(\mathcal{H}_{2}\right)$
by $\mathcal{CP}\left(\mathcal{H}_{1},\mathcal{H}_{2}\right)$. 
\begin{defn}
A positive map $\Lambda\in\mathcal{P}\left(\mathcal{H}_{1},\mathcal{H}_{2}\right)$
is called completely positive if and only if it satisfies $\Lambda\otimes\mathbb{I}_{p}\in\mathcal{P}\left(\mathcal{H}_{1}\otimes\mathbb{C}^{p},\,\mathcal{H}_{2}\otimes\mathbb{C}^{p}\right)$
for all natural $p$. 
\end{defn}
A CP map which is trace preserving is called a quantum channel. The
set of all quantum channels $\Lambda:\mathrm{End}\left(\mathcal{H}\right)\rightarrow\mathrm{End}\left(\mathcal{H}\right)$,
which we denote by $\mathcal{CP}_{0}\left(\mathcal{H}\right)$ can
be viewed as a set consisting of all possible evolutions that states
from $\rho$ can undergo%
\footnote{It is important to mention that a restriction of a channel $\Lambda\in\mathcal{CP}_{0}\left(\mathcal{H}_{A}\otimes\mathcal{H}_{B}\right)$
to a subsystem $\mathcal{H}_{A}$ in general does not lead to a completely
positive map, provided the initial sate of a total system is not separable
\citep{NielsenChaung2010}. %
}. Latter on in the thesis we will need some basic properties of CP
maps which we state here for completeness. The criterion for being
a CP map is given by the Jamiołkowski-Choi mapping. 
\begin{fact}
(Jamiołkowski-Choi \citep{NielsenChaung2010}) There is one to one
correspondence between completely positive maps on $N$ dimensional
complex Hilbert space $\mathcal{H}$ and non-negative positive operators
on $\mathcal{H}\otimes\mathcal{H}$. The isomorphism is given by the
Jamiołkowski-Choi mapping $J:\mathrm{\mathcal{CP}\left(\mathcal{H}\right)}\rightarrow\mathrm{Herm}_{+}\left(\mathcal{H}\otimes\mathcal{H}\right)$,
\begin{equation}
J\left(\Lambda\right)=\left(\mathbb{I}\otimes\Lambda\right)\left(\kb{\Psi}{\Psi}\right)\,,\label{eq:Jamio=000142kowski isomoprhism}
\end{equation}
where $\mathbb{I}$ is the identity operator on $\mathcal{H}$, $\Lambda\in\mathcal{CP}\left(\mathcal{H}\right)$
and $\ket{\Psi}=\sum_{i=1}^{i=N}\ket i\ket i$ for some fixed orthonormal
basis $\left\{ \ket i\right\} _{i=1}^{i=N}$ of $\mathcal{H}$. 
\end{fact}
Checking the complete positivity can be thus reduced to determining
whether the corresponding operator $J\left(\Lambda\right)$ acting
on $\mathcal{H}\otimes\mathcal{H}$ is positive semi-definite. To
this end we need the inverse of the Jamiołkowski-Choi mapping $J^{-1}:\mathrm{Herm}_{+}\left(\mathcal{H}\otimes\mathcal{H}\right)\rightarrow\mathcal{CP}\left(\mathcal{H}\right)$.
It is given by the following formula,
\begin{equation}
\left(J^{-1}\left(A\right)\right)\left(\rho\right)=\mathrm{tr}_{1}\left[\left(\rho^{T}\otimes\mathbb{I}\right)A\right]\,,\label{eq:inverse jamio=000142kowski}
\end{equation}
where $A\in\mathrm{Herm}_{+}\left(\mathcal{H}\otimes\mathcal{H}\right)$,
$\rho\in\mathrm{End}\left(\mathcal{H}\right)$, $\mathrm{tr}_{1}:\mathrm{End}\left(\mathcal{H}\otimes\mathcal{H}\right)\rightarrow\mathrm{End}\left(\mathcal{H}\right)$
is the partial trace over the first factor of $\mathcal{H}\otimes\mathcal{H}$,
and $\rho^{T}$ is a transpose of the operator $\rho$ in the basis
$\left\{ \ket i\right\} _{i=1}^{i=N}$. We will also use the so-called
Kraus decomposition (see \citep{NielsenChaung2010}) of CP maps.
\begin{fact}
\label{kraus decomposition}For each $\Lambda\in\mathcal{CP}\left(\mathcal{H}\right)$
there exist a set of operators $T_{\alpha}:\mathcal{H}\rightarrow\mathcal{H}$
($\alpha\in\mathcal{A}$, where $\mathcal{A}$ is the set of indices)
such that for all $\rho\in\mathrm{End}\left(\mathcal{H}\right)$,
\begin{equation}
\Lambda\left(\rho\right)=\sum_{\alpha\in\mathcal{A}}T_{\alpha}\rho T_{\alpha\,.}^{\dagger}\label{eq:krauss decomposition}
\end{equation}

\end{fact}
The map $\Lambda$ expressed in the form of Kraus decomposition \eqref{eq:krauss decomposition}
is manifestly CP. Operators $T_{\alpha}$ appearing in \eqref{eq:krauss decomposition}
are called Kraus operators. Decomposition \eqref{eq:krauss decomposition}
is by no means unique. Note however, that in the special case of the
unitary evolution (which is, of course a CP map) $\Lambda\left(\rho\right)=U\rho U^{\dagger}$
it is enough to consider just one Kraus operator $T_{\alpha}=U$. 

Not all operations important in quantum mechanics are represented
by linear operators. An important example is a class of antiunitary
operators. 
\begin{defn}
A mapping $\theta:\mathcal{H}\rightarrow\mathcal{H}$ is called antiunitary
if and only if for all $\ket{\psi},\ket{\phi}\in\mathcal{H}$, $a,b\in\mathbb{C}$
\[
\theta\left(a\ket{\psi}+b\ket{\phi}\right)=a^{\ast}\theta\ket{\psi}+b^{\ast}\theta\ket{\phi}
\]
and
\begin{equation}
\bk{\theta\psi}{\theta\phi}=\bk{\psi}{\phi}^{\ast}\,.\label{eq:antiunitary}
\end{equation}

\end{defn}
Antiunitary operators are important as they describe symmetries of
quantum systems such as the operation of time reversal \citep{Wigner}.
The structure of antiunitary operators was given by Wigner \citep{NormalWigner1960}.
He showed that every antiunitary operator $\theta$ has the form $\theta=UK$,
where $U$ is an unitary operator and $K$ is a complex conjugation
is some fixed orthonormal basis $\left\{ \ket i\right\} _{i=1}^{i=N}$
of $\mathcal{H}$,
\begin{equation}
K\left(\sum_{i=1}^{N}a_{i}\ket i\right)=\sum_{i=1}^{N}a_{i}^{\ast}\ket i\,.\label{eq:complex conjugation}
\end{equation}

The class of antiunitary operators will be extensively used in Chapter
\ref{chap:Complete-characterisation} as they allow, in certain cases,
to describe analytically classes of mixed correlated states (in a
sense used throughout this thesis, see Section \ref{sec:General-motivation}).

\section{Methods of entanglement theory\label{sec:Methods-from-entanglement}}

As it was indicated in Chapter \ref{chap:Introduction}, this thesis
focuses mostly on the generalized separability problem, i.e. a problem
to characterize the set of separable (or equivalently entangled) states.
Many of the tools used in the context of the separability problem
can be also applied to its generalization. In this section we present
a quick survey of basic facts from the theory of entanglement detection,
focusing on the bipartite and finite dimensional case. See \citep{Guehne2009}
for the review of methods used to detect entanglement.

Let us start with the formal definition of bipartite separable states.
A set of separable bipartite states $\mathcal{S}\subset\mathcal{D}\left(\mathbb{C}^{d_{1}}\otimes\mathbb{C}^{d_{2}}\right)$
is defined as a convex hull of pure product states $\mathcal{M}_{d}\subset\mathcal{D}_{1}\left(\mathbb{C}^{d_{1}}\otimes\mathbb{C}^{d_{2}}\right)$.
\begin{equation}
.\mathcal{S}=\mbox{\ensuremath{\mathcal{M}}}_{d}^{c}=\left\{ \left.\sum_{i}p_{i}\kb{\psi_{i}}{\psi_{i}}\,\right|\,\sum_{i}p_{i}=1,\, p_{i}\geq0\,,\,\ket{\psi_{i}}=\ket{\phi_{i}^{A}}\ket{\phi_{i}^{B}}\right\} \,.\label{eq:separable}
\end{equation}

Due to the importance of the concept of entanglement (see Subsection
\ref{sub:Entanglement-of-distingioshable}) one would like to have
a universal method to detect whether arbitrary state $\rho\in\mathcal{D}\left(\mathbb{C}^{d_{1}}\otimes\mathbb{C}^{d_{2}}\right)$
is separable or not. Entanglement detection have been proven to be
very hard from the perspective of the theory of computational complexity.
Let $\left\Vert \cdot\right\Vert $ denote the trace norm on $\mathrm{Herm}\left(\mathbb{C}^{d_{1}}\otimes\mathbb{C}^{d_{2}}\right)$.
\begin{problem*}
($\mathrm{WMSP}\left(\epsilon\right)$, Weak membership problem for
separability) Given a density matrix $\rho\in\mathcal{D}\left(\mathbb{C}^{d_{1}}\otimes\mathbb{C}^{d_{2}}\right)$
with the premise that either (i) $\rho\in\mathcal{S}$ or (ii) $\min_{\sigma\in\mathcal{S}}\left\Vert \rho-\sigma\right\Vert \geq\epsilon$,
decide which is the case. 
\end{problem*}
In \citep{Gurvits2003} it was proven that the above problem is $\mathrm{NP}$-hard
for $\epsilon=\mathrm{exp}\left(-O\left(d\right)\right)$, where $d=\sqrt{d_{1}d_{2}}$.
This result was later improved in \citep{Gharibian} where the $\mathrm{NP}$-hardness
was proven for $\epsilon=O\left(\frac{1}{\mathrm{poly}\left(d\right)}\right)$,
where $\mathrm{poly}\left(d\right)$ denotes some polynomial of finite
degree in $d$. The intuitive meaning of $\mathrm{NP}$-hardness of
$\mathrm{WMSP}\left(\epsilon\right)$ is that every problem from the
class of $\mathrm{NP}$ can be reduced to $\mathrm{WMSP}\left(\epsilon\right)$
in the polynomial time \citep{Aaronson2013}. This in turn strongly
suggest the hardness of $\mathrm{WMSP}\left(\epsilon\right)$ since
for many problems belonging to the class $\mathrm{NP}$ there are
no known effective (polynomial time) algorithms that solve them (a
well known example is the traveling salesman problem).

\subsubsection*{Entanglement witnesses and positive maps}

For the reason presented above one suspects that in general separability
problem is ``hard''. Therefore it is natural to relax our expectations
and focus only on partial characterization of the set of separable/entangled
states. The basic tool for detecting entanglement are entanglement
witnesses \citep{Horodecki1996,Terhal2000}.
\begin{defn*}
A Hermitian operator $W\in\mathrm{Herm}\left(\mathbb{C}^{d_{1}}\otimes\mathbb{C}^{d_{2}}\right)$
is called entanglement witness if and only if
\begin{enumerate}
\item $\mathrm{tr\left(\kb{\psi}{\psi}W\right)}\geq0$ for all pure separable
states $\kb{\psi}{\psi}\in\mathcal{M}_{d}$.
\item $\mathrm{tr\left(\rho W\right)}<0$ for some entangled state $\rho\in\mathcal{D}\left(\mathbb{C}^{d_{1}}\otimes\mathbb{C}^{d_{2}}\right)$.
\end{enumerate}
The set of entanglement witnesses will be denoted by $\mathcal{W}$.
\end{defn*}
From the definition of the entanglement witness it follows that $\mathrm{tr\left(\sigma W\right)\geq0}$
for all separable states $\sigma\in\mathcal{S}$. Consequently, if
for a given mixed state $\rho\in\mathcal{D}\left(\mathbb{C}^{d_{1}}\otimes\mathbb{C}^{d_{2}}\right)$
we can find an entanglement witness $\tilde{W}\in\mathcal{W}$ such
that $\mathrm{tr\left(\rho\tilde{W}\right)}<0$, then we can infer
that $\rho$ is entangled. In such a situation we say that the entanglement
witness $\tilde{W}$ detects the state $\rho$. The following fact
shows that with the help of entanglement witnesses we can describe
the whole set of separable states $\mathcal{S}$.
\begin{fact}
(\citep{Horodecki1996}) Every entangled state $\rho\in\mathcal{D}\left(\mathbb{C}^{d_{1}}\otimes\mathbb{C}^{d_{2}}\right)$
is detected by some entanglement witness $W\in\mathcal{W}$.\label{witness completeness}
\end{fact}
The proof of the above relies on the fact that the set of separable
state $\mathcal{S}$, understood as a subset of $\mathrm{Herm}\left(\mathbb{C}^{d_{1}}\otimes\mathbb{C}^{d_{2}}\right)$
is closed (in the standard trace norm on $\mathrm{Herm}\left(\mathbb{C}^{d_{1}}\otimes\mathbb{C}^{d_{2}}\right)$)
and convex. Therefore, for an operator $X\in\mathrm{Herm}\left(\mathbb{C}^{d_{1}}\otimes\mathbb{C}^{d_{2}}\right)$,
outside $\mathcal{S}$, there exist a hyperplane separating this point
from $\mathcal{S}$. Let us note that the analogous result holds for
every convex and closed subset $\tilde{\mathcal{S}}$ of $\mathcal{D}\left(\mathcal{H}\right)$,
where $\mathcal{H}$ is a finite-dimensional Hilbert space. For this
reason one can introduce the ``generalized correlation witnesses''
for the scenarios introduced in Chapter \ref{chap:Introduction} (provided
the set of non-correlated states $\mathcal{M}$ is closed). The result
analogous to Fact \ref{witness completeness} holds also in this generalized
setting. The similar considerations can be repeated in  the infinite-dimensional
setting but the proof requires Hahn-Banach separation theorem \citep{Reed1972}.

An elegant characterization of entanglement witnesses can obtained
via applying the following theorem%
\footnote{In Section \ref{sec:Standard-mathematical-structures} we introduced
the  isomorphism between $\mathcal{CP}\left(\mathcal{H}\right)$ and
$\mathrm{Herm}_{+}\left(\mathcal{H}\otimes\mathcal{H}\right)$. In
fact formulas \eqref{eq:Jamio=000142kowski isomoprhism} and \eqref{eq:inverse jamio=000142kowski}
can be used to define isomorphism $\mathcal{J}:\mathrm{End}\left(\mathrm{End}\left(\mathcal{H}\right)\right)\rightarrow\mbox{\ensuremath{\mathrm{End}}}\left(\mathcal{H}\otimes\mathcal{H}\right)$.
Analogous formulas can be used to get the isomorphism (also called
Jamiołkowski-Choi isomorphism)
\[
\tilde{J}:\mathrm{End}\left(\mathrm{End}\left(\mathcal{H}_{1}\right),\mathrm{End}\left(\mathcal{H}_{2}\right)\right)\rightarrow\mbox{\ensuremath{\mathrm{End}}}\left(\mathcal{H}_{1}\otimes\mathcal{H}_{2}\right)\,.
\]
}. 
\begin{fact}
\label{positive maps}(Characterization of entanglement via positive
maps \citep{Horodecki1996}) Let $\rho\in\mathcal{D}\left(\mathbb{C}^{d_{1}}\otimes\mathbb{C}^{d_{2}}\right)$.
The state $\rho$ is separable if and only if for any natural $d$
and for any positive map $\Lambda\in\mathcal{P}\left(\mathbb{C}^{d_{1}},\mathbb{C}^{d}\right)$
an operator $\left(\Lambda\otimes\mathbb{I}_{d_{2}}\right)\left(\rho\right)$
is positive. In fact it suffices to check the positivity of $\left(\Lambda\otimes\mathbb{I}_{d_{2}}\right)\left(\rho\right)$
for $\Lambda\in\mathcal{P}\left(\mathbb{C}^{d_{1}},\mathbb{C}^{d_{2}}\right)$.
\end{fact}
Since $\left(\Lambda\otimes\mathbb{I}_{d_{2}}\right)\left(\rho\right)$
is automatically positive if $\Lambda$ is completely positive, the
result given in Fact \ref{positive maps} motivates an interest in
positive but not completely positive maps. Bellow we discuss two examples
of maps that are positive but not completely positive: the transposition
map and the reduction map. They give rise to important criteria for
detection of entanglement.

The transposition map $T:\mathrm{End}\left(\mathbb{C}^{d}\right)\rightarrow\mathrm{End}\left(\mathbb{C}^{d}\right)$
is defined by first specifying the basis (not necessarily orthonormal)
$\left\{ \ket i\right\} _{i=1}^{d}$ of $\mathbb{C}^{d}$. The action
of $T$ is first defined on rank one operators 
\[
T\left(\kb ij\right)=\kb ji
\]
and than extended by linearity to the whole $\mathrm{End}\left(\mathbb{C}^{d}\right)$.
In other words, if we identify $X\in\mathrm{End}\left(\mathbb{C}^{d}\right)$
with its representation in some basis, the transposition map maps
$X$ to its transposition. It is a known fact \citep{Horodecki1996,Peres1996}
that $T\in\mathcal{P}\left(\mathbb{C}^{d}\right)$ but $T\notin\mathcal{CP}\left(\mathbb{C}^{d}\right)$.
Therefore, by the virtue of Fact \ref{positive maps}, the transposition
map can be used to formulate the criterion for entanglement. This
criterion, called a positive partial transposition (PPT) or Peres-Horodecki
criterion, takes the following form
\begin{equation}
\rho\in\mathcal{S}\,\Longrightarrow\,\left(T\otimes\mathbb{I}_{d_{2}}\right)\left(\rho\right)\geq0\,.\label{eq:ppt criterion}
\end{equation}
If we expand a bipartite state $\rho$ in a natural basis of a tensor
product $\mathbb{C}^{d_{1}}\otimes\mathbb{C}^{d_{2}}$,
\[
\rho=\sum_{i,j,k,l}\rho_{j,l}^{i,k}\kb ij\otimes\kb kl
\]
we see that 
\[
\left(T\otimes\mathbb{I}_{d_{2}}\right)\left(\rho\right)=\sum_{i,j,k,l}\rho_{j,l}^{i,k}\kb jk\otimes\kb kl=\sum_{i,j,k,l}\left[\rho^{T_{1}}\right]_{i,l}^{j,k}\kb ij\otimes\kb kl,
\]
where $\left[\rho^{T_{1}}\right]_{i,l}^{j,k}=\rho_{i,l}^{j,k}$. This
motivates the name of the criterion. By the virtue of \eqref{eq:ppt criterion},
if a given state does not have a positive partial transposition then
it is automatically entangled. The importance of PPT criterion comes
from the fact that for low dimensional single particle Hilbert spaces
the PPT criterion gives the complete characterization of separability.
\begin{fact}
(\citep{Horodecki1996}) Let $\rho\in\mathcal{D}\left(\mathbb{C}^{d_{1}}\otimes\mathbb{C}^{d_{2}}\right)$
and let $d_{1}\cdot d_{2}\leq6$. Then the following equivalence holds
\[
\rho\in\mathcal{S}\,\Longleftrightarrow\,\left(T\otimes\mathbb{I}_{d_{2}}\right)\left(\rho\right)\geq0\,.
\]

\end{fact}
For systems with higher dimensions the PPT criterion is not strong
enough to solve a separability problem. In particular for $d_{1}\cdot d_{2}>6$
there exist mixed states that posses positive partial transposition
despite being entangled. Such states are called PPT entangled states%
\footnote{The states which are PPT entangled for an important class of states.
In particular they exhibit so-called bound entanglement \citep{EntantHoro}.
However, we will not study this class of states in this thesis.%
}.

Another example of a map which is positive but not completely positive
is a reduction map \citep{Horodecki1999} $\mathcal{R}:\mathrm{End}\left(\mathbb{C}^{d}\right)\rightarrow\mathrm{End}\left(\mathbb{C}^{d}\right)$.
It is defined by the following expression
\begin{equation}
\mathcal{R}\left(X\right)=\mathrm{Tr}\left(X\right)\mathbb{I}_{d}-X\,.\label{eq:reduction map}
\end{equation}
From Fact \ref{positive maps} we infer the following criterion for
entanglement 
\begin{equation}
\rho\in\mathcal{S}\,\Longrightarrow\,\mathbb{I}_{d_{1}}\otimes\rho_{2}-\rho\geq0\,,\label{eq:reduction criterion}
\end{equation}
where $\rho_{2}=\mathrm{tr}_{1}\left(\rho\right)$. The reduction
map is positive but not completely positive. Moreover in \citep{Horodecki1999}
it was proven that $\mathcal{R}$ is decomposable, i.e.
\begin{equation}
\mathcal{R}=\Lambda_{1}+\Lambda_{2}\circ T\label{eq:decomposable}
\end{equation}
where $\Lambda_{1,2}\in\mathcal{CP}\left(\mathbb{C}^{d}\right)$ and
$T$ is a transposition map defined above. From Eq.\eqref{eq:decomposable}
it easily follows that $\mathcal{R}$ cannot detect PPT entangled
sates.

\subsubsection*{Other criteria for detection of entanglement}

In this part we discuss briefly a number of criteria for entanglement
which are not given by an entanglement witness or a positive map. 

The first example is the criterion for entanglement based on Renyi
entropy \citep{Terhal2002a}. Quantum Renyi entropy of order $\alpha\in\left(0,\infty\right)$
is given by%
\footnote{For $\alpha\rightarrow1$ quantum Renyi entropy converges to von-Neumann
entropy $H\left(\rho\right)$.%
}
\[
S_{\alpha}=\frac{1}{1-\alpha}\mathrm{ln}\left(\mathrm{tr}\left(\rho^{\alpha}\right)\right)\,.
\]
The criterion takes the following form 
\begin{equation}
\mbox{\ensuremath{\rho\in\mathcal{S}}}\,\Longrightarrow\, S_{\alpha}\left(\rho_{2}\right)\leq S_{\alpha}\left(\rho\right)\,,\,\alpha\in\left(0,\infty\right)\,.\label{eq:renyi}
\end{equation}
The violation of the inequality $S_{\alpha}\left(\rho_{2}\right)\leq S_{\alpha}\left(\rho\right)$
for some $\alpha$ is therefore an indicator that a state $\rho$
is entangled. The criterion \eqref{eq:renyi} was proven in \citep{Terhal2002a}.
It was showed that the  inequality $S_{\alpha}\left(\rho_{2}\right)\leq S_{\alpha}\left(\rho\right)$
follows from the inequality $\mathbb{I}_{d_{1}}\otimes\rho_{2}-\rho\geq0$.
Consequently, PPT entangled states cannot violate the inequality $S_{\alpha}\left(\rho_{2}\right)\leq S_{\alpha}\left(\rho\right)$. 

We now discuss the method of entanglement detection based on so-called
convex roof extension. We decide to describe this concept in the general
setting as we will make use of it in Chapter \ref{sec:Proofs-concerning-Chapter rigorous}
of this thesis. Let $V$ be a real, finite-dimensional vector space.
Let $\mathcal{E}\subset V$ be a compact subset of $W$ and let $\mathcal{E}^{c}$
be its convex hull (due to compactness of $\mathcal{M}$ it is also
compact). Let $f:\mathcal{E}\rightarrow\mathbb{R}$ be a continuous
function. We define its convex roof extension \citep{Uhlmann2010},
denoted by $f^{\cup}$, by 
\begin{equation}
f^{\cup}\left(x\right)=\inf_{\sum_{k}p_{k}x_{k}=x}\sum_{k}p_{k}f\left(x_{k}\right)\,,\label{eq:convex roof extension0}
\end{equation}

where the infimum is taken over all possible convex decompositions
of $x$ onto vectors from the set of extremal points $\mathcal{E}$.
Let $\mathcal{E}_{0}\subset\mathcal{E}$ be some compact subset of
the set of extremal points and let $\mathcal{E}_{0}^{c}\subset\mathcal{C}$
be its convex hull. If $\left.f\right|_{\mathcal{E}_{0}}=c$ and $\left.f\right|_{\mathcal{E}\setminus\mathcal{E}_{0}}>c$,
then $f^{\cup}\left(x\right)=c$ if and only if $x\in\mathcal{E}_{0}^{c}$.
Therefore $f^{\cup}$ can serve as an identifier of the set $\mathcal{E}_{0}^{c}$.
In the context of separability problem $\mathcal{D}\left(\mathbb{C}^{d_{1}}\otimes\mathbb{C}^{d_{2}}\right)$
correspond to $\mathcal{E}^{c}$ whereas pure states $\mathcal{D}_{1}\left(\mathbb{C}^{d_{1}}\otimes\mathbb{C}^{d_{2}}\right)$
correspond to $\mathcal{E}$. Separable pure states $\mathcal{M}_{d}$
play the role of $\mathcal{E}_{0}$ and therefore the set of separable
states $\mathcal{S}$ can be identified with $\mathcal{E}_{0}^{c}$.
Assume now that we have a function $g:\mathcal{D}_{1}\left(\mathbb{C}^{d_{1}}\otimes\mathbb{C}^{d_{2}}\right)\rightarrow\mathbb{R}$
that $g\left(\kb{\psi}{\psi}\right)\geq0$ and $g\left(\kb{\psi}{\psi}\right)=0$
if and only if $\kb{\psi}{\psi}\in\mathcal{M}$. Then the convex roof
extension $g^{\cup}$ can be used as an indicator of entanglement.
Functions satisfying the desired properties are known. An example
is entanglement of formation \citep{Bennett1996} $E_{f}$ given by%
\footnote{In Chapter \ref{chap:Multilinear-criteria-for-pure-states} we give
a number of such functions for the problem of entanglement and its
generalizations. %
} 
\begin{equation}
E_{F}\left(\kb{\psi}{\psi}\right)=H\left(\mathrm{tr}_{1}\left(\kb{\psi}{\psi}\right)\right)\,.\label{eq:entanglement of formation}
\end{equation}
To sum up, the problem of characterization of $\mathcal{S}$ is equivalent
to computation of a convex roof 
\begin{equation}
g^{\cup}\left(\rho\right)=\inf_{\sum_{k}p_{k}\kb{\psi_{k}}{\psi_{k}}=\rho}\sum_{k}p_{k}g\left(\kb{\psi_{k}}{\psi_{k}}\right)\label{eq:convex roof entanglement}
\end{equation}

for the suitably chosen function $g$. Unfortunately the explicit
computation of $g^{\cup}$ is possible only for special functions
in low dimensional cases (see \citep{Uhlmann2010} for a review of
the concept of convex roof extensions). The physically relevant examples
of such situations include entanglement of formation and concurrence
for two qbits. The concurrence \citep{Wootters2001} $C:\mathcal{D}_{1}\left(\mathbb{C}^{2}\otimes\mathbb{C}^{2}\right)\rightarrow\mathbb{R}$
is defined by the following formula:
\[
C\left(\kb{\psi}{\psi}\right)=\left|\bra{\psi}\sigma_{y}\otimes\sigma_{y}\ket{\psi^{\ast}}\right|\,,
\]
where $\sigma_{y}=\begin{pmatrix}0 & i\\
-i & 0
\end{pmatrix}$ and $\ket{\psi^{\ast}}$is the complex conjugate of a vector $\ket{\psi}$
written in a standard basis of $\mathbb{C}^{2}\otimes\mathbb{C}^{2}$,
$\left\{ \ket 0\ket 0,\ket 0\ket 1,\ket 1\ket 0,\ket 1\ket 1\right\} $.
The convex roof extension of $C$ (we dropped the superscript for
simplicity) is given \citep{Wootters2001} by the formula%
\footnote{The convex roof extension of $E_{F}$ can be computed explicitly from
Eq.\eqref{eq:wooters concurrence} (see \citep{Wootters2001} for
details). %
} 
\begin{equation}
C\left(\rho\right)=\max\left\{ 0,\lambda_{1}-\lambda_{2}-\lambda_{3}-\lambda_{4}\right\} \,,\label{eq:wooters concurrence}
\end{equation}
where $\lambda_{i}$ are eigenvalues of the Hermitian operator $R=\sqrt{\rho\tilde{\rho}}$,
where 
\[
\tilde{\rho}=\sigma_{y}\otimes\sigma_{y}\rho^{\ast}\sigma_{y}\otimes\sigma_{y}\,,
\]
and $\rho^{\ast}$ denotes the complex conjugation of the matrix $\rho$
written in the standard basis of $\mathbb{C}^{2}\otimes\mathbb{C}^{2}$. 

Since the problem of computing of the convex roof is difficult one
can also hope to gain some information about entanglement by deriving
lower bounds of the convex roof extensions of the relevant functions.
This approach was initiated in \citep{Mintert2004} and where authors
derived efficiently computable lower bounds for the convex roof of
the bipartite concurrence. It was latter developed further in \citep{Mintert2005,Aolita2006,Mintert2007}.
For the extension to the scenario of generalized entanglement (c.f.
Section \ref{sec:General-motivation}) see \citep{Kotowski2010}.

Let us finish our considerations with the brief explanation of the
method of \textit{symmetric extensions }\citep{Doherty2004}. Let
us label the Hilbert spaces associated to subsystems of the bipartite
system by $\mathcal{H}_{A}$ and $\mathcal{H}_{B}$ respectively.
The idea of the method of symmetric extensions comes from the following
observation. If a state $\rho\in\mathcal{D}\left(\mathcal{H}_{A}\otimes\mathcal{H}_{B}\right)$
is separable then it admits a symmetric extension to $\mathcal{H}_{A}\otimes\left(\mathcal{H}_{B}^{\left(1\right)}\otimes\ldots\otimes\mathcal{H}_{B}^{\left(k\right)}\right)$
for any natural $k$. In other words for each $k$ there exist a state
$\rho^{\left(n\right)}\in\mathcal{D}\mathcal{H}_{A}\otimes\left(\mathcal{H}_{B}^{\left(1\right)}\otimes\ldots\otimes\mathcal{H}_{B}^{\left(k\right)}\right)$
such that
\begin{enumerate}
\item The state $\rho^{\left(k\right)}$ is symmetric under exchange of
subsystems $B_{1},\ldots,B_{k}$.
\item $\mathrm{tr}_{B_{2},\ldots,B_{k}}\left(\rho^{\left(k\right)}\right)=\rho$.
\end{enumerate}
In \citep{Doherty2004} it was proven that checking whether a state
$\rho$ posses a symmetric extension to $\mathcal{H}_{A}\otimes\left(\mathcal{H}_{B}^{\left(1\right)}\otimes\ldots\otimes\mathcal{H}_{B}^{\left(k\right)}\right)$
can be directly formulated as a semidefinite programme. Semidefinite
programs are optimization problems that can not only be solved efficiently,
but under weak conditions the global optimality of the found solution,
can also be proven to be optimal. Consequently, the problem of finding
$\rho^{\left(k\right)}$ can be directly tackled with standard numerical
packages, and if no extension is found, the algorithm can also prove
that no extension exists, consequently the state must be entangled.
Due to the symmetry requirements, the number of parameters in the
semidefinite program increases only polynomially in the number of
extensions. The method of symmetric extensions delivers a hierarchy
of separability criteria, as the existence of symmetric extension
to $k+1$ parties automatically ensure the existence of the symmetric
extension to $k$ parties. Authors of \citep{Doherty2004} used the
quantum de-Finetti theorem to show that the hierarchy is complete:
any entangled state is detected in some step of the hierarchy. The
method from \citep{Doherty2004} has been latter on developed in a
series of papers \citep{Navascues2009,Navascues2009a,Brandao2011,Brandao2012}.

\section{Differential geometry, Lie groups and Lie algebras\label{sec:Differential-geometry,-Lie}}

Lie groups and Lie algebras have been accompanying quantum mechanics
ever since the beginning of this theory, starting from the pioneering
work of Wigner \citep{Wigner}. The role of these objects is prominent
as, just like in the classical setting, they describe symmetries and
conservation laws in quantum systems. Although Lie groups and Lie
algebras appear naturally in quantum mechanics, the proper language
of this theory seems to be that of functional analysis and linear
algebra. It turns out however, that many concepts and structures appearing
in quantum theory have a natural geometrical interpretation. Application
of methods of differential geometry and topology in quantum mechanics
has led not only to its elegant formulation but also enabled the discovery
of many important physical phenomena (among others: celebrated Aharonov-Bohm
effect, Berry phase and the existence of anyons \citep{GeomPhase}).
The aim of this part is to give basic definitions and facts form differential
geometry and the theory of Lie groups and algebras that will be applied
in other parts of the thesis. In the course of the presentation of
general structures we will also give examples of concrete objects
that will be used latter. For a survey of application of geometric
and topological methods in quantum mechanics see excellent textbooks
\citep{GeometryQuantum} and \citep{DiffGeomPhys}. For a neat introduction
to the theory of Lie groups and Lie algebras see \citep{HallGroups}.
The modern account to application of the theory of Lie groups in physics
can be found in \citep{BatutRaczka,HallQM2013}.

\subsection{Differential geometry\label{sub:Differential-geometry}}

In this section we outline concepts and facts from differential geometry
that will be used in other parts of the thesis. We do not give any
proofs nor reasoning leading to the presented results. Interested
reader is welcomed to consult the relevant specialized literature
\citep{DiffGeomPhys,KobayashiNomizu}.

\subsubsection*{Basic definitions and examples}
\begin{defn}
A differential manifold $\mathcal{M}$ is a topological space which
is provided with the family of pairs $\left\{ \left(U_{\alpha},\phi_{\alpha}\right)\right\} _{\alpha\in\mathcal{A}}$,
where \end{defn}
\begin{itemize}
\item Family $\left\{ U_{\alpha}\right\} _{\alpha\in\mathcal{A}}$ is the
open covering of $\mathcal{N}$: each $U_{\alpha}\subset\mathcal{M}$
is open and $\cup_{\alpha\in\mathcal{A}}U_{\alpha}=\mathcal{M}$;
\item Mappings $\phi_{\alpha}:U_{\alpha}\rightarrow\mathbb{R}^{N}$ are
homeomorphisms onto open subsets $\phi_{\alpha}\left(U_{\alpha}\right)\subset\mathbb{R}^{N}$;
\item Mappings from the family $\left\{ \phi_{\alpha}\right\} _{\alpha\in\mathcal{A}}$
satisfy compatibility conditions: whenever $U_{\alpha}\cap U_{\beta}\neq\emptyset$,
a function 
\[
\left.\mbox{\ensuremath{\phi}}_{\alpha}\circ\phi_{\beta}^{-1}\right|_{\phi_{\beta}\left(U_{\alpha}\cap U_{\beta}\right)}:\phi_{\beta}\left(U_{\alpha}\cap U_{\beta}\right)\rightarrow\phi_{\alpha}\left(U_{\alpha}\cap U_{\beta}\right)\,,
\]

\end{itemize}
is infinitely differentiable.

\noindent Intuitively speaking above conditions mean that $\mathcal{M}$
``locally looks'' like the open subset of $\mathbb{R}^{N}$ (the
number $N$ is called the dimension of the manifold $\mathcal{M}$).
A pair $\left(U_{\alpha},\phi_{\alpha}\right)$ is called a chart.

Let us list basic examples of differential manifolds. Some of them
will be used in latter chapters of the thesis.
\begin{itemize}
\item \textbf{Real coordinate space} $\mathbb{R}^{k}$. The space $\mathbb{R}^{k}$
is a manifold of dimension $N=k$. 
\item \textbf{Unit sphere} $\mathbb{S}_{k}$ in $k+1$ dimensions. The sphere
$\mathbb{S}_{k}$ is a manifold of dimension $N=k$. It is defined
by
\[
\mathbb{S}_{k}=\left\{ \left(x_{1},\ldots,x_{k+1}\right)\in\mathbb{R}^{k+1}|\,\sum_{i=1}^{i=k}x_{i}^{2}=1\right\} \,.
\]

\item \textbf{Set of isospectral density matrices} in finite dimensional
Hilbert space $\mathcal{H}$ of dimension $D$. The set of isospectral
density matrices is defined by
\[
\Omega_{\left\{ p_{1},\ldots,p_{D}\right\} }=\left\{ \rho\in\mathcal{D}\left(\mathcal{H}\right)\left|\mathrm{Sp}\left(\rho\right)=\left(p_{1},\ldots,p_{D}\right)\right.\right\} ,
\]
where $\mathrm{Sp}\left(\rho\right)$ denotes the ordered spectrum
of $\rho$, i.e. $p_{1}\geq p_{2}\geq\ldots\geq p_{D}\geq0$. For
$\mathrm{sp}\left(\rho\right)=\left(1,0,\ldots,0\right)$ we $\Omega_{\left\{ p_{1},\ldots,p_{D}\right\} }=\mathcal{D}_{1}\left(\mathcal{H}\right)$.
This family of manifolds will be extensively used throughout the thesis
(see Chapters \ref{chap:Multilinear-criteria-for-pure-states}-\ref{chap:Typical-properties-of}).
The manifold $\Omega_{\left\{ p_{1},\ldots,p_{D}\right\} }$ will
be discussed in Subsection \ref{sub:Lie-groups-and} in the context
of the action of a Lie group $\mathrm{SU}\left(N\right)$ on the set
of density matrices $\mathcal{D}\left(\mathcal{H}\right)$.
\end{itemize}

\subsubsection*{Smooth maps, tangent and cotangent spaces }
\begin{defn}
A function $f$ on $\mathcal{M}$ is called a smooth map from $\mathcal{M}$
to $\mathbb{R}$ if and only if for all coordinate functions $\phi_{\alpha}$
a function $f\circ\phi_{\alpha}^{-1}$ is a smooth real-valued function
on $\phi_{\alpha}\left(U_{\alpha}\right)$. The set of smooth functions
on $\mathcal{M}$ is denoted by $\mathcal{F}\left(\mathcal{M}\right)$.
\end{defn}
It is also important to introduce smooth mappings (functions) between
differential manifolds. Definition of smooth mapping between two manifolds
$\mathcal{M}$ and $\mathcal{N}$ is analogous to the definition of
a smooth functions on $\mathcal{M}$.
\begin{defn}
Let $\mathcal{M}$ and $\mathcal{N}$ be manifolds having atlases
$\left\{ \left(U_{\alpha},\phi_{\alpha}\right)\right\} _{\alpha\in\mathcal{A}}$
and $\left\{ \left(V_{\beta},\psi_{\beta}\right)\right\} _{\beta\in\mathcal{B}}$.
Assume that $\mathrm{dim}\mathcal{M}=N_{1}$ and $\mathrm{dim}\mathcal{N}=N_{2}$.
A mapping $f:\mathcal{M}\rightarrow\mathcal{N}$ is called smooth
if and only if mappings $\psi_{\beta}\circ f\circ\phi_{\alpha}^{-1}:\phi_{\alpha}\left(U_{\alpha}\right)\rightarrow\psi_{\beta}\left(V_{\beta}\right)$
are smooth (as mappings between open subsets of $\mathbb{R}^{N_{1}}$
and $\mathbb{R}^{N_{2}}$ respectively) whenever $f(U_{\alpha})\cap V_{\beta}\neq\emptyset$.
The set of smooth functions between $\mathcal{M}$ and $\mathcal{N}$
is denoted by $\mathcal{F}\left(\mathcal{M},\,\mathcal{N}\right)$.
\end{defn}
A special class of smooth functions consists of curves on $\mathcal{M}$
i.e smooth mappings $\gamma:U\rightarrow\mathcal{M}$, where $U$
is the open subset of $\mathbb{R}$. Definition of smooth functions
on the differentiable manifold is indispensable for introducing concepts
of a tangent space and a cotangent space.
\begin{defn}
A tangent space to a manifold $\mathcal{M}$ at the point $p\in\mathcal{M}$,
denoted by $T_{p}\mathcal{M}$, is a vector space consisting of linear
mappings $V:\mathcal{F}\left(\mathcal{M}\right)\rightarrow\mathbb{R}$
satisfying the Leibniz property: $V\left(f\cdot g\right)=f(p)X\left(g\right)+V\left(f\right)g\left(p\right),$
where $f,\, g\in\mathcal{F}\left(\mathcal{M}\right)$.
\end{defn}
The space $T_{p}\mathcal{M}$ is a real vector space of finite dimension
equal to $N$, the dimension of the manifold itself. Every element
$V\in T_{p}\mathcal{M}$ can be expressed uniquely as a linear combination
of basis vectors%
\footnote{Each coordinate system $\phi_{\alpha}=\left(x^{1}\left(\cdot\right),\ldots,x^{N}\left(\cdot\right)\right)$
covering the neighborhood of $p$ gives rise to the basis of $T_{p}\mathcal{M}$.
Let $\left(x^{1}(p),\ldots,x^{N}(p)\right)$ denote the coordinates
of the point $p$. The basis of $T_{p}\mathcal{M}$ is given by the
set of operators of ``partial derivative'', $\left\{ \left.\frac{\partial}{\partial x^{i}}\right|_{p}\right\} _{i=1}^{N}$,
defined by
\[
\mbox{\ensuremath{\left.\frac{\partial}{\partial x^{i}}\right|_{p}}}\left(f\right)=\left.\frac{\partial\left(f\circ\phi_{\alpha}^{-1}\right)}{\partial x^{i}}\right|_{\left(x^{1},\ldots,x^{N}\right)=\left(x^{1}(p),\ldots,x^{N}(p)\right)}\,,
\]
where $f\in\mathcal{F}\left(\mathcal{M}\right)$ and $f\circ\phi_{\alpha}^{-1}$
on the right hand side of the above expression is treated as a usual
smooth function on $U_{\alpha}\subset\mathbb{R}^{N}$. In what follows,
unless it causes a confusion, we will drop the subscript $p$ in $\mbox{\ensuremath{\left.\frac{\partial}{\partial x^{i}}\right|_{p}}}$.%
}

\begin{equation}
V=\sum_{i=1}^{N}V^{i}\frac{\partial}{\partial x^{i}}\,,\label{eq:tangent vector expression}
\end{equation}
Let us remark that on the intuitive level tangent vectors can be interpreted
as generalized directional derivatives of smooth function computed
at the point $p\in\mathcal{M}$. 

We now define the cotangent space, the vector space dual to a tangent
space.
\begin{defn}
A cotangent space to a manifold $\mathcal{M}$ at the point $p\in\mathcal{M}$,
$T_{p}^{\ast}\mathcal{M}$ is defined as a vector space dual to $T_{p}\mathcal{M}$,
i.e. $T_{p}^{\ast}\mathcal{M}=\left(T_{p}\mathcal{M}\right)^{\ast}$. 
\end{defn}
Just like the tangent space, the cotangent space $T_{p}^{\ast}\mathcal{M}$
is a real vector space of finite dimension $N$. In the coordinate
system $\phi_{\alpha}=\left(x^{1}\left(\cdot\right),\ldots,x^{N}\left(\cdot\right)\right)$
defined in the neighborhood of $p\in\mathcal{M}$ the basis of $T_{p}^{\ast}\mathcal{M}$
is given by $\left\{ \left.dx^{i}\right|_{p}\right\} _{i=1}^{N}$,
defined by conditions
\[
\left\langle \left.dx^{i}\right|_{p},\,\ensuremath{\left.\frac{\partial}{\partial x^{j}}\right|_{p}}\right\rangle =\delta_{ij}\,,
\]
where $\left\langle \omega,\, V\right\rangle $ denotes the pairing
of $\omega\in T_{p}^{\ast}\mathcal{M}$ with $V\in T_{p}\mathcal{M}$
and $\delta_{ij}$ denotes the Kronecker delta. An alternative interpretation
of vectors $\left.dx^{i}\right|_{p}$ comes from the fact that in
the coordinate system $\phi_{\alpha}$ they can be identified with
the usual derivatives of functions $x^{i}\left(\cdot\right)$ computed
at $\left(x^{1},\ldots,x^{N}\right)=\left(x^{1}(p),\ldots,x^{N}(p)\right)$.
In what follows, unless it leads a confusion, we will drop the subscript
$p$ in $\mbox{\ensuremath{\left.dx^{i}\right|_{p}}}$. Every element
$\omega\in T_{p}^{\ast}\mathcal{M}$ can be expressed uniquely as
a linear combination of basis vectors,
\begin{equation}
\omega=\sum_{i=1}^{N}\omega_{i}dx^{i}\,.\label{eq:one form expression}
\end{equation}
Let us remind that due to the paring $T_{p}^{\ast}\mathcal{M}\times T_{p}\mathcal{M}\ni\left(\omega,\, V\right)\rightarrow\left\langle \omega,\, V\right\rangle \in\mathbb{R}$,
we can interpret $T_{p}\mathcal{M}$ as a space dual to $T_{p}^{\ast}\mathcal{M}$,
$T_{p}\mathcal{M}\approx\left(T_{p}^{\ast}\mathcal{M}\right)^{\ast}$.

Assume now that we have a smooth mapping between two manifolds $f:\mathcal{M}\rightarrow\mathcal{N}$.
For every $p\in\mathcal{M}$ we have the induced linear mapping $f^{\ast}:T_{p}\mathcal{M}\rightarrow T_{f(p)}\mathcal{\mathcal{N}}$
defined by, 
\[
f^{\ast}V\left(F\right)=V\left(F\circ f\right)\,,
\]
for $F\in\mathcal{F}\left(\mathcal{N}\right)$ and $V\in T_{p}\mathcal{M}$.
Analogously, we have a linear map being a transpose of $f^{\ast}$,
$f_{\ast}:T_{f(p)}^{\ast}\mathcal{\mathcal{N}}\rightarrow T_{p}^{\ast}\mathcal{M}$.
It is defined via the condition
\[
\left\langle f_{\ast}\omega,\, V\right\rangle =\left\langle \omega,\, f^{\ast}V\right\rangle ,
\]
 for $\omega\in T_{f(p)}^{\ast}\mathcal{\mathcal{N}}$ and $V\in T_{p}\mathcal{\mathcal{M}}$.
Let and $f:\mathcal{M}\rightarrow\mathcal{N}$ and $g:\mathcal{N}\rightarrow\mathcal{K}$
be smooth mappings between differential manifolds. The following identities
hold:
\[
\left(g\circ f\right)^{\ast}=g^{\ast}\circ f^{\ast},\,\left(g\circ f\right)_{\ast}=f_{\ast}\circ g_{\ast}\,,
\]
where the domains of the involved mappings are implicit.

\subsubsection*{Vector fields and one-forms}
\begin{defn}
A tangent bundle to the manifold $\mathcal{M}$, $T\mathcal{M}$,
is defined by
\[
T\mathcal{M}=\bigcup_{p\in\mathcal{M}}\left\{ \left(p,\, v\right)|\, v\in T_{p}\mathcal{M}\right\} \,.
\]
 On $T\mathcal{M}$ we have a natural projection $\pi:T\mathcal{M}\rightarrow\mathcal{M}$
given by $\pi\left(\left[p,\, v\right]\right)=p$. A tangent bundle
$T\mathcal{M}$ can be equipped with a differential structure that
makes from it a $2N$ dimensional manifold and for which the projection
$\pi:T\mathcal{M}\rightarrow\mathcal{M}$ is a smooth mapping. 
\end{defn}
The concept of a tangent bundle can be used to define, among other
things, vector fields on a manifold $\mathcal{M}$.
\begin{defn}
A vector field $X$ on a manifold $\mathcal{M}$ is a smooth mapping
$X:\mathcal{M}\rightarrow T\mathcal{M}$ which satisfies the condition
$\pi\left(X\left(p\right)\right)=p$ for each $p\in\mathcal{M}$.
The collection of all vector fields on $\mathcal{M}$ is denoted by
$\mathcal{X}\left(\mathcal{M}\right)$. 
\end{defn}
Informally one can think of a vector field $X$ as a smooth mapping
that assigns to any point $p\in\mathcal{M}$ exactly one vector $X_{p}$
in the corresponding tangent space $T_{p}\mathcal{M}$. In what follows
we will identify, for the sake of simplicity, a vector $X_{p}\in T_{p}\mathcal{M}$
with the element of the associated element of the tangent bundle $X\left(p\right)=\left(p,\, X_{p}\right)$.
The set of vector fields $\mathcal{X}\left(\mathcal{M}\right)$ posses
a natural structure of a vector space over real numbers. Moreover,
for any $f\in\mathcal{F}\left(\mathcal{M}\right)$ and $X\in\mathcal{X}\left(\mathcal{M}\right)$
the object $f\cdot X$, defined by $\left(f\cdot X\right)\left(g\right)=fX\left(g\right)$
(for $g\in\mathcal{F}\left(\mathcal{M}\right)$), is again a vector
field%
\footnote{Actually, this action of $\mathcal{F}\left(\mathcal{M}\right)$ on
gives $\mathcal{X}\left(\mathcal{M}\right)$ the structure of a ring
over $\mathcal{F}\left(\mathcal{M}\right)$.%
} The set of vector fields $\mathcal{X}\left(\mathcal{M}\right)$ is
equipped with the \textit{Lie bracket} of vector fields, $\left[\cdot,\,\cdot\right]:\mathcal{X}\left(\mathcal{M}\right)\times\mathcal{X}\left(\mathcal{M}\right)\rightarrow\mathcal{X}\left(\mathcal{M}\right)$,
defined by,
\[
\left[X,\, Y\right]\left(f\right)=X\left(Y\left(f\right)\right)-Y\left(X\left(f\right)\right)\,,
\]
for $X,\, Y\in\mathcal{X}\left(\mathcal{M}\right)$, $F\in\mathcal{F}\left(\mathcal{M}\right)$.
Lie bracket is bilinear, antisymmetric and satisfies the Jacobi identity
\eqref{eq:jacobi identity}.
\begin{defn}
A cotangent bundle to the manifold $\mathcal{M}$, $T\mathcal{^{\ast}M}$,
is defined by
\[
T^{\ast}\mathcal{M}=\bigcup_{p\in\mathcal{M}}\left\{ \left(p,\,\omega\right)|\,\omega\in T_{p}^{\ast}\mathcal{M}\right\} \,.
\]
Just like in the case of tangent bundle on cotangent bundle we have
the projection $\pi:T^{\ast}\mathcal{M}\rightarrow\mathcal{M}$ given
by $\pi\left(\left[p,\,\omega\right]\right)=p$ (in order not to complicate
the notation we will consequently use the same symbol for projections
defined on $T\mathcal{M}$ and $T^{\ast}\mathcal{M}$ respectively).
Analogously, cotangent bundle $T^{\ast}\mathcal{M}$ can be given
a differential structure that makes from it a $2N$ dimensional manifold
and for which the projection $\pi:T^{\ast}\mathcal{M}\rightarrow\mathcal{M}$
is a smooth mapping. 
\end{defn}
Analogously to the concept of vector fields differential one-forms
are defined as sections of the cotangent bundle. 
\begin{defn}
A one-form $\omega$ on a manifold $\mathcal{M}$ is a smooth mapping
$\omega:\mathcal{M}\rightarrow T^{\ast}\mathcal{M}$ which satisfies
the condition $\pi\left(\omega\left(p\right)\right)=p$ for each $p\in\mathcal{M}$.
The collection of all 1-forms on $\mathcal{M}$ is denoted by $\Omega_{1}\left(\mathcal{M}\right)$. 
\end{defn}
From the informal perspective the one-form $\omega$ is a smooth mapping
that assigns to any point $p\in\mathcal{M}$ exactly one covector
$\omega_{p}$ in the corresponding cotangent space $T_{p}^{\ast}\mathcal{M}$.
In what follows we identify, for the sake of simplicity, a covector
$\omega_{p}\in T_{p}^{\ast}\mathcal{M}$ with the associated element
of the cotangent bundle $\omega\left(p\right)=\left(p,\,\omega{}_{p}\right)$.
Every smooth function $F\in\mathcal{F}\left(\mathcal{M}\right)$ defines,
in a natural manner, one-form $dF\in\Omega_{1}\left(\mathcal{M}\right)$,
given by
\begin{equation}
\left\langle dF_{p},\, V\right\rangle =V\left(F\right)\,,\label{eq:derrivative}
\end{equation}
for $V\in T_{p}\mathcal{M}$. One form $dF$ can be also interpreted
as a derivative of the function $F$.

\subsubsection*{General tensor fields}

Notions of the tangent vectors and covectors can be extended to tensors
of arbitrary type. By taking multiple tensor powers of $T_{p}\mathcal{M}$
and $T_{p}^{\ast}\mathcal{M}$ to each point of the manifold $\mathcal{M}$
we attach a vector space $T_{p}^{\left(m,\, n\right)}\mathcal{M}$
consisting of tensors of the type $\left(n,\, m\right)$,
\begin{equation}
T_{n,\, p}^{m}\mathcal{M}=\left(T_{p}\mathcal{M}\right)^{\otimes m}\otimes\left(T_{p}^{\ast}\mathcal{M}\right)^{\otimes n}\,.\label{eq:tensors tangent}
\end{equation}
Any element $\mathrm{T}\in T_{n,\, p}^{m}\mathcal{M}$ can be interpreted
as a multilinear form on $m$ copies of $T_{p}^{\ast}\mathcal{M}$
and $n$ copies of $T_{p}\mathcal{M}$,
\[
\mathrm{T}:\left(T_{p}^{\ast}\mathcal{M}\right)^{\times m}\times\left(T_{p}\mathcal{M}\right)^{\times n}\rightarrow\mathbb{R}\,.
\]
Using the notation from \eqref{eq:tangent vector expression} and
\eqref{eq:one form expression} we express $\mathrm{T}\in T_{n,\, p}^{m}\mathcal{M}$
in the following manner
\begin{equation}
\mathrm{T}=\sum_{i_{1}=1}^{N}\ldots\sum_{i_{m}=1}^{N}\sum_{j_{1}=1}^{N}\ldots\sum_{j_{n}=1}^{N}\mathrm{T}_{j_{1}\ldots j_{n}}^{i_{1}\ldots i_{m}}\frac{\partial}{\partial x^{i_{1}}}\otimes\ldots\otimes\frac{\partial}{\partial x^{i_{m}}}\otimes dx^{j_{1}}\otimes\ldots\otimes dx^{j_{n}}\,,\label{eq:tensor expression}
\end{equation}
where vectors of the form $\frac{\partial}{\partial x^{i_{1}}}\otimes\ldots\otimes\frac{\partial}{\partial x^{i_{m}}}\otimes dx^{j_{1}}\otimes\ldots\otimes dx^{j_{n}}$
form the basis of $\left(T_{p}\mathcal{M}\right)^{\otimes m}\otimes\left(T_{p}^{\ast}\mathcal{M}\right)^{\otimes n}$.
Analogously to cases of the tangent and the cotangent bundle it is
possible to define a general tensor bundle $\mathcal{T}_{n,\, p}^{m}\mathcal{M}$
which is a $N^{m+n+1}$ manifold. As before, we have a natural projection
$\pi:\mathcal{T}_{n,\, p}^{m}\mathcal{M}\rightarrow\mathcal{M}$ which
is a smooth mapping. We will be particularly interested in sections
of the bundle $\mathcal{T}_{n,\, p}^{m}\mathcal{M}$ (tensor fields
of the type $\left(m,\, n\right)$ on $\mathcal{M}$). The reason
for that comes from the fact that many interesting geometrical structures
on the manifold $\mathcal{M}$ can be interpreted as tensor fields
on $\mathcal{M}$.
\begin{defn}
A tensor field of the type $\left(m,\, n\right)$ on a manifold $\mathcal{M}$
is a smooth mapping $\mathrm{T}:\mathcal{M}\rightarrow\mathcal{T}_{n,\, p}^{m}\mathcal{M}$
which satisfies the condition $\pi\left(\mathrm{T}\left(p\right)\right)=p$
for each $p\in\mathcal{M}$. The collection of all tensor fields of
the type $\left(m,\, n\right)$ is denoted by $\mathcal{T}_{n}^{m}\mathcal{\left(M\right)}$. 
\end{defn}

\subsection{Elements of Riemannian geometry\label{sub:Elements-of-Riemannian}}

The Riemannian structure is a way to introduce the natural distance
on a manifold. It also allows to transfer to the realm of manifolds
many concepts known from the euclidean geometry (such as the parallel
transport or the volume form). We will touch the subject of Riemannian
geometry only superficially, limiting ourselves only to listing the
structures that will be used in the thesis.
\begin{defn}
A Riemannian metric $g$ on a manifold $\mathcal{M}$ is a tensor
field of type $\left(0,\,2\right)$, $g\in\mathcal{T}_{2}^{0}\left(\mathcal{M}\right)$,
such that for every $p\in\mathcal{M}$ the linear map $g_{p}:\left(T_{p}\mathcal{M}\right)^{\otimes2}\rightarrow\mathbb{R}$
is positive definite and symmetric:
\[
g_{p}\left(V,\, W\right)=g_{p}\left(W,\, V\right)\,\text{for each }V,W\in T_{p}\mathcal{M}\,,
\]
\[
g_{p}\left(V,\, V\right)\geq0\,\text{and }g_{p}\left(V,V\right)=0\,\text{if and only if }V=0,\,\text{for }V\in T_{p}\mathcal{M}\,.
\]
Intuitively, a Riemannian metric $g$ attach to each tangent space
$T_{p}\mathcal{M}$ an euclidean inner product $g_{p}$. The length
of $V\in T_{p}\mathcal{M}$ is defined via the expression $\sqrt{g_{p}\left(V,\, V\right)}$. 
\end{defn}
Given a smooth curve $\gamma:\left[0,\,1\right]\rightarrow\mathcal{M}$
we can define its length by
\[
L_{\gamma}=\int_{\left[0,\,1\right]}\sqrt{g\left(\frac{d\gamma}{dt},\,\frac{d\gamma}{dt}\right)}dt\,.
\]
Given a connected manifold $\mathcal{M}$ one defines the distance
between two points $p,q\in\mathcal{M}$ as the infimum over the curves
$\gamma$ that start at the point $p$ and end at the point $q$,
\[
\mathrm{d}\left(p,\, q\right)=\inf_{\gamma:\,\gamma(0)=p,\,\gamma(1)=q}\, L_{\gamma}\,.
\]
Existence of a metric $g$ allows to introduce a measure $\mu$ on
$\mathcal{M}$. Let us first express the Riemann metric in coordinates
$\left(x^{1},\ldots,x^{N}\right)$ associated to some chart $\left(U_{\alpha},\phi_{\alpha}\right)$.
The metric $g$ on $U_{\alpha}$ is specified by the matrix-valued
function on $\phi_{\alpha}\left(U_{\alpha}\right)$: $g_{ij}=g\left(\frac{\partial}{\partial x^{i}},\,\frac{\partial}{\partial x^{j}}\right)$.
On every chart $\left(U_{\alpha},\phi_{\alpha}\right)$ the measure
$\mu$ is defined by specifying the value of integrals on smooth functions
$f$ with the support%
\footnote{The support of the function $f$ is defined by $\left\{ p\in\mathcal{M}|\, f\left(p\right)\neq0\right\} ^{cl}$,
where the superscript $cl$ denotes the closure.%
} in $U_{\alpha}$,
\begin{equation}
\int_{\mathcal{M}}fd\mu=\int_{\phi_{\alpha}\left(U_{\alpha}\right)}f\left(\phi_{\alpha}^{-1}\left(x^{1},\ldots,x^{N}\right)\right)\sqrt{\mathrm{det}\left(g_{ij}\right)}dx^{1}\ldots dx^{N}\,.\label{eq:integral coord}
\end{equation}
The integral on the right hand side of \eqref{eq:integral coord}
is the usual Lebesgue integral%
\footnote{The intuition behind this definition of a measure $\mu$ comes from
the fact that applied to the ``infinitesimal cube'' it gives the
right result, $\mu\left(\phi_{\alpha}^{-1}\left(\times{}_{i=1}^{N}\left[x_{0}^{i},\, x_{0}^{i}+\Delta x^{i}\right]\right)\right)\approx\mathrm{det}\left(g_{ij}\right)\prod_{i=1}^{N}\Delta x^{i}$.%
} over the open subset $\phi_{\alpha}\left(U_{\alpha}\right)\subset\mathbb{R}^{N}$.
it Integral for the general function $f\in\mathcal{F}\left(\mathcal{M}\right)$
can by defined by ``gluing together'' integrals of the form \eqref{eq:integral coord}
for different charts. If the manifold $\mathcal{M}$ is compact then
its measure is finite,
\[
\mu\left(\mathcal{M}\right)=\int_{\mathcal{M}}d\mu=C<\infty\,.
\]
When the metric tensor $g$ is rescaled by some $\alpha>0$, the measure
$\mu$ changes accordingly,
\begin{equation}
g\rightarrow g'=\alpha\cdot g\,\Longrightarrow\,\mu\rightarrow\mu'=\alpha^{\frac{N}{2}}\mu\,.\label{eq:rescalling measure}
\end{equation}
Consequently, by the appropriate rescaling of the metric $g$ the
measure $\mu$ on a compact manifold $\mathcal{M}$ can be chosen
to be probabilistic i.e $\mu\left(\mathcal{M}\right)=1$.

Having a metric $g$ and a function $F\in\mathcal{F}\left(\mathcal{M}\right)$,
we define the gradient of $F$ as the unique vector field $\nabla F\in\mathcal{X}\left(\mathcal{M}\right)$
satisfying the equation
\begin{equation}
dF=g\left(\nabla F,\cdot\right)\,,\label{eq:gradient definition}
\end{equation}
where the above expression is understood as the equality of two one-forms.

\subsection{Lie groups and Lie algebras \label{sub:Lie-groups-and}}

\subsubsection*{Definitions and examples}
\begin{defn}
A Lie group $G$ is a group that has a structure of a differentiable
manifold such that the group operations,
\[
\cdot:G\times G\rightarrow G,\,\left(g_{1},\, g_{2}\right)\rightarrow g_{1}\cdot g_{2}\,,
\]
\[
\left(\cdot\right)^{-1}:G\rightarrow G,\, g\rightarrow g^{-1}\,,
\]

\end{defn}
\noindent are smooth mappings.

Throughout the thesis we will adopt the multiplicative notation for
the group law. By $e$ we will denote the neutral element of the group.
Usually, unless it causes a confusion, we will drop the multiplication
symbol~$\cdot$. Let us introduce left and right actions of a Lie
group on itself. For any $g\in G$ its left and right action on elements
of $G$ are defined by, 
\begin{gather}
L_{g}:G\rightarrow G\,,\, x\rightarrow gx\,,\nonumber \\
R_{g}:G\rightarrow G\,,\, x\rightarrow xg\,.\label{eq:left right action}
\end{gather}
The above smooth mappings satisfy
\[
L_{g_{1}}\circ L_{g_{2}}=L_{g_{1}g_{2}},\, R_{g_{1}}\circ R_{g_{2}}=R_{g_{2}g_{1}}\,,L_{g_{1}}\circ R_{g_{2}}=R_{g_{2}}\circ L_{g_{1}}\,,
\]
For all $g_{1},\, g_{2}\in G$. 

The basic example of a Lie group is a general linear group $GL\left(N,\,\mathbb{R}\right)$.
It is defined as a set of invertible real $N\times N$ matrices with
the natural group law coming from matrix multiplication. Before proceeding
to the definition of a Lie algebra of a Lie group let us present first
the strictly algebraic definition of a Lie algebra.
\begin{defn}
A real finite-dimensional Lie algebra is a finite-dimensional vector
space $\mathfrak{g}$ over $\mathbb{R}$ equipped with antisymmetric
bilinear operation (called the \textit{Lie bracket}) $\left[\cdot,\,\cdot\right]_{\mathfrak{g}}:\mathfrak{g}\times\mathfrak{g}\rightarrow\mathfrak{g}$
satisfying the \textit{Jacobi identity,}
\begin{equation}
\left[X,\,\left[Y,\, Z\right]_{\mathfrak{g}}\right]_{\mathfrak{g}}+\left[Z,\,\left[X,\, Y\right]_{\mathfrak{g}}\right]_{\mathfrak{g}}+\left[Y,\,\left[Z,\, X\right]_{\mathfrak{g}}\right]_{\mathfrak{g}}=0\,,\label{eq:jacobi identity}
\end{equation}
for all $X,\, Y,\, Z\in\mathfrak{g}$. 
\end{defn}
Lie algebras can be also defined over the field of complex numbers.
Also, Lie algebras can be infinite-dimensional. The definitions are
analogous to the one presented above. Although Lie algebras are interesting
objects to study on their own, in this thesis we will be particularly
interested in Lie algebras associated with Lie groups. 
\begin{defn}
The Lie algebra of a Lie group $G$ is a real vector space $\mathfrak{g}$
consisting of the left invariant vector fields on $G$, i.e. the vector
fields satisfying $L_{g}^{\ast}X=X$ for all $g\in G$. The Lie bracket
structure on $\mathfrak{g}$ is given by the commutator of vector
fields,
\[
\left[X\,,Y\right]_{\mathfrak{g}}=\left[X,\, Y\right]\,\text{ (in a sense of vector fields on \ensuremath{G})\,,}
\]
for $X,\, Y\in\mathfrak{g}$.
\end{defn}
Actually, every finite-dimensional real Lie algebra $\mathfrak{g}$
is a Lie algebra of some Lie group $G$ \citep{FultonHarris}. For
this reason in what follows, unless it causes a confusion, we will
drop the subscript $\mathfrak{g}$ when referring to a Lie bracket
of a Lie algebra. Because of the condition of the left invariance
and the transitivity of the action of $G$ on itself we have $X_{g}=L_{g}^{\ast}X_{e}$
for $X\in\mathfrak{g}$. As a result every $X\in\mathfrak{g}$ is
uniquely specified by its value at the neutral element $e$. Thus,
we have the linear isomorphism of real vector spaces $\mathfrak{g}\approx T_{e}G$.
In what follows we will consequently use this identification. Throughout
the thesis we will be using only matrix Lie groups which are Lie groups
of particularly simple structure. 
\begin{defn}
A matrix Lie group $G$ is a closed%
\footnote{With respect to the topology on $\mathrm{GL}\left(N,\,\mathbb{R}\right)$
induced form the norm on $\mathrm{End}\left(\mathbb{R}^{N}\right)$. %
} subgroup of the group $\mathrm{GL}\left(N,\,\mathbb{R}\right)$ for
some natural number $N$.
\end{defn}
Because every closed subgroup of a Lie group automatically inherits
the structure of a Lie group \citep{BatutRaczka}, the above definition
makes sense. Every matrix Lie group $G$ is an embedded%
\footnote{A manifold $\mathcal{N}$ is an embedded submanifold of a manifold
$\mathcal{M}$ if there exist bijective differential mapping $F:\mathcal{N}\rightarrow\mathcal{M}$.%
} submanifold of $\mathrm{End}\left(\mathbb{R}^{N}\right)$ for some
natural number $N$. Consequently, the tangent space $T_{e}G$ can
be treated as a subspace of $\mathrm{End}\left(\mathbb{R}^{N}\right)$.
For this reason Lie algebras of matrix Lie groups have a simple description
as vector subspaces of $\mathrm{End}\left(\mathbb{R}^{N}\right)$.
In what follows $\mathrm{exp}:\mathrm{End}\left(\mathbb{R}^{N}\right)\rightarrow\mathrm{GL\left(N,\,\mathbb{R}\right)}$
denotes the exponential of matrices.
\begin{fact}
\label{matrix lie algebra}Let $\mathfrak{g}$ be a Lie algebra of
the matrix Lie group $G\subset GL\left(N,\,\mathbb{R}\right)$. Lie
algebra $\mathfrak{g}$ is described by the condition,
\begin{equation}
X\in\mathfrak{g}\,\Longleftrightarrow\,\mathrm{exp}\left(tX_{e}\right)\in G\,,\text{for sufficiently small \ensuremath{t\,.}}\label{eq:condition Lie}
\end{equation}
The value of the vector field $X\in\mathfrak{g}$ at $g\in G$ is
given by $X_{g}=X_{e}g$, where the former expression should be understood
is a sense of matrix multiplication.
\end{fact}
Explicit computation of the Lie bracket of $X,\, Y\in\mathfrak{g}$
gives
\begin{equation}
\left[X,\, Y\right]_{e}=\left[X_{e},\, Y_{e}\right]_{\mathrm{Mat}}\,,\label{eq:matrix bracket}
\end{equation}
where $\left[X_{e},\, Y_{e}\right]_{\mathrm{Mat}}$ is the usual commutator
of operators, $\left[A,\, B\right]_{\mathrm{Mat}}=AB-BA$, for $A,B\in\mathrm{End}\left(\mathbb{R}^{N}\right)$.
In what follows we will be interested in values of $X\in\mathfrak{g}$
only at the neutral element of the group. For this reason, unless
we state otherwise, we will identify $X\in\mathfrak{g}$ with $X_{e}\in T_{e}G\subset\mathrm{End}\left(\mathbb{R}^{N}\right)$.
Moreover, unless it causes a confusion, we will make no distinction
between various types of Lie brackets. We will also use the convention
$e=\mathbb{I}$, where $\mathbb{I}$ is the identity matrix of the
appropriate dimension. Below we give examples of matrix Lie groups
and the corresponding Lie algebras that will be used in the description
of various classes of correlations in quantum systems. In each case
considered below the Lie bracket structure is given by \eqref{eq:matrix bracket}.

\paragraph*{1. General linear group and special linear group }

The general linear group $GL(N,\mathbb{\, R})$ is defined as a group
of invertible real $N\times N$ matrices,
\begin{equation}
\mathrm{GL}(N,\mathbb{\, R})=\left\{ g\in\mathbb{M}_{N\times N}\left(\mathbb{R}\right)|\,\mathrm{det}\left(g\right)\ne0\right\} \,.
\end{equation}
By the virtue of the Fact \ref{matrix lie algebra} its Lie algebra,
$\mathfrak{gl}\left(N,\,\mathbb{R}\right)$ corresponds the space
of all $N\times N$ matrices, 
\begin{equation}
\mathfrak{gl}\left(N,\,\mathbb{R}\right)=\left\{ X\in\mathbb{M}_{N\times N}\left(\mathbb{R}\right)\right\} =\mathrm{End}\left(\mathbb{R}^{N}\right)\,.
\end{equation}
In what follows we will use the notation $\mathfrak{gl}\left(N,\,\mathbb{R}\right)$
and $\mathrm{End}\left(\mathbb{R}^{N}\right)$ interchangeably. The
(real) dimension of $\mathfrak{gl}\left(N,\,\mathbb{R}\right)$ equals
$N^{2}$. The special linear group $SL(N,\,\mathbb{R})$ consists
of invertible matrices having the unit determinant,

\begin{equation}
SL(N,\mathbb{\, R})=\left\{ g\in\mathbb{M}_{N\times N}\left(\mathbb{R}\right)|\,\mathrm{det}\left(g\right)=1\right\} \,.
\end{equation}
The Lie algebra of $SL\left(N,\mathbb{R}\right)$, $\mathfrak{sl}\left(N,\,\mathbb{R}\right),$
consists of $N\times N$ real traceless matrices, 
\begin{equation}
\mathfrak{sl}\left(N,\,\mathbb{R}\right)=\left\{ X\in\mathbb{M}_{N\times N}\left(\mathbb{R}\right)|\,\mathrm{tr}\left(X\right)=0\right\} .
\end{equation}
The (real) dimension of $\mathfrak{sl}\left(N,\,\mathbb{R}\right)$
equals $N^{2}-1$. In what follows we will be considering the complex
analogue of $SL(N,\mathbb{\, R})$ and $\mathfrak{sl}\left(N,\,\mathbb{R}\right)$,
\begin{gather}
SL(N,\mathbb{\, C})=\left\{ g\in\mathbb{M}_{N\times N}\left(\mathbb{C}\right)|\,\mathrm{det}\left(g\right)=1\right\} \,,\nonumber \\
\mathfrak{sl}\left(N,\,\mathbb{C}\right)=\left\{ X\in\mathbb{M}_{N\times N}\left(\mathbb{C}\right)|\,\mathrm{tr}\left(X\right)=0\right\} .\label{eq:special linear def}
\end{gather}
The Lie algebra $\mathfrak{sl}\left(N,\,\mathbb{C}\right)$ can be
interpreted as a Lie algebra over the field of complex or real numbers.
We have $\mathrm{dim}_{\mathbb{C}}\left(\mathfrak{sl}\left(N,\,\mathbb{C}\right)\right)=N^{2}-1,\,\mathrm{dim}_{\mathbb{R}}\left(\mathfrak{sl}\left(N,\,\mathbb{C}\right)\right)=2N^{2}-2$.

\paragraph*{2. Unitary and special unitary group}

The unitary group in dimension $N$ is defined as a set of $N\times N$
unitary matrices,
\begin{equation}
U\left(N\right)=\left\{ g\in\mathbb{M}_{N\times N}\left(\mathbb{C}\right)|\, gg^{\dagger}=\mathbb{I}\right\} \,.
\end{equation}
The Lie algebra of $U\left(N\right)$, $\mathfrak{u}\left(N\right)$,
corresponds to the space of all $N\times N$ anti-Hermitian matrices,
\begin{equation}
\mathfrak{u}\left(N\right)=\left\{ X\in\mathbb{M}_{N\times N}\left(\mathbb{C}\right)|\, X^{\dagger}=-X\right\} \,.
\end{equation}
The algebra $\mathfrak{u}\left(N\right)$ is a real Lie algebra of
dimension $N^{2}$. 

At this point we would like to make a remark concerning the difference
between conventions used by physicists and mathematicians in the definition
of the Lie algebra of the matrix Lie group being a subgroup of the
unitary group. In the physical literature in such cases one usually
defines the Lie algebra $\mathfrak{g}$ via
\begin{equation}
X\in\mathfrak{g}_{\mathrm{Phys}}\,\Longleftrightarrow\,\mathrm{exp}(iX)\in G\,.\label{eq:physical convention Lie}
\end{equation}
The above is in contrast to the condition \eqref{eq:condition Lie}.
As a consequence of \eqref{eq:physical convention Lie} we have $i\left[X,\, Y\right]\in\mathfrak{g}_{\mathrm{Phys}}$,
for $X,\, Y\in\mathfrak{g}_{\mathrm{Phys}}$. Adapting the ``physical''
convention to the case of the unitary group we get that $\mathfrak{u}\left(N\right)_{\mathrm{Phys}}$
consisting of Hermitian matrices. In what follows we will consequently
use the mathematical convention for the definition of a Lie algebra. 

The special unitary group $SU\left(N\right)$ and the corresponding
Lie algebra $\mathfrak{su}\left(N\right)$ are defined analogously
as the special linear group and its Lie algebra,
\begin{gather}
SU(N)=\left\{ g\in\mathbb{M}_{N\times N}\left(\mathbb{C}\right)|\, gg^{\dagger}=\mathbb{I},\,\mathrm{det}\left(g\right)=1\right\} \,,\nonumber \\
\mathfrak{su}\left(N\right)=\left\{ X\in\mathbb{M}_{N\times N}\left(\mathbb{C}\right)|\, X^{\dagger}=-X,\,\mathrm{tr}\left(X\right)=0\right\} .\label{eq:special unitary}
\end{gather}
The real Lie algebra $s\mathfrak{u}\left(N\right)$ has dimension
$N^{2}-1$.

\paragraph*{3. Orthogonal and special orthogonal groups}

The orthogonal group $O\left(N\right)$ is defined via,
\begin{equation}
O(N)=\left\{ g\in\mathbb{M}_{N\times N}\left(\mathbb{R}\right)|\, gg^{T}=\mathbb{I}\right\} \,,
\end{equation}
where $X^{T}$ denotes the transpose of the matrix $X$. The Lie algebra
of $O\left(N\right)$, $\mathfrak{o}\left(N\right)$, consists of
antisymmetric real matrices,
\begin{equation}
\mathfrak{o}\left(N\right)=\left\{ X\in\mathbb{M}_{N\times N}\left(\mathbb{R}\right)|\, X^{T}=-X\right\} \,.
\end{equation}
The dimension of $\mathfrak{o}\left(N\right)$ equals $\frac{N\left(N-1\right)}{2}$.
The special orthogonal group $SO\left(N\right)$ and its Lie algebra
are defined by

\begin{gather}
SO(N)=\left\{ g\in\mathbb{M}_{N\times N}\left(\mathbb{R}\right)|\, gg^{T}=\mathbb{I},\,\mathrm{det}\left(g\right)=1\right\} \,,\nonumber \\
\mathfrak{so}\left(N\right)=\mathfrak{o}\left(N\right)=\left\{ X\in\mathbb{M}_{N\times N}\left(\mathbb{R}\right)|\, X^{T}=-X\right\} .\label{eq:special orthogonal definition}
\end{gather}
The equality $\mathfrak{o}\left(N\right)=\mathfrak{so}\left(N\right)$
comes from the fact that the condition $\mathrm{tr}\left(X\right)=0$
is automatically satisfied for real antisymmetric matrices.

\subsubsection*{Geometrical and algebraic properties of Lie groups and Lie algebras}
\begin{defn}
A Lie group homomorphism between Lie groups $G$ and $H$ is defined
as a smooth mapping $\Phi:G\rightarrow H$ which is at the same time
a homeomorphism of groups: $\Phi\left(g_{1}g_{2}\right)=\Phi\left(g_{1}\right)\Phi\left(g_{2}\right)$
for arbitrary $g_{1},g_{2}\in G$. If, in addition, $\Phi$ is a diffeomorphism
of manifolds, then $\Phi$ is called the isomorphism of Lie Groups.
\end{defn}
\noindent We have an analogous notion of a Lie algebra homomorphisms
and isomorphisms.
\begin{defn}
A Lie algebra homomorphism between Lie algebras $\mathfrak{g}$ and
$\mathfrak{h}$ is a linear mapping $\phi:\mathfrak{g}\rightarrow\mathfrak{h}$
that preserves a lie bracket: $\phi\left(\left[X,\, Y\right]\right)=\left[\phi\left(X\right),\,\phi\left(Y\right)\right]$
for all $X,Y\in\mathfrak{g}$. If, in addition, $\phi$ is invertible,
then it is called an isomorphism of Lie algebras.\end{defn}
\begin{fact}
Every homomorphism of Lie groups $\Phi:G\rightarrow H$ induces the
homomorphism of the associated Lie algebras $\Phi_{\ast}:\mathfrak{g}\rightarrow\mathfrak{h}$.
The induced homomorphism $\phi$ is specified uniquely by 
\begin{equation}
\mbox{\ensuremath{\Phi_{\ast}\left(X\right)=\left.\frac{d}{dt}\right|_{t=0}}\ensuremath{\Phi}}\left(\mathrm{exp}\left[tX\right]\right)\,,\, X\in\mathfrak{g}\,.\label{eq:induced homomorphism}
\end{equation}
\end{fact}
\begin{defn}
A \textit{(finite-dimensional) representation of a Lie group}%
\footnote{By dropping the condition that $G$ is a Lie group we obtain the usual
definition of a representation of a group (without assumed additional
structure). %
} is a homomorphism
\[
\Pi:G\rightarrow GL\left(\mathcal{H}\right)
\]
from a Lie group $G$ to $GL\left(\mathcal{H}\right)$, where $\mathcal{H}$
is some finite dimensional, complex or real, vector space \textit{called
the carrier space} of the representation. Analogously, \textit{representations}
of Lie algebras are defined as homomorphisms from a Lie group $\mathfrak{g}$
to $End\left(\mathcal{H}\right)$. In what follows we will use symbols
$\Pi$ and $\pi$ to denote representations of respectively Lie groups
and Lie algebras. 

In what follows, unless it causes the ambiguity, we will refer to
the carrier space $\mathcal{H}$ as to the representation itself.
Also, unless we state otherwise, \uline{we will implicitly assume
that the dimension of the carrier spaces of considered representations
is finite}. We will also refer to, for the sake of simplicity, the
carrier space $\mathcal{H}$ having in mind the representation $\Pi$
(or $\pi$) of the group (respectively its Lie algebra).
\end{defn}
Many representation-theoretic properties of Lie groups (and the associated
algebras) depend on their topological properties.
\begin{defn}
A matrix Lie group group $G$ is connected if and only if any two
elements $g_{1},g_{2}\in G$ can be connected by a continuous path.
In other words there exist a continuous curve $\gamma:\left[0,1\right]\rightarrow G$
such that $\gamma\left(0\right)=g_{1}$ and $\gamma\left(1\right)=g_{2}$.

\noindent Examples of connected matrix Lie groups include: $U\left(N\right)$,
$SU(N)$, $SL\left(N,\,\mathbb{R}\right)$ and $SO(N)$. Examples
of disconnected Lie groups include: $\mathrm{GL}\left(N,\,\mathbb{R}\right)$
and $O\left(N\right)$.
\end{defn}
An important topological notion, stronger than the one of connectedness,
is simple-connectedness of Lie groups. Intuitively, a Lie group is
simply-connected if and only if every continuous loop on it can be
contracted (in a continuous fashion) to a point. A formal definition
reads as follows.
\begin{defn}
A matrix Lie group $G$ is simply connected if for every continuous
loop $\gamma:\left[0,1\right]\rightarrow G$ there exist a continuous
mapping $A_{\gamma}:\left[0,1\right]\times\left[0,1\right]\rightarrow G$
such that for all $0\leq s,t\leq1$ $A_{\gamma}\left(0,t\right)=\gamma(t)$,
$A\left(s,0\right)=A\left(s,1\right)$ and $A\left(1,t\right)=A\left(1,0\right)$.
\end{defn}
\noindent Examples of simply connected matrix Lie groups include:
$SU(N)$, $SL\left(N,\mathbb{R}\right)$ and $SL\left(N,\mathbb{C}\right)$.
The most relevant (for the purpose of this thesis, see Subsection
\ref{sub:Fermionic-Gaussian-states}) example of a Lie group, which
is not simply-connected is $SO\left(N\right)$. A representation $\Pi$
of a Lie group $G$ in a vector space $\mathcal{H}$ induces, via
\eqref{eq:induced homomorphism} the representation $\pi=\Pi_{\ast}$
of its Lie algebra in $\mathcal{H}$. The importance property of being
a simply-connected Lie group stems from the fact that for simply-connected
Lie group $G$ every representation $\pi$ of $\mathfrak{g}$ comes
from some representation of the group $G$. 
\begin{fact}
\label{lifting of representation} Let $G$ be a simply connected
Lie group and let $\mathfrak{g}$ be its Lie algebra. Then, given
a representation $\pi$ of Lie algebra $\mathfrak{g}$ in a vector
space $\mathcal{H}$, there exist a unique representation $\Pi$ of
the Lie group $G$ such that $\pi=\Pi_{\ast}$. In other words the
representation $\Pi$ is uniquely specified by the condition 
\[
\mathrm{exp}\left(\pi\left(X\right)\right)=\left.\frac{d}{dt}\right|_{t=0}\Pi\left(\mathrm{exp}\left[tX\right]\right)\,,\, X\in\mathfrak{g}\,.
\]

\end{fact}
The above fact is important in practice as it is often more convenient
to study representations of the Lie algebras rather than that of Lie
groups. An important example \citep{HallGroups} of a Lie group such
that not every representation of its Lie algebra can be lifted to
a representation of a group is the group $SO\left(3\right)$. Lie
algebra $\mathfrak{so}\left(3\right)$ is isomorphic to $\mathfrak{su}\left(2\right)$,
the Lie algebra of $SU\left(2\right)$. Due to simply-connectedness
of $SU\left(2\right)$, every representation of $\mathfrak{so}\left(3\right)$
comes from some representation of $SU\left(2\right)$. In fact for
every connected Lie group $G$ there exist a ``universal cover''
$G'$ i.e simply-connected Lie group whose representations are in
one-to-one correspondence with representations of $\mathfrak{g}$. 
\begin{fact}
Given a connected Lie group $G$ there exist a universal covering
group $G'$, i.e. a simply-connected Lie group $G'$ such such that
there exist a homomorphism $\Phi:G\rightarrow G'$ (called the covering
homomorphism) such that the induced Lie algebra homomorphism $\Phi_{\ast}:\mathfrak{g}\rightarrow\mathfrak{g}'$
is a Lie algebra isomorphism. 
\end{fact}
Another property of some Lie groups, which simplifies a great deal
their representation theory (see Section \ref{sec:Rep theory of semisimple}
below), is compactness.
\begin{defn}
A matrix Lie group $G\subset GL\left(N,\,\mathbb{R}\right)$ is compact
if and only if it is (as a set) bounded and closed subset of $GL\left(N,\,\mathbb{R}\right)$.
\end{defn}
\noindent Examples of compact matrix Lie groups include: $U\left(N\right)$,
$SU(N)$, $O(N)$, and $SO(N)$. Examples of non-compact matrix Lie
groups include: $GL\left(N,\,\mathbb{R}\right)$, $GL\left(N,\,\mathbb{C}\right)$,
$SL\left(N,\,\mathbb{R}\right)$ and $SL\left(N,\,\mathbb{C}\right)$.
The relevance and usefulness of the notion of compactness for representation
theory will be discussed in Section \ref{sec:Rep theory of semisimple}.
Here we state only some basic analytical and geometrical features
of matrix Lie groups.
\begin{fact}
\label{Existence-of-the haar}(\citep{HallGroups}) On a compact matrix
Lie group $G$ there exist a bi-invariant Haar measure $\mu$, i.e.
a Borel measure on $G$ such that for any Borel subset%
\footnote{For the introduction to the measure theory see for example \citep{Reed1972}.%
} $A\subset G$ and any $g\in G$ we have
\begin{equation}
\mu\left(A\right)=\mu\left(gA\right)=\mu\left(Ag\right)\,,\,\label{eq:invariance measure}
\end{equation}
where: $gE=\left\{ x\in G|\, x=ga,\, a\in A\right\} $, \textup{$Ag=\left\{ x\in G|\, x=ag,\, a\in A\right\} $.
The measure $\mu$ is finite, $\mu\left(G\right)<\infty$ and is unique
(up to a constant multiplicative factor).}
\end{fact}
In what follows we will always assume that the Haar measure on $G$
is normalized, i.e. $\mu\left(G\right)$. The concept of the Haar
measure enables to introduce integration over the considered Lie group
\footnote{We assume, in order to omit measure-theoretic details, that $f$ is
a smooth function on $G$. Obviously integration can be also defined
for more general classes of functions. %
},
\begin{equation}
\mathcal{F}\left(G\right)\ni f\rightarrow\int_{G}d\mu(a)f\left(a\right)\in\mathbb{R}\,.\label{eq:integration}
\end{equation}
By the virtue of \ref{eq:invariance measure} we have 
\begin{equation}
\int_{G}d\mu(a)f\left(ha\right)=\int_{G}d\mu(a)f\left(ah\right)=\int_{G}d\mu(a)f\left(a\right)\,,\label{eq:invariance integral}
\end{equation}
for all $h\in G$. Intuitively, one can think of the Haar measure
$\mu$ as of the ``most homogenous'' measure on $G$ (it is an analogue
of the Lebesgue measure on $\mathbb{R}^{N}$). Existence of $\mu$
gives rise to many invariant geometric structures on $G$. The concept
of Haar measure (in the case of $G=\mathrm{U}\left(N\right)$) will
be used extensively in Chapter \ref{chap:Typical-properties-of} where
this measure will be used to typical properties of correlations in
multiparty quantum systems.
\begin{fact}
On a compact matrix Lie group $G$ there exist a bi-invariant metric
$g\in\mathcal{T}_{2}^{0}\left(G\right)$ i.e an inner product satisfying
\begin{equation}
L_{a}^{\ast}g=R_{a}^{\ast}g=g\,,\,\text{for all}\, g\in G\,.\label{eq:invariant inner product}
\end{equation}
A metric $g$ is unique up to scaling by a positive factor.
\end{fact}
The existence of bi-invariant metric relies on the existence of bi-invariant
finite measure \ref{eq:invariance measure}. On the other hand, the
Haar measure \ref{eq:invariance measure} can be recovered from the
bi-invariant inner product $g$ as a canonical measure induced on
$G$ by the metric \eqref{eq:integral coord}.

\subsubsection*{Action of Lie groups on manifolds}

Let us conclude our considerations by defining the action of Lie groups
on manifolds\@. We say that the group $G$ acts on a set $X$ if
there exist a mapping 
\[
\tau:G\times X\rightarrow X\,,\,\left(g,x\right)\rightarrow\tau_{g}\left(x\right)\in X
\]
that for all $g_{1},g_{2}\in G$ and for all $x\in X$ satisfies $\tau_{g_{1}g_{2}}\left(x\right)=\tau_{g_{1}}\left(\tau_{g_{2}}\left(x\right)\right)$
and $\tau_{e}\left(x\right)=x$. We will use the commonly used notation
$\tau_{g}\left(x\right)\equiv g.x$. If $X=\mathcal{M}$ is a manifold,
$G$ is a Lie group and the mapping $\tau$ is smooth we say that
the Lie group $G$ acts on the manifold $\mathcal{M}$. For a given
element $x\in\mathcal{M}$ the set
\[
G.x=\left\{ \left.g.x\right|\, g\in G\right\} 
\]
 is called the orbit of the action of $G$ through $x$. Let $\mathrm{Stab}\left(x\right)\subset G$
denotes the stabilizer group of $x$,
\[
\mathrm{Stab}\left(x\right)=\left\{ \left.g\in G\right|\, g.x=x\right\} \,.
\]
The orbit $G.x$ can be identified with the set the coset space $G/\mathrm{Stab}\left(x\right)$
of the equivalence relation in $G$,
\[
g_{1}\sim g_{2}\,\Longleftrightarrow\, g_{1}=g_{2}h,\,\text{for some}\, h\in\mathrm{Stab}\left(x\right)\,.
\]
Let us denote by $\left[g\right]$ the equivalence class of the element
$g\in G$, $\left[g\right]=g\mathrm{Stab}\left(x\right)$. The group
$G$ acts on $G.x=G/\mathrm{Stab}\left(x\right)$ in a natural manner.
This action is defined by
\[
g_{1}.\left[g\right]=\left[g_{1}g\right],\,\text{for all}\, g_{1}\in G\,\text{and\,}\left[g\right]\in G.x\,.
\]
The orbit $G.x$ posses a natural structure of the manifold \citep{Arvanitogeorgos2003}
which is compatible with the action of $G$ on it. This action is
transitive, i.e. for arbitrary $p,q\in G.x$ there exist $\tilde{g}\in G$
such that $p=\tilde{g}.q$. Orbits of the action of the Lie group
$G$ are also called homogenous spaces of $G$. The following homogenous
spaces will be used throughout the thesis:
\begin{itemize}
\item The manifold of isospectral density matrices $\Omega_{\left\{ p_{1},\ldots,p_{D}\right\} }$
in $D$ dimensional Hilbert space $\mathcal{H}$.
\[
\Omega_{\left\{ p_{1},\ldots,p_{D}\right\} }=\left\{ \rho\in\mathcal{D}\left(\mathcal{H}\right)\left|\,\mathrm{Sp}\left(\rho\right)=\left(p_{1},\ldots,p_{D}\right)\right.\right\} ,
\]
where $\mathrm{sp}\left(\rho\right)$ denotes the ordered spectrum
of $\rho$, i.e. $p_{1}\geq p_{2}\geq\ldots\geq p_{D}\geq0$. The
Lie group $\mathrm{SU}\left(\mathcal{H}\right)$ acts naturally in
the vector space of Hermitian operators $\mathrm{Herm}\left(\mathcal{H}\right)$
via conjugation%
\footnote{Instead of the action of $\mathrm{SU}\left(\mathcal{H}\right)$ we
could have equivalently consider the action of $\mathrm{U}\left(\mathcal{H}\right)$.%
}
\[
X\rightarrow UXU^{\dagger}\:,\, X\in\mathrm{Herm}\left(\mathcal{H}\right),\, U\in\mathrm{SU}\left(\mathcal{H}\right)\,.
\]
The manifold $\Omega_{\left\{ p_{1},\ldots,p_{D}\right\} }$ can be
then interpreted as an orbit of the action of  $\mathrm{SU}\left(\mathcal{H}\right)$
through the state $\rho_{0}=\mathrm{diag}\left(p_{1},\ldots,p_{D}\right)\in\mathrm{Herm}\left(\mathcal{H}\right)$.
In particular we see that the set of pure states $\mathcal{D}_{1}\left(\mathcal{H}\right)$
is also a manifold since $\mathcal{D}_{1}\left(\mathcal{H}\right)=\Omega_{\left\{ 1,0,\ldots,0\right\} }$. 
\item Orbits of the action of the Lie subgroups of $G\subset\mathrm{U}\left(\mathcal{H}\right)$
in $\mathcal{D}_{1}\left(\mathcal{H}\right)$. In particular the sets
of ``non-correlated pure states'' considered in this thesis will
be of the form
\[
\Pi\left(K\right).\kb{\psi_{0}}{\psi_{0}}=\left\{ \left.\kb{\psi}{\psi}\right|\,\kb{\psi}{\psi}=\Pi\left(k\right)\kb{\psi_{0}}{\psi_{0}}\Pi\left(k\right)^{\dagger},\, k\in K\right\} \,,
\]
where $\kb{\psi_{0}}{\psi_{0}}$ is some fixed pure state, $K$ is
a compact simply-connected Lie group $K$ and $\Pi:K\rightarrow\mathrm{U}\left(\mathcal{H}\right)$
is a irreducible representation of $K$ in $\mathcal{H}$.
\end{itemize}

\section{Representation theory of Lie groups and Lie algebras \label{sec:Rep theory of semisimple}}

In this section we give a survey of the representation theory of Lie
groups and Lie algebras. Representation theory of Lie groups and algebras
is a vast and beautiful field and we will present only parts of it
that will be used in the forthcoming chapters. For a more detailed
treatment of the subject consult the relevant literature \citep{HallGroups,BatutRaczka,FultonHarris}.
The section is structured as follows. We begin with recalling the
standard representation-theoretic notions like irreducible representations,
unitary representations, Schur lemma and basic operations on representations.
In the latter part we discuss the structural theory of compact simply-connected
Lie groups and semisimple Lie algebras and their representation theory.
Next give a detailed description of the representation theory of the
group $\mathrm{SU\left(N\right)}$ and the concept of Schur-Weyl duality
which links the representation theory of $\mathrm{SU\left(N\right)}$
with that of the discrete permutation group $S_{m}$. Representation
theory of $\mathrm{SU}\left(N\right)$ will be used extensively in
all remaining Chapters. In the last part of this section we present
briefly some necessary facts from representation theory of $\mathrm{Spin}\left(2d\right)$
that will be useful for discussing fermionic Gaussian states in Chapters
\ref{chap:Multilinear-criteria-for-pure-states}, \ref{chap:Complete-characterisation}
and \ref{chap:Polynomial-mixed states}.

\subsection{Basic representation theory}
\begin{defn}
Let $\Pi:G\rightarrow GL\left(\mathcal{H}\right)$ be a representation
of a group $G$. Representation $\Pi$ is called\textit{ reducible}
if and only if there exist a proper subspace $\mathcal{V}\subset\mathcal{H}$
which is preserved by $G$, i.e. $\Pi\left(g\right)\mathcal{V}\subset\mathcal{V}$
for all $g\in G$. If the representation $\Pi$ is not reducible than
it is called irreducible. Analogously we define the notion of reducible
and irreducible representation of a Lie algebra. 
\end{defn}
In the study of representations of a given group $G$ it is convenient
to introduce the concept of intertwining maps. Given two representations
\[
\Pi_{1}:G\rightarrow\mathrm{GL}\left(\mathcal{H}_{1}\right),\,\Pi_{2}:G\rightarrow\mathrm{GL}\left(\mathcal{H}_{2}\right)
\]
of the group $G$ we say that a linear map $F:\mathcal{H}_{1}\rightarrow\mathcal{H}_{2}$
\textit{intertwines }representations $\Pi_{1}$ and $\Pi_{2}$ if
and only if 
\[
F\circ\Pi_{1}\left(g\right)=\Pi_{2}\left(g\right)\circ F
\]
for all $g\in G$. Representations $\Pi_{1}$ and $\Pi_{2}$ are said
to be equivalent if and only if there exist an invertible linear map
$F$ that intertwines $\Pi_{1}$ and $\Pi_{2}$. Intertwining maps
between representations of Lie algebras are defined analogously. A
very useful tool in the study of irreducible representations of groups
and Lie algebras is Schur Lemma.
\begin{fact}
\label{Schur-Lemma}(Schur Lemma) Let $F,F'$ be intertwining mappings
between two complex irreducible representations $\Pi_{1}$ and $\Pi_{2}$
of the group $G$. Then
\[
F=\alpha F'\,,
\]
where $\alpha$ is a complex number (possibly equal $0$). Analogous
result holds for intertwining maps between irreducible finite dimensional
representations of a Lie algebra $\mathfrak{g}$.
\end{fact}
Given a pair of representations $\Pi_{1},\,\Pi_{2}$ of the group
$G$, we can define their direct sum  $\Pi_{1}\oplus\Pi_{2}$ and
the tensor product $\Pi_{1}\otimes\Pi_{2}$. These are representations
of the group $G$ defined in the following manner,
\begin{equation}
\Pi_{1}\oplus\Pi_{2}:G\rightarrow\mathrm{GL}\left(\mathcal{H}_{1}\oplus\mathcal{H}_{2}\right),\, g\rightarrow\Pi_{1}\left(g\right)\oplus\Pi_{2}\left(g\right)\,,\label{eq:direct sum}
\end{equation}
\begin{equation}
\Pi_{1}\otimes\Pi_{2}:G\rightarrow\mathrm{GL}\left(\mathcal{H}_{1}\otimes\mathcal{H}_{2}\right),\, g\rightarrow\Pi_{1}\left(g\right)\otimes\Pi_{2}\left(g\right)\,.\label{eq:tensor product of reps}
\end{equation}
Analogously, for a pair of two representations
\[
\pi_{1}:\mathfrak{g}\rightarrow\mathrm{End}\left(\mathcal{H}_{1}\right),\,\pi_{2}:\mathfrak{g}\rightarrow\mathrm{End}\left(\mathcal{H}_{2}\right)
\]
 of a Lie algebra $\mathfrak{g}$, we can define their sum $\pi_{1}\oplus\pi_{2}$
and the representation $\pi_{1}\otimes\pi_{2}$ (whose carrier space
is $\mathcal{H}_{1}\otimes\mathcal{H}_{2}$)
\begin{equation}
\pi_{1}\oplus\pi_{2}:\mathfrak{g}\rightarrow\mathrm{End}\left(\mathcal{H}_{1}\oplus\mathcal{H}_{2}\right),\, X\rightarrow\pi_{1}\left(X\right)\oplus\pi_{2}\left(X\right)\,,\label{eq:direct sum-alg}
\end{equation}
\begin{equation}
\pi_{1}\otimes\pi_{2}:\mathfrak{g}\rightarrow\mathrm{End}\left(\mathcal{H}_{1}\otimes\mathcal{H}_{2}\right),\, X\rightarrow\pi_{1}\left(X\right)\otimes\mathbb{I}+\mathbb{I}\otimes\pi_{2}\left(X\right)\,.\label{eq:tensor product of reps alg}
\end{equation}
The construction of the tensor product of representations can be also
defined if we have a pair of representations of two different groups,
\[
\Pi_{1}:G_{1}\rightarrow\mathrm{GL}\left(\mathcal{H}_{1}\right)\,,\,\Pi_{1}:G_{2}\rightarrow\mathrm{GL}\left(\mathcal{H}_{2}\right)\,.
\]
The tensor product of representations $\Pi_{1}\otimes\Pi_{2}$ is
a representation of the group $G_{1}\times G_{2}$ defined in a natural
way
\begin{equation}
\Pi_{1}\otimes\Pi_{2}:G_{1}\times G_{2}\rightarrow\,\mathrm{GL}\left(\mathcal{H}_{1}\otimes\mathcal{H}_{2}\right),\,\left(g_{1},g_{2}\right)\rightarrow\Pi_{1}\left(g_{1}\right)\otimes\Pi_{2}\left(g_{2}\right)\,.\label{eq:tensor product two reps}
\end{equation}
Similarly we define the ``tensor product'' of a pair of representations
of two different Lie algebras.
\begin{fact}
\label{irreducibility of product}Let $\Pi_{1}:G_{1}\rightarrow\mathrm{GL}\left(\mathcal{H}_{1}\right)$,~$\Pi_{2}:G_{2}\rightarrow\mathrm{GL}\left(\mathcal{H}_{2}\right)$
be two irreducible representations of groups $G_{1}$ and $G_{2}$
respectively. Then the representation $\Pi_{1}\otimes\Pi_{2}:G_{1}\times G_{2}\rightarrow\,\mathrm{GL}\left(\mathcal{H}_{1}\otimes\mathcal{H}_{2}\right)$
(defined as in Eq.\eqref{eq:tensor product two reps}) is irreducible.\end{fact}
\begin{defn}
A finite-dimensional representation $\Pi:G\rightarrow\mathrm{GL}\left(\mathcal{H}\right)$
of the group $G$ is called completely reducible if and only if there
exists a natural number $K$ and irreducible representations $\Pi_{i}:G\rightarrow\mathrm{GL}\left(\mathcal{H}_{i}\right)$,
$i=1,\ldots,M$, such that
\begin{equation}
\mathcal{H}=\mathcal{H}_{1}\oplus\ldots\oplus\mathcal{H}_{M}\,,\,\Pi=\Pi_{1}\oplus\ldots\oplus\Pi_{M}\,.\label{eq:completelly reducible}
\end{equation}
The concept of completely reducible representations of Lie algebras
is defined in the analogous manner.\end{defn}
\begin{fact}
Let $\Pi:G\rightarrow\mathrm{GL}\left(\mathcal{H}\right)$ be a completely
reducible representation of a Lie group $G$. Then the induced representation
$\pi=\Pi_{\ast}:\mathfrak{g}\rightarrow\mathrm{End}\left(\mathcal{H}\right)$
of a Lie algebra $\mathfrak{g}=\mathrm{Lie}\left(G\right)$ is also
completely reducible. 
\end{fact}
Let us introduce a convenient operation, useful in representation
theory: a complexification of a real Lie algebra. Given a real Lie
algebra $\mathfrak{g}$ we define its complexification as a complex
vector space, denoted by $\mathfrak{g}^{\mathbb{C}}$, being a complexification
of $\mathfrak{g}$ as a vector space%
\footnote{A complexification of a real vector space $\mathcal{V}$ with a basis
$\left\{ \ket{v_{i}}\right\} _{i=1}^{i=N}$ is a complex vector space
$\mathcal{V}^{\mathbb{C}}=\mathcal{V}\otimes\mathbb{C}$ where multiplication
by scalars is extended to the field of complex numbers and as set
$\left\{ \ket{v_{i}}\right\} _{i=1}^{i=N}$ remains a basis. for the
formal definition see \citep{BatutRaczka}.%
} and endowed with a Lie product $\left[\cdot,\cdot\right]^{\mathbb{C}}$defined
by
\begin{equation}
\left[X+i\cdot Y,Z+i\cdot T\right]^{\mathbb{C}} =\left[X,Z\right]-\left[Y,T\right]+i\cdot\left(\left[X,T\right]+\left[Y,Z\right]\right)\,,\label{eq:complexified bracket}
\end{equation}
where $i$ is an imaginary unit, $X,Y,Z,T\in\mathfrak{g}$, and $\left[\cdot,\cdot\right]$
denotes the Lie bracket on $\mathfrak{g}$. Above definition ensures
that $\mathfrak{g}^{\mathbb{C}}$ together with the bracket $\left[\cdot,\cdot\right]^{\mathbb{C}}$
is a Lie algebra over the field of complex numbers. In what follows
we will consequently drop the superscript $\mathbb{C}$ when referring
to the Lie product on $\mathfrak{g}^{\mathbb{C}}$. In general, the
complexification of a real matrix Lie algebra $\mathfrak{g}$ is given
by allowing complex combinations of elements form $\mathfrak{g}$
while keeping the same matrix commutator. An important example of
this construction is the relation $\mathfrak{su}\left(N\right)^{\mathbb{C}}=\mathfrak{sl}\left(N,\,\mathbb{C}\right)$
(see Subsection \ref{sub:Lie-groups-and}). Given a representation
$\pi$ of a real Lie algebra $\mathfrak{g}$, it is possible to extend
it to the representation of $\mathfrak{g}^{\mathbb{C}}$, denoted
by $\pi^{\mathbb{C}}$, in a natural manner,
\[
\pi^{\mathbb{C}}:\mathfrak{g}^{\mathbb{C}}\rightarrow\mathrm{End}\left(\mathcal{H}\right),\,\mathfrak{g}^{\mathbb{C}}\ni X+i\cdot Y\rightarrow\pi\left(X\right)+i\cdot\pi\left(Y\right)\,,
\]
where $X,Y\in\mathfrak{g}$. The complexified representation $\pi^{\mathbb{C}}$
is often simpler to investigate than $\pi$ itself. On the hand, from
$\pi^{\mathbb{C}}$ one can often regain the relevant representation
theoretic data about $\pi$.
\begin{fact}
\label{fact:complexifiactaion of representation}A finite-dimensional
representation $\pi:\mathfrak{g}\rightarrow\mathrm{End}\left(\mathcal{H}\right)$
of the Lie algebra $\mathfrak{g}$ is completely reducible if and
only if the complexified representation $\pi^{\mathbb{C}}:\mathfrak{g}^{\mathbb{C}}\rightarrow\mathrm{End}\left(\mathcal{H}\right)$
is completely reducible. Moreover, we have
\begin{equation}
\pi=\pi_{1}\oplus\ldots\oplus\pi_{M}\,,\,\pi^{\mathbb{C}}=\pi_{1}^{\mathbb{C}}\oplus\ldots\oplus\pi_{M}^{\mathbb{C}}\,.\label{eq:completelly reducible-complexification}
\end{equation}
Therefore, in order to find decomposition onto irreducible components
we can either decompose $\pi$ ore $\pi^{\mathbb{C}}$.
\end{fact}
In quantum mechanics (but also in pure representation theory) special
role is played by unitary representations of groups, i.e. representations
taking values in unitary operators on some Hilbert space,
\begin{equation}
\Pi:G\rightarrow\mathrm{U}\left(\mathcal{H}\right)\,.\label{eq:unitary representation}
\end{equation}
Using \eqref{eq:induced homomorphism} one easily checks that in the
case when $G$ is a Lie group, the induced representation $\pi=\Pi_{\ast}$
of a Lie algebra $\mathfrak{g}=\mathrm{Lie}\left(G\right)$ takes
values in anti-Hermitian operators (we call such representation anti-Hermitian),
\begin{equation}
\pi:\mathfrak{g}\rightarrow i\cdot\mathrm{Herm}\left(\mathcal{H}\right)\,.\label{eq:antihermitian rep}
\end{equation}
The relevance of unitary representations in physics comes from the
fact that they describe symmetries of quantum systems \citep{BatutRaczka}
(that will be the case also in this thesis). On the mathematical side
it is often possible to consider a given representation as a unitary
representation (with respect to the suitably chosen inner product).
\begin{fact}
\label{Unitary-reepresentations-irreducible}Unitary representations
of groups are completely reducible. Every representation $\Pi:G\rightarrow\mathrm{U}\left(\mathcal{H}\right)$
of a compact Lie group $G$ is a unitary representation with respect
to a suitably chosen inner product.
\end{fact}
A given Lie group $G$ acts in a natural manner on its Lie algebra
$\mathfrak{g}$ via the so-called adjoint representation, denoted
by $\mathrm{Ad}$. For simplicity we assume that we deal with matrix
Lie groups but the general construction can be repeated also for the
general Lie groups. The adjoint representation of the matrix Lie group
$G$ is given by
\begin{equation}
\mathrm{Ad}:G\rightarrow\mathrm{GL}\left(\mathfrak{g}\right),\,\mathrm{Ad}_{g}X=g\cdot X\cdot g^{-1}\,,\label{eq:adjoint representation}
\end{equation}
where $X\in\mathfrak{g}$, $g\in G$ and the expression $\mathrm{Ad}_{f}X=g\cdot X\cdot g^{-1}$
is to be understood in a sense of multiplication of matrices. In what
follows we will use the customary notation for the adjoint representation,
identifying $\mathrm{Ad}\left(g\right)\equiv\mathrm{Ad}_{g}$. The
action of $G$ on $\mathfrak{g}$ induces, via \eqref{eq:induced homomorphism},
the representation of $\mathfrak{g}$ on itself, also called the adjoint
representation and denoted by $\mathrm{ad}$. Straightforward computation
leads to the following expression,
\begin{equation}
\mathrm{ad}:\mathfrak{g}\rightarrow\mathrm{End}\left(\mathfrak{g}\right),\,\mathrm{ad}_{Y}X=\left[Y,\, X\right]\,,\label{eq:adjoint rep Lie algebra}
\end{equation}
where $X,Y\in\mathfrak{g}$ and $\left[\cdot,\cdot\right]$ is the
Lie bracket in $\mathfrak{g}$. In what follows we will use the customary
notation, identifying $\mathrm{ad}_{Y}\equiv\mathrm{ad}\left(Y\right)$.
The adjoint representation $\mathrm{ad}$ of a Lie algebra is, as
we will see in the next subsection, a useful tool in the representation
theory of semisimple Lie algebras. The Killing form $\mathrm{Kill}$
is a bilinear symmetric form on a (real) Lie algebra $\mathfrak{g}$
defined by%
\footnote{If the Lie algebra $\mathfrak{g}$ is defined over $\mathbb{C}$ then
$B$ takes values in $\mathbb{C}$%
} 
\begin{equation}
\mathrm{Kill}:\mathfrak{g}\times\mathfrak{g}\rightarrow\mathbb{R},\,\mathrm{Kill}\left(X,Y\right)=\mathrm{tr}\left(\mathrm{ad}_{X}\circ\mathrm{ad}_{Y}\right)\,,\label{eq:killing fornm}
\end{equation}
where $X,Y\in\mathfrak{g}$. One easily checks that $\mathrm{Kill}$
is indeed symmetric (it follows from the fact that $\mathrm{tr}\left(A\circ B\right)=\mathrm{tr}\left(B\circ A\right)$
for arbitrary linear operators $A,B\in\mathrm{End}\left(\mathfrak{g}\right)$
) and that the following identities hold
\begin{equation}
\mathrm{Kill}\left(\mbox{\ensuremath{\mathrm{Ad}}}_{g}X,\mbox{\ensuremath{\mathrm{Ad}}}_{g}Y\right)=\mathrm{Kill}\left(X,Y\right)\,,\label{eq:orthogonality}
\end{equation}
\begin{equation}
\mathrm{Kill}\left(\mbox{\ensuremath{\mathrm{ad}}}_{Z}X,Y\right)=-\mathrm{Kill}\left(X,\mbox{\ensuremath{\mathrm{ad}}}_{Z}Y\right)\,,\label{eq:antisymmetric}
\end{equation}
For $X,Y,Z\in\mathfrak{g}$ and $g\in G$. Equation \eqref{eq:orthogonality}
means that $\mathrm{\mathrm{Kill}}$ is invariant under the action
of $G$ on $\mathfrak{g}$ via $\mathrm{Ad}$. Condition \eqref{eq:antisymmetric}
follows easily from \eqref{eq:orthogonality} and shows that for all
$Z\in\mathfrak{g}$ the mapping $\mathrm{ad}_{Z}:\mathfrak{g}\rightarrow\mathfrak{g}$
is skew-symmetric (with respect to the symmetric form $\mathrm{Kill}$).
\begin{fact}
\label{Cartan-criterion}(Cartan criterion) A real Lie algebra $\mathfrak{g}$
is a Lie algebra of some compact Lie group $G$ if and only if $\mathrm{Kill}\left(X,X\right)\leq0$
for all $X\in\mathfrak{g}$.
\end{fact}

\subsection{Representation theory of compact simply-connected Lie groups and
semisimple Lie algebras \label{sub:Structural-theory-of}}

Compact simply-connected Lie groups form an important class of groups
whose representation theory is very well understood \citep{HallGroups}.
Notions introduced in this subsection will be used throughout the
thesis as we will make use almost exclusively compact simply-connected
Lie groups (with particular emphasis on groups $\mathrm{SU}\left(N\right)$
and $\mathrm{Spin}\left(2d\right)$ - see Subsections \ref{sub:Representation-theory-of}
and \ref{sub:Spinor-represenations-of} below). 

Let us start with a definition of the simple complex Lie algebra.
In order to define it we need a notion of an ideal in a Lie algebra.
An ideal $\mathfrak{I}$ in a (complex or real) Lie algebra $\mathfrak{g}$
is a subspace of $\mathfrak{g}$ such that
\[
\left[X,Y\right]\in\mathfrak{J}
\]
for all $X\in\mathfrak{J}$ and $Y\in\mathfrak{g}$. Every Lie algebra
$\mathfrak{g}$ possesses two trivial ideals: null ideal$\left\{ 0\right\} $
and $\mathfrak{g}$ itself.
\begin{defn}
A complex Lie algebra $\mathfrak{g}$ is called simple if it is nonabelian
and does not posses nontrivial ideals. 
\end{defn}
We will now give a list of equivalent definitions of semisimple complex
Lie algebras
\begin{defn}
\label{def:semisimple}A complex Lie algebra $\mathfrak{g}$ is called
semisimple if and only if the following equivalent (the proof of the
equivalence can be found in \citep{FultonHarris}) conditions are
satisfied
\begin{itemize}
\item Lie algebra $\mathfrak{g}$ is a direct sum of simple Lie algebras
\[
\mathfrak{g}=\mathfrak{g}_{1}\oplus\ldots\oplus\mathfrak{g}_{m}\,.
\]

\item The Killing form $\mathrm{Kill}$ of $\mathfrak{g}$ is non-degenerate:
\[
\mathrm{Kill}\left(X,X\right)=0\,\Longleftrightarrow X=0\,.
\]

\item We have $\mathfrak{g}=\mathfrak{k}^{\mathbb{C}}$, where $\mathfrak{k}=\mathrm{Lie}\left(K\right)$
is a Lie algebra of compact simply-connected Lie group $K$.
\end{itemize}
\end{defn}
From the third point from the list above and Facts \ref{lifting of representation},
\ref{fact:complexifiactaion of representation} and \ref{Unitary-reepresentations-irreducible}
we get the following characterization of representations of semisimple
Lie algebras. In the following, unless we specify otherwise, symbols
$\mbox{\ensuremath{\mathfrak{g}},}\mathfrak{k}$ and $K$ will have
the meaning described in Definition \ref{def:semisimple}.
\begin{fact}
\label{fact:equivalence of repr theory semisimple}Let  $\mathfrak{g}$
be a semisimple complex Lie algebra. Then, the following hold 
\begin{itemize}
\item Every representation $\pi:\mathfrak{g}\rightarrow\mathrm{End}\left(\mathcal{H}\right)$
is completely reducible:
\begin{equation}
\pi=\pi_{1}\oplus\ldots\oplus\pi_{n}\,,\,\mathcal{H}=\mathcal{H}_{1}\oplus\ldots\oplus\mathcal{H}_{n}\,,\label{eq:irrpos}
\end{equation}
where $\pi_{i}:\mathfrak{g}\rightarrow\mathcal{H}_{i}$ are irreducible
representations. Subspaces $\mathcal{H}_{i}$ are also carrier spaces
of irreducible representations of $\mathfrak{k}$ and $K$ respectively.
\item Each irreducible representation $\pi$ of $\mathfrak{g}$ defines
a unique irreducible representation of $\mathfrak{k}$ denoted, by
the abuse of notation, also by $\pi$ (and vice versa). 
\item Each irreducible representation$\pi$ of $\mathfrak{g}$ defines a
unique irreducible representation of $K$ denoted by $\Pi$ (and vice
versa). 
\end{itemize}
\end{fact}
By the virtue of the above Fact one can study interchangeably representations
of $\mathfrak{g}$, $\mathfrak{k}$ and $K$. 

Directly from the definition of semisimple complex Lie algebra we
have.
\begin{fact}
\label{fact:imposible on dim reps}Let $\mathcal{H}$ be a one-dimensional
carrier space of a representation $\pi$ of a semisimple complex Lie
algebra $\mathfrak{g}$. Then the representation $\pi$ must be trivial,
i.e. $\pi\left(X\right)=0$ for all $X\in\mathfrak{g}$.
\end{fact}
We have the analogous results for one-dimensional representations
of compact simply-connected Lie group $K$ (in this case to every
$k\in K$ a representation $\Pi$ associates the identity operator),
or its Lie algebra, $\mathfrak{k}=\mathrm{Lie}\left(K\right)$).

In what follows we will present a brief survey of abstract structure
theory of semisimple Lie algebras and their representation theory.
Before that we will, however, discuss briefly the simplest example
of a semisimple Lie algebra: $\mathfrak{sl}\left(2,\mathbb{C}\right)$
which is a complexification of $\mathfrak{su}\left(2\right)$, a Lie
algebra of a compact simply-connected Lie group $\mathrm{K=SU}\left(2\right)$
(see subsection \ref{sub:Lie-groups-and}). During the presentation
of the case of $\mathfrak{sl}\left(2,\mathbb{C}\right)$ we will introduce
the terminology that will be used during the discussion of general
semisimple Lie algebras.

Recall that that the Lie algebra $\mathfrak{su}\left(2\right)$ consists
of traceless anti-Hermitian $2\times2$ matrices. Let us chose the
following basis of this Lie algebra
\[
\tau_{x}=\begin{pmatrix}0 & i\\
i & 0
\end{pmatrix}\,,\,\tau_{y}=\begin{pmatrix}0 & 1\\
-1 & 0
\end{pmatrix}\,,\,\tau_{z}=\begin{pmatrix}i & 0\\
0 & -i
\end{pmatrix}\,.
\]
Notice that $\tau_{\alpha}=i\sigma_{\alpha}$ ($\alpha=x,y,z$), where
$\sigma_{\alpha}$ denote the standard Pauli matrices. In the complexified
algebra $\mathfrak{su}\left(2\right)^{\mathbb{C}}=\mathfrak{sl}\left(2,\mathbb{C}\right)$
we chose the following basis
\begin{equation}
E_{+}=\begin{pmatrix}0 & 1\\
0 & 0
\end{pmatrix}\,,\, E_{-}=\begin{pmatrix}0 & 0\\
1 & 0
\end{pmatrix}\,,\, H=\begin{pmatrix}1 & 0\\
0 & -1
\end{pmatrix}\,.\label{eq:matrix complexified}
\end{equation}
Notice that $E_{\pm}=\frac{1}{2}\left(\sigma_{x}\pm i\sigma_{y}\right)$
correspond to the usual ladder operators%
\footnote{Recall that the ladder operators do not belong to the Lie algebra
$\mathfrak{su}\left(2\right)$.%
} in quantum mechanics and that $H=\sigma_{z}$. Let us write down
the commutation relation for the operators from Eq.\eqref{eq:matrix complexified}.
\begin{equation}
\left[E_{+},E_{-}\right]=H,\,\left[H,E_{+}\right]=2E_{+},\,\left[H,E_{-}\right]=2E_{-}\,.\label{eq:commutation sl2}
\end{equation}

Let us introduce now some terminology. One dimensional complex subspace
$\mathrm{span}_{\mathbb{C}}\left\{ H\right\} $ is called the \textit{Cartan
algebra} of $\mathfrak{sl}\left(2,\mathbb{C}\right)$. Operators $E_{+}$
and $E_{-}$ are called respectively positive and negative \textit{root
vectors}. In what follows we will rephrase in the mathematical language
the usual discussion of the structure of irreducible finite-dimensional
representations of $\mathfrak{sl}\left(2,\mathbb{C}\right)$. Let
$\pi:\mathfrak{sl}\left(2,\mathbb{C}\right)\rightarrow\mathrm{End}\left(\mathcal{H}\right)$
be a finite dimensional representation of Lie algebra $\mathfrak{sl}\left(2,\mathbb{C}\right)$.
From Fact \eqref{fact:equivalence of repr theory semisimple} it follows
that $\pi$ is induced from the unitary representation of the Lie
group $\mathrm{SU}\left(2\right)$. Consequently, the operator $\pi\left(H\right)$
is Hermitian and thus diagonalizable. Let $\ket{\psi_{\lambda}}$
denote the eigenvector of $\pi\left(H\right)$ corresponding to the
eigenvalue $\lambda$. From the commutation relations \eqref{eq:commutation sl2}
we have
\begin{align}
\pi\left(H\right)\pi\left(E_{+}\right)\ket{\psi_{\lambda}} & =\left(\lambda+2\right)\pi\left(E_{+}\right)\ket{\psi_{\lambda}}\,,\nonumber \\
\pi\left(H\right)\pi\left(E_{-}\right)\ket{\psi_{\lambda}} & =\left(\lambda-2\right)\pi\left(E_{-}\right)\ket{\psi_{\lambda}}\,.\label{eq:eigenvalues sl2}
\end{align}
Therefore either $\pi\left(E_{+}\right)\ket{\psi_{\lambda}}=0$ or
$\pi\left(E_{+}\right)\ket{\psi_{\lambda}}$ is an eigenvector of
$\pi\left(H\right)$ with the eigenvalue $\lambda+2$ (the analogous
statement holds for $\pi\left(E_{-}\right)\ket{\psi_{\lambda}}$).
Eigenvectors of the operator $\pi\left(H\right)$ are called \textit{weight
vectors} and the corresponding eigenvalues are called \textit{weights}.
Weight vectors $\ket{\psi_{\lambda}}$ satisfying $\pi\left(E_{+}\right)\ket{\psi_{\lambda}}=0$
are called the \textit{highest weight vectors}. Respectively weight
vectors satisfying satisfying $\pi\left(E_{-}\right)\ket{\psi_{\lambda}}=0$
are called the \textit{lowest weight vectors. }

We now describe the structure of irreducible representations of $\mathfrak{sl}\left(2,\mathbb{C}\right)$.
Let $m\geq0$ be a non-negative integer. We define an irreducible
representation $\pi_{m}:\mathfrak{sl}\left(2,\mathbb{C}\right)\rightarrow\mathrm{End}\left(\mathcal{H}\right)$
in $m+1$ dimensional Hilbert space $\mathcal{H}$. A representation
$\pi_{m}$ is specified uniquely by the highest weight $\ket{\psi_{0}}\in\mathcal{H}$
having the following properties
\begin{equation}
\pi_{m}\left(H\right)\ket{\psi_{0}}=m\ket{\psi_{m}}\,,\,\pi_{m}\left(E_{+}\right)\ket{\psi_{0}}=0\,,\,\pi_{m}\left(E_{-}\right)\ket{\psi_{0}}\neq0\,.\label{eq:highest weight sl2}
\end{equation}
We have
\[
\mathcal{H}=\mathrm{span}_{\mathbb{C}}\left\{ \left\{ \ket{\psi_{k}}\right\} _{k=0}^{k=m}\right\} \,,
\]
where $\ket{\psi_{k}}=\pi_{m}\left(E_{-}\right)^{k}\ket{\psi_{0}}$
satisfy $\pi_{m}\left(H\right)\ket{\psi_{k}}=\left(m-2k\right)\ket{\psi_{k}}$
and $\pi_{m}\left(E_{-}\right)\ket{\psi_{m}}=0$. The action of $\pi_{m}\left(E_{+}\right)$
on $\ket{\psi_{k}}$ can be deduced from \ref{eq:commutation sl2}
and from the fact that we require $\pi_{m}$ to be a representation
of $\mathfrak{sl}\left(2,\mathbb{C}\right)$. The representation described
above is equivalent to the well-known spin $2m$ representation of
$\mathfrak{sl}\left(2,\mathbb{C}\right)$. Moreover, every finite-dimensional
irreducible representation of $\mathfrak{sl}\left(2,\mathbb{C}\right)$
is equivalent to $\pi_{m}$ specified by Eq.\eqref{eq:highest weight sl2}
for some $m\geq0$.

We can now move to the presentation of the representation theory of
semisimple complex Lie algebras. We will only glimpse this beautiful
subject. Almost all definitions and Facts which we state below stem
from a comprehensive book of Brian C. Hall \citep{HallGroups}.
\begin{defn}
\label{The-Cartan-subalgebra}The Cartan subalgebra $\mathfrak{h}\subset\mathfrak{g}$
is a complex subalgebra of $\mathfrak{g}$ characterized by either
of the following equivalent conditions:
\begin{itemize}
\item $\mathfrak{h}=\mathfrak{t}^{\mathbb{C}}$, where $\mathfrak{t}$ is
a Lie algebra of a maximal torus $T$ in the Lie group $K$%
\footnote{A torus $T$ in a compact Lie group $K$ is a compact connected abelian
subgroup of $K$. A maximal torus in $K$ is a torus which is not
contained in any other torus of $K$. %
}.
\item $\mathfrak{h}$ is an abelian subalgebra of $\mathfrak{g}$ satisfying
self-normalizing property,
\begin{equation}
[X,\, Y]\in\mathfrak{h}\,\Longrightarrow\, Y\in\mathfrak{h}\,,\label{eq:self normalising}
\end{equation}
for all $X\in\mathfrak{g}$ and all $Y\in\mathfrak{h}$.
\end{itemize}
\end{defn}
Note that form \eqref{eq:self normalising} and form the fact that
$\mathfrak{h}$ is abelian it follows that $[X,\, Y]\in\mathfrak{h}$
implies $Y=0$. The dimension of $\mathfrak{h}$ is called the rank
of $\mathfrak{g}$ (or equivalently, the rank of $\mathfrak{k}$ or
$K$) and will be denoted by $r$. The 
\begin{fact}
\label{fact:invariant product on semisimple}(Hall \citep{HallGroups})
On a semisimple Lie algebra $\mathfrak{g}$ there exist a $1\frac{1}{2}$
linear inner product $\bk{\cdot}{\cdot}$such that the adjoint action
of $K$ on $\mathfrak{g}$ is unitary,
\begin{equation}
\bk{\mathrm{Ad}_{k}X}{\mathrm{Ad}_{k}Y}=\bk XY\,,\label{eq:unitarty inner product}
\end{equation}
for all $X,Y\in\mathfrak{g}$ and all $k\in K$. 
\end{fact}
In what follows we will assume that we have chosen a fixed $K$-invariant
inner product $\bk{\cdot}{\cdot}$ on $\mathfrak{g}$. The facts that
we are going to present do not depend on this choice. For the cases
of groups $K=\mathrm{SU}\left(N\right)$ and $K=\mathrm{Spin\left(2d\right)}$
(and the corresponding semisimple Lie algebras) that will be used
throughout this thesis (c.f. Chapters \ref{chap:Multilinear-criteria-for-pure-states},
\ref{chap:Polynomial-mixed states} and \ref{chap:Typical-properties-of})
we will use a concrete $K$-invariant inner product on $\mathfrak{g}=\mathrm{Lie}\left(K\right)^{\mathbb{C}}$(see
Subsections \ref{sub:Representation-theory-of} and \ref{sub:Spinor-represenations-of}). 

Form the Fact \ref{fact:invariant product on semisimple} it follows
that the adjoint representation of a Lie algebra $\mathfrak{k}$ on
$\mathfrak{g}$ is anti-Hermitian,
\begin{equation}
\mathfrak{k}\rightarrow i\cdot\mathrm{Herm}\left(\mathfrak{g}\right)\,,\, X\rightarrow\mathrm{ad}_{X}\,.\label{eq:ad anti hermitian}
\end{equation}
Consequently, form the definition of the Cartan subalgebra it follows
that $\mathfrak{g}$ decomposes onto joint eigenspaces of mutually
commuting operators $\mathrm{ad}_{H}$, where $H\in\mathfrak{h}$.
The decomposition reads,
\begin{equation}
\mathfrak{g}=\mathfrak{h}\oplus\bigoplus_{\alpha}\mathfrak{g}_{\alpha}\,,\label{eq:root space decomposition}
\end{equation}

where the subspace $\mathfrak{g}_{\alpha}$ is spanned by $X\in\mathfrak{g}$
such that there exist a nonzero element $\alpha\in\mathfrak{h}^{\ast}$
(the space dual to $\mathfrak{h}$) for which
\begin{equation}
\mathrm{ad}{}_{H}(X)=\alpha(H)X\,,\label{eq:definition root}
\end{equation}
for all $H\in\mathfrak{h}$. Such $\alpha\in\mathfrak{h}^{\ast}$
are called \textit{roots} whereas $\mathfrak{g}_{\alpha}$ are called
\textit{root spaces}. The collection of all roots of a Lie algebra
$\mathfrak{g}$ will be denoted by $\mathcal{R}\left(\mathfrak{g}\right)$.
From \eqref{eq:ad anti hermitian} we conclude that for $H\in\mathfrak{t}$
we have $\alpha\left(H\right)\in i\cdot\mathbb{R}$ (in other words
$\alpha\left(H\right)$ is purely imaginary). The following fact summarizes
important properties of roots and root vectors.
\begin{fact}
\label{fact:root root spaces}Let $\mathfrak{g}_{\alpha}$ and $\mathcal{R}\left(\mathfrak{g}\right)$
be as above.
\begin{enumerate}
\item For $\alpha\in\mathcal{R}\left(\mathfrak{g}\right)$ only multiples
of $\alpha$ that belong to $\mathcal{R}\left(\mathfrak{g}\right)$
are $\alpha$ and $-\alpha$.
\item The set of roots $\mathcal{R}\left(\mathfrak{g}\right)$ spans $\mathfrak{h}^{\ast}$.
\item Subspaces $\mathfrak{g}_{\alpha}$ are one dimensional and manually
orthogonal (with respect to inner product \eqref{eq:unitarty inner product}).
Moreover for $X\in\mathfrak{g}_{\alpha},\, Y\in\mathfrak{g}_{\beta}$
we have
\begin{equation}
[X,Y]\in\mathfrak{g}_{\alpha+\beta},\,\label{eq:roots addition}
\end{equation}
where it is assumed that $\mathfrak{g}_{\alpha+\beta}=\left\{ 0\right\} $
whenever $\alpha+\beta\notin\mathcal{R}\left(\mathfrak{g}\right)$.
\end{enumerate}
\end{fact}
As we will see bellow roots and root vectors play a crucial role in
the representation theory of $\mathfrak{g}$ (and thus, by virtue
of Fact \ref{fact:equivalence of repr theory semisimple}, in the
representation theory of $\mathfrak{k}$ and $K$). We now proceed
with a few more technicalities necessary to state the\textit{ highest
weight }theorem (see Fact \ref{Highest-weight-theorem}), a theorem
that characterizes all irreducible representations of semisimple complex
Lie algebras $\mathfrak{g}$. 
\begin{fact}
It is possible to chose the basis of $\mathfrak{h}^{\ast}$, denoted
$\Delta=\left\{ \alpha_{i}\right\} _{i=1}^{r}$ (elements of this
set are called positive simple roots), in such a way that every root
$\alpha\in\mathcal{R}\left(\mathfrak{g}\right)$ has the decomposition
\[
\alpha=\sum_{i=1}^{r}m_{i}\alpha_{i}\,,
\]
where numbers $n_{i}$ are integers, either all non-negative or all
non-positive. Roots such that for all $i$, $m_{i}\geq0$ (respectively
for all $i$, $m_{i}\leq0$) are called positive roots (respectively
negative roots). We have the following decomposition (called a root
decomposition) of a Lie algebra $\mathfrak{g}$, 
\begin{equation}
\mathfrak{g}=\mathfrak{n}_{-}\oplus\mathfrak{h}\oplus\mathfrak{n}_{+}.\label{eq:root decomposition of g}
\end{equation}
In the above the subspaces $\mathfrak{n}_{\pm}$ are spanned by the
root vectors corresponding to the positive and, respectively, the
negative roots. By the virtue of \eqref{eq:roots addition} subspaces
$\mathfrak{n}_{\pm}$ are actually subalgebras of $\mathfrak{g}$. 
\end{fact}
A choice of positive simple roots $\Delta$ introduces on $\mathfrak{h}^{\ast}$
the relation of partial order. It is defined in the following way.
\begin{equation}
\alpha\succeq\beta\,\Longleftrightarrow\,\alpha-\beta=\sum_{i=1}^{r}x_{i}\alpha_{i}\,,\, x_{i}\geq0\,,\label{eq:partial order}
\end{equation}
where $\alpha,\beta\in\mathfrak{h}^{\ast}$.

In what follows we will consequently identify $\mathfrak{h}$ with
$\mathfrak{h}^{\ast}$ via the (antilinear) pairing $\varphi$ induced
by the $\mathrm{Ad}_{K}$-invariant inner product \eqref{eq:unitarty inner product},
\begin{equation}
\varphi:\mathfrak{h}\rightarrow\mathfrak{h}^{\ast}\,,\,\varphi_{H_{1}}\left(H_{2}\right)=\left\langle H_{1},H_{2}\right\rangle \,,\label{eq:cartan algebra identification}
\end{equation}
where $H_{1},H_{2}\in\mathfrak{h}$. For the sake of simplicity, will
not be using the mapping $\varphi$ explicitly. Instead, we will just
refer to $\alpha\in\mathfrak{h}^{\ast}$ as to an element of $\mathfrak{h}$
by writing $\alpha\left(H\right)=\left\langle \alpha,H\right\rangle $.
Consequently roots will be identified with elements of,$\mathcal{R}\left(\mathfrak{g}\right)\subset\mathfrak{h}$. 

A convenient way of description of a representation of $\mathfrak{g}$,
$\pi:\mathfrak{g}\rightarrow\mathrm{End}\left(\mathcal{H}\right)$,
uses the notion of \textit{weight vectors}, i.e. simultaneous eigenvectors
of representatives of all elements from the Cartan subalgebra $\mathfrak{h}$.
It means that $\ket{\psi_{\lambda}}\in\mathcal{H}$ is a weight vector
if:
\begin{equation}
\pi(H)\ket{\psi_{\lambda}}=\left\langle \lambda,H\right\rangle \ket{\psi_{\lambda}}\,,
\end{equation}
where $H\in\mathfrak{h}$ and $\lambda\in\mathfrak{h}$ is called
the \textit{weight} associated to the weight vector $\ket{\psi_{\lambda}}$.
The set of all weight that correspond to weight vectors in $\mathcal{H}$
are denoted by $\mathrm{wt}\left(\pi\right)$. Subspaces of $\mathcal{H}$
corresponding to the same weight are called the \textsl{weight spaces}
and are denoted by $\mathcal{H}_{\lambda}$. 
\begin{fact}
\label{fact:weight spaces-1}Let $\pi:\mathfrak{g}\rightarrow\mathrm{End}\left(\mathcal{H}\right)$
be a representation of semisimple complex Lie algebra $\mathfrak{g}$.
Let $\mathcal{H}_{\lambda},\mathfrak{g}_{\alpha}$ and $\mathrm{wt}\left(\pi\right)$
be as above.
\begin{enumerate}
\item The set weight vectors span $\mathcal{H}$.
\item Subspaces $\mathfrak{\mathcal{H}}_{\lambda}$ are manually orthogonal
(with respect to the $K$-invariant inner product on $\mathcal{H}$).
Moreover for $X\in\mathfrak{g}_{\alpha},\,\ket{\psi}\in\mathfrak{g}_{\beta}$
we have
\begin{equation}
\pi\left(X\right)\ket{\psi_{\lambda}}\in\mathcal{H}_{\lambda+\alpha},\,\label{eq:weight addition}
\end{equation}
where it is assumed that $\mathcal{H}_{\lambda+\alpha}=\left\{ 0\right\} $
whenever $\lambda+\beta\notin\mathrm{wt}\left(\pi\right)$.
\end{enumerate}
\end{fact}
We say that a weight $\lambda_{0}$ is the \textit{highest weight}
if and only if $\lambda_{0}\succeq\lambda$ for all weights $\lambda$
in the representation%
\footnote{Note that not all representations $\pi$ admit the existence of the
highest weight.%
} $\pi$. The corresponding weight vector is called the \textit{highest
weight vector }and will be denoted by $\ket{\psi_{\lambda_{0}}}$
or simply $\ket{\psi_{0}}$.
\begin{defn}
An element $\mu$ of $\mathfrak{h}$ is called an \textsl{integral
element} if $2\frac{\left\langle \mu,\alpha\right\rangle }{\left\langle \alpha,\alpha\right\rangle }$
is an integer for all $\alpha\in\Delta$. If $2\frac{\left\langle \mu,\alpha\right\rangle }{\left\langle \alpha,\alpha\right\rangle }$
is a non-negative integer then $\mu$ is called a \textsl{dominant
integral element}%
\footnote{The condition that $2\frac{\left\langle \mu,\alpha\right\rangle }{\left\langle \alpha,\alpha\right\rangle }$
is integer depends upon the identification \eqref{eq:cartan algebra identification}.
However this requirement is invariant if we think of $\mu$ as an
element of $\mathfrak{h}^{\ast}$.%
}.

It turns out that weights are automatically integral elements.\end{defn}
\begin{fact}
(\citep{HallGroups}) \label{integral element}Let $\pi:\mathfrak{g}\rightarrow\mathrm{End}\left(\mathcal{H}\right)$
be the representation of the semisimple complex Lie algebra $\mathfrak{g}$.
Every weight $\lambda\in\mathfrak{h}$ is an integral element. \end{fact}
\begin{defn}
A representation $\pi:\mathfrak{g}\rightarrow\mathrm{End}\left(\mathcal{H}\right)$
is called a highest weight cyclic representation if there exist $\ket{\psi_{\lambda_{0}}}\in\mathcal{H}$
such that
\begin{itemize}
\item $\ket{\psi_{\lambda_{0}}}$ is a highest weight vector in the representation
$\pi$;
\item $\pi\left(X\right)\ket{\psi_{\lambda_{0}}}$ for all $X\in\mathfrak{n}_{+}$;
\item The smallest invariant subspace of $\mathcal{H}$ is whole $\mathcal{H}$.
\end{itemize}
\end{defn}
\begin{fact}
(The highest weight theorem \citep{HallGroups}) \label{Highest-weight-theorem}Let
$\mathfrak{g}$ be a semisimple complex Lie algebra. The following
holds
\begin{enumerate}
\item Every irreducible finite dimensional representation $\pi$ of $\mathfrak{g}$
has a highest weight $\lambda_{0}$, i.e.
\begin{equation}
\lambda_{0}-\lambda=\sum_{i=1}^{r}x_{i}\alpha_{i}\,,x_{i}\geq0\,;\label{eq:highest weight}
\end{equation}
for all weights $\lambda$ of $\pi$.
\item Irreducible representations $\pi,\pi'$ having the same highest weight
are equivalent.
\item Every highest weight $\lambda_{0}$ is a dominant integral element%
\footnote{Not every weight $\lambda$which is dominant integral element comes
from a highest weight vector.%
}.
\item Every irreducible representation $\pi$ is a highest weight cyclic
representation.
\end{enumerate}
\end{fact}
In what follows we will use the symbol $\pi^{\lambda_{0}}$ and to
refer to a representation of the Lie algebra $\mathfrak{g}$ that
is characterized to a highest weight $\lambda_{0}$. The carrier space
of this representation will be denoted by $\mathcal{H}^{\lambda_{0}}$.
Let us remark that in the above discussion we chose some additional
structures in order to describe the structure of irreducible representations
of semisimple complex Lie algebras: compact real form $\mathfrak{k}$,
Cartan subalgebra $\mathfrak{h}$, set of positive simple roots $\Delta$,
$K$-invariant inner product $\left\langle \cdot,\cdot\right\rangle $.
However, it can be shown that different choices result with equivalent
results \citep{HallGroups}.

Throughout the thesis we will use Fact \ref{Highest-weight-theorem}
extensively as the majority of considered cases will correspond precisely
to the setting in which a compact simply-connected Lie group $K$
(and thus also a Lie algebra $\mathfrak{g}=\mathfrak{k}^{\mathbb{C}}$)
is irreducibly represented in a Hilbert space $\mathcal{H}^{\lambda_{0}}$.
In the next two subsections we will describe in detail how the machinery
introduced above works for two compact simply-connected Lie groups
$\mbox{\ensuremath{\mathrm{SU}}}\left(N\right)$ (the corresponding
semisimple complex Lie algebra is $\mathfrak{sl}\left(N,\mathbb{C}\right)=\mathfrak{su}\left(N\right)^{\mathbb{C}})$
and $\mathrm{Spin}\left(2d\right)$ (the corresponding semisimple
complex Lie algebra is $\mathfrak{so}\left(N,\mathbb{C}\right)=\mathfrak{so}\left(N\right)^{\mathbb{C}})$. 

We conclude this subsection with introducing the Casimir operator
of a Lie algebra $\mathfrak{g}$. This concept will be used extensively
in Chapter \ref{chap:Multilinear-criteria-for-pure-states}. A Casimir
operator $\mathcal{C}_{2}$ is not an element of a Lie algebra $\mathfrak{g}$
but a so-called universal enveloping algebra $\mathfrak{U}\left(\mathfrak{g}\right)$.
For the sake of completeness let us briefly introduce concept of $\mathfrak{U}\left(\mathfrak{g}\right)$.
The universal enveloping algebra $\mathfrak{U}\left(\mathfrak{g}\right)$
is the ``smallest associative algebra'' generated by $\mathfrak{g}$.
Intuitively speaking $\mathfrak{U}\left(\mathfrak{g}\right)$ is a
complex vector which is generated by formal products of elements from
$\mathfrak{g}$, subject to the relation coming from the Lie bracket.
In other words every element of $A\in\mathfrak{U}\left(\mathfrak{g}\right)$
is a finite sum of products of the form 
\[
X_{1}\cdot X_{2}\cdot\ldots\cdot X_{k}\,,
\]
where $X_{i}\in\mathfrak{g}$ and it is assumed that the product $\cdot\mathfrak{U}\left(\mathfrak{g}\right)\times\mathfrak{U}\left(\mathfrak{g}\right)\rightarrow\mathfrak{U}\left(\mathfrak{g}\right)$
respects the commutation relation generated by the Lie bracket,
\[
X\cdot Y-Y\cdot X=\left[X,Y\right]\,.
\]
For the formal definition of $\mathfrak{U}(\mathfrak{g})$ see \citep{BatutRaczka}.
Every representation $\pi:\mathfrak{g}\rightarrow\mathrm{End}\left(\mathcal{H}\right)$
induces a representation $\tilde{\pi}:\mathfrak{U}\left(\mathfrak{g}\right)\rightarrow\mathrm{End}\left(\mathcal{H}\right)$
of the universal enveloping Lie algebra defined on monomials in the
following way
\[
\tilde{\pi}\left(X_{1}\cdot X_{2}\cdot\ldots\cdot X_{k}\right)=\tilde{\pi}\left(X_{1}\right)\tilde{\pi}\left(X_{2}\right)\ldots\tilde{\pi}\left(X_{k}\right)
\]
and extended to the whole $\mathfrak{U}(\mathfrak{g})$ by linearity.
For the sake of simplicity in what follows we use the symbol $\pi$
while referring to representations of either Lie algebra or the corresponding
universal enveloping algebra. Let $\left\langle \cdot,\cdot\right\rangle $
be $K$-invariant $1\frac{1}{2}$-linear inner product on $\mathfrak{g}$.
Let use the same symbol to denote the $K$-invariant inner product
on $\mathfrak{k}\subset\mathfrak{g}$. Let $\left\{ X_{i}\right\} _{i=1}^{\mathrm{dim}\left(\mathfrak{k}\right)}$
be the orthonormal basis of $\mathfrak{k}$. The Casimir invariant
is defined in the following way
\begin{equation}
\mathcal{C}_{2}=-\sum_{i=1}^{\mathrm{dim}\left(\mathfrak{g}\right)}X_{i}^{2}\,\label{eq:casimir definition}
\end{equation}
It can be checked that $\mathcal{C}_{2}$ commutes with all elements
of $\mathfrak{g}$ and therefore belongs to the center of $\mathfrak{U}(\mathfrak{g})$.
As a result, due to Schur lemma (see Fact \ref{Schur-Lemma}), representations
of $\mathcal{C}_{2}$ act as a multiplication by a scalar on every
irreducible representation of $\mathfrak{g}$ (the scalar depends
upon the considered representation). It is perhaps instructive to
mention at this point that in the case of $\mathrm{SU}\left(2\right)$
the second order Casimir equals (up to the positive constant depending
on the choice of the inner product on $\mathfrak{su}\left(2\right)$)
the square of the total angular momentum,
\[
L^{2}=\frac{1}{4}\left(\sigma_{x}^{2}+\sigma_{y}^{2}+\sigma_{z}^{2}\right)\,,
\]
where $\sigma_{x},\sigma_{y},\sigma_{z}$ denote the standard Pauli
matrices. The representation of $L^{2}$ for the representation of
$\mathrm{SU}\left(2\right)$ labeled by the total spin $j$ equals
$j\left(j+1\right)$. For the general semisimple Lie algebra and the
representation labeled by the highest weight $\lambda_{0}$, the representation
of the $\mathcal{C}_{2}$, $\pi^{\lambda_{0}}(C_{2})$ is given by
the formula:
\begin{equation}
\pi^{\lambda_{0}}(C_{2})=(\lambda_{0}+2\delta,\lambda_{0})\mathbb{I}\,,\label{eq:casimir irrep value}
\end{equation}
where $\mathbb{I}$ is the identity operator on $\mathcal{H}^{\lambda_{0}}$,
$\delta=\frac{1}{2}\sum_{\alpha\in\Delta}\alpha$ and $(\cdot,\cdot)$
is the inner product on $\mathfrak{h}$ induced from the inner product
$\left\langle \cdot,\cdot\right\rangle $ on $\mathfrak{g}$.

\subsection{Representation theory of the $\mathrm{SU\left(N\right)}$ and Schur-Weyl
duality\label{sub:Representation-theory-of}}

\noindent In this part we state necessary facts concerning the representation
theory of $\mathrm{SU}\left(N\right)$ that we employ in Chapters
\ref{chap:Multilinear-criteria-for-pure-states}-\ref{chap:Typical-properties-of}
to describe correlations of distinguishable particles, fermions and
bosons. Representation theory of $\mathrm{SU}\left(N\right)$ is closely
related to the representation theory of symmetric groups. This relation
is known under the name Schur-Weyl duality'' \citep{FultonHarris,Cvitanovic,Fulton1997}.
First we present the classification of all irreducible representations
of $\mathrm{SU}\left(N\right)$ and the representation theory of symmetric
group of $L$ elements, $\mathfrak{S}_{L}$. Then we turn to their
concrete realizations as sub representations of $\left(\mathbb{C}^{N}\right)^{\otimes L}$.
Schur-Weyl duality is employed here to give the pairing of irreducible
representations of $\mathrm{SU}\left(N\right)$ and the symmetric
group $\mathfrak{S}_{L}$ that acts on $\left(\mathbb{C}^{N}\right)^{\otimes L}$
in a natural manner.

\subsubsection*{Highest weights of $\mathrm{SU}\left(N\right)$}

\noindent Lie group $\mathrm{SU}\left(N\right)$ fits into the formalism
presented in Subsection \ref{sub:Structural-theory-of} - it is compact
and simply-connected. Therefore, finite-dimensional irreducible representations
of the group, its complexification and the corresponding Lie algebras
are in one to one correspondence. The Lie algebra of $\mathrm{SU}\left(N\right)$,
$\mathfrak{su}(N)$, consists of traceless anti-Hermitian matrices
acting on $\mathbb{C}^{N}$:
\begin{equation}
\mathfrak{su}(N)=\left\{ X\in\mathbb{M}_{N\times N}\left(\mathbb{C}\right)|X^{\dagger}=-X,\,\mathrm{tr}X=0\right\} \,.
\end{equation}

\noindent The Lie algebra $\mathfrak{t}$ of a maximal torus in $\mathrm{SU}\left(N\right)$
can be chosen to consist of diagonal, traceless and anti-Hermitian
matrices:
\begin{equation}
\mathfrak{t}=\left\{ X\in\mathbb{M}_{N\times N}\left(\mathbb{C}\right)|X\,\text{- diagonal, }X^{\dagger}=-X,\,\mathrm{tr}X=0\right\} \,.
\end{equation}

\noindent The complexification of $\mathfrak{su}(N)$ consists of
complex traceless matrices:
\begin{equation}
\mathfrak{su}(N)^{\mathbb{C}}=\mathfrak{sl}(N,\mathbb{C})=\left\{ X\in\mathbb{M}_{N\times N}\left(\mathbb{C}\right)|\,\mathrm{tr}X=0\right\} \,.
\end{equation}

\noindent The corresponding complexified Cartan subalgebra $\mathfrak{t}^{\mathbb{C}}=\mathfrak{h}$
consists of diagonal traceless matrices:
\begin{equation}
\mathfrak{t}^{\mathbb{C}}=\mathfrak{h}=\left\{ X\in\mathbb{M}_{N\times N}\left(\mathbb{C}\right)|X\,\text{- diagonal,}\,\mathrm{tr}X=0\right\} \,.
\end{equation}

\noindent On $\mathfrak{sl}(N,\mathbb{C})$ we have an $\mathrm{SU}\left(N\right)$-invariant
skew-linear inner product given by the Hilbert-Schmidt inner product
given by:
\begin{equation}
\left\langle X,Y\right\rangle _{\mathrm{HS}}=\mathrm{tr}\left(X^{\dagger}Y\right),\label{eq:hilber schidt inner product}
\end{equation}
for $X,Y\in\mathfrak{sl}(N,\mathbb{C})$. We chose the following set
of simple positive roots $\Delta\subset\mathfrak{h}$:
\begin{equation}
\Delta=\left\{ E_{1,1}-E_{2,2},\, E_{2,2}-E_{3,3},\ldots,\, E_{N-1,N-1}-E_{N,N}\right\} \,,
\end{equation}

\noindent where: $E_{i,j}=\kb ij$ and $\ket i$ denotes the i'th
element of the standard basis in $\mathbb{C}^{N}$. This choice of
positive roots $\Delta$ corresponds to
\[
\mathfrak{n}_{+}=span_{\mathbb{C}}\left\{ E_{i,j}|i<j\right\} \,\text{ (strictly upper triangular matrices)}
\]
and
\[
\mathfrak{n}_{-}=span_{\mathbb{C}}\left\{ E_{i,j}|i>j\right\} \,\text{ (strictly lower triangular matrices)}
\]
Moreover, one easily checks that the positive root corresponding to
the root vector $E_{i,j}$ ($i\neq j$) equals $E_{i,i}-E_{j,j}$.
Let now $\pi:\mathfrak{sl}\left(N,\mathbb{C}\right)\rightarrow\mathrm{End}\left(\mathcal{H}\right)$
denotes the representation of $\mathfrak{sl}\left(N,\mathbb{C}\right)$
in a Hilbert space $\mathcal{H}$. As it was explained in the previous
subsection any weight $\lambda\in\mathfrak{h}$ encodes the simultaneous
eigenvalues of representatives of operators $\pi\left(E_{i,i}-E_{i+1,i+1}\right)$
on a weight vector $\ket{\psi_{\lambda}}$:
\begin{equation}
\pi\left(E_{i,i}-E_{i+1,i+1}\right)\ket{\psi_{\lambda}}=\lambda_{i}\ket{\psi_{\lambda}}\, i=1,\ldots,N-1\,.
\end{equation}
where $\lambda_{i}$ labels (real) eigenvalue of the operator $E_{i,i}-E_{i+1,i+1}$.
From Fact \ref{integral element} it follows that for each weight
$\lambda$ the corresponding $\lambda_{i}$ are integers. The highest
weights $\lambda\in\mathfrak{h}$ are dominant integral elements,
i.e. every coefficient $\lambda_{i}$ is non-negative. Let us introduce
the following relabeling of the highest weights of $\mathrm{sl}\left(N,\mathbb{C}\right)$:
\begin{equation}
\lambda\equiv\left(\tilde{\lambda}_{1},\ldots,\tilde{\lambda}_{d-1}\right)\,,\label{eq:relabelling weights}
\end{equation}
where
\[
\tilde{\lambda}_{i}=\sum_{k=i}^{d-1}\lambda_{k}\,.
\]
Notice that numbers $\left\{ \tilde{\lambda}_{i}\right\} _{i=1}^{d-1}$
satisfy $\tilde{\lambda}_{1}\geq\tilde{\lambda}_{2}\geq\ldots\geq\tilde{\lambda}_{d-1}$.
Therefore they form \textit{a partition} of a number $\left|\lambda\right|=\sum_{k=i}^{d-1}\tilde{\lambda_{i}}$
which we call the \textit{length of the weight} $\lambda$. For the
highest weight $\lambda$ we denote the corresponding irreducible
representation of $\mathfrak{sl}(N)$ as $\mathcal{H}_{N}^{\lambda}$.
We observe that dominant integral weights of length $L\in\mathbb{N}$
correspond to partition of the number $m$ onto array of $N-1$ non-negative
and non-increasing integers. We shall write $\lambda\vdash L$ when
we refer to such partitions. One can also represent $\lambda$ by
a \textit{Young diagram} - a collection of boxes arranged in left-justified
rows, with non increasing lengths when looked  from top to the bottom.
For example, Young tableau corresponding to $\lambda=(4,2,1)$ is: 

\noindent \begin{center}
\ydiagram{4,2,1}
.
\end{center}

\noindent For the purpose of the latter discussion we define the height
of the Young diagram, $d(\lambda)$, as a number of rows composing
the diagram. The above interpretations of highest weights of $\mathrm{SU}\left(N\right)$
establish the link between representation theory of $\mathrm{SU}\left(N\right)$
and that of permutation groups. In the rest of this section this section
we will denote by $\mathcal{H}_{N}^{\lambda}$ the irreducible representation
of $\mathrm{SU\left(N\right)}$ characterized by the highest weight
$\lambda$.

\subsubsection*{Representations of the permutation group $\mathfrak{S}_{L}$}

All irreducible representations of the permutation group $\mathfrak{S}_{L}$
can be labeled by Young diagrams of size $L$. Let us briefly describe
the construction of irreducible representation $\mathcal{V}^{\lambda}$
corresponding to shape $\lambda\vdash L$. A \textit{Young tableaux}
is obtained from the Young diagram by filling it with numbers from
the set $1,\ldots,L$ (with no repetitions). For instance from $\lambda=(4,2,1)$
we can construct the following Young tableau:

\begin{center}
\ytableausetup{centertableaux}
\begin{ytableau}
7 & 3 & 1 & 2 \\ 
4 & 5 \\ 
6 
\end{ytableau}
.
\end{center}

Let $V^{\lambda}$ be a complex vector space formed by formal linear
combinations of all $L!$ Young tableaux of the shape $\lambda$.
Permutation group $\mathfrak{S}_{L}$ acts on different Young tableaux
by permuting numbers that fill them. This action can be extended (by
linearity) to the whole $V^{\lambda}$ making it the carrier space
of the representation of $\mathfrak{S}_{L}$. Representation $\mathbb{S}^{\lambda}$
is obtained as a quotient of $V^{\lambda}$ under the equivalence
relation defined as follows. First, one identifies Young tableaux
that have the same numbers in their rows. For example:

\begin{center}
\[
\ytableaushort[]{1357,24,6} \sim\ \ytableaushort[]{1537,42,6} \ \text{.} 
\]
\end{center}

Moreover, two Young tableaux having the same numbers in their columns
are identified up to the sign of the permutation that transforms one
onto another. For example:

\begin{center}
\[
\ytableaushort[]{1357,24,6} \sim\ -\ \ytableaushort[]{2357,14,6} \ \text{.} 
\]
\end{center}

The resulting quotient of the space $V^{\lambda}$ respects the action
of the group $\mathfrak{S}_{L}$ will be denoted by $\mathcal{V}^{\lambda}$.
It can be checked that the space $\mathbb{S}^{\lambda}$ has a basis
formed by \textit{standard Young tableaux} \citep{FultonHarris}.
Standard young tableau is a tableau whose entries in each row and
each column are increasing (when read from left to the right and from
the top to the bottom). For example 

\begin{center}
\ytableausetup{centertableaux}
\begin{ytableau}
1 & 2 & 3 & 5 \\ 
4 & 7 \\ 
6 
\end{ytableau}

\end{center}is a standard Young tableau.

\subsubsection*{Schur-Weyl duality }

Consider the diagonal representation $\Pi^{\otimes L}:\mathrm{SU}\left(N\right)\rightarrow\mathrm{U}\left(\left(\mathbb{C}^{N}\right)^{\otimes L}\right)$
of a Lie group $\mathrm{SU}\left(N\right)$ in$\left(\mathbb{C}^{N}\right)^{\otimes L}$,
\begin{equation}
\Pi^{\otimes m}\left(U\right)\left(\ket{\phi_{1}}\otimes\ket{\phi_{2}}\otimes\ldots\otimes\ket{\phi_{L}}\right)=U\ket{\phi_{1}}\otimes U\ket{\phi_{2}}\otimes\ldots\otimes U\ket{\phi_{L}}\,,
\end{equation}
where $\ket{\phi_{i}}\in\mathbb{C}^{N}$, $i=1,\ldots,m$. On $\left(\mathbb{C}^{N}\right)^{\otimes L}$
we have also a natural representation of the symmetric group $\mathfrak{S}_{L}$,
\begin{equation}
\rho\left(\sigma\right)\left(\ket{\phi_{1}}\otimes\ket{\phi_{2}}\otimes\ldots\otimes\ket{\phi_{L}}\right)=\ket{\phi_{\sigma^{-1}(1)}}\otimes\ket{\phi_{\sigma^{-1}(2)}}\otimes\ldots\otimes\ket{\phi_{\sigma^{-1}(L)}}\,,\label{eq:permut action}
\end{equation}

where $\sigma\in\mathfrak{S}_{L}$. The Schur-Weyl duality describes
the pairing between irreducible representations of $\mathrm{SU}\left(N\right)$
and $\mathfrak{S}_{L}$ that appear in $\left(\mathbb{C}^{N}\right)^{\otimes L}$.
In this paragraph we state main ideas associated with this duality.
Let us first note that
\begin{equation}
\left[\rho\left(\sigma\right),\,\Pi^{\otimes L}\left(U\right)\right]=0\,,\,
\end{equation}
for all $\sigma\in\mathfrak{S}_{L}$ and all $U\in\mathrm{SU}\left(N\right)$.
Actually a much stronger condition holds. In order to state this result
let us first define the concept of the \textit{commutant}. Let $\mathcal{H}$
be a finite dimensional Hilbert space and let $S\subset\mathrm{End}\left(\mathcal{H}\right)$
be a subset of linear operators on $\mathcal{H}$ which is closed
under Hermitian conjugation%
\footnote{Actually it is possible to define the concept of the commutant of
the set $\mathcal{S}$ without this requirement. However, all the
cases that will be considered by us are of this form.%
} (for $s\in S$ we have $s^{\dagger}\in S$). The \textit{commutant
}of $S$, denoted by $\mathrm{comm}\left(S\right)$, consists of all
operators that commute with elements of $S$:
\begin{equation}
\mathrm{Comm}(S)=\left\{ X\in\mathrm{End}\left(\mathcal{H}\right)|\left[X,\, s\right]=0,\,\forall_{s\in S}\right\} \,.
\end{equation}

One checks that $\mathrm{comm}(S)$ form a $\mathbb{C}^{\ast}$ algebra%
\footnote{A $\mathbb{C}^{\ast}$ algebra in $\mathrm{End}\left(\mathcal{H}\right)$
is a vector subspace of $\mathrm{End}\left(\mathcal{H}\right)$ which
us closed under operator composition and taking the Hermitian conjugate. %
} in $\mathrm{End}\left(\mathcal{H}\right)$. Let $\mathbb{C}\left[\rho\left(\mathfrak{S}_{L}\right)\right]$
and $\mathbb{C}\left[\Pi^{\otimes L}\left(\mathrm{SU}\left(N\right)\right)\right]$
be $\mathbb{C}^{\ast}$ algebras of operators in $\left(\mathbb{C}^{N}\right)^{\otimes L}$
generated by the image of the representation $\rho$ and $\Pi^{\otimes L}$
respectively. The relation between representation of $\mathfrak{S}_{L}$
and $\mathrm{SU}\left(N\right)$ is given by the following fact.
\begin{fact}
(Schur-Weyl duality \citep{FultonHarris}) Let $\Pi^{\otimes L}$
and $\rho$ be defined as above. The following equalities hold:
\begin{align}
\mathrm{Comm}\left(\Pi^{\otimes L}\left(\mathrm{SU}\left(N\right)\right)\right) & =\mathbb{C}\left[\rho\left(\mathfrak{S}_{L}\right)\right]\,,\nonumber \\
\mathrm{Comm}\left(\rho\left(\mathfrak{S}_{L}\right)\right) & \mathbb{=C}\left[\Pi^{\otimes L}\left(\mathrm{SU}\left(N\right)\right)\right]\,.\label{eq:double commutant}
\end{align}
Moreover, the Hilbert space $\left(\mathbb{C}^{N}\right)^{\otimes L}$
treated as a carrier space of the representation of the group $\mathfrak{S}_{L}\times\mathrm{SU}\left(N\right)$
decomposes onto irreducible components as follows
\begin{equation}
\left(\mathbb{C}^{N}\right)^{\otimes L}\approx\bigoplus_{\begin{array}[t]{c}
\lambda\vdash L\\
d(\lambda)\leq N
\end{array}}\mathcal{V}^{\lambda}\otimes\mathcal{H}_{N}^{\lambda}\,,\label{eq:Schur -Weyl}
\end{equation}
where$\approx$ is the equivalence of representations and $\mathbb{S}^{\lambda}\otimes\mathcal{H}_{N}^{\lambda}$
is an irreducible representation of $\mathfrak{S}_{L}\times\mathrm{SU}\left(N\right)$. 
\end{fact}
The Meaning of the tensor product $\mathcal{V}^{\lambda}\otimes\mathcal{H}_{N}^{\lambda}$
in \eqref{eq:Schur -Weyl}is the following: $\mathfrak{S}_{L}$ acts
on the first factor of the tensor product whereas $SU(N)$ acts in
the second one. 
\begin{rem*}
The Schur-Weyl duality is usually stated in the form of duality between
the action of $\mathfrak{S}_{L}$ and $\mathrm{U}\left(N\right)$
(or $\mathrm{GL}\left(N\right)$). However, due to the equalities%
\footnote{The representation $\Pi^{\otimes L}$ of $\mathrm{U}\left(N\right)$
is defined in analogous manner as the representation of $\mathrm{SU}\left(N\right)$.%
}
\begin{align*}
\mathrm{Comm}\left(\Pi^{\otimes L}\left(\mathrm{SU}\left(N\right)\right)\right) & =\mathrm{Comm}\left(\Pi^{\otimes L}\left(\mathrm{U}\left(N\right)\right)\right)\,,\\
\mathbb{C}\left[\Pi^{\otimes L}\left(\mathrm{SU}\left(N\right)\right)\right] & =\mathbb{C}\left[\Pi^{\otimes L}\left(\mathrm{U}\left(N\right)\right)\right]\,,
\end{align*}
equations \eqref{eq:double commutant} remain valid if we replace
the group $\mathrm{SU}\left(N\right)$ by $\mathrm{U}\left(N\right)$.
In the beginning of this section we saw that irreducible representations
$\mathcal{H}_{N}^{\lambda}$ of $\mathrm{SU}\left(N\right)$ are labeled
by $\lambda\vdash L$ satisfying the condition $d\left(\lambda\right)\leq N-1$.
For $L\geq N$ in Eq.\eqref{eq:Schur -Weyl} we have Young diagrams
having the height $d\left(\lambda\right)=L$. It turns out (see that
discussion of Schur-Weyl duality in \citep{Goodman1998}) we have
the equivalence of irreducible representations $\mathcal{H}_{N}^{\lambda},\mathcal{H}_{N}^{\lambda'}$
of $\mathrm{SU}\left(N\right)$ characterized by the Young diagrams
$\lambda,\lambda'$ ($d\left(\lambda\right)\leq N$, $d\left(\lambda'\right)\leq N$
) satisfying
\begin{equation}
\,\lambda=\lambda'+m\left(\stackrel{N}{\overbrace{1,\ldots,1}}\right)\,.\label{eq:equivalence weights}
\end{equation}
A decomposition \eqref{eq:Schur -Weyl} can be considered as a generalization
the well-known decomposition $\mathbb{C}^{N}\otimes\mathcal{\mathbb{C}}^{N}=\mathrm{Sym}^{2}\left(\mathbb{C}^{N}\right)\oplus\bigwedge^{2}\left(\mathbb{C}^{N}\right)$.
In this case we have
\begin{equation}
\mathrm{Comm}\left(\Pi^{\otimes2}\left(\mathrm{SU}\left(N\right)\right)\right)=\mathrm{Span}_{\mathbb{C}}\left\{ \mathbb{I}\otimes\mathbb{I},\mathbb{S}\right\} \,,\label{eq:schur weyl ytwo copies}
\end{equation}
where $\mathbb{I}\otimes\mathbb{I}$ is the identity operator on $\mathbb{C}^{N}\otimes\mathbb{C}^{N}$
and $\mathbb{S}$ is the standard swap operator. Note that the case
with $L=2$ is exceptional in a sense that $\left(\mathbb{C}^{N}\right)^{\otimes2}$
is a multiplicity-free representation%
\footnote{A representation is multiplicity-free if and only if it decomposition
onto irreducible components does not contain two copies of the same
irreducible representation. %
} of $\mathrm{SU}\left(N\right)$. For $L>2$ this is no longer the
case as representations $\mathcal{V}^{\lambda}$ are not one dimensional.
For every standard Young tableau $\mathcal{T}$ of the shape $\lambda$
the corresponding vector $v_{\mathcal{T}}\in\mathbb{\mathcal{V}}^{\lambda}$
defines an irreducible subrepresentations of $\mathrm{SU}\left(N\right)$
of type $\mathcal{H}_{N}^{\lambda}$: $\mathcal{H}_{N}^{\lambda,\mathcal{T}}\approx v_{\mathcal{T}}\otimes\mathcal{H}_{N}^{\lambda}$.
Linear subspace $\mathcal{H}_{N}^{\lambda,\mathcal{T}}\subset\left(\mathbb{C}^{N}\right)^{\otimes L}$
is uniquely characterized by the action of $\mathfrak{S}_{L}$ on
the vector $v_{\mathcal{T}}$. According to the construction of the
representation $\mathbb{\mathcal{V}}^{\lambda}$, the vector $v_{\mathcal{T}}$
is symmetric under the action of $\sigma\in\mathfrak{S}_{L}$, when
$\sigma$ preserves the rows of $\mathcal{T}$. Similarly $v_{\mathcal{T}}$
is completely antisymmetric under the action of $\sigma$ that independently
act on different columns of $\mathcal{T}$. For example the Young
tableau 
\end{rem*}
\begin{center}
\begin{equation}
\mathcal{T}= \ytableaushort[]{1235,47,6} 
\end{equation}
\end{center}defines the subspace $\mathcal{H}_{N}^{\lambda,\mathcal{T}}$ of $\left(\mathbb{C}^{N}\right)^{\otimes7}$
which is totally symmetric under independent permutations within sets
of indices (in a sense of Eq. \eqref{eq:permut action}) forming rows
of $\mathcal{T}$: $\left\{ 1,2,3,5\right\} $, $\left\{ 4,7\right\} $and
$\left\{ 6\right\} $. Similarly, $\mathcal{H}_{N}^{\lambda,\mathcal{T}}$
is totally antisymmetric under independent permutations within indices
forming columns of $\mathcal{T}$ : $\left\{ 1,4,6\right\} $, $\left\{ 2,7\right\} $
and $\left\{ 5\right\} $. For the general Young tableau $\mathcal{T}$
one can construct the projector onto $\mathcal{H}_{N}^{\lambda,\mathcal{T}}$,
$\mathbb{P}_{\lambda,\mathcal{T}}:\,\left(\mathbb{C}^{N}\right)^{\otimes L}\rightarrow\left(\mathbb{C}^{N}\right)^{\otimes L}$.
It is given by the formula \citep{Cvitanovic}:
\begin{equation}
\mathbb{P}_{\lambda,\mathcal{T}}=\alpha\left(\lambda\right)\mathbb{P}_{r(\mathcal{T})}^{+}\mathbb{P}_{c(\mathcal{T})}^{-},\,\label{eq:projector irreducible schur}
\end{equation}
where $\mathbb{P}_{r(\mathcal{T})}^{+}$ is the projector onto subspace
of $\left(\mathbb{C}^{N}\right)^{\otimes L}$ totally symmetric under
permutations that preserve rows of $\mathcal{T}$ and $\mathbb{P}_{c(\mathcal{T})}^{-}$
is the projector onto the subspace that is totally antisymmetric under
permutations that preserve columns of $\mathcal{T}$. Scalar factor
$\alpha\left(\lambda\right)$ can in Eq. \eqref{eq:projector irreducible schur}
is a rational number depending upon the shape of the Young tableau
$\mathcal{T}$. It ensures that $\left(\mathbb{P}_{\lambda,\mathcal{T}}\right)^{2}=\mathbb{P}_{\lambda,\mathcal{T}}$
(operators $\mathbb{P}_{r(\mathcal{T})}^{+}$ and $\mathbb{P}_{c(\mathcal{T})}^{-}$
do not commute in general). The scalar $\alpha\left(\lambda\right)$
can be computed according to the formula \citep{Cvitanovic}
\begin{equation}
\alpha\left(\lambda\right)=\frac{c_{\lambda}\cdot r_{\lambda}}{g_{\lambda}}\,.\label{eq:coefficient}
\end{equation}
The meaning of symbol that appear in \eqref{eq:coefficient} is the
following
\[
c_{\lambda}=\prod_{i=1}^{\mathrm{nc}\left(\lambda\right)}c_{i}!\,,r_{\lambda}=\prod_{i=1}^{\mathrm{nr}\left(\lambda\right)}r_{i}!\,,
\]
where $\left\{ c_{i}\right\} _{i=1}^{\mathrm{nc}\left(\lambda\right)}$
($\left\{ r_{i}\right\} _{i=1}^{\mathrm{nr}\left(\lambda\right)}$)
is a set consisting of number of boxes that appear in each column
(respectively each row) of the Young diagram $\lambda$. The scalar
$g_{\lambda}$ is given by the ``hook'' rule: enter into each box
of $\lambda$ the number of boxes below and to the left of the box,
including the box itself. Then $g_{\lambda}$ is the product of the
numbers in all the boxes. 

A particularly simple examples of Young tableaux are these coming
only form a Young diagram having only one row or one column. In terms
of the notation \eqref{eq:relabelling weights} these are
\[
\lambda_{b}=\left(L,0,\ldots,0\right)\,,\,\lambda_{f}=\left(\overset{L}{\overbrace{1,1,\ldots1}},0,\ldots,0\right)\,.
\]
The representations that correspond to these Young diagrams are: the
standard ``bosonic'' representation of $\mathrm{SU}\left(N\right)$
and respectively the standard ``fermionic'' representation of $\mathrm{SU}\left(N\right)$
(see Section \eqref{sub:Representation-theory-of}). The projectors
that correspond to the appropriate Young tableaux $\mathcal{T}_{b},\mathcal{T}_{f}$
are
\[
\mathbb{P}_{\lambda_{b},\mathcal{T}_{b}}=\mathbb{P}_{\left\{ 1,\ldots,L\right\} }^{\mathrm{sym}}\,\text{and }\,\mbox{\ensuremath{\mathbb{P}_{\lambda_{f},\mathcal{T}_{f}}}}=\mathbb{P}_{\left\{ 1,\ldots,L\right\} }^{a\mathrm{sym}}.
\]

We now give the formula for the dimension of $\mathcal{H}_{N}^{\lambda}$
and the scalar in terms of the Young diagram $\lambda$ (in accordance
to \eqref{eq:equivalence weights} we now allow weights to have the
height $d\left(\lambda\right)$ at most $N$). 
\begin{fact}
\textup{(\citep{Cvitanovic}) \label{dimension irrep}The dimension
of the irreducible representation $\mathcal{H}_{N}^{\lambda}$ of
the group $\mathrm{SU}\left(N\right)$ is given by the by }
\begin{equation}
\mathrm{dim}\left(\mathcal{H}_{N}^{\lambda}\right)=\frac{f_{\lambda}\left(N\right)}{g_{\lambda}}\,.\label{eq:dimension formula}
\end{equation}
In the above formula $g_{\lambda}$ is as in \eqref{eq:coefficient}
and $f_{\lambda}\left(N\right)$ is the polynomial in $N$ obtained
from the $\lambda$ by multiplying the numbers written in the boxes
of $\lambda$, according to the following rules:
\begin{enumerate}
\item The upper left box contains $N$.
\item The numbers in a row increase by one when reading from left to right..
\item The numbers in a column decrease by one when reading from top to bottom.
\end{enumerate}
\end{fact}

\subsection{Spinor representations of $\mathrm{Spin}\left(2d\right)$\label{sub:Spinor-represenations-of}}

In this part we introduce group-theoretic tools useful for description
of non-Gaussian correlations in fermionic systems (see Subsection
\ref{sub:Fermionic-Gaussian-states}). We first define the group $\mathrm{Spin}\left(N\right)$
and introduce its spinor representations: $\mathcal{H}_{\mathrm{Fock}}^{+}\left(\mathbb{C}^{d}\right)$
and $\mathcal{H}_{\mathrm{Fock}}^{-}\left(\mathbb{C}^{d}\right)$
of the group $\mathrm{Spin}\left(2d\right)$. Then we describe these
representations from the perspective of the highest weight theory
introduced in Subsection \ref{sub:Structural-theory-of}. 

From Subsection \ref{sub:Structural-theory-of} we know that there
is one to one correspondence between irreducible finite dimensional
representations of: compact simply-connected Lie group $K$, its Lie
algebra $\mathfrak{k}$ and semi-simple complex Lie algebra $\mathfrak{g}=\mathfrak{k}^{\mathbb{C}}.$
Let us consider the compact Lie group $K=\mathrm{SO}\left(N\right)$
($N>2$). It is known \citep{HallGroups} that $\mathrm{SO}\left(N\right)$
is not simply-connected. However, the complexified Lie algebra $\mathfrak{g}=\mathfrak{so}\left(N\right)^{\mathbb{C}}$
is semi-simple \citep{HallGroups}. Not every irreducible representation
of $\mathfrak{so}\left(N\right)$ will then lift to the representation
of $\mathrm{SO}\left(N\right)$ (as in the simplest case of the group
$\mathrm{SO}\left(3\right)$ and its Lie algebra $\mathfrak{so}\left(3\right)\approx\mathfrak{su}\left(2\right)$).
However, due to the fact that $\mathrm{SO}\left(N\right)$ is compact
and simply connected there exist \citep{HallGroups} its \textit{universal
cover}, i.e. a compact simply-connected Lie group $\widetilde{\mathrm{SO}\left(N\right)}$
such that there exist a surjective homomorphism of Lie groups
\begin{equation}
h:\widetilde{\mathrm{SO}\left(N\right)}\rightarrow\mathrm{SO}(N)\label{eq:covering homo}
\end{equation}
and $\mathrm{Lie}\left(\widetilde{\mathrm{SO}\left(N\right)}\right)=\mathrm{Lie}\left(\mathrm{SO}\left(N\right)\right)=\mathfrak{so}\left(N\right)$.
In what follows we will refer to $\widetilde{\mathrm{SO}\left(N\right)}$
as the Spin group $\mathrm{Spin}\left(N\right)$. We will not describe
this group more explicitly here. For $N=3$ we have $\mathrm{Spin}\left(3\right)\approx\mathrm{SU}\left(2\right)$.
For the general $N$ the group $\mathrm{Spin}\left(N\right)$ can
be given a concrete realization via the use of Clifford algebras \citep{Goodman1998}. 

Let us now focus one the Lie algebra $\mathfrak{so}\left(2d\right)$
of the group $\mathrm{Spin}\left(2d\right)$%
\footnote{The coherent states of the group $\mathrm{Spin}\left(2d\right)$ in
its spinor representations will be identified with fermionic Gaussian
states (see Subsection \ref{sub:Fermionic-Gaussian-states}).%
} ($d>1$). Bellow we describe, after \citep{HallGroups}, the root
structure of $\mathfrak{so}\left(2d\right)^{\mathbb{C}}$and the structure
of its highest weights. The Lie algebra of $\mathrm{Spin}\left(2d\right)$
consists of real $2d\times2d$ antisymmetric matrices
\[
\mathfrak{so}\left(2d\right)=\left\{ X\in\mathbb{M}_{N\times N}\left(\mathbb{R}\right)|X^{T}=-X\right\} 
\]

\noindent The Lie algebra $\mathfrak{t}$ of a maximal torus in $\mathfrak{so}\left(2d\right)$
is spanned by $d$ two by two block diagonal matrices 
\[
\mathfrak{t}\mathfrak{=\mathrm{Span}_{\mathbb{R}}}\left(\theta_{1},\ldots,\theta_{d}\right)\,,
\]
where $\theta_{i}=-\kb{2i-1}{2i}+\kb{2i}{2i-1}$. The complexification
of $\mathfrak{so}\left(2d\right)$ consists of $2d\times2d$ complex
antisymmetric matrices
\[
\mathfrak{so}\left(2d\right)^{\mathbb{C}}=\mathfrak{so}\left(2d,\mathbb{C}\right)=\left\{ X\in\mathbb{M}_{N\times N}\left(\mathbb{C}\right)|X^{T}=-X\right\} \,.
\]

\noindent The corresponding complexified Cartan subalgebra $\mathfrak{t}^{\mathbb{C}}=\mathfrak{h}$
is given by
\begin{equation}
\mathfrak{h=\mathrm{Span}_{\mathbb{C}}}\left(\theta_{1},\ldots,\theta_{d}\right)\,.
\end{equation}

In the matrix language every operator operators $X\in\mathfrak{h}$
has a block-diagonal form. Each of $d$ block corresponds to one $\theta_{i}$
and has the form
\[
\begin{pmatrix}0 & a_{i}\\
-a_{i} & 0
\end{pmatrix}\,,
\]
for $\alpha_{i}\in\mathbb{C}$. On $\mathfrak{so}\left(2d,\mathbb{C}\right)$
we have an $\mathrm{Spin\mathrm{\left(2d\right)}}$-invariant skew-linear
inner product given by the Hilbert-Schmidt inner product given by:
\[
\left\langle X,Y\right\rangle _{\mathrm{HS}}=\mathrm{tr}\left(X^{\ast}Y\right)\,,
\]
where $\left(\cdot\right)^{\ast}$ denotes the complex conjugate and
$X,Y\in\mathfrak{so}\left(2d,\mathbb{C}\right)$. We hose the following
set of positive simple roots $\Delta\subset\mathfrak{h}$:
\begin{equation}
\Delta=\left\{ \alpha_{i}\right\} _{i=1}^{d},\label{eq:positive simple so(2d)}
\end{equation}
where 
\[
\alpha_{i}=\frac{i}{2}\left(\theta_{i}-\theta_{i+1}\right)\,,\, i=1,\ldots,d-1\,,
\]
 and 
\[
\alpha_{d}=\frac{i}{2}\left(\theta_{i}+\theta_{i+1}\right)\,.
\]

Before we describe the positive and negative root spaces let us introduce
the mapping
\[
\left[\cdot\right]^{kl}:\mathbb{M}\left(2,\mathbb{C}\right)\rightarrow\mathfrak{so}\left(2d,\mathbb{C}\right)\,,\,1\leq k<l\leq d
\]
that associates to any matrix $C\in\mathbb{M}\left(2,\mathbb{C}\right)$
a block matrix $\left[C\right]^{kl}\in\mathfrak{so}\left(2d,\mathbb{C}\right)$
that have the matrix $C$ in the block $\left(k,l\right)$ and the
matrix $-C^{T}$ in the block $\left(k,l\right)$ defined in the following.
We can now describe positive and negative roots corresponding to \eqref{eq:positive simple so(2d)}
\begin{equation}
\mathfrak{n}_{+}=\mathrm{Span}_{\mathbb{C}}\left\{ \left.\left[C\right]^{kl}\right|\, C=\begin{pmatrix}1 & 0\\
-i & 0
\end{pmatrix}\,\text{or}\, C=\begin{pmatrix}0 & i\\
0 & 1
\end{pmatrix},\,1\leq k<l\leq d\right\} \,,\label{eq:positive roots spin}
\end{equation}
\begin{equation}
\mathfrak{n}_{-}=\mathrm{Span}_{\mathbb{C}}\left\{ \left.\left[C\right]^{kl}\right|\, C=\begin{pmatrix}1 & 0\\
i & 0
\end{pmatrix}\,\text{or}\, C=\begin{pmatrix}0 & -i\\
0 & 1
\end{pmatrix},\,1\leq k<l\leq d\right\} \,.\label{eq:negative roots spin}
\end{equation}

We now briefly describe the spinor representations of $\mathrm{Spin}\left(2d\right)$.
The physical significance of these representations will be explained
in subsection \ref{sub:Fermionic-Gaussian-states}. Consider a Fermionic
Fock space with $d$ dimensional single particle Hilbert space,
\[
\mathcal{H}_{\mathrm{Fock}}\left(\mathbb{C}^{d}\right)=\bigoplus_{k=0}^{d}\bigwedge^{k}\left(\mathbb{C}^{d}\right)\,,
\]

On this space we define a set of $2d$ Majorana operators,
\[
c_{2k-1}=a_{k}+a_{k}^{\dagger}\,,\, c_{2k}=i\left(a_{k}-a_{k}^{\dagger}\right)\,,k=1,\ldots,d\,,
\]
where $a_{k},$$a_{k}^{\dagger}$ denote the standard annihilation
and creation fermionic operators. Majorana operators are Hermitian,
traceless and satisfy the following anticommutation relations
\begin{equation}
\left\{ c_{k},\, c_{l}\right\} =2\delta_{kl}\,.\label{eq:anticommutation majorana1}
\end{equation}
We define the spinor representation $\pi_{s}$ of $\mathfrak{so}\left(2d\right)$
by defining it on a standard basis of $\mathfrak{so}\left(2d\right)$
given by the operators
\[
L_{kl}=\kb kl-\kb lk\,,\,,\,(1\leq k<l\leq2d\,.
\]
We set
\begin{equation}
\pi_{s}\left(L_{kl}\right)=\frac{1}{2}c_{k}c_{l}\,,\,(1\leq k<l\leq2d)\,.\label{eq:spin representation algebra}
\end{equation}
It turns out that the above mapping actually defines a representation
of $\mathfrak{so}\left(2d\right)$. Consequently, by the virtue of
Fact \ref{lifting of representation}, the representation $\pi_{s}$
induces a representation $\Pi_{s}:\mathrm{Spin}\left(2d\right)\rightarrow\mathrm{U}\left(\mathcal{H}_{\mathrm{Fock}}\left(\mathbb{C}^{d}\right)\right)$
of group $\mathrm{Spin}\left(2d\right)$. The representation $\Pi_{s}$
encodes exactly the parity-preserving Bogolyubov transformations on
$\mathcal{H}_{\mathrm{Fock}}\left(\mathbb{C}^{d}\right)$ (c.f. Subsection
\ref{sub:Fermionic-Gaussian-states} ). Application of \eqref{eq:spin representation algebra}
to operators $\theta_{i}$ forming the Cartan subalgebra $\mathfrak{h}$
shows that
\begin{equation}
\pi_{s}\left(\theta_{i}\right)=\frac{i}{2}\left(2\hat{n}_{k}-\mathbb{I}\right)\,,\label{eq:torous repres}
\end{equation}
where $\hat{n}_{k}$ is the occupation mode operator associated to
mode $k$ and $\mathbb{I}$ is the identity operator on $\mathcal{H}_{\mathrm{Fock}}\left(\mathbb{C}^{d}\right)$.
Consequently, the standard Fock states $\ket{n_{1},\ldots,n_{d}}$
(see Section \ref{sub:Fermionic-Gaussian-states} for the explanation
of this notation) are weight vectors of $\mathfrak{so}\left(2d\right)$.
Applying again \eqref{eq:spin representation algebra} to \eqref{eq:positive roots spin}
we get 
\begin{equation}
\pi_{s}\left(\mathfrak{n}_{+}\right)=\mathrm{Span}_{\mathbb{C}}\left\{ \left\{ \left.a_{k}a_{l}\right|1\leq k<l\leq d\right\} ,\,\left\{ \left.a_{k}a_{l}^{\dagger}\right|1\leq k<l\leq d\right\} \right\} \,.\label{eq:positive root vectors spinor rep}
\end{equation}
From the above equation an the characterization of the highest weight
vectors (see Fact \ref{Highest-weight-theorem}) it follows that in
$\mathcal{H}_{\mathrm{Fock}}\left(\mathbb{C}^{d}\right)$ we have
two (up to a scalar factor) highest weight vectors
\[
\ket{\psi_{0}^{+}}=\ket 0\,\text{(the Fock vacuum) }\,,\,\ket{\psi_{0}^{-}}=\ket{0,\ldots,0,1}\,.
\]
The irreducible representations of $\mathrm{Spin}\left(2d\right)$
corresponding to these highest weight vectors%
\footnote{We define $\mathcal{H}_{\mathrm{Fock}}^{+}\left(\mathbb{C}^{d}\right)$
as a subspace of $\mathcal{H}_{\mathrm{Fock}}\left(\mathbb{C}^{d}\right)$
spanned by states with even number of excitations. Analogously, $\mathcal{H}_{\mathrm{Fock}}^{+-}\left(\mathbb{C}^{d}\right)$
is a subspace of $\mathcal{H}_{\mathrm{Fock}}\left(\mathbb{C}^{d}\right)$
spanned by states with odd number of excitations %
} are
\[
\Pi_{s}^{+}:\mathrm{Spin}\left(2d\right)\rightarrow\mathrm{U}\left(\mathcal{H}_{\mathrm{Fock}}^{+}\left(\mathbb{C}^{d}\right)\right)\,
\]
and
\[
\Pi_{s}^{-}:\mathrm{Spin}\left(2d\right)\rightarrow\mathrm{U}\left(\mathcal{H}_{\mathrm{Fock}}^{-}\left(\mathbb{C}^{d}\right)\right)
\]
 respectively.

\chapter{Multilinear criteria for detection of general correlations for pure
states \label{chap:Multilinear-criteria-for-pure-states}}

In this chapter we present a unified method for describing various
classes of pure states that define physically interesting types of
quantum correlations. The scheme, in which we define correlations
for mixed states, has already been discussed in Chapter \ref{chap:Introduction}
but we present it here briefly for completeness. The general idea
is that we start from the class of ``not-correlated'' pure states
$\mathcal{M}\subset\mathcal{D}_{1}\left(\mathcal{H}\right)$ and define
non-correlated (with respect to the choice of $\mathcal{M}$) mixed
states as states that can be expressed as convex mixtures of states
belonging to $\mathcal{M}$. Mathematically this corresponds to taking
the convex hull, $\mathcal{M}^{c}$, in the set of all density matrices
$\mathcal{D}\left(\mathcal{M}\right)$. Consequently, correlated states
are defined as all mixed states that are outside $\mathcal{M}^{c}$. 

Clearly, the so-defined concept of correlations depends upon the choice
of the considered class $\mathcal{M}$. As we will see, with the appropriate
choice of $\mathcal{M}$ we can capture many interesting classes of
correlations in multipartite quantum systems. In the current Chapter
we will address the following problem.
\begin{problem}
\label{problem: detection pure}Let $\mathcal{M}\subset\mathcal{D}_{1}\left(\mathcal{H}\right)$
be a class of non-correlated pure states. Given a pure state $\kb{\psi}{\psi}\in\mathcal{D}_{1}\left(\mathcal{H}\right)$
decide whether it belongs to $\mathcal{M}$.
\end{problem}
Let us illustrate the above problem on a simple example. Consider
a system of two qbits, $\mathcal{H}=\mathbb{C}^{2}\otimes\mathbb{C}^{2}$,
and let $\mathcal{M}$ consists of separable states. Let $\ket +=\frac{1}{\sqrt{2}}\left(\ket 0+\ket 1\right)$,
where $\ket 0,\ket 1$ form a standard basis of $\mathbb{C}^{2}.$
Obviously a vector $\ket{\psi}=\ket +\ket +$ defines a separable
state. However, when we express it in the standard basis $\mathbb{C}^{2}\otimes\mathbb{C}^{2}$
we get
\[
\ket{\psi}=\frac{1}{2}\left(\ket 0\ket 0+\ket 0\ket 1+\ket 1\ket 0+\ket 1\ket 1\right)\,,
\]
and it is no longer obvious that $\kb{\psi}{\psi}$ is separable%
\footnote{Obviously, theory of quantum information offers us tools (Schmidt
decomposition, reduction to one-body density matrices etc.) with the
help of which we can solve this problem. However in general the class
$\mathcal{M}$ can be very different from the class of separable states. %
}. Our aim is to give a basis-independent criterion for deciding whether
an arbitrary pure state $\kb{\psi}{\psi}$ belongs to $\mathcal{M}$.
Obviously, the solution to Problem \ref{problem: detection pure}
depends on both the class $\mathcal{M}$ and the relevant Hilbert
space $\mathcal{H}$. In the present chapter we solve Problem \ref{problem: detection pure}
by describing the physically-interesting classes of pure states as
null sets of some real homogenous polynomials defined on the relevant
Hilbert spaces $\mathcal{H}$. This approach, at the first glance,
seems to be rather abstract and high-level. Nevertheless, in latter
chapters of the thesis we will show that such a characterization can
be used to derive non-trivial criteria for detection of the correlations
for general mixed states. 

The chapter is organized as follows. In Section \ref{sec:semisimple-quadratic-characterisation}
we will focus on the situation when the class of pure states $\mathcal{M}$
consists of ``generalized coherent states'' of a compact simply-connected
Lie group $K$ \citep{GenCohPer,CoherentMathPhys,MaxWeightCoh,Kus2009}.
In this case it is possible to characterize the class $\mathcal{M}$
as the zero set of a single $K$-invariant quartic (quadratic in the
density matrix) polynomial. Many interesting types of correlations
can be defined in this way. Examples include: entanglement of distinguishable
particles \citep{Werner1989}, entanglement in bosonic \citep{KilloranPlenio2014,EckertFermions2002}
and fermionic \citep{SchliemannTwoFermions2001,EckertFermions2002,SchliemannFermions2001}
systems, as well as ``non-Gaussian'' correlations in fermionic systems
\citep{powernoisy2013,universalfracBravyi,BravyiKoeningSimul}. In
all these cases the dimension of the relevant Hilbert spaces is finite.
In Section \ref{sec:inf dimension} we discuss entanglement of distinguishable
particles, bosons and fermions, without the restriction on dimensions
of the single particle Hilbert spaces. Also in these cases the relevant
classes of pure states can be characterized as the zero sets of a
single quartic polynomial. In Section \ref{sec:Multilinear-characterization-of-pure}
 we discuss classes of pure states $\mathcal{M}$ that can be described
as zeros of homogenous polynomials of degree higher than two in the
density matrix. Via such homogenous polynomials we can cover all the
classes mentioned above, as well as other classes that cannot be described
by polynomials of too low degree. We focus on classes of pure states
that give rise to a more detailed description of quantum entanglement.
In the bipartite setting we describe polynomials defining the class
of pure states having the Schmidt number at most $n$ \citep{SchmidtNumHoro,SchmidtNumberLew}%
\footnote{We focus on the system of qdits meaning that the relevant Hilbert
space is $\mathcal{H}=\mathbb{C}^{d_{1}}\otimes\mathbb{C}^{d_{2}}$.
The natural number $n$ cannot exceed $d=\mathrm{max}\left\{ d_{1},d_{2}\right\} $%
}. The corresponding class of correlated states consists of mixed states
with Schmidt number bigger then $n$. In the multipartite setting
we deal with the class of ``absolutely separable'' pure states,
i.e. states that remain separable in all possible bipartitions of
the composite Hilbert space \citep{EntantHoro,Guehne2009}. States
that cannot be written as convex combinations of ``absolutely separable''
pure states are called genuinely multipartite entangled. The chapter
concludes with Section \ref{sec:pure Discussion-and-open} where we
summarize the obtained results and present some open problems.

The content of this chapter can be treated as an extension of works
\citep{GenCohPer}, \citep{Kotowski2010} and \citep{Kus2009} where
authors describe classes of non-correlated pure states via a single
quartic (quadratic in the density matrix) $K$-invariant polynomials.
The presented approach has the following advantages:
\begin{itemize}
\item Simple and computable description of polynomials describing the relevant
classes of \linebreak{}
``non-correlated''~pure states discussed in \citep{GenCohPer},
\citep{Kus2009} and \citep{Kotowski2010}.
\item Treatment of broader classes of correlations than in cases described
in \citep{GenCohPer}, \citep{Kotowski2010} and \citep{Kus2009}
(e.g., fermionic Gaussian states).
\item It uses the higher order polynomials in order to describe more complicated
types of entanglement or correlations.
\end{itemize}
The results presented in Sections \ref{sec:semisimple-quadratic-characterisation}
and \ref{sec:inf dimension} constitute parts of the papers: \citep{detection2012,UniversalFramework2013}.
The results presented in Subsection \ref{sec:Multilinear-characterization-of-pure}
will contribute to the forthcoming paper \citep{oszman2014a}. 

Finally, let us mention that during the discussion of polynomial description
of subsequent classes of states $\mathcal{M}\subset\mathcal{D}_{1}\left(\mathcal{H}\right)$
(and thus, the corresponding types of correlations) we will introduce
the notation and terminology that will be used in the remaining chapters
of the thesis.

\section{Quadratic characterization of coherent states of compact simply-connected
Lie groups\label{sec:semisimple-quadratic-characterisation}}

In this section we discuss the Perelomov generalized coherent states
of a compact simply-connected Lie group $K$ \citep{GenCohPer,CoherentMathPhys,Kus2009}.
We first show, using facts from representation theory, that in this
situation the set of coherent states $\mathcal{M}$ is the zero set
of a single $K$-invariant quartic (quadratic in the density matrix)
polynomial on the relevant Hilbert space $\mathcal{H}$. In other
words the following holds,
\begin{equation}
\kb{\psi}{\psi}\in\mathcal{M}\,\Longleftrightarrow\mbox{\ensuremath{\bra{\psi}\bra{\psi}}}A\ket{\psi}\ket{\psi}=0\,,\label{eq:criterion ver 1}
\end{equation}
where $A$ is the $K$-invariant orthonormal projector onto some subspace
of $\mathcal{H}\otimes\mathcal{H}$ %
\footnote{Expression $\mbox{\ensuremath{\bra{\psi}\bra{\psi}}}A\ket{\psi}\ket{\psi}=\mathrm{tr}\left(\kb{\psi}{\psi}\otimes\kb{\psi}{\psi}\, A\right)$
is to be understood as a degree four real polynomial on the Hilbert
space $\mathcal{H}$ (or the degree two polynomial in the density
matrix $\kb{\psi}{\psi}$)%
}. In the latter part of the section we will present the algebraic
form of the operator $A$ for four physically relevant classes of
coherent states: product states of distinguishable particles, product
bosonic states, Slater determinants (in systems with fixed number
of fermions) and fermionic Gaussian states. 

\begin{figure}[h]
\centering{}\includegraphics[width=8cm]{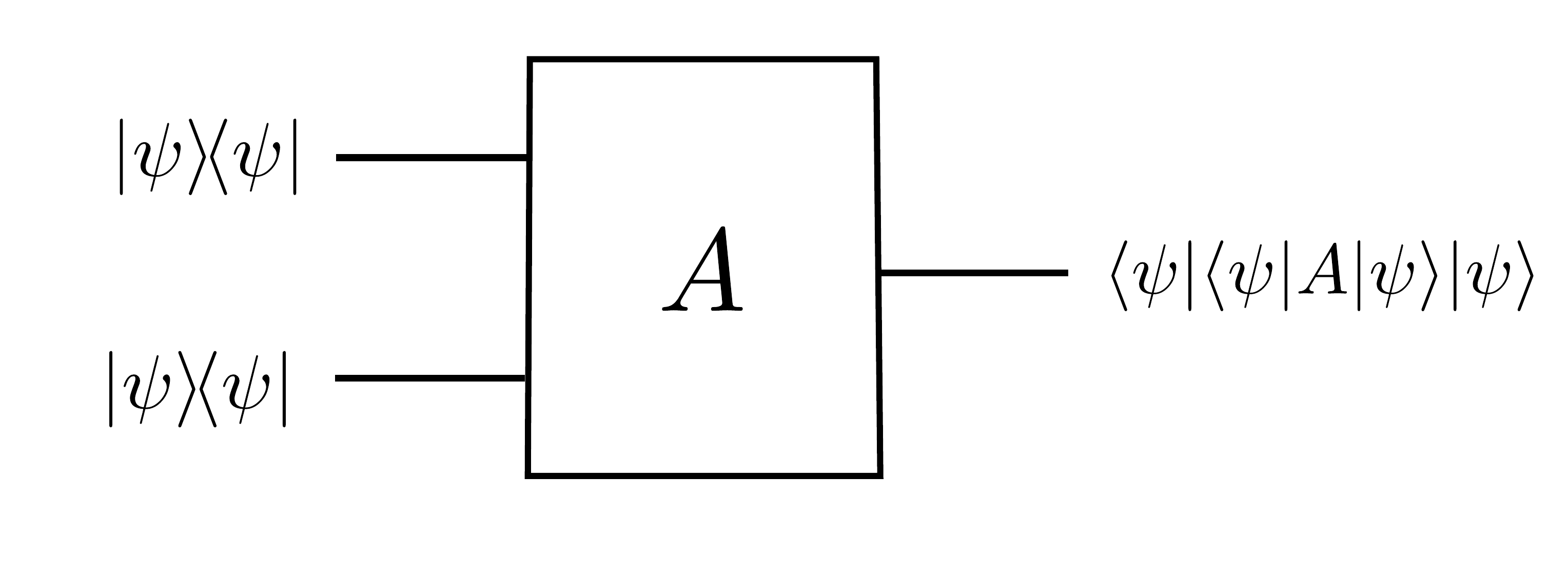}\protect\caption{\label{fig:Pictorial-representation-of-Atwo}Pictorial representation
of the expression $\mbox{\ensuremath{\protect\bra{\psi}\protect\bra{\psi}}}A\protect\ket{\psi}\protect\ket{\psi}$
defining, via Eq. \eqref{eq:criterion ver 1}, class of generalized
coherent states $\mathcal{M}$.}
\end{figure}

Let us define precisely the concept of generalized coherent states
for compact simply-connected Lie groups \citep{GenCohPer,MaxWeightCoh}.
Let $K$ be a compact simply-connected Lie group and let
\begin{equation}
\mbox{\ensuremath{\Pi}:}K\rightarrow U\left(\mathcal{H}^{\lambda_{0}}\right),\, k\rightarrow\Pi\left(k\right)\label{eq:RepCompact}
\end{equation}
 be an irreducible representation of $K$ characterized by the highest
weight $\lambda_{0}$. Due to the compactness of $K$ (see Section
\ref{sec:Rep theory of semisimple}) the dimension of $\mathcal{H}^{\lambda_{0}}$
is finite. Let $\ket{\psi_{0}}$ be a normalized highest weight vector
in $\mathcal{H}^{\lambda_{0}}$. The set of the generalized coherent
states, $\mathcal{M}_{\lambda_{0},K}$ is defined as the orbit of
$K$ through the state $\kb{\psi_{0}}{\psi_{0}}\in\mathcal{D}_{1}\left(\mathcal{H}^{\lambda_{0}}\right)$,
\begin{equation}
\mathcal{M}_{\lambda_{0},K}=\left\{ \Pi\left(k\right)\kb{\psi_{0}}{\psi_{0}}\Pi\left(k\right)^{\dagger}\,|\, k\in K\right\} \,.\label{eq:GenCohStates}
\end{equation}
Where we have used the subscripts $K$ and $\lambda_{0}$ to express
the fact that the set of coherent states depends both on the group
$K$ and on its irreducible representation (labeled by $\lambda_{0}$).
In what follows we will usually drop these subscripts for the sake
of simplicity (we will be working with the fixed representation $\mathcal{H}^{\lambda_{0}}$
of a fixed group $K$).

There exists a simple, purely algebraic characterization of the set
$\mathcal{M}$, given by Liechtenstein \citep{Lichtenstein1982} which
we now briefly describe. Let $\pi=\Pi_{\ast}$ be the induced representation
of the Lie algebra $\mathfrak{k}$ in $\mathcal{H}^{\lambda_{0}}$,
\[
\mbox{\ensuremath{\pi}:}\mathfrak{k}\rightarrow i\cdot\mathrm{Herm}\left(\mathcal{H}^{\lambda_{0}}\right),\, X\rightarrow\pi\left(X\right)\,.
\]
Let $\mathfrak{h}$ be the Cartan subalgebra of $\mathfrak{k}^{\mathbb{C}}$,
i.e. the complexification of some maximal torus $\mathfrak{t}\subset\mathfrak{k}$
(see Section \ref{sec:Rep theory of semisimple} for details). The
representations $\Pi$ and $\pi$ can be promoted, by the usual construction
(see Section \ref{sec:Rep theory of semisimple}), to the representation
on $\mathcal{H}^{\lambda_{0}}\otimes\mathcal{H}^{\lambda_{0}}$,
\begin{gather}
\Pi\otimes\Pi:K\rightarrow\mathrm{U}\left(\mathcal{H}^{\lambda_{0}}\otimes\mathcal{H}^{\lambda_{0}}\right),\, k\rightarrow\Pi\left(k\right)\otimes\Pi\left(k\right),\,\label{eq:rep group two copies}\\
\pi\otimes\pi:\mathfrak{k}\rightarrow i\cdot\mathrm{Herm}\left(\mathcal{H}^{\lambda_{0}}\otimes\mathcal{H}^{\lambda_{0}}\right),\, X\rightarrow\pi\left(X\right)\otimes\mathbb{I}+\mathbb{I}\otimes\pi\left(X\right)\,.\label{eq:rep algebra two copies}
\end{gather}

\begin{prop}
\label{prop: two copies} The representation $\Pi\otimes\Pi$ is,
in general, reducible. The decomposition of $\mathcal{H}^{\lambda_{0}}\otimes\mathcal{H}^{\lambda_{0}}$
onto irreducible components reads,
\begin{equation}
\mathcal{H}^{\lambda_{0}}\otimes\mathcal{H}^{\lambda_{0}}=\mathcal{H}^{2\lambda_{0}}\oplus\bigoplus_{\beta<2\lambda_{0}}\mathcal{H}^{\beta}\,,\label{eq:decomposition}
\end{equation}
where $\mathcal{H}^{2\lambda_{0}}$ denotes the representation characterized
by the highest weight $2\lambda_{0}$ and $\bigoplus_{\beta<2\lambda_{0}}\mathcal{H}^{\beta}$
denotes the direct sum of all other irreducible subrepresentations
of the group $K$ in $\mathcal{H}^{\lambda_{0}}\otimes\mathcal{H}^{\lambda_{0}}$. \end{prop}
\begin{proof}
It is enough to show that $2\lambda_{0}$ is a highest weight and
that in $\mathcal{H}^{\lambda_{0}}\otimes\mathcal{H}^{\lambda_{0}}$
there exist only one irreducible representation of $K$ characterized
by the highest weight $2\lambda_{0}$. This however is equivalent
to showing that in $\mathcal{H}^{\lambda_{0}}\otimes\mathcal{H}^{\lambda_{0}}$
there exist only one weight vector corresponding to the highest weight
$2\lambda_{0}$. Let $\mathrm{wt}\left(\mathcal{H}^{\lambda_{0}},\Pi\right)$
denotes the set of weights corresponding to weight vectors in $\mathcal{H}^{\lambda_{0}}$.
By the definition of the tensor product of representation we have
\begin{equation}
\mathrm{wt}\left(\mathcal{H}^{\lambda_{0}}\otimes\mathcal{H}^{\lambda_{0}},\Pi\otimes\Pi\right)=\left\{ \lambda+\mu\,\left|\lambda,\mu\in\mathrm{wt}\left(\mathcal{H}^{\lambda_{0}},\Pi\right)\right.\right\} \,.\label{eq:possible weights}
\end{equation}
From \eqref{eq:possible weights} we get that for arbitrary $\tilde{\lambda}\in\mathrm{wt}\left(\mathcal{H}^{\lambda_{0}}\otimes\mathcal{H}^{\lambda_{0}},\Pi\otimes\Pi\right)$
and from the definition of the highest weight \eqref{eq:highest weight}
applied to the highest weight $\lambda_{0}$ we get
\[
2\lambda_{0}-\tilde{\lambda}=\sum_{i=1}^{r}x_{i}\alpha_{i}\,,x_{i}\geq0\,,
\]
where $\left\{ \alpha_{i}\right\} _{i=1}^{r}=\Delta$. Consequently,
we by the theorem of the highest weight (c.f. Fact \ref{Highest-weight-theorem})
and the fact that every finite-dimensional representation of $K$
is completely reducible (c.f. Eq.\eqref{eq:completelly reducible})
we get that $2\lambda_{0}$ is a highest weight. There is only one
weight vector $\ket{\psi}$ in $\mathcal{H}^{\lambda_{0}}\otimes\mathcal{H}^{\lambda_{0}}$
that characterized by the weight $2\lambda_{0}$. Note that for $\lambda,\mu\in\mathrm{wt}\left(\mathcal{H}^{\lambda_{0}},\Pi\right)$
we have the following implication,
\[
2\lambda_{0}-\mu-\lambda=0\,\Longleftrightarrow\,\lambda=\mu=\lambda_{0}\,.
\]
which is a consequence of the fact that $\lambda_{0}$ is the unique
highest weight in $\mathcal{H}^{\lambda_{0}}$. As all weight vectors
in $\mathcal{H}^{\lambda_{0}}\otimes\mathcal{H}^{\lambda_{0}}$ are
linear combinations of simple tensors $\ket{\psi_{\lambda}}\ket{\psi_{\mu}}$,
where $\ket{\psi_{\lambda}},\ket{\psi_{\mu}}\in\mathcal{H}^{\lambda_{0}}$
are weight vectors corresponding to weights $\lambda,\mu\in\mathrm{wt}\left(\mathcal{H}^{\lambda_{0}},\Pi\right)$
we conclude that the only weight vector characterized by highest weight
vector is $\ket{\psi}=\ket{\psi_{0}}\ket{\psi_{0}}$.
\end{proof}
Let $\left\{ X_{i}\right\} _{i=1}^{\mathrm{dim}\left(\mathfrak{k}\right)}$
be the orthonormal (with respect to the chosen $\mathrm{Ad}_{K}$-invariant
inner product on $\mathfrak{k}$) basis of the Lie algebra $\mathfrak{k}$
and let $C_{2}=-\sum_{i=1}^{\mathrm{dim}\left(\mathfrak{k}\right)}X_{i}^{2}$
be the second order Casimir operator of a semisimple Lie algebra $\mathfrak{k}$
(see Eq. \ref{eq:casimir definition} ). Let $L_{2}$ denote the representation%
\footnote{As it was noted in the Section \ref{sec:Rep theory of semisimple}
every representation of Lie algebra $\mathfrak{k}$ induces the representation
of the universal enveloping algebra $\mathfrak{U}\left(\mathfrak{k}\right)$
which contains $C_{2}$.%
} of $C_{2}$ in $\mathcal{H}^{\lambda_{0}}\otimes\mathcal{H}^{\lambda_{0}}$
\begin{equation}
L_{2}:\mathcal{H}^{\lambda_{0}}\otimes\mathcal{H}^{\lambda_{0}}\rightarrow\mathcal{H}^{\lambda_{0}}\otimes\mathcal{H}^{\lambda_{0}}\,,\, L_{2}=-\sum_{i=1}^{N}\left(\pi\left(X_{i}\right)\otimes\mathbb{I}+\mathbb{I}\otimes\pi\left(X_{i}\right)\right)^{2}\,.\label{eq:second order cas rep}
\end{equation}
The operator $L_{2}$ commutes with $\Pi(k)\otimes\Pi(k)$ for each
$k\in K$ so, by the Schur lemma, it is proportional to the identity
operator on each irreducible component $\mathcal{H}^{\beta}$ in the
decomposition \eqref{eq:decomposition}. According to \eqref{eq:casimir irrep value}
we have 
\begin{equation}
\left.L_{2}\right|_{\mathcal{H}^{\beta}}=\left(\beta,\,\beta+2\delta\right)\left.\mathbb{I}\right|_{\mathcal{H}^{\beta}}\,,\label{eq:casimir irrep}
\end{equation}
where $\left(\cdot,\cdot\right):\mathfrak{h}\times\mathfrak{h}\rightarrow\mathbb{C}$
is the $\mathrm{Ad}_{K}$-invariant inner product on $\mathfrak{h}$
(c.f. Subsection \eqref{sub:Structural-theory-of}).
\begin{fact}
(\citep{Lichtenstein1982}) \label{Lichteinstein charact}The set
of generalized coherent states $\mathcal{M}$ can be characterized
by the following condition
\begin{equation}
\kb{\psi}{\psi}\in\mathcal{M}\,\Longleftrightarrow L_{2}\ket{\psi}\ket{\psi}=\left(2\lambda_{0},\,2\lambda_{0}+2\delta\right)\ket{\psi}\ket{\psi}\,.\label{eq:Lichneinstein1}
\end{equation}

\end{fact}
Using \eqref{eq:casimir irrep} and \eqref{eq:Lichneinstein1} we
get the following proposition.
\begin{prop}
\label{prop:proj two}The set of generalized coherent states $\mathcal{M}$
can be characterized by the following condition
\begin{equation}
\kb{\psi}{\psi}\in\mathcal{M}\,\Longleftrightarrow\,\bra{\psi}\bra{\psi}\left(\mathbb{P}^{\mathrm{sym}}-\mathbb{P}^{2\lambda_{0}}\right)\ket{\psi}\ket{\psi}=0\,,\label{eq:polynomial characterisation}
\end{equation}
where $\mathbb{P}^{\mathrm{sym}}$ is a projector onto the two-fold
symmetric tensor power of $\mathcal{H}^{\lambda}$ and the operator
$\mathbb{P}^{2\lambda_{0}}:\mathcal{H}^{\lambda_{0}}\otimes\mathcal{H}^{\lambda_{0}}\rightarrow\mathcal{H}^{\lambda_{0}}\otimes\mathcal{H}^{\lambda_{0}}$
is the orthonormal projector onto the irreducible representation $\mathcal{H}^{2\lambda_{0}}\subset\mathcal{H}^{\lambda_{0}}\otimes\mathcal{H}^{\lambda_{0}}$. 
\end{prop}
It follows that the operator $A=\mathbb{P}^{\mathrm{sym}}-\mathbb{P}^{2\lambda_{0}}$
is invariant under the action of $K$,
\begin{equation}
\Pi(k)\otimes\Pi(k)A\Pi^{\dagger}(k)\otimes\Pi^{\dagger}(k)=A\,,\,\text{for all}\, k\in K\,.\label{eq:invariance A}
\end{equation}
Consequently, the function
\[
C^{2}:\mathcal{D}_{1}\left(\mathcal{H}^{\lambda_{0}}\right)\ni\kb{\psi}{\psi}\rightarrow\bra{\psi}\bra{\psi}A\ket{\psi}\ket{\psi}=\mathrm{tr}\left(\left[\kb{\psi}{\psi}^{\otimes2}\right]A\right)\in\mathbb{R}\,,
\]
is a polynomial invariant%
\footnote{Function $C^{2}$ is an example of a polynomial invariant of the action
of the group $K$ in $\mathcal{D}\left(\mathcal{H}^{\lambda_{0}}\right)$.
For the comprehensive study of application of invariant polynomials
in the context of entanglement see \citep{Vrana2011}. %
} under the action of the group $K$. From \eqref{eq:polynomial characterisation}
and \eqref{eq:invariance A} we easily conclude the criterion \eqref{eq:criterion ver 1}
presented at the beginning of this Section. 

We present here, following \citep{Klyachko2008}, an alternative point
of view on the Fact \ref{Lichteinstein charact}. This perspective
will sheds some light on the physical interpretation of generalized
coherent states of compact simply-connected Lie groups. Racall that
due to the fact that representation $\Pi$ is unitary, operators $\pi\left(X\right)$
($X\in\mathfrak{k}$) are anti-Hermitian and do not represent physical
observables. Let us therefore introduce the following notation%
\footnote{It is necessary to introduce the additional imaginary unit factor
due to the fact that throughout the thesis we decided to use a ``mathematical''
definition of a Lie algebra.%
}:$\bar{\pi}\left(X\right)=i\pi(X)$. Let us consider now the generalized
variance, $\mathrm{Var}_{K}\left(\kb{\psi}{\psi}\right)$, of a state
$\kb{\psi}{\psi}\in\mathcal{D}_{1}\left(\mathcal{H}^{\lambda_{0}}\right)$
defined by 
\begin{equation}
\mathrm{Var}_{K}\left(\kb{\psi}{\psi}\right)=\sum_{i=1}^{\mathrm{dim}\left(\mathfrak{k}\right)}\left(\bk{\psi}{\bar{\pi}\left(X_{i}\right)^{2}\psi}-\bk{\psi}{\bar{\pi}\left(X_{i}\right)\psi}^{2}\right),\label{eq:definition variance}
\end{equation}
where elements of $X_{i}\in\mathfrak{k}$ are defined as in \eqref{eq:second order cas rep}.
One can think of $\mathrm{Var}_{K}\left(\kb{\psi}{\psi}\right)$ as
of the generalized measure of uncertainty associated to a state $\kb{\psi}{\psi}$while
measuring Hermitian operators which are representatives of the orthonormal
basis of $\mathfrak{k}$. The variance $\mathrm{Var}_{K}\left(\kb{\psi}{\psi}\right)$
can be expressed via the expectation value of the Casimir $L_{2}$,
\begin{align}
\mathrm{Var}_{K}\left(\kb{\psi}{\psi}\right) & =\left(\lambda_{0},\lambda_{0}+2\delta\right)-\sum_{i=1}^{\mathrm{dim}\left(\mathfrak{k}\right)}\bk{\psi}{\bar{\pi}\left(X_{i}\right)\psi}^{2}\,,\label{eq:var1}\\
 & =2\left(\lambda_{0},\lambda_{0}+2\delta\right)-\frac{1}{2}\bra{\psi}\bra{\psi}L_{2}\ket{\psi}\ket{\psi}\,.\label{eq:var2}
\end{align}
In the above the first line follows from the fact that $\sum_{i=1}^{\mathrm{dim}\left(\mathfrak{k}\right)}\bar{\pi}\left(X_{i}\right)^{2}$
is the representation of the Casimir in $\mathcal{H}^{\lambda_{0}}$.
The second line is a consequence of \eqref{eq:second order cas rep}.
From \eqref{eq:var2} we see that $\mathrm{Var}_{K}\left(\kb{\psi}{\psi}\right)$
is invariant under the action of the group $K$,
\[
\mathrm{Var}_{K}\left(\kb{\psi}{\psi}\right)=\mathrm{Var}_{K}\left(\Pi\left(k\right)\kb{\psi}{\psi}\Pi\left(k\right)^{\dagger}\right)\,.
\]
Fact \ref{Lichteinstein charact} can be now deduced from the following
fact which characterizes coherent states as states that minimize the
variance $\mathrm{Var}_{K}\left(\kb{\psi}{\psi}\right)$.
\begin{fact}
Let $\mathrm{Var}_{K}\left(\kb{\psi}{\psi}\right)$ be the generalized
$K$-invariant variance defined above. The following equality holds
\begin{equation}
\min_{\kb{\psi}{\psi}\in\mathcal{D}_{1}\left(\mathcal{H}^{\lambda_{0}}\right)}\mathrm{Var}_{K}\left(\kb{\psi}{\psi}\right)=\left(\lambda_{0},2\delta\right)\,.\label{eq:minimal variance}
\end{equation}
Moreover, $\mathrm{Var}_{K}\left(\kb{\psi}{\psi}\right)=\left(\lambda_{0},2\delta\right)$
if and only if $\kb{\psi}{\psi}=\Pi\left(k\right)\kb{\psi_{0}}{\psi_{0}}\Pi\left(k\right)^{\dagger}$,
for $k\in K$.
\end{fact}
From the above fact we have
\begin{equation}
\max_{\kb{\psi}{\psi}\in\mathcal{D}_{1}\left(\mathcal{H}^{\lambda_{0}}\right)}\bra{\psi}\bra{\psi}L_{2}\ket{\psi}\ket{\psi}=\left(2\lambda_{0},\,2\lambda_{0}+2\delta\right)\,,\label{eq:maximum casimir}
\end{equation}
and $\bra{\psi}\bra{\psi}L_{2}\ket{\psi}\ket{\psi}=\left(2\lambda_{0},\,2\lambda_{0}+2\delta\right)$
if and only if $\kb{\psi}{\psi}$ is a coherent state. 

We now apply \eqref{eq:polynomial characterisation} to describe,
by a single polynomial condition, four physically-relevant classes
of pure states: 
\begin{itemize}
\item Product states,
\item Symmetric product states
\item Slater determinants
\item Pure fermionic Gaussian states. 
\end{itemize}
Motivation for the study of above listed classes of states was presented
in Section \ref{chap:Introduction}. As we present the polynomial
characterization of these classes we will introduce the necessary
notation and group-theoretical formalism needed for the description
of these classes them (that will be used also in other parts of the
thesis).  Results presenting the polynomial characterization of each
class of states are stated in the form of lemmas whose proofs can
be omitted without affecting the understanding of the rest of the
thesis.

\subsection{Product states (non-entangled pure states)\label{sub:Product-states}}

In the case of $L$ distinguishable particles, the Hilbert space of
states is the tensor product of the Hilbert spaces of single particles
$\mathcal{H}_{1},\ldots,\mathcal{H}_{L}$ of dimensions $\ensuremath{N_{1},\ldots,N_{L}}$,
which we conveniently identify with the complex spaces of the same
dimensions,
\begin{equation}
\mathcal{H}_{d}=\mathcal{H}_{1}\otimes\cdots\otimes\mathcal{H}_{L}=\mathbb{C}^{N_{1}}\otimes\ldots\otimes\mathbb{C}^{N_{L}}=\bigotimes_{i=1}^{L}\mathbb{C}^{N_{i}}\,.\label{eq:dist hilbert space}
\end{equation}
We take the group $K$ as the group of local unitary operations i.e
as the direct product of special unitary groups $SU(N_{k})$, each
acting independently in the respective one-particle space $\mathcal{H}_{k}=\mathbb{C}^{N_{k}}$%
\footnote{As a group of local unitary transformations we could have also taken
can take the group $K'=\mathrm{U}(N_{1})\times\ldots\times\mathrm{U}(N_{L})$.
However, for physical applications the global phase of the wave function
does not play a role and hence the distinction between $K$ and $K'$
is irrelevant.%
}. Therefore we have
\begin{gather}
K=\mathrm{LU}=\mathrm{SU}(N_{1})\times\ldots\times\mathrm{SU}\left(N_{L}\right)=\times_{i=1}^{L}SU(N_{i})\,,\nonumber \\
\Pi_{d}(U_{1}\ldots,U_{L})\ket{\psi_{1}}\otimes\cdots\otimes\ket{\psi_{L}}=U_{1}\ket{\psi_{1}}\otimes\cdots\otimes U_{L}\ket{\psi_{L}}\,.\label{eq:action dist}
\end{gather}
In the notation above we made explicit that the group $K$ acts on
the Hilbert space $\mathcal{H}_{d}$ via the particular representation
$\Pi_{d}$, defined here by its action on simple tensors. Representation
$\Pi_{d}$ is a tensor product of the defining (and thus irreducible)
representations of groups $\mathrm{SU}\left(N_{i}\right)$. From the
Fact \eqref{irreducibility of product} it follows that the representation
$\Pi_{d}$ is actually irreducible. One easily checks that in this
case generalized coherent states are precisely product states,
\begin{equation}
\mathcal{M}_{d}=\left\{ \kb{\psi}{\psi}\in\mathcal{D}_{1}\left(\mathcal{H}_{d}\right)|\,\ket{\psi}=\ket{\psi_{1}}\ket{\psi_{2}}\ldots\ket{\psi_{L}},\,\ket{\psi_{i}}\in\mathcal{H}_{i}\right\} \,.\label{eq:product distinguishable def}
\end{equation}
We now describe explicitly the operator $A$ that defines, via \eqref{eq:polynomial characterisation},
the set $\mathcal{M}_{d}$. Let us first introduce some notation:
\begin{equation}
\mathcal{H}_{d}\otimes\mathcal{H}_{d}=\left(\bigotimes_{i=1}^{i=L}\mathcal{H}_{i}\right)\otimes\left(\bigotimes_{i=1'}^{i=L'}\mathcal{H}_{i}\right)\,,\label{eq:nonation distinguishable}
\end{equation}
where $L=L'$ and we decided to label spaces from the second copy
of the total space with primes in order to avoid ambiguity. 

\begin{figure}[h]
\begin{centering}
\includegraphics[width=8cm]{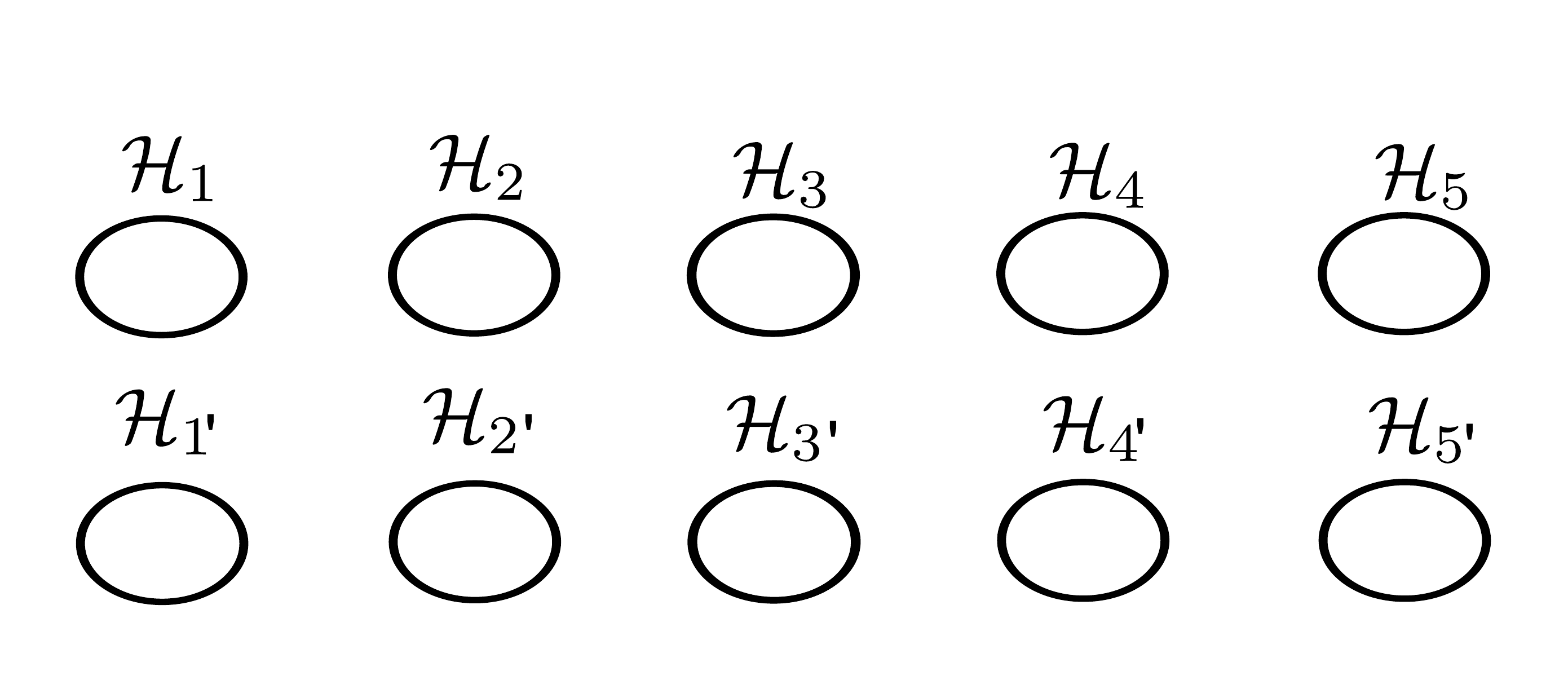}\protect\caption{\label{fig:disting_figure}Graphical illustration of the tensor product
$\mathcal{H}_{d}\otimes\mathcal{H}_{d}$ for $L=5$. Note that $\mathcal{H}_{i}\approx\mathcal{H}_{i'}$
for $i=1,\ldots,L$}

\par\end{centering}

\end{figure}
The action of $K$ on $\mathrm{Sym}^{2}\left(\mathcal{H}_{d}\right)$
is given by the restriction to the symmetric (with respect to the
interchange of copies of $\mathcal{H}_{d}$) tensors of the action
defined on $\mathcal{H}_{d}\otimes\mathcal{H}_{d}$ (Eq. \eqref{eq:rep group two copies}).
Let us also introduce the symmetrization operators $\mathbb{P}_{ii'}^{+}:\mathcal{H}_{d}\otimes\mathcal{H}_{d}\rightarrow\mathcal{H}_{d}\otimes\mathcal{H}_{d}$
that project onto the subspace of $\mathcal{H}_{d}\otimes\mathcal{H}_{d}$
completely symmetric under interchange spaces labeled by $i$ and
$i'$ (the operators of anti-symmetrization $\mathbb{P}_{ii'}^{-}$
are defined analogously). 
\begin{lem}
\label{lemma crit dist part}Under the introduced notation the closed
expression for the projector operator $\mathbb{P}^{2\lambda_{0}}$
(see \eqref{eq:polynomial characterisation}) for the case of $K=\times^{L}\big(\mathrm{SU}(N)\big)$
and $\mathcal{H}^{\lambda_{0}}=\mathcal{H}_{d}$ reads
\begin{equation}
\mathbb{P}^{2\lambda_{0}}=\mathbb{P}_{d}=\mathbb{P}_{11'}^{+}\circ\mathbb{P}_{22'}^{+}\circ\ldots\circ\mathbb{P}_{LL'}^{+}\,.\label{eq:dist p2lambda}
\end{equation}
Consequently, by the virtue of \eqref{eq:polynomial characterisation}
we have the following characterization of the set of product states
\begin{equation}
\kb{\psi}{\psi}\in\mathcal{M}_{dist}\,\Longleftrightarrow C_{dist}^{2}\left(\kb{\psi}{\psi}\right)=\bra{\psi}\bra{\psi}\left(\mathbb{P}^{\mathrm{sym}}-\mathbb{P}_{11'}^{+}\circ\mathbb{P}_{22'}^{+}\circ\ldots\circ\mathbb{P}_{LL'}^{+}\right)\ket{\psi}\ket{\psi}=0\,.\label{eq:crit prod states}
\end{equation}

\end{lem}
\begin{figure}[h]
\centering{}\includegraphics[width=8cm]{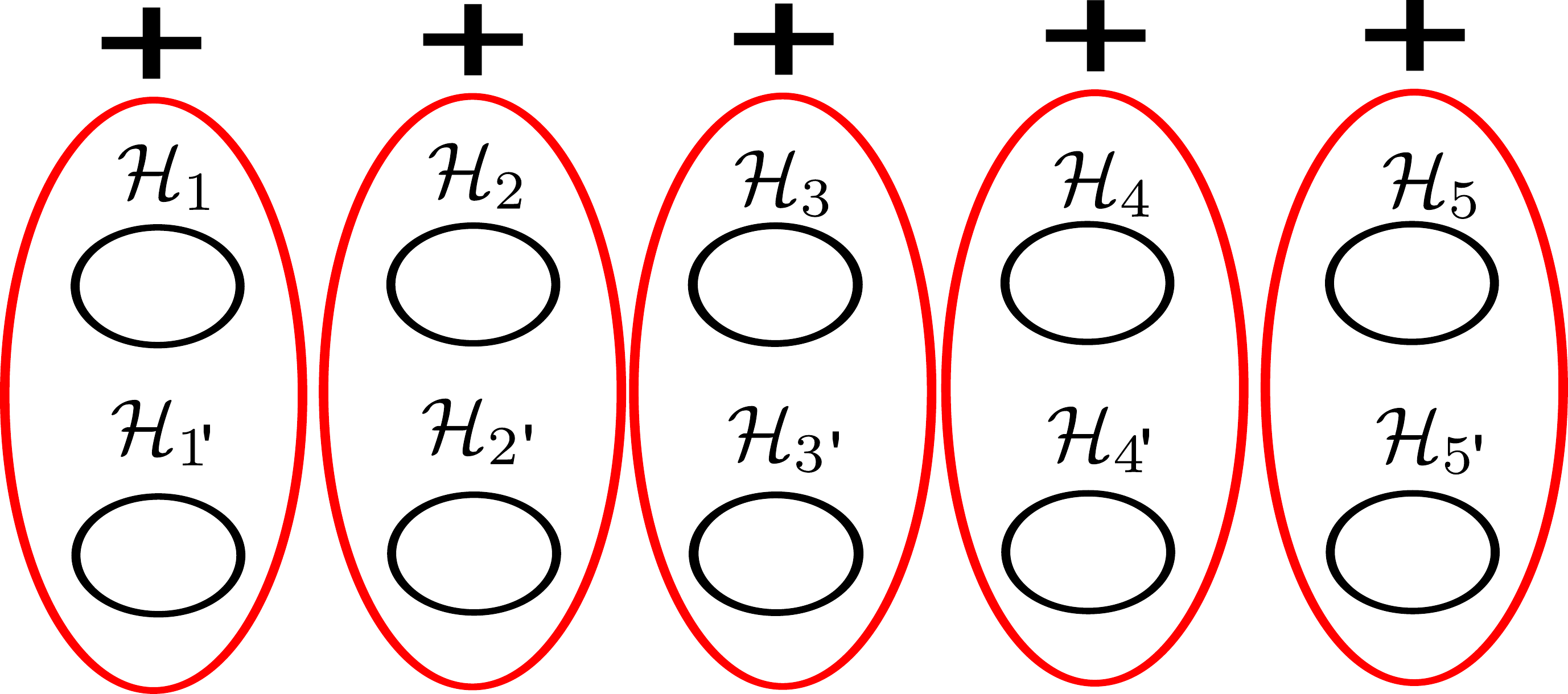}\protect\caption{\label{fig:disting_fig_op}Graphical illustration he action of the
operator $\mathbb{P}_{d}$ in $\mathcal{H}_{d}\otimes\mathcal{H}_{d}$.
Red circles correspond to symmetrizations in the relevant factors
of the space $\mathcal{H}_{d}\otimes\mathcal{H}_{d}$.}
\end{figure}

\begin{proof}[Proof of Lemma \ref{lemma crit dist part}]
The proof of \eqref{eq:dist p2lambda} is the following. First, notice
that for separable $\ket{\psi}$ we have $\mathbb{P}_{d}\ket{\psi}\otimes\ket{\psi}=\ket{\psi}\otimes\ket{\psi}$.
Secondly, we have the equivalence of representations of $K$,
\[
\mathbb{P}_{11'}^{+}\circ\mathbb{P}_{22'}^{+}\circ\ldots\circ\mathbb{P}_{LL'}^{+}\,\left(\mathrm{Sym}^{2}\left(\mathcal{H}_{d}\right)\right)\approx\mathrm{Sym}^{2}\left(\mathcal{H}_{1}\right)\otimes\mathrm{Sym}^{2}\left(\mathcal{H}_{2}\right)\otimes\ldots\otimes\mathrm{Sym}^{2}\left(\mathcal{H}_{L}\right)\,.
\]
Therefore, subspace $\mathbb{P}_{11'}^{+}\circ\mathbb{P}_{22'}^{+}\circ\ldots\circ\mathbb{P}_{LL'}^{+}\,\left(\mathrm{Sym}^{2}\left(\mathcal{H}_{d}\right)\right)$
is an irreducible representation of $K$. Talking into account criterion
\eqref{eq:polynomial characterisation} and the fact that the separable
states are exactly the coherent states of $K$ finishes the proof. 
\end{proof}

\subsection{Symmetric product states (non-entangled pure states of bosons)\label{sub:Symmetric-product-states}}

The Hilbert space describing $L$ bosonic particles in $N$ modes
has the structure $\mathcal{H}^{\lambda_{0}}=\mathcal{H}_{b}=\mathrm{Sym}^{L}\left(\mathcal{H}\right)$,
where $\mathcal{H}\approx\mathbb{C}^{N}$ is the single particle Hilbert
space. The relevant symmetry group is a compact simply-connected Lie
group $K=\mathrm{SU}(N)$ represented on $\mathcal{H}_{b}$ via the
irreducible representation $\Pi_{b}$ describing an arbitrary ``single-particle''
evolution
\begin{equation}
\Pi_{b}(U)\left(\ket{\psi_{1}}\vee\cdots\vee\ket{\psi_{L}}\right)=U\ket{\psi_{1}}\vee\cdots\vee U\ket{\psi_{L}},\label{eq:bos rep}
\end{equation}
where $\ket{\psi_{1}}\vee\cdots\vee\ket{\psi_{L}}$ denotes the (normalized)
projection onto $\mathrm{Sym}^{L}\left(\mathcal{H}\right)$ of the
vector%
\footnote{It is clear that vectors of the form $\ket{\psi_{1}}\vee\cdots\vee\ket{\psi_{L}}$
span $\mathrm{Sym}^{L}\left(\mathcal{H}\right)$.%
} $\ket{\psi_{1}}\otimes\cdots\otimes\ket{\psi_{L}}$. In the notation
for the highest weights of $\mathrm{SU}\left(N\right)$ introduced
in Subsection \eqref{sub:Representation-theory-of}, the representation
$\Pi_{b}$ is labeled by the highest weight $\lambda_{0}=\left(L,0,\ldots,0\right)$.
The set of generalized coherent states coincides with the set of symmetric
product states,
\begin{equation}
\mathcal{M}_{b}=\left\{ \kb{\psi}{\psi}\in\mathcal{D}_{1}\left(\mathcal{H}_{b}\right)\,|\,\ket{\psi}=\ket{\phi}\ket{\phi}\ldots\ket{\phi}\,,\,\ket{\phi}\in\mathbb{C}^{N}\right\} \,.\label{eq:crit prod bos states}
\end{equation}
We proceed with giving the explicit form of the operator $A$ describing
$\mathcal{M}_{b}$. We first embed $\mathrm{Sym}^{L}\left(\mathcal{H}\right)$
in the Hilbert space of $L$ identical distinguishable particles,
\begin{equation}
\mathrm{Sym}^{L}\left(\mathcal{H}\right)\subset\mathcal{H}_{1}\otimes\ldots\otimes\mathcal{H}_{L}\,,\label{eq:encoding bos}
\end{equation}
where $\mathcal{H}_{i}\approx\mathcal{H}$. We have the analogous
embedding of $\mathrm{Sym}^{L}\left(\mathcal{H}\right)\vee\mathrm{Sym}^{L}\left(\mathcal{H}\right)$,
\[
\mathrm{Sym}^{L}\left(\mathcal{H}\right)\vee\mathrm{Sym}^{L}\left(\mathcal{H}\right)\subset\left(\bigotimes_{i=1}^{i=L}\mathcal{H}_{i}\right)\otimes\left(\bigotimes_{i=1'}^{i=L'}\mathcal{H}_{i}\right)=\mathcal{H}_{d}\otimes\mathcal{H}_{d}\,,
\]
where, as before, $L=L'$. Let $\mathbb{P}_{\left\{ 1,\ldots,L\right\} }^{\mathrm{sym}}:\mathcal{H}_{d}\otimes\mathcal{H}_{d}\rightarrow\mathcal{H}_{d}\otimes\mathcal{H}_{d}$
be the projector onto the subspace of $\mathcal{H}_{d}\otimes\mathcal{H}_{d}$
which is completely symmetric with respect to the interchange of spaces
labeled by indices from the set $\left\{ 1,2,\ldots,L\right\} $.
We define $\mathbb{P}_{\left\{ 1',\ldots,L'\right\} }^{\mathrm{sym}}$
in the analogous way. 
\begin{lem}
\label{lema bosons p2}Under the introduced notation a closed expression
for the projector operator $\mathbb{P}^{2\lambda_{0}}$ (see \eqref{eq:polynomial characterisation})
for the case of $K=\mathrm{SU}(N)$ and $\mathcal{H}^{\lambda_{0}}=\mathcal{H}_{b}$
reads
\begin{equation}
\mathbb{P}^{2\lambda_{0}}=\mathbb{P}_{b}=\left(\mathbb{P}_{11'}^{+}\circ\mathbb{P}_{22'}^{+}\circ\ldots\circ\mathbb{P}_{LL'}^{+}\right)\left(\mathbb{P}_{\left\{ 1,\ldots,L\right\} }^{\mathrm{sym}}\circ\mathbb{P}_{\left\{ 1',\ldots,L'\right\} }^{\mathrm{sym}}\right)\,,\label{eq:bos p2lambda}
\end{equation}
The operators $\mathbb{P}_{ii'}^{+}$ are the same as in \eqref{eq:dist p2lambda}
and it is understood that $\,\mathbb{P}^{2\lambda_{0}}$ acts on the
space $\mathcal{H}_{d}\otimes\mathcal{H}_{d}$. By \eqref{eq:polynomial characterisation}
we have the following characterization of the set of product bosonic
states
\begin{equation}
\kb{\psi}{\psi}\in\mathcal{M}_{b}\,\Longleftrightarrow C_{b}^{2}\left(\ket{\psi}\right)=\bra{\psi}\bra{\psi}\left(\mathbb{P}^{\mathrm{sym}}-\mathbb{P}_{b}\right)\ket{\psi}\ket{\psi}=0\,.\label{eq:criterion bos}
\end{equation}

\end{lem}
\noindent Before we prove Lemma \ref{lema bosons p2} let us note
that we may write 
\[
\left.\mathbb{P}_{b}\right|_{\mathrm{Sym}^{L}\left(\mathcal{H}\right)\otimes\mathrm{Sym}^{L}\left(\mathcal{H}\right)}=\mathbb{P}_{11'}^{+}\circ\mathbb{P}_{22'}^{+}\circ\ldots\circ\mathbb{P}_{LL'}^{+}\,,
\]
as for any $\ket{\Psi}\in\mathrm{Sym}^{L}\left(\mathcal{H}\right)\otimes\mathrm{Sym}^{L}\left(\mathcal{H}\right)$
we have $\left(\mathbb{P}_{\left\{ 1,\ldots,L\right\} }^{\mathrm{sym}}\otimes\mathbb{P}_{\left\{ 1',\ldots,L'\right\} }^{\mathrm{sym}}\right)\ket{\Psi}=\ket{\Psi}$.

\begin{figure}[h]
\centering{}\includegraphics[width=8cm]{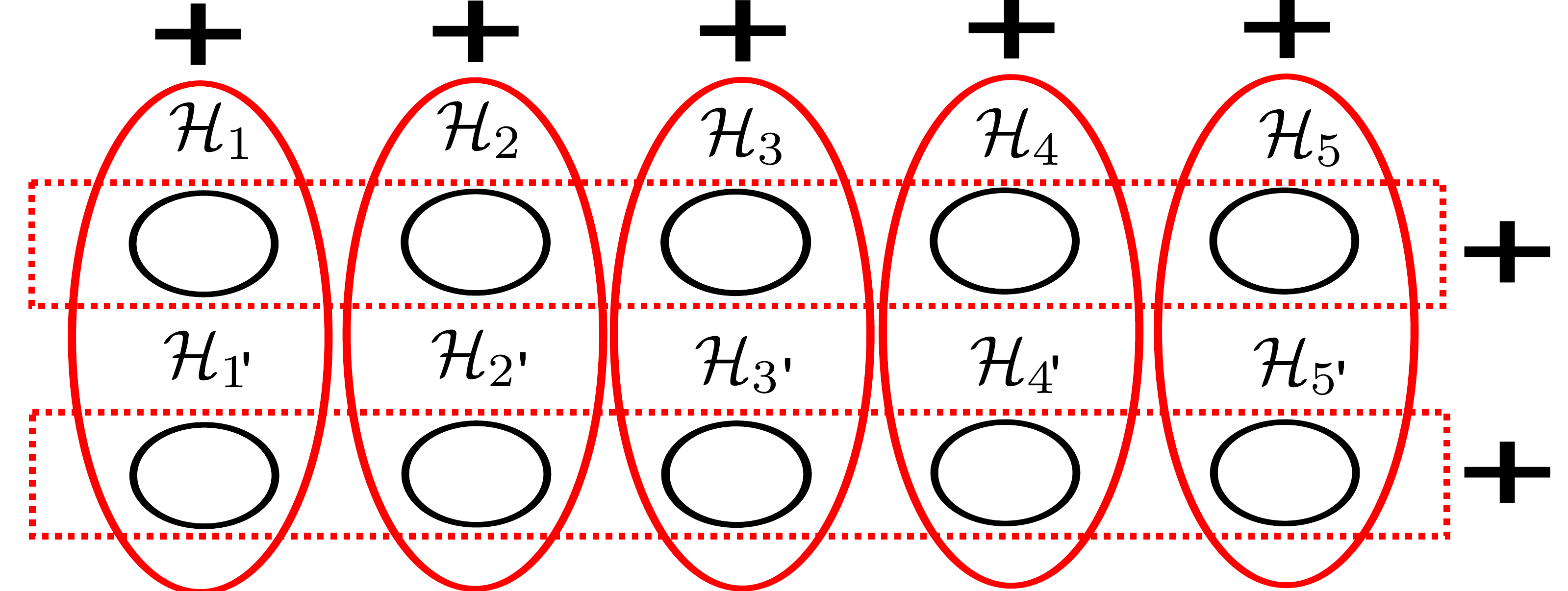}\protect\caption{\label{fig:disting_bos_op-1}Graphical illustration of the action
of the operator $\mathbb{P}_{b}$ in $\mathcal{H}_{d}\otimes\mathcal{H}_{d}$.
Red circles correspond to symmetrizations in the relevant factors
of the space $\mathcal{H}_{d}\otimes\mathcal{H}_{d}$.}
\end{figure}

\begin{proof}[Proof of Lemma \ref{lema bosons p2}]
We show that the operator $\mathbb{P}_{b}:\bigotimes^{2L}\mathcal{H}\rightarrow\bigotimes^{2L}\mathcal{H}$,
given by
\[
\mathbb{P}_{b}=\left(\mathbb{P}_{11'}^{+}\circ\mathbb{P}_{22'}^{+}\circ\ldots\circ\mathbb{P}_{LL'}^{+}\right)\left(\mathbb{P}_{\left\{ 1,\ldots,L\right\} }^{\mathrm{sym}}\circ\mathbb{P}_{\left\{ 1',\ldots,L'\right\} }^{\mathrm{sym}}\right)\,,
\]
equals $\mathbb{P}^{2\lambda_{0}}$. Notice that $\mathbb{P}_{b}\left(\mathrm{Sym}^{L}\left(\mathcal{H}\right)\vee\mathrm{Sym}^{L}\left(\mathcal{H}\right)\right)\subset\mathrm{Sym}^{L}\left(\mathcal{H}\right)\vee\mathrm{Sym}^{L}\left(\mathcal{H}\right)$.
Moreover, $\mathbb{P}_{b}$ is a projector onto $\mathrm{Sym}^{2L}\left(\mathcal{H}\right)$,
a completely symmetric subspace of $\mathcal{H}_{d}\otimes\mathcal{H}_{d}$.
Subspace $\mathrm{Sym}^{2L}\left(\mathcal{H}\right)$ is an irreducible
representation of $K$ labeled by the highest weight $\lambda_{0}=\left(2L,0,\ldots,0\right)$.
For a coherent bosonic state $\ket{\psi}\in\mathcal{H}_{b}$, we have
$\mathbb{P}_{b}\ket{\psi}\otimes\mathbb{\ket{\psi}}=\ket{\psi}\otimes\mathbb{\ket{\psi}}$.
As a result, by criterion \eqref{eq:polynomial characterisation},
$\mathbb{P}_{b}=\mathbb{P}^{2\lambda_{0}}$. 
\end{proof}

\subsection{Slater determinants\label{sub:Slater-determinants}}

The Hilbert space describing $L$ fermionic particles in $N$ modes
is%
\footnote{By $\bigwedge^{L}\left(\mathcal{H}\right)$ we denote the $L$-fold
anti-symmetrization of the Hilbert space $\mathcal{H}$. %
} $\mathcal{H}^{\lambda_{0}}=\mathcal{H}_{f}=\bigwedge^{L}\left(\mathcal{H}\right)$,
where $\mathcal{H}\approx\mathbb{C}^{N}$ (obviously we have to assume
that $L\leq N$). As a symmetry group we take again $K=\mathrm{SU}(N)$.
The group $K$ is represented on $\mathcal{H}_{f}$ via the irreducible
representation $\Pi_{f}$, describing an arbitrary ``single-particle''
evolution, 
\begin{equation}
\Pi_{f}(U)\left(\ket{\psi_{1}}\wedge\cdots\wedge\ket{\psi_{L}}\right)=U\ket{\psi_{1}}\wedge\cdots\wedge U\ket{\psi_{L}},\label{eq:ferm rep-1}
\end{equation}
where $\ket{\psi_{1}}\wedge\cdots\wedge\ket{\psi_{L}}$ denotes the
projection onto $\bigwedge^{L}\left(\mathcal{H}\right)$ of the vector
$\ket{\psi_{1}}\otimes\cdots\otimes\ket{\psi_{L}}$ (it is clear that
vectors of the form $\ket{\psi_{1}}\wedge\cdots\wedge\ket{\psi_{L}}$
span $\bigwedge^{L}\left(\mathcal{H}\right)$). The representation
$\Pi_{f}$ is characterized by the highest weight
\[
\lambda_{0}=\left(\overset{L}{\overbrace{1,\ldots,1}},0,\ldots,0\right)\,.
\]
The set of generalized coherent states coincides with the set consisting
of projectors onto Slater determinants,
\begin{equation}
\mathcal{M}_{f}=\left\{ \kb{\psi}{\psi}\in\mathcal{D}_{1}\left(\mathcal{H}_{f}\right)\,|\,\ket{\psi}=\ket{\phi_{1}}\wedge\ket{\phi_{2}}\wedge\ldots\wedge\ket{\phi_{L}},\,\ket{\phi_{i}}\in\mathbb{C}^{N},\,\bk{\phi_{i}}{\phi_{j}}=\delta_{ij}\right\} \,.\label{eq:ferm slater}
\end{equation}
This time the explicit computation of the operator $A$ describing
the set $\mathcal{M}_{f}$ in a bit more complicated but nevertheless
can be done analytically. Just like in the case of bosons (see \eqref{eq:encoding bos})
we have
\begin{equation}
\bigwedge^{L}\left(\mathcal{H}\right)\subset\mathcal{H}_{1}\otimes\ldots\otimes\mathcal{H}_{L},\,\bigwedge^{L}\left(\mathcal{H}\right)\vee\bigwedge^{L}\left(\mathcal{H}\right)\subset\left(\bigotimes_{i=1}^{i=L}\mathcal{H}_{i}\right)\otimes\left(\bigotimes_{i=1'}^{i=L'}\mathcal{H}_{i}\right)=\mathcal{H}_{d}\otimes\mathcal{H}_{d}\,.\label{eq:encoding ferm}
\end{equation}
By $\mathbb{P}_{\left\{ 1,\ldots,L\right\} }^{a\mathrm{sym}}:\mathcal{H}_{d}\otimes\mathcal{H}_{d}\rightarrow\mathcal{H}_{d}\otimes\mathcal{H}_{d}$
we denote the projector onto the subspace of $\mathcal{H}_{d}\otimes\mathcal{H}_{d}$
which is completely asymmetric with respect to interchange of spaces
labeled by indices from the set $\left\{ 1,2,\ldots,L\right\} $.
We define $\mathbb{P}_{\left\{ 1',\ldots,L'\right\} }^{\mathrm{asym}}$
in the analogous way. 
\begin{lem}
\label{lema crit fermions} Under the introduced notation a closed
expression for the projector operator $\mathbb{P}^{2\lambda_{0}}$
(see \eqref{eq:polynomial characterisation}) for the case of $K=\mathrm{SU}(N)$
and $\mathcal{H}^{\lambda_{0}}=\mathcal{H}_{f}$ reads
\begin{equation}
\mathbb{P}^{2\lambda_{0}}=\mathbb{P}_{f}=\mbox{\ensuremath{\alpha}}\left(\mathbb{P}_{11'}^{+}\circ\mathbb{P}_{22'}^{+}\circ\ldots\circ\mathbb{P}_{LL'}^{+}\right)\left(\mathbb{P}_{\left\{ 1,\ldots,L\right\} }^{\mathrm{asym}}\circ\mathbb{P}_{\left\{ 1',\ldots,L'\right\} }^{\mathrm{asym}}\right)\,,\label{eq:operator two copies}
\end{equation}
where $\alpha=\mbox{\ensuremath{\frac{2^{L}}{L+1}}}$ and it is understood
that $\mathbb{P}^{2\lambda_{0}}$ acts on the space $\mathcal{H}_{d}\otimes\mathcal{H}_{d}$.
By \eqref{eq:polynomial characterisation} we have the following characterization
of the set of projectors onto Slater determinants
\begin{equation}
\kb{\psi}{\psi}\in\mathcal{M}_{f}\,\Longleftrightarrow C_{f}^{2}\left(\ket{\psi}\right)=\bra{\psi}\bra{\psi}\left(\mathbb{P}^{\mathrm{sym}}-\mathbb{P}_{f}\right)\ket{\psi}\ket{\psi}=0\,.\label{eq:criterion ferm}
\end{equation}

\end{lem}
\noindent In analogy to the case of bosons we have 
\[
\left.\mathbb{P}_{f}\right|_{\mathrm{\bigwedge}^{L}\left(\mathcal{H}\right)\otimes\mathrm{\bigwedge}^{L}\left(\mathcal{H}\right)}=\mbox{\ensuremath{\alpha}}\mathbb{P}_{11'}^{+}\circ\mathbb{P}_{22'}^{+}\circ\ldots\circ\mathbb{P}_{LL'}^{+}\,,
\]
since for any $\ket{\Psi}\in\mathrm{\bigwedge}^{L}\left(\mathcal{H}\right)\otimes\mathrm{\bigwedge}^{L}\left(\mathcal{H}\right)$
we have $\left(\mathbb{P}_{\left\{ 1,\ldots,L\right\} }^{a\mathrm{sym}}\circ\mathbb{P}_{\left\{ 1',\ldots,L'\right\} }^{a\mathrm{sym}}\right)\ket{\Psi}=\ket{\Psi}$. 

\begin{figure}[h]
\centering{}states\includegraphics[width=8cm]{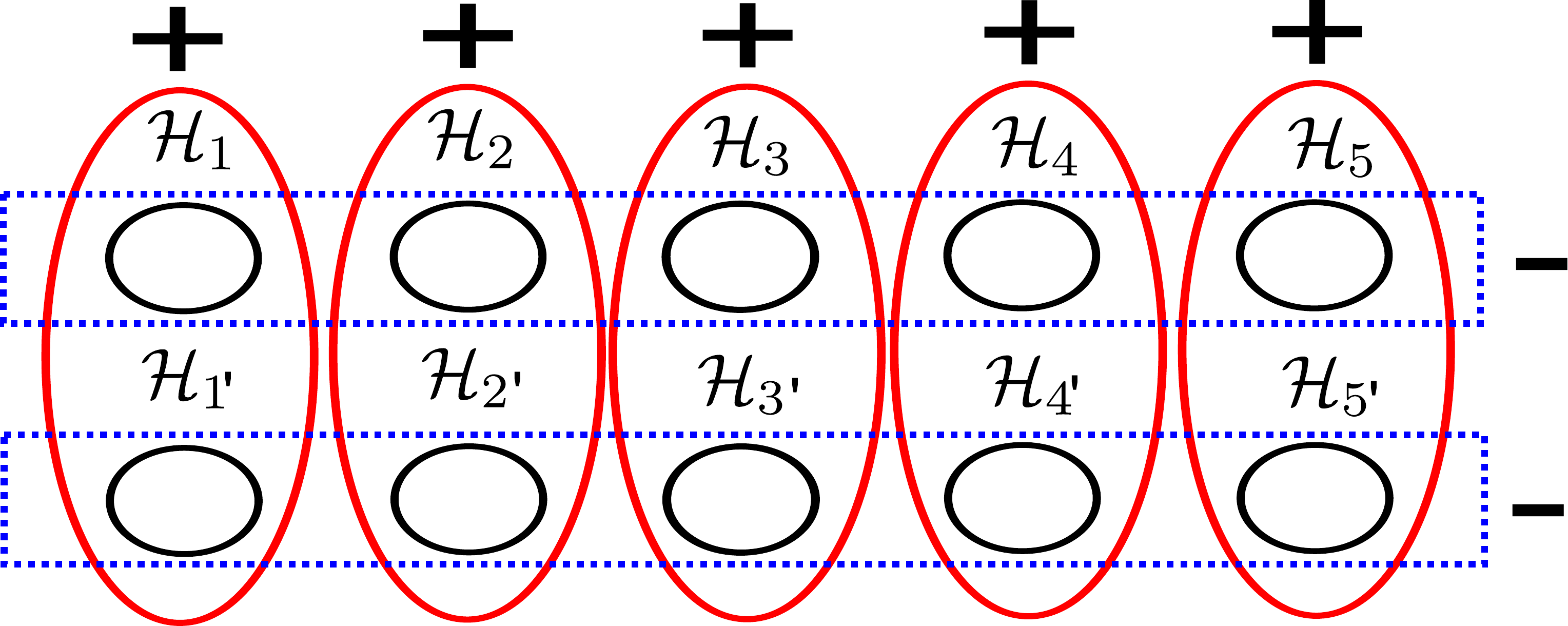}\protect\caption{\label{fig:disting_bos_op-1-1}Graphical illustration he action of
the operator $\mathbb{P}_{f}$. Red circles correspond to symmetrization
in the relevant factors of the space $\mathcal{H}_{d}\otimes\mathcal{H}_{d}$.
Blue rectangles enclose factors on which we have antisymmetrization.}
\end{figure}

\begin{proof}
The proof is presented in Section \ref{sec:Proofs-chapter polyn charact}
of the Appendix (see page \pageref{sub:Proof-of-Lemma alg form ferm}). 
\end{proof}

\subsection{Physical interpretation of the invariant polynomial for distinguishable
particles, bosons and fermions}

Expressions for the invariant polynomial $C^{2}\left(\kb{\psi}{\psi}\right)=\mbox{\ensuremath{\bra{\psi}\bra{\psi}}}A\ket{\psi}\ket{\psi}$
that detects coherent states (in a sense of \ref{eq:criterion ver 1})
depend, in the case of distinguishable particles, bosons and fermions,
only upon $\bra{\psi}\bra{\psi}\mathbb{P}_{11'}^{+}\circ\mathbb{P}_{22'}^{+}\circ\ldots\circ\mathbb{P}_{LL'}^{+}\ket{\psi}\ket{\psi}$
(see Eq. \eqref{eq:crit prod states}, \eqref{eq:criterion bos} and
\eqref{eq:criterion ferm}). One can show (c.f. \citep{Augusiak2009})
that for arbitrary $L$-particle states the following expression holds,
\begin{equation}
\bra{\psi}\bra{\psi}\mathbb{P}_{11'}^{+}\circ\mathbb{P}_{22'}^{+}\circ\ldots\circ\mathbb{P}_{LL'}^{+}\ket{\psi}\ket{\psi}=2^{-L}\left(\sum_{k}\mathrm{tr}\left(\rho_{k}^{2}\right)+2\right)\,,\label{eq:phys form}
\end{equation}
where the summation is over all different $2^{L}-2$ proper subsystems
of $L$-particle systems and $\rho_{k}$ is the reduced density matrix
describing the particular subsystem $k$. Notice that the expression
\eqref{eq:phys form} is also valid for bosons and fermions because
we can formally embed bosonic and fermionic Hilbert spaces in $\bigotimes^{L}\left(\mathbb{C}^{N}\right)$. 

Although in our approach we cared only whether a given multiparty
pure state is ``classical'' or not, it is tempting to ask what are
the ``maximally entangled'' states corresponding to measures $C_{dist}^{2}$,
$C_{b}^{2}$ and $C_{f}^{2}$ in each of three considered contexts.
Equation \eqref{eq:phys form} enables us to formally answer this
question. Clearly, $C_{dist}^{2}$, $C_{b}^{2}$ and $C_{f}^{2}$
will be maximal once for each proper subsystem $k$ the corresponding
reduced density matrix will be maximally mixed. For the case of distinguishable
particles states $\kb{\psi}{\psi}$ satisfying this condition are
called ``absolutely maximally entangled states''. The problem of
deciding whether for a given $L$ and $N$ such states at all exist
is in general unsolved \citep{absmaxent2012}. Therefore, one cannot
hope for an easy characterization of states that maximize $C_{dist}^{2}$,
$C_{b}^{2}$ or $C_{f}^{2}$ (or equivalently, minimize \eqref{eq:phys form}
once $\ket{\psi}\in\mathcal{H}_{dist},\,\mathcal{H}_{b}$ or $\mathcal{H}_{f}$
respectively). Nevertheless, the characterization of ``absolutely
maximally entangled'' bosonic and fermionic states is certainly an
interesting open problem. 

Let us mention also the connection of the polynomials $C_{dist}^{2}$,
$C_{b}^{2}$ and $C_{f}^{2}$ to the results known before. In the
case of distinguishable particles the function $C_{dist}\left(\kb{\psi}{\psi}\right)$
coincides with the multiparty concurrence considered before in \citep{Aolita2006,Mintert2007}.
In these papers the authors derived lower bounds for the convex roof
extension of $C_{dist}\left(\kb{\psi}{\psi}\right)$, which is the
faithful indicator of entanglement for mixed states%
\footnote{For the definition of the convex roof extension of a function see
Subsection \ref{sub:Antiunitary-conjugations-andUW}.%
} Polynomials $C_{b}^{2}$ and $C_{f}^{2}$ were also derived in \citep{Kus2009}
for the case of $L=2$ particles.

\subsection{Pure Fermionic Gaussian states\label{sub:Fermionic-Gaussian-states}}

The discussion of the case of pure fermionic Gaussian states needs
to be preceded by a more thorough introduction then the cases considered
before in this section. The reason for this is a rather involved definition
of the class of fermionic Gaussian states. Because of that, despite
the extensive use of fermionic Gaussian states in physics (see Section
\ref{sec:General-motivation} and Section \ref{sec:Classical-simulation-of-FLO})
their group-theoretical interpretation is not well-known. It is perhaps
worth mentioning that in the mathematical literature pure fermionic
Gaussian states are referred to as ``pure spinors'' \citep{Chevalley1997,Manivel2009}.
This subsection is organized as follows. We first present the setting
and notation necessary to define pure Fermionic Gaussian states. Then
we will give standard ``physical'' definitions of general (mixed
and pure) Fermionic Gaussian states. We then show that pure fermionic
Gaussian states, of $d$ mode fermionic system can be identified with
the Perelomov coherent states%
\footnote{Strictly speaking pure Gaussian states will be identified with the
disjoint union of the coherent states of $\mathrm{Spin}\left(2d\right)$
in the direct sum of two irreducible representations.%
} of the of the group $\mathrm{Spin}\left(2d\right)$ in its spinor
representations%
\footnote{The group $\mathrm{Spin}\left(2d\right)$ is a compact simply-connected
Lie group (see Subsection \eqref{sub:Spinor-represenations-of}) which
is isomorphic to the group of above-mentioned Bogolyubov transformations.%
}. In the next part of we give (using results of a recent paper of
de Melo, Ćwikliński and Terhal \citep{powernoisy2013}) an explicit
form of the operator $A$ that allows to characterize the set of pure
Gaussian states, denoted by $\mathcal{G}$, via Eq.\eqref{eq:polynomial characterisation},
as a zero set of a polynomial invariant under the group of Bogolyubov
transformations. Finally, due to the group-theoretical characterization
of pure Gaussian states, the results concerning the operator $A$
presented in \citep{powernoisy2013} can be understood from the perspective
of a general framework introduced in this Section. 

Let us consider a fermionic system whose particles can be in $d$
modes. If there is no restriction for the number of particles the
relevant Hilbert space associated to a system is the fermionic Fock
space which we denote $\mathcal{H}_{\mathrm{Fock}}\left(\mathbb{C}^{d}\right)$.
Mathematically $\mathcal{H}_{\mathrm{Fock}}\left(\mathbb{C}^{d}\right)$
is defined by
\[
\mathcal{H}_{\mathrm{Fock}}\left(\mathbb{C}^{d}\right)=\bigoplus_{k=0}^{d}\bigwedge^{k}\left(\mathbb{C}^{d}\right)\,,
\]

where $\bigwedge^{0}\left(\mathbb{C}^{d}\right)$ is a one dimensional
space spanned by the Fock vacuum which we denote by $\ket 0$. To
every choice of the orthonormal basis $\left\{ \ket{\phi_{k}}\right\} _{k=1}^{k=d}$
of the single-particle Hilbert space $\mathbb{C}^{d}$ one associates
the standard creation and annihilation operators: $a_{k}^{\dagger}$,
$a_{k}$, $k=1,\ldots,d$, satisfying canonical anticommutation relations,
\begin{equation}
\left\{ a_{k},a_{l}^{\dagger}\right\} =\delta_{kl}\mathbb{I}\,,\label{eq:CAR}
\end{equation}
where $k,l=1,\ldots,d$ and $\left\{ \cdot,\cdot\right\} $ denotes
the standard anticommutator of operators, $\left\{ A,B\right\} =AB+BA$.
The Fock space $\mathcal{H}_{\mathrm{Fock}}\left(\mathbb{C}^{d}\right)$
has dimension $2^{d}$ and is spanned by the set of orthonormal Fock
states:
\begin{equation}
\ket{n_{1},\ldots,n_{d}}=\left(a_{1}^{\dagger}\right)^{n_{1}}\cdots\left(a_{d}^{\dagger}\right)^{n_{d}}\ket 0\,,\,\label{eq:fock states}
\end{equation}
where $n_{k}=0,1$, $k=1,\dots,d$, are the occupation numbers of
the modes characterized by the single particle states $\ket{\phi_{k}}\in\mathbb{C}^{d}$.
For each $k=1,\ldots,d$ the operator of the particle number in the
mode $k$ is given by $\hat{n}_{k}=a_{k}^{\dagger}a_{k}$. Due to
the fact that Fock states \eqref{eq:fock states} form a basis of
$\mathcal{H}_{\mathrm{Fock}}\left(\mathbb{C}^{d}\right)$, every linear
operator $A\in\mathrm{End}\left(\mathcal{H}_{\mathrm{Fock}}\left(\mathbb{C}^{d}\right)\right)$
can be written as a polynomial in the creation and annihilation operators. 

It will be convenient to describe fermionic Gaussian states via so-called
Majorana fermion operators \citep{powernoisy2013,Lagr2004}, 
\[
c_{2k-1}=a_{k}+a_{k}^{\dagger}\,,\, c_{2k}=i\left(a_{k}-a_{k}^{\dagger}\right)\,,
\]
$k=1,\ldots,d$. Majorana operators are Hermitian, traceless and satisfy
the following anticommutation relations, 
\begin{equation}
\left\{ c_{k},\, c_{l}\right\} =2\delta_{kl}\,,\label{eq:anticommutation majorana}
\end{equation}
where $k,l=1,\ldots,d$. Because creation and annihilation operators
can be expressed uniquely via the Majorana operators every operator
$X\in\mathrm{End}\left(\mathcal{H}_{\mathrm{Fock}}\left(\mathbb{C}^{d}\right)\right)$
can be written as a polynomial in $\left\{ c_{i}\right\} _{i=1}^{d}$
\footnote{One can show that the complex algebra generated by polynomials in
Majorana operators is isomorphic with the complex Clifford algebra
$Cl\left(2d,\mathbb{C}\right)$ (\textit{see Subsection \ref{sub:Spinor-represenations-of}}).%
},
\begin{equation}
X=\alpha_{0}\mathbb{I}+\sum_{k=1}^{2d}\sum_{1<l_{1}<l_{2}<\ldots<l_{k}\leq2d}\alpha_{l_{1}l_{2}\ldots l_{k}}c_{l_{1}}c_{l_{2}}\ldots c_{l_{k}}\,,\label{eq:arbitrary operator}
\end{equation}
where coefficients $\alpha_{0}$ and $\alpha_{l_{1}l_{2}\ldots l_{k}}$
are in general complex. It what follows we will consider only \textit{even}
operators, i.e. operators commuting with the total parity operator%
\footnote{This requirement can be understood as a kind of a superselection rule
imposed on a system in question \citep{Banuls2007}.%
} 
\begin{equation}
Q=\prod_{k=1}^{d}(-1)^{\hat{n}_{k}}=\prod_{k=1}^{d}\left(\mathbb{I}-2a_{k}^{\dagger}a_{k}\right)=i^{d}\prod_{k=1}^{2d}c_{k}\,.\label{eq:Q majorana}
\end{equation}

The operator $Q$ has the eigenvalues $\pm1$ and the whole Fock space
decomposes onto manually orthogonal eigenspaces of $Q$, 
\begin{equation}
\mathcal{H}_{\mathrm{Fock}}\left(\mathbb{C}^{d}\right)=\mathcal{H}_{\mathrm{Fock}}^{+}\left(\mathbb{C}^{d}\right)\oplus\mathcal{H}_{\mathrm{Fock}}^{-}\left(\mathbb{C}^{d}\right)\,,\label{eq:fock decomposition}
\end{equation}
where $\mathcal{H}_{\mathrm{Fock}}^{+}\left(\mathbb{C}^{d}\right)$
and $\mathcal{H}_{\mathrm{Fock}}^{-}\left(\mathbb{C}^{d}\right)$
are spanned by states with respectively even and odd number of excitations.
The corresponding orthonormal projectors onto these subspaces are
$\mathbb{P}_{\pm}=\frac{1}{2}\left(\mathbb{I}\pm Q\right)$. Even
operators $\mathrm{End}_{\mathrm{even}}\left(\mathcal{H}_{\mathrm{Fock}}\left(\mathbb{C}^{d}\right)\right)$
form a subalgebra of $\mathrm{End}\left(\mathcal{H}_{\mathrm{Fock}}\left(\mathbb{C}^{d}\right)\right)$
that respect the direct sum \eqref{eq:fock decomposition},
\begin{equation}
\mathrm{End}_{\mathrm{even}}\left(\mathcal{H}_{\mathrm{Fock}}\left(\mathbb{C}^{d}\right)\right)=\mathrm{End}\left(\mathcal{H}_{\mathrm{Fock}}^{+}\left(\mathbb{C}^{d}\right)\right)\oplus\mathrm{End}\left(\mathcal{H}_{\mathrm{Fock}}^{-}\left(\mathbb{C}^{d}\right)\right).\label{eq:splitting algebra}
\end{equation}

Application of the anticommutation relations \eqref{eq:anticommutation majorana}
and \eqref{eq:arbitrary operator} shows that $X$ is even if and
only if it can be written as a polynomial in Majorana operators involving
only monomials of even degree,
\begin{equation}
X=\alpha_{0}\mathbb{I}+\sum_{k=1}^{d}i^{k}\sum_{1\leq l_{1}<l_{2}<\ldots<l_{2k}\leq2d}\alpha_{l_{1}l_{2}\ldots l_{2k}}c_{l_{1}}c_{l_{2}}\ldots c_{l_{2k}}\,,\label{eq:even decomposition}
\end{equation}
where all the coefficients $\alpha_{0}$ and $\alpha_{l_{1}l_{2}\ldots l_{2k}}$
are real. Let $\mathcal{D}_{\mathrm{even}}\left(\mathcal{H}_{\mathrm{Fock}}\left(\mathbb{C}^{d}\right)\right)$
denote the set of even mixed states on $\mathcal{H}_{\mathrm{Fock}}\left(\mathbb{C}^{d}\right)$.
For an even mixed state $\rho\in\mathcal{D}_{\mathrm{even}}\left(\mathcal{H}_{\mathrm{Fock}}\left(\mathbb{C}^{d}\right)\right)$
the correlation matrix $M\left(\rho\right)=\left[M_{kl}\left(\rho\right)\right]$
is defined by
\begin{equation}
M_{kl}\left(\rho\right)=\frac{i}{2}\mathrm{Tr}\left(\rho\,\left[c_{k},\, c_{l}\right]\right)\,,\, k,l=1,\ldots,2d\,.\label{eq:correlation matrix}
\end{equation}
One checks that $M\left(\rho\right)$ is real and antisymmetric and
therefore (see Eq.\eqref{eq:special orthogonal definition}) can be
considered as an element of the Lie algebra $\mathrm{\mathfrak{so}}\left(2d\right)$.
The general Gaussian states, denoted by $Gauss\subset\mathcal{D}_{\mathrm{even}}\left(\mathcal{H}_{\mathrm{Fock}}\left(\mathbb{C}^{d}\right)\right)$,
are defined as states $\rho\in\mathcal{D}_{\mathrm{even}}\left(\mathcal{H}_{\mathrm{Fock}}\left(\mathbb{C}^{d}\right)\right)$
of the form
\begin{equation}
\rho=K\cdot\mathrm{exp}\left(i\sum_{k\neq l}h_{kl}c_{k}c_{l}\right)\,,\label{eq:mixed gaussian def}
\end{equation}
where the $2d\times2d$ matrix $h=\left[h_{kl}\right]$ is real and
antisymmetric ($h\in\mathfrak{so}\left(2d\right)$) and $K$ is the
normalization constant. By virtue of the fermionic Wick theorem \citep{Fetter2003,Lagr2004}
Gaussian states are fully determined by their correlation matrix $M$.
The correlation matrix of a Gaussian state $\rho$ satisfies (see
\citep{powernoisy2013,Lagr2004})
\[
M\left(\rho\right)M\left(\rho\right)^{T}\leq\mathbb{I}_{2d}\,,
\]
where $\mathbb{I}_{2d}$ is the $2d\times2d$ identity matrix. Our
main object of interest are pure fermionic Gaussian states, which
we denote by $\mathcal{M}_{g}$. By the virtue of the following fact
the class $\mathcal{M}_{g}$ is effectively characterized by orthogonal
correlation matrices.
\begin{fact}
\label{gauss terhal fact}\citep{powernoisy2013} The class of pure
fermionic Gaussian states $\mathcal{M}_{g}$ consists precisely of
states $\rho\in\mathcal{D}_{\mathrm{even}}\left(\mathcal{H}_{\mathrm{Fock}}\left(\mathbb{C}^{d}\right)\right)$
that have orthogonal correlation matrix.
\begin{equation}
\mbox{\ensuremath{\kb{\psi}{\psi}\in\mathcal{M}_{g}} }\,\Longleftrightarrow\, M\left(\rho\right)M\left(\rho\right)^{T}=\mathbb{I}_{2d}\,.\label{eq:gauss}
\end{equation}

\end{fact}
We now introduce an important symmetry group acting on $\mathcal{H}_{\mathrm{Fock}}\left(\mathbb{C}^{d}\right)$
called the group of Bogolyubov transformations and denoted by $\mathcal{B}$.
The group $\mathcal{B}$ is generated by arbitrary Hamiltonians quadratic
in the creation and annihilation operators
\[
H=i\sum_{k,l=1}^{2d}H_{kl}c_{k}c_{l}\,,
\]
where the matrix $H=\left[H_{kl}\right]$ is real and antisymmetric.
In other words we have 
\begin{equation}
\mbox{\ensuremath{\mathcal{B}}}=\left\{ U\in\mathrm{U}\left(\mathcal{H}_{\mathrm{Fock}}\left(\mathbb{C}^{d}\right)\right)\,|\, U=\mathrm{exp}\left(i\left(i\sum_{k\neq l}H_{kl}c_{k}c_{l}\right)\right)\right\} \,.\label{eq:bogolubov transf}
\end{equation}
The class of unitary operations $\mathcal{B}$ is also called sometimes
Fermionic Linear Optics \citep{Lagr2004,BravyiKoeningSimul,powernoisy2013}.
Elements of $\mathcal{B}$ commute with the operator of total parity
$Q$ (Eq.\eqref{eq:Q majorana}). Rigorously speaking in Eq.\eqref{eq:bogolubov transf}
we should have taken as $\mathcal{B}$ the closure (with respect to,
say, the norm topology in $\mathrm{U}\left(\mathcal{H}_{\mathrm{Fock}}\left(\mathbb{C}^{d}\right)\right)$)
of the set consisting of finite products of operators of the form
$\mathrm{exp}\left(iH\right)$. However one can show that due to the
simple-connectedness and compactness of the group $\mathrm{Spin}\left(2d\right)$
(see Subsection \ref{sub:Spinor-represenations-of}) that these two
definitions of $\mathcal{B}$ coincide. Moreover, directly from the
definition of $\mathcal{B}$ we get that for $U\in\mathcal{B}$ the
following properties hold
\begin{align}
Uc_{l}U^{\dagger} & =\sum_{k=1}^{2d}R_{kl}\left(U\right)c_{k}\,,\nonumber \\
\, M\left(U\rho U^{\dagger}\right) & =R\left(U\right)M\left(\rho\right)R\left(U\right)^{T}\,,\label{eq:transformation}
\end{align}
where $R\left(U\right)\in\mathrm{SO}\left(2d\right)$. Moreover, the
mapping $U\rightarrow R\left(U\right)$ is surjective, i.e. every
special orthogonal transformation can appear in \eqref{eq:transformation}
for a suitable choice of $U\in\mathcal{B}$. In Subsection \ref{sub:Spinor-represenations-of}
it was shown that the group of Bogolyubov transformations $\mathcal{B}$
can be realized as a spinor representation of the group $\mbox{\ensuremath{\mathrm{Spin}}}\left(2d\right)$.
To be more precise we have $\mathcal{B}=\Pi_{s}\left(\mbox{\ensuremath{\mathrm{Spin}}}\left(2d\right)\right)$,
where $\Pi_{s}:\mbox{\ensuremath{\mathrm{Spin}}}\left(2d\right)\rightarrow\mathrm{U}\left(\mathcal{H}_{\mathrm{Fock}}\left(\mathbb{C}^{d}\right)\right)$
is a so-called spinor representation of $\mbox{\ensuremath{\mathrm{Spin}}}\left(2d\right)$
(see Subsection \ref{sub:Spinor-represenations-of} for more details).
The spinor representation is reducible and decomposes onto two irreducible
components corresponding to subspaces of $\mathcal{H}_{\mathrm{Fock}}\left(\mathbb{C}^{d}\right)$
spanned by the Fock states with respectively even and odd number of
excitations,
\begin{equation}
\Pi_{s}=\Pi_{s}^{+}\oplus\Pi_{s}^{-}\,,\label{eq:decomposition spinor}
\end{equation}
where
\[
\Pi_{s}^{\pm}:\mbox{\ensuremath{\mathrm{Spin}}}\left(2d\right)\rightarrow\mathrm{U}\left(\mathcal{H}_{\mathrm{Fock}}^{\pm}\left(\mathbb{C}^{d}\right)\right)\,.
\]
are representations described in detail in Subsection \ref{sub:Spinor-represenations-of}.
Note that Eq.\eqref{eq:decomposition spinor} is in agreement with
the observation that $\mathcal{B}$ commutes with the total parity
operator. The pure fermionic Gaussian states are closely related to
coherent states of $\mathrm{Spin}\left(2d\right)$ in representations
$\Pi_{s}^{\pm}$ as specified in the following proposition.
\begin{prop}
\label{geometric interpretation of fermionic Gaussian}Let $\mathcal{M}_{g}\subset\mathcal{D}_{1}\left(\mathcal{H}_{\mathrm{Fock}}\left(\mathbb{C}^{d}\right)\right)$
be the class of pure fermionic Gaussian states defined as above. Let
$\ket{\psi_{0}^{\pm}}\in\mathcal{H}_{\mathrm{Fock}}^{\pm}\left(\mathbb{C}^{d}\right)$
be some fixed highest weight vectors%
\footnote{Without the loss of generality we can take $\ket{\psi_{0}^{+}}=\ket 0$
(Fock vacuum) and $\ket{\psi_{0}^{-}}=\ket{0,0,\ldots,1}$ (excitation
only in the last mode $\ket{\phi_{1}}$).%
} of the representations $\Pi_{s}^{\pm}$. Let $\mathcal{M}_{g}^{\pm}\subset\mathcal{D}_{1}\left(\mathcal{H}_{\mathrm{Fock}}^{\pm}\left(\mathbb{C}^{d}\right)\right)$
be the orbits of $\mathrm{Spin}\left(2d\right)$ through the states
\[
\kb{\psi_{0}^{\pm}}{\psi_{0}^{\pm}}\in\mathcal{D}_{1}\left(\mathcal{H}_{\mathrm{Fock}}^{\pm}\left(\mathbb{C}^{d}\right)\right)\,.
\]
Then, 
\begin{equation}
\mathcal{M}_{g}=\mathcal{M}_{g}^{+}\cup\mathcal{M}_{g}^{-}\,,\label{eq:pure gting proofaussan split}
\end{equation}
where elements of $\mathcal{D}_{1}\left(\mathcal{H}_{\mathrm{Fock}}^{+}\left(\mathbb{C}^{d}\right)\right)$
and $\mathcal{D}_{1}\left(\mathcal{H}_{\mathrm{Fock}}^{-}\left(\mathbb{C}^{d}\right)\right)$
are treated as elements of $\mathcal{D}_{1}\left(\mathcal{H}_{\mathrm{Fock}}\left(\mathbb{C}^{d}\right)\right)$
in a natural manner by the virtue of the decomposition \eqref{eq:fock decomposition}. \end{prop}
\begin{proof}
We present the proof of Proposition \ref{geometric interpretation of fermionic Gaussian}
in Section \ref{sec:Proofs-chapter polyn charact} of the Appendix
(see page \pageref{sub:proof of geometric interpretation}).
\end{proof}
The graphical representation of Eq.\eqref{eq:pure gting proofaussan split}
is presented on Figure \ref{fig:gaussian split}.

\begin{figure}[h]
\centering{}\includegraphics[width=9cm]{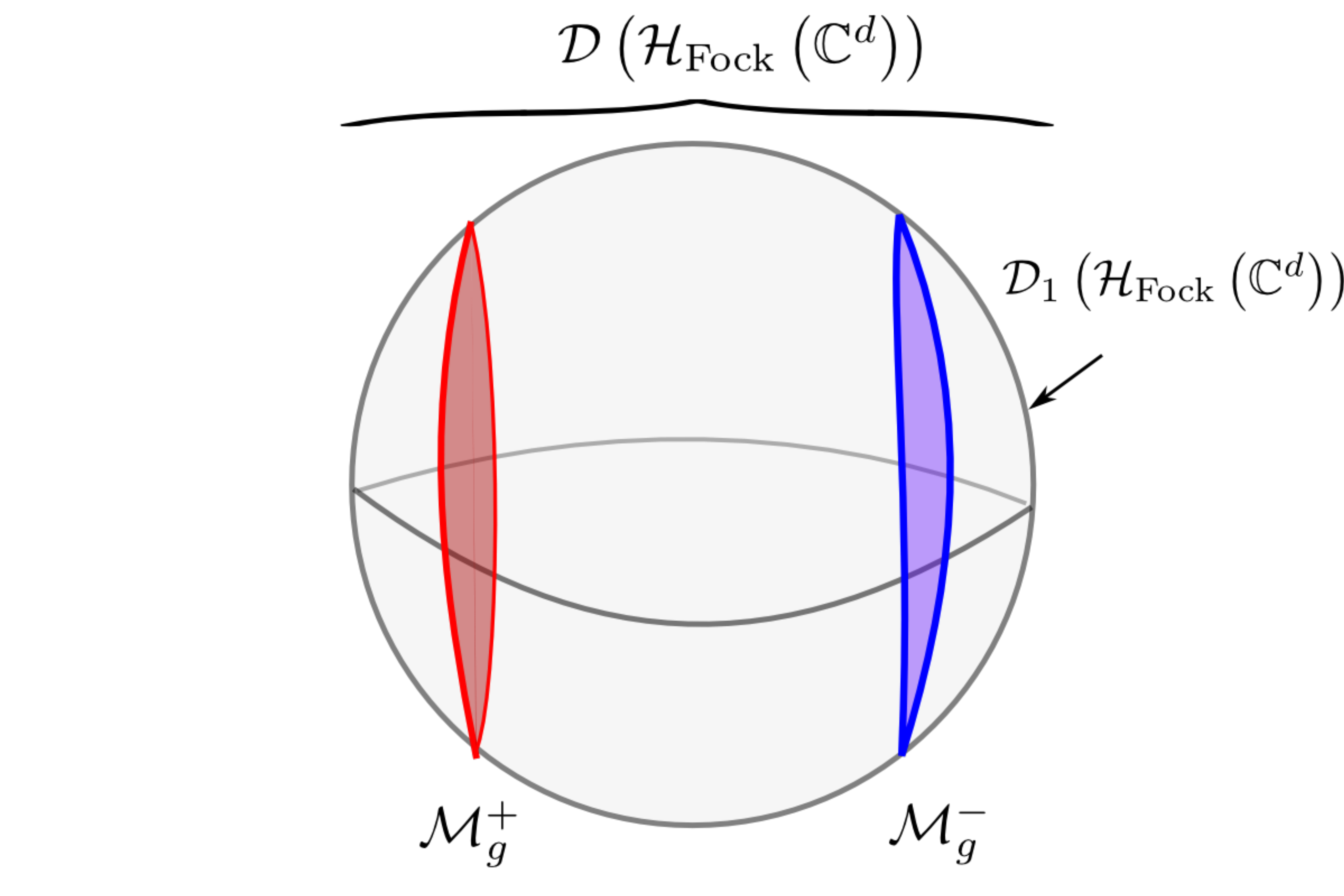}\protect\caption{\label{fig:gaussian split}Graphical representation of the relation
$\mathcal{M}_{g}=\mathcal{M}_{g}^{+}\cup\mathcal{M}_{g}^{-}$, $\mathcal{M}_{g}^{+}\cap\mathcal{M}_{g}^{-}=\emptyset$.}
\end{figure}

We now present, after \citep{powernoisy2013,Lagr2004}, the characterization
of pure fermionic Gaussian states $\mathcal{M}_{g}$ as the null set
of the $\mathrm{Spin}\left(2d\right)$-invariant polynomial defined
via the projector $A$ acting on $\mathrm{Sym}^{2}\left(\mathcal{H}_{\mathrm{Fock}}\left(\mathbb{C}^{d}\right)\right)$
(c.f. Eq.\eqref{eq:polynomial characterisation}). This characterization
is completely analogous to the one presented previously in this section.
The only difference comes from the fact that the relevant representation
of the symmetry group $\Pi_{s}$ is reducible and decomposes onto
two irreducible representations. Let us define the operator $\Lambda\in\mathrm{Herm}\left(\mathcal{H}_{\mathrm{Fock}}\left(\mathbb{C}^{d}\right)\otimes\mathcal{H}_{\mathrm{Fock}}\left(\mathbb{C}^{d}\right)\right)$
\begin{equation}
\Lambda=\sum_{k=1}^{2d}c_{i}\otimes c_{i}\,.\label{eq:operator Lambda}
\end{equation}
A straightforward computation shows that $\Lambda$ is $\mathrm{Spin}\left(2d\right)$-invariant,
i.e.
\[
\left[\Pi_{s}\left(g\right)\otimes\Pi_{s}\left(g\right),\Lambda\right]=0
\]
 for all $g\in\mathrm{Spin}\left(2d\right)$.
\begin{fact}
(\citep{powernoisy2013,Lagr2004}) Let $\mathbb{P}_{0}$ denotes the
projector onto the null eigenspace of the operator $\Lambda$. Then,
the following equivalence holds
\begin{equation}
\mbox{\ensuremath{\kb{\psi}{\psi}}}\in\mathcal{M}_{g}\,\Longleftrightarrow\,\left(\mathbb{P}^{\mathrm{sym}}-\mathbb{P}_{0}\right)\ket{\psi}\ket{\psi}=0\,.\label{eq:equivalence gaussian}
\end{equation}

\end{fact}
The projector $A=\mathbb{P}^{\mathrm{sym}}-\mathbb{P}_{0}$ is manifestly
$\mathrm{Spin}\left(2d\right)$-invariant. The following proposition
presents $\mathbb{P}_{0}$ directly in terms of Majorana operators.
\begin{prop}
\label{algebraic form operator}. The closed-form expression for $\mathbb{P}_{0}$,
the projector onto the null eigenspace of $\Lambda$ (see Eq.\eqref{eq:operator Lambda}),
in terms of the Majorana operators reads
\begin{align}
\mathbb{P}_{0} & =\frac{1}{2^{2d}}\sum_{k=0}^{d}f_{k,d}\left(\sum_{\begin{array}[t]{c}
X\subset\left\{ 1,\ldots,2d\right\} \\
\left|X\right|=2k
\end{array}}\prod_{i\in X}c_{i}\otimes c_{i}\right)\,,\label{eq:closed form expression ferm}
\end{align}
where the inner summation is over all subsets $X$ of the set $\left\{ 1,\ldots,2d\right\} $
with the cardinality $\left|X\right|=2k$ and
\begin{equation}
f_{k,d}=\left(-1\right)^{k}\frac{\left(2k\right)!\left(2d-2k\right)!}{d!k!\left(d-k\right)!}\,.\label{eq:gauss scalar expression}
\end{equation}
\end{prop}
\begin{proof}
The proof is given in Section \ref{sec:Proofs-chapter polyn charact}
of the Appendix (see page \pageref{sub:Proof-of-Proposition-alg form gauss}). \end{proof}
\begin{rem*}
The projectors $\mathbb{P}_{\pm}^{2\lambda_{0}}$ (appearing in Eq.\eqref{eq:polynomial characterisation}
) associated to representations $\Pi_{s}^{\pm}$ can be obtained from
$\mathbb{P}_{0}$ by restricting it to the subspaces $\mathcal{H}_{\mathrm{Fock}}^{\pm}\left(\mathbb{C}^{d}\right)\otimes\mathcal{H}_{\mathrm{Fock}}^{\pm}\left(\mathbb{C}^{d}\right)$.
\end{rem*}

\section{Quadratic characterization of the generalized coherent states in
the infinite dimensional Hilbert spaces \label{sec:inf dimension}}

In this section we provide polynomials that detect correlations of
pure quantum states when dimensions of the relevant Hilbert spaces
are infinite. We will cover the following types of correlations, defined
by the suitable choice of the class of ``non-correlated'' pure states
$\mathcal{M}$ (as explained in Section \ref{chap:Introduction} and
in the beginning of this chapter).
\begin{itemize}
\item Entanglement in system of distinguishable particles; $\mathcal{M}=\mathcal{M}_{d}$
- pure product states;
\item Entanglement in system of finite number of bosons; $\mathcal{M}=\mathcal{M}_{b}$
- symmetric product states;
\item ``Entanglement'' in system of finite number of fermions; $\mathcal{M}=\mathcal{M}_{f}$
- symmetric product states;
\end{itemize}
The section is organized as follows. We first make a few technical
remarks about infinite dimensional setting. In the rest of the section
we prove that the criteria for correlations are given by formally
the same expressions is in the finite dimensional case (see Equations
\eqref{eq:crit prod states}, \eqref{eq:criterion bos} and \eqref{eq:criterion ferm}). 

A separable Hilbert space $\mathcal{H}$ is, by definition, a Hilbert
space in which it is possible to chose a countable basis. Almost all
Hilbert spaces that occur in physics are separable \citep{Reed1972}.
Examples include all finite dimensional Hilbert spaces or the space
space of square integrable (with respect to the Lebesgue measure)
functions on $\mathbb{R}^{d}$, $L^{2}\left(\mathbb{R}^{d},dx\right)$.
In this section we consider only separable Hilbert spaces. The space
of pure states of a quantum system is a set $\mathcal{D}_{1}\left(\mathcal{H}\right)$
consisting of rank one orthogonal projectors acting on $\mathcal{H}$,
i.e. operators that project orthogonally onto one dimensional subspaces
of $\mathcal{H}$. The space $\mathcal{D}_{1}\left(\mathcal{H}\right)$
is a metric space with respect to the Hilbert–Schmidt metric \citep{Reed1972}.
That is, for $\kb{\psi}{\psi},\kb{\phi}{\phi}\in\mathcal{D}_{1}\left(\mathcal{H}\right)$
we have
\begin{equation}
\mathrm{d}\left(\kb{\psi}{\psi}\kb{\phi}{\phi}\right)=\sqrt{\mathrm{tr}\left[\left(\kb{\psi}{\psi}-\kb{\phi}{\phi}\right)^{2}\right]}=\sqrt{2\left(1-\left|\bk{\psi}{\phi}\right|^{2}\right)}\,,\label{eq:metric structure}
\end{equation}
where $\mathrm{d}\left(\cdot,\,\cdot\right)$ denotes the metric. 
\begin{prop}
\label{infinite dimensional metric}The set of pure states $\mathcal{D}_{1}\left(\mathcal{H}\right)$
endowed with the metric \eqref{eq:metric structure} is a complete
metric space, i.e. every Cauchy sequence of elements from $\mathcal{D}_{1}\left(\mathcal{H}\right)$
converges. \end{prop}
\begin{proof}
We will use the proof of this proposition to briefly introduce the
notion of Hilbert-Schmidt operators. A bounded operator%
\footnote{The set of bounded (continuous) operators on $\mathcal{H}$ is consists
of operators $A:\mathcal{H}\rightarrow\mathcal{H}$ that satisfy $\left\Vert A\ket v\right\Vert \leq C\left\Vert \ket v\right\Vert $
for some $C\geq0$ and all $\ket v\in\mathcal{H}$.%
} $A$ is a Hilbert Schmidt-operator if and only if for some basis
$\left\{ \ket{e_{i}}\right\} _{i\in\mathcal{I}}$ of $\mathcal{H}$
we have
\[
\sum_{i\in\mathcal{I}}\left\Vert A\ket{e_{i}}\right\Vert ^{2}<\infty\,.
\]
The space of Hilbert-Schmidt operators $\mathrm{HS}\left(\mathcal{H}\right)$
is a Hilbert space itself when equipped with the Hilbert-Schmidt inner
product $\left\langle \cdot,\cdot\right\rangle _{\mathrm{HS}}$ defined
in the following way
\begin{equation}
\left\langle A,B\right\rangle _{\mathrm{HS}}=\sum_{i\in\mathcal{I}}\bk{e_{i}}{A^{\ast}B|e_{i}}=\mathrm{tr}\left(A^{\ast}B\right)\,,\label{eq:Hilbert Schmidt inner}
\end{equation}
where, as above, $\left\{ \ket{e_{i}}\right\} _{i\in\mathcal{I}}$
is a basis of $\mathcal{H}$, and $A^{\ast}$ denotes the adjoint
of the operator $A$. The Hilbert-Schmidt inner product induces a
trace distance between elements of $A,B\in\mathrm{HS}\left(\mathcal{H}\right)$
\begin{equation}
\left\Vert A-B\right\Vert =\sqrt{\mathrm{tr}\left[\left(A-B\right)^{2}\right]}\,.\label{eq:hilbert schmidt distance}
\end{equation}
Simple calculation shows that the metric \eqref{eq:metric structure}
is a restriction of the trace distance to $\mathcal{D}_{1}\left(\mathcal{H}\right)$
(rank one projectors on $\mathcal{H}$). The completeness of $\mathcal{D}_{1}\left(\mathcal{H}\right)$
follows from the fact $\mathcal{D}_{1}\left(\mathcal{H}\right)$ is
a closed subset of Hilbert of $\mathrm{HS}\left(\mathcal{H}\right)$
(endowed with the trace distance \eqref{eq:hilbert schmidt distance}).
The set of pure states $\mathcal{D}_{1}\left(\mathcal{H}\right)$
is closed in $\mathrm{HS}\left(\mathcal{H}\right)$ because it is
specified uniquely by the following conditions
\begin{equation}
\mathbb{P}\in\mbox{\ensuremath{\mathcal{D}}}_{1}\left(\mathcal{H}\right)\,\Longleftrightarrow\,\mbox{\ensuremath{\mathbb{P}}}^{2}=\mathbb{P}\,,\,\mbox{\ensuremath{\mathbb{P}}}^{\ast}=\mathbb{P}\,\text{and}\,\mbox{\ensuremath{\mathrm{tr}\left(\mathbb{P}\right)}=1}\,.\label{eq:characterisation pure states}
\end{equation}
Mappings 
\[
\mbox{\ensuremath{\mathrm{HS}\left(\mathcal{H}\right)}}\rightarrow\mathrm{HS}\left(\mathcal{H}\right)\,,\, A\rightarrow A^{2}-A\,,
\]
\[
\mbox{\ensuremath{\mathrm{HS}\left(\mathcal{H}\right)}}\rightarrow\mathrm{HS}\left(\mathcal{H}\right)\,,\,,A\rightarrow A-A^{\ast}\,,
\]
\[
\mbox{\ensuremath{\mathrm{HS}\left(\mathcal{H}\right)}}\rightarrow\mathbb{C}\,,\, A\rightarrow\mathrm{tr}\left(A\right)\,,
\]
are manifestly continuous (with respect to the trace distance \eqref{eq:hilbert schmidt distance})
and thus the set $\mbox{\ensuremath{\mathcal{D}}}_{1}\left(\mathcal{H}\right)$
is the intersection of closed subsets of $\mathrm{HS}\left(\mathcal{H}\right)$
and therefore is a closed subset of $\mathrm{HS}\left(\mathcal{H}\right)$.
\end{proof}

\subsection{Distinguishable particles\label{sub:Distinguishable-particles-inf}}

We first study entanglement of $L$ distinguishable particles, described
by the Hilbert space $\mathcal{H}_{d}=\bigotimes_{i=1}^{i=L}\mathcal{H}_{i}$,
where single particle Hilbert spaces $\mathcal{H}_{i}$ are in general
infinite dimensional. The notion of the tensor product of infinite
dimensional Hilbert spaces involves, by definition, taking into account
tensors having infinite rank, i.e. tensors that cannot be written
as a finite combination of elements of the form $\ket{\psi_{1}}\ket{\psi_{2}}\ldots\ket{\psi_{L}}$.
This phenomenon, which is the main obstacle when we want to extend
finite-dimensional results to the infinite-dimensional setting,  does
not occur when dimensions of single particle Hilbert spaces are finite.
The set of product states consists, as before, of states having the
form of simple tensors from $\mathcal{H}_{d}$,
\begin{equation}
\mathcal{M}_{d}=\left\{ \kb{\psi}{\psi}\in\mathcal{D}_{1}\left(\mathcal{H}_{d}\right)|\,\ket{\psi}=\ket{\psi_{1}}\ket{\psi_{2}}\ldots\ket{\psi_{L}},\,\ket{\psi_{i}}\in\mathcal{H}_{i}\right\} \,.\label{eq:prod infinite dim}
\end{equation}
One can identify $\mathcal{M}_{d}$ with the orbit of $K=\mathrm{U}(\mathcal{H}_{1})\times\mathrm{U}(\mathcal{H}_{2})\times\ldots\times\mathrm{U}(\mathcal{H}_{L})$
through one exemplary separable state $\kb{\psi_{0}}{\psi_{0}}$.
The main difference with the finite dimensional setting is that the
group $K$ is not a Lie group, not to mention it is compact. Therefore,
methods of representation theory of Lie group cannot be applied to
get result of the form \eqref{eq:crit prod states}. Nevertheless,
we argue that the result analogous to \eqref{eq:crit prod states}
indeed holds.
\begin{lem}
\label{lemma nfinite dim prod criterion}We have the following characterization
of product states defined on a general separable composite Hilbert
space $\mathcal{H}_{d}=\bigotimes_{i=1}^{i=L}\mathcal{H}_{i}$
\begin{equation}
\kb{\psi}{\psi}\in\mathcal{M}_{d}\Longleftrightarrow C_{d}^{2}\left(\ket{\psi}\right)=\bra{\psi}\bra{\psi}\mathbb{P}^{\mathrm{sym}}-\mathbb{P}{}_{11'}^{+}\circ\mathbb{P}_{22'}^{+}\circ\ldots\circ\mathbb{P}_{LL'}^{+}\ket{\psi}\ket{\psi}=0\,,\label{eq:dist inf dim}
\end{equation}
where $\mathbb{P}_{ii'}^{+}:\mathcal{H}_{d}\otimes\mathcal{H}_{d}\rightarrow\mathcal{H}_{d}\otimes\mathcal{H}_{d}$
are the symmetrization operators defined as under Eq. \eqref{eq:nonation distinguishable}. \end{lem}
\begin{proof}[Proof of Lemma \eqref{lemma nfinite dim prod criterion}]
In order to prove \eqref{eq:dist inf dim} we first observe that
$\bra{\psi}\bra{\psi}\mathbb{I}\otimes\mathbb{I}-\mathbb{P}_{11'}^{+}\circ\mathbb{P}_{22'}^{+}\circ\ldots\circ\mathbb{P}_{LL'}^{+}\ket{\psi}\ket{\psi}=0$
for separable states. Therefore, we only need to prove the inverse
implication. Let us denote by $\mathcal{M}_{d}^{i}$ the set of states
that are separable with respect to the bipartition $\mathcal{H}_{d}=\mathcal{H}_{i}\otimes\left(\bigotimes_{j\neq i}\mathcal{H}_{j}\right)$.
That is,
\begin{equation}
\mathcal{M}_{d}^{i}=\left\{ \left.\kb{\psi}{\psi}\in\mathcal{D}_{1}\left(\mathcal{H}_{d}\right)\right|\,\ket{\psi}\in\mathcal{H}_{i}\,,\,\ket{\phi}\in\left(\bigotimes_{j\neq i}\mathcal{H}_{j}\right)\right\} .\label{eq:sep part}
\end{equation}
One checks that $\mbox{\ensuremath{\kb{\psi}{\psi}}}\in\mathcal{M}_{d}$
if and only if $\mbox{\ensuremath{\kb{\psi}{\psi}}}\in\mathcal{M}_{d}^{i}$
for all $i=1,\ldots,L$ (in other words $\kb{\psi}{\psi}$ is separable
with respect to any bipartition $\mathcal{H}_{d}=\mathcal{H}_{i}\otimes\left(\bigotimes_{j\neq i}\mathcal{H}_{j}\right)$).
Note that in order not to complicate the notation we abuse the notation
of the tensor product in \eqref{eq:sep part} (we do not respect the
order of terms in the tensor product). The proof of the above statement
is straightforward. If $\kb{\psi}{\psi}\in\mathcal{M}_{d}^{i}$, then
$\ket{\psi}$ is an eigenvector (with eigenvalue $1$) of the operator
\begin{equation}
\kb{\phi_{i}}{\phi_{i}}\otimes\mathbb{I}_{i},\label{eq:proj 1}
\end{equation}
where $\ket{\phi_{i}}\in\mathcal{H}_{i}$ and $\mathbb{I}_{i}$ is
the identity operator on $\left(\bigotimes_{j\neq i}\mathcal{H}_{j}\right)$.
Note that we can repeat the above reasoning for all other $i=1,\ldots,L$.
As a result, we get that $\ket{\psi}$ is an eigenvector with the
eigenvalue $1$ of the operator 
\[
\mathbb{P}=\kb{\phi_{1}}{\phi_{1}}\otimes\kb{\phi_{2}}{\phi_{2}}\otimes\ldots\otimes\kb{\phi_{L}}{\phi_{L}},
\]
where $\ket{\psi_{i}}\in\mathcal{H}_{i}$. Operator $\mathbb{P}$
is a projector onto a separable state and we must have $\kb{\psi}{\psi}=\mathbb{P}$.

We can now prove that $\bra{\psi}\bra{\psi}\mathbb{I}\otimes\mathbb{I}-\mathbb{P}_{11'}^{+}\otimes\mathbb{P}_{22'}^{+}\otimes\ldots\otimes\mathbb{P}_{LL'}^{+}\ket{\psi}\ket{\psi}=0$
implies that $\kb{\psi}{\psi}$ is a product states. Assume that $\mbox{\ensuremath{\kb{\psi}{\psi}}}$
is entangled. By the discussion above it must be non-separable with
respect to some bipartition $\mathcal{H}_{i_{0}}\otimes\left(\bigotimes_{j\neq i_{0}}\mathcal{H}_{j}\right)$.
We write the Schmidt decomposition of $\ket{\psi}$ with respect to
this bipartition \citep{EisertInfDim2002},
\begin{equation}
\ket{\psi}=\sum_{l}\lambda_{l}\ket{\psi_{l}}\otimes\ket{\phi_{l}},\label{eq:shmidt form}
\end{equation}
where $\ket{\psi_{l}}\in\mathcal{H}_{i_{0}}$, $\ket{\phi_{l}}\in\left(\bigotimes_{j\neq i_{0}}\mathcal{H}_{j}\right)$
and $\bk{\psi_{i}}{\psi_{j}}=\bk{\phi_{i}}{\phi_{j}}=\delta_{ij}$.
Moreover, we fix the normalization of the sate by setting $\sum_{i}\left|\lambda_{i}\right|^{2}=1$.
We have:
\[
\bra{\psi}\bra{\psi}\mathbb{P^{\mathrm{sym}}}-\mathbb{P}_{11'}^{+}\circ\mathbb{P}_{22'}^{+}\circ\ldots\circ\mathbb{P}_{LL'}^{+}\ket{\psi}\ket{\psi}\geq\bra{\psi}\bra{\psi}\mathbb{P^{\mathrm{sym}}}-\mathbb{P}_{i_{0}i_{0}'}^{+}\otimes\mathbb{I}\otimes\ldots\otimes\mathbb{I}\ket{\psi}\ket{\psi}.
\]
Direct computation based on \eqref{eq:shmidt form} shows that $\bra{\psi}\bra{\psi}\mathbb{P}_{i_{0}i_{0}'}^{+}\otimes\mathbb{I}\otimes\ldots\otimes\mathbb{I}\ket{\psi}\ket{\psi}<1$
which implies $\bra{\psi}\bra{\psi}\mathbb{I}\otimes\mathbb{I}-\mathbb{P}_{11'}^{+}\circ\mathbb{P}_{22'}^{+}\circ\ldots\circ\mathbb{P}_{LL'}^{+}\ket{\psi}\ket{\psi}>0$.
This concludes the proof of \eqref{eq:dist inf dim}. 
\end{proof}

\subsection{Bosons\label{sub:Bosons-inf}}

A criterion analogous to \eqref{eq:dist inf dim} holds also for the
arbitrary finite number of bosonic particles with the infinite-dimensional
single particle Hilbert space. We have $\mathcal{H}_{b}=\mathrm{Sym}^{L}\left(\mathcal{H}\right)$,
where $\mathcal{H}$ is infinite dimensional. In analogy with the
finite dimensional case \eqref{eq:crit prod bos states} we distinguish
symmetric product states (bosonic coherent states):
\begin{lem}
\label{lemma nfinite dim bos criterion}We have the following characterization
of symmetric product product states defined on $\mathcal{H}_{b}=\mathrm{Sym}^{L}\left(\mathcal{H}\right)$,
\begin{equation}
\kb{\psi}{\psi}\in\mathcal{M}_{b}\Longleftrightarrow C_{b}^{2}=\bra{\psi}\bra{\psi}\mathbb{P}^{\mathrm{sym}}-\mathbb{P}_{b}\ket{\psi}\ket{\psi}=0\,,\label{eq:bos inf dim}
\end{equation}
where $\mathbb{P}^{\mathrm{sym}}$ projects onto symmetric subspace
$\mathcal{H}_{b}\otimes\mathcal{H}_{b}$ and $\mathbb{P}_{b}$ is
defined as in \eqref{eq:bos p2lambda}. \end{lem}
\begin{proof}
We notice that the coherent bosonic states are precisely the completely
symmetric separable states of the system of identical distinguishable
particles with the single particle Hilbert spaces $\mathcal{H}$.
Thus, we can apply criterion \eqref{eq:dist inf dim} restricted to
$\mathrm{Sym}^{L}\left(\mathcal{H}\right)$ to distinguish coherent
bosonic states. More precisely, for $\ket{\psi}\in\mathrm{Sym}^{L}\left(\mathcal{H}\right)$
we have
\[
\kb{\psi}{\psi}\in\mathcal{M}_{b}\Longleftrightarrow\bra{\psi}\bra{\psi}\mathbb{P^{\mathrm{sym}}}-\mathbb{P}_{11'}^{+}\circ\mathbb{P}_{22'}^{+}\circ\ldots\circ\mathbb{P}_{LL'}^{+}\ket{\psi}\ket{\psi}=0\,,
\]
where operator $\mathbb{P}_{11'}^{+}\circ\mathbb{P}_{22'}^{+}\circ\ldots\circ\mathbb{P}_{LL'}^{+}$
is assumed to act on the Hilbert space $\mathrm{Sym}^{L}\left(\mathcal{H}\right)\otimes\mathrm{Sym}^{L}\left(\mathcal{H}\right)\subset\mathcal{H}_{d}\otimes\mathcal{H}_{d}$,
with $\mathcal{H}_{d}$ defined in \ref{sub:Distinguishable-particles-inf}
and each single particle Hilbert space $\mathcal{H}_{i}$ equal to
$\mathcal{H}$. We conclude the proof by noting that for $\ket{\psi}\in\mathrm{Sym}^{L}\left(\mathcal{H}\right)$
we have 
\[
\mathbb{P}_{\left\{ 1,\ldots,L\right\} }^{\mathrm{sym}}\otimes\mathbb{P}_{\left\{ 1,\ldots,L'\right\} }^{\mathrm{sym}}\ket{\psi}\ket{\psi}=\ket{\psi}\ket{\psi}\,.
\]

\end{proof}

\subsection{Fermions\label{sub:Fermions-inf}}

The case of fermionic particles turns out to be the most demanding,
albeit also the most interesting. System of $L$ fermionic particles
is described by $\mathcal{H}_{f}=\bigwedge^{L}\left(\mathcal{H}\right)$,
where the single particle Hilbert space $\mathcal{H}$ is infinite
dimensional. We define ``non-entangled'' or coherent fermionic states
analogously to the finite dimensional case,
\begin{equation}
\mathcal{M}_{f}=\left\{ \kb{\psi}{\psi}\in\mathcal{D}_{1}\left(\mathcal{H}_{f}\right)\,|\,\ket{\psi}=\ket{\phi_{1}}\wedge\ket{\phi_{2}}\wedge\ldots\wedge\ket{\phi_{L}},\,\ket{\phi_{i}}\in\mathcal{H},\,\bk{\phi_{i}}{\phi_{j}}=\delta_{ij}\right\} \,.
\end{equation}
We prove that the criterion \eqref{eq:criterion ferm} holds also
in the infinite dimensional situation.
\begin{lem}
\label{lemma: ferm inf dimensional}We have the following characterization
of Slater determinants on a general separable Hilbert space $\mathcal{H}_{f}=\bigwedge^{L}\left(\mathcal{H}\right)$,
\begin{equation}
\kb{\psi}{\psi}\in\mathcal{M}_{f}\Longleftrightarrow C_{f}^{2}\left(\kb{\psi}{\psi}\right)=\bra{\psi}\bra{\psi}\mathbb{P^{\mathrm{sym}}}-\alpha\mathbb{P}_{11'}^{+}\circ\mathbb{P}_{22'}^{+}\circ\ldots\circ\mathbb{P}_{LL'}^{+}\ket{\psi}\ket{\psi}=0\,,\label{eq:ferm-crit- infinite}
\end{equation}
where $\alpha=\frac{2^{l}}{l+1}$ and $\mathbb{P^{\mathrm{sym}}}$
projects onto symmetric subspace $\mathcal{H}_{f}\otimes\mathcal{H}_{f}$.
It is assumed that $\mathbb{P}_{11'}^{+}\circ\mathbb{P}_{22'}^{+}\circ\ldots\circ\mathbb{P}_{LL'}^{+}$
acts on $\bigwedge^{L}\left(\mathcal{H}\right)\otimes\bigwedge^{L}\left(\mathcal{H}\right)\subset\mathcal{H}_{d}\otimes\mathcal{H}_{d}$,
with $\mathcal{H}_{d}$ defined in \ref{sub:Distinguishable-particles-inf}
and each single particle Hilbert space $\mathcal{H}_{i}$ equal to
$\mathcal{H}$.\end{lem}
\begin{proof}
In order to prove \eqref{eq:ferm-crit- infinite} we consider the
equivalent problem,
\begin{equation}
\kb{\psi}{\psi}\in\mathcal{M}_{f}\Longleftrightarrow\bra{\psi}\bra{\psi}\mathbb{P}_{f}\ket{\psi}\ket{\psi}=1\,,\label{eq:ferm-inf-form}
\end{equation}
for $\mathbb{P}_{f}=\alpha\mathbb{P}_{11'}^{+}\circ\mathbb{P}_{22'}^{+}\circ\ldots\circ\mathbb{P}_{LL'}^{+}$
and a normalized $\ket{\psi}\in\bigwedge^{L}\left(\mathcal{H}\right)$.
Note that if the rank of $\ket{\psi}$ (i.e. the minimal number of
elements of the form $\ket{\phi_{1}}\wedge\ket{\phi_{2}}\wedge\ldots\wedge\ket{\phi_{L}}$
needed to express $\ket{\psi}$) is finite, we have $\ket{\psi}\in\bigwedge^{L}\left(\mathcal{H}_{0}\right)$,
where $\mathcal{H}_{0}$ is some finite dimensional subspace of $\mathcal{H}$.
Therefore, in this case \eqref{eq:ferm-inf-form} is proven as we
can apply results from Subsection \ref{sub:Slater-determinants}.
If rank of $\ket{\psi}$ is infinite and $\bra{\psi}\bra{\psi}\mathbb{P}_{f}\ket{\psi}\ket{\psi}<1$
there is nothing to prove. The only case left is when $\ket{\psi}$
has infinite rank and $\bra{\psi}\bra{\psi}\mathbb{P}_{f}\ket{\psi}\ket{\psi}=1$.
In the course of argumentation we will need the fact that the set
of coherent fermionic states $\mathcal{M}_{f}$ is closed in $\mathcal{D}_{1}\left(\bigwedge^{L}\left(\mathcal{H}\right)\right)$
(with respect to the metric \eqref{eq:metric structure}). In order
to prove this we identify $\mathcal{M}_{f}$ with the set $\mathcal{D}_{L}\left(\mathcal{H}\right)$
consisting of orthogonal projectors onto $L$ dimensional subspaces
of $\mathcal{H}$. The bijection $\alpha:\mathcal{M}\rightarrow\mathcal{D}_{L}\left(\mathcal{H}\right)$
is defined as follows. For $\kb{\psi}{\psi}\in\mathcal{M}_{f}$ we
have
\begin{equation}
\alpha\left(\kb{\psi}{\psi}\right)=\mbox{\ensuremath{\mathbb{P}}}_{\mathrm{Span}\left\{ \ket{\phi_{1}},\ldots,\ket{\phi_{L}}\right\} }\,,\label{eq:identification of slaters}
\end{equation}
where%
\footnote{Mapping $\alpha$ does not depend upon the choice of orthonormal vectors
$\ket{\phi_{1}},\ldots,\ket{\phi_{L}}$ such that $\ket{\psi}=\ket{\phi_{1}}\wedge\ket{\phi_{2}}\wedge\ldots\wedge\ket{\phi_{L}}.$%
} $\ket{\psi}=\ket{\phi_{1}}\wedge\ket{\phi_{2}}\wedge\ldots\wedge\ket{\phi_{L}}$
and $\mbox{\ensuremath{\mathbb{P}}}_{\mathrm{Span}\left\{ \ket{\phi_{1}},\ldots,\ket{\phi_{L}}\right\} }:\mathcal{H}\rightarrow\mathcal{H}$
is an orthogonal projector onto $L$ dimensional subspace $\mathrm{Span}\left\{ \ket{\phi_{1}},\ldots,\ket{\phi_{L}}\right\} $.
On the set $\mathcal{D}_{L}\left(\mathcal{H}\right)$ we have the
metric induced from the trace distance on $\mathrm{HS}\left(\mathcal{\mathcal{H}}\right)$
\begin{equation}
\mathrm{d_{L}}\left(\mathbb{P},\mathbb{P}^{'}\right)=\sqrt{2\left[L-\mathrm{tr}\left(\mathbb{P}\mathbb{P}'\right)\right]}\,.\label{eq:metric L inf}
\end{equation}
Using arguments analogous to the proof of Proposition \ref{infinite dimensional metric}
we see that the set $\mathcal{D}_{L}\left(\mathcal{H}\right)$ is
a complete with respect to the metric \eqref{eq:metric L inf}. The
closedness of $\mathcal{M}_{f}$ with respect to the natural metric
\eqref{eq:metric structure} on $\mathcal{D}_{1}\left(\bigwedge^{L}\left(\mathcal{H}\right)\right)$
follows now the following inequality
\[
\mathrm{d}_{L}\left(\alpha\left(\kb{\psi}{\psi}\right),\alpha\left(\kb{\psi'}{\psi'}\right)\right)\leq\sqrt{L}\mathrm{d}\left(\kb{\psi}{\psi},\kb{\psi'}{\psi'}\right)
\]
which is a direct consequence of  the definitions of the mapping $\alpha$
as well as the metrics $\mathrm{d}$ and $\mathrm{d}_{L}$.

We can now return to the original problem. We introduce the sequence
of finite dimensional subspaces
\begin{equation}
\mathcal{H}_{1}\subset\mathcal{H}_{2}\subset\ldots\subset\mathcal{H}_{k}\subset\ldots\subset\mathcal{H}\,,\label{eq:sequance subspaces}
\end{equation}
such that $\bigcup_{l=1}^{l=\infty}\mathcal{H}_{i}=\mathcal{H}$.
To the above sequence we associate corresponding sequence of subspaces
of $\bigwedge^{L}\left(\mathcal{H}\right)$,
\begin{equation}
\bigwedge^{L}\left(\mathcal{H}_{1}\right)\subset\bigwedge^{L}\left(\mathcal{H}_{2}\right)\subset\ldots\subset\bigwedge^{L}\left(\mathcal{H}_{k}\right)\subset\ldots\subset\bigwedge^{L}\left(\mathcal{H}\right)\,.
\end{equation}
Obviously we have $\bigcup_{l=1}^{l=\infty}\bigwedge^{L}\left(\mathcal{H}_{l}\right)=\bigwedge^{L}\left(\mathcal{H}\right)$.
We fix the index $k$ and consider the following orthogonal splittings
of $\bigwedge^{L}\left(\mathcal{H}\right)$ and $\bigwedge^{L}\left(\mathcal{H}\right)\otimes\bigwedge^{L}\left(\mathcal{H}\right)$,
\begin{equation}
\bigwedge^{L}\left(\mathcal{H}\right)=\bigwedge^{L}\left(\mathcal{H}_{k}\right)\oplus\left[\bigwedge^{L}\left(\mathcal{H}_{k}\right)\right]^{\perp},\label{eq:single split}
\end{equation}
\begin{equation}
\bigwedge^{L}\left(\mathcal{H}\right)\otimes\bigwedge^{L}\left(\mathcal{H}\right)=\left[\bigwedge^{L}\left(\mathcal{H}_{k}\right)\otimes\bigwedge^{L}\left(\mathcal{H}_{k}\right)\right]\oplus\left[\bigwedge^{L}\left(\mathcal{H}_{k}\right)\otimes\bigwedge^{L}\left(\mathcal{H}_{k}\right)\right]^{\perp},\label{eq:double split}
\end{equation}
where the orthogonal complements are taken with respect to the usual
inner products on $\bigwedge^{L}\left(\mathcal{H}\right)$ and $\bigwedge^{L}\left(\mathcal{H}\right)\otimes\bigwedge^{L}\left(\mathcal{H}\right)$
respectively. By $\mathbb{P}_{k}:\bigwedge^{L}\left(\mathcal{H}\right)\rightarrow\bigwedge^{L}\left(\mathcal{H}_{k}\right)$
we denote the orthogonal projector on $\bigwedge^{L}\left(\mathcal{H}_{k}\right)$.
We now prove that the infinite rank of $\ket{\psi}$ and $\bra{\psi}\bra{\psi}\mathbb{P}_{f}\ket{\psi}\ket{\psi}=1$
yield to the contradiction. Let us first note that for a normalized
$\ket{\psi}$ the condition $\bra{\psi}\bra{\psi}\mathbb{P}_{f}\ket{\psi}\ket{\psi}=1$
is equivalent to $\mathbb{P}_{f}\ket{\psi}\ket{\psi}=\ket{\psi}\ket{\psi}$.
Consider the decomposition
\begin{equation}
\ket{\psi}\ket{\psi}=\ket{\Psi_{k}}+\ket{\Psi_{k}^{\perp}},
\end{equation}
where $\ket{\Psi_{k}}\in\bigwedge^{L}\left(\mathcal{H}_{k}\right)\otimes\bigwedge^{L}\left(\mathcal{H}_{k}\right)$
and $\ket{\Psi_{k}^{\perp}}\in\left[\bigwedge^{L}\left(\mathcal{H}_{k}\right)\otimes\bigwedge^{L}\left(\mathcal{H}_{k}\right)\right]^{\perp}$.
We have $\mathbb{P}_{f}\ket{\psi}\ket{\psi}=\ket{\psi}\ket{\psi}$
and thus 
\begin{equation}
\mathbb{P}_{f}\ket{\Psi_{k}}+\mathbb{P}_{f}\ket{\Psi_{k}^{\perp}}=\ket{\Psi_{k}}+\ket{\Psi_{k}^{\perp}}\,.
\end{equation}
Because $\mathbb{P}_{f}$ preserves $\bigwedge^{L}\left(\mathcal{H}_{k}\right)\otimes\bigwedge^{L}\left(\mathcal{H}_{k}\right)$
we have $\bra{\Psi_{k}^{\perp}}\mathbb{P}_{f}\ket{\Psi_{k}}=0$ and
therefore $\mathbb{P}_{f}\ket{\Psi_{k}}=\ket{\Psi_{k}}$. Notice that
$\ket{\Psi_{k}}=\mathbb{P}_{k}\otimes\mathbb{P}_{k}\left(\ket{\psi}\ket{\psi}\right)$
and therefore $\ket{\Psi_{k}}$ is a product state. Because $\mathbb{P}_{f}\ket{\Psi_{k}}=\ket{\Psi_{k}}$
we see that $\mathbb{P}_{k}\ket{\psi}$ is actually an (non-normalized)
representative of some coherent fermionic state. We can repeat the
above construction for the arbitrary number $k$. We get
\begin{equation}
1=\lim_{k\rightarrow\infty}\bra{\psi}\mathbb{P}_{k}\ket{\psi},\,\label{eq:limit}
\end{equation}
where $\mathbb{P}_{k}\ket{\psi}$ is the (non-normalized) representative
of some coherent state. We therefore get the convergence (in a sense
of \eqref{eq:metric structure}) 
\[
\frac{1}{\bra{\psi}\mathbb{P}_{k}\ket{\psi}}\mathbb{P}_{k}\kb{\psi}{\psi}\mathbb{P}_{k}\overset{k\rightarrow\infty}{\rightarrow}\kb{\psi}{\psi}\,.
\]
Since 
\[
\frac{1}{\bra{\psi}\mathbb{P}_{k}\ket{\psi}}\mathbb{P}_{k}\kb{\psi}{\psi}\mathbb{P}_{k}\in\mathcal{M}_{f}\,,
\]
we get that $\kb{\psi}{\psi}\in\mathcal{M}_{f}$ as the set of coherent
fermionic states $\mathcal{M}_{f}$ is closed in $\mathcal{D}_{1}\left(\mathcal{H}_{f}\right)$.
This is clearly in contradiction with the assumption that $\ket{\psi}$
has infinite rank. 
\end{proof}

\section{Multilinear characterization of correlations in pure states\label{sec:Multilinear-characterization-of-pure}}

In this section we will extend the framework introduced in previous
parts of the chapter. We will study polynomials of degree higher then
two in the density matrix. It will be showed that polynomials of this
kind do not only capture correlation types discussed previously but
also cover more complicated types of correlations (e.g., genuine multipartite
entanglement or correlations based on the notion of Schmidt number). 

\begin{figure}[h]
\centering{}\includegraphics[width=8cm]{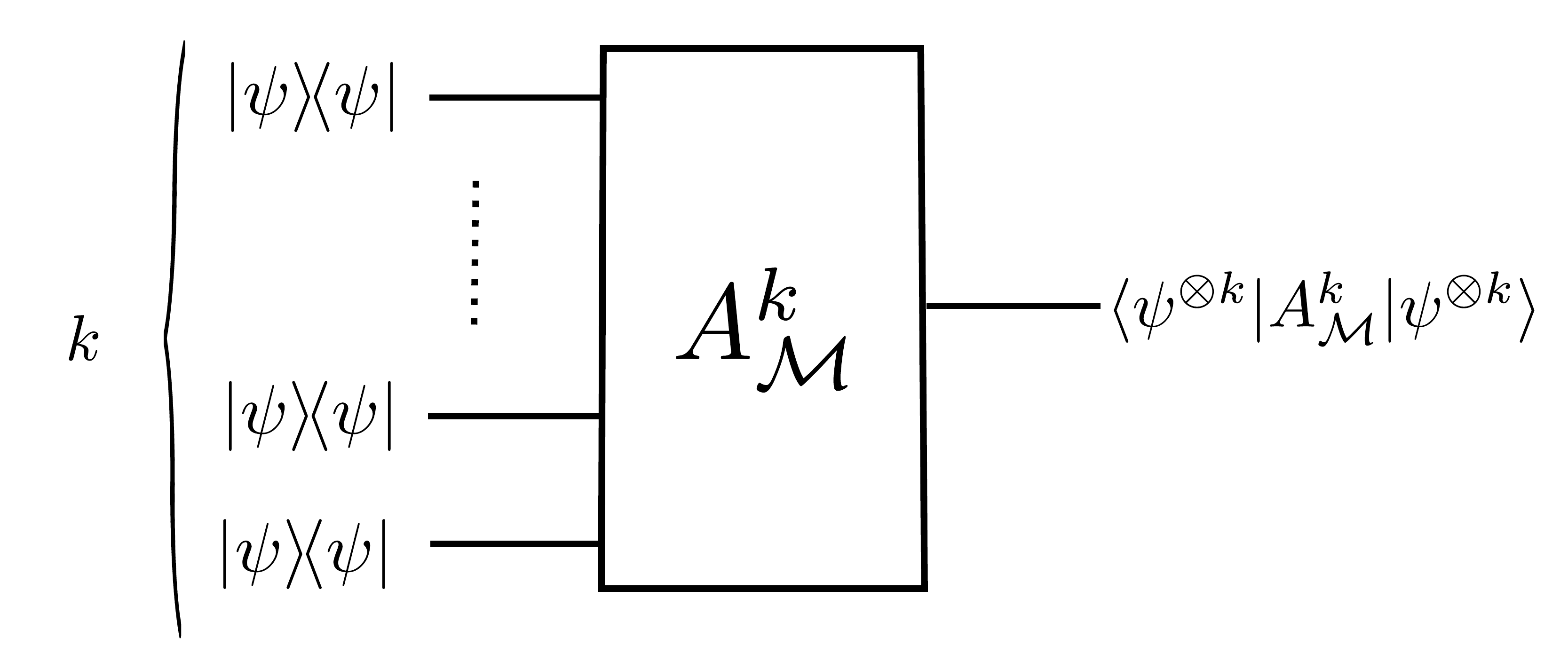}\protect\caption{\label{fig:Pictorial-representation-of-Aktimes}Pictorial representation
of the expression $\mbox{\ensuremath{\protect\bra{\psi^{\otimes k}}}}A_{\mathcal{M}}^{\left(k\right)}\protect\ket{\psi^{\otimes k}}$
defining, via Eq. \eqref{eq:multilinear characterisation}, the class
of non-correlated $\mathcal{M}$.}
\end{figure}

Let $\mathcal{H}$ be the arbitrary finite-dimensional Hilbert space
describing some physical system%
\footnote{Discussion presented here can be generalized, under mild technical
assumptions, to the infinite-dimensional case. %
}. In what follows we will present of classes of pure states $\mbox{\ensuremath{\mathcal{M}}}\subset\mathcal{D}_{1}\left(\mathcal{H}\right)$
that can be described via the following polynomial condition
\begin{equation}
\ket{\psi}\in\mathcal{M}\,\Longleftrightarrow\bra{\psi^{\otimes k}}\, A_{\mathcal{M}}^{\left(k\right)}\ket{\psi^{\otimes k}}=0\,,\label{eq:multilinear characterisation}
\end{equation}
where $A_{\mathcal{M}}^{\left(k\right)}:\mathcal{H}^{\otimes k}\rightarrow\mathcal{H}^{\otimes k}$
is a Hermitian operator acting on $k$-fold tensor product of $\mathcal{H}$.
The subscript $\mathcal{M}$ expresses the fact that different operators
$A_{\mathcal{M}}^{k}$ define, via \ref{eq:multilinear characterisation},
different classes of pure states $\mathcal{M}$. Superscript $\left(k\right)$
denotes the number of copies of $\mathcal{H}$ which is needed to
implement $k$. In what follows, unless it causes an ambiguity, we
will drop both the subscript and the superscript when referring to
$A_{\mathcal{M}}^{\left(k\right)}$. In accordance to the terminology
introduced previously we will refer to the class of states $\mathcal{M}$
as a class of not-correlated states. Because $\ket{\psi^{\otimes k}}\in\mathrm{Sym}^{k}\left(\mathcal{H}\right)$
($\mathrm{Sym}^{k}\left(\mathcal{H}\right)$ is a subspace of $\mathcal{H}^{\otimes k}$
completely symmetric under the exchange of factors in the multiple
tensor product) we can, assume without loss of generality that operator
$A$ is actually defined on $\mathrm{Sym}^{k}\left(\mathcal{H}\right)$.
We will also assume%
\footnote{Note that by imposing such a condition we do not lose the generality
as the condition $\eqref{eq:multilinear characterisation}$ is invariant
under scaling. %
} that $A\leq\mathbb{P}^{\mathrm{sym},k}$, where $\mathbb{P}^{\mathrm{sym},k}:\mathcal{H}^{\otimes k}\rightarrow\mathrm{Sym}^{k}\left(\mathcal{H}\right)$
denotes the orthogonal projector onto $\mathrm{Sym}^{k}\left(\mathcal{H}\right)$.
An interest in classes of non-correlated pure states described by
\eqref{eq:multilinear characterisation} is motivated by the following
reasons:
\begin{itemize}
\item Condition \eqref{eq:multilinear characterisation} it is a natural
generalization of \eqref{eq:criterion ver 1}.
\item Via condition \eqref{eq:multilinear characterisation} it is possible
to describe all classes of generalized coherent sates discussed before
in this Chapter.
\item The class of pure states $\mathcal{M}$ does not have to be an orbit
of some symmetry group represented in $\mathcal{H}$. 
\item In Chapter \ref{chap:Polynomial-mixed states} we will show that the
characterization \eqref{eq:multilinear characterisation} can be used
to derive polynomial criteria for detection of correlated (with respect
to the class $\mathcal{H}$) mixed states.
\end{itemize}
This Section is organized as follows. First we present the explicit
form of operators $A$ that characterize
\begin{itemize}
\item Multipartite pure states that do not exhibit ``genuine multiparty
entanglement'' (GME) \citep{EntantHoro,Guehne2009}; 
\item Bipartite states with Schmidt rank bounded above by some some number
$n$ \citep{SchmidtNumHoro,SchmidtNumberLew}.
\end{itemize}
For definitions of above concepts consult Subsections \eqref{sub:GME}
and \eqref{sub: Schmid rank} respectively. One can think of correlations
defined by these classes of states as of possible generalizations
of the usual notion of entanglement. 

In the second part of the Section we prove that generalized coherent
states discussed above can be equivalently characterized by a general
multilinear condition \eqref{eq:multilinear characterisation}. We
give an explicit form of operators $A$, acting on $\mathrm{Sym}^{k}\left(\mathcal{H}\right)$
that characterize previously discussed classes of states: product
states, product bosonic states and Slater determinants.

\subsection{States that do not exhibit GME\label{sub:GME}}

Let us first recall the concept of genuine multiparty entanglement. 
\begin{defn}
\label{def L-party sep}Consider a Hilbert space associated to the
system of $L$ distinguishable particles, $\mathcal{H}_{d}=\mathcal{H}_{1}\otimes\ldots\otimes\mathcal{H}_{L}$.
We say that a pure state $\kb{\psi}{\psi}\in\mathcal{D}_{1}\left(\mathcal{H}_{d}\right)$
is $2$-separable if and only if it can be written as a product state
with respect to some bipartition of the Hilbert space $\mathcal{H}_{d}$.
In other words, 
\begin{equation}
\ket{\psi}=\ket{\phi_{1}}_{X}\otimes\ket{\phi_{1}}_{\bar{X}}\label{eq:GME definition}
\end{equation}
 for some subset $X\subset\left\{ 1,2,\ldots,L\right\} $ which define
splitting of the system of $L$ particles onto two subsystems (see
Figure \ref{fig:bipartition}):
\begin{equation}
\left\{ 1,2,\ldots,L\right\} =X\cup\bar{X},\,\mathcal{H}_{d}=\left(\bigotimes_{i\in X}\mathcal{H}_{i}\right)\otimes\left(\bigotimes_{i\in\bar{X}}\mathcal{H}_{i}\right)\,.\label{eq:bipartition}
\end{equation}
Note that in \eqref{eq:GME definition} and \eqref{eq:bipartition}
we reordered terms in the tensor product appearing in the definition
of $\mathcal{H}_{d}$. In what follows we will denote the set of $2$-separable
states by $\mathcal{M}_{d}^{2}$. A state $\ket{\psi}\in\mathcal{H}_{d}$
will be called called genuinely multiparty entangled if and only if
it is not $2$-separable. 
\end{defn}
\begin{figure}[h]
\centering{}\includegraphics[width=6cm]{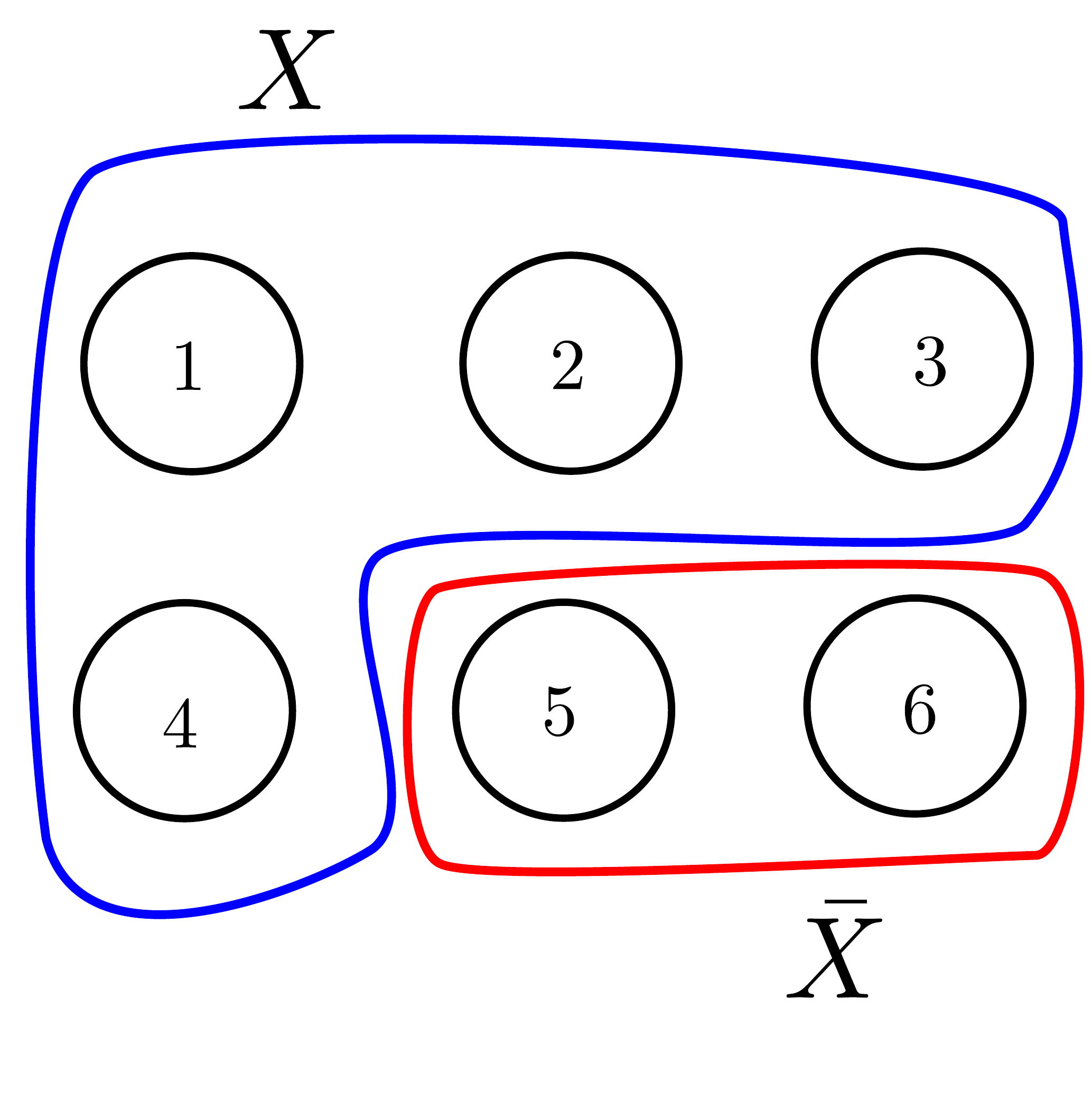}\protect\caption{\label{fig:bipartition}Exemplary bipartition of a composite system
consisting of six subsystems.}
\end{figure}

A class of ``non-correlated'' states $\mathcal{M}_{d}^{2}\subset\mathcal{D}_{1}\left(\mathcal{H}_{d}\right)$
is a natural generalization of separable states from the bipartite
setting: it consists of states that are separable with respect to
some bipartition of the composite system onto two subsystem. Let us
remark that $\mathcal{M}_{d}^{L2}$ is invariant under the action
of local unitary group but is not an orbit of this group in $\mathcal{D}_{1}\left(\mathcal{H}_{d}\right)$.
Let us consider a specific example of $L=3$ Indeed, for $L=3$ we
have $\mathcal{H}_{d}=\mathcal{H}_{1}\otimes\mathcal{H}_{2}\otimes\mathcal{H}_{3}$
and
\begin{equation}
\mathcal{M}_{d}^{2}=\mathcal{M}_{d}^{1:23}\cup\mathcal{M}_{d}^{2:13}\cup\mathcal{M}_{d}^{3:12}\,,\label{eq:decomposition GME}
\end{equation}
where $\mathcal{M}_{d}^{X:YZ}\subset\mathcal{D}_{1}\left(\mathcal{H}_{A}\otimes\mathcal{H}_{B}\otimes\mathcal{H}_{C}\right)$
describes the class of the splitting of a composite system in the
cut $X:YZ$. The graphical presentation of the set $\mathcal{M}_{d}^{2}\subset\mathcal{D}_{1}\left(\mathcal{H}_{1}\otimes\mathcal{H}_{2}\otimes\mathcal{H}_{3}\right)$
is presented in Figure \ref{fig:GME pure3}. As $\mathcal{M}_{d}^{X:YZ}$
itself consists of many orbits of local unitary group then the same
concerns $\mathcal{M}_{d}^{2}$. 

\begin{figure}[h]
\centering{}\includegraphics[width=6cm]{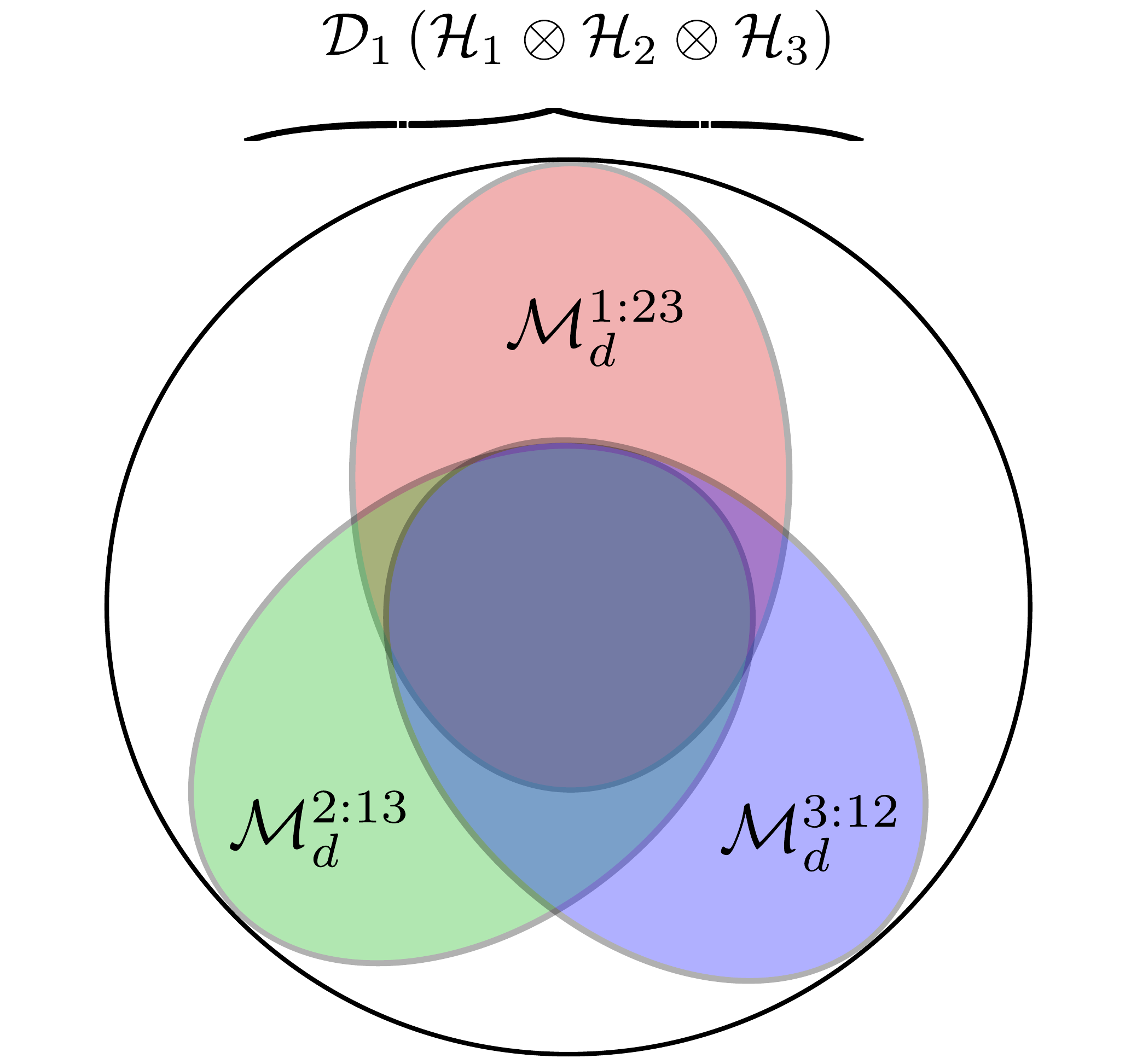}\protect\caption{\label{fig:GME pure3}Graphical presentation of the structure of the
set $\mathcal{M}_{d}^{2}=\mathcal{M}_{d}^{1:23}\cup\mathcal{M}_{d}^{2:13}\cup\mathcal{M}_{d}^{3:12}\,\subset\mathcal{D}_{1}\left(\mathcal{H}_{1}\otimes\mathcal{H}_{2}\otimes\mathcal{H}_{3}\right)$.}
\end{figure}

We now present a polynomial condition that characterizes the class
of $L$-particle separable states. For simplicity we restrict our
considerations to the special case of $L=3$ subsystems. Arguments
presented below can be easily extended to the arbitrary number of
subsystems. 
\begin{lem}
\label{GME pure}The class of pure states $\mathcal{M}_{d}^{2}\subset\mathcal{D}_{1}\left(\mathcal{H}_{1}\otimes\mathcal{H}_{2}\otimes\mathcal{H}_{3}\right)$
that do not exhibit GME can be characterized by the operator $A_{\mathrm{GME}}$
acting on $\mathrm{Sym}^{6}\left(\mathcal{H}_{1}\otimes\mathcal{H}_{2}\otimes\mathcal{H}_{3}\right),$
\begin{equation}
\ket{\psi}\in\mathcal{M}_{d}^{3}\Longleftrightarrow\bra{\psi{}^{\otimes6}}A_{\mathrm{GME}}\ket{\psi{}^{\otimes6}}=0\,.\label{eq:characterisation GME}
\end{equation}
Operator $A_{\mathrm{GME}}$ is given by
\begin{equation}
A_{\mathrm{GME}}=\mbox{\ensuremath{\mathbb{P}}}^{\mathrm{sym,6}}\circ\left(A^{1:23}\otimes A^{2:13}\otimes A^{3:12}\right)\circ\mbox{\ensuremath{\mathbb{P}}}^{\mathrm{sym,6}}\,,\label{eq:operator gme}
\end{equation}
where
\[
\mbox{\ensuremath{\mathbb{P}}}^{\mathrm{sym,6}}\,\text{- projector onto completelly symmetric sybspace of }\mbox{\ensuremath{\left(\mathcal{H}_{1}\otimes\mathcal{H}_{2}\otimes\mathcal{H}_{3}\right)}}^{\otimes6}\,,
\]

\begin{equation}
A^{1:23}=\mathbb{P}_{11'}^{-}\circ\mathbb{P}_{\left\{ 23\right\} \left\{ 2'3'\right\} }^{-}\label{eq:splitting explain}
\end{equation}
and the remaining operators $A^{X:YZ}$ are defined in the analogous
manner. %
\footnote{In expression \eqref{eq:splitting explain} we extended in a natural
fashion the notation introduced in Subsection \ref{sub:Product-states}.
The operator $\mathbb{P}_{\left\{ 23\right\} \left\{ 2'3'\right\} }^{-}$
denotes the orthogonal projector onto the subspace $\left(\mathcal{H}_{1}^{\otimes2}\right)\otimes\bigwedge\left(\mathcal{H}_{2}\otimes\mathcal{H}_{3}\right)\subset\mathrm{Sym}^{2}\left(\mathcal{H}_{1}\otimes\mathcal{H}_{2}\otimes\mathcal{H}_{3}\right)$. %
}.\end{lem}
\begin{proof}
The proof follows immediately from \eqref{eq:decomposition GME} and
from the fact that
\[
\bra{\psi{}^{\otimes6}}A_{\mathrm{GME}}\ket{\psi{}^{\otimes6}}=C_{1:23}^{2}\left(\ket{\psi}\right)\cdot C_{2:13}^{2}\left(\ket{\psi}\right)\cdot C_{3:12}^{2}\left(\ket{\psi}\right)\,,
\]
where by $C_{X:YZ}^{2}$ we denote the quadratic polynomial that characterizes
the separability (in a sense of Eq. \eqref{eq:crit prod states} )
of a state $\ket{\psi}\in\mathcal{D}_{1}\left(\mathcal{H}_{1}\otimes\mathcal{H}_{2}\otimes\mathcal{H}_{3}\right)$
with respect to the bipartition $X:YZ$. 
\end{proof}

\subsection{Bipartite pure states with the bounded Schmidt rank\label{sub: Schmid rank}}

Consider a system of two distinguishable particles with a Hilbert
space $\mathcal{H}_{A}\otimes\mathcal{H}_{B}$. Let $d=\mathrm{max}\left\{ \mathrm{dim}\mathcal{H}_{A},\mathrm{dim}\mathcal{H}_{b}\right\} $.
Every pure state $\kb{\psi}{\psi}\in\mathcal{D}_{1}\left(\mathcal{H}_{A}\otimes\mathcal{H}_{B}\right)$
admits a Schmidt decomposition \citep{NielsenChaung2010},
\begin{equation}
\mbox{\ensuremath{\ket{\psi}}}=\sum_{i=1}^{i=d}\lambda_{i}\ket{e_{i}}\ket{f_{i}}\,,\label{eq:schmidt decomposition}
\end{equation}
where $\lambda_{i}\geq0$, $\sum_{i=1}^{i=d}\lambda_{i}^{2}=1$. Moreover,
sets $\left\{ \ket{e_{i}}\right\} _{i=1}^{i=d}$ and $\left\{ \ket{f_{i}}\right\} _{i=1}^{i=d}$
are orthogonal normalized vectors in $\mathcal{H}_{A}$ and respectively
in $\mathcal{H}_{B}$.
\begin{defn}
\label{def:Schmidt number}A Schmidt rank of a pure state $\ket{\psi}$
is a number of nonzero coefficients in \eqref{eq:schmidt decomposition}.
We define the class of bipartite pure states $\mathcal{M}_{n}$ as
the set of states that have Schmidt number less or equal than $n$
($1\leq n\leq d$),
\begin{equation}
\mathcal{M}_{n}=\left\{ \ket{\psi}\in\mathcal{H}_{A}\otimes\mathcal{H}_{B}\,|\,\text{at most \ensuremath{n}of \ensuremath{\lambda_{i}}from\,\eqref{eq:schmidt decomposition} are nonzero}\right\} \,.\label{eq:non-correlated schmidt}
\end{equation}
We have the following chain of inclusions (see Figure \ref{fig:k-separability}for
a graphical representation),
\begin{equation}
\mathcal{M}_{d}=\mathcal{M}_{1}\subset\mathcal{M}_{2}\subset\ldots\subset\mathcal{M}_{d}=\mathcal{D}_{1}\left(\mathcal{H}_{A}\otimes\mathcal{H}_{B}\right)\,.\label{eq:chain1}
\end{equation}
\begin{figure}[h]
\centering{}\includegraphics[width=6.5cm]{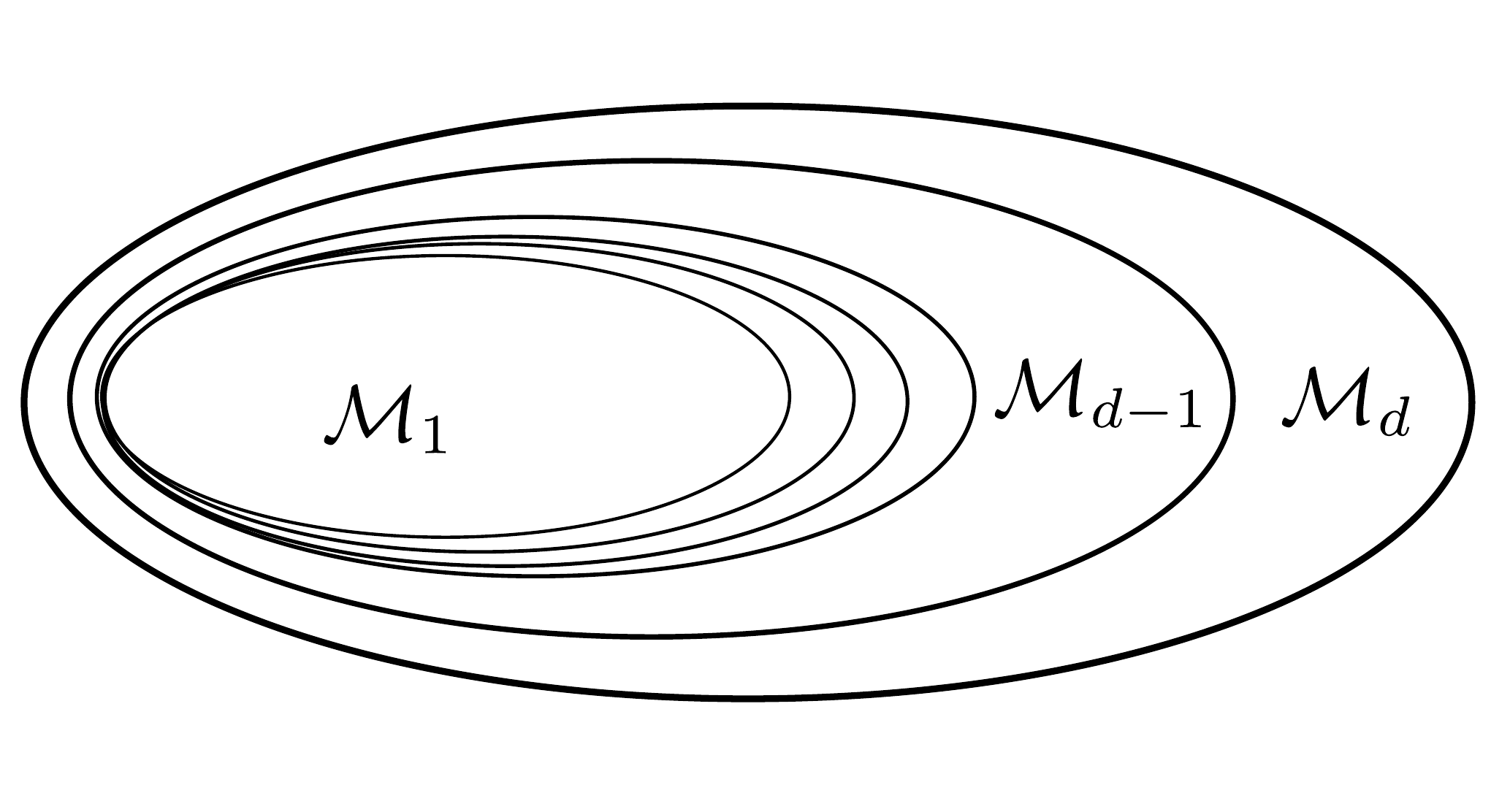}\protect\caption{\label{fig:k-separability}Graphical presentation of the hierarchy
of inclusions \eqref{eq:chain1}.}
\end{figure}

\end{defn}
In what follows we use the following isomorphism of multiple tensor
powers of $\left(\mathcal{H}_{A}\otimes\mathcal{H}_{B}\right)^{\otimes k},$
\begin{equation}
\left(\mathcal{H}_{A}\otimes\mathcal{H}_{B}\right)^{\otimes k}\approx\left(\mathcal{H}_{A}^{\otimes k}\right)\otimes\left(\mathcal{H}_{B}^{\otimes k}\right)\,.\label{eq:identification of tensor power}
\end{equation}
We can now state the advertised polynomial characterization of bipartite
pure states with the Schmidt rank equal at most $n$.
\begin{lem}
\label{characterization schmidt rank}The class of pure states $\mathcal{M}_{n}$
that have a Schmidt rank at most $n$ (see \eqref{eq:non-correlated schmidt})
can be characterized by the operator $A_{n}$ acting on $\mathrm{Sym}^{n+1}\left(\mathcal{H}_{A}\otimes\mathcal{H}_{B}\right),$
\begin{equation}
\kb{\psi}{\psi}\in\mathcal{M}_{n}\Longleftrightarrow\bra{\psi{}^{\otimes\left(n+1\right)}}A_{n}\ket{\psi{}^{\otimes\left(n+1\right)}}=0\,.\label{eq:characterisation of schmidt}
\end{equation}
Operator $A_{k}$ is given by
\begin{equation}
A_{n}=\mathbb{P}_{A}^{\mathrm{asym,n+1}}\otimes\mathbb{P}_{B}^{\mathrm{asym,n+1}}\,,\label{eq:operator schmidt}
\end{equation}
where
\[
\mathbb{P}_{A}^{\mathrm{asym,n+1}}\,\text{ -is a projector onto the subspace }\bigwedge^{n+1}\left(\mathcal{H}_{A}\right)\otimes\left(\mathcal{H}_{B}^{\otimes\left(n+1\right)}\right)\subset\left(\mathcal{H}_{A}\otimes\mathcal{H}_{B}\right)^{\otimes(n+1)}\,,
\]
\[
\mathbb{P}_{b}^{\mathrm{asym,n+1}}\,\text{ -is a projector onto the subspace }\left(\mathcal{H}_{A}^{\otimes(n+1)}\right)\otimes\bigwedge^{n+1}\left(\mathcal{H}_{B}\right)\subset\left(\mathcal{H}_{A}\otimes\mathcal{H}_{B}\right)^{\otimes(n+1)}\,.
\]
 Note that in the above we have used identification \eqref{eq:identification of tensor power}.\end{lem}
\begin{proof}
For a given $\ket{\psi}\in\mathcal{M}_{n}$ expanding its $n+1$ copies,
$\ket{\psi{}^{\otimes\left(n+1\right)}}$, in the space $\left(\mathcal{H}_{A}^{\otimes\left(n+1\right)}\right)\otimes\left(\mathcal{H}_{B}^{\otimes\left(n+1\right)}\right)$
gives terms of the form
\[
\left(\ket{e_{i_{1}}}\otimes\ldots\otimes\ket{e_{i_{k+1}}}\right)\otimes\left(\ket{f_{i_{1}}}\otimes\ldots\otimes\ket{f_{i_{k+1}}}\right)\,,
\]
 where at least one pair of indices overlap. Therefore, by definition
of the operator $A_{n}$, we have $A_{n}\ket{\psi{}^{\otimes\left(n+1\right)}}=0$.
Conversely, one easily checks that that for $\ket{\psi}\notin\mathcal{M}_{n}$
we have
\[
\bra{\psi{}^{\otimes\left(n+1\right)}}A\ket{\psi{}^{\otimes\left(n+1\right)}}\neq0\,.
\]
Operator $A_{n}$ has a support on $\mathrm{Sym}^{n+1}\left(\mathcal{H}_{A}\otimes\mathcal{H}_{B}\right)$.
It follows from the fact that for any $\tau\in S_{n+1}$, represented
in a natural manner (see Subsection \ref{sub:Representation-theory-of}
for details) on $\left(\mathcal{H}_{A}\otimes\mathcal{H}_{B}\right)^{\otimes n+1}$,
the following identity holds
\begin{equation}
A_{n}\tau=\tau A_{n}=A_{n}\,.\label{eq:permutation invariance}
\end{equation}
The above expression can be verified using the structure of $A_{n}$
and observing that $\tau=\tau_{A}\otimes\tau_{B}$, where $\tau_{A(B)}$
denote the permutation on Alice’s (Bob's) registers under the identification
\eqref{eq:identification of tensor power}. 
\end{proof}
Let us remark that from the definition of operators $A_{n}:\mathrm{Sym}^{n+1}\left(\mathcal{H}_{A}\otimes\mathcal{H}_{B}\right)\rightarrow\mathrm{Sym}^{n+1}\left(\mathcal{H}_{A}\otimes\mathcal{H}_{B}\right)$
it is clear that they are $LU$-invariant, i.e.
\[
U^{\otimes n+1}A_{n}\left(U^{\dagger}\right)^{\otimes n+1}=A_{n}\,,
\]
for $U=U_{A}\otimes U_{B}$, where $U_{A,B}\in\mathrm{SU}\left(\mathcal{H}_{A,B}\right)$. 

Let us conclude our considerations on Schmidt rank of quantum states
by explicitly computing the polynomial $\bra{\psi{}^{\otimes\left(n+1\right)}}A_{n}\ket{\psi{}^{\otimes\left(n+1\right)}}$
for arbitrary pure state $\kb{\psi}{\psi}\in\mathcal{D}_{1}\left(\mathcal{H}_{A}\otimes\mathcal{H}_{B}\right)$.
\begin{prop}
\label{schmidt number explicit}Let $\kb{\psi}{\psi}\in\mathcal{D}_{1}\left(\mathcal{H}_{A}\otimes\mathcal{H}_{B}\right)$.
Let the Schmidt decomposition of $\kb{\psi}{\psi}$ be given by Eq.\eqref{eq:schmidt decomposition}.
The following formula holds
\begin{equation}
\bra{\psi{}^{\otimes\left(n+1\right)}}A_{n}\ket{\psi{}^{\otimes\left(n+1\right)}}=\frac{1}{\left(n+1\right)!}\sum_{\begin{array}[t]{c}
X\subset\left\{ 1,\ldots,d\right\} \\
\left|X\right|=n+1
\end{array}}\left(\prod_{i\in X}\lambda_{i}^{2}\right)\,.\label{eq:explicit form Schmidt}
\end{equation}
\end{prop}
\begin{proof}
The proof is given in Section \ref{sec:Proofs-chapter polyn charact}
of the Appendix (see page \pageref{sub:Proof-of-Proposition schmidt number}).
\end{proof}
The above result shows that there is a connection between the polynomial
in\textbackslash{}invariant based on $A_{n}$ and $d$-concurrence
\citep{Gour2005}. More precisely we have
\[
^{n+1}\sqrt{\bra{\psi{}^{\otimes\left(n+1\right)}}A_{n}\ket{\psi{}^{\otimes\left(n+1\right)}}}\propto C_{n+1}\left(\kb{\psi}{\psi}\right)\,,
\]
where $C_{n+1}\left(\cdot\right)$ is the concurrence \citep{Gour2005}
of order $n+1$.

\subsection{Generalized coherent states of compact simply-connected Lie groups\label{sub:Generalized-coherent-states-multilin}}

In this subsection we present the multilinear characterization of
the manifold of generalized coherent states of compact simply-connected
Lie group. The considered setting is exactly the same as in Section
\ref{sec:semisimple-quadratic-characterisation}: we have a compact
simply-connected Lie group $K$ irreducibly represented on the Hilbert
space $\mathcal{H}^{\lambda_{0}}$ ($\lambda_{0}\in\mathfrak{h}^{\ast}$
denotes the highest weight characterizing the representation). We
use the same notation as in Section \ref{sec:semisimple-quadratic-characterisation}:
$\Pi,$ $\pi$ and $\mathcal{M}_{\lambda_{0}}$ denote respectively
representation of Lie group $K$, representation of the corresponding
Lie algebra $\mathfrak{k}$ Lie algebra, and the class of generalized
coherent states in $\mathcal{D}_{1}\left(\mathcal{H}^{\lambda_{0}}\right)$.
We would like to remark that the result presented here is of different
type then ones presented in Subsections \ref{sub:GME} and \ref{sub: Schmid rank}.
There we had to consider polynomials of degree higher that two in
order to grasp a given class of pure states $\mathcal{M}$. Here however
we fix the class of pure states $\mathcal{M}_{\lambda_{0}}$ (we know
its characterization in terms of polynomial of degree two - see Proposition
\ref{prop:proj two}) and describe its characterizations via polynomials
of arbitrary degree in the density matrix. In order to state our results
we have to first introduce some notation. Let $k>1$ be a natural
number. By the usual tensor product procedure (see Eq. \eqref{eq:tensor product of reps})
we construct representations of $K$ and $\mathfrak{k}$ respectively
on $\left(\mathcal{H}^{\lambda_{0}}\right)^{\otimes k}$, 
\begin{equation}
\Pi^{\otimes k}:K\rightarrow\mathrm{U}\left(\left(\mathcal{H}^{\lambda_{0}}\right)^{\otimes k}\right),\, U\rightarrow\Pi\left(U\right)^{\otimes k}\,\label{eq:k fold group}
\end{equation}
\begin{equation}
\mbox{\ensuremath{\pi}}^{\otimes k}:\mathfrak{k}\rightarrow i\cdot\mathrm{Herm}\left(\left(\mathcal{H}^{\lambda_{0}}\right)^{\otimes k}\right),\, X\rightarrow\pi\left(X\right)\otimes\mathbb{I}^{\otimes(k-1)}+\mathbb{I}\otimes\pi\left(X\right)\otimes\mathbb{I}^{\otimes(k-2)}+\ldots\,\label{eq:k fold algebra}
\end{equation}
We have a property analogous to the one presented in Proposition \eqref{prop: two copies}.
\begin{prop}
\label{prop:K copies}Representation $\Pi^{\otimes k}$ is, in general,
reducible. The decomposition of $\left(\mathcal{H}^{\lambda_{0}}\right)^{\otimes k}$
onto irreducible components reads,
\begin{equation}
\left(\mathcal{H}^{\lambda_{0}}\right)^{\otimes k}=\mathcal{H}^{k\lambda_{0}}\oplus\bigoplus_{\beta<k\lambda_{0}}\mathcal{H}^{\beta}\,,\label{eq:decomposition-1}
\end{equation}
where $\mathcal{H}^{k\lambda_{0}}$ denotes the representation characterized
by the highest weight $k\lambda_{0}$ and $\bigoplus_{\beta<k\lambda_{0}}\mathcal{H}^{\beta}$
denotes the direct sum of all other irreducible subrepresentations
of the group $K$ in $\left(\mathcal{H}^{\lambda_{0}}\right)^{\otimes k}$. \end{prop}
\begin{proof}
The proof is a simple generalization of the reasoning contained in
the proof of Proposition \ref{prop: two copies}.
\end{proof}
We can now present a multilinear characterization of the set of coherent
states which is a direct generalization of the result presented in
Proposition \ref{prop:proj two}. The result we state bellow is certainly
well-known to mathematicians working on algebraic geometry but the
author is not aware of any specific reference where it is stated. 
\begin{lem}
\label{k lin characterization coheren}The set of generalized coherent
states $\mathcal{M}_{\lambda_{0}}$ can be characterized by the following
condition
\begin{equation}
\kb{\psi}{\psi}\in\mathcal{M}_{\lambda_{0}}\,\Longleftrightarrow\,\bra{\psi^{\otimes k}}\left(\mathbb{P}^{\mathrm{sym,k}}-\mathbb{P}^{k\lambda_{0}}\right)\ket{\psi^{\otimes k}}=0\,,\label{eq:multilinear characterisation coherent}
\end{equation}
where $\mathbb{P}^{\mathrm{sym,k}}$ is a projector onto $\mathrm{Sym}^{k}\left(\mathcal{H}^{\lambda_{0}}\right)$
and the operator $\mathbb{P}^{k\lambda_{0}}:\left(\mathcal{H}^{\lambda_{0}}\right)^{\otimes k}\rightarrow\left(\mathcal{H}^{\lambda_{0}}\right)^{\otimes k}$
is an orthonormal projector onto irreducible representation $\mathcal{H}^{k\lambda_{0}}\subset\mathcal{H}^{\lambda_{0}}\otimes\mathcal{H}^{\lambda_{0}}$. \end{lem}
\begin{proof}
The proof of Lemma \eqref{k lin characterization coheren} is given
in Section \ref{sec:Proofs-chapter polyn charact} of the Appendix
(see \pageref{sub:Proof-of-Lemma k casimir}). 
\end{proof}
In what follows we present explicit formulas for $\mathbb{P}^{k\lambda_{0}}$
for classes of coherent states discussed already in this chapter:
separable sates, separable bosonic states and Slater determinants.
We will use the following convenient embeddings of vectors spaces
\begin{equation}
\left(\mathrm{Sym}^{L}\left(\mathcal{H}\right)\right)^{\otimes k}\subset\left(\bigotimes_{i=1}^{L}\mathcal{H}_{i}\right)^{\otimes k}\approx\bigotimes_{i=1}^{L}\mathcal{H}_{i}^{\otimes k}\,,\label{eq:embed bos}
\end{equation}
\begin{equation}
\left(\bigwedge^{L}\left(\mathcal{H}\right)\right)^{\otimes k}\subset\left(\bigotimes_{i=1}^{L}\mathcal{H}_{i}\right)^{\otimes k}\approx\bigotimes_{i=1}^{L}\mathcal{H}_{i}^{\otimes k}\,,\label{eq:embed ferm}
\end{equation}
where $\mathcal{H}\approx\mathcal{H}_{i}\approx\mathbb{C}^{N}$ is
the single particle Hilbert space of the either bosonic or fermionic
system considered. Note that in the above we used the natural relabeling
of factors in the composite tensor product,
\begin{equation}
\left(\bigotimes_{i=1}^{L}\mathcal{H}_{i}\right)^{\otimes k}\approx\bigotimes_{i=1}^{L}\mathcal{H}_{i}^{\otimes k}\,.\label{eq:relabelling}
\end{equation}
In what follows, in order to distinguish different spaces that appear
in $\mathcal{H}_{i}^{\otimes k}$ we will use additional labeling,
\begin{equation}
\mathcal{H}_{i}^{\otimes k}=\bigotimes_{j=1}^{k}\mathcal{H}_{i^{\left(j\right)}}\,.\label{eq:relabelling two}
\end{equation}

Consequently, we obtain the following decomposition

\begin{equation}
\left(\bigotimes_{i=1}^{L}\mathcal{H}_{i}\right)^{\otimes k}\approx\bigotimes_{i=1}^{L}\mathcal{H}_{i}^{\otimes k}=\bigotimes_{i=1}^{L}\left(\bigotimes_{j=1}^{k}\mathcal{H}_{i^{\left(j\right)}}\right)\,.\label{eq:relabeling of terms}
\end{equation}

In accordance to \ref{eq:relabeling of terms} we introduce the following
notation
\begin{itemize}
\item $\mathbb{P}_{\left\{ i^{\left(1\right)},\ldots,i^{\left(k\right)}\right\} }^{\mathrm{sym}}$
- orthonormal projector onto totally symmetric subspace of $\mathcal{H}_{i}^{\otimes k}$
(see Eq. \eqref{eq:relabeling of terms}).
\item $\mathbb{P}_{\left\{ 1^{\left(j\right)},\ldots,L^{\left(j\right)}\right\} }^{\mathrm{sym}}-$is
the orthonormal projector onto totally symmetric subspace of $\bigotimes_{i=1}^{L}\mathcal{H}_{i^{\left(j\right)}}$.
\item $\mathbb{P}_{\left\{ 1^{\left(j\right)},\ldots,L^{\left(j\right)}\right\} }^{\mathrm{asym}}-$is
the orthonormal projector onto totally asymmetric subspace of $\bigotimes_{i=1}^{L}\mathcal{H}_{i^{\left(j\right)}}$.\end{itemize}
\begin{lem}
Under the notation introduced in Subsection \ref{sub:Product-states}
the closed expression for the projector operator $\mathbb{P}^{k\lambda_{0}}$
(see \eqref{eq:multilinear characterisation coherent}) for the case
of $K=\times^{L}\big(\mathrm{SU}(N)\big)$ and $\mathcal{H}^{\lambda_{0}}=\mathcal{H}_{dist}=\bigotimes_{i=1}^{L}\mathcal{H}_{i}$
reads
\begin{equation}
\mathbb{P}^{k\lambda_{0}}=\mathbb{P}_{dist}^{(k)}=\bigotimes_{i=1}^{L}\mathbb{P}_{\left\{ i^{\left(1\right)},\ldots,i^{\left(k\right)}\right\} }^{\mathrm{sym}}\,..\label{eq:dist k copies lambda}
\end{equation}
\end{lem}
\begin{proof}[Sketch of the proof]
 The proof of \eqref{eq:dist k copies lambda} is a direct generalization
of the proof of Lemma \ref{eq:dist p2lambda}. The main idea is to
observe that the subspace (note that we are using convention \eqref{eq:relabeling of terms})
\[
\bigotimes_{i=1}^{L}\mathrm{Sym}^{k}\left(\mathcal{H}_{i}\right)\subset\mathrm{Sym}^{k}\left(\bigotimes_{i=1}^{L}\mathcal{H}_{i}\right)
\]
is the irreducible representation of the group $\times^{L}\big(\mathrm{SU}(N)\big)$.
\end{proof}
Figure \ref{fig:disting k copies} presents a graphical representation
of the operator $\mathbb{P}_{d}^{(k)}$.

\begin{figure}[h]
\centering{}\includegraphics[width=8.5cm]{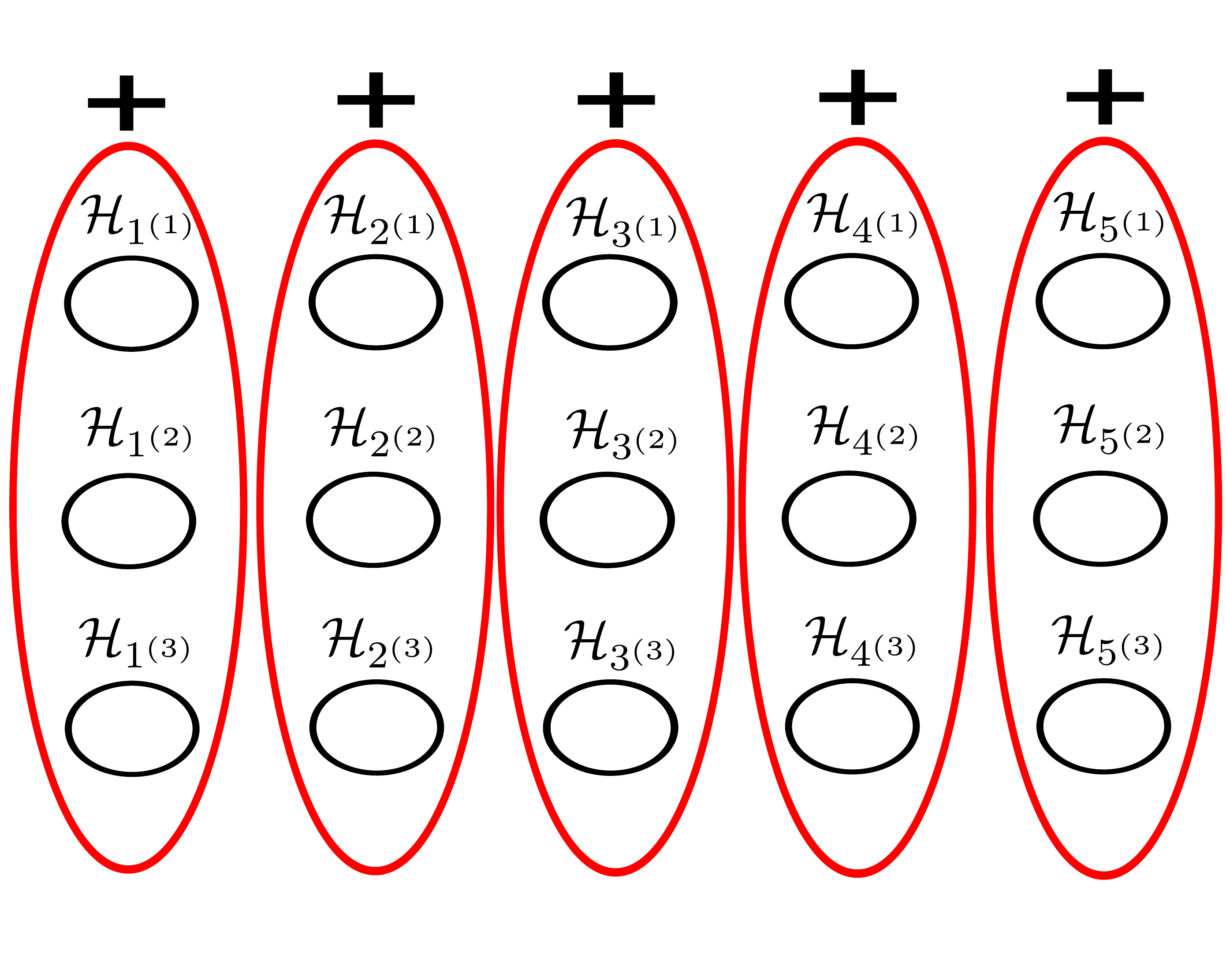}\protect\caption{\label{fig:disting k copies}Graphical presentation of the operator
$\mathbb{P}_{d}^{\left(k\right)}$ in $\mathcal{H}_{d}^{\otimes k}$.
Red circles correspond to symmetrizations in the relevant factors
of $\mathcal{H}_{d}^{\otimes k}$.}
\end{figure}

\begin{lem}
Under the notation introduced in Subsection \ref{sub:Symmetric-product-states}
the closed expression for the projector operator $\mathbb{P}^{k\lambda_{0}}$
(see \eqref{eq:multilinear characterisation coherent}) for the case
of $K=SU(N)$ and $\mathcal{H}^{\lambda_{0}}=\mathcal{H}_{b}=\mathrm{Sym}^{L}\left(\mathcal{H}\right)$
reads
\begin{equation}
\mathbb{P}^{k\lambda_{0}}=\mathbb{P}_{b}^{(k)}=\left(\bigotimes_{i=1}^{L}\mathbb{P}_{\left\{ i^{\left(1\right)},\ldots,i^{\left(k\right)}\right\} }^{\mathrm{sym}}\right)\circ\left(\bigotimes_{j=1}^{k}\mathbb{P}_{\left\{ 1^{\left(j\right)},\ldots,L^{\left(j\right)}\right\} }^{\mathrm{sym}}\right)\,..\label{eq:bosons k poies lambda}
\end{equation}
\end{lem}
\begin{proof}[Sketch of the proof]
The proof is completely analogous to the proof of Lemma \ref{lema bosons p2}.
The operator $\left(\bigotimes_{i=1}^{L}\mathbb{P}_{\left\{ i^{\left(1\right)},\ldots,i^{\left(k\right)}\right\} }^{\mathrm{sym}}\right)\circ\left(\bigotimes_{j=1}^{k}\mathbb{P}_{\left\{ 1^{\left(j\right)},\ldots,L^{\left(j\right)}\right\} }^{\mathrm{sym}}\right)$
is a projector onto $\mathrm{Sym}^{k\cdot L}\left(\mathcal{H}\right)\subset\mathrm{Sym}^{k}\left(\mathrm{Sym}^{L}\left(\mathcal{H}\right)\right)$,
which is an irreducible representation of $SU\left(N\right)$. 
\end{proof}
Figure \ref{fig:bos k copies} presents a graphical representation
of the operator $\mathbb{P}_{b}^{(k)}$.

\begin{figure}[h]
\centering{}\includegraphics[width=9cm]{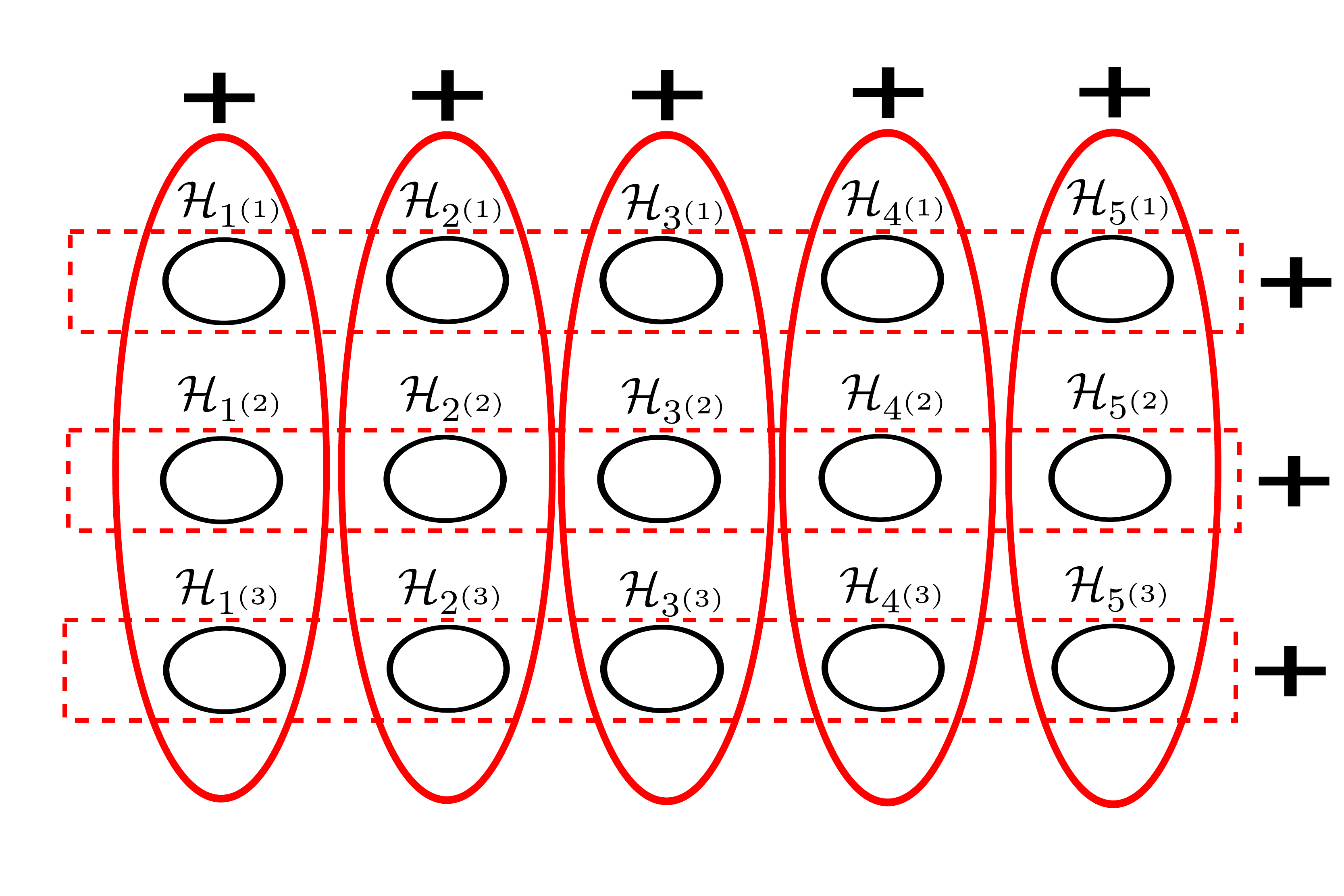}\protect\caption{\label{fig:bos k copies}Graphical presentation of the operator $\mathbb{P}_{k}^{\left(k\right)}$
treated as an operator on $\mathcal{H}_{d}^{\otimes k}$. Red circles
and rectangles correspond to symmetrizations in the relevant factors
of $\mathcal{H}_{d}^{\otimes k}$ containing $\mathcal{H}_{b}^{\otimes k}$.}
\end{figure}

The polynomial invariant
\[
\mathcal{D}_{1}\left(\mathrm{Sym}^{L}\left(\mathcal{H}\right)\right)\ni\kb{\psi}{\psi}\longrightarrow\mathrm{tr}\left(\left[\kb{\psi}{\psi}^{\otimes k}\right]\mathbb{P}_{b}^{(k)}\right)\in\mathbb{R}
\]
was used in \citep{Migdal2014} to study the problem of unitary equivalence
of multiphoton states under the action of linear optics.
\begin{lem}
Under the notation introduced in Subsection \ref{sub:Slater-determinants}
the closed expression for the projector operator $\mathbb{P}^{k\lambda_{0}}$
(see \eqref{eq:multilinear characterisation coherent}) for the case
of $K=SU(N)$ and $\mathcal{H}^{\lambda_{0}}=\mathcal{H}_{f}=\mathrm{\bigwedge}^{L}\left(\mathcal{H}\right)$
reads
\begin{equation}
\mathbb{P}^{k\lambda_{0}}=\mathbb{P}_{f}^{(k)}=\alpha_{k}\left(\bigotimes_{i=1}^{L}\mathbb{P}_{\left\{ i^{\left(1\right)},\ldots,i^{\left(k\right)}\right\} }^{\mathrm{sym}}\right)\circ\left(\bigotimes_{j=1}^{k}\mathbb{P}_{\left\{ 1^{\left(j\right)},\ldots,L^{\left(j\right)}\right\} }^{\mathrm{asym}}\right)\,,\label{eq:ferm k copies lambda}
\end{equation}
where
\[
{\color{red}{\normalcolor \alpha_{k}=\alpha\left(\left(\stackrel{L}{\overbrace{k,k,\ldots,k}},0,\ldots0\right)\right)}}
\]
and the function $\alpha\left(\lambda\right)$ is given by Eq.\eqref{eq:coefficient}. \end{lem}
\begin{proof}
Let us first notice that $\mathrm{Im}\left(\mathbb{P}_{f}^{(k)}\right)\subset\mathcal{H}_{f}^{\otimes k}$.
Moreover, $\mathbb{P}_{f}^{(k)}:\mathcal{H}^{\otimes k\cdot L}\rightarrow\mathcal{H}^{\otimes k\cdot L}$
projects onto subspace characterized by the highest weight
\[
k\cdot\lambda_{f}==\left(\stackrel{L}{\overbrace{k,k,\ldots,k}},0,\ldots0\right)\,.
\]
By the virtue of Proposition \eqref{prop:K copies} there is only
one representation of $\mathrm{SU}\left(N\right)$ characterized by
the highest weight $k\cdot\lambda_{f}$. 
\end{proof}
Figure \ref{fig:ferm  k copies} presents a graphical representation
of the operator $\mathbb{P}_{f}^{(k)}$.

\begin{figure}[h]
\centering{}\includegraphics[width=9cm]{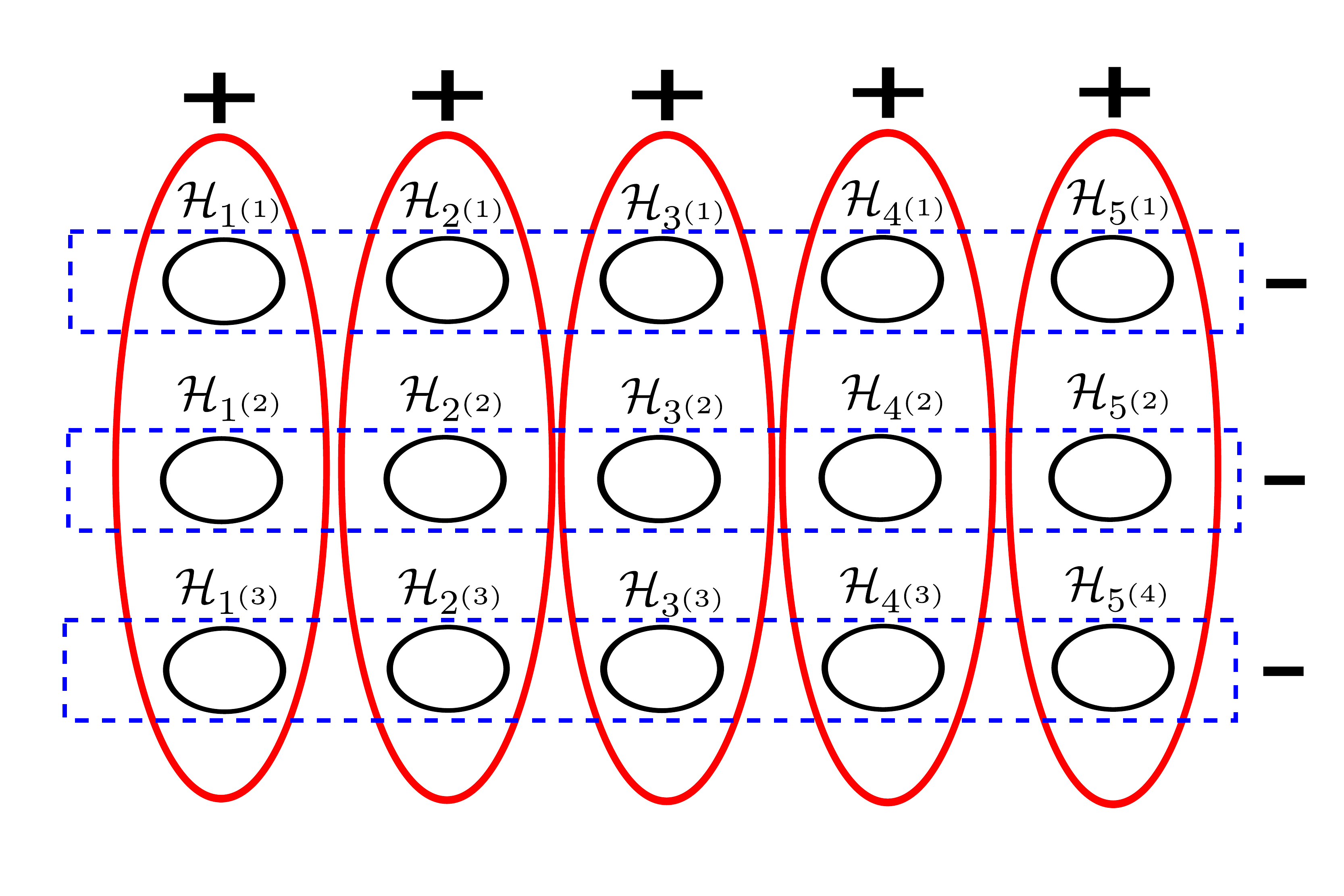}\protect\caption{\label{fig:ferm  k copies}Graphical presentation of the operator
$\mathbb{P}_{f}^{\left(k\right)}$ treated as an operator on $\mathcal{H}_{d}^{\otimes k}$.
Red circles correspond to symmetrizations and blue rectangles to antisymmetrizations
in the relevant factors of the space $\mathcal{H}_{d}^{\otimes k}$
containing $\mathcal{H}_{f}^{\otimes k}$.}
\end{figure}

\section{Summary and open problems\label{sec:pure Discussion-and-open}}

In this chapter we provided a polynomial characterization of different
classes of non-correlated pure states $\mathcal{M}\subset\mathcal{D}_{1}\left(\mathcal{H}\right)$
that appear in many physical problems. Our approach was independent
on the choice of the basis and the characterization we gave was group-invariant,
provided there existed a symmetry group preserving the set $\mathcal{M}$.
We described the classes $\mathcal{M}$ as the null sets of real non-negative
homogenous polynomials in the density matrix $\kb{\psi}{\psi}$ of
a state (also in infinite-dimensional settings). In general the descriptions
studied were of the form
\begin{equation}
\kb{\psi}{\psi}\in\mathcal{M}\,\Longleftrightarrow\,\mathrm{tr}\left(\left[\kb{\psi}{\psi}^{\otimes k}\right]\, A\right)=0\,,\label{eq:again charact}
\end{equation}
where $A$ is non-negative operator satisfying $A\leq\mathbb{P}^{\mathrm{sym},k}.$ 

The following list contains the most important results contained in
this chapter.
\begin{itemize}
\item \textbf{Section \ref{sec:semisimple-quadratic-characterisation}}.
A characterization of coherent states for irreducible representations
of compact simply-connected Lie groups $\mathcal{M}\subset\mathcal{D}_{1}\left(\mathcal{H}\right)$
as zero sets of a group-invariant quadratic polynomials ($k$ in Eq.\eqref{eq:again charact}
equals $2$). We presented an explicit form of the operator $A$ for
the following classes of pure states: product states of distinguishable
particles, bosonic product states, Slater determinants, and pure fermionic
Gaussian states.
\item \textbf{Section \ref{sec:inf dimension}}. A polynomial characterization
of the following classes of pure states in infinite-dimensional Hilbert
spaces: product states of distinguishable particles, bosonic product
states, Slater determinants. We represent these classes as a zero
sets of a quadratic polynomial ($k$ in Eq.\eqref{eq:again charact}
equals $2$) which is group-invariant. 
\item \textbf{Section \ref{sec:Multilinear-characterization-of-pure}}.
A characterization of coherent states in irreducible representations
of compact simply-connected Lie groups $\mathcal{M}\subset\mathcal{D}_{1}\left(\mathcal{H}\right)$
as a zero sets of a group-invariant polynomials of degree $k$. We
presented an explicit for of the operator $A$ for the following classes
of pure states: product states of distinguishable particles, bosonic
product states, Slater determinants. We also provided a similar characterization
of classes of pure states consisting of Characterization of pure bipartite
states with bounded Schmidt rank and $2$-separable states in multiparty
setting.
\end{itemize}

\subsection*{Open problems}

We now give some open problems connected to the subject of this Chapter.
\begin{itemize}
\item It is natural to ask whether a similar polynomial characterizations
is also possible for other classes of states which naturally appear
in quantum information and were not captured in this chapter. We give
here examples of such classes. Depending upon concrete physical system
in question it is certainly possible to invent many other classes. 

\begin{itemize}
\item Matrix Product States (MPS) and Projected Entangled Pair States (PEPS)
\citep{Verstraete2008}. These are important classes of variational
states that turned useful in describing respectively 1D and 2D lattice
quantum systems.
\item A notion of $k$-coherent states was recently introduced in the context
of the theory of coherent delocalization in \citep{Levi2014}. From
the mathematical perspective $k$-coherent pure states are specified
by the choice of the class of pure states $\mathcal{M}\subset\mathcal{D}_{1}\left(\mathcal{H}\right)$
and the natural number $k$. We say that a pure state $\kb{\psi}{\psi}$
is $k$-coherent if its vector representative can be written as
\[
\ket{\psi}=\sum_{i=1}^{k}\alpha_{i}\ket{\phi_{i}}\,,
\]
where $\alpha_{i}\neq0$ and $\ket{\phi_{i}}$ are vector representatives
of states $\kb{\phi_{i}}{\phi_{i}}\in\mathcal{M}$. If the class $\mathcal{M}$
is an algebraic variety (i.e. it is gives as a null set of some polynomial)
there is a natural connection between sets consisting of $k$-coherent
states and so-called Secant varieties \citep{Sawicki2013}.
\item Classes of states $\mathcal{M}$ relevant for quantum optics. In this
context the appropriate Hilbert space is the bosonic Fock space (corresponding
to one or many optical modes) \citep{Puri2001}. The relevant classes
of coherent states are: optical coherent sates, squeezed states and
pure bosonic Gaussian states. Each of these classes is an orbit of
the appropriate symmetry group: the Heisenberg group \citep{Puri2001},
Metaplectic group \citep{Derezinski2011} and affine-Metaplectic group
\citep{Derezinski2011} receptively. The main technical difficulty
that appears while studying this problem is the fact that the relevant
Hilbert space is infinite-dimensional and therefore functional-theoretic
aspects have to be taken int account.
\end{itemize}
\item A polynomial characterization of the class of pure Gaussian states
$\mathcal{M}_{g}\subset\mathcal{H}_{\mathrm{Fock}}\left(\mathbb{C}^{d}\right)$
and quantum de-Finetti theorem \citep{Harrow2013} was used in \citep{powernoisy2013}
to derive a complete hierarchy of semidefinite programs that decide
whether a given even state $\mathcal{D}_{\mathrm{even}}\left(\mathcal{H}_{\mathrm{Fock}}\left(\mathbb{C}^{d}\right)\right)$
belongs to the convex hull of $\mathcal{M}$, $\mathcal{M}_{g}^{c}$.
Similar techniques have been used before in the context of entanglement
\citep{Doherty2004}. It would be a very interesting to derive a complete
hierarchy of criteria characterizing convex hulls of classes of states
$\mathcal{M}$ defined via Eq.\eqref{eq:again charact}.
\item In \citep{Badziag2013} it was showed that entanglement of a mixed
bipartite state can be characterized by checking the weak optimality%
\footnote{A witness $W\in\mathrm{Herm}\left(\mathcal{H}_{A}\otimes\mathcal{H}_{B}\right)$
is said to be weakly optimal if and only if its expectation value
vanishes on some product vector.%
} of an entanglement witness acting on an auxiliary extended Hilbert
space. This approach can be generalized to the setting considered
in this chapter and throughout this thesis. It would be therefore
interesting to investigate the implications of this on the problem
of generalized entanglement.
\end{itemize}

\chapter{Complete characterization of correlations in mixed states \label{chap:Complete-characterisation}}

In the previous chapter we gave a polynomial characterization of classes
of non-correlated pure states $\mathcal{M}\subset\mathcal{D}_{1}\left(\mathcal{H}\right)$.
We showed that a number physically relevant of classes of pure states
$\mathcal{M}$ are a zero sets of a real homogenous polynomial
\begin{equation}
\kb{\psi}{\psi}\in\mathcal{M}\,\Longleftrightarrow\mathrm{tr}\left(A\left[\kb{\psi}{\psi}^{\otimes k}\right]\right)=0\,,\label{eq:characterization again}
\end{equation}
where $A\in\mathrm{Herm}_{+}\left(\mathrm{Sym}^{k}\left(\mathcal{H}\right)\right)$.
This chapter as well as Chapters \ref{sec:Proofs-concerning-Chapter multilinear witnesses}
and \ref{sec:Proofs-concerning-Chapter typicality} will be devoted
to study notion of correlations defined by the choice of the class
$\mathcal{M}$. As explained in Section \ref{sec:General-motivation}
in this thesis we will be concerned with correlations present in general
mixed states specified by the probabilistic mixtures of pure states
belonging to the class $\mathcal{M}$. Given a class of non-correlated
pure states $\mathcal{M}$ we define non-correlated mixed states as
a convex hull, $\mathcal{M}^{c}$, of $\mathcal{M}$ in the set of
all mixed states of a considered system: 
\[
\mathcal{M}^{c}=\left\{ \rho=\sum_{i}p_{i}\kb{\psi_{i}}{\psi_{i}}\left|\kb{\psi_{i}}{\psi_{i}}\in\mathcal{M},\, p_{i}\geq0,\,\sum_{i}p_{i}=1\right.\right\} \,.
\]

In the current chapter we will consider cases for which the class
$\mathcal{M}$ consists of coherent states of compact simply-connected
Lie groups irreducibly represented in $\mathcal{H}$. In particular
we will be concerned with the following problem.
\begin{problem}
\label{exact characterization}Let $\mathcal{M}$ consists of coherent
states of compact  simply-connected Lie group $K$ irreducibly represented
in the Hilbert space $\mathcal{H}$. Characterize the cases when it
is possible to describe the set of non-correlated states $\mathcal{M}^{c}$
by an explicit analytical criterion%
\footnote{By ``an explicit analytical'' criterion we mean a function $f:\mathcal{D}\left(\mathcal{H}\right)\rightarrow\mathbb{R}$
that have a closed form (the dependance of $f$ on coefficients of
$\rho$ is known) and values of $f$ separates $\mathcal{M}^{c}$
from its complement (for instance we may have $f\left(\rho\right)=0$
if and only if $\rho\in\mathcal{M}^{c}$).%
}.
\end{problem}
The above problem is not stated, for the sake of clarity, in a rigorous
manner. Let us now specify more precisely what kind of ``analytical
criterion'' we have in mind. In Section \ref{sec:Methods-from-entanglement}
we discussed the Wotters concurrence \citep{Wootters1998} $C\left(\rho\right)$
(see Eq.\eqref{eq:wooters concurrence}) for the case of entanglement
in the case of two qbits. The concurrence $C\left(\rho\right)$ posses
the properties we desire: it has a closed-form expression and satisfies

\[
C\left(\rho\right)>0\,\Longleftrightarrow\mbox{\ensuremath{\rho}\,\text{is entangled }\,},\, C\left(\rho\right)=0\,\Longleftrightarrow\mbox{\ensuremath{\rho}\,\text{is separable }\,}.
\]
where property being separable corresponds to being inside the convex
hull of the set of pure product states for two qbits ($\rho\in\mathcal{M}_{dist}^{c}$).
It turns out that in general it is possible to write down a formula
analogous to Eq.\eqref{eq:wooters concurrence} provided there exists
an antiunitary conjugation%
\footnote{The definition of the antiunitary conjugation is given in the next
section. %
} $\theta$ such that $\left|\bk{\psi}{\theta\psi}\right|=0$ if and
only if $\kb{\psi}{\psi}\in\mathcal{M}$. The reason for which this
condition allows for analytic characterization of $\mathcal{M}^{c}$
will be explained in the next section and is known under the name
of Uhlmann-Wotters construction \citep{Uhlmann2000}. We are now ready
to state rigorously the problem we will investigate in this chapter.
\begin{problem}
\label{limited expression}Let $\mathcal{M}$ consists of coherent
states of a compact simply-connected Lie group $K$ irreducibly represented,
via the representation $\Pi:K\rightarrow\mathrm{U}\left(\mathcal{H}\right)$,
in the Hilbert space $\mathcal{H}$. Characterize the cases when $\mathcal{M}$
can be described by the condition
\begin{equation}
\kb{\psi}{\psi}\in\mathcal{M}\,\Longleftrightarrow\,\left|\bk{\psi}{\theta\psi}\right|=0\,,\label{eq:antiunitary coherent}
\end{equation}
where $\theta$ is a unitary conjugation which is $K$-invariant,
i.e.
\begin{equation}
\Pi\left(k\right)\theta\Pi^{\dagger}\left(k\right)=\theta\,,\label{eq:k invariant antiunitary}
\end{equation}
for all $k\in K$.
\end{problem}
The motivation for the study of this problem stems from the fact that,
except for a number of distributed, seemingly unrelated results \citep{SchliemannTwoFermions2001,EckertFermions2002,Giraud2008},
so far there has been no complete understanding of cases when coherent
states allow for the characterization via \eqref{eq:antiunitary coherent}.
A certain degree of unification was achieved in \citep{Kus2009} and
the investigations presented here can be treated as a generalization
of the results these obtained in \citep{Kus2009}.

This chapter is organized as follows. In Section \ref{sec:Ulhmann-Wooters}
we describe the Uhlmann-Wotters construction and present the complete
solution to the Problem \ref{limited expression}. In Theorems \ref{theta representation}
and \ref{theta epimorphism} we present a complete group-theoretical
characterization of coherent states of compact simply-connected Lie
groups that admit a characterization via the antiunitary, group-invariant
conjugation. In Section \ref{sec:Classical-simulation-of-FLO} we
apply our result to study the type of correlations that has not been
studied before by this formalism: correlations defined by choosing
$\mathcal{M}$ to consist of pure fermionic Gaussian states $\mathcal{M}_{g}$
(c.f. Subsection \ref{sub:Fermionic-Gaussian-states}). This type
of correlations is important for the problem of classical  simulability
of a model of quantum computation consisting of Fermionic Linear Optics
augmented with noisy ancilla state $\rho$. We define this model of
quantum computation and explain that the problem of its classical
simulability can be partially solved by determining whether a mixed
state $\rho$ is a convex combination of pure Gaussian states (we
call such states convex-Gaussian states). It turns out that the methods
developed in Section \ref{sec:Ulhmann-Wooters} allow to characterize
$\mathcal{M}_{g}^{c}$ for the lowest dimensional non-trivial case
of the Fermionic Fock space of $d=4$ modes. This characterization
allows us to solve an open problem recently posed in \citep{powernoisy2013}
and to study the geometry of the set of convex-Gaussian states in
the space of all states defined on four mode Fock space $\mathcal{H}_{\mathrm{Fock}}\left(\mathbb{C}^{4}\right)$.
We conclude the chapter in Section \ref{sec:sumarry Exact} where
we summarize the obtained results and state some open problems.

Results presented in Sections \ref{sec:Ulhmann-Wooters} and \ref{sec:Classical-simulation-of-FLO}
were published in \citep{detection2012} and \citep{GaussSim2014}
respectively.

\section{Classes of coherent states for which Uhlmann--Wotters construction
works\label{sec:Ulhmann-Wooters}}

In this section we classify all the cases when the set of coherent
states%
\footnote{In order to simplify the notation we will use the symbol $\mathcal{M}$
to refer to the set of coherent states of the particular $K$ in the
representation $\mathcal{H}^{\lambda_{0}}$of interest.%
} $\mathcal{M}$ of a compact simply-connected Lie group can be characterized
by condition \eqref{eq:antiunitary coherent}. This description allows
to give a simple analytic characterization of the class of non-correlated
states $\mathcal{M}^{c}$. The section is organized as follows. In
Subsection \ref{sub:Antiunitary-conjugations-andUW} we present some
technical tools that will be useful in further considerations. We
introduce there the following concepts: antiunitary conjugation, convex
roof extensions and Uhlmann-Wotters construction. We also discuss
in more details some properties of completely positive maps and Jamiołkowski-Choi
isomorphism. In Subsection \ref{sub:Classification-of-coherent-antiunitary}
we give a complete solution to Problem \ref{limited expression} which
employs the methods mentioned above as well as some facts form representation
theory and differential geometry.

\subsection{Technical tools \label{sub:Antiunitary-conjugations-andUW}}

We first describe some properties of antilinear operators. An antilinear
operator $\vartheta:\mathcal{H}\rightarrow\mathcal{H}$ is defined
by the condition:
\begin{equation}
\vartheta\left(\alpha\ket{\phi}+\beta\ket{\psi}\right)=\alpha^{\ast}\vartheta\left(\ket{\phi}\right)+\beta^{\ast}\vartheta\left(\ket{\psi}\right)\,,\label{eq:antilinear}
\end{equation}
valid for all $\ket{\phi},\ket{\psi}\in\mathcal{H}$ and $\alpha,\beta\in\mathbb{C}$.
Just like in the case of linear operators it is customary to abbreviate
$\vartheta\left(\ket{\psi}\right)$ by $\vartheta\ket{\psi}$ or $\ket{\vartheta\psi}$.
The Hermitian conjugate of the antilinear operator $\vartheta$, denoted
by $\vartheta^{\dagger}$, is defined by
\begin{equation}
\bra{\psi}\left(\vartheta^{\dagger}\ket{\phi}\right)=\bra{\phi}\left(\vartheta\ket{\psi}\right)\,.\label{eq:hermitian conjugation antiunitary}
\end{equation}
We will always assume implicitly that the antilinear operator acts
``on the right'', when put between ``bra'' and ``ket''. Consequently,
we will use the following convention
\begin{equation}
\bra{\psi}\theta\ket{\phi}\equiv\bra{\phi}\left(\theta^{\dagger}\ket{\psi}\right)\,.\label{eq:convention antilinear}
\end{equation}

A product of linear and antilinear operator is antilinear. The product
of two antilinear operators is a linear operator. For two operators
$A,B$ (each being linear or antilinear) we have $\left(AB\right)^{\dagger}=B^{\dagger}A^{\dagger}$.
Moreover, we have $\left(\vartheta^{\dagger}\right)^{\dagger}=\vartheta$
for each antilinear operator $\vartheta$. Recall (c.f. Section \ref{sec:Methods-from-entanglement})
that a mapping $\theta:\mathcal{H}\rightarrow\mathcal{H}$ is called
antiunitary if it is antilinear and satisfies
\[
\bk{\theta\psi}{\theta\phi}=\bk{\psi}{\phi}^{\ast}\,,
\]
for all $\ket{\psi},\ket{\phi}\in\mathcal{H}$. Every anti-unitary
operator $\theta$ admits a decomposition $\theta=U\mathcal{K}$,
where $\mathcal{K}$ is the conjugation in some fixed orthonormal
basis (see Eq.\ref{eq:complex conjugation}) and $U$ is the unitary
operator. We can now define 
\begin{defn}
\label{antiunitary conjugation}An antiunitary operator $\theta$
is called \textit{anti-unitary conjugation} if and only if it is antiunitary
and $\theta^{\dagger}=\theta$. 
\end{defn}
The above definition is equivalent to $U\mathcal{K}-\mathcal{K}U^{\dagger}=0$.
Consequently the conjugation $\mathcal{K}$ preserves eigenvectors
of $U$. Let $\ket{\psi_{k}}$ be an eigenvector of $U$ corresponding
to eigenvalue $e^{i\phi_{k}}$. We have
\[
\theta e^{i\alpha}\ket{\psi_{k}}=e^{i(\phi_{k}-\alpha)}\ket{\psi_{k}}\,,
\]
and therefore, by gauging out the phase $\alpha$, we can obtain a
vector $\ket{\psi'}$ which is the eigenvector of $\theta$ with eigenvalue
one. The same procedure can be repeated for all eigenvectors of $U$.
Consequently we obtain that $\theta=\mathcal{K}'$, where $\mathcal{K}'$
is a complex conjugation in a suitably-chosen orthonormal basis of
$\mathcal{H}$.

Let us now present how the characterization of coherent states via
the polynomial condition \eqref{eq:characterization again} can be
used to formally solve the problem of characterization of the set
$\mathcal{M}^{c}$ via the method of convex roof extension (see Section
\ref{sec:Methods-from-entanglement}). Let $\mathcal{H=}\mathcal{H}^{\lambda_{0}}$
be the carrier space of the irreducible representation of a compact
simply-connected Lie group $K$ and $\mathcal{M}\subset\mathcal{D}_{1}\left(\mathcal{H}^{\lambda_{0}}\right)$
consists of coherent states of this group. Let us define a function
$g:\mathcal{D}_{1}\left(\mathcal{H}^{\lambda_{0}}\right)\rightarrow\mathbb{R}$,
\begin{equation}
g\left(\kb{\psi}{\psi}\right)=\sqrt{\mathrm{tr}\left(\left(\mathrm{\mathbb{P}^{\mathrm{sym}}-\mathbb{P}^{2\lambda_{0}}}\right)\left[\kb{\psi}{\psi}^{\otimes2}\right]\right)}\,,\label{eq:proposition f pure}
\end{equation}
where $\mathbb{P}^{2\lambda_{0}}$ is the projection onto the representation
with the highest weight $2\lambda_{0}$ embedded in $\mathrm{Sym}^{2}\left(\mathcal{H}\right)$
(see Eq.\eqref{eq:decomposition}). It is clear that $g$ is well
defined, continuous and reaches the minimum (equal to $0$) on the
set of coherent states $\mathcal{M}$. Therefore $g^{\cup}:\mathcal{D}\left(\mathcal{H}\right)\rightarrow\mathbb{R}$
distinguishes between correlated and non-correlated states (c.f. Section
\ref{sec:Methods-from-entanglement}): 
\[
g^{\cup}\left(\rho\right)\geq0\,\text{ and }\, g^{\cup}\left(\rho\right)=0\,\text{ if and only if }\rho\in\mathcal{M}^{c}\,.
\]
Due to the fact that both $g$ and $g^{\cup}$ are 1-homogenous, we
can write
\begin{equation}
g^{\cup}\left(\rho\right)=\inf_{\sum_{k}\kb{\psi_{k}}{\psi_{k}}=\rho}\sum_{k}\sqrt{\mathrm{tr}\left(\left(\mathrm{\mathbb{P}^{\mathrm{sym}}-\mathbb{P}^{2\lambda_{0}}}\right)\left[\kb{\psi_{k}}{\psi_{k}}^{\otimes2}\right]\right)}\,,\label{eq:homogenous convex roof}
\end{equation}
where the infimum is taken over all decompositions of $\rho$ into
a sum of operators of rank 1 (not necessary normalized). In general
the infimum in the formula \eqref{eq:homogenous convex roof} cannot
be computed explicitly for an arbitrary $\rho\in\mathcal{D}\left(\mathcal{H}\right)$;
one then has to rely on various, relatively easily computable estimates,
which, however, give only sufficient criteria for detection of correlations
\citep{Kotowski2010} leaving a margin of uncertainty in discriminating
mixed classical states. There are cases when the effective computation
of the infimum is possible \citep{Kus2009}. They correspond to situations
when the operator expectation value in \eqref{eq:proposition f pure}
can be expressed in terms of some anti-unitary conjugation%
\footnote{Note that we have used the convention from Eq.\eqref{eq:convention antilinear},
i.e. we assume that the antiunitary operator acts always ``on the
right''. %
} $\theta:\mathcal{H}^{\lambda_{0}}\rightarrow\mathcal{H}^{\lambda_{0}}$,
\begin{equation}
\mathrm{tr}\left(\left(\mathrm{\mathbb{P}^{\mathrm{sym}}-\mathbb{P}^{2\lambda_{0}}}\right)\left[\kb{\psi}{\psi}^{\otimes2}\right]\right)=c\left|\bra{\psi}\theta\ket{\psi}\right|^{2}\,,\label{eq:simplification formula}
\end{equation}
where $c>0$ is a constant.

We will use the following result by Uhlmann and Wotters.
\begin{fact}
(\citep{Uhlmann2000}) \label{fact uhlmann construction}Let $\theta:\mathcal{H}\rightarrow\mathcal{H}$
be the antiunitary conjugation on $N$ dimensional Hilbert space $\mathcal{H}$.
Let 
\begin{equation}
g^{\cup}\left(\rho\right)=\inf_{\sum_{k}\kb{\psi_{k}}{\psi_{k}}=\rho}\left|\bra{\psi_{k}}\theta\ket{\psi_{k}}\right|\,.\label{eq:convex roof antiunitary}
\end{equation}
We have
\begin{equation}
g^{\cup}\left(\rho\right)=\max\left\{ 0,\lambda_{1}-\sum_{k=2}^{N}\lambda_{k}\right\} \,,\label{eq:explict formula}
\end{equation}
where $\left\{ \lambda_{k}\right\} _{k=1}^{k=N}$ are increasingly
ordered eigenvalues of the operator $\sqrt{\rho\tilde{\rho}}$, where
$\tilde{\rho}=\theta\rho\theta$ (compare Eq.\eqref{eq:wooters concurrence}). 

Moreover, the optimal decomposition in \eqref{eq:convex roof antiunitary}
may be constructed out of $2n+1$ rank one operators, where $2n<N\le2n+1$.
\end{fact}
According to our knowledge the situations expressed by \eqref{eq:simplification formula}
are the only ones in which it is possible to compute $g$ explicitly.
The list of known examples of this kind described in the literature
\citep{SchliemannTwoFermions2001,EckertFermions2002,Giraud2008} is
short and contains only three examples which we list below.
\begin{enumerate}
\item The three-dimensional (labeled by spin $S=1$) representation $\mathcal{H}^{1}$
of $\mathrm{SU}\left(2\right)$. It is a well known fact that $\mathrm{Sym}^{2}\left(\mathcal{H}^{1}\right)=\mathcal{H}^{2}\oplus\mathcal{H}^{0}$,
where $\mathcal{H}^{0}$ is the trivial representation (labeled by
spin $S=0$) and $\mathcal{H}^{2}$ is the five-dimensional representation
(labeled by spin $S=2$). This representation is used in the description
of two bosons of spin $S=1$. The corresponding class of coherent
states $\mathcal{M}$ consists of standard spin coherent states in
the representation $\mathcal{H}^{1}$ \citep{Giraud2008}.
\item The four-dimensional representation of $\mathrm{SU}\left(2\right)\times\mathrm{SU}\left(2\right)$
defined by its natural action on $\mathcal{H}=\mathbb{C}^{2}\otimes\mathbb{C}^{2}$.
This representation is used to describe entanglement of two qbits
\citep{Wootters1998}.
\item The six-dimensional representation of $\mathrm{SU}\left(4\right)$
labeled by highest weight%
\footnote{Consult Subsection \ref{sub:Representation-theory-of} for the notation
used here.%
}
\[
\lambda_{0}=\left(1,1,0,0\right)\,.
\]
The carrier space of this representation is  isomorphic to the six
dimensional representation of $\mathrm{SU}\left(4\right)$ acting
on $\bigwedge^{2}\mathbb{C}^{4}$. This representation is natural
for the description of the entanglement of two fermions with spin
$S=3/2$ \citep{SchliemannTwoFermions2001,EckertFermions2002}.
\end{enumerate}
It is important to note that in each of those cases there exists an
epimorphism%
\footnote{Epimorphism between two groups is a homomorphism which is surjective. %
} $h$ of the appropriate group $K$ onto the group $\mathrm{SO}\left(N\right)$,
where $N$ is the dimension of the irreducible representation of $K$.
\begin{align}
h:\mathrm{SU}\left(2\right) & \rightarrow\mathrm{SO\left(3\right)}\,\text{(Spin-1 coherent states)}\,,\label{eq:list homomorphism}\\
h:\mathrm{SU}\left(2\right)\times\mathrm{SU}\left(2\right) & \rightarrow\mathrm{SO}\left(3\right)\,\text{(two qbits)\,, }\nonumber \\
h:\mathrm{SU}\left(4\right) & \rightarrow\mathrm{SO}\left(6\right)\,\text{(two four state fermions)\,.}\nonumber 
\end{align}
This observation was first made in \citep{Kus2009}. In what follows
we will prove that these examples are not accidental and are the manifestation
of a rather general principle relating epimorphisms of $K$ and some
$\mathrm{SO}\left(N\right)$, anti-unitary conjugations and the decomposition
of the symmetric power of the representation considered onto irreducible
components.

Before we proceed to the main results of this chapter let us mention
briefly some properties of completely positive maps and Jamiołkowski-Choi
isomorphism. According to Fact \ref{kraus decomposition} every completely
positive map $\Lambda\in\mathcal{CP}\left(\mathcal{H}\right)$ admits
a Kraus decomposition,
\begin{equation}
\Lambda\left(\rho\right)=\sum_{\alpha\in\mathcal{A}}T_{\alpha}\rho T_{\alpha}^{\dagger}\,,\label{eq:krauss decomp again}
\end{equation}
for $T_{\alpha}:\mathcal{H}\rightarrow\mathcal{H}$. A Kraus decomposition
is not unique yet there is a distinguished one associated with the
spectral decomposition of the image of $\Lambda$ under the Jamiołkowski-Choi
isomorphism $\mathcal{J}$,
\begin{equation}
A=\mathcal{J}\left(\Lambda\right)=\left(\mathbb{I}\otimes\Lambda\right)\left(\kb{\Psi}{\Psi}\right)\,,\label{eq:jam again}
\end{equation}
where
\begin{equation}
\ket{\Psi}=\frac{1}{\sqrt{N}}\sum_{i=1}^{N}\ket i\ket i\,,\label{eq:maximally entangled}
\end{equation}
is the maximally entangled state on $\mathcal{H}\otimes\mathcal{H}$.
If $\left\{ \ket{f_{\alpha}}\right\} _{\alpha\in\mathcal{A}}$ is
the orthonormal basis of eigenvectors of $A$ that correspond to (necessary
non-negative) eigenvalues $\left\{ \lambda_{\alpha}\right\} _{\alpha\in\mathcal{A}}$,
we define
\[
T_{\alpha}=\lambda_{\alpha}^{\frac{1}{2}}\left(\bra{\Psi}\otimes\mathbb{I}\right)\left(\mathbb{I}\otimes\ket{f_{\alpha}}\right)\,.
\]
The notation used in the above formula, although commonly used, probably
needs some elucidation. Observe that both $\bra{\Psi}$ and $\ket{f_{\alpha}}$
are linear combinations of simple tensors (the former by its definition,
the latter as an eigenvector of $A\in\mathrm{Herm}\left(\mathcal{H}\otimes\mathcal{H}\right)$).
For simple tensors $\ket a\otimes\ket b,\,\ket c\otimes\ket d$ the
corresponding formula reads
\[
\left(\bra a\otimes\bra b\otimes\mathbb{I}\right)\left(\mathbb{I}\otimes\ket c\otimes\ket d\right)=\kb da\bk bc\,,
\]
which is indeed a linear operator on $\mathcal{H}$. It turns out
that $T_{\alpha}$ indeed form a Kraus decomposition of $\Lambda$.
The importance of this particular Kraus decomposition is twofold.
Firstly, operators $T_{\alpha}$ are orthogonal to each other with
respect to the standard Hilbert–Schmidt inner product on $\mathrm{Lin}\left(\mathcal{H}\right)$.
Secondly, the cardinality of the set $\mathcal{A}$ is minimal. It
is possible to express matrix coefficients of any $A\in\mathrm{Herm}_{+}\left(\mathcal{H}\otimes\mathcal{H}\right)$
in terms of operators from Kraus decomposition \eqref{eq:krauss decomp again}
of the CP map corresponding to it.
\begin{equation}
\bra{\psi_{1}}\bra{\psi_{2}}A\ket{\psi_{3}}\ket{\psi_{4}}=\sum_{\alpha\in\mathcal{A}}\bk{\psi_{1}}{T_{\alpha}\mathcal{K}\psi_{2}}\bk{T_{\alpha}\mathcal{K}\psi_{3}}{\psi_{4}}\,,\label{eq:matrix elements choi}
\end{equation}
where $\mathcal{K}$ is the complex conjugation in the basis $\left\{ \ket i\right\} _{i=1}^{i=N}$
of $\mathcal{H}$ used to define the maximally entangled state \eqref{eq:maximally entangled}
used in the definition of the Jamiołkowski-Choi isomorphism $\mathcal{J}$.
An important class of CP maps is the class of quantum channels, i.e.
CP maps that preserve traces, $\mathrm{tr}\left(\Lambda\left(\rho\right)\right)=\mathrm{tr}\left(\rho\right)$.
On the level of Kraus decomposition this condition reduces to the
requirement that $\sum_{\alpha\in\mathcal{A}}T_{\alpha}T_{\alpha}^{\dagger}=\mathbb{I}$.
How is this condition realized on the level of the operator $A=\mathcal{J}\left(\Lambda\right)\in\mathrm{Herm}_{+}\left(\mathcal{H}\otimes\mathcal{H}\right)$?
The necessary and sufficient condition turns out to be $\mathrm{tr}_{1}\left(A\right)=\mathbb{I}$.
In what follows we focus on the situation when we have some non-negative
$A$ with only one nonzero eigenvalue. As discussed above, this situation
allows us to choose only one Kraus operator in the decomposition of
the corresponding $\Lambda$. If we assume that $\Lambda$ is a quantum
channel we get that the corresponding Kraus operator $T_{\alpha_{0}}$
is unitary,
\[
T_{\alpha_{0}}T_{\alpha_{0}}^{\dagger}=\mathbb{I}\,.
\]
Note that if $T_{\alpha_{0}}T_{\alpha_{0}}^{\dagger}\propto\mathbb{I}$
one can rescale the initial $A$ ($A\rightarrow A'=cA,\, c>0$) so
that resulting $T_{\alpha_{0}}$ is unitary. By the virtue of Eq.\eqref{eq:matrix elements choi},
in the case of unitary $T_{\alpha_{0}}$ expectation value of $A$
can be expressed in terms of anti-unitary operator $\theta=T_{\alpha_{0}}\mathcal{K}$,
\begin{equation}
\bra{\psi}\bra{\psi}A\ket{\psi}\ket{\psi}=\left|\bra{\psi}\theta\ket{\psi}\right|^{2}\,.\label{eq:expectation value choi}
\end{equation}

\subsection{Classification of coherent states characterized by antiunitary conjugation\label{sub:Classification-of-coherent-antiunitary}}

In this part, we characterize in terms of the representation theory
of compact semi-simple Lie groups all situations in which equation
\eqref{eq:simplification formula} holds and explicit computation
of the ``correlation witness'' $g^{\cup}\left(\rho\right)$ is possible
(see Eq.\eqref{eq:explict formula}). Let us first introduce the concept
of anti-unitary conjugation that ‘detects correlations’. It will prove
to be useful in our considerations.
\begin{defn}
\label{definition conjugation}Let $\Pi$ be an irreducible representation
of the compact semi-simple Lie group $K$. We shall say that an anti-unitary
conjugation%
\footnote{In \citep{detection2012,Oszmaniec2014} the formulation of results
was not precise. There it was stated that $\theta$ is just an antiunitary
operator. Here we have corrected this mistake. The results  from \citep{detection2012,Oszmaniec2014}
are not affected.  %
} operator $\theta:\mathcal{H}^{\lambda_{0}}\rightarrow\mathcal{H}^{\lambda_{0}}$
\textit{detects correlations} if and only if it satisfies the following:
\begin{itemize}
\item $\theta$ is $K$-invariant, that is, $\Pi\left(k\right)\theta\Pi\left(k\right)^{\dagger}=\theta$
for each $k\in K$;
\item The expectation value of $\theta$ vanishes exactly on classical states
$\bk{\psi}{\theta\psi}=0$ if and only if $\kb{\psi}{\psi}\in\mathcal{M}$.
\end{itemize}
\end{defn}
We present our results in two theorems. The first relates the existence
of anti-unitary conjugation detecting correlations to the decomposition
of $\mathrm{Sym}^{2}\left(\mathcal{H}^{\lambda_{0}}\right)$ into
irreducible components. The second theorem connects this kind of anti-unitary
conjugation with the existence of epimorphisms of the group $K$ onto
some orthogonal group.
\begin{thm}
\label{theta representation}Let $K$ be a semi-simple, compact and
connected Lie group. Let $\Pi$ be some irreducible unitary representation
of the group $K$ in the Hilbert space $\mathcal{H}^{\lambda_{0}}$
with the highest weight $\lambda_{0}$. The following two statements
are equivalent:
\begin{enumerate}
\item There exist an anti-unitary operator $\theta:\mathcal{H}^{\lambda_{0}}\rightarrow\mathcal{H}^{\lambda_{0}}$
detecting correlations.
\item The following decomposition holds, $\mathrm{Sym}^{2}\left(\mathcal{H}^{\lambda_{0}}\right)=\mathcal{H}^{2\lambda_{0}}\oplus\mathcal{H}^{0}$,
where $\mathcal{H}^{0}$ is the trivial representation of the group
$K$. 
\end{enumerate}
\end{thm}
\begin{proof}
($1\longrightarrow2$) Let $\theta=T\tilde{K}$ where $T$ is an unitary
operator and $\tilde{\mathcal{K}}$ is the operator of the complex
conjugation is some fixed basis of $\mathcal{H}^{\lambda_{0}}$, say$\left\{ \ket i\right\} _{i=1}^{i=N}$.
Define an operator $A\in\mathrm{Herm}_{+}\left(\mathcal{H}^{\lambda_{0}}\otimes\mathcal{H}^{\lambda_{0}}\right)$
as an image of the Jamiołkowski-Choi map (defined with respect to
the basis $\left\{ \ket i\right\} _{i=1}^{i=N}$) of the CP map $\Lambda(\rho)=T\rho T^{\dagger}$
(see Eq.\eqref{eq:jam again}). Matrix elements of $A$ are given
by the following formula (see Eq.\eqref{eq:matrix elements choi}),
\[
\bra{\psi_{1}}\bra{\psi_{2}}A\ket{\psi_{3}}\ket{\psi_{4}}=\bk{\psi_{1}}{T\mathcal{K}\psi_{2}}\bk{T\mathcal{K}\psi_{3}}{\psi_{4}}=\bk{\psi_{1}}{\theta\psi_{2}}\bk{\theta\psi_{3}}{\psi_{4}}\,.
\]
We now claim that the operator $A$ is proportional to $\mathrm{\mathbb{P}^{\mathrm{sym}}-\mathbb{P}^{2\lambda_{0}}}$.
Let us first notice that $A$ is symmetric. This follows from the
fact that $\theta$ is an antiunitary conjugation and thus is Hermitian
(see Definition \ref{antiunitary conjugation}). Therefore we have
\[
\bk{\psi_{1}}{\theta\psi_{2}}=\bk{\psi_{2}}{\theta\psi_{1}},\,\bk{\theta\psi_{3}}{\psi_{4}}=\bk{\theta\psi_{4}}{\psi_{3}}\,.
\]
The operator $A$ is also non-negative which follows from the formula
$\bra{\psi_{1}}\bra{\psi_{2}}A\ket{\psi_{1}}\ket{\psi_{2}}=\left|\bk{\psi_{1}}{\theta\psi_{2}}\right|^{2}$.
It is also $K$-invariant due to the $K$-invariance of $\theta$.
Therefore, by the virtue of Eq.\eqref{eq:decomposition} we have 
\[
A=a_{2\lambda_{0}}\mathbb{P}^{2\lambda_{0}}+\mathbb{V},
\]
where $a_{2\lambda_{0}}\geq0$ and $\mathbb{V}$ is the non-negative
operator commuting with the action of $K$ having the support on $\left(\mathcal{H}^{2\lambda_{0}}\right)^{\perp}\subset\mathrm{Sym}^{2}\left(\mathcal{H}^{\lambda_{0}}\right)$.
By definition, $A$ has only one eigenvector (see our remarks below
Eq.\eqref{eq:matrix elements choi}). Projection on this eigenvector
cannot belong to $\mathcal{H}$ because the expectation value of $\theta$
vanishes on coherent states and consequently $A=\mathbb{V}$. On the
other hand, by \eqref{eq:polynomial characterisation} and the definition
of $\theta$ we have 
\[
\left|\bk{\psi}{\theta\psi}\right|>0\Longleftrightarrow\mathrm{\left(\mathbb{P}^{\mathrm{sym}}-\mathbb{P}^{2\lambda_{0}}\right)}\left(\ket{\psi}\ket{\psi}\right)\neq0\,.
\]
The condition $\left|\bk{\psi}{\theta\psi}\right|>0$ is equivalent
to $A\ket{\psi}\ket{\psi}\neq0$. Therefore we get that
\[
A=c\left(\mathbb{P}^{\mathrm{sym}}-\mathbb{P}^{2\lambda_{0}}\right)\,.
\]
where $c>0$. Consequently $\mathbb{P}^{\mathrm{sym}}-\mathbb{P}^{2\lambda_{0}}$
is a projector onto one dimensional $K$-invariant subspace of $\mathrm{Sym}^{2}\left(\mathcal{H}^{\lambda_{0}}\right)$.
By the fact that $K$ is comapct and simply-connected this representation
must be trivial (see Fact \ref{fact:imposible on dim reps}).

($2\longrightarrow1$) If $\mathrm{Sym}^{2}\left(\mathcal{H}^{\lambda_{0}}\right)=\mathcal{H}^{2\lambda_{0}}\oplus\mathcal{H}^{0}$
then the operator $\mathbb{P}^{\mathrm{sym}}-\mathbb{P}^{2\lambda_{0}}$
has rank 1 and is non-negative. If we apply to it the inverse of the
Jamiołkowski-Choi isomorphism (with respect to some fixed basis $\left\{ \ket i\right\} _{i=1}^{i=N}$)
we get $T=\mathcal{J}^{-1}\left(\mathbb{P}^{0}\right)$. By Eq.\eqref{eq:matrix elements choi}
we have
\begin{equation}
\bra{\psi_{1}}\bra{\psi_{2}}\mathbb{P}^{0}\ket{\psi_{3}}\ket{\psi_{4}}=\bk{\psi_{1}}{T\mathcal{K}\psi_{2}}\bk{T\mathcal{K}\psi_{3}}{\psi_{4}}\,,\label{eq:antiunitary proof 01}
\end{equation}
where $K$ is the complex conjugation in the basis $\left\{ \ket i\right\} _{i=1}^{i=N}$.
We claim that the antilinear operator $\theta=T\mathcal{K}$ is proportional
to the anti-unitary operator detecting correlations. The $K$-invariance
of $\theta$ follows from \eqref{eq:antiunitary proof 01} and $K$-invariance
$\mathbb{P}^{0}$. Because of the decomposition $\mathrm{Sym}^{2}\left(\mathcal{H}^{\lambda_{0}}\right)=\mathcal{H}^{2\lambda_{0}}\oplus\mathcal{H}^{0}$
and Eq.\eqref{eq:antiunitary proof 01},
\[
\bk{\psi}{\theta\psi}=0\,\Longleftrightarrow\kb{\psi}{\psi}\in\mathcal{M}\,.
\]
Moreover, operator $\theta$ is Hermitian, $\theta=\theta^{\dagger}$
which follows from the fact that $\mathbb{P}^{0}$ has support on
$\mathrm{Sym}^{2}\left(\mathcal{H}^{\lambda_{0}}\right)$. The only
thing that needs to be proved is that $T$ can be rescaled to the
unitary operator. This follows from the discussion of the relation
between non-negative operators on the product of Hilbert spaces and
quantum channels explained below Eq.\eqref{eq:antiunitary proof 01}.
The necessary and sufficient condition for $T$ to be proportional
to the unitary operator is $\mathrm{tr}_{1}\left(\mathbb{P}^{0}\right)=\mathbb{I}$.
The operator $\mathbb{P}^{0}$ is the orthogonal projection onto one-dimensional
trivial representation $\mathcal{H}^{0}$ in the decomposition of
$\mathrm{Sym}^{2}\left(\mathcal{H}^{\lambda_{0}}\right)$. It can
be thus written in the form of the integral with respect to the normalized
Haar measure $\mu$ over the whole $K$ \citep{BatutRaczka}.
\[
\mathbb{P}^{0}=\int_{K}\Pi\left(k\right)\otimes\Pi\left(k\right)d\mu\left(k\right)\,.
\]
As a result we have 
\[
\mathrm{tr}_{1}\left(\mathbb{P}^{0}\right)=\int_{K}\mathrm{tr}\left(\Pi\left(k\right)\right)\Pi\left(k\right)d\mu\left(k\right)=\int_{K}\chi_{\lambda_{0}}\left(k\right)\Pi\left(k\right)d\mu\left(k\right)\,,
\]
where $\chi_{\lambda_{0}}\left(k\right)$ is the character of the
representation $\Pi$. By the general representation theory of compact
Lie groups \citep{BatutRaczka}, we have
\[
\int_{K}\chi_{\lambda_{0}}\left(k\right)\Pi\left(k\right)d\mu\left(k\right)=\frac{\mathbb{I}}{N}\,,
\]
Therefore, the proof is completed.
\end{proof}
Note that in the assumptions of the above theorem there is no reference
to the dimension $N$ of the considered representation $\mathcal{H}^{\lambda_{0}}$.
It is nevertheless clear that when $N=1$ and $N=2$ both statements
(that are meant to be equivalent) are at the same time false%
\footnote{For the case $N=1$ this statement follows form the fact that $K$
is compact and simply-connected and therefore one dimensional representation
must be trivial. The case of $N=2$ follows from the fact that (by
the virtue of semi-simplicity) $K$ contains as a subgroup a group
isomorphic to $\mathrm{SU}\left(2\right)$. Consequently for $N=2$
the considered representation $\mathcal{H}^{\lambda_{0}}$ is the
defining represenation of $\mathrm{SU}\left(2\right)$in $\mathbb{C}^{2}$.
In this representation we have $\mathcal{M}=\mathcal{D}_{1}\left(\mathbb{C}^{2}\right)$.%
}. 

The theorem proved above states that cases when operator $\mathbb{P}^{\mathrm{sym}}-\mathbb{P}^{2\lambda_{0}}$
has rank $1$ correspond exactly to the appearance of anti-unitary
conjugations that detect correlations. The following results shows
that such cases are related to the existence of an epimorphism between
the group $K$ and one of three groups: $\mathrm{SO}\left(N\right)$
( for $N=\mathrm{dim}\left(\mathcal{H}^{\lambda_{0}}\right)$), $G_{2}$
or $\mathrm{Spin}\left(7\right)$. 
\begin{thm}
\label{theta epimorphism}Let $K$ be a semi-simple, compact and connected
Lie group. The following two statements are equivalent
\begin{enumerate}
\item There exists an irreducible unitary representation $\Pi$ of the group
$K$ in the Hilbert space $\mathcal{H}^{\lambda_{0}}$ with the highest
weight $\lambda_{0}$ ($N=\mathrm{dim}\left(\mathcal{H}^{\lambda_{0}}\right)>2$).
On $\mathcal{H}^{\lambda_{0}}$ there exists an anti-unitary conjugation
$\theta:\mathcal{H}^{\lambda_{0}}\rightarrow\mathcal{H}^{\lambda_{0}}$
detecting correlations.
\item There exists an epimorphism $h:K\rightarrow\mathrm{SO}\left(N\right)$,
or $h:K\rightarrow G_{2}$ (the exceptional Lie group $G_{2}$ c.f.
\citep{Adams1996}) with $N=7$, or $h:K\rightarrow\mathrm{Spin}\left(7\right)$
with $N=8$.
\end{enumerate}
\end{thm}
\begin{proof}
($1\longrightarrow2$)~Because $\theta$ is anti-unitary conjugation
it is possible to choose the orthonormal basis $\left\{ \ket i\right\} _{i=1}^{i=N}$
of $\mathcal{H}^{\lambda_{0}}$ in such a way that each vector from
the basis is an eigenvector of $\theta$ with an eigenvalue $1$:
\[
\theta\ket i=\ket i\,,\, i=1,\ldots,N\,.
\]
In this basis $\theta$ acts as a complex conjugation,
\[
\theta\left(\sum_{i=1}^{N}\alpha_{i}\ket i\right)=\sum_{i=1}^{N}\bar{v}_{i}\ket i\,.
\]
Let us denote by $\mathcal{H}_{\mathbb{R}}^{\lambda_{0}}$ the real
$N$ subspace of $\mathcal{H}^{\lambda_{0}}$ spanned by real combinations
of vectors stabilized by $\theta$. From the $K$-invariance of $\theta$
it follows that in the basis $\left\{ \ket i\right\} _{i=1}^{i=N}$
operators $\Pi\left(k\right)$ are orthogonal. Because $\Pi:K\rightarrow\mathrm{U}\left(\mathcal{H}^{\lambda_{0}}\right)$
is continuous, the image of a connected group $K$ must be connected
and therefore $\Pi$ defines a homeomorphism $h:K\rightarrow\mathrm{SO}\left(N\right)$.
Note that each vector representative $\ket{\psi}$ of a state $\kb{\psi}{\psi}\in\mathcal{D}_{1}\left(\mathcal{H}^{\lambda_{0}}\right)$
can be decomposed onto its real and imaginary part,
\begin{equation}
\ket{\psi}=\ket u+i\cdot\ket w,\label{eq:real decomposition}
\end{equation}
where $\ket u,\ket w\in\mathcal{H}_{\mathbb{R}}^{\lambda_{0}}$. If
$\kb{\psi}{\psi}$ is a coherent state we have

\[
0=\bk{\psi}{\theta\psi}=\bk uu-\bk vv+2i\cdot\bk uv\,.
\]
Therefore all coherent states are represented by vectors $\ket{\psi}=\ket u+i\cdot\ket w$,
where $\bk uu=\bk ww$ and $\bk uw=0$. In particular this holds for
a highest weight vector $\ket{\psi_{0}}=\ket{u_{0}}+i\cdot\ket{w_{0}}$
for some orthogonal and appropriately normalized $\ket{u_{0}},\ket{w_{0}}\in\mathcal{H}_{\mathbb{R}}^{\lambda_{0}}$.
For $\kb{\tilde{\psi}}{\tilde{\psi}}\in\mathcal{M}$ let $\ket{\tilde{\psi}}=\ket{\tilde{u}}+i\cdot\ket{\tilde{w}}$
be its vector representative. We claim that 
\[
\ket{\tilde{u}}=\Pi\left(k\right)\ket u\,,\,\ket{\tilde{w}}=\Pi\left(k\right)\ket w\,,
\]
for some $k\in K$. Indeed, we must have $\ket{\tilde{\psi}}=\epsilon\Pi\left(k_{1}\right)\ket{\psi_{0}}$($\left|\epsilon\right|=1$,
$k_{1}\in K$) due to the fact that $\mathcal{M}$ is an orbit of
$K$. Because representation $\Pi$ is nontrivial we can chose $k_{2}\in K$
(actually $k_{2}$ belongs to the maximal torus $T\subset K$, see
Subsection \ref{sub:Structural-theory-of}) such that $\epsilon\ket{\psi_{0}}=\Pi\left(k_{2}\right)\ket{\psi_{0}}$.
Therefore, if we take $k=k_{1}k_{2}$ we get the desired result (matrices
corresponding to $\Pi\left(k\right)$ are real in the considered basis).
Thus, it is possible to generate all pairs of orthonormal vectors
by the action of $h\left(K\right)$ on vectors $\ket{u_{0}}$ and
$\ket{v_{0}}$, i.e. $h\left(K\right)$ acts transitively on pairs
of orthonormal vectors (for the formal definition of the action of
a group see Subsection \ref{sub:Lie-groups-and}) from $\mathcal{H}_{\mathbb{R}}^{\lambda_{0}}$.
Consequently, the group $h\left(K\right)$ acts transitively on $M=N-1$
dimensional unit sphere sphere $\mathbb{S}_{N-1}$ in $\mathcal{H}_{\mathbb{R}}^{\lambda_{0}}$.
This action is also effective%
\footnote{The action of the group $K$ on a set $X$ is effective if and only
if give two different elements of a group, $k_{1},k_{2}\in K$ there
exist an element $x\in X$ such that $k_{1}.x\neq k_{2}.x$.%
} because $h\left(K\right)$ is a subgroup of $\mathrm{SO}\left(N\right)$
whose action on $\mathbb{S}_{N-1}$ is effective. This fact suffices
to prove that $h(K)$ equals $\mathrm{SO}\left(N\right)$, $G_{2}$
or $\mathrm{Spin}\left(7\right)$. In order to see this, we refer
to the classical result of Montgomery and Samelson \citep{Montgomery1943}
that classifies all compact and connected Lie groups acting transitively
and effectively on $N$-dimensional spheres. The list of such groups
is short and consists of seven cases: $\mathrm{SO}\left(N\right)$
itself, its three proper subgroups $\mathrm{SU}\left(\frac{N}{2}\right)$,
$\mathrm{Sp}\left(\frac{N}{4}\right)$ and $\mathrm{Sp}\left(1\right)\times\mathrm{Sp}\left(\frac{N}{4}\right)$
(where $\mathrm{Sp}\left(\cdot\right)$ denotes the compact symplectic
group%
\footnote{The compact symplectic group $\mathrm{Sp}\left(N\right)$ is defined
by $\mathrm{U}(2N)\cap\mathrm{Sp}(2N,\mathbb{C})$, where $\mathrm{Sp}(2N,\mathbb{C})$
is a group preserving the standard anti-linear form on $\mathbb{C}^{2N}$.%
}), $G_{2}\subset\mathrm{SO\left(7\right)}$, $\mathrm{Spin}\left(7\right)\subset\mathrm{SO}\left(8\right)$,
and $\mathrm{Spin}\left(9\right)\subset\mathrm{SO}\left(16\right)$.
We first consider the last three ''exceptional'' cases. Groups $G_{2}$,
$\mathrm{Spin}\left(7\right)$ and $\mathrm{Spin}\left(9\right)$
act transitively on, respectively, $6$-, $7$- and $15$-dimensional
spheres. Those actions come from the following (faithful) representations:
defining representation of $G_{2}$, eight-dimensional spinor representation
of $\mathrm{Spin}\left(7\right)$ and $16$ dimensional spinor representation
of $\mathrm{Spin}\left(9\right)$. Actions of $G_{2}$ and $\mathrm{Spin}\left(7\right)$
are transitive on orthonormal pairs of vectors (see \citep{Adams1996},
page 32). Therefore those groups are permissible. On the other hand,
it is known \citep{Klyachko2008} that the $16$-dimensional representation
of $\mathrm{Spin}\left(9\right)$ does not have the desired property.
Let us now consider the special unitary and symplectic subgroups of
$\mathrm{SO}\left(N\right)$. Those groups can appear only when $2$
(in the case of $\mathrm{SU}\left(\frac{N}{2}\right)$) or $4$ (in
cases of $\mathrm{Sp}\left(\frac{N}{4}\right)$ and $\mathrm{Sp}\left(1\right)\times\mathrm{Sp}\left(\frac{N}{4}\right)$)
are divisors of $N$. Therefore when $N$ is odd the proof is finished.
Now assume that $2$ or $4$ divides $N$. Since $h\left(K\right)$
acts transitively on orthonormal pairs of vectors, a stabilizer subgroup
$\mathrm{Stab}\left(\ket{u_{0}}\right)\subset h\left(K\right)$ must
act transitively on
\[
\mathbb{S}_{N-2}=\mathbb{S}_{N-1}\cap\ket{u_{0}}^{\perp}\,,
\]
where $\ket{u_{0}}^{\perp}$ is the orthogonal complement of $\ket{u_{0}}$
in $\mathcal{H}_{\mathbb{R}}^{\lambda_{0}}$. We can now apply the
theorem of Montgomery and Samelson for the dimension $N-1$. Since
$N-1$ is now odd, we infer that we have three possibilities
\begin{enumerate}
\item $\mathrm{Stab}\left(\ket{u_{0}}\right)=\mathrm{Spin}\left(7\right)$,
$N=8$;
\item $\mathrm{Stab}\left(\ket{u_{0}}\right)=\mathrm{Spin}\left(9\right)\subset SO\left(16\right)$,
$N=16$;
\item $\mathrm{Stab}\left(\ket{u_{0}}\right)=\mathrm{SO}\left(N-1\right)$,
$N$ - arbitrary.
\end{enumerate}
In the first case we have that $\mathrm{Stab}\left(\ket{u_{0}}\right)=\mathrm{Spin}\left(7\right)$
must be a subgroup of $h\left(K\right)$ which is either $SU\left(4\right)$,
or $Sp\left(2\right)$ or $\mathrm{Sp}\left(1\right)\times\mathrm{Sp}\left(2\right)$.
These cases can be however rejected since the equality
\begin{equation}
\mathrm{dim}\left(\mathrm{Stab}\left(\ket{u_{0}}\right)\right)+\mathrm{dim}\left(\mathbb{S}_{N-1}\right)=\mathrm{dim}\left(h\left(K\right)\right)\,,\label{eq:dimension equality}
\end{equation}
is not satisfied. In the second case we have that $\mathrm{Stab}\left(\ket{u_{0}}\right)=\mathrm{Spin}\left(9\right)$
a subgroup of $h\left(K\right)$ which is either $SU\left(8\right)$,
or $Sp\left(4\right)$ or $\mathrm{Sp}\left(4\right)\times\mathrm{Sp}\left(2\right)$.
The cases of $Sp\left(4\right)\text{ or }\mathrm{Sp}\left(4\right)\times\mathrm{Sp}\left(2\right)$
can be disregarded again by the virtue of \eqref{eq:dimension equality}.
The case of $\mathrm{Stab}\left(\ket{u_{0}}\right)=SU\left(8\right)$
is not obvious but form the discussion contained on page 12 of \citep{Friedrich1999}
it follows that $\mathrm{Spin}\left(9\right)\cap SU\left(8\right)\neq\mathrm{Spin}\left(9\right)$
and we have to reject also this possibility. Consequently we have
to consider the last possibility: $\mathrm{Stab}\left(\ket{u_{0}}\right)=\mathrm{SO}\left(N-1\right)$.
As a consequence of \eqref{eq:dimension equality}, we have
\[
\mathrm{dim}\left(h\left(K\right)\right)\geq\frac{\left(N-1\right)\left(N-2\right)}{2}=\mathrm{dim}\left(\mathrm{SO\left(N-1\right)}\right)\,.
\]
Since the dimensions of $\mathrm{SU}\left(\frac{N}{2}\right)$, $\mathrm{Sp}\left(\frac{N}{4}\right)$
and $\mathrm{Sp}\left(1\right)\times\mathrm{Sp}\left(\frac{N}{4}\right)$
are, respectively, $\frac{N^{2}}{4}-1,\,\frac{N}{4}\left(\frac{N}{2}+1\right)$
and $\frac{N}{4}\left(\frac{N}{2}+1\right)+3$ we can exclude those
groups. At the end, we conclude that only possibilities are that $h(K)=\mathrm{SO}\left(N\right)$,
$h(K)=G_{2}$ (when $N=7$) or $h(K)=\mathrm{Spin}\left(7\right)$
(when $N=8$).

($2\longrightarrow1$) We treat groups $\mathrm{SO}\left(N\right)$,
$G_{2}$ and $\mathrm{Spin}\left(7\right)$ together. We consider
defining representations of $\mathrm{SO}\left(N\right)$ and $G_{2}$
and the eight-dimensional spinor representation of $\mathrm{Spin}\left(7\right)$.
We shall show that symmetric powers of those irreducible faithful
representations (clearly those are also irreducible representations
of the group $K)$ decompose into two ingredients: $\mathrm{Sym}^{2}\left(\mathcal{H}^{\lambda_{0}}\right)=\mathcal{H}^{2\lambda_{0}}\oplus\mathcal{H}^{0}$.
Then, combining this with Theorem \ref{theta representation}, we
conclude the existence of the anti-unitary operator $\theta$ that
detects coherent states for each of considered representations. To
prove the above decomposition, we note that each representation respects
the Euclidean structure in the relevant $\mathcal{H}_{\mathbb{R}}^{\lambda_{0}}$
(when viewed as subgroups of $\mathrm{SO}\left(N\right)$, $\mathrm{SO}\left(7\right)$
and $\mathrm{SO}\left(8\right)$ accordingly) and therefore we have
a following chain of equivalences of representations
\begin{align}
\mathrm{Sym}^{2}\left(\mathcal{H}^{\lambda_{0}}\right) & \approx\mathbb{C}\otimes\mathrm{Sym}^{2}\left(\mathcal{H}_{\mathbb{R}}^{\lambda_{0}}\right)\approx\mathbb{C}\otimes\mathrm{SEnd}\left(\mathcal{H}_{\mathbb{R}}^{\lambda_{0}}\right)\,,\label{eq:first decomp}
\end{align}
where $\mathrm{SEnd}\left(\mathcal{H}_{\mathbb{R}}^{\lambda_{0}}\right)$
is the vectors space of symmetric real operators on $\mathcal{H}_{\mathbb{R}}^{\lambda_{0}}$.
The action of the relevant group $h\left(K\right)$ on $X\in\mathrm{SEnd}\left(\mathcal{H}_{\mathbb{R}}^{\lambda_{0}}\right)$
is via conjugation,
\[
X\rightarrow g.X=gXg^{T}\,,
\]
where $g\in h\left(K\right)\subset\mathrm{SO}\left(N\right)$. The
decomposition $\mathrm{SEnd}\left(\mathcal{H}_{\mathbb{R}}^{\lambda_{0}}\right)$
onto irreducible components of $h\left(K\right)$ reads
\begin{equation}
\mathrm{SEnd}\left(\mathcal{H}_{\mathbb{R}}^{\lambda_{0}}\right)\approx\mathrm{Lin}_{\mathbb{R}}\mathbb{I}\oplus\mathrm{SEnd}_{0}\left(\mathcal{H}_{\mathbb{R}}^{\lambda_{0}}\right)\,,\label{eq:second decomp}
\end{equation}
where $\mathrm{Lin}_{\mathbb{R}}\mathbb{I}$ is the one dimensional
trivial representation spanned by identity operator and and $\mathrm{SEnd}_{0}\left(\mathcal{H}_{\mathbb{R}}^{\lambda_{0}}\right)$
denote the real vector vectors space of symmetric traceless operators
on $\mathcal{H}_{\mathbb{R}}^{\lambda_{0}}.$ The component $\mathrm{SEnd}_{0}\left(\mathcal{H}_{\mathbb{R}}^{\lambda_{0}}\right)$
is irreducible which follows from the transitivity of the action of
the each group on pairs of orthonormal vectors. Taking into account
\eqref{eq:first decomp} and \eqref{eq:second decomp} and using the
fact that complexification of a real irreducible representation remains
irreducible we get $\mathrm{Sym}^{2}\left(\mathcal{H}^{\lambda_{0}}\right)=\mathcal{H}^{2\lambda_{0}}\oplus\mathcal{H}^{0}$
which finishes the proof.
\end{proof}
From the proof of Theorem \ref{theta epimorphism} we get the following
Corollary.
\begin{cor}
\label{action of the orthogonal group}Let $K$ be a semi-simple,
compact and connected Lie group irreducibly represented, via representation
$\Pi$, in Hilbert space $\mathcal{H}^{\lambda_{0}}$. Assume that
on $\mathcal{H}^{\lambda_{0}}$ there exists an anti-unitary conjugation
$\theta:\mathcal{H}^{\lambda_{0}}\rightarrow\mathcal{H}^{\lambda_{0}}$
detecting correlations. Let $\mathcal{H}_{\mathbb{R}}^{\lambda_{0}}$
be the $N$ dimensional subspace subspace spanned by real combinations
of vectors satisfying $\theta\ket{\psi}=\ket{\psi}$. Then the following
holds
\begin{itemize}
\item The representation $\Pi$ preserves $\mathcal{H}_{\mathbb{R}}^{\lambda_{0}}$,
i.e. $\Pi\left(k\right)$ are orthogonal in the orthonormal basis
of $\mathcal{H}_{\mathbb{R}}^{\lambda_{0}}$.
\item Given two pairs of orthogonal vectors from $\mathcal{H}_{\mathbb{R}}^{\lambda_{0}}$,
$\left(\ket{\psi_{1}},\ket{\psi_{2}}\right)$, $\left(\ket{\phi_{1}},\ket{\phi_{2}}\right)$,
there exists an element $k\in K$ such that 
\[
\ket{\phi_{1}}=\Pi\left(k\right)\ket{\psi_{1}}\,,\,\ket{\phi_{2}}=\Pi\left(k\right)\ket{\psi_{2}}\,.
\]
In other words the group $K$ acts transitively, via representation
$\Pi$ on pairs of orthonormal vectors from $\mathcal{H}_{\mathbb{R}}^{\lambda_{0}}$.
\end{itemize}
\end{cor}
Let us now remark on the obtained results.
\begin{itemize}
\item The list of groups appearing in point two of Theorem \eqref{theta epimorphism}
consists precisely of groups appearing in Theorem 3.8.1 in \citep{Klyachko2008}.
In the cited paper Alexander Klyachko considers the generalization
of the concept of entanglement based on the analogies between some
aspect of the entanglement theory and geometric invariant theory (see
Section \ref{sec:General-motivation}). The groups $\mathrm{SO}\left(N\right),\, G_{2}$
and $\mathrm{Spin}\left(7\right)$ correspond precisely to irreducible
representations of the group $K$ for which ``all unstable%
\footnote{The notion of ``unstable states'' in $\mathcal{H}^{\lambda_{0}}$
is related to the action of the complexified group $G=K^{\mathbb{C}}$
in $\mathcal{H}^{\lambda_{0}}$. For details see \citep{Klyachko2008,Sawicki2012,Sawicki2014}. %
} states are coherent''. This is not entirely accidental as ``systems
in which all unstable states are coherent’ considered by Klyachko
in his paper can be, in fact, equivalently characterized by our Theorem
\ref{theta representation}.
\item The existence of the anti-unitary conjugation $\theta:\mathcal{H}^{\lambda_{0}}\rightarrow\mathcal{H}^{\lambda_{0}}$
commuting with the representation $\Pi$ of the group $K$ is one
of the equivalent conditions \citep{FultonHarris} that ensure that
representation $\Pi$ is real, i.e. there exist a basis of $\mathcal{H}^{\lambda_{0}}$
such all operators $\Pi\left(k\right)$ in this basis are orthogonal.
We have used this fact in the course of the proof of \eqref{theta epimorphism}.
\item In the proof of Theorem \eqref{theta epimorphism}, we referred to
the classical work by Montgomery and Samelson \citep{Montgomery1943}.
Although it may seem to be a trick from a rather ‘high floor’, we
would like to stress that the problem is not as easy as it may seem
at the first sight. It turns out that when $N$ is even, there are
proper subgroups of $\mathrm{SO}\left(N\right)$ that act transitively
on $\mathbb{S}_{N-1}$ (this fact is directly related to the classification
of the holonomy%
\footnote{The holonomy group of a connected Riemannian manifold $\mathcal{M}$
of dimension $N$ is a subgroup of the orthogonal group $\mathrm{O}\left(N\right)$
defined by the Levi-Civita connection on $\mathcal{M}$ in the following
way. To every smooth loop $\gamma:\left[0,1\right]\rightarrow\mathcal{M}$
based at $x\in\mathcal{M}$ a Levi-Civita connection associates a
unique parellel transoprt $P_{\gamma}:T_{x}\mathcal{M}\rightarrow T_{x}\mathcal{M}$.
The holonomy group $\mathrm{Hol}_{x}\left(\mathcal{M}\right)$ based
at $x\in\mathcal{M}$ is a group generated by pararell transport maps
associated to all loops,
\[
\mathrm{Hol}_{x}\left(\mathcal{M}\right)=\left\{ \left.P_{\gamma}\right|\,\gamma:\left[0,1\right]\rightarrow\mathcal{M},\,\gamma\left(0\right)=\gamma\left(1\right)=x\right\} \,.
\]

Given any two points $x,y\in\mathcal{M}$ we have $\mathrm{Hol}_{y}\left(\mathcal{M}\right)=g_{yx}\mathrm{Hol}_{x}\left(\mathcal{M}\right)g_{yx}^{T}\,$
where $g_{xy}\in\mathrm{O}\left(N\right)$. Therefore the holonomy
group is defined uniqelly, up to isomorphism. For the comprehensive
introduction to the concept of holonomy group see \citep{DiffGeomPhys}.%
} groups of irreducible non-locally symmetric Riemannian spaces \citep{Olmos2005}).
Nevertheless, our assumption about the existence of an ``anti-unitary
conjugation detecting correlations'' is strong enough to guarantee
that the image of the homomorphism we consider is the whole $\mathrm{SO}\left(N\right),\, G_{2}$
or $\mathrm{Spin}\left(7\right)$. 
\end{itemize}
We presented the group-theoretical conditions for the cases when the
anti-unitary conjugation detecting correlations exists and it is possible
to compute $g^{\cup}\left(\rho\right)$ (see Eq. \eqref{eq:convex roof antiunitary})
exactly. Our results reproduce all physically-relevant cases where
anti-unitary conjugation was known to exist (c.f. Eq.\eqref{eq:list homomorphism}).
In the next section we apply the obtained results to study ``non-Gaussian''
correlations in fermionic systems, a type of correlations that does
not appear on the list \eqref{eq:list homomorphism}.

\section{Classical simulation of Fermionic Linear Optics augmented with noisy
ancillas\label{sec:Classical-simulation-of-FLO}}

For any model of quantum computation it is vital to characterize its
computational power. Probably the most important practical question
is how a given model compares to universal classical, respectively
quantum, computation%
\footnote{In the circuit model of classical computation a discrete set of gates
$G$ (elementary operations that can be performed on a input string
$x\in2^{\left\{ 1,\ldots,N\right\} }$) is said to be \textit{universal
}if and only if it is possible to compute arbitrary function
\[
f:2^{\left\{ 1,\ldots,N\right\} }\rightarrow\left\{ 0,1\right\} \,,
\]
via the composition of elements from $\mathcal{G}$ (for arbitrary
natural number $n$). 

In the circuit model of quantum computing a finite set of gates $\mathcal{G}\subset S\mathrm{U}\left(\left(\mathbb{C}^{2}\right)^{\otimes L}\right)$
is said to be \textit{computationally universal} if by product of
elements from $\mathcal{G}$ we can aproximate arbitrary unitary operation
$U\in\mathrm{SU}\left(\left(\mathbb{C}^{2}\right)^{\otimes L}\right)$.
The foundations quantum computation are beyound the scope of this
thesis. For the comprehensive introduction to this exiting field of
study see \citep{NielsenChaung2010} or \citep{Kitaev2002}.%
}. If protocols allowed by the model are efficiently simulable on a
classical computer, the corresponding physical system may be accessible
to numerical studies, but is unlikely to be a suitable candidate for
a quantum computer. On the other hand, simulability by quantum circuits
ensures that the underlying physics can be effectively studied using
a quantum computer \citep{Feynman1982}. Lastly, if the resources
provided by the model enable one to implement a universal quantum
computation, the corresponding physical system is a candidate for
the realization of a quantum computer. In this part we apply result
of the previous section to study the problem of classical simulation
of model of quantum computation in which Fermionic Linear Optics (FLO)
is augmented with noisy ancilla state. In this model of quantum computation
the auxiliary state $\rho\in\mathcal{D}\left(\mathcal{H}_{\mathrm{Fock}}\left(\mathbb{C}^{d}\right)\right)$
from the $d$-mode Fock space is a parameter of the model. If $\rho$
belong to the convex hull, $\mathcal{M}_{g}^{c}$, of pure Gaussian
states $\mathcal{M}_{g}$ (see Subsection \ref{sub:Fermionic-Gaussian-states})
than the corresponding model of computation is classically simulable
\citep{powernoisy2013}. In turns out that using the methods developed
in the previous section we will be able to characterize the set of
convex-Gaussian fermionic states $\mathcal{M}_{g}^{c}$ in the first
nontrivial case of $d=4$ fermionic modes. 

The section is structured as follows. In Subsection \ref{sub:Definition-of-ancilla-asisted}
we introduce the model of quantum computation based on FLO augmented
with noisy ancilla states and present the relation between FLO and
the model of Topological Quantum Computation (TQC) based on braiding
of Ising anyons. In Subsection \ref{sub:Analitical-characterization-of}
we present the analytical characterization of four mode convex-Gaussian
states in four mode fermionic Fock space thus partially solving the
problem of classical simulation of above mentioned model of computation
(in particular we solve an open problem recently posed in \citep{powernoisy2013}).
We also derive the analogue of Schmidt decomposition for pure fermionic
states with fixed parity in $\mathcal{H}_{\mathrm{Fock}}\left(\mathbb{C}^{4}\right)$.
This allows us to describe the geometry of the inclusion $\mathcal{M}_{g}^{c}\subset\mathcal{D}\left(\mathcal{H}_{\mathrm{Fock}}\left(\mathbb{C}^{d}\right)\right)$.
We conclude our considerations in Subsection \ref{sub:Discussion FLO}
where we discuss importance of the obtained results to the computation
model based on ancilla-assisted FLO. Throughout the section we will
use the notation from Subsection \ref{sub:Fermionic-Gaussian-states}.\pagebreak{}

\subsection{Ancilla-assisted Fermionic Linear Optics\label{sub:Definition-of-ancilla-asisted}}

Many physically motivated models of quantum computation are defined
by specifying the available set of initial states, gates and measurements
\citep{Gottesman1998,Raussendorf2003}. The computation model based
on Fermionic Linear Optics was introduced in \citep{Terhal2002}.
The allowed operations of the model are:
\begin{enumerate}
\item[i.] Preparation of the Fock vacuum $\kb 00\in\mathcal{D}_{1}\left(\mathcal{H}_{\mathrm{Fock}}\left(\mathbb{C}^{d}\right)\right)$;
\item[ii.] Measurement of the occupation numbers $\hat{n}_{k}=a_{k}^{\dagger}a_{k}=\frac{1}{2}\left(\mathbb{I}+ic_{2k-1}c_{2k}\right)$
for any mode $k=1,\ldots,d$;
\item[iii.] Evolution under the von-Neumann equation, $\frac{d}{dt}\rho=-i\left[H,\,\rho\right]$,
for time $t$. The Hamiltonian $H=i\sum_{k,l=1}^{2N}h_{kl}c_{k}c_{l}$
is an arbitrary Hamiltonian quadratic in Majorana operators. This
operation is equivalent to application of arbitrary parity preserving
Bogolyubov transformation $U\in\mathcal{B}$ (see Eq.\eqref{eq:bogolubov transf})
on a state $\rho$. 
\end{enumerate}
Operations (i), (ii) and (iii) can be preformed in arbitrary order
and may depend upon measurement results obtained during the previous
stages of the computation. The protocol concludes with the final measurement
whose (binary) outcome is the result of the computation. The above
model of computation can be efficiently simulated in a polynomial
time on a probabilistic classical computer (a classical computer having
access to random bits). The proof (see \citep{Terhal2002} for details)
relies on the flowing properties of the model
\begin{itemize}
\item Efficient encoding of Gaussian states%
\footnote{Recall that the class of mixed Gaussian states consists of states
satisfying Eq.\eqref{eq:mixed gaussian def}. In general Gaussian
states form a subset of convex-Gaussian states, $\mathrm{Gauss}\subset\mathcal{M}_{g}^{c}$.%
} $\mathrm{Gau}ss\subset\mathcal{D}_{1}\left(\mathcal{H}_{\mathrm{Fock}}\left(\mathbb{C}^{d}\right)\right)$.
Gaussian states in $d$ mode fermionic systems are fully described
by their correlation matrix $M\left(\rho\right)$ (c.f. Eq.\eqref{eq:correlation matrix}),
with $O\left(d^{2}\right)$ elements.
\item FLO transformations (3) map Gaussian states onto Gaussian states,
with the efficient update rule. Action of $U\in\mathcal{B}$ induces
the orthogonal transformation on $M\left(\rho\right)$ (given by Eq.\eqref{eq:transformation})
that can be evaluated in $O\left(d^{3}\right)$ steps.
\item Efficient read-out of measurement probability distributions occupation
number. Thanks to Fermionic Wick’s theorem \citep{powernoisy2013,Lagr2004}
the probability of measuring the population in fermionic modes can
be evaluated in $O\left(d^{3}\right)$ steps. Furthermore, number
operator measurements transform Gaussian states onto Gaussian states.
\end{itemize}
Thee Fermionic Linear Optics model of quantum computation presented
above is the fermionic analogue of well-know bosonic linear optics
\citep{Kok2007,Aaronson2011}.  Fermionic linear optics describes
systems of non-interacting fermions, i.e. fermionic systems that can
be described exactly by Bogolyubov mean field theory. The model of
computation based on FLO alone is not computationally universal and
can, as described above, be effectively simulated by a classical probabilistic
computer \citep{BravyiKoeningSimul,Terhal2002}. Nevertheless, the
physics beyond FLO is rich and captures a number of systems of interest
for condensed matter physics, including Kitaev’s Majorana chain \citep{Kitaev2001},
honeycomb model \citep{Kitaev2006}, $\nu=5/2$ fractional Quantum
Hall systems \citep{Alicea2012}. These systems possess a topological
order and can be used as fault-tolerant quantum memories \citep{Mazza2013}
or Topological Quantum Computation (TQC) with Ising anyons \citep{Kitaev2006}.
This motivates an interest in extending FLO in such a way that the
resulting model will become computationally universal. One of the
possible extensions is the model in which FLO is augmented with a
noisy ancilla state. In this model, introduced in \citep{powernoisy2013},
the traditional model of FLO is extended by introducing additional
fermionic modes in which, at the beginning of computation, one stores
certain number of copies of a, perhaps noisy, ancilla state. Initial
state of the system in the original ``computational'' modes remains
the vacuum state and the class of allowed operations remains intact.
Formally, the model (i-iii) is extended by the following operation.
\begin{enumerate}
\item[iv.]  Multiple usage of the even%
\footnote{For the definition of even fermionic states see Subsection \ref{sub:Fermionic-Gaussian-states}.%
} ancilla state $\rho\in\mathcal{D}_{\mathrm{even}}\left(\mathcal{H}_{\mathrm{Fock}}\left(\mathbb{C}^{m}\right)\right)$
that is stored in auxiliary $m$ fermionic modes. If $k$ auxiliary
states are available, the total Hilbert space of the system becomes%
\footnote{This isomorphism is basis-dependent. The linear mapping $\phi:\mathcal{H}_{\mathrm{Fock}}\left(\mathbb{C}^{d_{1}}\right)\otimes\mathcal{H}_{\mathrm{Fock}}\left(\mathbb{C}^{d_{2}}\right)\rightarrow\mathcal{H}_{\mathrm{Fock}}\left(\mathbb{C}^{d_{1}+d_{2}}\right)$
defined on a simple tensors by
\[
\phi\left(\ket{n_{1},\ldots,n_{d}}\otimes\ket{n_{1'},\ldots,n_{d_{2}'}}\right)=\ket{n_{1},\ldots,n_{d},n_{1'},\ldots,n_{d_{2}'}}\,,
\]
gives an isomorphism $\mathcal{H}_{\mathrm{Fock}}\left(\mathbb{C}^{d_{1}}\right)\otimes\mathcal{H}_{\mathrm{Fock}}\left(\mathbb{C}^{d_{2}}\right)\approx\mathcal{H}_{\mathrm{Fock}}\left(\mathbb{C}^{d_{1}+d_{2}}\right)$.%
}
\begin{equation}
\mathcal{H}_{\mathrm{Fock}}\left(\mathbb{C}^{d}\right)\otimes\left(\mathcal{H}_{\mathrm{Fock}}\left(\mathbb{C}^{m}\right)\right)^{\otimes k}\approx\mathcal{H}_{\mathrm{Fock}}\left(\mathbb{C}^{d+m\cdot k}\right)\label{eq:extra fock}
\end{equation}
 and we allow arbitrary operations of the form (2) and (3) to be performed
on the initial state of the total system of the form $\kb 00\otimes\rho^{\otimes k}$. 
\end{enumerate}
This model of computation is analogous to other ancilla-assisted models
of quantum computation such as ancilla-assisted computation with Clifford
gates \citep{Bravyi2005,Veitch2014} or ancilla-assisted TQC with
Ising anyons \citep{universalfracBravyi}. It will be convenient for
us to describe the model of TQC with Ising anyons in more detail.
This model is defined essentially in the same manner as FLO \citep{universalfracBravyi}.
The only difference is in the step (iii) of the model:
\begin{enumerate}
\item[iii'.] Application of the braiding unitary operation $B_{p,q}\in\mathcal{B}$
defined by
\begin{equation}
B_{p,q}=\mathrm{exp}\left(-\frac{\pi}{4}c_{p}c_{q}\right)\,,\,1\leq p<q\leq2d\label{eq:braiding operator}
\end{equation}

\end{enumerate}
The rule (iii') changes the model dramatically - it distinguishes
one particular basis of a single-particle Hilbert space $\mathbb{C}^{d}$.
We will not discuss here in detail the motivation beyond (iii'). Let
us only mention that $B_{p,q}$ represents nonabelian exchange between
``particles'' associated to Majorana operators $c_{k}$. The operation
\eqref{eq:braiding operator} models \citep{universalfracBravyi}
a ``topologically protected'' operation of exchange of excitations
in two dimensional electron gas in Fractional quantum hall regime
characterized by the filling factor $\nu=\frac{5}{2}$. The ancilla-assisted
TQC is defined by extending the model by the rule (iv).

~

In \citep{universalfracBravyi} it was shown that ancilla-assisted
TQC with Ising anyons can be promoted to the universal model of computation
if one has access to the following ancilla states%
\footnote{By ``universal model of computation'' we understand the following:
it is possible to approximate arbitrary unitary operator $U\in\mathrm{SU}\left(\mathcal{H}_{\mathrm{Fock}}^{+}\left(\mathbb{C}^{d}\right)\right)$
if we can perform operations (ii) and (iii') on the Hilbert space
$\mathcal{H}_{\mathrm{Fock}}\left(\mathbb{C}^{d+k_{1}\cdot2+k_{2}\cdot4}\right)$
provided initially the state of the system is $\kb 00\otimes\rho_{1}^{\otimes k_{1}}\otimes\rho_{2}^{\otimes k_{2}}$,
for suitable $k_{1},k_{2}$. %
}
\begin{itemize}
\item $\rho_{1}\in\mathcal{D}\left(\mathcal{H}_{\mathrm{Fock}}^{+}\left(\mathbb{C}^{2}\right)\right)$
such that $\mathrm{tr}\left(\rho_{1}\kb{a_{4}}{a_{4}}\right)\geq0.86$,
where $\kb{a_{4}}{a_{4}}\in\mathcal{D}_{1}\left(\mathcal{H}_{\mathrm{Fock}}^{+}\left(\mathbb{C}^{2}\right)\right)$
is some fixed pure state;
\item $\rho_{2}\in\mathcal{D}\left(\mathcal{H}_{\mathrm{Fock}}^{-}\left(\mathbb{C}^{2}\right)\right)$
such that $\mathrm{tr}\left(\rho_{2}\kb{a_{8}}{a_{8}}\right)\geq0.62$,
where $\kb{a_{8}}{a_{8}}\in\mathcal{D}_{1}\left(\mathcal{H}_{\mathrm{Fock}}^{+}\left(\mathbb{C}^{4}\right)\right)$
is some fixed pure state (see below).
\end{itemize}
In \citep{powernoisy2013} it was proven that for $m\leq3$ every
pure state $\kb{\psi}{\psi}\in\mathcal{D}_{1}\left(\mathcal{H}_{\mathrm{Fock}}^{+}\left(\mathbb{C}^{m}\right)\right)$
is Gaussian. Therefore, having access to FLO we can always generate
the state $\kb{a_{4}}{a_{4}}$ from the vacuum. 

Let us now come back to the discussion of ancilla-assisted FLO. The
discussion above shows that, depending on the properties of the auxiliary
state $\rho$, it may be possible to implement, with the help of traditional
FLO operations, gates that are necessary for computational universality
on registers describing the actual computation (this follows from
the fact that the class of operation allowed by TQC with Ising anyons
is a subclass of operations allowed by FLO). 

On the other hand the computation model (i-iv) was shown \citep{powernoisy2013}
to be effectively classically simulable if the auxiliary state $\rho$
is convex-Gaussian, i.e.
\[
\rho=\sum_{i}p_{i}\kb{\psi_{i}}{\psi_{i}}\,,
\]
where $\kb{\psi_{i}}{\psi_{i}}\in\mathcal{M}_{g}$. The simulation
scheme consists of sampling pure Gaussian states $\left\{ \kb{\psi_{i}}{\psi_{i}}\right\} $
according to the probability distribution $\left\{ p_{i}\right\} $
followed by the classical simulation of the evolution of Gaussian
states described above.  This is the reason why the characterization
of the convex-Gaussian states $\mathcal{M}_{g}^{c}$ is important
for this model of computation. We would like to remark that convex-Gaussian
ancilla states lead to an effectively classically simulable model
also when one replaces FLO with its dissipative counterpart, recently
introduced in \citep{BravyiKoeningSimul}. For this reason results
presented in this section are also valid for the dissipative FLO.

\subsection{Analytical characterization of four mode fermionic convex-Gaussian
states\label{sub:Analitical-characterization-of}}

In this subsection we characterize analytically the set of convex-Gaussian
states $\mathcal{M}_{g}^{c}$ in the first nontrivial case of $d=4$
modes. From the discussion given in the previous subsection it follows
that this result is relevant for the classical simulability of the
model of computation in which FLO is assisted by noisy ancilla state
(convex-Gaussianity of the ancilla state implies the classical simulability
of the model). In this subsection we state and prove a number of technical
results. As a byproduct of our considerations we also solve an open
problem posed recently in the paper of de Melo et. al. \citep{powernoisy2013}.
The implications of obtained results to the above mentioned model
of quantum computation will be discussed in the next section. 
\begin{fact}
(\citep{powernoisy2013}) Any even pure state $\kb{\psi}{\psi}\in\mathcal{D}_{\mathrm{even}}\left(\mathcal{H}_{\mathrm{Fock}}\left(\mathbb{C}^{d}\right)\right)$
for $d\leq3$ is Gaussian.
\end{fact}
Before we present our results let us recall (c.f. subsection \ref{sub:Fermionic-Gaussian-states})
that all pure fermionic Gaussian states $\mathcal{M}_{g}$ are even,
i.e. they commute with the total parity operator $Q$. Consequently
operators belonging to the convex hull $\mathcal{M}_{g}^{c}$ are
also even operators. In what follows we will need the following decomposition
of even operators $X\in\mathrm{Herm}_{\mathrm{even}}\left(\mathcal{H}_{\mathrm{Fock}}\left(\mathbb{C}^{d}\right)\right)$,

\begin{equation}
X=\alpha_{0}\mathbb{I}+\sum_{k=1}^{d}i^{k}\sum_{1\leq l_{1}<l_{2}<\ldots<l_{2k}\leq2d}\alpha_{l_{1}l_{2}\ldots l_{2k}}c_{l_{1}}c_{l_{2}}\ldots c_{l_{2k}}\,,\label{eq:even decomposition2}
\end{equation}
where all the coefficients $\alpha_{0}$ and $\alpha_{l_{1}l_{2}\ldots l_{2k}}$
are real. 

Let $\rho\in\mathcal{D}_{\mathrm{even}}\left(\mathcal{H}_{\mathrm{Fock}}\left(\mathbb{C}^{4}\right)\right)$
be an arbitrary four mode even mixed state having the decomposition
\eqref{eq:even decomposition2}. Let
\[
\rho_{+}=\mathbb{P}_{+}\rho\mathbb{P}_{+}\,,\,\rho_{-}=\mathbb{P}_{-}\rho\mathbb{P}_{-}
\]
denote restrictions of $\rho$ to $\mathcal{H}_{\mathrm{Fock}}^{+}\left(\mathbb{C}^{d}\right)$
and $\mathcal{H}_{\mathrm{Fock}}^{-}\left(\mathbb{C}^{d}\right)$,
respectively. By $\tilde{X}$ we denote the ``complex conjugate''
of the operator $X$, i.e., an operator constructed from $X$ by changing
all $i$ to $-i$ in the decomposition \eqref{eq:even decomposition2}.
Let us introduce non-negative functions $C_{+}$ and $C_{-}$ that
are the analogues of Uhlmann-Wotters concurrence from the previous
section. They are defined by
\begin{equation}
C_{\pm}\left(\rho\right)=\mathrm{max}\left\{ 0,\,\lambda_{1}^{\pm}-\sum_{k=2}^{8}\lambda_{k}^{\pm}\right\} ,\label{eq:gen concurrence}
\end{equation}
where $\left(\lambda_{1}^{\pm},\,\lambda_{2}^{\pm},\ldots,\,\lambda_{8}^{\pm}\right)$
denote non-increasingly ordered eigenvalues of the operator $\sqrt{\rho_{\pm}\tilde{\rho}_{\pm}}$.
\begin{thm}
\label{analitical characterization gaussian} Convex-Gaussianity of
$\rho\in\mathcal{D}_{\mathrm{even}}\left(\mathcal{H}_{\mathrm{Fock}}\left(\mathbb{C}^{4}\right)\right)$
is characterized by the values of these generalized concurrences,
\begin{equation}
\rho\text{ is convex-Gaussian }\Longleftrightarrow C_{+}\left(\rho\right)=C_{-}\left(\rho\right)=0\,.\label{eq:criterion gaussian}
\end{equation}
Moreover, an arbitrary four mode convex-Gaussian state $\rho\in\mathcal{M}_{g}^{c}$
can be represented as a convex combination
\[
\rho=\sum_{i=1}^{i=\mathcal{N}}p_{i}\kb{\psi_{i}}{\psi_{i}}\,,\,
\]
of at most $\mathcal{N}=16$ of pure Gaussian states $\kb{\psi_{i}}{\psi_{i}}\in\mathcal{M}_{g}$.\end{thm}
\begin{proof}
Let us first note that pure fermionic Gaussian states have a fixed
parity. In other words: $\mathcal{G}=\mathcal{M}_{g}^{+}\cup\mathcal{M}_{g}^{-}$,
where $\mathcal{M}_{g}^{\pm}\subset\mathcal{D}_{\mathrm{1}}\left(\mathcal{H}_{\mathrm{Fock}}^{\pm}\left(\mathbb{C}^{d}\right)\right)$
(see Figure \ref{fig:gaussian split}). For this reason it is enough
to consider the problem of convex-Gaussianity separately on $\mathrm{Fock}_{\pm}\left(\mathbb{C}^{d}\right)$.
In other words an even state $\rho$ is convex-Gaussian if and only
if both $\rho_{+}$ and $\rho_{-}$ are convex-Gaussian. We will show
below that there exist antiunitary operators $\theta_{\pm}$, each
acting on $\mathrm{Fock}_{\pm}\left(\mathbb{C}^{4}\right)$, such
that
\begin{equation}
\kb{\psi}{\psi}\in\mathcal{M}_{g}^{\pm}\Longleftrightarrow C_{\pm}\left(\kb{\psi}{\psi}\right)=\left|\bk{\psi}{\theta_{\pm}\psi}\right|=0\,.\label{eq:pure concurrence}
\end{equation}
We can now use the Uhlmann-Wotters construction and results of the
previews section to compute the convex roof extension of $C_{\pm}$,
\[
C_{\pm}\left(\sigma\right)=\mathrm{inf}_{\sum\kb{\psi_{i}}{\psi_{i}}=\sigma}\left(\sum_{i}C_{\pm}\left(\kb{\psi_{i}}{\psi_{i}}\right)\right)
\]
for $\sigma$ a non-negative operator on $\mathrm{Fock}_{\pm}\left(\mathbb{C}^{4}\right)$.
From the definition of the convex roof extension and the discussion
above we have that $\rho$ is convex-Gaussian if and only if
\[
C_{+}\left(\rho_{+}\right)=C_{-}\left(\rho_{-}\right)=0\,.
\]
Explicit formulas for $C_{\pm}\left(\rho_{\pm}\right)$ are given
by \eqref{eq:gen concurrence}, where $\left(\lambda_{1}^{\pm},\,\lambda_{2}^{\pm},\ldots,\,\lambda_{8}^{\pm}\right)$
denote non-increasingly ordered eigenvalues of the operator $\sqrt{\rho_{\pm}\theta_{\pm}\rho_{\pm}\theta_{\pm}}$
(see Eq.\eqref{eq:explict formula}). From Fact \ref{fact uhlmann construction}
it follows that for convex-Gaussian states supported in $\mathrm{Fock}_{\pm}\left(\mathbb{C}^{4}\right)$
we need at most $\mathcal{N}=8=\mathrm{dim}\left(\mathrm{Fock}_{\pm}\left(\mathbb{C}^{4}\right)\right)$
pure Gaussian states in the convex decomposition. Consequently, for
an arbitrary convex-Gaussian state in $\mathrm{Fock}\left(\mathbb{C}^{4}\right)$
this number equals $\mathcal{N}=16$.

The existence of antiunitary operators $\theta_{\pm}$ follows from
group-theoretical interpretation of pure Gaussian states $\mathcal{M}_{g}^{\pm}$
presented in Subsection \ref{sub:Fermionic-Gaussian-states}. The
group of Bogolyubov transformations $\mathcal{B}$ is precisely the
image of a representation of $\mathrm{Spin}\left(2d\right),$ a compact
semi-simple Lie group. The Hilbert space $\mathrm{Fock}\left(\mathbb{C}^{d}\right)$
decomposes into two irreducible representations of $\mathrm{Spin}\left(2d\right)$:
$\mathrm{Fock}_{+}\left(\mathbb{C}^{m}\right)$ and $\mathrm{Fock}_{-}\left(\mathbb{C}^{m}\right)$
respectively. Sets of pure Gaussian states $\mathcal{M}_{g}^{\pm}\subset\mathcal{D}_{\mathrm{1}}\left(\mathcal{H}_{\mathrm{Fock}}^{\pm}\left(\mathbb{C}^{d}\right)\right)$
are precisely the highest weight orbits of this group in $\mathrm{Spin}\left(2d\right)$.
Semi-simple compact Lie groups $K$ and irreducible representations
$\mathcal{H}^{\lambda_{0}}$  admitting an antiunitary operator $\theta:\mathcal{H}^{\lambda_{0}}\rightarrow\mathcal{H}^{\lambda_{0}}$
detecting the orbit through the highest-weight vector had been classified
in Theorems \ref{theta representation} and \ref{theta epimorphism}.
In order to guarantee the existence of such $\theta$ it suffices
to check that the following decomposition holds:
\begin{equation}
\mathrm{Sym}^{2}\left(\mathcal{H}^{\lambda_{0}}\right)=\mathcal{H}^{2\lambda_{0}}\oplus\mathcal{H}^{0}\,,\label{eq:irrep decomp}
\end{equation}
where $\mathcal{H}^{2\lambda_{0}}$ is an irreducible representation
of $K$ characterized by the highest weight $2\lambda_{0}$, and $\mathcal{H}^{0}$
is a trivial (one dimensional) representation of $K$. In our case
we have $K=\mathrm{Spin}\left(8\right)$ and $\mathcal{H}^{\lambda_{0}}=\mathcal{H}_{\mathrm{Fock}}^{+}\left(\mathbb{C}^{d}\right)$
or $\mathcal{H}^{\lambda_{0}}=\mathcal{H}_{\mathrm{Fock}}^{-}\left(\mathbb{C}^{d}\right)$.
For these particular representations the decomposition \eqref{eq:irrep decomp}
indeed holds (This follows from Eq. \eqref{eq:closed form expression ferm},
see also, for example \citep{Manivel2009}) and thus existence of
$\mathrm{Spin}\left(8\right)$-invariant antiunitaries $\theta_{\pm}$
is thus guaranteed. We conclude the proof of \eqref{eq:criterion gaussian}
by showing that
\[
\theta_{\pm}X\theta_{\pm}=\tilde{X}
\]
for every operator $X$ supported in either $\mathrm{Fock}_{+}\left(\mathbb{C}^{4}\right)$
or $\mathrm{Fock}_{-}\left(\mathbb{C}^{4}\right)$. We present here
a proof only for the case when $X$ has a support in $\mathrm{Fock}_{+}\left(\mathbb{C}^{4}\right)$.
The desired property of $\theta_{+}$ follows from its invariance
under the action of $\mathrm{Spin}\left(8\right)$. Because $\mathrm{Spin}\left(8\right)$
is generated by anti-Hermitian operators $c_{i}c_{j}$, it follows
that $c_{i}c_{j}\theta_{+}=\theta_{+}c_{i}c_{j}$ for every pair of
Majorana operators. Using the fact that $\theta_{+}$ satisfies (as
every antiunitary operator) $\theta_{+}^{2}=\mathbb{I}$, $\theta i=-i\theta$
and noting that every operator $X$ with support in $\mathrm{Fock}_{+}\left(\mathbb{C}^{4}\right)$
is an even operator (and thus has a decomposition \eqref{eq:decomposition})
proves $\theta_{\pm}X\theta_{\pm}=\tilde{X}$. 
\end{proof}
Let us remark on the result given above.
\begin{itemize}
\item The maximal number of pure Gaussian $\mathcal{N}=16$ needed to represent
arbitrary four mode convex-Gaussian state is much smaller then the
upper bound $\tilde{\mathcal{N}}=48$ obtained in \citep{powernoisy2013}. 
\item It is natural to ask whether similar results hold also for number
of modes greater than $4$. Unfortunately the decomposition \eqref{eq:irrep decomp}
does not hold for the symmetric product of $\mathrm{Spin}\left(2d\right)$
representations $\mathcal{H}_{\mathrm{Fock}}^{\pm}\left(\mathbb{C}^{d}\right)$
for $d>4$. \end{itemize}
\begin{example}
\label{terhal problem}Let us now apply Theorem \ref{analitical characterization gaussian}
to give the noise threshold $p_{cr}$ above which a depolarization
of the state $\kb{a_{8}}{a_{8}}\in\mathcal{D}_{\mathrm{even}}\left(\mathcal{H}_{\mathrm{Fock}}\left(\mathbb{C}^{4}\right)\right)$
becomes convex-Gaussian. In other words we consider a state
\begin{equation}
\rho\left(p\right)=\left(1-p\right)\kb{a_{8}}{a_{8}}+p\frac{\mathbb{I}}{16}\,,\label{eq:depolarised a8}
\end{equation}
where $p\in[0,\,1${]}, $\mathbb{I}$ is the identity operator on
$\mathrm{\mathrm{Fock}_{+}\left(\mathbb{C}^{4}\right)}$. The pure
state $\kb{a_{8}}{a_{8}}$ is defined by
\begin{equation}
\kb{a_{8}}{a_{8}}=\frac{1}{16}\left(\mathbb{I}+S_{1}\right)\left(\mathbb{I}+S_{2}\right)\left(\mathbb{I}+S_{3}\right)\left(\mathbb{I}+Q\right)\,,\label{eq:a8 definition}
\end{equation}
where $S_{1}=-c_{1}c_{2}c_{5}c_{6}$, $S_{2}=-c_{2}c_{3}c_{6}c_{7}$,
$S_{3}=-c_{1}c_{2}c_{3}c_{4}$. The state $\kb{a_{8}}{a_{8}}$ can
be used to implement a $\mathrm{CNOT}$ gate that is needed to promote
FLO to be computationally universal \citep{powernoisy2013,universalfracBravyi}.
The problem of finding $p_{cr}$ was considered in \citep{powernoisy2013}
where authors showed that $\rho\left(p\right)$ is non-convex-Gaussian
for $p\leq p_{1}^{\ast}=\frac{8}{15}$ and is convex-Gaussian for
$p\geq p_{2}^{\ast}=\frac{8}{9}$. Application of criterion \eqref{eq:criterion gaussian}
to \eqref{eq:depolarised a8} is straightforward because $\tilde{\rho}\left(p\right)=\rho\left(p\right)$.
Simple algebra shows that $\rho(p)$ is convex-Gaussian if and only
if $p\geq\frac{8}{11}=p_{cr}$. This result is particularly interesting
as it opens a possibility for existence of more noise-resilient protocols
of distillation of the state pure $\kb{a_{8}}{a_{8}}$ from copies
of a noisy state $\rho\left(p\right)$ via FLO or TQC with Ising Anyons
(protocol based on TQC introduced in \citep{universalfracBravyi}
works for $p\leq0.4$). The Figure \ref{fig:comparison of results}
presents the comparison of results presented here to the ones given
in \citep{powernoisy2013}.
\end{example}
\begin{figure}[h]
\noindent \begin{centering}
\includegraphics[width=13cm]{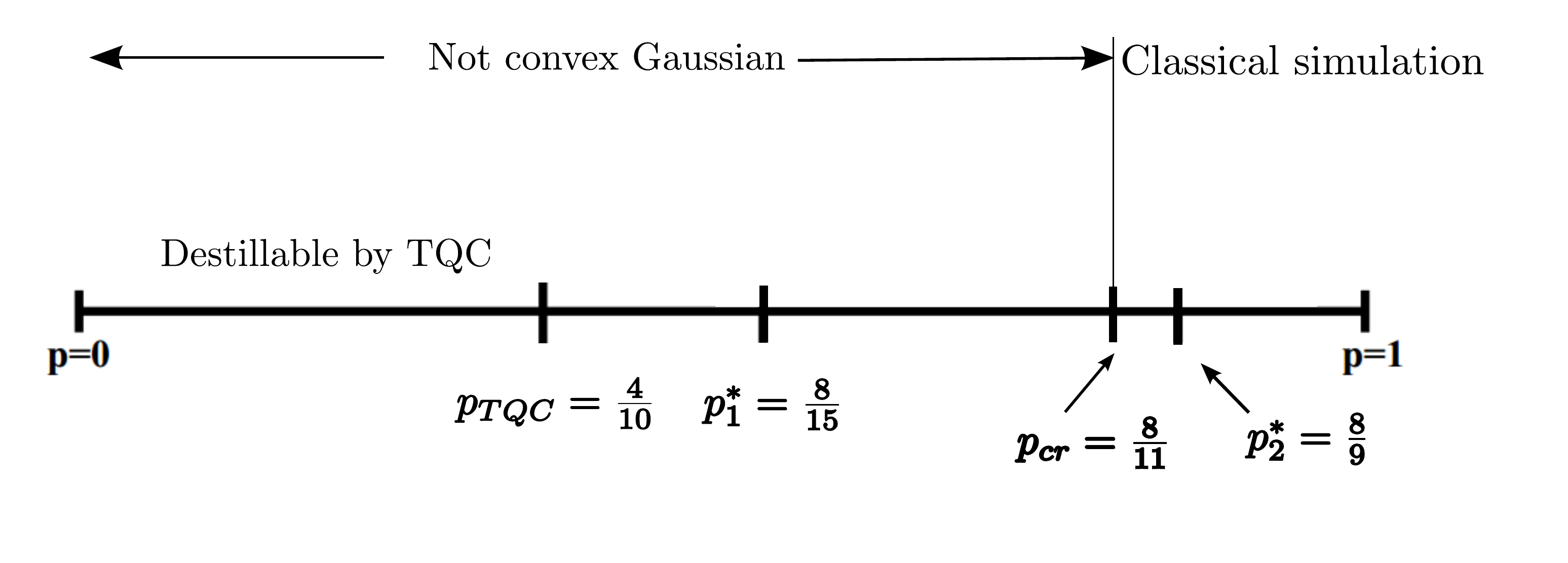}
\par\end{centering}

\noindent \protect\caption{\label{fig:comparison of results}The properties of a state $\rho\left(p\right)$
as a function of the depolarization parameter $p$. The state $\rho\left(p\right)$
is not convex-Gaussian for $p\in\left[0,p_{\mathrm{cr}}\right)$ and
is convex-Gaussian for $p\in\left[p_{\mathrm{cr}},1\right]$, where
$p_{\mathrm{cr}}=\frac{8}{11}$. For $p\in\left[0,\frac{4}{10}\right)$
the state $\rho\left(p\right)$ can be used to distill \citep{universalfracBravyi},
via TQC which is a subset of FLO, pure states $\protect\kb{a_{8}}{a_{8}}$that
can be subsequently used to make FLO compuationally universal. The
parameters $p_{1}^{\ast}=\frac{8}{15}$ and $p_{2}^{\ast}=\frac{8}{9}$
are respectivelly lower and upper bounds for $p_{\mathrm{cr}}$ obtained
in \citep{powernoisy2013}.}
\end{figure}

Insights from the proofs of results given in the previous section
allow us to describe the geometry of the pure-Gaussian and convex-Gaussian
states in four mode fermionic Fock space. We present our results in
the form of two technical Lemmas. For simplicity we restrict our considerations
to the subspace $\mathcal{H}_{\mathrm{Fock}}^{+}\left(\mathbb{C}^{4}\right)$.
The analogous statements hold also for states defined on $\mathcal{H}_{\mathrm{Fock}}^{-}\left(\mathbb{C}^{4}\right)$.
\begin{lem}
\label{eneralised Schmidt decomposition}Let $\kb{\psi}{\psi}\in\mathcal{D}\left(\mathcal{H}_{\mathrm{Fock}}^{+}\left(\mathbb{C}^{4}\right)\right)$.
Let $\ket{\psi}\in\mathcal{H}_{\mathrm{Fock}}^{+}\left(\mathbb{C}^{4}\right)$
be a normalized vector representing the state $\kb{\psi}{\psi}$.
Up to a global phase $\ket{\psi}$ can be expressed as a combination
of orthogonal vectors representing Gaussian states: $\ket{\psi_{G}}$
and $\theta_{+}\ket{\psi_{G}}$,
\begin{equation}
\ket{\psi}=\sqrt{1-p^{2}}\ket{\psi_{G}}+p\theta_{+}\ket{\psi_{G}}\,,\label{eq:generalised Schmidt}
\end{equation}
where $p\in\left[0,\frac{1}{\sqrt{2}}\right]$ and $\ket{\psi_{G}}=U\ket 0$
for $U\in\mathcal{B}$. The expression \eqref{eq:generalised Schmidt}
can be considered as a generalized Schmidt decomposition%
\footnote{It is important to remark that in \eqref{eq:generalised Schmidt}
the phase of $\ket{\psi_{G}}$ \emph{does} matter (unless $\kb{\psi}{\psi}$
is itself Gaussian).%
}. \end{lem}
\begin{proof}
The proof is given in the Section \ref{sec:Proofs-concerning-Chapter rigorous}
Appendix (see page \ref{sub:Proof-ofgeneral schmidt}).
\end{proof}
Let us note that $p=0$ in Eq.\eqref{eq:generalised Schmidt} corresponds
to the set of pure Gaussian states in $\mathcal{D}\left(\mathcal{H}_{\mathrm{Fock}}^{+}\left(\mathbb{C}^{4}\right)\right)$.
On the other hand states for which $p=\frac{1}{\sqrt{2}}$ belong
to the orbit $\mathcal{B}$ through the state $\kb{a_{8}}{a_{8}}$.
\begin{lem}
\label{Fildelity distance four mode}Let $F\left(\rho,\sigma\right)=\left(\mathrm{tr}\left[\sqrt{\sqrt{\rho}\sigma\sqrt{\rho}}\right]\right)^{2}$
denotes Uhlmann fidelity between states $\rho,\sigma\in\mathcal{D}\left(\mathcal{H}_{\mathrm{Fock}}^{+}\left(\mathbb{C}^{4}\right)\right)$.
The following formula holds 
\begin{equation}
F_{\mathrm{Gauss}}\left(\rho\right)=\mathrm{max}_{\sigma\in\mathcal{M}_{g}^{+}}F\left(\rho,\sigma\right)=\frac{1}{2}+\frac{1}{2}\sqrt{1-C_{+}^{2}\left(\rho\right)}\,,\label{eq:Fidelity}
\end{equation}
where $C_{+}\left(\rho\right)$ is given by \eqref{eq:gen concurrence}.\end{lem}
\begin{proof}
The proof is given in the Section \ref{sec:Proofs-concerning-Chapter rigorous}
of the Appendix (see page \pageref{sub:Proof-of-Lemma fidelity}).
\end{proof}
Let us note that Fuchs-van de Graaf inequalities \citep{Fuchs1999}
can be used to bound the statistical (trace) distance \citep{GeometryQuantum}
of any state $\rho$ on $\mathcal{H}_{\mathrm{Fock}}^{+}\left(\mathbb{C}^{4}\right)$
to the set of convex-Gaussian states, $D\left(\rho,\mathcal{G}^{c}\right)=\mathrm{min}_{\sigma\in\mathcal{G}^{c}}\frac{1}{2}\sqrt{\mathrm{tr}\left(\left(\rho-\sigma\right)^{2}\right)}$,
by
\begin{equation}
1-\sqrt{F_{\mathrm{Gauss}}\left(\rho\right)}\leq D\left(\rho,\mathcal{G}^{conv}\right)\leq\sqrt{1-F_{\mathrm{Gauss}}\left(\rho\right)}\,.\label{eq:statistical distance}
\end{equation}
Inequalities \eqref{eq:statistical distance} together with \eqref{eq:Fidelity}
show that for a not convex-Gaussian state $\rho$, with support on
$\mathcal{H}_{\mathrm{Fock}}^{+}\left(\mathbb{C}^{4}\right)$, the
generalized concurrence $C_{+}\left(\rho\right)$ can be used do assess
the resilience of the property of being not convex-Gaussian against
noise.

\subsection{Discussion\label{sub:Discussion FLO}}

We have presented a complete analytical characterization \eqref{eq:criterion gaussian}
of convex-Gaussian states in the four-mode fermionic Fock space $\mathrm{Fock}\left(\mathbb{C}^{4}\right)$.
Using methods taken from entanglement theory and theory of Lie groups
we described quantitatively (c.f. \eqref{eq:Fidelity}) how the property
of being not convex-Gaussian is resilient to noise. These result have
immediate consequence for the computation power of models in which
FLO or TQC with Ising anyons are assisted with a noisy ancilla state.
This follows form the fact that computations augmented with convex-Gaussian
states are automatically classically simulable. We have used our methods
to give a precise value of the noise threshold $p_{cr}=\frac{8}{11}$
above which the state $\kb{a_{8}}{a_{8}}$ (used to make FLO or TQC
computationally universal \citep{powernoisy2013}), when depolarized,
becomes convex-Gaussian. This result is especially interesting as
the threshold value $p_{cr}$ is much higher then previously known
lower bounds. This opens a possibility for the existence of protocols
based on FLO that would distill $\kb{a_{8}}{a_{8}}$ from a noisy
state $\rho$ and would be much more noise tolerant then the currently
known ones. At the and we would like to mention that our method is
limited to only to the case of $d=4$ fermionic modes. From the proof
of Theorem \ref{analitical characterization gaussian} if follows
that it is impossible to extend the presented reasoning beyond $d=4$
modes. A general method for detecting whether a state $\rho\in\mathcal{D}_{\mathrm{even}}\left(\mathcal{H}_{\mathrm{Fock}}\left(\mathbb{C}^{d}\right)\right)$
is not convex-Gaussian was given in \citep{powernoisy2013}. Authors
of \citep{powernoisy2013} introduced a hierarchy of criteria based
on the application of quantum de-Finetti theorem and a polynomial
characterization of the set of pure fermionic Gaussian states. The
hierarchy was proved to be exact (every not convex-Gaussian state
is detected on some stage of the hierarchy). However, due to its high
numerical complexity, the method was unfeasible in practice even for
the case of four fermionic modes.

\section{Summary and open problems\label{sec:sumarry Exact}}

In this chapter we investigated cases in which coherent states $\mathcal{M}\subset\mathcal{D}_{1}\left(\mathcal{H}^{\lambda_{0}}\right)$
of compact simply-connected Lie group $K$ (irreducibly represented
on in the Hilbert space $\mathcal{H}^{\lambda_{0}}$) admit a characterization
in terms of the expectation value of the anti-unitary conjugation
\begin{equation}
\kb{\psi}{\psi}\in\mathcal{M}\,\Longleftrightarrow\,\left|\bk{\psi}{\theta\psi}\right|=0\,.\label{eq:again conjugation}
\end{equation}
This characterization enables, by the virtue of Uhlmann-Wotters construction
(c.f. Subsection \ref{sub:Antiunitary-conjugations-andUW}), to derive
a closed-form analytic criterion characterizing the convex hull $\mathcal{M}^{c}\subset\mathcal{D}\left(\mathcal{H}^{\lambda_{0}}\right)$
of the set $\mathcal{M}$ (the convex hull $\mathcal{M}^{c}$ consists
exclusively of non-correlated states if we chose $\mathcal{M}$ as
the class of ``non-correlated'' pure states),
\begin{equation}
\rho\in\mathcal{M}^{c}\,\Longleftrightarrow\,\lambda_{1}\leq\sum_{k=2}^{N}\lambda_{k}\,,\label{eq:criterion analytic again}
\end{equation}
where $\left(\lambda_{1},\lambda_{2},\ldots,\lambda_{N}\right)$ is
the non-increasingly ordered spectrum of the operator $\sqrt{\rho\theta\rho\theta}$.
In Section \ref{sec:Ulhmann-Wotters} we classified, using methods
of group theory and Riemannian geometry, all the cases when the class
of coherent states can be described by the condition \eqref{eq:again conjugation}.
The main results (which have been published in \citep{detection2012})
are contained in Theorems \ref{theta representation} and \ref{theta epimorphism}.
There, the characterization \eqref{eq:again conjugation} is associated
with particular properties of the group $K$ and the decomposition
of the representation $\mathrm{Sym}^{2}\left(\mathcal{H}^{\lambda_{0}}\right)$
onto irreducible components. In Section \ref{sec:Classical-simulation-of-FLO}
we apply Theorems \ref{theta representation} and \ref{theta epimorphism}
to the problem of classical simulation of a computation model in which
Fermionic Linear Optics is augmented with noisy ancilla state (which
is a parameter of the model). If the auxiliary state $\rho$ can be
written as a convex combination of pure fermionic Gaussian states
$\mathcal{M}_{g}\subset\mathcal{H}_{\mathrm{Fock}}\left(\mathbb{C}^{d}\right)$,
then the corresponding model of computation is classically simulable.
Using tools developed in Section \ref{sec:Ulhmann-Wotters} and group-theoretical
interpretation of $\mathcal{M}_{g}$ (c.f. Subsection \ref{sub:Fermionic-Gaussian-states})
we gave a complete characterization of convex-Gaussian in the first
nontrivial case of $d=4$ modes (Theorem \ref{analitical characterization gaussian}).
As a byproduct of our considerations we solved an open problem recently
given in \ref{terhal problem} (see Example \ref{terhal problem}).
We also introduced a generalized Schmidt decomposition for four mode
pure fermionic states with fixed parity (see Lemma \ref{eneralised Schmidt decomposition}).
This decomposition can be of importance for the distillation protocols
which can promote ancilla-assisted FLO or ancilla-assisted TQC to
the universal quantum computer.

\subsubsection*{Open problems}

We now state open problems connected with the results presented in
this chapter.
\begin{itemize}
\item Find the physical relevance of the ``exceptional cases'' (Spinor
representation of $\mathrm{Spin}\left(7\right)$ and defining representation
of $G_{2}$) in which coherent states are characterized by \eqref{eq:again conjugation}
(see Theorem \ref{theta epimorphism})
\item Develop improved protocols that would purify the state $\kb{a_{8}}{a_{8}}$
from a noisy mixed state from $\mathcal{D}\left(\mathcal{H}_{\mathrm{Fock}}\left(\mathbb{C}^{d}\right)\right)$
and will be based on FLO or TQC. We expect that the analogy between
entanglement and ``non-Gaussianity'' stated in \eqref{eq:generalised Schmidt}
will contribute to this goal.
\item Develop a full resource theory (analogous to the one existing in the
context of entanglement \citep{EntantHoro} or ancilla-assisted Clifford
computation\citep{Veitch2014}) for the ancilla-assisted FLO or TQC. 
\end{itemize}

\chapter{Polynomial criteria for detection of correlations for mixed states
\label{chap:Polynomial-mixed states}}

This chapter deals with the problem of detection of correlations for
mixed quantum states. In Chapter \ref{chap:Multilinear-criteria-for-pure-states}
we provided a polynomial characterization of various interesting classes
of non-correlated pure quantum states. In the most general form our
polynomial characterization is the following
\begin{equation}
\kb{\psi}{\psi}\in\mathcal{M}\,\Longleftrightarrow\mathrm{tr}\left(A\left[\kb{\psi}{\psi}^{\otimes k}\right]\right)=0\,,\label{eq:characterisation multilin}
\end{equation}
where $\mathcal{M}\subset\mathcal{D}_{1}\left(\mathcal{H}\right)$
is the relevant class of pure states and $A:\mathrm{Sym}^{k}\left(\mathcal{H}\right)\rightarrow\mathrm{Sym}^{k}\left(\mathcal{H}\right)$
is a Hermitian operator characterizing (via \eqref{eq:characterisation multilin})
the class $\mathcal{M}$. In accordance with the discussion presented
in Chapter \ref{chap:Introduction} we define the class of non-correlated
states (with respect to the class $\mathcal{M}$) as a convex hull
of $\mathcal{M}$, denoted by%
\footnote{In what follows we will assume that $\mathcal{H}$ is finite dimensional
and ignore the necessity to work with topological closures (in, say,
trace norm) of the purely algebraic convex hulls. %
} $\mathcal{M}^{c}$,
\begin{equation}
\mathcal{M}^{c}=\left\{ \rho=\sum_{i}p_{i}\kb{\psi_{i}}{\psi_{i}}\left|\kb{\psi_{i}}{\psi_{i}}\in\mathcal{M},\, p_{i}\geq0,\,\sum_{i}p_{i}=1\right.\right\} \,.\label{eq:non-correlated states}
\end{equation}
In the previous chapter we characterized cases when the group-theoretical
characterization of the set $\mathcal{M}$ allows for the complete
analytical description of the set $\mathcal{M}^{c}$. In this Chapter
we will investigate the following general problem. 
\begin{problem}
(Detection problem) \label{detection problem}Let $\mathcal{M}\subset\mathcal{D}_{1}\left(\mathcal{H}\right)$
be a class of ``non-correlated'' pure states characterized via \eqref{eq:characterisation multilin}.
Is it possible to derive, using the characterization \eqref{eq:characterisation multilin},
criteria for correlations for states $\rho\in\mathcal{D\left(\mathcal{H}\right)}$? 
\end{problem}
In this section we will partially solve the above problem by deriving
a collection of correlation criteria that will take the general form:
\begin{equation}
\mathrm{tr}\left(V\left[\rho_{1}\otimes\ldots\otimes\rho_{k}\right]\right)>0\,\Longrightarrow\,\rho_{1}\notin\mathcal{M}^{c}\,\text{(\ensuremath{\rho_{1}}is correlated)}\,,\label{eq:criteria mixed idea}
\end{equation}
where $\rho_{1},\ldots,\rho_{L}\in\mathcal{D}\left(\mathcal{H}\right)$
and $V:\mathcal{H}^{\otimes k}\rightarrow\mathcal{H}^{\otimes k}$
is a Hermitian operator suitably tailored for the class $\mathcal{M}$.
The structure of \eqref{eq:criteria mixed idea} resembles the standard
criterion based on entanglement witnesses (see discussion in Section
\ref{sec:Methods-from-entanglement}) and for this reason we will
call operator $V$ appearing in \eqref{eq:criteria mixed idea} \textit{a
multilinear correlation witness}%
\footnote{We decided to use a different sign convention than the standard one
used in the context of entanglement witnesses (see Section \ref{sec:Methods-from-entanglement}).
This convention is a residue of the original motivation that led us
to study criteria of the form \eqref{eq:criteria mixed idea}. We
arrived at criteria of this form by trying to find lower bounds for
``generalized concurrences'' (see Section \ref{sec:Methods-from-entanglement}).
Consequantly the positive value of a function giving a lower bound
was implying the correlations present in a given state.%
}. 

The variety of criteria in which multilinear witnesses appear (for
different number $k$ and different classes $\mathcal{M}$) will be
presented in the forthcoming sections. Before we concentrate on particular
examples we first present a general discussion of criteria based on
multilinear correlation witnesses. In what follows we will concentrate
on multilinear witness that allow to infer about correlations of all
states that appear in the criterion,
\begin{equation}
\mathrm{tr}\left(V\left[\rho_{1}\otimes\ldots\otimes\rho_{k}\right]\right)>0\,\Longrightarrow\,\rho_{1},\ldots,\rho_{k}\notin\mathcal{M}^{c}\,\text{(\ensuremath{\rho_{1}},\ensuremath{\ldots},\ensuremath{\rho_{k}} are correlated)}\,.\label{eq:criteria mixed idea stonger}
\end{equation}
Multilinear correlation witness satisfying \eqref{eq:criteria mixed idea stonger}
automatically satisfy \eqref{eq:criteria mixed idea} but not vice
versa. Below we list advantages of studying multilinear witnesses. 
\begin{itemize}
\item The general form of multilinear correlation witness \eqref{eq:criteria mixed idea}
contains, as a specific example, the usual entanglement witnesses%
\footnote{It is enough to put $V=-W\otimes\mathbb{I}^{\otimes\left(k-1\right)}$,
where $W$ is the usual entanglement witness.%
}. As it was pointed out in Section \ref{sec:Methods-from-entanglement},
the usual linear witnesses also suffice to characterize the set of
non-correlated states $\mathcal{M}^{c}$. However, for arbitrary class
of pure states $\mathcal{M}$ (and the corresponding set of non-correlated
states) we do not have a structural theory of the corresponding linear
witnesses (see Section \ref{sec:Methods-from-entanglement}). For
this reason construction of criteria based on multilinear correlation
witnesses for a fairly general classes of pure states $\mathcal{M}$
(given by Eq. \eqref{eq:characterisation multilin}) can be beneficial.
We will use such a criterion in next chapter to derive an interesting
quantitative characterization of typical correlation properties of
general mixed states defined on the Hilbert space $\mathcal{H}$. 
\item Criteria \eqref{eq:criteria mixed idea stonger} allow, in principle,
to certify correlations in multiple mixed states $\rho_{1},\ldots,\rho_{k}$
via the measurement of the expectation value of a single collective
operator $V$.
\item Multilinear witnesses allow to consider nontrivial criteria for detection
of correlations that are invariant under the action of the relevant
symmetry group. This would not have been possible if we had restricted
ourselves only to criteria based on linear witnesses (see Section
\ref{sec:Optimal--bilinear} for details). Invariance of the criterion
under the action of the symmetry group means that for all states $\rho_{1},\rho_{2},\ldots,\rho_{k}\in\mathcal{D}\left(\mathcal{H}\right)$
we have
\begin{equation}
\mathrm{tr}\left(V\left[\rho_{1}\otimes\rho_{2}\otimes\ldots\otimes\rho_{k}\right]\right)=\mathrm{tr}\left(V\left[\rho'_{1}\otimes\rho'_{2}\ldots\otimes\rho'_{k}\right]\right)\,,\label{eq:invariance multilin}
\end{equation}
where $\rho'_{i}=U\rho_{i}U^{\dagger}$ and $U$ is a unitary implementation
of some element of the symmetry group (note that the same $U$ acts
on all parties $i=1,2,\ldots,k$). Condition \eqref{eq:invariance multilin}
is natural as in many cases considered (for instance in the entanglement
problem) the class of pure states $\mathcal{M}$ is invariant under
the action of some symmetry group and, consequently, the corresponding
notion of correlations is invariant under the action of this group
(see Section \ref{sec:Original-contributions-to} for a more thorough
discussion of this matter). 
\end{itemize}
The chapter is organized as follows. In Section \ref{sec:Bilinear-correlation-witnesses}
we present a construction of a bilinear correlation witness valid
for all classes of states which can be characterized by the operator
$A$ acting on two copies of the relevant Hilbert space $\mathcal{H}$
(number $k$ in Eq. \eqref{eq:characterisation multilin} is two).
In Section \ref{sec:Multilinear-correlation-witness} we generalize
this construction to the arbitrary natural number $k$. To be more
precise we construct, starting from the condition \eqref{eq:characterisation multilin},
the operator $V$, acting on $\mathcal{H}^{\otimes k}$, that satisfy
\eqref{eq:criteria mixed idea stonger} for a given class of pure
states $\mathcal{M}$. Section \ref{sec:Optimal--bilinear} contains
the complete description of the structure of bilinear correlation
witnesses that are invariant under the action of the relevant symmetry
group in cases when the class $\mathcal{M}$ consists of group-theoretic
coherent states and certain technical conditions are satisfied. In
this section we use extensively representation theory of semisimple
Lie algebras introduced in Subsection \ref{sub:Representation-theory-of}.
This methodology allows us to analyze analytically the strength of
criteria for classes of states presented in Section \ref{sec:semisimple-quadratic-characterisation}:
product states, product bosonic states, Slater determinants and fermionic
Gaussian states. We conclude the chapter in Section \ref{sub:Discussion-bilin-detec}
where we summarize the obtained results and state some open problems.

Throughout this chapter we will consequently use the notation introduced
in Chapter \ref{chap:Multilinear-criteria-for-pure-states}. The announced
criteria for detection of various correlations will be tested in practice
on particular examples and in some cases confronted with the existing
literature of the subject (see in particular Section \ref{sec:Optimal--bilinear}).
The main technical advantage of the criteria of the form \eqref{eq:criteria mixed idea stonger}
is that they can be easily used to study typical correlation properties
of states belonging to the set of density matrices $\mathcal{D}\left(\mathcal{H}\right)$
and on its submanifolds. The forthcoming Chapter \ref{chap:Typical-properties-of}
is devoted to study of this problem. The list of criteria of particular
type for detection of different kinds of correlations in mixed states
are given in Table \ref{tab:list of criteria}

\begin{table}[h]
\noindent \begin{centering}
{\footnotesize{}}%
\begin{tabular}{|c|c|c|c|}
\hline 
{\footnotesize{}Type of correlations } & {\footnotesize{}Bilinear witness} & {\footnotesize{}Multilinear witness} & {\footnotesize{}Optimal bilinear witness}\tabularnewline
\hline 
\hline 
{\footnotesize{}Entanglement } & Yes & Yes & Yes\tabularnewline
\hline 
{\footnotesize{}Entanglement of bosons} & Yes & Yes & Yes\tabularnewline
\hline 
{\footnotesize{}``Entanglement'' of fermions} & Yes & Yes & Yes\tabularnewline
\hline 
{\footnotesize{}Not convex-Gaussian} & Yes & Yes & Yes\tabularnewline
\hline 
{\footnotesize{}GME} & No & Yes & No\tabularnewline
\hline 
{\footnotesize{}Schmidt rank $\geq m$} & No & Yes  & No\tabularnewline
\hline 
\end{tabular}
\par\end{centering}{\footnotesize \par}

\noindent \protect\caption{\label{tab:list of criteria} Criteria detecting correlations in mixed
states presented in this chapter. Each row of the table correspond
to a different kind of correlation studied (the class of pure states
$\mathcal{M}$ giving rise to a particular type of correlations is
implicit). The last tree columns specify whether the correlation witness
of particular type is provided for the particular kind of correlations.}

\noindent 
\end{table}

Results presented in Section \ref{sec:Bilinear-correlation-witnesses}
constitute parts of \citep{Oszmaniec2014}. Results presented in Section
\ref{sec:Multilinear-correlation-witness} will contribute to the
forthcoming paper

\noindent \begin{center}
\begin{minipage}[t]{0.8\columnwidth}%
\begin{lyxlist}{00.00.0000}
\item [{\textit{``Multilinear}}] \textit{criteria for detection of generalized
entanglement''}, Michał Oszmaniec, et al. (in preparation),\end{lyxlist}
\end{minipage}
\par\end{center}

whereas results contained in Section \ref{sec:Optimal--bilinear}
will constitute a major part of another forthcoming publication

\noindent \begin{center}
\begin{minipage}[t]{0.8\columnwidth}%
\textit{``Criteria for quantum correlations based on second moments
of the density matrices''}, Michał Oszmaniec, et al. (in preparation).%
\end{minipage}
\par\end{center}

\section{Bilinear correlation witnesses based on quadratic characterization
of pure non-correlated states\label{sec:Bilinear-correlation-witnesses}}

In this section we derive a bilinear correlation witness for correlations
defined via the choice of the class of pure states $\mathcal{M}\subset\mathcal{D}_{1}\left(\mathcal{H}\right)$
specified by the quadratic (in the state's density matrix) condition
\begin{equation}
\kb{\psi}{\psi}\in\mathcal{M}\,\Longleftrightarrow\bra{\psi}\bra{\psi}A\ket{\psi}\ket{\psi}=0\,,\label{eq:definition class}
\end{equation}
where, just like in Section \ref{sec:Multilinear-characterization-of-pure},
$A$ is a Hermitian operator acting on $\mathrm{Sym}^{2}\mathcal{H}$
that satisfies $A\leq\mathbb{P}^{\mathrm{sym}}$, where $\mathbb{P}^{\mathrm{sym}}$
denotes the orthonormal projection onto $\mathrm{Sym}^{2}\mathcal{H}\subset\mathcal{H}\otimes\mathcal{H}$.
Our criterion for detection of correlations in the mixed states takes
a particularly simple form.
\begin{thm}
\label{thm:Main result bilin}Let the class of pure states $\mathcal{M}\subset\mathcal{D}_{1}\left(\mathcal{H}\right)$
be specified by Eq. \eqref{eq:definition class}. Then, the following
implication holds 
\begin{equation}
\mathrm{tr}\left(\rho_{1}\otimes\rho_{2}\, V\right)>0\,\Longrightarrow\,\rho_{1}\,\mbox{\text{and} }\rho_{2}\text{ are correlated}\,\left(\rho_{1},\rho_{2}\notin\mathcal{M}^{c}\right),\label{eq:criterion formula}
\end{equation}
where
\begin{equation}
V=A-\mathbb{P}^{\mathrm{asym}}\label{eq:form of V bilin}
\end{equation}
and $\mathbb{P}^{\mathrm{asym}}$ denotes the orthogonal projection
onto the two fold antisymmetrization, $\bigwedge^{2}\mathcal{H}$,
of $\mathcal{H}$.
\end{thm}
We would like to note that the above criterion is independent upon
the dimension of $\mathcal{H}$ and uses only the algebraic structure
of the set $\mathcal{M}$.  The proof of Theorem \ref{thm:Main result bilin}
follows from Lemmas \ref{lem:general bilin}, \ref{lem: particular bilin}
which we present below.
\begin{lem}
\label{lem:general bilin}\emph{Consider a class of pure states} $\mathcal{M}\subset\mathcal{D}_{1}\left(\mathcal{H}\right)$.
Assume that there exists an operator $V\in\mathrm{Herm\left(\mathcal{H}\otimes\mathcal{H}\right)}$
satisfying the property $\bra v\bra wV\ket v\ket w\le0$ for all $\kb vv\in\mathcal{M}$
and for arbitrary $\ket w\in\mathcal{H}$. Then, for any state $\rho\in\mathcal{M}^{C}$
and for arbitrary $B\ge0$ acting on $\mathcal{H}$, we have
\begin{equation}
\mathrm{tr}\left(\left[\rho\otimes B\right]V\right)\le0\,.\label{eq:basic ineuality}
\end{equation}
Moreover, if we assume $\mathbb{S}V\mathbb{S}=V$, for $\mathbb{S}$
being a swap operator%
\footnote{A swap operator is defined by its action on simple tensors $\mathbb{S}\ket{\psi}\ket{\phi}=\ket{\phi}\ket{\psi}$for
$\ket{\psi},\ket{\phi}\in\mathcal{H}$.%
} in $\mathcal{H}\otimes\mathcal{H}$, then except for \ref{eq:basic ineuality}
we have 
\begin{equation}
\mathrm{tr}\left(\left[B\otimes\rho\right]V\right)\le0\,.\label{eq:reverse order ineq}
\end{equation}
\end{lem}
\begin{proof}[Proof of Lemma \ref{lem:general bilin} ]
Since the expression \eqref{eq:basic ineuality} is linear in $B$
and every non-negative operator is of the form $B=\sum_{i}\kb{w_{i}}{w_{i}}$
for some (not necessary normalized) $\ket{w_{i}}\in\mathcal{H}$,
it is enough to consider $B=\kb ww$, where $\ket w\in\mathcal{H}$
is arbitrary. By definition the condition $\rho\in\mathcal{M}^{c}$
is equivalent to $\rho=\sum_{j}p_{j}\kb{v_{j}}{v_{j}}$ for $\kb{v_{i}}{v_{i}}\in\mathcal{M}$
and $p_{j}\ge0$, $\sum_{j}p_{j}=1$. Using that and the assumption
(\ref{eq:basic ineuality}) about the operator $V$ we get
\[
\mathrm{tr}\left(\left[\rho\otimes\kb ww\right]V\right)=\sum_{j}p_{j}\bra{v_{j}}\bra wV\ket w\ket{v_{j}}\le0\,.
\]
This concludes the proof of \eqref{eq:basic ineuality}. Inequality
\eqref{eq:reverse order ineq} follows immediately from the definition
of $\mathbb{S}$ and inequality \eqref{eq:basic ineuality}. 
\end{proof}
The permutation invariant operator $V$ (assuming that it exists)
from Lemma \ref{lem:general bilin} gives a criterion of the type
\eqref{eq:criterion formula}. Indeed, whenever we find a non-negative
operator $B$ for which $\mathrm{tr}\left(\left(\rho_{1}\otimes B\right)V\right)>0$
we know that $\rho_{1}\notin\mathcal{M}^{C}$. Note that without the
loss of generality we can take $B=\rho_{2}$. Now, using the permutation
symmetry of $V$ and \eqref{eq:reverse order ineq} we can exchange
roles of $\rho_{1}$ and $\rho_{2}$ which proves that
\[
\mathrm{tr}\left(\left[\rho_{1}\otimes\rho_{2}\right]V\right)>0\,\Longrightarrow\mathrm{{\color{green}{\color{black}both}}}\,\rho_{1}\mbox{ and }\rho_{2}\text{ are correlated}\,\left(\rho_{1},\rho_{2}\notin\mathcal{M}^{C}\right)\,.
\]
We remark that the condition $\bra v\bra wV\ket v\ket w\le0$ ($\kb vv\in\mathcal{M}$
and $\ket w\in\mathcal{H}$) from Lemma \ref{lem:general bilin} is
relatively easy check as it only only involves computing expectation
values on pure states. This observation will be useful in proofs of
Lemma \ref{lem: particular bilin} and Theorem \ref{finitelly generated cone}.

Lemma \ref{lem:general bilin} does not say anything about the existence
of the operator $V$ for a given class of pure states $\mathcal{M}$.
The following theorem ensures that as the operator one can take $V=A-\mathbb{P}^{\mathrm{asym}}$
wherever $\mathcal{M}$ is given by the condition \eqref{eq:definition class}.
\begin{lem}
\textup{\label{lem: particular bilin}}\textup{\emph{Consider the
class of pure states}}\textup{ $\mathcal{M}$ }\textup{\emph{given
by}}\textup{ }\eqref{eq:definition class}\textup{\emph{. The operator}}\textup{
$V=A-\mathbb{P}^{\mathrm{asym}}$ }\textup{\emph{satisfies}}\textup{
$\bra v\bra wV\ket v\ket w\le0$ for all $\ket v\in\mathcal{M}$ and
for arbitrary $\ket w\in\mathcal{H}$.}\end{lem}
\begin{proof}[Proof of Lemma \ref{lem: particular bilin} ]
Let $\ket v\in\mathcal{M}$ and let $\ket w\in\mathcal{H}$. Let
us write $\ket w=\ket{v_{||}}+\ket{v_{\perp}}$, where $\ket{v_{||}}\propto\ket v$
and $\ket{v_{\perp}}\perp\ket v$. We have the following equalities
\[
\bra v\bra wA\ket v\ket w=\bra v\bra{v_{\perp}}A\ket v\ket{v_{\perp}}\,,
\]
\[
\bra v\bra w\mathbb{P}^{\mathrm{asym}}\ket v\ket w=\bra v\bra{v_{\perp}}\mathbb{P}^{\mathrm{asym}}\ket v\ket{v_{\perp}}=\frac{1}{2}\bk{v_{\perp}}{v_{\perp}}\,.
\]
We have used the fact that operator $A$ is an orthonormal projector
and therefore we have $A\ket v\ket v=A\ket v\ket{v_{||}}=0$. Consequently,
we get the desired inequality
\begin{equation}
\bra v\bra wV\ket v\ket w=\bra v\bra{v_{\perp}}A\ket v\ket{v_{\perp}}-\frac{1}{2}\bk{v_{\perp}}{v_{\perp}}\le0\,,\label{eq:proof bilin crit step}
\end{equation}
where the estimate stems from the fact that $\bra v\bra{v_{\perp}}A\ket v\ket{v_{\perp}}\le\bra v\bra{v_{\perp}}\mathbb{P}^{\mathrm{sym}}\ket v\ket{v_{\perp}}=\frac{1}{2}\bk{v_{\perp}}{v_{\perp}}$,
where $\mathbb{P}^{\mathrm{sym}}$ is the projector onto $\mathrm{Sym}^{2}\left(\mathcal{H}\right)$.
\end{proof}
Analysis of proof of Lemma \ref{lem: particular bilin} leads to a
simple extensions of Theorem \ref{thm:Main result bilin} given by
the following corollary
\begin{cor}
\label{optimality simple bilin}Operator $V$ appearing in Theorem
\ref{thm:Main result bilin} can be taken to be 
\begin{equation}
V=A-c\cdot\mathbb{P}^{\mathrm{asym}}\,,\label{eq:refined criterion}
\end{equation}
where the constant $c$ is given by
\begin{equation}
c=2\cdot\left(\underset{\kb vv\in\mathcal{M}}{\mathrm{max}}\,\underset{\ket{v_{\perp}}\neq0}{\mathrm{max}}\frac{\bra v\bra{v_{\perp}}A\ket v\ket{v_{\perp}}}{\bk{v_{\perp}}{v_{\perp}}}\right)\,,\label{eq:maximum}
\end{equation}
and $\ket{v_{\perp}}$ denotes the arbitrary vector in $\mathcal{H}$
perpendicular to $\ket v$. The constant $c$ given by the above equation
is the smallest possible number $\beta$ such that operator $A-\beta\cdot\mathbb{P}^{\mathrm{asym}}$
satisfies the assumptions of Lemma \ref{lem:general bilin}.\end{cor}
\begin{proof}
Follows directly from Eq.\eqref{eq:proof bilin crit step}.\end{proof}
\begin{cor}
The above derived criteria for detection of correlations defined by
the choice of the class of pure states $\mathcal{M}$ are also valid
if the operator $A$ and the class $\mathcal{M}$ are related in the
following manner 
\begin{equation}
\kb{\psi}{\psi}\in\mathcal{M}\,\Longrightarrow\bra{\psi}\bra{\psi}A\ket{\psi}\ket{\psi}=0\,,\label{eq:modified condition}
\end{equation}
where operator $A$ satisfy $0\leq A\leq\mathbb{P}^{\mathrm{sym}}$.\end{cor}
\begin{proof}
The proof follows from Lemma \ref{lem:general bilin} and from the
straightforward modification of the proof of Lemma \ref{lem: particular bilin}
based on the inequality
\[
\bra v\bra wA\ket v\ket w\leq\bra v\bra{v_{\perp}}A\ket v\ket{v_{\perp}}\,.
\]

\end{proof}
Before moving to the concrete examples let us note that due to the
characterization \eqref{eq:definition class} the criterion based
on $V=A-c\mathbb{P}^{\mathrm{asym}}$ is exact for pure states in
the following sense
\begin{equation}
\kb{\psi}{\psi}\in\mathcal{M}\,\Longleftrightarrow\,\mathrm{tr}\left(\left[\kb{\psi}{\psi}^{\otimes2}\right]\, V\right)=0\,.\label{eq:exactnes bilin criterion}
\end{equation}
We will now derive the optimal constants $c$ (in a sense of Corollary
\ref{optimality simple bilin}) for four types of correlations defined
via classes of coherent states discussed in Section \ref{sec:semisimple-quadratic-characterisation}:
product states, product bosonic states, Slater determinants and pure
fermionic Gaussian states. We will also apply the obtained criteria
to witness correlations for some exemplary families of states.

\subsection{Entanglement of distinguishable particles}

We now apply Corollary \ref{optimality simple bilin} to the case
when $\mathcal{M}=\mathcal{M}_{d}\subset\mathcal{D}_{1}\left(\mathcal{H}_{d}\right)$
consists of pure product states of $L$ distinguishable particles.
In the following proposition we use the notation introduced in Subsection
\ref{sub:Product-states}.
\begin{prop}
\label{simple optimality distinguishable particles}Let
\[
A=\mathbb{P}^{\mathrm{sym}}-\mathbb{P}_{11'}^{+}\otimes\mathbb{P}_{22'}^{+}\otimes\ldots\otimes\mathbb{P}_{LL'}^{+}
\]
 and let $\mathcal{M}=\mathcal{M}_{d}$. The constant $c$ appearing
in Corollary \ref{optimality simple bilin} equals
\begin{equation}
c=1-2^{1-L}\,.\label{eq:constant distinguishable}
\end{equation}
Consequently, the operator 
\begin{equation}
V_{d}=\mathbb{P}^{\mathrm{sym}}-\mathbb{P}_{11'}^{+}\otimes\mathbb{P}_{22'}^{+}\otimes\ldots\otimes\mathbb{P}_{LL'}^{+}-\left(1-2^{1-L}\right)\mathbb{P}^{\mathrm{asym}}\label{eq:mintert buchleitner operator}
\end{equation}
is the largest operator having a structure $A-\beta\mathbb{P}^{\mathrm{asym}}$
that satisfies assumptions of Lemma \ref{lem:general bilin} and can
be used to detect entanglement.\end{prop}
\begin{proof}
The proof of Eq.\eqref{eq:constant distinguishable} is simple yet
technical. The detailed proof is given in Section \ref{sec:Proofs-concerning-Chapter multilinear witnesses}
of the Appendix (see page \pageref{sub:Proof-of-Proposition-optimaldist}).
\end{proof}
The criterion based on the operator \eqref{eq:mintert buchleitner operator}
is identical to the criterion based on the lower bound of the $L$-partite
concurrence (see Subsection \ref{sec:Methods-from-entanglement})
obtained in a paper by Aolita and Mintert \citep{Aolita2006}. Our
analysis shows that no stronger criterion based on the operator of
the structure $A-\beta\mathbb{P}^{\mathrm{asym}}$ is possible. The
general analysis of all $\mathrm{LU}$-invariant bilinear entanglement
witnesses will be given in Subsection \ref{sub:Entanglement-of-distinguishable-bilin}.

\subsection{Entanglement of bosons}

We apply Corollary \ref{optimality simple bilin} to the case when
$\mathcal{M}=\mathcal{M}_{b}\subset\mathcal{D}_{1}\left(\mathcal{H}_{b}\right)$
consists of pure product states of a system consisting of $L$ bosons.
In the following proposition we use the notation introduced in Subsection
\ref{sub:Symmetric-product-states}.
\begin{prop}
\label{simple optimality bos}Let
\[
A=\mathbb{P}^{\mathrm{sym}}-\mathbb{P}_{11'}^{+}\otimes\mathbb{P}_{22'}^{+}\otimes\ldots\otimes\mathbb{P}_{LL'}^{+}\left(\mathbb{P}_{\left\{ 1,\ldots,L\right\} }^{\mathrm{sym}}\otimes\mathbb{P}_{\left\{ 1',\ldots,L'\right\} }^{\mathrm{sym}}\right)\,
\]
and let $\mathcal{M}=\mathcal{M}_{b}$. The constant $c$ appearing
in Corollary \ref{optimality simple bilin} for the equals
\begin{equation}
c=1-2^{1-L}\,.\label{eq:constant bos}
\end{equation}
Consequently, the operator 
\begin{equation}
V_{b}=\mathbb{P}^{\mathrm{sym}}-\mathbb{P}_{11'}^{+}\otimes\mathbb{P}_{22'}^{+}\otimes\ldots\otimes\mathbb{P}_{LL'}^{+}\left(\mathbb{P}_{\left\{ 1,\ldots,L\right\} }^{\mathrm{sym}}\otimes\mathbb{P}_{\left\{ 1',\ldots,L'\right\} }^{\mathrm{sym}}\right)-\left(1-2^{1-L}\right)\mathbb{P}^{\mathrm{asym}}\label{eq:mintert buchleitner operator bos}
\end{equation}
is the largest operator having a structure $A-\beta\mathbb{P}^{\mathrm{asym}}$
that satisfies assumptions of Lemma \ref{lem:general bilin} and can
be used to detect entanglement.\end{prop}
\begin{proof}
The proof is completely analogous to the proof of Proposition \ref{simple optimality distinguishable particles}
and follows from the fact that
\[
V_{b}=\left(\mathbb{P}_{\left\{ 1,\ldots,L\right\} }^{\mathrm{sym}}\otimes\mathbb{P}_{\left\{ 1',\ldots,L'\right\} }^{\mathrm{sym}}\right)V_{d}\left(\mathbb{P}_{\left\{ 1,\ldots,L\right\} }^{\mathrm{sym}}\otimes\mathbb{P}_{\left\{ 1',\ldots,L'\right\} }^{\mathrm{sym}}\right)
\]
 in the case when all single-particle Hilbert spaces in the tensor
product $\mathcal{H}_{d}=\otimes_{i=1}^{L}\mathcal{H}_{i}$ are identical.
\end{proof}

\subsection{``Entanglement'' of fermions }

We will now apply Corollary \ref{optimality simple bilin} to the
case when $\mathcal{M}=\mathcal{M}_{f}\subset\mathcal{D}_{1}\left(\mathcal{H}_{f}\right)$
consists of projectors onto Slater determinants in the system of $L$
fermions. Recall that the correlations induced by such choice of the
class $\mathcal{M}$ are correlations in states that do not follow
merely form the antisymmetrizations of wave functions in $\mathcal{H}_{f}$
(c.f. Section \ref{sec:General-motivation}). In what follows we will
use the  notation introduced in Subsection \ref{sub:Slater-determinants}. 
\begin{prop}
\label{simple optimal ferm}Let
\[
A=A_{f}=\mathbb{P}^{\mathrm{sym}}-\frac{2^{L}}{L+1}\mathbb{P}_{11'}^{+}\otimes\mathbb{P}_{22'}^{+}\otimes\ldots\otimes\mathbb{P}_{LL'}^{+}\left(\mathbb{P}_{\left\{ 1,\ldots,L\right\} }^{\mathrm{asym}}\otimes\mathbb{P}_{\left\{ 1',\ldots,L'\right\} }^{\mathrm{asym}}\right)\,
\]
and let $\mathcal{M}=\mathcal{M}_{f}$. The constant $c$ appearing
in Proposition \ref{optimality simple bilin} equals
\begin{equation}
c=1-\frac{2}{L+1-\max\left\{ 0,2L-d\right\} }\,.\label{eq:constant ferm}
\end{equation}
Consequently, the operator 
\begin{equation}
V_{f}=\mathbb{P}^{\mathrm{sym}}-\frac{2^{L}}{L+1}\mathbb{P}_{11'}^{+}\otimes\mathbb{P}_{22'}^{+}\otimes\ldots\otimes\mathbb{P}_{LL'}^{+}\left(\mathbb{P}_{\left\{ 1,\ldots,L\right\} }^{\mathrm{asym}}\otimes\mathbb{P}_{\left\{ 1',\ldots,L'\right\} }^{\mathrm{asym}}\right)-c\mathbb{P^{\mathrm{asym}}}\label{eq:mintert buchleitner operator ferm}
\end{equation}
is the largest operator having a structure $A_{f}-\beta\mathbb{P}^{\mathrm{asym}}$
that satisfies assumptions of Lemma \ref{lem:general bilin} and can
be used to detect whether a given state $\rho$ is correlated ($\rho\notin\mathcal{M}_{f}^{c}$). \end{prop}
\begin{proof}
The proof is analogous to the proof of Proposition \ref{simple optimality distinguishable particles}
and is given in Section \ref{sec:Proofs-concerning-Chapter multilinear witnesses}
of the Appendix (see page \pageref{sub:Proof-of-Proposition-optimal ferm}).
Let us explain here the reason for the occurrence of the number $\max\left\{ 0,2L-d\right\} $
in \eqref{eq:constant ferm}. Let $\ket{\psi_{0}},\ket{\psi_{1}}\in\mathcal{H}_{f}$
be two Slater determinants constructed from single-particle orbitals
$\left\{ \ket{\phi_{i}}\right\} _{i=1}^{d}$, forming a basis of a
single particle Hilbert space $\mathcal{H}\approx\mathbb{C}^{d}$.
The number $\max\left\{ 0,2L-d\right\} $ is the minimal number of
times that the same orbitals can occur in both $\ket{\psi_{0}}$and
$\ket{\psi_{1}}$. \end{proof}
\begin{example}
We now apply the criterion stated above to study correlations of a
depolarization of an arbitrary pure state of two fermions. Since the
problem for $d=4$ can be solved exactly (see Section \eqref{sec:Ulhmann-Wotters})
we consider situation when $d>4$. Any normalized vector $\ket{\psi}\in\bigwedge^{2}\left(\mathbb{C}^{d}\right)$
can be written \citep{EckertFermions2002} as 
\begin{equation}
\ket{\psi}=\sum_{i=1}^{i=\left\lfloor \frac{d}{2}\right\rfloor }\lambda_{i}\ket{\phi_{2i-1}}\wedge\ket{\phi_{2i}},\,\label{eq:arbitrary 2 ferm}
\end{equation}
where $\lambda_{i}\ge0$, $\sum_{i=1}^{i=\left\lfloor \frac{d}{2}\right\rfloor }\lambda_{i}^{2}=1$
and vectors $\ket{\phi_{i}}$($i=1,\ldots,d$) form the orthonormal
basis of $\mathbb{C}^{d}$. The arbitrary depolarization of state
$\kb{\psi}{\psi}$ is given by
\begin{equation}
\rho_{\kb{\psi}{\psi}}\left(p\right)=\left(1-p\right)\kb{\psi}{\psi}+p\frac{2\mathbb{I}}{d\left(d-1\right)}\,,\label{eq:slater family}
\end{equation}
Direct usage of the criterion \eqref{eq:criterion formula} based
on the operator \eqref{eq:mintert buchleitner operator ferm} for
$\rho_{1}=\rho_{\kb{\psi}{\psi}}\left(p\right)$ and $\rho_{2}=\kb{\psi}{\psi}$
gives that $\rho_{\kb{\psi}{\psi}}\left(p\right)$ is correlated for
\begin{equation}
p<p_{cr}=\frac{\left(1-\sum_{i=1}^{i=\left\lfloor \frac{d}{2}\right\rfloor }\lambda_{i}^{4}\right)}{\left(1-\sum_{i=1}^{i=\left\lfloor \frac{d}{2}\right\rfloor }\lambda_{i}^{4}\right)+f\left(d\right)}\,,\label{eq:critical p ferm}
\end{equation}
where $f\left(d\right)=4\frac{d-2}{d\left(d-1\right)}$. The proof
of \eqref{eq:critical p ferm} follows from identities:
\[
\bra{\psi}\bra{\psi}A_{f}\ket{\psi}\ket{\psi}=\frac{1}{3}\left(1-\sum_{i=1}^{i=\left\lfloor \frac{d}{2}\right\rfloor }\lambda_{i}^{4}\right)\,,
\]
\[
\mathrm{tr}\left(\left[\mathbb{I}\otimes\kb{\psi}{\psi}\right]A_{f}\right)=\frac{\left(d-2\right)\left(d-3\right)}{6\cdot2}\,,
\]
which are the consequence of Eq.\eqref{eq:expectation ferm}.
\end{example}

\subsection{Not convex-Gaussian correlation \label{Not-convex-Gaussian-correlation opt}}

In the situation of the class of pure fermionic Gaussian states (see
Subsection \ref{sub:Fermionic-Gaussian-states}) we have two situations
to consider. We can either take $\mathcal{M}=\mathcal{M}_{g}\subset\mathcal{D}_{1}\left(\mbox{\ensuremath{\mathcal{H}}}_{\mathrm{Fock}}\left(\mathbb{C}^{d}\right)\right)$
or impose parity superselection rule and take $\mathcal{M}=\mathcal{M}_{g}^{+}\subset\mathcal{D}_{1}\left(\mbox{\ensuremath{\mathcal{H}}}_{\mathrm{Fock}}^{+}\left(\mathbb{C}^{d}\right)\right)$
(the treatment of $\mathcal{M}_{g}^{-}\subset\mathcal{D}_{1}\left(\mbox{\ensuremath{\mathcal{H}}}_{\mathrm{Fock}}^{-}\left(\mathbb{C}^{d}\right)\right)$)
is analogous). In the discussion that follows we will use the  notation
introduced in Subsection \ref{sub:Fermionic-Gaussian-states}.

\paragraph*{No parity superselection rule}

Recall that in the case $\mathcal{M}=\mathcal{M}_{g}\subset\mathcal{D}_{1}\left(\mbox{\ensuremath{\mathcal{H}}}_{\mathrm{Fock}}\left(\mathbb{C}^{d}\right)\right)$
we have
\[
A=A_{g}=\mathbb{P}^{\mathrm{sym}}-\mathbb{P}_{0}\,,
\]
where operator $\mbox{\ensuremath{\mathbb{P}}}_{0}:\mathrm{Sym}^{2}\left(\mbox{\ensuremath{\mathcal{H}}}_{\mathrm{Fock}}\left(\mathbb{C}^{d}\right)\right)\rightarrow\mathrm{Sym}^{2}\left(\mbox{\ensuremath{\mathcal{H}}}_{\mathrm{Fock}}\left(\mathbb{C}^{d}\right)\right)$
is given by Eq.\eqref{eq:closed form expression ferm}). Application
of Corollary \ref{optimality simple bilin} to the case when $\mathcal{M}=\mathcal{M}_{g}\subset\mathcal{D}_{1}\left(\mathrm{Fock}\left(\mathbb{C}^{d}\right)\right)$
gives $c=1$. This follows from the fact that the operator has a support
solely on the subspace 
\[
\left(\mbox{\ensuremath{\mathcal{H}}}_{\mathrm{Fock}}^{+}\left(\mathbb{C}^{d}\right)\otimes\mbox{\ensuremath{\mathcal{H}}}_{\mathrm{Fock}}^{+}\left(\mathbb{C}^{d}\right)\right)\oplus\left(\mbox{\ensuremath{\mathcal{H}}}_{\mathrm{Fock}}^{-}\left(\mathbb{C}^{d}\right)\otimes\mbox{\ensuremath{\mathcal{H}}}_{\mathrm{Fock}}^{-}\left(\mathbb{C}^{d}\right)\right)\,,
\]
of the total tensor product $\mbox{\ensuremath{\mathcal{H}}}_{\mathrm{Fock}}\left(\mathbb{C}^{d}\right)\otimes\mbox{\ensuremath{\mathcal{H}}}_{\mathrm{Fock}}\left(\mathbb{C}^{d}\right)$
and $\bra 0\bra{\psi}\mathbb{P}_{0}\ket 0\ket{\psi}=0$ for $\ket{\psi}\in\mbox{\ensuremath{\mathcal{H}}}_{\mathrm{Fock}}^{-}\left(\mathbb{C}^{d}\right)$.

\paragraph*{Parity superselection rule imposed}

In the case $\mathcal{M}=\mathcal{M}_{g}^{+}\subset\mathcal{D}_{1}\left(\mbox{\ensuremath{\mathcal{H}}}_{\mathrm{Fock}}^{+}\left(\mathbb{C}^{d}\right)\right)$
we get a nontrivial value of the constant $c$, as described by the
following proposition.

The general discussion of bilinear non-Gaussianity witnesses invariant
under the action of Bogolyubov transformations will be given in Subsection
\ref{sub:Not-convex-Gaussian-correlation-bilin}.
\begin{prop}
\label{simple optiml gauss}Let $A:\mathrm{Sym}^{2}\left(\mbox{\ensuremath{\mathcal{H}}}_{\mathrm{Fock}}^{+}\left(\mathbb{C}^{d}\right)\right)\rightarrow\mathrm{Sym}^{2}\left(\mbox{\ensuremath{\mathcal{H}}}_{\mathrm{Fock}}^{+}\left(\mathbb{C}^{d}\right)\right)$
be given by
\begin{equation}
A=A_{g}=\mathbb{P}^{\mathrm{sym}}-\mathbb{P}_{+}\mathbb{P}_{0}\mathbb{P}_{+}\,,\label{eq:appropiate A}
\end{equation}
where
\[
\mathbb{P}_{+}=\frac{1}{4}\left(\mathbb{I}+Q\right)\otimes\left(\mathbb{I}+Q\right)\,.
\]
Let $\mathcal{M}=\mathcal{M}_{g}^{+}\subset\mathcal{D}_{1}\left(\mbox{\ensuremath{\mathcal{H}}}_{\mathrm{Fock}}^{+}\left(\mathbb{C}^{d}\right)\right)$.
Assume that $d\leq1000$. The constant $c$ appearing in Corollary
\ref{optimality simple bilin} equals
\begin{equation}
c_{d}=1-a_{d}\,,\label{eq:constant gauss}
\end{equation}
where
\[
a_{d}=\frac{1}{2^{2d}}\binom{2d}{d}\sum_{k=0}^{d}\sum_{m=0}^{\mathrm{min\left\{ k,2\left\lfloor \frac{d}{2}\right\rfloor \right\} }}\left(-2\right)^{m}\frac{\binom{d}{k}}{\binom{2d}{2k}}\binom{d-m}{k-m}\binom{\left\lfloor \frac{d}{2}\right\rfloor }{m}\,.
\]
Consequently, the operator 
\begin{equation}
V_{g}=\mathbb{P}^{\mathrm{sym}}-\mathbb{P}_{+}\mathbb{P}_{0}\mathbb{P}_{+}-c_{d}\mathbb{P^{\mathrm{asym}}}\label{eq:mintert buchleitner gauss}
\end{equation}
is the largest operator having a structure $A-\beta\mathbb{P}^{\mathrm{asym}}$
that satisfies assumptions of Lemma \ref{lem:general bilin} and can
be used, via Lemma to detect whether a given state $\rho\in\mathcal{D}\left(\mbox{\ensuremath{\mathcal{H}}}_{\mathrm{Fock}}^{+}\left(\mathbb{C}^{d}\right)\right)$
is not convex-Gaussian ($\rho\notin\left(\mathcal{M}_{g}^{+}\right)^{c}$). \end{prop}
\begin{proof}
The reasoning is completely analogous to that from the proof of Proposition
\ref{simple optimality distinguishable particles}. The sketch of
the proof is given in Section \ref{sec:Proofs-concerning-Chapter multilinear witnesses}
of the Appendix (see page  \pageref{sub:Proof-of-Proposition-optimapl gauss}). \end{proof}
\begin{rem}
The limitation of the number of fermionic modes $d\leq1000$ in Proposition
\ref{simple optiml gauss} stems solely from  the combinatorial complexity
of the formulas involved in the proof of \ref{simple optiml gauss}.
We proved Eq.\eqref{eq:constant gauss} by a direct computation with
Mathematica up to $d=1000$. We conjecture that \eqref{eq:constant gauss}
holds for arbitrary positive integer $d$.
\end{rem}

\section{Multilinear correlation witness\label{sec:Multilinear-correlation-witness}}

In this section we derive a k-linear correlation witness for correlations
defined via the choice of the class of pure ``non-correlated'' states
$\mathcal{M}\subset\mathcal{D}_{1}\left(\mathcal{H}\right)$ specified
by the k-linear (in the state's density matrix) condition 
\begin{equation}
\kb{\psi}{\psi}\in\mathcal{M}\,\Longleftrightarrow\bra{\psi{}^{\otimes k}}A\ket{\psi{}^{\otimes k}}=0\ ,\label{eq:k-lin class}
\end{equation}
where $A$ is the Hermitian operator acting on $\mathrm{Sym}^{k}\mathcal{H}$
($k$-fold symmetric power of $\mathcal{H}$) satisfying $A\le\mathbb{P}^{\mathrm{sym},k}$
(operator $\mathbb{P}^{\mathrm{sym,k}}$ is the orthogonal projector
onto $\mathrm{Sym}^{k}\mathcal{H}$). For exemplary types of correlations
that can be defined by such a choice of ``non-correlated'' pure
states see Section \ref{sec:Multilinear-characterization-of-pure}.
After presenting the main result, the k-linear generalization of theorem
\ref{thm:Main result bilin}, we will apply the obtained criteria
to detect genuine multiparty entanglement (see Subsection \ref{sub:GME})
and Schmidt rank (see Subsection \ref{sub: Schmid rank}) for exemplary
families of mixed states. 
\begin{thm}
\label{main result multilinear witness}Consider the class of ``non-correlated''
pure states $\mathcal{M}$ specified by 
\begin{equation}
\kb{\psi}{\psi}\in\mathcal{M}\,\Longleftrightarrow\bra{\psi{}^{\otimes k}}A\ket{\psi{}^{\otimes k}}=0\ .\label{eq:k linear}
\end{equation}
Let
\begin{equation}
V=A-\left(k-1\right)\left(\mathbb{I}^{\otimes k}-\mathbb{P}^{\mathrm{sym,k}}\right)\,.\label{eq:operator V}
\end{equation}
The following implication holds
\begin{equation}
\mathrm{tr}\left(\left[\rho_{1}\otimes\rho_{2}\otimes\ldots\otimes\rho_{k}\right]V\right)>0\ \Longrightarrow\ \rho_{1},\ldots,\rho_{k}\notin\mathcal{M}^{c}\,\text{(\ensuremath{\rho_{1}},\ensuremath{\ldots},\ensuremath{\rho_{k}} are correlated)}\,.\label{eq:multilinear criterion}
\end{equation}
\end{thm}
\begin{proof}
The straightforward generalization of Lemma \ref{lem:general bilin}
shows that for the operator $V\in\mathrm{Herm\left(\mathcal{H}^{\otimes k}\right)}$,
which is permutation-symmetric%
\footnote{The operator $V\in\mathrm{Herm\left(\mathcal{H}^{\otimes k}\right)}$
is permutation invariant if and only if it satisfies $\sigma V\sigma^{\dagger}=V$,
where $\sigma$ is an arbitrary permutation operator flipping factors
of the tensor product $\mathcal{H}^{\otimes k}$ (See Section \ref{sub:Representation-theory-of}
for the description of the action of the permutation group $S_{k}$
on $\mathcal{H}^{\otimes k}$).%
}, and arbitrary class of pure states $\mathcal{M}\subset\mathcal{D}_{1}\left(\mathcal{H}\right)$,
the implication \eqref{eq:multilinear criterion} is equivalent to
the following condition 
\begin{equation}
\bra{\psi}\bra{\phi_{1}}\ldots\bra{\phi_{k-1}}V\ket{\psi}\ket{\phi_{1}}\ldots\ket{\phi_{k-1}}\leq0\ \,,\label{eq:interidiate multilin}
\end{equation}
where $\ket{\psi}\in\mathcal{M}$ and $\ket{\phi_{i}}\in\mathcal{H}$
($i=1,\ldots k-1$) are arbitrary vectors. In what follows we will
check that the operator
\begin{equation}
V=A-\left(k-1\right)\left(\mathbb{I}^{\otimes k}-\mathbb{P}^{\mathrm{sym,k}}\right)\label{eq:V alpha}
\end{equation}
satisfies \eqref{eq:interidiate multilin}. We proceed analogously
to the case of the proof of Lemma \eqref{lem: particular bilin}.
Let us fix $\ket{\psi}\in\mathcal{M}$ and $\ket{\phi_{i}}\in\mathcal{H}$
($i=1,\ldots k-1$) appearing \eqref{eq:interidiate multilin}. For
each $i=1,\ldots,k-1$ we have a decomposition
\begin{equation}
\ket{\phi_{i}}=\ket{\phi_{i}^{||}}+\ket{\phi_{i}^{\perp}}\,,\label{eq:expantion single}
\end{equation}
where $\ket{\phi_{i}^{||}}\propto\ket{\psi}$ and $\ket{\phi_{i}^{\perp}}\perp\ket{\psi}$.
Note that the vectors appearing in \eqref{eq:expantion single} are
in general not normalized. We introduce the following notation
\begin{equation}
\ket{\Phi^{X}}=\bigotimes_{i=1}^{k-1}\ket{\phi_{i}^{X(i)}}\,,\label{eq:convention expantion}
\end{equation}
where $X\subset\left\{ 1,\ldots,k-1\right\} $ and
\[
\ket{\phi_{i}^{X(i)}}=\begin{cases}
\ket{\phi_{i}^{||}} & \text{for }i\in X\,,\\
\ket{\phi_{i}^{\perp}} & \text{for }i\notin X\,.
\end{cases}
\]
Using \eqref{eq:expantion single} we get 
\[
\bra{\psi}\bra{\phi_{1}}\ldots\bra{\phi_{k-1}}A\ket{\psi}\ket{\phi_{1}}\ldots\ket{\phi_{k-1}}=\sum_{X,Y\subset\left\{ 1,\ldots,k-1\right\} }\bra{\psi}\bra{\Phi^{X}}A\ket{\psi}\ket{\Phi^{Y}}\,.
\]
Using definitions of the class $\mathcal{M}$ (see Eq.\eqref{eq:k linear})
and vectors $\ket{\phi^{X}}$ we get
\begin{equation}
\bra{\psi}\bra{\phi_{1}}\ldots\bra{\phi_{k-1}}A\ket{\psi}\ket{\phi_{1}}\ldots\ket{\phi_{k-1}}=\sum_{l,m=0}^{k-2}\bra{\psi}\bra{\Phi^{l}}A\ket{\psi}\ket{\Phi^{m}}\,,\label{eq:rewriting}
\end{equation}
where 
\[
\ket{\Phi^{l}}=\sum_{\begin{array}[t]{c}
X\subset\left\{ 1,\ldots,k-1\right\} \\
\left|X\right|=l
\end{array}}\ket{\Phi^{X}}\,.
\]
We have the following chain of estimates
\begin{align}
\sum_{l,m=0}^{k-2}\bra{\psi}\bra{\Phi^{l}}A\ket{\psi}\ket{\Phi^{m}} & \leq\sum_{l,m=0}^{k-2}\bra{\psi}\bra{\Phi^{l}}\mathbb{P}^{k,\mathrm{sym}}\ket{\psi}\ket{\Phi^{m}}\,,\\
 & =\sum_{l=0}^{k-2}\bra{\psi}\bra{\Phi^{l}}\mathbb{P}^{k,\mathrm{sym}}\ket{\psi}\ket{\Phi^{l}}\label{eq:equality 1}\\
 & =\sum_{l=0}^{k-2}\frac{l+1}{k}\bra{\psi}\bra{\Phi^{l}}\mathbb{I}\otimes\mathbb{P}^{k-1,\mathrm{sym}}\ket{\psi}\ket{\Phi^{l}}\,.\label{eq:equality 2}\\
 & \leq\frac{k-1}{k}\sum_{l=0}^{k-2}\bk{\psi}{\psi}\bk{\Phi^{l}}{\Phi^{l}}\label{eq:estym another}
\end{align}
The equality \eqref{eq:equality 1} follows from the fact that $\mathbb{P}^{k,\mathrm{sym}}\ket{\psi}\ket{\Phi^{m}}\perp\ket{\psi}\ket{\Phi^{l}}$
for $l\neq m$. The equality \eqref{eq:equality 2} follows from the
fact that $\bk{\psi}{\phi_{i}^{X(i)}}=0$ for $i\notin X$. Combining
\eqref{eq:estym another} with \eqref{eq:rewriting} we get 
\begin{equation}
\bra{\psi}\bra{\phi_{1}}\ldots\bra{\phi_{k-1}}A\ket{\psi}\ket{\phi_{1}}\ldots\ket{\phi_{k-1}}\leq\frac{k-1}{k}\sum_{l=0}^{k-2}\bk{\psi}{\psi}\bk{\Phi^{l}}{\Phi^{l}}.\label{eq:A multilin estimate}
\end{equation}
On the other hand we have 
\[
\bra{\psi}\bra{\phi_{1}}\ldots\bra{\phi_{k-1}}\left(\mathbb{I}^{\otimes k}-\mathbb{P}^{\mathrm{sym,k}}\right)\ket{\phi_{1}}\ldots\ket{\phi_{k-1}}=\sum_{l=0}^{k-2}\bra{\psi}\bra{\Phi^{l}}\left(\mathbb{I}^{\otimes k}-\mathbb{P}^{\mathrm{sym,k}}\right)\ket{\psi}\ket{\Phi^{l}}\,.
\]
Using the analogous estimates to \eqref{eq:estym another} for $\bra{\psi}\bra{\Phi^{l}}\mathbb{P}^{\mathrm{sym,k}}\ket{\psi}\ket{\Phi^{l}}$
shows that the operator $V$, defined by \eqref{eq:V alpha}, satisfies
\eqref{eq:interidiate multilin}.
\end{proof}
Before moving to the examples of application of the criterion \eqref{eq:multilinear criterion}
let us note that due to the characterization \eqref{eq:k-lin class}
the criterion \eqref{eq:multilinear criterion} is exact for pure
states in the following sense
\begin{equation}
\kb{\psi}{\psi}\in\mathcal{M}\,\Longleftrightarrow\,\mathrm{tr}\left(\left[\kb{\psi}{\psi}^{\otimes k}\right]\, V\right)=0\,.\label{eq:exactnes bilin criterion-1}
\end{equation}

\subsection{Witnessing Genuine Multiparty Entanglement\label{sub:GME multiparty witness}}

Let us now use Theorem \ref{thm:Main result bilin} to witness genuine
multiparty entanglement in tripartite system. In this scenario the
set of non-correlated pure states consists of pure 2-separable states%
\footnote{For simplicity we assume that dimensions of single particle Hilbert
spaces are identical.%
} $\mathcal{M}_{d}^{2}\subset\mathcal{D}_{1}\left(\mathbb{C}^{d}\otimes\mathbb{C}^{d}\otimes\mathbb{C}^{d}\right)$
(see Subsection \ref{sub:GME}). 
\begin{defn}
A tripartite mixed state $\rho\in\mathcal{D}\left(\mathbb{C}^{d}\otimes\mathbb{C}^{d}\otimes\mathbb{C}^{d}\right)$
exhibit genuine multiparty entanglement if and only of $\rho\notin\left(\mathcal{M}_{d}^{2}\right)^{c}$.
\end{defn}
We now apply the criterion \eqref{eq:multilinear criterion} and the
polynomial characterization of $\mathcal{M}_{d}^{2}$ (given in Lemma
\ref{GME pure}) to obtain a criterion for genuine multiparty entanglement.
\begin{cor}
\label{GME witness} Let $\rho\in\mathcal{D}\left(\mathbb{C}^{d}\otimes\mathbb{C}^{d}\otimes\mathbb{C}^{d}\right)$.
Let
\begin{equation}
V=A-5\left(\mathbb{I}^{\otimes6}-\mathbb{P}^{\mathrm{sym,6}}\right)\,,\label{eq:operator V gme}
\end{equation}
where%
\footnote{For the explanation of the notation used in Eq.\eqref{eq:A gme again}
see Subsection \ref{sub:GME}.%
} 
\begin{equation}
\mbox{\ensuremath{\mathbb{P}}}^{\mathrm{sym,6}}\circ\left(A^{1:23}\otimes A^{2:13}\otimes A^{3:12}\right)\circ\mbox{\ensuremath{\mathbb{P}}}^{\mathrm{sym,6}}\,,\label{eq:A gme again}
\end{equation}
The following implication holds
\begin{equation}
\mathrm{tr}\left(\left[\rho_{1}\otimes\rho_{2}\otimes\ldots\otimes\rho_{6}\right]V\right)>0\ \Longrightarrow\ \rho_{1},\ldots,\rho_{6}\,\text{are genuinly multiparty entangled}.\label{eq:gme criterion}
\end{equation}

\end{cor}

\subsection{Witnessing states with Schmidt number greater than $n$ \label{sub:multilin states with schmidt rank}}

We now apply Theorem \ref{thm:Main result bilin} to witness correlations
based on the notion of Schmidt rank of bipartite state. Recall (see
Subsection \ref{sub: Schmid rank}) that the class $\mathcal{M}_{n}\subset\mathcal{D}_{1}\left(\mathcal{H}_{A}\otimes\mathcal{H}_{B}\right)$
was defined as the set consisting of pure states with Schmidt rank
at most $n$ ($n\leq\min\left\{ \mathrm{dim}\left(\mathcal{H}_{A}\right),\mathrm{dim}\left(\mathcal{H}_{B}\right)\right\} $).
Let us first define, after \citep{SchmidtNumHoro}, the notion of
the Schmidt number of a bipartite states.
\begin{defn}
\label{def:schmidt number of mixed}A bipartite density matrix $\rho\in\mathcal{D}\left(\mathcal{H}_{A}\otimes\mathcal{H}_{B}\right)$
has Schmidt number $k$ ($\mathrm{Sr}\left(\rho\right)=k$) if and
only if $\rho\in\mathcal{M}_{n}^{c}$ and $\rho\notin\mathcal{M}_{n-1}^{c}$. 
\end{defn}
Straightforward application of criterion \eqref{eq:multilinear criterion}
and Lemma \ref{characterization schmidt rank} allow us to witness
states with Schmidt number greater than $n$.
\begin{cor}
\label{schmid number witness} Let $\rho\in\mathcal{D}\left(\mathcal{H}_{A}\otimes\mathcal{H}_{B}\right)$.
Let
\begin{equation}
V_{n}=A_{n}-\left(n\right)\left(\mathbb{I}^{\otimes n+1}-\mathbb{P}^{\mathrm{sym,n+1}}\right)\,,\label{eq:operator V schmidt}
\end{equation}
where%
\footnote{For the explanation of the notation used in Eq.\eqref{eq:A schmidt again}
see Subsection \ref{sub: Schmid rank}.%
} 
\begin{equation}
A_{n}=\mathbb{P}_{A}^{\mathrm{asym},n+1}\otimes\mathbb{P}_{B}^{\mathrm{asym},n+1}\,,\label{eq:A schmidt again}
\end{equation}
The following implication holds
\begin{equation}
\mathrm{tr}\left(\left[\rho_{1}\otimes\rho_{2}\otimes\ldots\otimes\rho_{n+1}\right]V_{n}\right)>0\ \Longrightarrow\ \rho_{1},\ldots,\rho_{n+1}\text{ have Schmidt number higher than \ensuremath{n}}\,.\label{eq:schmidt selection}
\end{equation}
\end{cor}
\begin{example}
We now apply the criterion stated above to study Schmidt number $\rho\left(p\right)\in\mathcal{D}\left(\mathbb{C}^{d}\otimes\mathbb{C}^{d}\right)$
of a depolarized maximally entangled state of two qdits,
\[
\rho\left(p\right)=\left(1-p\right)\kb{\Psi}{\Psi}+p\frac{\mathbb{I}}{d^{2}}\,,
\]
where $\ket{\Psi}=\frac{1}{\sqrt{d}}\sum_{i=1}^{d}\ket i\ket i$ ($\left\{ \ket i\right\} _{i=1}^{i=d}$
is a standard orthonormal basis of $\mathbb{C}^{d}$). The usage of
criterion \eqref{eq:schmidt selection} for $\rho_{1}=\rho\left(p\right)$
and $\rho_{i}=\kb{\Psi}{\Psi}$ for $i>1$ gives the following information
about Schmidt number of the state $\rho\left(p\right)$.
\begin{equation}
p<p_{cr,n}=\frac{\mathcal{A}_{n}}{\left(n\cdot\left(1-\mathcal{C}_{n}\right)-\mathcal{A}_{n}-\mathcal{B}_{n}-\right)}\,\Longrightarrow\,\rho\left(p\right)\,\text{has Schmidt number greater than }n\,,\label{eq:schmidt number num}
\end{equation}
where
\begin{align}
\mathcal{A}_{n} & =\mathrm{tr}\left(\left[\kb{\psi}{\psi}^{\otimes n+1}\right]A_{n}\right)=\frac{\binom{d}{n+1}}{\left[\left(n+1\right)!\right]d^{n+1}}\,,\label{eq:aleph schmidt}\\
\mathcal{B}_{n} & =\mathrm{tr}\left(\left[\frac{1}{d^{2}}\mathbb{I}\otimes\left\{ \kb{\psi}{\psi}^{\otimes n}\right\} \right]A_{n}\right)=\frac{d-n}{d^{n-1}}\binom{d-1}{n}\,,\label{eq:bethschmidt}\\
\mathcal{C}_{n} & =\mathrm{tr}\left(\left[\frac{1}{d^{2}}\mathbb{I}\otimes\left\{ \kb{\psi}{\psi}^{\otimes n}\right\} \right]\mathbb{P}^{\mathrm{sym,n+1}}\right)\frac{1}{d^{2}}\left(1+\frac{d^{2}-1}{n+1}\right)\,.\label{eq:sym schmidt}
\end{align}
The only difficult part is proof of \eqref{eq:bethschmidt} (Eq.\eqref{eq:aleph schmidt}
follows directly from from Eq.\eqref{eq:explicit form Schmidt} whereas
the proof of Eq.\eqref{eq:sym schmidt} is straightforward). This
formula can be proven by performing the computation analogous to the
one given in the proof of Proposition \ref{schmidt number explicit}
(see page \pageref{sub:Proof-of-Proposition schmidt number}).
\end{example}

\section{Optimal  group-invariant bilinear correlation witness\label{sec:Optimal--bilinear}}
\begin{problem}
What is the structure of the bilinear, group invariant, criteria in
cases when the class of pure states $\mathcal{M}$ consists of generalized
coherent sates of a compact simply-connected Lie group? 
\end{problem}
Before we formulate the answer to the above question in a rigorous
manner we recall the notation used in Section \ref{sec:semisimple-quadratic-characterisation}
and define formally the concept of ``bilinear, group invariant, correlation
criteria''. Let $K$ be a compact simply-connected Lie group irreducibly
represented, by the representation $\Pi$ on the finite dimensional
Hilbert space $\mathcal{H}^{\lambda_{0}}$ (see Sections \ref{sec:Rep theory of semisimple}
and \ref{sec:semisimple-quadratic-characterisation} for the explanation
of the terminology used in this section). Let the class of states
considered consists of Perelomov's coherent states, i.e. 
\begin{equation}
\mathcal{M}=\left\{ \Pi\left(k\right)\kb{\psi_{0}}{\psi_{0}}\Pi\left(k\right)^{\dagger}\,|\, k\in K\right\} \,,\label{eq:coherent sates another def}
\end{equation}
where $\ket{\psi_{0}}$ is the highest weight vector corresponding
to the highest weight $\lambda_{0}$. In what follows we will consequently
drop the superscript $\lambda_{0}$ when referring to the carrier
space of the representation $\Pi$. Let us define the concept of $K$-invariant
bilinear correlation witness.
\begin{defn}
An operator $V\in\mathrm{Herm}\left(\mathcal{H}\otimes\mathcal{H}\right)$
is called $K$-invariant bilinear correlation witness if and only
if the following two conditions are satisfied
\begin{enumerate}
\item For all $\rho,\sigma\in\mathcal{D}\left(\mathcal{H}\right)$ we have
\begin{equation}
\mathrm{tr}\left(\left[\rho\otimes\sigma\right]\, V\right)>0\,\Longrightarrow\mathrm{{\color{green}{\color{black}both}}}\,\rho\,\mbox{\text{and} }\sigma\text{ are correlated}\,\left(\rho,\sigma\notin\mathcal{M}^{c}\right)\,.\label{eq:condition bilin witness definition}
\end{equation}

\item Operator $V$ is $K$-invariant, i.e. for all $k\in K$
\begin{equation}
V=\left(\Pi\left(k\right)\otimes\Pi\left(k\right)\right)V\left(\Pi\left(k\right)^{\dagger}\otimes\Pi\left(k\right)^{\dagger}\right)\,.\label{eq:invariance def}
\end{equation}

\end{enumerate}
The set of $K$-invariant bilinear correlation witnesses will be denoted
by $\mathcal{W}_{2}^{K}\left(\mathcal{M}\right)$.
\end{defn}
Let us note that condition \eqref{eq:invariance def} is equivalent
to the $K$-invariance of the criterion specified by $V$, i.e. the
requirement that for all $\rho,\sigma\in\mathcal{D}\left(\mathcal{H}\right)$
and for all $k\in K$
\begin{equation}
\mathrm{tr}\left(\left[\rho\otimes\sigma\right]\, V\right)=\mathrm{tr}\left(\tilde{\rho}\otimes\tilde{\sigma}V\right)\,,\label{eq:invariance alternative}
\end{equation}
where $\tilde{X}=\Pi(k)X\Pi(k)^{\dagger}$. The equivalence of \eqref{eq:invariance def}
and \eqref{eq:invariance alternative} follows from the fact that
trace is a unitary invariant function and the existence of the Haar
measure on $K$. By the trace invariance we mean $\mathrm{tr}\left(A\right)=\mathrm{tr}\left(UAU^{\dagger}\right)$
for all unitary $U$. The existence of the Haar measure $\mu$ on
$K$ shows that the operator $V$ in \eqref{eq:invariance alternative}
can be replaced by
\[
V'=\int_{K}d\mu(k)\left(\Pi\left(k\right)\otimes\Pi\left(k\right)\right)V\left(\Pi\left(k\right)\otimes\Pi\left(k\right)\right)^{\dagger}\,,
\]
which is a manifestly $K$-invariant operator and gives, by the virtue
of \eqref{eq:invariance alternative}, the same criterion as $V$. 

Let us also remark that a mapping $k\rightarrow\Pi\left(k\right)\otimes\Pi\left(k\right)$
can be interpreted as a standard tensor product of representation
$\Pi$ with itself (see Eq. \eqref{eq:tensor product of reps}). Consequently,
the condition \eqref{eq:invariance alternative} can be interpreted
as the invariance of the criterion \eqref{eq:condition bilin witness definition},
when a sate $\rho,\sigma$ undergo the same unitary evolution described
by $\Pi\left(k\right)$, $k\in K$ (see Figure \ref{fig:group simult}
for the illustration of this idea). This is a natural requirement
as, by definition, the notion of correlations is invariant under the
action of the group $K$.

\begin{figure}[h]
\centering{}\includegraphics[width=8.5cm]{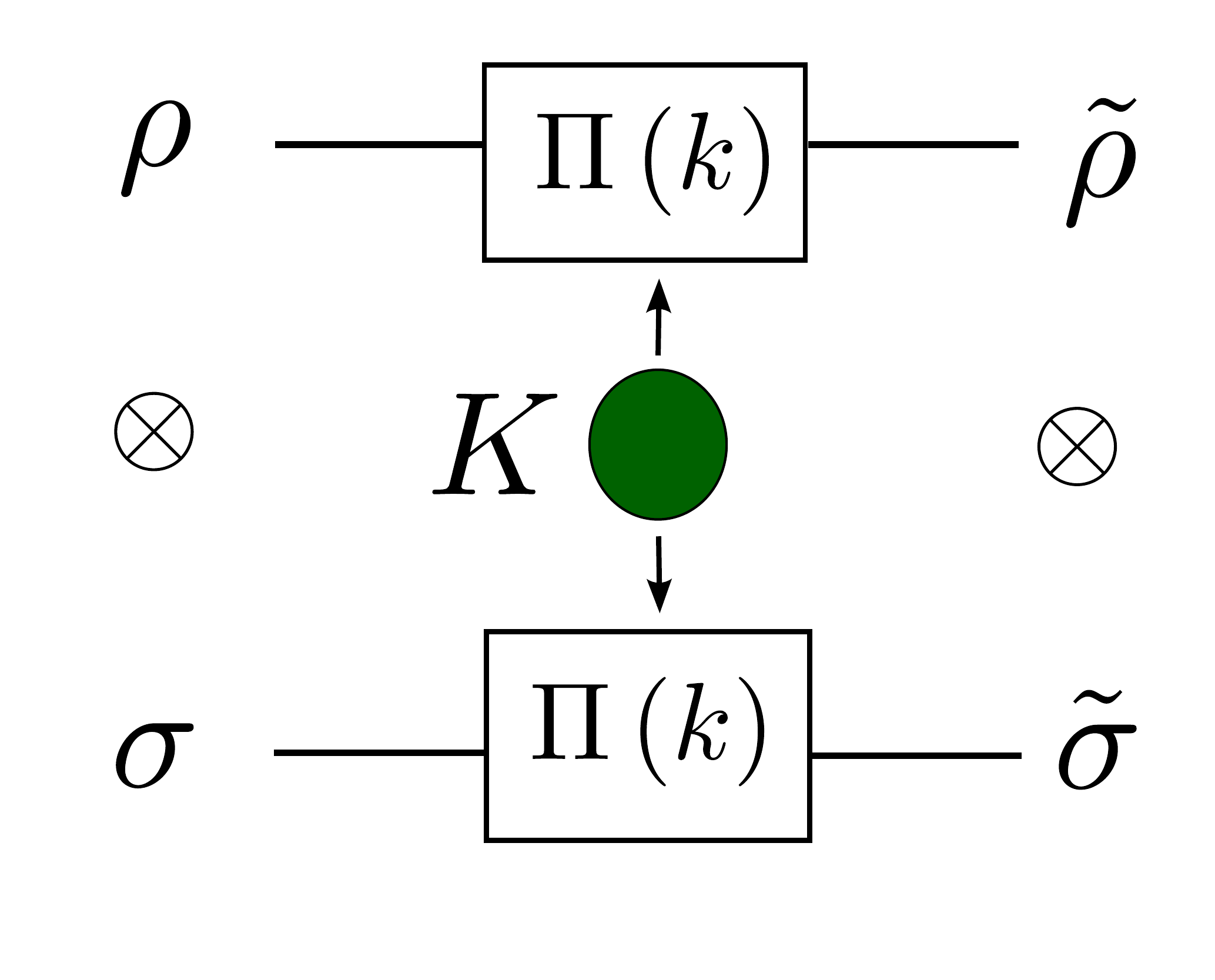}\protect\caption{\label{fig:group simult}To states $\rho,\sigma\in\mathcal{D}\left(\mathcal{H}\right)$
undergo a simultaneous evolution according to the unitary transformation
$\Pi\left(k\right)$ which represents the same element of the symmetry
group $k\in K$. Mathematically this transformation is described by
the mapping $\rho\otimes\sigma\rightarrow\tilde{\rho}\otimes\tilde{\sigma}$,
where $\tilde{X}=\Pi(k)X\Pi(k)^{\dagger}$. }
\end{figure}

The rest of the section is organized as follows. First in Subsection
\ref{sub:general structure bilin} we describe in detail the structure
of the set $\mathcal{W}_{2}^{K}\left(\mathcal{M}\right)$ and the
strength of correlation criteria based on $V\in\mathcal{W}_{2}^{K}\left(\mathcal{M}\right)$.
After the general discussion valid for all compact simply-connected
Lie groups we will focus on four specific types of correlations that
bear a physical relevance:
\begin{itemize}
\item Entanglement in system of distinguishable particles; $\mathcal{M}=\mathcal{M}_{d}$
- pure product states (Subsection \ref{sub:Entanglement-of-distinguishable-bilin});
\item Entanglement in system of finite number of bosons; $\mathcal{M}=\mathcal{M}_{b}$
- symmetric product states (Subsection \ref{sub:Enatnglement-of-bosons-twolin});
\item ``Entanglement'' in system of finite number of fermions; $\mathcal{M}=\mathcal{M}_{f}$
- symmetric product states (Subsection \ref{sub:Enatnglement-of-fermions-twolin});
\item Not convex-Gaussian correlation in fermionic systems; $\mathcal{M}_{g}$
- Fermionic Gaussian states (Subsection \ref{sub:Not-convex-Gaussian-correlation-bilin}).
\end{itemize}
In each of above cases we will give explicitly the structure $\mathcal{W}_{2}^{K}\left(\mathcal{M}\right)$
and discuss the strength of the corresponding criteria. We conclude
the section with the discussion of the obtained results presented
in Subsection \ref{sub:Discussion bilin}.

\subsection{Structure of $\mathcal{W}_{2}^{K}\left(\mathcal{M}\right)$\label{sub:general structure bilin}}

We start with the description of the basic properties of the set $\mathcal{W}_{2}^{K}\left(\mathcal{M}\right)$.
\begin{lem}
The set $\mathcal{W}_{2}^{K}\left(\mathcal{M}\right)$ is a finite
dimensional convex cone%
\footnote{A convex cone $\mathcal{C}\subset\mathcal{V}$ in a finite dimensional
real vector space $\mathcal{V}$ is a set satisfying: (i) $c_{1}+c_{2}\in\mathcal{C}$
for $c_{1,2}\in\mathcal{C}$, and (ii) $\alpha\cdot c\in\mathcal{C}$
for a scalar $\alpha\geq0$ and $c\in\mathcal{C}$.%
} in $\mathrm{Herm}\left(\mathcal{H}\otimes\mathcal{H}\right)$.\end{lem}
\begin{proof}
Condition \eqref{eq:condition bilin witness definition}, which essentially
describes the property of being a correlation witness, is equivalent
to
\begin{equation}
\mathrm{tr}\left(\left[\rho\otimes\sigma\right]\, V\right)\leq0\,\Longleftarrow\,\rho\mbox{ or }\sigma\text{ are not-correlated }\,\left(\rho\text{ or }\sigma\in\mathcal{M}^{c}\right)\,.\label{eq:equvalen witness}
\end{equation}
The above implication is certainly satisfied for $V'=V_{1}+V_{2}$
if it is satisfied for both $V_{1}$ and $V_{2}$. The same concerns
dilation by a non-negative scalar. Moreover, the property of being
$K$-invariant (See Eq. \eqref{eq:invariance def}) is intact if we
take sum of operators or multiply it by a positive number. Therefore,
$\mathcal{W}_{2}^{K}\left(\mathcal{M}\right)$ is a convex cone in
$\mathrm{Herm}\left(\mathcal{H}\otimes\mathcal{H}\right)$. A set
$\mathcal{W}_{2}^{K}\left(\mathcal{M}\right)$ is finite-dimensional
because $\mathcal{H}$ is finite-dimensional (by definition) and therefore
dimension of $\mathrm{Herm}\left(\mathcal{H}\otimes\mathcal{H}\right)$
if finite.
\end{proof}
Let us recall that the commutant (see Subsection \ref{sub:Representation-theory-of}
for a more detailed discussion of this concept) of the representation
$\Pi\otimes\Pi$ of $K$ in $\mathcal{H}\otimes\mathcal{H}$, denoted
by $\mathrm{Comm}\left(\mathbb{C}\left[\Pi\otimes\Pi\left(K\right)\right]\right)$,
is a subalgebra of $\mathrm{End}\left(\mathcal{H}\otimes\mathcal{H}\right)$
consisting of operators that commute with the representation of $\Pi\otimes\Pi$
of the group $K$,
\begin{equation}
X\in\mathrm{Comm}\left(\mathbb{C}\left[\Pi\otimes\Pi\left(K\right)\right]\right)\,\Longleftrightarrow\left[X,\,\Pi\left(k\right)\otimes\Pi\left(k\right)\right]=0\,,\label{eq:commutant recall}
\end{equation}
for all $k\in K.$ We observe that
\begin{equation}
\mathcal{W}_{2}^{K}\left(\mathcal{M}\right)\subset\mathrm{Comm}\left(\mathbb{C}\left[\Pi\otimes\Pi\left(K\right)\right]\right)\cap\mathrm{Herm}\left(\mathcal{H}\otimes\mathcal{H}\right)\,.\label{eq:commutant versus witness}
\end{equation}
In many cases the structure of the commutant is known or can be easily
deduced. This will prove to be useful and will allow to describe $\mathcal{W}_{2}^{K}\left(\mathcal{M}\right)$
explicitly in examples that follow. 

The following theorem states that under certain technical condition
on the representation $\Pi$ the cone $\mathcal{W}_{2}^{K}\left(\mathcal{M}\right)$
is the intersection of finite number of half-spaces.
\begin{thm}
\label{finitelly generated cone}Let $\Pi:K\rightarrow\mathrm{U}\left(\mathcal{H}\right)$
be an irreducible representation of the group $K$ with the property
that all weight spaces are one-dimensional%
\footnote{For the definition of weight spaces see Subsection \ref{sub:Structural-theory-of}.%
}. Let $\mathrm{Comm}\left(\Pi\otimes\Pi\left(K\right)\right)$ posses
the following property 
\begin{equation}
X\in\mathrm{Comm}\left(\mathbb{C}\left[\Pi\otimes\Pi\left(K\right)\right]\right)\Longrightarrow\mathbb{S}X\mathbb{S}=X\,,\label{eq:swap in the center of commutant}
\end{equation}
for $\mathbb{S}$ being a swap operator in $\mathcal{H}\otimes\mathcal{H}$.
Under this condition the cone of $K$-invariant correlation witnesses
$\mathcal{W}_{2}^{K}\left(\mathcal{M}\right)$ is the intersection
of finite number of half-spaces given by the following inequalities%
\footnote{If we had dropped out the condition $\tau X\tau=X$ for $X\in\mathrm{Comm}\left(\Pi\otimes\Pi\left(K\right)\right)$
we would have to impose another set of inequalities :
\[
\bra{\psi_{\lambda}}\bra{\psi_{0}}V\ket{\psi_{\lambda}}\ket{\psi_{0}}\leq0\,.
\]
}

\begin{equation}
\bra{\psi_{0}}\bra{\psi_{\lambda}}V\ket{\psi_{0}}\ket{\psi_{\lambda}}\leq0\,,\label{eq:inequalities structure general}
\end{equation}
where $V\in\mathrm{Comm}\left(\Pi\otimes\Pi\left(K\right)\right)\cap\mathrm{Herm}\left(\mathcal{H}\otimes\mathcal{H}\right)$,
$\ket{\psi_{0}}$is the highest weight vector and vectors $\ket{\psi_{\lambda}}$ranges
over all weight vectors of the representation $\Pi$ in $\mathcal{H}$. \end{thm}
\begin{proof}
We first note that by the virtue of Lemma \ref{lem:general bilin}
it is enough to consider 
\[
V\in\mathrm{Comm}\left(\mathbb{C}\left[\Pi\otimes\Pi\left(K\right)\right]\right)\cap\mathrm{Herm}\left(\mathcal{H}\otimes\mathcal{H}\right)
\]
that satisfy
\[
\bra{\psi}\bra{\phi}V\ket{\psi}\ket{\phi}\leq0\,,
\]
where $\kb{\psi}{\psi}\in\mathcal{M}$ and $\ket{\phi}\in\mathcal{H}$.
Because of \eqref{eq:coherent sates another def} and $K$-invariance
of $V$ the above inequality is equivalent to
\begin{equation}
\bra{\psi_{0}}\bra{\phi}V\ket{\psi_{0}}\ket{\phi}\leq0\,,\label{eq:simplifed condition}
\end{equation}
for all $\ket{\phi}\in\mathcal{H}$. Decomposing the vector $\ket{\phi}$
in the basis of weight vectors (see Subsection \ref{sub:Structural-theory-of}
for details) we get
\begin{equation}
\ket{\phi}=\sum_{\lambda}c_{\lambda}\ket{\psi_{\lambda}}\,,\label{eq:decomposition weights}
\end{equation}
where the sum ranges over all weights that occur in the representation
$\mathcal{H}$. Using \eqref{eq:decomposition weights} we get 
\begin{equation}
\sum_{\lambda,\lambda'}c_{\lambda}^{\ast}c_{\lambda'}\bra{\psi_{0}}\bra{\phi_{\lambda}}V\ket{\psi_{0}}\ket{\phi_{\lambda'}}\leq0\,.\label{eq:simplified condition 2}
\end{equation}
Operator $V$ belongs to the commutant of the representation $\Pi\otimes\Pi$
in $\mathcal{H}\otimes\mathcal{H}$ and consequently it commutes with
the representation of the associated representation of the Lie algebra
$\pi\otimes\pi$ (see Eq.\eqref{eq:tensor product of reps alg}) restricted
to the Cartan algebra $\mathfrak{h}\subset\mathfrak{k}^{\mathbb{C}}$,
\begin{equation}
\left[V,\pi\otimes\pi\left(H\right)\right]=0\,,\label{eq:commutation with Cartan}
\end{equation}
where
\[
\pi\otimes\pi\left(H\right)=\pi\left(H\right)\otimes\mathbb{I}+\mathbb{I}\otimes\pi\left(H\right)\,,
\]
for $H\in\mathfrak{h}$. From \eqref{eq:commutation with Cartan}
it follows that $V$ preserves weight spaces in $\mathcal{H}\otimes\mathcal{H}$.
Using this and the fact that weight spaces are mutually orthogonal
(see Subsection \eqref{sub:Structural-theory-of}) the double sum
in \eqref{eq:simplified condition 2} reduces to the single sum
\[
\sum_{\lambda}\left|c_{\lambda}\right|^{2}\bra{\psi_{0}}\bra{\phi_{\lambda}}V\ket{\psi_{0}}\ket{\phi_{\lambda}}\leq0\,
\]
which is equivalent to \eqref{eq:inequalities structure general}.
\end{proof}
We now move to the discussion of the strength of the criteria based
on operators $V\in\mathcal{W}_{2}^{K}\left(\mathcal{M}\right)$. More
precisely, we are interested in the following problem.
\begin{problem}
Given a mixed state $\rho\in\mathcal{D}\left(\mathcal{H}\right)$,
what is the optimal strategy, based on the criterion \eqref{eq:condition bilin witness definition}
and $V\in\mathcal{W}_{2}^{K}\left(\mathcal{M}\right)$, that allows
to detect correlations in state $\rho$?
\end{problem}
Before we answer the above question we have to introduce some terminology
from convex geometry (for the comprehensive introduction to the field
of convex geometry see \citep{ConvexAnalysis1997}). A \textit{face}
of a convex set $\mathcal{C}$ in a real vector space $\mathcal{V}$
is a nonempty convex subset, $F\subset\mathcal{C}$ with the property
that if $x,y\in\mathcal{C},\,\theta\in(0,1$), and $\theta x+(1-\theta)y\in F$,
then $x,y\in F$. An \textit{extreme ray} $\mathcal{R}$ of a convex
cone $\mathcal{C}$ is a face which is a half-line originating from
the origin of the cone (a point $0\in\mathcal{V}$). For technical
reasons we restrain ourselves to so-called \textit{finitely-generated
cones} %
\footnote{In examples that follow we will only deal with finitely-generated
cones. Moreover, for finitely-generated cones it is very simple to
gibe an algorithm for deciding a possible sign of a linear functional
on it (see Proof of Lemma \ref{ultimate streangth}). %
}i.e cones defined via,
\begin{equation}
\mathcal{C}_{v_{1},\ldots,v_{N}}=\left\{ \lambda_{1}v_{1}+\dots+\lambda_{N}v_{N}|\,\lambda_{i}\geq0,\, v_{i}\in\mathcal{V}\right\} \,,\label{eq:finitely generated cone}
\end{equation}
where $\mathcal{V}$ is the ambient vector space. We will refer to
vectors $v_{1},\ldots,v_{N}$ as the generators of the cone and denote
this (in general not-unique) set as $\mathcal{G}\left(\mathcal{C}\right)$.
For a pictorial representation of a finitely-generated cone see Figure
\ref{fig:finitelly generated cone}. 

\begin{figure}[h]
\centering{}\includegraphics[width=8.5cm]{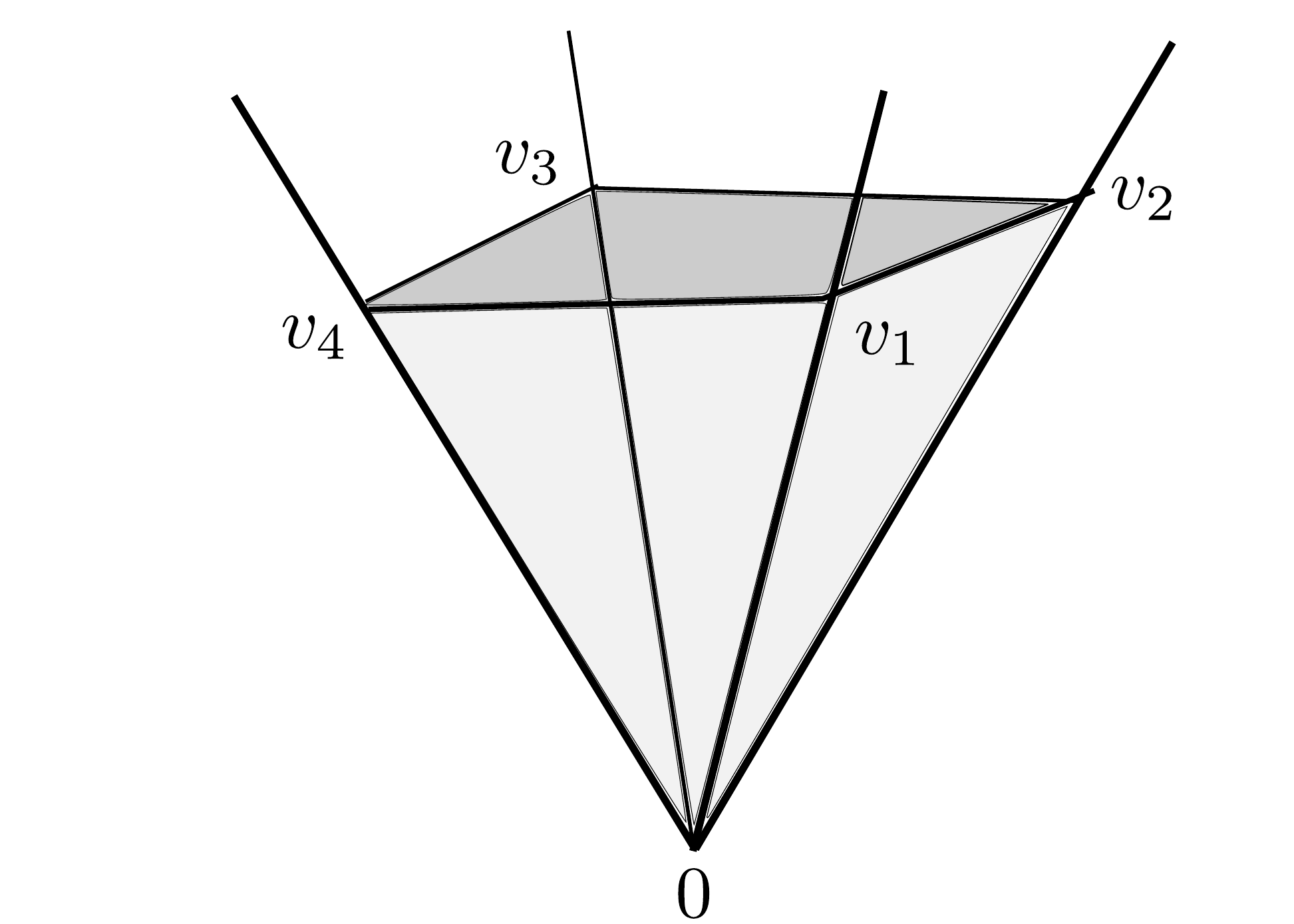}\protect\caption{\label{fig:finitelly generated cone}A graphical presentation of an
exemplary finitely-generated cone generated by vectors $v_{1},v_{2},v_{3},v_{4}\in\mathbb{R}^{3}$.
The cone $\mathcal{C}_{v_{1},v_{2},v_{3},v_{4}}$ consists of all
linear combinations of vectors $v_{1},v_{2},v_{3},v_{4}$ with non-negative
coefficients.}
\end{figure}

The following theorem characterizes the detection power of bilinear,
$K$-invariant witnesses of correlations.
\begin{lem}
\label{ultimate streangth}Let $\rho,\sigma\in\mathcal{D}\left(\mathcal{M}\right)$.
The following statements hold
\begin{itemize}
\item Correlations in states $\rho$ and $\sigma$ can be detected via criteria
based on \textup{$\mathcal{W}_{2}^{K}\left(\mathcal{M}\right)$ if
and only if }there exist an operator $V\in\mathcal{G}\left(\mathcal{W}_{2}^{K}\left(\mathcal{M}\right)\right)$
such that for some nonzero $V\in\mathcal{R}$
\begin{equation}
\mathrm{tr}\left(\left[\rho\otimes\sigma\right]V\right)>0\,.\label{eq:cond streangth}
\end{equation}

\item \textup{Correlations in the state $\rho$ can be detected are detected
}via criteria based on \textup{$\mathcal{W}_{2}^{K}\left(\mathcal{M}\right)$
if and only if }there exist an operator $V\in\mathcal{G}\left(\mathcal{W}_{2}^{K}\left(\mathcal{M}\right)\right)$
such that
\begin{equation}
\lambda_{\mathrm{max}}\left\{ \mathrm{tr}_{1}\left(\left[\rho\otimes\mathbb{I}\right]V\right)\right\} >0\,,\label{eq:final condition streangth}
\end{equation}
where $\lambda_{\mathrm{max}}$ denotes the maximum eigenvalue and
$\mathrm{tr}_{1}\left(\cdot\right)$ denotes the partial trace over
first factor of the tensor product $\mathcal{H}\otimes\mathcal{H}$.
\end{itemize}
\end{lem}
\begin{proof}
Let us first prove the first statement. Since $\mathcal{W}_{2}^{K}\left(\mathcal{M}\right)$
is finitely-generated (this follows from the fact that it is given
by the intersection of finitely many inequalities \eqref{eq:inequalities structure general}),
every element of $\mathcal{W}_{2}^{K}\left(\mathcal{M}\right)$ can
be written as a non-negative combination of elements from the generating
set 
\[
V=\sum_{i=1}^{n}\lambda_{i}V_{i}\,,
\]
where $\lambda_{i}\geq0$ and $\mathcal{G}\left(\mathcal{W}_{2}^{K}\left(\mathcal{M}\right)\right)=\left\{ V_{1},\ldots,V_{n}\right\} $.
Due to the homogeneity of condition \eqref{eq:condition bilin witness definition}
in $V$, we can without the loss of generality assume that $\sum_{i=1}^{n}\lambda_{i}=1$.
This condition cuts out in the cone $\mathcal{W}_{2}^{K}\left(\mathcal{M}\right)$
a polytope which we denote by $\Delta$. For a given states $\rho,\sigma\in\mathcal{D}\left(\mathcal{M}\right)$,
the mapping
\[
V\rightarrow\mathrm{tr}\left(\left[\rho\otimes\sigma\right]V\right)\,,
\]
 is a linear function on $\mathrm{Herm}\left(\mathcal{H}\otimes\mathcal{H}\right)$.
This function, when restricted to $\Delta$, attains its maximum on
vertices of $\Delta$. From the definition of $\mathcal{G}\left(\mathcal{W}_{2}^{K}\left(\mathcal{M}\right)\right)$
it follows that vertices of $\Delta$ are contained in $\mathcal{G}\left(\mathcal{W}_{2}^{K}\left(\mathcal{M}\right)\right)$.
This finishes the proof of the first statement. 

Proof of the second statement is straightforward. From arguments given
above it follows that criterion based on $\mathcal{W}_{2}^{K}\left(\mathcal{M}\right)$
will detect correlations in the state $\rho$ if and only if 
\[
\underset{\sigma\in\mathcal{D}\left(\mathcal{H}\right)}{\mathrm{max}}\underset{V\in\mathcal{G}\left(\mathcal{W}_{2}^{K}\left(\mathcal{M}\right)\right)}{\mathrm{max}}\mathrm{tr}\left(\left[\rho\otimes\sigma\right]V\right)>0\,.
\]
The order of $\mathrm{max}$ operators can be inverted. The maximum
over states $\sigma\in\mathcal{D}\left(\mathcal{H}\right)$ can be
computed using the following identity, valid for all $V\in\mathrm{Herm}\left(\mathcal{H}\otimes\mathcal{H}\right)$,

\begin{equation}
\underset{\sigma\in\mathcal{D}\left(\mathcal{H}\right)}{\mathrm{max}}\mathrm{tr}\left(\left[\rho\otimes\sigma\right]V\right)=\underset{\sigma\in\mathcal{D}\left(\mathcal{H}\right)}{\mathrm{max}}\mathrm{tr}\left(\sigma\mathrm{tr}_{1}\left(\left[\rho\otimes\mathbb{I}\right]V\right)\right)\,.\label{eq:maximum idenity}
\end{equation}
Right-hand side of \eqref{eq:maximum idenity} equals $\lambda_{\mathrm{max}}\left\{ \mathrm{tr}_{1}\left(\left[\rho\otimes\mathbb{I}\right]V\right)\right\} $
which follows from the linearity of 
\[
\mathrm{tr}\left(\sigma\mathrm{tr}_{1}\left(\left[\rho\otimes\mathbb{I}\right]V\right)\right)
\]
 in $\sigma$. This concludes the proof of \eqref{eq:final condition streangth}.
\end{proof}
Theorem \ref{finitelly generated cone} together with Lemma \ref{ultimate streangth}
give a constructive algorithm to describe the structure of $\mathcal{W}_{2}^{K}\left(\mathcal{M}\right)$
and to gauge the power of criteria based on it. The procedure is the
following
\begin{enumerate}
\item Describe the structure of the commutant $\mathrm{Comm}\left(\Pi\otimes\Pi\left(K\right)\right)$
(find its basis or a generating set)
\item Use inequalities \eqref{eq:inequalities structure general} to find
a generating set of the cone $\mathcal{W}_{2}^{K}\left(\mathcal{M}\right)$.
\item Apply Lemma \ref{ultimate streangth} to gauge the strength of the
criteria based on $\mathcal{W}_{2}^{K}\left(\mathcal{M}\right)$.
\end{enumerate}
In the next four subsections we will apply the above algorithm to
four physically-motivated classes of correlations.

\subsection{Entanglement of distinguishable particles\label{sub:Entanglement-of-distinguishable-bilin}}

In this part we describe in detail the structure of the set $\mathcal{W}_{2}^{K}\left(\mathcal{M}\right)$
in the context of detecting standard entanglement in the system of
$L$ distinguishable particles. The relevant class of pure states
$\mathcal{M=}\mathcal{M}_{d}$ consists of pure product states (see
Eq.\eqref{eq:product distinguishable def}). In what follows we will
use the notation introduced in Subsection \ref{sub:Product-states}.
Recall that the symmetry group in the context of entanglement is the
local unitary group $K=\mbox{\ensuremath{\mathrm{LU}}=}\times_{i=1}^{i=L}SU(N_{i})$
irreducibly represented on a Hilbert space 
\[
\mathcal{H}_{d}=\mathcal{H}_{1}\otimes\cdots\otimes\mathcal{H}_{L}=\mathbb{C}^{N_{1}}\otimes\ldots\otimes\mathbb{C}^{N_{L}}=\bigotimes_{i=1}^{L}\mathbb{C}^{N_{i}}\,,
\]
via the tensor product of the defining representations of $\mathrm{SU}(N_{i})$,
\[
\Pi_{d}(U_{1}\ldots,U_{L})=U_{1}\otimes\cdots\otimes U_{L}\,.
\]
It will be convenient for us to use the identification of multiple
tensor products 
\begin{equation}
\mathcal{H}_{d}\otimes\mathcal{H}_{d}=\bigotimes_{i=1}^{i=L}\left(\mathcal{H}_{i}\otimes\mathcal{H}_{i'}\right)\,,\label{eq:two copies dist notation}
\end{equation}
where we label spaces from the second copy of the total space with
primes in order to avoid ambiguity (note that \eqref{eq:two copies dist notation}
is analogous to the notation introduced in Subsection \ref{sub:Generalized-coherent-states-multilin}).
Let $\mathbb{S}_{ii'}:\mathcal{H}_{i}\otimes\mathcal{H}_{i'}\rightarrow\mathcal{H}_{i}\otimes\mathcal{H}_{i'}$
be the swap operation flipping the factors of $\mathcal{H}_{i}\otimes\mathcal{H}_{i'}$
and let $\mathbb{I}_{ii'}:\mathcal{H}_{i}\otimes\mathcal{H}_{i'}\rightarrow\mathcal{H}_{i}\otimes\mathcal{H}_{i'}$
be the identity operator in this space. Note that under the above
identification the tensor product of representation $\Pi_{d}$ takes
a convenient form
\begin{equation}
\Pi_{d}\otimes\Pi_{d}(U_{1}\ldots,U_{L})=\bigotimes_{i=1}^{i=L}\left(U_{i}\otimes U_{i}\right)\,.\label{eq:tensor product of represnetations ditinguishable}
\end{equation}
The following lemma characterizes the commutant of the representation
$\Pi_{d}\otimes\Pi_{d}$ in $\mathcal{H}_{dist}\otimes\mathcal{H}_{dist}$.
\begin{lem}
\label{commutant distinguishable}The commutant $\mathrm{Comm}\left(\mathbb{C}\left[\Pi_{d}\otimes\Pi_{d}\left(\mathrm{LU}\right)\right]\right)$
is commutative. Let $X\subset\left\{ 1,\ldots,L\right\} $ be the
arbitrary subset of the set $\left\{ 1,\ldots,L\right\} $. Let $\mathbb{S}^{X}:\mathcal{H}_{d}\otimes\mathcal{H}_{d}\rightarrow\mathcal{H}_{d}\otimes\mathcal{H}_{d}$
be a Hermitian operator defined by
\begin{equation}
\mathbb{S}^{X}=\left(\bigotimes_{i\in X}\mathbb{S}_{ii'}\right)\otimes\left(\bigotimes_{i\in\bar{X}}\mathbb{I}_{ii'}\right)\,,\label{eq:complicated swap operator}
\end{equation}
where it is understood that the relevant factors in \eqref{eq:complicated swap operator}
act on the appropriate factors in the tensor product \eqref{eq:two copies dist notation}
and $\bar{X}$ is the complement of the set $X$. 

Operators $\mathbb{S}^{X}$ for all possible choices of $X$ form
the basis of $\mathrm{Comm}\left(\Pi_{d}\otimes\Pi_{d}\left(\mathrm{LU}\right)\right)$.\end{lem}
\begin{proof}
Let us note that the operator algebra generated by the unitary representation
$\Pi_{d}\otimes\Pi_{d}$ has the following structure
\begin{equation}
\mathbb{C}\left[\Pi_{d}\otimes\Pi_{d}\left(\mathrm{LU}\right)\right]=\bigotimes_{i=1}^{L}\mathbb{C}\left[\Pi_{i}\otimes\Pi_{i}\left(\mathrm{SU}\left(N_{i}\right)\right)\right]\,,\label{eq:decomposition of product}
\end{equation}
where the multiple tensor product in \eqref{eq:decomposition of product}
corresponds to the decompositions \eqref{eq:two copies dist notation}
and \eqref{eq:tensor product of represnetations ditinguishable}.
Representation $\Pi_{i}$ appearing in \eqref{eq:decomposition of product}
is a defining representation of $\mathrm{SU}\left(N_{i}\right)$.
From Eq.\eqref{eq:schur weyl ytwo copies}, given in Subsection \eqref{sub:Representation-theory-of},
we have
\begin{equation}
\mathrm{Comm}\left(\Pi_{i}\otimes\Pi_{i}\left[\mathrm{SU}\left(N_{i}\right)\right]\right)=\mathrm{Span}_{\mathbb{C}}\left\{ \mathbb{I}\otimes\mathbb{I},\,\mbox{\ensuremath{\mathbb{S}}}\right\} \,,\label{eq:commutant of tensor product}
\end{equation}
where $\mathbb{I}\otimes\mathbb{I}$ is the identity operator on $\mathcal{H}_{i}\otimes\mathcal{H}_{i}$
and $\mathbb{S}$ is the swap operator on this space. We now use the
known property of a commutant discussed in Subsection \ref{sub:Representation-theory-of}.
For any operator algebras $\mathcal{A}_{1}$ and $\mathcal{A}_{2}$
(over the field of complex numbers) $\mathrm{Comm}\left(\mathcal{A}_{1}\otimes\mathcal{A}_{2}\right)=\mathrm{Comm}\left(\mathcal{A}_{1}\right)\otimes\mathrm{Comm}\left(\mathcal{A}_{2}\right)$.
Operators $\mathbb{S}^{X}$ from Eq.\eqref{eq:complicated swap operator}
clearly form a basis of
\[
\bigotimes_{i=1}^{L}\mathrm{Comm}\left(\Pi_{i}\otimes\Pi_{i}\left[\mathrm{SU}\left(N_{i}\right)\right]\right)
\]
which concludes the proof.
\end{proof}
In what follows we will make an extensive use of operators $\mathbb{S}^{X}$
from Lemma \ref{commutant distinguishable}. Pictorial representation
of the action of $\mathbb{S}^{X}$ is presented in Figure on the example
of $\mathbb{S}^{\left\{ 2,3\right\} }$.

\begin{figure}[h]
\centering{}\includegraphics[width=8.5cm]{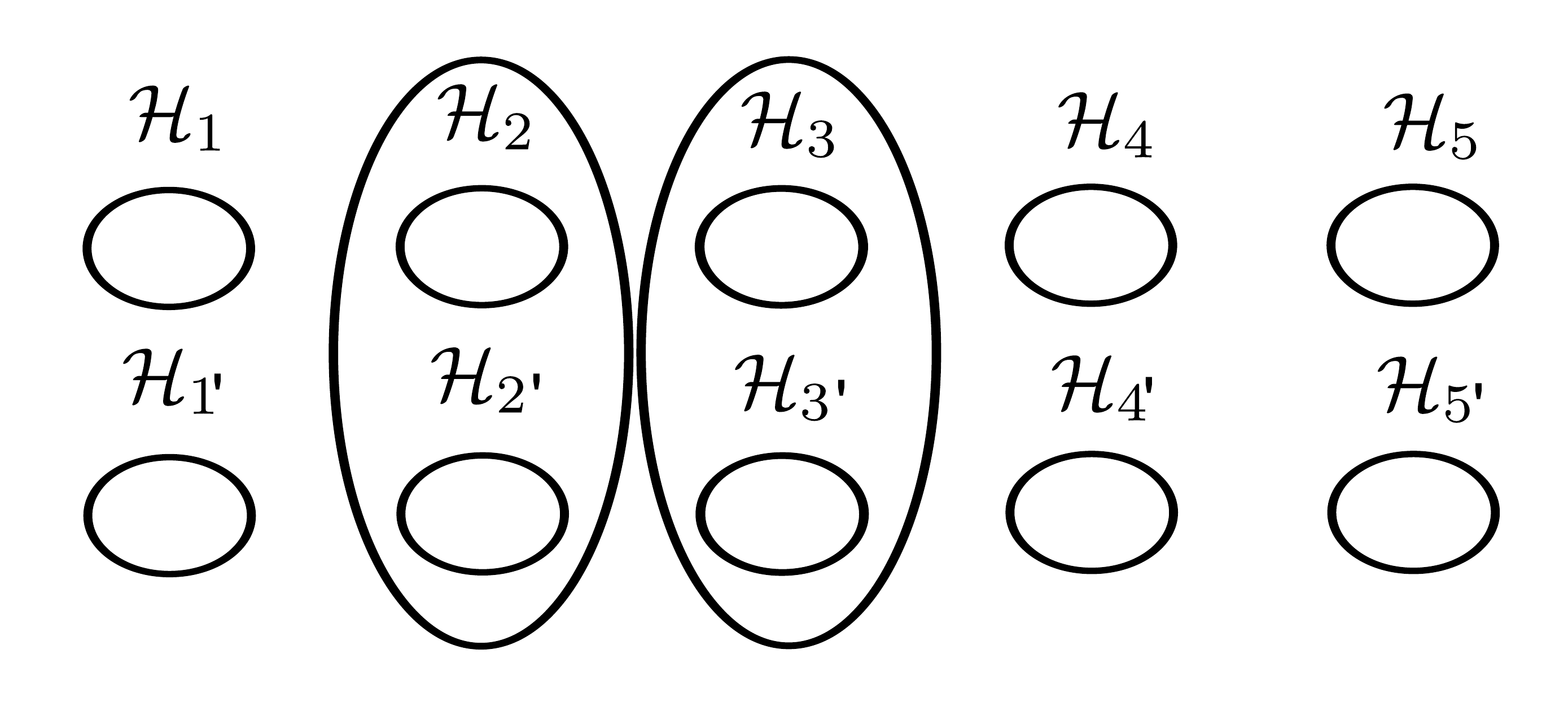}\protect\caption{\label{fig:complex swap}A graphical presentation of the action of
$\mathbb{S}^{\left\{ 2,3\right\} }$ in $\bigotimes_{i=1}^{i=5}\left(\mathcal{H}_{i}\otimes\mathcal{H}_{i'}\right)$.
The operator performs the independent swaps between factors of the
tentsor product enclosed by black curves. }
\end{figure}

In particular we will use of the following well-known fact.
\begin{fact}
\label{fact:encoding of reduction}Let $\rho,\sigma\in\mathcal{D}\left(\mathcal{H}_{dist}\right)$
and let $X\subset\left\{ 1,\ldots,L\right\} $. Then the following
equality holds,

\begin{equation}
\mathrm{tr}\left(\left[\rho\otimes\sigma\right]\mathbb{S}^{X}\right)=\mathrm{tr}\left(\rho_{X}\sigma_{X}\right)\,,\label{eq:swap expectation value}
\end{equation}
where $\rho_{X},\sigma_{X}\in\mathcal{D}\left(\bigotimes_{i\in X}\mathcal{H}_{i}\right)$
are reduction of the states $\rho,\sigma$ to the subsystem of a composite
$L$-partite system consisting of particles labeled indices belonging
to the subset $X$.\end{fact}
\begin{lem}
\label{LU invariant bilin theorem}Every $\mathrm{LU}$-invariant
bilinear entanglement witness $V$ ($V\in\mathcal{W}_{2}^{\mathrm{LU}}\left(\mathcal{M}_{dist}\right)$)
for $L$ distinguishable particles has the structure
\begin{equation}
V=\sum_{X\subset\left\{ 1,\ldots,L\right\} }a_{X}\mathbb{S}^{X}\,,\label{eq:bilin entangl structure}
\end{equation}
where $a_{X}\in\mathbb{R}$ are real parameters labeled by subsets
of the set $\left\{ 1,\ldots,L\right\} $ satisfying the following
inequalities (also labeled by subsets $Y$ of the set $\left\{ 1,\ldots,L\right\} $),
\begin{equation}
\sum_{X:X\subset Y}a_{X}\leq0\,.\label{eq:inequalities entanglement}
\end{equation}
\end{lem}
\begin{proof}
We will make a direct use of Theorem \ref{finitelly generated cone}
as the considered scenario satisfies its assumptions. We will show
that inequalities \eqref{eq:inequalities entanglement} are equivalent,
in this particular example, to inequalities \eqref{eq:inequalities structure general}.
Let us fix some orthonormal bases
\[
\left\{ \ket{e_{j}^{(i)}}\right\} _{j=1}^{j=N_{j}}\,,\, i=1,\ldots,L
\]
of the spaces $\mathcal{H}_{i}$ composing the full tensor product
$\mathcal{H}_{dist}=\bigotimes_{i=1}^{L}\mathcal{H}_{i}$. As a highest
weight state $\ket{\psi_{0}}$ of $\mathrm{LU}$ we take 
\begin{equation}
\ket{\psi_{0}}=\bigotimes_{i=1}^{L}\ket{e_{1}^{(i)}}\,.\label{eq:high weight vector disting}
\end{equation}
As weight vectors we can take (see Subsection \ref{sub:Structural-theory-of})
standard product vectors in $\mathcal{H}_{d}$,
\begin{equation}
\ket{\psi_{\lambda}}=\ket{e_{j_{1}}^{\left(1\right)}}\otimes\ket{e_{j_{2}}^{\left(2\right)}}\otimes\ldots\otimes\ket{e_{j_{L}}^{\left(L\right)}}\,,\label{eq:weight vector distinguish}
\end{equation}
where the parameter $\lambda$ encodes the subsequent subscripts $j_{1},\ldots,j_{L}$
and thus we write $\lambda=\left(j_{1},\ldots,j_{L}\right)$. To every
$\lambda$ we associate an auxiliary set $X_{\lambda}\subset\left\{ 1,\ldots,L\right\} $
defined as follows,
\begin{equation}
X_{\lambda}=\left\{ \left.i\right|j_{i}=1\right\} \,.\label{eq:auxiliary set}
\end{equation}
Simple calculation shows that for all $Y\subset\left\{ 1,\ldots,L\right\} $
\begin{equation}
\bra{\psi_{0}}\bra{\psi_{\lambda}}\mathbb{S}^{Y}\ket{\psi_{0}}\ket{\psi_{\lambda}}=\begin{cases}
1 & \text{for \ensuremath{Y}\ensuremath{\subset}}X_{\lambda}\\
0 & \text{otherwise}
\end{cases}\,.\label{eq:matrix element swap}
\end{equation}
By the virtue of Lemma \ref{commutant distinguishable} every operator
\[
V\in\mathrm{Comm}\left(\mathbb{C}\left[\Pi_{dist}\otimes\Pi_{dist}\left(\mathrm{LU}\right)\right]\right)\cap\mathrm{Herm}\left(\mathcal{H}_{dist}\otimes\mathcal{H}_{dist}\right)
\]
can be written as
\[
V=\sum_{X\subset\left\{ 1,\ldots,L\right\} }a_{X}\mathbb{S}^{X}\,,
\]
where $a_{X}\in\mathbb{R}$ are real parameters. Inequalities \eqref{eq:inequalities entanglement}
follow from applying to the above expression inequalities \eqref{eq:inequalities structure general},
Eq.\eqref{eq:matrix element swap} and by noting that as we range
over all $\lambda$, subsets $X_{\lambda}$ range over all possible
subsets of $\left\{ 1,\ldots,L\right\} $.
\end{proof}
It turns out that inequalities \eqref{eq:inequalities entanglement}
are independent. Moreover, it is possible to compute extreme rays
of the cone given by them explicitly. This is settled by the following
lemma
\begin{lem}
\label{lem: extrem rays entanglement}The following operators form
a generating set of the cone $\mathcal{W}_{2}^{\mathrm{LU}}\left(\mathcal{M}_{d}\right)$,
\begin{equation}
V^{Y}=-\left(-1\right)^{\left|Y\right|}\sum_{X:X\supset Y}\left(-1\right)^{\left|X\right|}\mathbb{S}^{X}\,,\label{eq:formula extremal entanglement}
\end{equation}
where $Y\subset\left\{ 1,\ldots,L\right\} $, $\left|\cdot\right|$
denotes the cardinality of a discrete set and the summation is over
all subsets of $\left\{ 1,\ldots,L\right\} $ that contain the set
$Y$. Moreover, half-lines in $\mathrm{Herm}\left(\mathcal{H}_{d}\otimes\mathcal{H}_{d}\right)$
passing through the null operator and one of the operators $V^{Y}$
form $2^{L}$ extreme rays of the cone $\mathcal{W}_{2}^{\mathrm{LU}}\left(\mathcal{M}_{d}\right)$. 

Consequently, entanglement of states $\rho,\sigma\in\mathcal{D}\left(\mathcal{H}_{d}\right)$
can be detected by criterion based on operators belonging to $\mathcal{W}_{2}^{\mathrm{LU}}\left(\mathcal{M}_{d}\right)$
if and only if
\begin{equation}
\underset{Y\subset\left\{ 1,\ldots,L\right\} }{\mathrm{max}}\mathrm{tr}\left(\left[\rho\otimes\sigma\right]V^{Y}\right)>0\,,\label{eq:ultimate power}
\end{equation}
Moreover, entanglement of $\rho\in\mathcal{D}\left(\mathcal{H}_{d}\right)$
can be detected by $V\in\mathcal{W}_{2}^{\mathrm{LU}}\left(\mathcal{M}_{d}\right)$
if and only if
\begin{equation}
\underset{Y\subset\left\{ 1,\ldots,L\right\} }{\mathrm{max}}\lambda_{\mathrm{max}}\left\{ \mathrm{tr}_{1}\left(\left[\rho\otimes\mathbb{I}\right]V^{Y}\right)\right\} >0\,,\,\label{eq:alternative power}
\end{equation}
where notation in \eqref{eq:alternative power} is the same as in
\eqref{eq:final condition streangth}.\end{lem}
\begin{proof}
The proof that the generating set of the cone $\mathcal{W}_{2}^{\mathrm{LU}}\left(\mathcal{M}_{d}\right)$
is given by operators \eqref{eq:formula extremal entanglement} follows
from the fact that the following linear transformation is invertible
\begin{equation}
b_{Y}=\sum_{X:X\subset Y}a_{X}\,,\label{eq:subsets relations}
\end{equation}
where $\left\{ a_{X}\right\} _{X\subset\left\{ 1,\ldots,L\right\} }$
and $\left\{ b_{Y}\right\} _{Y\subset\left\{ 1,\ldots,L\right\} }$
are $2^{L}$ dimensional real vectors whose coordinates are labeled
by subsets of the set $\left\{ 1,\ldots,L\right\} $. The inverse
formula to \eqref{eq:subsets relations} is given by the well-known
inclusion-exclusion principle \citep{Knuth1989},
\begin{equation}
a_{X}=\sum_{Y:Y\subset X}\left(-1\right)^{\left|X\right|+\left|Y\right|}b_{Y}\,.\label{eq:inclusion-exclusion}
\end{equation}
The invertibility of \eqref{eq:subsets relations} ensures that in
the new variables $\left\{ b_{Y}\right\} _{Y\subset\left\{ 1,\ldots,L\right\} }$
the cone $\mathcal{W}_{2}^{\mathrm{LU}}\left(\mathcal{M}_{d}\right)$
is simply the intersection of half-spaces
\begin{equation}
b_{Y}\leq0\,,\label{eq:new ineq}
\end{equation}
for $Y\subset\left\{ 1,\ldots,L\right\} $. In these new variables
the generating vectors $V^{\tilde{Y}}$ of $\mathcal{W}_{2}^{\mathrm{LU}}\left(\mathcal{M}_{d}\right)$
obviously have coordinates of the form $b_{Y}=-\delta_{Y,\tilde{Y}}$,
where $\delta_{Y,\tilde{Y}}$ denotes the Kronecker delta. Using this
relation and inserting it to the inverse formula \eqref{eq:inclusion-exclusion}
we get the ``old'' coordinates of $V^{Y}$,
\[
a_{X}=-\sum_{Y:Y\subset X}\left(-1\right)^{\left|X\right|+\left|Y\right|}\delta_{Y,\tilde{Y}}=\begin{cases}
-\left(-1\right)^{\left|X\right|+\left|\tilde{Y}\right|} & \text{for }\tilde{Y}\subset X\\
0 & \text{otherwivse}
\end{cases}\,.
\]
Inserting the above to \eqref{eq:bilin entangl structure} we obtain
\eqref{eq:formula extremal entanglement}. The proof of the fact that
operators $V^{Y}$ define the extreme rays of $\mathcal{W}_{2}^{\mathrm{LU}}\left(\mathcal{M}_{d}\right)$
is immediate. The statements that proceed Eq.\eqref{eq:ultimate power}
and Eq.\eqref{eq:alternative power} follow from Lemma \ref{ultimate streangth}.
\end{proof}

\subsection{Particle entanglement of bosons\label{sub:Enatnglement-of-bosons-twolin}}

In this subsection we analyze in detail the set $\mathcal{W}_{2}^{K}\left(\mathcal{M}\right)$
for the case of particle entanglement of $L$ bosonic particles. We
will consequently recall the notation and terminology from Subsection
\ref{sub:Symmetric-product-states}. The appropriate class of pure
states $\mathcal{M}$ is $\mathcal{M}_{b}$ consisting of bosonic
product states (see Eq.\eqref{eq:crit prod bos states}). The relevant
symmetry group in this context is the ``bosonic local unitary group''
$\mathrm{LU}_{b}=\mathrm{SU}\left(d\right)$, irreducibly represented
on a $L$-boson Hilbert space
\[
\mathcal{H}_{b}=\mathrm{Sym}^{L}\left(\mathcal{H}\right)\,,\,\mathcal{H}\approx\mathbb{C}^{d}
\]
via the representation
\[
\Pi_{b}\left(U\right)=\stackrel{L\,\text{times}}{\overbrace{U\otimes\ldots\otimes U}}\,.
\]
We will use the following embedding of multiple tensor products 

\begin{equation}
\mathcal{H}_{b}\otimes\mathcal{H}_{b}\subset\mathcal{H}_{d}\otimes\mathcal{H}_{d}=\bigotimes_{i=1}^{i=L}\left(\mathcal{H}_{i}\otimes\mathcal{H}_{i'}\right)\,,\label{eq:two copies bos notation}
\end{equation}
where $\mathcal{H}_{i}\approx\mathcal{H}_{i'}\approx\mathbb{C}^{d}$.
Note that under the above identification the tensor product of representation
$\Pi_{b}$ takes a convenient form
\begin{equation}
\Pi_{b}\otimes\Pi_{b}(U)=\bigotimes_{i=1}^{i=L}\left(U\otimes U\right)\,.\label{eq:tensor product of represnetations bos}
\end{equation}
Just like in the case of distinguishable particles the commutant of
$\mathrm{SU}\left(N\right)$ in $\mathcal{H}_{b}\otimes\mathcal{H}_{b}$
turns out to be commutative.
\begin{lem}
\label{lem:commutant bosons-}The commutant $\mathrm{Comm}\left(\mathbb{C}\left[\Pi_{b}\otimes\Pi_{b}\left(\mathrm{LU}_{b}\right)\right]\right)$
is commutative. Let $k=0,\ldots,L$ and let $\mathbb{S}_{b}^{k}:\mathcal{H}_{b}\otimes\mathcal{H}_{b}\rightarrow\mathcal{H}_{b}\otimes\mathcal{H}_{b}$
be a Hermitian operator defined by%
\footnote{We use the embedding \eqref{eq:two copies bos notation} and the analogous
to that of Eq. \eqref{eq:complicated swap operator}.%
}
\begin{equation}
\mathbb{S}_{b}^{k}=\mbox{\ensuremath{\mathbb{P}}}^{\mathrm{sym}}\otimes\mbox{\ensuremath{\mathbb{P}}}^{\mathrm{sym}}\left[\left(\bigotimes_{i\in X}\mathbb{S}_{ii'}\right)\otimes\left(\bigotimes_{i\in\bar{X}}\mathbb{I}_{ii'}\right)\right]\mbox{\ensuremath{\mathbb{P}}}^{\mathrm{sym}}\otimes\mbox{\ensuremath{\mathbb{P}}}^{\mathrm{sym}}\,,\label{eq:complicated swap operator bosons}
\end{equation}
where $\mbox{\ensuremath{\mathbb{P}}}^{\mathrm{sym}}:\mathcal{H}^{\otimes L}\rightarrow\mathrm{Sym}^{L}\left(\mathcal{H}\right)$
is an orthonormal projector onto $\mathrm{Sym}^{L}\left(\mathcal{H}\right)$
and $X\subset\left\{ 1,\ldots,L\right\} $ is a subset%
\footnote{The definition of $\mathbb{S}^{m}$ does not depend upon the choice
of $X\subset\left\{ 1,\ldots,L\right\} $ provided it has the cardinality
equal $k$.%
} having the cardinality $m$. Operators $\mathbb{S}_{b}^{m}$ for
$k=0,\ldots,L$ span $\mathrm{Comm}\left(\Pi_{b}\otimes\Pi_{b}\left(\mathrm{LU}_{b}\right)\right)$.\end{lem}
\begin{proof}
The proof is presented in Section \ref{sec:Proofs-concerning-Chapter multilinear witnesses}
of the Appendix (see page \pageref{sub:Proof-of-Lemmas-commutants}).
\end{proof}
Although the formula \eqref{eq:complicated swap operator bosons}
for the operator $\mathbb{S}_{b}^{k}$ may seem to be quite complicated,
the physical interpretation of $\mathbb{S}_{b}^{m}$ is actually very
simple. It is due to the following simple fact.
\begin{fact}
\label{Fact:reductions bosons}(\citep{Fetter2003}) Let $\rho\in\mathcal{D}\left(\mathcal{H}_{b}\right)$
be a state of system consisting of $L$ bosons. Due to the embedding%
\footnote{We have in mind a trivial embedding $\mathrm{Sym}^{L}\mathcal{\left(H\right)}\subset\mathcal{H}^{\otimes L}$.%
} $\mathcal{H}_{b}\subset\mathcal{H}_{d}$ the sate $\rho$ can be
treated as a state of $L$ distinguishable identical particles, $\rho\in\mathcal{D}\left(\mathcal{H}_{d}\right)$.
Let $\rho_{X}$ denote a reduced density matrix describing the subsystem
of $L$-partite system consisting of particles labeled by indices
belonging to the subset $X\subset\left\{ 1,\ldots,L\right\} $. For
all $X\subset\left\{ 1,\ldots,L\right\} $ satisfying $\left|X\right|=k$
the reduced state $\rho_{X}$ is the same and depends only on the
integer $k$, 
\begin{equation}
\rho_{X}=\rho_{\left(k\right)}\,,\label{eq:restrictions bosons}
\end{equation}
where $\rho_{\left(k\right)}\in\mathcal{D}\left(\mathrm{Sym}^{k}\left(\mathcal{H}\right)\right)$.
\end{fact}
Using Facts \ref{fact:encoding of reduction} and \ref{Fact:reductions bosons}
we obtain the following formula,
\begin{equation}
\mathrm{tr}\left(\left[\rho\otimes\sigma\right]\mathbb{S}_{b}^{k}\right)=\mathrm{tr}\left(\rho_{\left(k\right)}\sigma_{\left(k\right)}\right)\,,\label{eq:fromula swap bosons}
\end{equation}
where $\rho,\sigma\in\mathcal{D}\left(\mathcal{H}_{b}\right)$ and
$\rho_{\left(k\right)},\sigma_{\left(k\right)}\in\mathcal{D}\left(\mathrm{Sym}^{k}\left(\mathcal{H}\right)\right)$
are as in \eqref{eq:restrictions bosons}. Equation \eqref{eq:fromula swap bosons}
specifies the operator $\mathbb{S}_{b}^{k}$ uniquely and gives it
a physical interpretation.

We now proceed with the description of the cone $\mathcal{W}_{2}^{\mathrm{LU}_{b}}\left(\mathcal{M}_{b}\right)$
consisting of $\mathrm{LU}_{b}$-invariant bilinear witnesses of particle
entanglement for bosonic systems.
\begin{lem}
\label{LUb invariant bilin theorem}Every $\mathrm{LU}_{b}$-invariant
bilinear particle-entanglement witness $V$ ($V\in\mathcal{W}_{2}^{\mathrm{LU}_{b}}\left(\mathcal{M}_{b}\right)$)
for $L$ bosonic particles has the structure
\begin{equation}
V=\sum_{k=0}^{L}a_{k}\mathbb{S}_{b}^{k}\,,\label{eq:bilin bos entangl structure}
\end{equation}
where $a_{k}\in\mathbb{R}$ ($k=0,\ldots,L$) are real parameters
satisfying the following $L+1$ inequalities 
\begin{equation}
\sum_{k=0}^{n}\frac{\binom{n}{k}}{\binom{L}{k}}a_{k}\leq0\,,\label{eq:inequalities entanglement bosons}
\end{equation}
parametrized by integers $n=0,\ldots,L$.\end{lem}
\begin{proof}
We prove \ref{LUb invariant bilin theorem} in an analogous way as
Theorem \ref{LU invariant bilin theorem}. The detailed proof is given
in Section \ref{sec:Proofs-concerning-Chapter multilinear witnesses}
of the Appendix (see page \pageref{sub:Proof-of-Lemma-Lub ineq}). 
\end{proof}
Just like in the case of distinguishable particles inequalities \eqref{eq:inequalities entanglement bosons}
are independent and it is possible to give a precise structure of
the cone $\mathcal{W}_{2}^{\mathrm{LU}_{b}}\left(\mathcal{M}_{b}\right)$.
\begin{lem}
\label{lem: extrem rays entanglement bos}The following operators
form a generating set of the cone $\mathcal{W}_{2}^{\mathrm{LU_{b}}}\left(\mathcal{M}_{b}\right)$,
\begin{equation}
V_{b}^{m}=-\left(-1\right)^{m}\sum_{k:k\geq m}\binom{L}{k}\binom{k}{m}\left(-1\right)^{k}\mathbb{S}_{b}^{k}\,,\label{eq:formula extremal bos entanglement}
\end{equation}
where $m=0,\ldots,L$. Moreover, half-lines in $\mathrm{Herm}\left(\mathcal{H}_{b}\otimes\mathcal{H}_{b}\right)$
passing through the null operator and one of the operators $V_{b}^{m}$
form $L+1$ extreme rays of the cone $\mathcal{W}_{2}^{\mathrm{LU_{b}}}\left(\mathcal{M}_{b}\right)$.

Consequently, particle entanglement of states $\rho,\sigma\in\mathcal{D}\left(\mathcal{H}_{b}\right)$
can be detected by some $V\in\mathcal{W}_{2}^{\mathrm{LU}_{b}}\left(\mathcal{M}_{b}\right)$
if and only if
\begin{equation}
\underset{0\leq m\leq L}{\mathrm{max}}\mathrm{tr}\left(\left[\rho\otimes\sigma\right]V_{b}^{m}\right)>0\,.\label{eq:ultimate power bos}
\end{equation}
Moreover, entanglement of $\rho\in\mathcal{D}\left(\mathcal{H}_{b}\right)$
can be detected via criteria based on $\mathcal{W}_{2}^{\mathrm{LU}_{b}}\left(\mathcal{M}_{b}\right)$
if and only if
\begin{equation}
\underset{0\leq m\leq L}{\mathrm{max}}\lambda_{\mathrm{max}}\left\{ \mathrm{tr}_{1}\left(\left[\rho\otimes\mathbb{I}\right]V_{b}^{m}\right)\right\} >0\,,\,\label{eq:alternative power bos}
\end{equation}
where the notation in \eqref{eq:alternative power bos} is the same
as in \eqref{eq:final condition streangth}.\end{lem}
\begin{proof}
The proof is presented in Section \ref{sec:Proofs-concerning-Chapter multilinear witnesses}
of the Appendix (see page \pageref{sub:Proof-of-Lemma-extr rays bos}).\end{proof}
\begin{rem*}
From the proof of Lemma \ref{lem: extrem rays entanglement bos} we
get as a corollary that operators $\mathbb{S}_{b}^{k}$, $k=0,\ldots,L$
form the basis of $\mathrm{Comm}\left(\mathbb{C}\left[\Pi_{b}\otimes\Pi_{b}\left(\mathrm{LU}_{b}\right)\right]\right)$.
\end{rem*}
Due to the fact that $L$ boson state $\rho$ is particle-entangled
if and only if it is entangled, when treated as a state of $L$ distinguishable
particles, criteria based on operators $V\in\mathcal{W}_{2}^{\mathrm{LU}}\left(\mathcal{M}_{d}\right)$
can be applied to the bosonic scenario. The following proposition
shows that criteria for detection of bosonic particle entanglement
based on $V\in\mathcal{W}_{2}^{\mathrm{LU}_{b}}\left(\mathcal{M}_{b}\right)$
are always restrictions of criteria based on $V\in\mathcal{W}_{2}^{\mathrm{LU}}\left(\mathcal{M}_{d}\right)$
valid for entanglement of distinguishable particles. 
\begin{prop}
\label{Lem:equivalence bos disting}Extremal rays $\mathcal{W}_{2}^{\mathrm{LU}_{b}}\left(\mathcal{M}_{b}\right)$
(given by Eq.\eqref{eq:formula extremal bos entanglement}) coincide
with extremal rays of $\mathcal{W}_{2}^{\mathrm{LU}}\left(\mathcal{M}_{d}\right)$
(given by Eq. \eqref{eq:formula extremal entanglement}) when the
latter are restricted to the space $\mathcal{H}_{b}\otimes\mathcal{H}_{b}\subset\mathcal{H}_{d}\otimes\mathcal{H}_{d}$.\end{prop}
\begin{proof}
The proof is presented in Section \ref{sec:Proofs-concerning-Chapter multilinear witnesses}
of the Appendix (see page \pageref{sub:Proof-of-Proposition-bosons oper compar}).
\end{proof}
We would like to stress that even though the result of Proposition
\ref{Lem:equivalence bos disting} may seem to be ``intuitive'',
it does not follow imminently from definitions of cones $\mathcal{W}_{2}^{\mathrm{LU}_{b}}\left(\mathcal{M}_{b}\right)$
and $\mathcal{W}_{2}^{\mathrm{LU}}\left(\mathcal{M}_{d}\right)$.
In principle it might have happened that criteria based on $\mathcal{W}_{2}^{\mathrm{LU}_{b}}\left(\mathcal{M}_{b}\right)$
detect entanglement of bosonic states more effectively then simply
restrictions of criteria based on $\mathcal{W}_{2}^{\mathrm{LU}}\left(\mathcal{M}_{d}\right)$.

\subsection{``Entanglement'' of fermions \label{sub:Enatnglement-of-fermions-twolin}}

In this subsection we consider the structure of the cone $\mathcal{W}_{2}^{K}\left(\mathcal{M}\right)$
for the case of ``entanglement'' (correlations that go beyond merely
antisymmetrization) in the system of $L$ fermions. We will adopt
the notation and the terminology from Subsection \ref{sub:Slater-determinants}.
The appropriate class of pure states $\mathcal{M}=\mathcal{M}_{f}$
consists of projectors onto Slater determinants (see Eq.\eqref{eq:ferm slater}).
The relevant symmetry group in this context is the ``fermionic local
unitary group'' $\mathrm{LU}_{f}=\mathrm{SU}\left(N\right)$, irreducibly
represented on a $L$-fermion Hilbert space
\[
\mathcal{H}_{f}=\mathrm{\bigwedge}^{L}\left(\mathcal{H}\right)\,,\,\mathcal{H}\approx\mathbb{C}^{Nd}
\]
via the representation
\[
\Pi_{f}\left(U\right)=\stackrel{L\,\text{times}}{\overbrace{U\otimes\ldots\otimes U}}\,.
\]

It will be convenient for us to use the following embedding of multiple
tensor products 
\begin{equation}
\mathcal{H}_{f}\otimes\mathcal{H}_{f}\subset\mathcal{H}_{d}\otimes\mathcal{H}_{d}=\bigotimes_{i=1}^{i=L}\left(\mathcal{H}_{i}\otimes\mathcal{H}_{i'}\right)\,,\label{eq:two copies ferm notation}
\end{equation}
where $\mathcal{H}_{i}\approx\mathcal{H}_{i'}\approx\mathbb{C}^{d}$.
Note that under the above identification the tensor product of representation
$\Pi_{f}$ takes the convenient form
\begin{equation}
\Pi_{f}\otimes\Pi_{f}(U)=\bigotimes_{i=1}^{i=L}\left(U\otimes U\right)\,.\label{eq:tensor product of represnetations ferm}
\end{equation}

Just like in the cases considered before, the commutant of $\mathrm{SU}\left(N\right)$
in $\mathcal{H}_{f}\otimes\mathcal{H}_{f}$ turns out to be commutative.
\begin{lem}
\label{commutant ferm}The commutant $\mathrm{Comm}\left(\mathbb{C}\left[\Pi_{f}\otimes\Pi_{f}\left(\mathrm{LU}_{f}\right)\right]\right)$
is commutative. Let $k=0,\ldots,L$ and let $\mathbb{S}_{f}^{k}:\mathcal{H}_{f}\otimes\mathcal{H}_{f}\rightarrow\mathcal{H}_{f}\otimes\mathcal{H}_{f}$
be a Hermitian operator defined by%
\footnote{We use the embedding \eqref{eq:two copies ferm notation} and the
analogous to that of Eq. \eqref{eq:complicated swap operator}.%
}
\begin{equation}
\mathbb{S}_{f}^{k}=\mbox{\ensuremath{\mathbb{P}}}^{a\mathrm{sym}}\otimes\mbox{\ensuremath{\mathbb{P}}}^{a\mathrm{sym}}\left[\left(\bigotimes_{i\in X}\mathbb{S}_{ii'}\right)\otimes\left(\bigotimes_{i\in\bar{X}}\mathbb{I}_{ii'}\right)\right]\mbox{\ensuremath{\mathbb{P}}}^{a\mathrm{sym}}\otimes\mbox{\ensuremath{\mathbb{P}}}^{\mathrm{asym}}\,,\label{eq:complicated swap operator fermions}
\end{equation}
where $\mbox{\ensuremath{\mathbb{P}}}^{\mathrm{asym}}:\mathcal{H}^{\otimes L}\rightarrow\mathrm{\bigwedge}^{L}\left(\mathcal{H}\right)$
is an orthonormal projector onto $\bigwedge^{L}\left(\mathcal{H}\right)$
and $X\subset\left\{ 1,\ldots,L\right\} $ is a subset%
\footnote{The definition of $\mathbb{S}_{f}^{m}$ does not depend upon the choice
of $X\subset\left\{ 1,\ldots,L\right\} $ provided it has the cardinality
equal $k$.%
} having the cardinality $m$. Operators $\mathbb{S}_{f}^{m}$ for
$k=0,\ldots,L$ span $\mathrm{Comm}\left(\Pi_{f}\otimes\Pi_{f}\left(\mathrm{LU}_{f}\right)\right)$.\end{lem}
\begin{proof}
The proof is presented alongside the proof of Lemma \ref{lem:commutant bosons-}
in Section \ref{sec:Proofs-concerning-Chapter multilinear witnesses}
of the Appendix (see page \pageref{sub:Proof-of-Lemmas-commutants}).
\end{proof}
Operators $\mathbb{S}_{f}^{k}$, just like their bosonic counterparts
have an appealing physical interpretation. 
\begin{fact}
\label{Fact:reductions fermions}(\citep{Fetter2003}) Let $\rho\in\mathcal{D}\left(\mathcal{H}_{f}\right)$
be a state of the system consisting of $L$ fermions. Due to the embedding%
\footnote{We have in mind a trivial embedding $\bigwedge^{L}\mathcal{\left(H\right)}\subset\mathcal{H}^{\otimes L}$.%
} $\mathcal{H}_{f}\subset\mathcal{H}_{dist}$ the sate $\rho$ can
be treated as a state of $L$ distinguishable identical particles,
$\rho\in\mathcal{D}\left(\mathcal{H}_{dist}\right)$. Let $\rho_{X}$
denote a reduced density matrix describing the subsystem of $L$-particle
system consisting of particles labeled by indices belonging to the
subset $X\subset\left\{ 1,\ldots,L\right\} $. For all $X\subset\left\{ 1,\ldots,L\right\} $
satisfying $\left|X\right|=k$ the reduced state $\rho_{X}$ is the
same and depends only on the integer $k$, 
\begin{equation}
\rho_{X}=\rho_{\left(k\right)}\,,\label{eq:restrictions fermions}
\end{equation}
where $\rho_{\left(k\right)}\in\mathcal{D}\left(\bigwedge^{k}\left(\mathcal{H}\right)\right)$.
\end{fact}
Usisng Facts \ref{fact:encoding of reduction} and \ref{Fact:reductions fermions}
we get,
\begin{equation}
\mathrm{tr}\left(\left[\rho\otimes\sigma\right]\mathbb{S}_{f}^{k}\right)=\mathrm{tr}\left(\rho_{\left(k\right)}\sigma_{\left(k\right)}\right)\,,\label{eq:fromula swap fermions}
\end{equation}
where $\rho,\sigma\in\mathcal{D}\left(\mathcal{H}_{f}\right)$ and
$\rho_{\left(k\right)},\sigma_{\left(k\right)}\in\mathcal{D}\left(\bigwedge^{k}\left(\mathcal{H}\right)\right)$
are as in \eqref{eq:restrictions fermions}. Equation \eqref{eq:fromula swap fermions}
specifies the operator $\mathbb{S}_{f}^{k}$ uniquely and gives it
a physical interpretation.

We now proceed with the description of the cone $\mathcal{W}_{2}^{\mathrm{LU}_{f}}\left(\mathcal{M}_{f}\right)$
consisting of $\mathrm{LU}_{f}$-invariant bilinear witnesses of ``entanglement''
in $L$-fermion systems. In what follows, for the sake of simplicity,
we restrict our considerations to the case $2L\leq d$. The relaxation
of assumption is possible but makes computations cumbersome.
\begin{lem}
\label{LUf invariant bilin theorem}Assume that $2L\leq d$. Then,
every $\mathrm{LU}_{f}$-invariant bilinear correlation $V$ ($V\in\mathcal{W}_{2}^{\mathrm{LU}_{f}}\left(\mathcal{M}_{f}\right)$)
for $L$ fermionic particles has the structure
\begin{equation}
V=\sum_{k=0}^{L}a_{k}\mathbb{S}_{f}^{k}\,,\label{eq:bilin ferm entangl structure}
\end{equation}
where $a_{k}\in\mathbb{R}$ ($k=0,\ldots,L$) are real parameters
satisfying the following $L+1$ inequalities 
\begin{equation}
\sum_{k\leq n}\frac{\binom{n}{k}}{\binom{L}{k}^{2}}a_{k}\leq0\,,\label{eq:inequalities entanglement fermions}
\end{equation}
parametrized by integers $n=0,\ldots,L$.\end{lem}
\begin{proof}
The proof is presented in Section \ref{sec:Proofs-concerning-Chapter multilinear witnesses}
of the Appendix (see page \pageref{LUf invariant bilin theorem}).
\end{proof}
In analogy to the previously considered cases we exploit inequalities
\eqref{eq:inequalities entanglement fermions} to describe the structure
of the cone $\mathcal{W}_{2}^{\mathrm{LU}_{f}}\left(\mathcal{M}_{f}\right)$
\begin{lem}
\label{lem: extrem rays entanglement ferm}Assume that $2L\leq d$
then, the following operators form a generating set of the cone $\mathcal{W}_{2}^{\mathrm{LU_{f}}}\left(\mathcal{M}_{f}\right)$,
\begin{equation}
V_{f}^{m}=-\left(-1\right)^{m}\sum_{k:k\geq m}\binom{L}{k}^{2}\binom{k}{m}\left(-1\right)^{k}\mathbb{S}_{f}^{k}\,,\label{eq:formula extremal ferm entanglement}
\end{equation}
where $m=0,\ldots,L$. Moreover, half-lines in $\mathrm{Herm}\left(\mathcal{H}_{f}\otimes\mathcal{H}_{f}\right)$
passing through the null operator and one of the operators $V_{f}^{m}$
form $L+1$ extreme rays of the cone $\mathcal{W}_{2}^{\mathrm{LU}_{f}}\left(\mathcal{M}_{f}\right)$.

Consequently, particle entanglement of states $\rho,\sigma\in\mathcal{D}\left(\mathcal{H}_{f}\right)$
can be detected by some $V\in\mathcal{W}_{2}^{\mathrm{LU}_{f}}\left(\mathcal{M}_{f}\right)$
if and only if
\begin{equation}
\underset{0\leq m\leq L}{\mathrm{max}}\mathrm{tr}\left(\left[\rho\otimes\sigma\right]V_{f}^{m}\right)>0\,.\label{eq:ultimate powerferm}
\end{equation}
Moreover, entanglement of $\rho\in\mathcal{D}\left(\mathcal{H}_{f}\right)$
can be detected via criteria based on $V\in\mathcal{W}_{2}^{\mathrm{LU}_{f}}\left(\mathcal{M}_{f}\right)$
if and only if
\begin{equation}
\underset{0\leq m\leq L}{\mathrm{max}}\lambda_{\mathrm{max}}\left\{ \mathrm{tr}_{1}\left(\left[\rho\otimes\mathbb{I}\right]V_{f}^{m}\right)\right\} >0\,,\,\label{eq:alternative power ferm}
\end{equation}
where notation in \eqref{eq:alternative power ferm} is the same as
in \eqref{eq:final condition streangth}.\end{lem}
\begin{proof}[Idea of the proof]
The proof of Lemma \ref{lem: extrem rays entanglement ferm} is completely
analogous to the proof of Lemmas \ref{lem: extrem rays entanglement}
and \ref{lem: extrem rays entanglement bos}. Just like in the case
of distinguishable particles and bosons one gets that under the assumption
$2L\leq d$ the operators $\mathbb{S}_{f}^{k}$ ($k=0,\ldots,L$)
form a basis of $\mathrm{Comm}\left(\Pi_{f}\otimes\Pi_{f}\left(\mathrm{LU}_{f}\right)\right)$.
\end{proof}
We would like to remark that, unlike in the case of entanglement of
bosons discussed in Subsection \ref{sub:Enatnglement-of-bosons-twolin},
the correlation criteria based on operators $V_{f}^{m}$ \eqref{eq:formula extremal ferm entanglement}
cannot be deduced from the entanglement criteria derived in Subsection
\ref{sub:Entanglement-of-distinguishable-bilin}. The main reason
for that is simple: non-correlated pure fermionic sates $\mathcal{M}_{f}\subset\mathcal{D}_{1}\left(\mathcal{H}_{f}\right)$
are automatically entangled, when treated as pure states on an ``auxiliary''
Hilbert space of $L$ distinguishable particles $\mathcal{H}_{dist}=\left(\mathbb{C}^{d}\right)^{\otimes L}$.

\subsection{Not convex-Gaussian correlations\label{sub:Not-convex-Gaussian-correlation-bilin}}

In this subsection we describe the cone $\mathcal{W}_{2}^{K}\left(\mathcal{M}\right)$
in the case when the correlations in question are ``not convex-Gaussian''
(see Section \ref{sec:General-motivation} and Subsection \ref{sub:Fermionic-Gaussian-states}).
The relevant class of pure states that defines our notion of correlations
consists of pure fermionic Gaussian states $\mathcal{M}_{g}$ (See
Eq.\eqref{gauss terhal fact}). In what follows we will use extensively
the notation and facts from Subsection \ref{sub:Fermionic-Gaussian-states}.
The relevant symmetry group $K=\mathrm{Spin}\left(2d\right)$ is represented
in a $d$ mode fermionic Fock space $\mathcal{H}_{\mathrm{Fock}}\left(\mathbb{C}^{d}\right)$
via $\Pi_{s}$ (the group of Bogolyubov transformations). Recall that
 pure fermionic Gaussian states does not form a single orbit of$\mathrm{Spin}\left(2d\right)$.
Instead we have
\begin{equation}
\mathcal{M}_{g}=\mathcal{M}_{g}^{+}\cup\mathcal{M}_{g}^{-}\,,\label{eq:split second}
\end{equation}
where $\mathcal{M}_{g}^{\pm}\subset\mathcal{D}_{1}\left(\mathcal{H}_{\mathrm{Fock}}^{\pm}\left(\mathbb{C}^{d}\right)\right)$
are highest-weight orbits of $\mathrm{Spin}\left(2d\right)$ which
is represented, via representations $\Pi_{s}^{\pm}$ in Hilbert spaces
$\mathcal{H}_{\mathrm{Fock}}^{\pm}\left(\mathbb{C}^{d}\right)$ (see
Section \ref{sub:Spinor-represenations-of}). For the sake of clarity
and in order to avoid unnecessary technical complications we will
describe only the case when $\mathcal{M}=\mathcal{M}_{g}^{+}$ and
the Hilbert space $\mathcal{H}_{\mathrm{Fock}}^{+}\left(\mathbb{C}^{d}\right)$.
The case of $\mathcal{M}_{g}^{-}$ can be treated analogously. From
the physical perspective studing solely $\mathcal{M}_{g}^{+}$ and
restricting to $\mathcal{H}_{\mathrm{Fock}}^{+}\left(\mathbb{C}^{d}\right)$
amounts to imposing the parity superselection rule and restricting
considerations to fermionic states characterized by the total party
$Q$ equal one. 

We proceed analogously as in the previous three subsections by first
describing the structure of the commutant $\mathrm{Comm}\left(\Pi_{s}^{+}\otimes\Pi_{s}^{+}\left(\mathrm{Spin}\left(2d\right)\right)\right)$.
Then, with the help of Theorem \ref{finitelly generated cone}, we
describe analytically the cone $\mathcal{W}_{2}^{\mathrm{Spin}\left(2d\right)}\left(\mathcal{M}_{g}^{+}\right)$
consisting of bilinear witnesses of non-Gaussian correlations being
invariant under the group of Bogolyubov transformations in $\mathcal{H}_{\mathrm{Fock}}^{+}\left(\mathbb{C}^{d}\right)$
(realized by the representation $\Pi_{s}^{+}$ of the group $\mathrm{Spin}\left(2d\right)$). 
\begin{lem}
\label{commutant fer gauss}The commutant $\mathrm{Comm}\left(\Pi_{s}^{+}\otimes\Pi_{s}^{+}\left(\mathrm{Spin}\left(2d\right)\right)\right)$
is commutative. The basis of $\mathrm{Comm}\left(\Pi_{s}^{+}\otimes\Pi_{s}^{+}\left(\mathrm{Spin}\left(2d\right)\right)\right)$
consists of operators $\mathbb{C}_{k}$ ($k=0,\ldots\left\lfloor \frac{d}{2}\right\rfloor $)
defined by
\begin{equation}
\mathbb{C}_{k}=\mbox{\ensuremath{\mathbb{P}}}_{+}\left(\sum_{\begin{array}[t]{c}
X\subset\left\{ 1,\ldots,2d\right\} \\
\left|X\right|=2k
\end{array}}\prod_{i\in X}c_{i}\otimes c_{i}\right)\,\mbox{\ensuremath{\mathbb{P}}}_{+},\label{eq:commutant gauss}
\end{equation}
where $\mathbb{P}_{+}=\frac{1}{4}\left(\mathbb{I}+Q\right)\otimes\left(\mathbb{I}+Q\right)$
is the orthonormal projector%
\footnote{In principle we could have written expression for $\mathbb{C}_{k}$
explicitly. However, the resulting expressions are more complitated
and will not be used by us. %
} onto $\mathcal{H}_{\mathrm{Fock}}^{+}\left(\mathbb{C}^{d}\right)\otimes\mathcal{H}_{\mathrm{Fock}}^{+}\left(\mathbb{C}^{d}\right)\subset\mathcal{H}_{\mathrm{Fock}}\left(\mathbb{C}^{d}\right)\otimes\mathcal{H}_{\mathrm{Fock}}\left(\mathbb{C}^{d}\right)$.\end{lem}
\begin{proof}
The proof of the proposition is technical and relies on the structure
of the commutant of the representation of the group $\mathrm{Pin}\left(2d\right)$
in $\mathcal{H}_{\mathrm{Fock}}\left(\mathbb{C}^{d}\right)\otimes\mathcal{H}_{\mathrm{Fock}}\left(\mathbb{C}^{d}\right)$.
The details of the reasoning will be presented in Section \ref{sec:Proofs-concerning-Chapter multilinear witnesses}
of the Appendix (see page \pageref{sub:Proof-of-Lemma comm gauss}).
\end{proof}
Using the above Proposition and Theorem \ref{finitelly generated cone}
we get the following characterization of the cone $\mathcal{W}_{2}^{\mathrm{Spin}\left(2d\right)}\left(\mathcal{M}_{g}^{+}\right)$.
\begin{lem}
\label{FLO invariant bilin theorem}Every $\mathrm{Spin}\left(2d\right)$-invariant
bilinear witness $V$ detecting non-convex-Gaussian correlations in
$\mathcal{D}\left(\mathcal{H}_{\mathrm{Fock}}^{+}\left(\mathbb{C}^{d}\right)\right)$,$V\in\mathcal{W}_{2}^{\mathrm{Spin}\left(2d\right)}\left(\mathcal{M}_{g}^{+}\right)$,
has the structure
\begin{equation}
V=\sum_{k=0}^{\left\lfloor \frac{d}{2}\right\rfloor }a_{k}\mathbb{C}_{k}\,,\label{eq:bilin gauss structure}
\end{equation}
where $a_{k}\in\mathbb{R}$ ($k=0,\ldots,\lfloor \frac{d}{2}$) are real parameters
satisfying the following $\left\lfloor \frac{d}{2}\right\rfloor +1$
inequalities 
\begin{equation}
\sum_{k=0}^{\left\lfloor \frac{d}{2}\right\rfloor }b_{k}^{n}a_{k}\leq0\,,\label{eq:inequalities gauss structure}
\end{equation}
\begin{equation}
b_{k}^{n}=\left(-1\right)^{k}\sum_{m=0}^{\mathrm{min}\left\{ 2n,k\right\} }\left(-2\right)^{m}\binom{d-m}{k-m}\binom{2n}{m}\,,\label{eq:coeffitient}
\end{equation}
parametrized by the integer $n=0,\ldots,\left\lfloor \frac{d}{2}\right\rfloor $.\end{lem}
\begin{proof}
The proof of Lemma \ref{FLO invariant bilin theorem} is technical
and completely analogous to the proof of Lemma \ref{LU invariant bilin theorem}.
We present the detailed reasoning in the Section \ref{sec:Proofs-concerning-Chapter multilinear witnesses}
of the Appendix (see page \pageref{sub:Proof-of-flo invariant}). 
\end{proof}
Unlike to cases previously considered in this section, we were not
able to obtain a closed-form formulas for coefficients $b_{k}^{n}$.
However the numerical investigations carried out for $d=4,\ldots,1000$
show that the $\left(\left\lfloor \frac{d}{2}\right\rfloor +1\right)\times\left(\left\lfloor \frac{d}{2}\right\rfloor +1\right)$
matrix $B=\left[b_{k}^{n}\right]$ is invertible and therefore we
conjecture that results analogous to those given in Lemmas \ref{lem: extrem rays entanglement},
\ref{lem: extrem rays entanglement bos} and \ref{lem: extrem rays entanglement ferm}
hold also for arbitrary number of modes in the context of fermionic
Gaussian states.

\subsection{Discussion\label{sub:Discussion bilin}}

In this part we discuss the results obtained in this Section. We have
described completely the structure of cones of group invariant bilinear
correlation witnesses $\mathcal{W}_{2}^{K}\left(\mathcal{M}\right)$
detecting correlations four classes of correlations: entanglement
of distinguishable particles, particle entanglement of bosons, ``entanglement''
of fermions and non-Gaussian correlations in fermionic systems. In
each of these cases we proved that $\mathcal{W}_{2}^{K}\left(\mathcal{M}\right)$
is a finitely generated cone and we provided the inequalities generating
it (see Theorem \ref{finitelly generated cone} and Lemmas \ref{LU invariant bilin theorem},
\ref{LUb invariant bilin theorem}, \ref{LUf invariant bilin theorem}
and \ref{FLO invariant bilin theorem}). We have also derived the
extremal rays of the cones in question. This allowed us to gauge,
in each of the considered cases, the ultimate strength of the obtained
criteria (they are given in Lemmas%
\footnote{For the case of non-Gaussian fermionic correlations we did not present
an explicit analytical formula. However, the extreme rays of the cone
$\mathcal{W}_{2}^{\mathrm{Spin}\left(2d\right)}\left(\mathcal{M}_{g}^{+}\right)$were
obtained numerically number of modes up to $d=1000$.%
} \ref{lem: extrem rays entanglement}, \ref{lem: extrem rays entanglement bos}
and \ref{lem: extrem rays entanglement ferm}). Let us first remark
that due to equations \eqref{eq:swap expectation value}, \eqref{eq:fromula swap bosons}
and \eqref{eq:matrix element swap ferm} the criteria given by equations
\eqref{eq:formula extremal entanglement}, \eqref{eq:formula extremal bos entanglement}
and \eqref{eq:formula extremal ferm entanglement} can be understood
as the inequalities which are linear in the linear entropy%
\footnote{A linear entropy of a state is defined by $S_{L}\left(\rho\right)=1-\mathrm{tr}\left(\rho^{2}\right)$%
} of reduced density matrices of a state $\rho$ in question and have
to be satisfied if $\rho$ is uncorrelated. Let us remark that it
is possible to consider a smaller cone  $\mathcal{\tilde{W}}_{2}^{K}\left(\mathcal{M}\right)\subset\mathcal{W}_{2}^{K}\left(\mathcal{M}\right)$consisting
of bilinear $K$-invariant entanglement witnesses which are exact
on pure states in the following sense
\[
\mathrm{tr}\left(\left[\kb{\psi}{\psi}^{\otimes2}\right]\mathrm{V}\right)=0\,\Longleftrightarrow\,\kb{\psi}{\psi}\in\mathcal{M}\,.
\]
However, for the sake of simplicity, we did not consider the cone
$\mathcal{\tilde{W}}_{2}^{K}\left(\mathcal{M}\right)$ here. Let us
now present a connection between presented results and the existing
literature. 
\begin{itemize}
\item For the case of entanglement the inequality of the type \eqref{eq:formula extremal entanglement}
was derived in \citep{Horodecki2003}. Similar inequalities were also
used in \citep{Aolita2006,Mintert2007} in the context of founding
lower bound for the convex roof of the generalized concurrence. 
\item Upon completion of the presented work we realized that \citep{Hall2005}
contains a family of criteria equivalent to the ones given by \eqref{eq:formula extremal entanglement}.
In \citep{Hall2005} it was proven that the this criteria are in general
more powerful than the reduction criterion but nevertheless do not
detect PPT entangled states. However, in \eqref{eq:formula extremal entanglement}
there was no discussion of the optimality of the obtained criteria.
The discussion contained in Subsection \ref{sub:general structure bilin}
shows that inequalities \eqref{eq:formula extremal entanglement}
correspond to extreme rays of the cone $\mathcal{W}_{2}^{\mathrm{LU}}\left(\mathcal{M}_{d}\right)$
and therefore can be considered as optimal. 
\item To our knowledge, the criteria for correlations in fermionic systems
analogous to these given in Lemmas \ref{LUf invariant bilin theorem}
and \ref{FLO invariant bilin theorem} have not been presented before.
\end{itemize}
Much of the discussion presented in this section can be transferred
to the case when $\mathcal{H}$ is infinite-dimensional. In particular
correlation criteria derived in Subsections \ref{sub:Entanglement-of-distinguishable-bilin},
\ref{sub:Enatnglement-of-bosons-twolin} and \ref{sub:Enatnglement-of-fermions-twolin}
remain valid when single particle Hilbert spaces are general separable
Hilbert spaces (the setting considered in Section \ref{sec:inf dimension}).

\section{Summary and open problems \label{sub:Discussion-bilin-detec}}

In the present chapter we have developed a collection of multilinear
witnesses detecting entanglement and its generalizations. The research
presented here was built-up on two observations: 
\begin{itemize}
\item A number of classes of pure states $\mathcal{M}\subset\mathcal{D}_{1}\left(\mathcal{H}\right)$
that give rise to physically-interesting notions of correlations can
be characterized by a polynomial condition in the state's density
matrix (the thorough discussion of this construction was given in
Chapter \ref{chap:Multilinear-criteria-for-pure-states});
\item Many types of correlations are invariant under the action of the relevant
symmetry group of the problem.
\end{itemize}
The following list summarizes the main results presented in this chapter.
\begin{itemize}
\item \textbf{Section \ref{sec:Bilinear-correlation-witnesses}}. A general
bilinear witness for detection of correlations defined via the choice
of the class of non-correlated pure states $\mathcal{M}\subset\mathcal{D}_{1}\left(\mathcal{H}\right)$
given as a zero set of a polynomial quadratic is state's density matrix.
The criterion based on a bilinear witness is given in Theorem \ref{thm:Main result bilin}.
The criterion is applied to the problem of detection of the following
classes of correlations: entanglement of distinguishable particles,
particle entanglement of bosons, ``entanglement'' of fermions, non-Gaussian
correlations in fermionic systems.
\item \textbf{Section \ref{sec:Multilinear-correlation-witness}}. A general
$k$-linear witness for detection of correlations defined via the
choice of the class of non-correlated pure states $\mathcal{M}\subset\mathcal{D}_{1}\left(\mathcal{H}\right)$
given as a zero set of a polynomial of degree $k$ in state's density
matrix. The criterion based on a $k$-linear witness is given in Theorem
\ref{main result multilinear witness}. The criterion is applied to
the problem of detection of Genuine Multiparty Entanglement and Schmidt
number of quantum states.
\item \textbf{Section \ref{sec:Optimal--bilinear}}. Description of the
structure of the cone $\mathcal{W}_{2}^{K}\left(\mathcal{M}\right)$
of bilinear group-invariant correlation witness for cases when the
class $\mathcal{M}\subset\mathcal{D}_{1}\left(\mathcal{H}\right)$
is an orbit of the action of a compact simply-connected Lie group
$K$ represented irreducibly on $\mathcal{H}$. The complete description
of the cone $\mathcal{W}_{2}^{K}\left(\mathcal{M}\right)$ of the
corresponding criteria is presented for the following types of correlations:
entanglement of distinguishable particles, particle entanglement of
bosons, ``entanglement'' of fermions, non-Gaussian correlations
in fermionic systems. The criteria obtained in Section \ref{sec:Bilinear-correlation-witnesses}
are, by definition, special cases of criteria based on operators belonging
to $\mathcal{W}_{2}^{K}\left(\mathcal{M}\right)$.
\end{itemize}

\subsubsection*{Open problems}

Below we give a number of open problems related to the content of
this Chapter.
\begin{itemize}
\item Describe the cone of group-invariant $2$-linear correlation witnesses
for classes of coherent states $\mathcal{M}$ relevant for quantum
optics. In this context the appropriate Hilbert space is the bosonic
Fock space (corresponding one mode or multiple modes). The relevant
classes of coherent states are: optical coherent sates, squeezed states
and pure bosonic Gaussian states. Each of these classes is an orbit
of the appropriate symmetry group: the Heisenberg group \citep{Puri2001},
Metaplectic group \citep{Derezinski2011} and affine-Metaplectic group
\citep{Derezinski2011} receptively. The main technical difficulty
that appears while studying representations of these groups in bosonic
Fock space comes from the fact that this Hilbert space is infinite-dimensional.
Consequently, it is necessary to take into account the functional-theoretic
aspects of the considered problem.
\item Describe the cone of group-invariant $k$-linear entanglement witness
for $k>2$. This problem becomes much more difficult than the one
considered in Section \ref{sec:Optimal--bilinear} due to the fact
that the commutant of the action of the relevant groups in $\mathcal{H}^{\otimes k}$
is no longer commutative.
\item Let $\mathcal{M}\subset\mathcal{D}_{1}\left(\mathcal{H}^{\lambda_{0}}\right)$
denote the class of coherent states of a compact simply-connected
Lie group $K$ in the representation $\Pi$ characterized by the highest
weight $\lambda_{0}$. The set of all mixed states $\mathcal{D}\left(\mathcal{H}^{\lambda_{0}}\right)$
stratifies onto orbits of the action of $K$ given by
\[
k.\rho=\Pi\left(k\right)\rho\Pi\left(k\right)^{\dagger}\,,
\]
where $\rho\in\mathcal{D}\left(\mathcal{H}^{\lambda_{0}}\right)$
and $k\in K$. The sets of non-correlated states $\mathcal{M}^{c}$
is the union of the orbits of the group $K$. Therefore it is possible,
in principle, to decide about the correlations in a state $\rho\in\mathcal{D}\left(\mathcal{H}^{\lambda_{0}}\right)$
by looking solely on the polynomial invariants of the action of $K$
in $\rho$ (for the excellent exposition of the theory of group-invariant
polynomials in the context of entanglement see \citep{Vrana2011,Vrana2011a}).
It is now a natural question to ask what relations between values
of polynomial invariants allow to infer about correlations in a state
$\rho$. In Section \ref{sec:Optimal--bilinear} we have solved this
problem completely but only for the invariants of the particular form
$\mathrm{tr}\left(\left[\rho^{\otimes2}\right]\, V\right)$. Obviously
this problem is of the greatest interest in the context of entanglement. 
\end{itemize}

\chapter{Typical properties of correlations \label{chap:Typical-properties-of}}

In Chapter \ref{chap:Multilinear-criteria-for-pure-states} we have
established a polynomial characterization of various kinds of pure
``non-correlated'' states. In Chapter \ref{chap:Polynomial-mixed states}
we derived a variety of polynomial-based criteria for detection of
correlation in mixed states. Recall that we call a mixed state correlated
if and only if it is not possible to express it as a convex combination
of pure states from a distinguished class of ``non-correlated''
pure states $\mathcal{M}\subset\mathcal{D}_{1}\left(\mathcal{H}\right)$
(see Section \ref{sec:General-motivation} for the motivation for
the introduction of correlations defined in this manner). In the current
chapter we will use results obtained in Chapter \ref{chap:Polynomial-mixed states}
to study typical properties of so-defined correlations in the space
of isospectral density matrices. To be more specific we will be interested
in the following question.
\begin{problem}
(Typicality problem) \label{(Typicality-problem)}Let $\mathcal{M}\subset\mathcal{D}_{1}\left(\mathcal{H}\right)$
be a class of ``non-correlated'' pure states in $N$ dimensional
Hilbert space $\mathcal{H}$. What fraction $\eta_{\Omega}^{\mathrm{corr}}$
of mixed states $\rho\in\mathcal{D}\left(\mathcal{H}\right)$ belonging
to the set ofmatrices with fixed spectrum is correlated (does not
belong to the convex hull $\mathcal{M}^{c}\subset\mathcal{D}\left(\mathcal{H}\right)$)? 
\end{problem}
Figure \ref{fig:typicality question} presents a graphical illustration
to the above question.

\begin{figure}[h]
\centering{}\includegraphics[width=8.5cm]{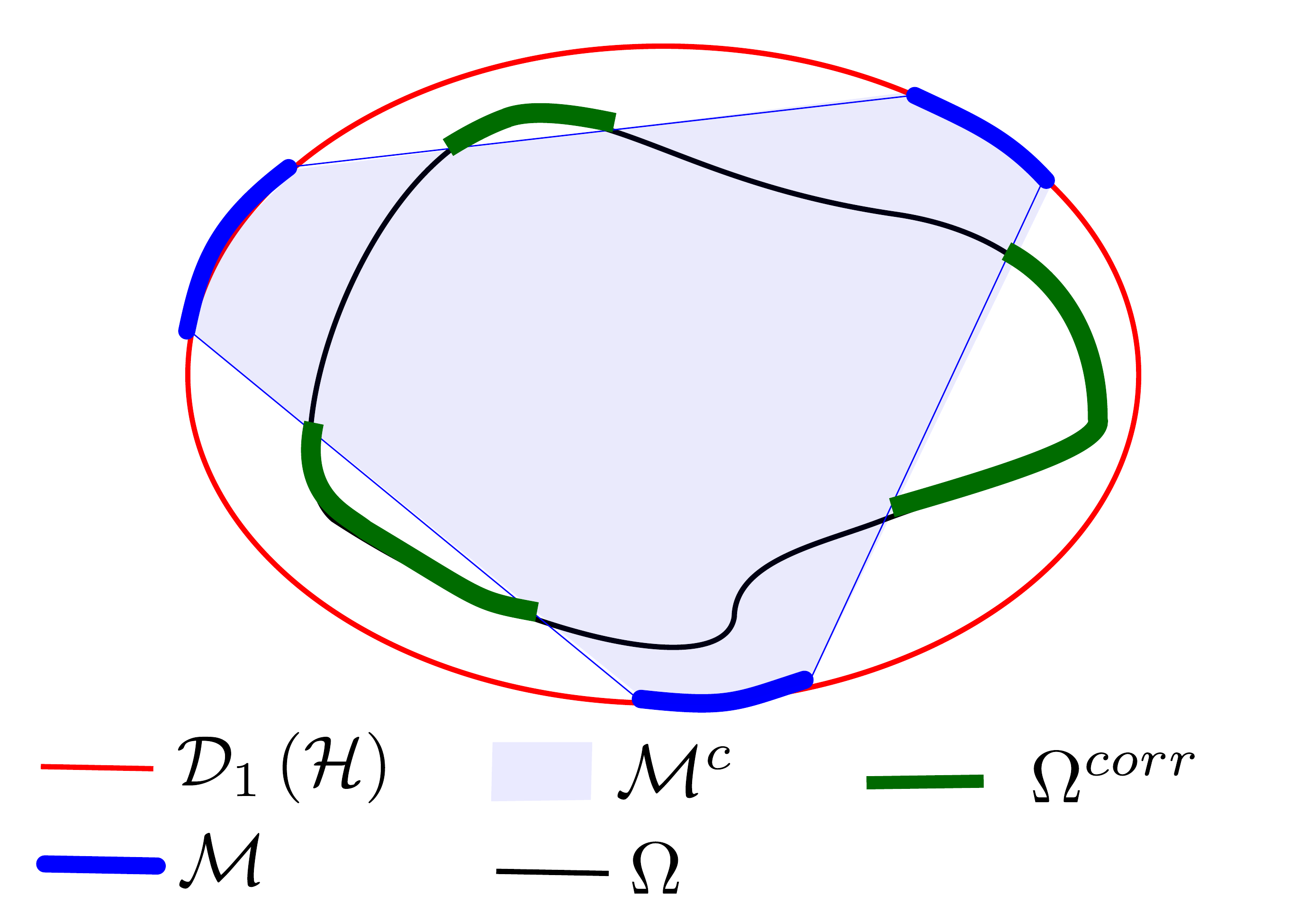}\protect\caption{\label{fig:typicality question}\textcolor{black}{Illustra}tion of
the geometry of correlated states in the space of density matrices
$\mathcal{D\left(\mathcal{H}\right)}$. The space $\mathcal{D\left(\mathcal{H}\right)}$
is located inside the region bounded by the red line. For simplicity
of presentation we identified the boundary of $\mathcal{D\left(\mathcal{H}\right)}$
with space of pure states $\mathcal{D}_{1}\left(\mathcal{H}\right)$.
The class of pure states $\mathcal{M}$ is given by the thick blue
segments laying on $\mathcal{D}_{1}\left(\mathcal{H}\right)$. The
class of non-correlated states, $\mathcal{M}^{c}$, is marked by a
shaded blue region. We ask what fraction of states belonging to the
space of isospectral density matrix $\Omega$ (marked by a black loop)
is correlated (does not belong to $\mathcal{M}^{c}$). }
\end{figure}
Before moving to a more rigorous formulation of the above problem
let us first give a motivation for our considerations.
\begin{itemize}
\item In many cases it is difficult to infer (due to the lack of knowledge
about the state or the experimental difficulties) about the correlations
for a particular state of interest. However, answer to the question
posed in the Problem \ref{(Typicality-problem)} may allow to infer
about the correlation properties of a given state ``with high probability''
(see discussion bellow), assuming that we have the knowledge only
about its spectrum.
\item Consider a bipartite state $\rho\in\mathcal{D}\left(\mathcal{H}_{A}\otimes\mathcal{H}_{B}\right)$.
It is known that there exists a dichotomy between entanglement and
the purity $\mathrm{tr}\left(\rho^{2}\right)$ of $\rho$ \citep{Zyczkowski1998,Zyczkowski1999}.
That is, the smaller the purity of $\rho$, the more likely $\rho$
is to be separable. In particular there exist a critical purity $P_{0}$
such that if $\mathrm{tr}\left(\rho^{2}\right)\leq P_{0}$ the state
$\rho$ is guarantied to be separable. The Problem \ref{(Typicality-problem)}
can be considered as a generalization of studies pioneered in \citep{Zyczkowski1998,Zyczkowski1999}
to a much more general context.
\item Solution to Problem \ref{(Typicality-problem)} gives an insight into
the geometry of the inclusion $\mathcal{M}^{c}\subset\mathcal{D}\left(\mathcal{H}\right)$.
\item Similar problem was considered before but only in the context of quantum
entanglement and for different ensembles of quantum states. Different
ensembles of quantum states and the fraction of correlated states
in the whole set of density matrices $\mathcal{D}\left(\mathcal{H}\right)$
(equipped with a suitable measure) were studied \citep{Szarek2011,Aubrun2012}
(see Section \ref{sec:Introduction-to-concentration} below).
\end{itemize}
In the rest of this chapter we will present results that give a partial
solution to the Problem \ref{(Typicality-problem)} in cases when
the class of ``non-correlated'' pure states can be defined via the
polynomial condition defined in terms of a single operator $\mathrm{A}\in\mathrm{Herm}_{+}\left(\mathrm{Sym}^{k}\left(\mathcal{H}\right)\right)$
(see Eq.\eqref{eq:characterisation multilin}). We will consider the
uniform measure (induced from the Haar measure on the total unitary
group $\mathrm{U}\left(\mathcal{H}\right)$) on the set of isospectral
density matrices $\Omega_{\left(p_{1},\ldots,p_{N}\right)}\subset\mathcal{D}\left(\mathcal{H}\right)$,
i.e. the set of states that posses the same (ordered) spectrum,
\begin{equation}
\Omega_{\left\{ p_{1},\ldots,p_{N}\right\} }=\left\{ \rho\in\mathcal{D}\left(\mathcal{H}\right)\left|\mathrm{Sp}\left(\rho\right)=\left(p_{1},\ldots,p_{N}\right)\right.\right\} ,\label{eq:set of isospectral density matrices}
\end{equation}
where $\mathrm{Sp}\left(\rho\right)$ denotes the ordered spectrum
of $\rho$, i.e. $p_{1}\geq p_{2}\geq\ldots\geq p_{N}\geq0$.

In what follows, unless it causes a confusion, we will write $\Omega_{\left\{ p_{1},\ldots,p_{N}\right\} }\equiv\Omega$
assuming that we deal with the  spectrum $\left(p_{1},\ldots,p_{N}\right)$.
We will be mostly concerned with the ``thermodynamic limit'' i.e.
a limit when a certain relevant parameter of a system (e.g., number
of particles, number of modes etc.) goes to infinity and consequently
$N\rightarrow\infty$. For this reason we will use the technique of
measure concentration. The phenomenon of measure concentration is
a powerful tool that allows to obtain strong large deviation bounds
for random variables on high-dimensional probability spaces. 

The rest the chapter is organized as follows. In Section \ref{sec:Introduction-to-concentration}
we present elements of the theory of measure concentration and give
a brief survey of the application of this technique in the field of
quantum information. In Section \ref{sec:Introduction-to-concentration}
we apply the concentration of measure techniques and results of Section
\ref{sec:Proofs-concerning-Chapter multilinear witnesses} to derive
lower bounds for $\eta_{\Omega}^{\mathrm{corr}}$.  The chapter concludes
in Section \ref{sec:Discussion-typicality} where we discuss the obtained
results and present some open problems. 

Results presented in Section \ref{sec:Introduction-to-concentration}
have been partially published in \citep{Oszmaniec2014}. The other
results presented here will constitute the forthcoming publication

\noindent \begin{center}
\begin{minipage}[t]{0.8\columnwidth}%
\noindent \begin{flushleft}
\textit{``Typical properties of correlations on manifolds of isospectral
density matrices'',} Michał Oszmaniec, et al. (in preparation),
\par\end{flushleft}%
\end{minipage}
\par\end{center}

\section{Introduction to concentration of measure phenomenon\label{sec:Introduction-to-concentration}}

The phenomenon of measure concentration occurs when values of Lipschitz
functions%
\footnote{A function $f:X\rightarrow\mathbb{R}$ on a metric space $X$ is called
Lipschitz if there exist a constant $L>0$ (called a Lipschitz constant)
such that for all $x,y\in X$ $\left|f\left(x\right)-f\left(y\right)\right|\leq L\cdot d\left(x,y\right)$,
where $d\left(\cdot,\cdot\right)$ is a metric on $X$. %
} defined on some probabilistic space concentrate around the value
of their mean (or median). To put it more precisely, measure of the
set on which a function differs from its mean value in a significant
way is small. Physically speaking, if measure concentration occurs
on the space $X$ (for some physical model), then in many practical
applications it is enough to work with mean values rather than with
functions themselves. Concentration of measure phenomenon takes place
on many spaces naturally occurring in quantum mechanics (spaces of
pure states $\mathcal{D}_{1}\left(\mathcal{H}\right)$), manifolds
of isospectral density matrices $\Omega_{\left(p_{1},\ldots,p_{N}\right)}$).
This is the reason for which the theory of measure concentration finds
many applications in the quantum information theory and in foundations
of quantum statistical mechanics. 

The concept of measure concentration has circulated in the area of
mathematics for a long time. In the year 1934 it was proved by Lévy
that the so-called Gaussian concentration of measure takes place on
the $n$ dimensional sphere $\mathbb{S}^{n}$ equipped with the uniform
measure \citep{Hayden2006}. More precisely
\begin{equation}
\mu\left(\left\{ x\in\mathbb{S}^{n}|\, f-\mathbb{E}f\gtrless\pm\epsilon\right\} \right)\leq2\cdot\mathrm{exp}\left(-\frac{C\left(n+1\right)\epsilon^{2}}{L^{2}}\right),\,\label{eq:sphere}
\end{equation}
where $\mu$ is the standard normalized measure on $\mathbb{S}^{n}$,
$C=\frac{1}{9\pi^{3}\mathrm{ln\left(2\right)}}$ and $\mathbb{E}f$
is the expectation value of the function $f$. The following example
gives perhaps the most intuitive account of the idea of measure concentration.
As a Lipschitz function $f(x)$ on $\mathbb{S}^{n}$ let us take the
coordinate of a given point $x$ (understood as an element of $\mathbb{R}^{n+1}$)
on the axis perpendicular to some arbitrarily chosen equator and having
the beginning in the center of a sphere. Average value of $f(x)$
is $0$ and because of the measure concentration the bigger the dimension
$n$, the smaller the measure of the points that are far away from
the equator. In this sense we see that $n$ sphere becomes ``flatter''
as $n\rightarrow\infty$. 

Before we introduce the concept of measure concentration more formally
we first give a short survey of its applications in quantum information
theory.

\subsubsection*{Application of concentration of measure in quantum information theory}

In the context of quantum information and the theory of entanglement
typical spaces where measure concentration is studied are the set
of all pure states $\mathcal{D}_{1}\left(\mathcal{H}\right)$ (identified
with the complex projective space $\mathbb{P}(\mathcal{H})$) and
the set consisting of$k$ dimensional (complex) subspaces of $\mathcal{H}$
(identified with complex Grassmanians $Gr_{k}(\mathcal{H})$) \citep{Hayden2006}.
In both cases natural measures coming from the Haar measure on $U(\mathcal{H})$
are considered. Statistical ensembles obtained in this way describe
the ``uniform distribution'' of pure states of a quantum system
or, respectively, the ``uniform distribution'' of $k$-dimensional
subspaces of $\mathcal{H}$. The most fundamental work in this context
is \citep{Hayden2006}. In this paper concentration of measure theory
is used to show, among other things, that in the case of a system
consisting of two subsystems, $\mathcal{H}=\mathcal{\mathbb{C}}^{d_{1}}\otimes\mathbb{C}^{d_{2}}$,
``random'' states and all states belonging to a ``random'' $k$-dimensional
subspace $\mathcal{H}_{k}\subset\mathcal{\mathbb{C}}^{d_{1}}\otimes\mathbb{C}^{d_{2}}$
are (for large enough $k$) typically maximally entangled. Another
important application of the ideas of measure concentration is the
Hastings counterexample of the minimal output entropy conjecture.
In \citep{Hastings2009} it was shown that the minimal output entropy%
\footnote{The minimal output entropy $H^{\mathrm{min}}\left(\Lambda\right)$
is defined by $H^{\mathrm{min}}\left(\Lambda\right)=\min_{\kb{\psi}{\psi}\in\mathcal{D}_{1}\left(\mathcal{H}\right)}H\left(\Lambda\left[\rho\right]\right)$,
where where $H\left(\cdot\right)$ denotes the von-Neumann entropy
of a quantum state. %
} $H^{\mathrm{min}}\left(\Lambda\right)$  is not additive, i.e. there
exist quantum channels $\Lambda_{1}\in\mathcal{CP}_{0}\left(\mathcal{H}_{1}\right)$,
$\Lambda_{2}\in\mathcal{CP}_{0}\left(\mathcal{H}_{2}\right)$ such
that
\[
H^{\mathrm{min}}\left(\Lambda_{1}\otimes\Lambda_{2}\right)<H^{\mathrm{min}}\left(\Lambda_{1}\right)+H^{\mathrm{min}}\left(\Lambda_{2}\right)\,.
\]
It is perhaps worth mentioning that tools of measure concentration
are also used extensively in classical information theory (for instance
in proofs of Shannon coding theorems \citep{NielsenChaung2010}).

Another important direction of research, closely related the results
presented in next section, are studies of typical properties of entanglement
with respect to a suitable probability measure (ensemble of states)
on the space of mixed states $\mathcal{D}\left(\mathcal{H}_{d}\right)$
of a given system of distinguishable particles. In the two pioneering
works \citep{Zyczkowski1998,Zyczkowski1999} authors showed, among
other things that there exist a ball (with respect to the Hilbert-Schmidt
norm) around the maximally mixed state $\frac{\mathbb{I}}{N}\in\mathcal{D}\left(\mathcal{H}_{d}\right)$
consisting solely of separable states and estimated (numerically and
analytically) in some cases fractions (calculated with respect to
suitable natural measures on $\mathcal{D}\left(\mathcal{H}_{d}\right)$)
of entangled and respectively separable states in $\mathcal{D}\left(\mathcal{H}_{d}\right)$. 

This line of research was then taken up in many works. Let us just
mention two notable examples. In \citep{Aubrun2012} the authors considered
the measure $\mu_{s}$ on the set of bipartite states $\mathcal{D}\left(\mathbb{C}^{d}\otimes\mathbb{C}^{d}\right)$
steaming from the uniform ensemble of pure states $\mathcal{D}_{1}\left(\mathbb{C}^{d}\otimes\mathbb{C}^{d}\otimes\mathbb{C}^{s}\right)$
traced over the ``environment'' register $\mathbb{C}^{s}$. It turns
out that in the limit $d\rightarrow\infty$ there exist a threshold
dimension $s_{0}$ such that whenever $s<s_{0}$ the states are typically
entangled and if $s>s_{0}$ they are typically separable (the analogous
threshold exist for the NPPT property). Another interesting study
\citep{Szarek2011} investigates the fraction of $k$-entangled%
\footnote{For the definition of $k$-entangled states see Subsection \ref{sub:multilin states with schmidt rank} %
} states in $\mathcal{D}\left(\mathbb{C}^{d}\otimes\mathbb{C}^{d}\right)$
with respect to the measure on $\mathcal{D}\left(\mathbb{C}^{d}\otimes\mathbb{C}^{d}\right)$
induced by the Hilbert-Schmidt distance.

\subsubsection*{Technical aspects of concentration of measure}

We will now introduce the concept of Gaussian concentration formally.
At the end of this part, in Fact \ref{concentration of measure unitary group},
we state large deviation inequalities valid for the special unitary
group $\mathrm{SU}\left(\mathcal{H}\right)$. These inequalities will
be our main technical tool in the next section.
\begin{defn}
\label{gaussian concentration of measure}We will now Let\textsl{
}$\left(X,\, d\right)$ be a compact metric space equipped with a
compatible structure of the probabilistic space $(X,\,\Sigma,\,\mu)$.
Let $f:X\rightarrow\mathbb{R}$ be a Lipschitz function with the mean
$\mathbb{E}_{\mu}f=\int_{X}f\left(x\right)d\mu\left(x\right)$ and
the Lipschitz constant $L$. We say that on the space $X$ we have
the Gaussian concentration if for each Lipschitz $f$ we have, 
\begin{gather}
\mu\left(\left\{ x\in X\left|\, f\left(x\right)-\mathbb{E}_{\mu}f\geq\epsilon\right.\right\} \right)\leq C\cdot\mathrm{exp}\left(-\frac{c\epsilon^{2}}{L^{2}}\right)\,,\label{eq:gaussian concentration}\\
\mu\left(\left\{ x\in X|\left|\, f\left(x\right)-\mathbb{E}_{\mu}f\leq-\epsilon\right.\right\} \right)\leq C\cdot\mathrm{exp}\left(-\frac{c\epsilon^{2}}{L^{2}}\right)\,,
\end{gather}
for arbitrary $\epsilon\geq0$ and some positive constants $c,C>0$.
\end{defn}
The smaller the expressions (for fixed $\epsilon$) on the right hand
side and \eqref{eq:gaussian concentration}, the more robust the concentration
phenomenon on $X$. The phenomenon of Gaussian measure concentration
automatically occurs on compact Riemannian manifolds and controlled
by their geometry (c.f. Subsection \ref{sub:Elements-of-Riemannian}
for introduction of Riemannian geometry).  Let $\left(\mathcal{M},\, g\right)$
be a compact and connected Riemannian manifold with the metric tensor
$g$ normalized in such a way that the measure $\mu$ it induces is
probabilistic (see Eq.\eqref{eq:rescalling measure}). On the manifold
$\mathcal{M}$ we have also the Ricci curvature tensor%
\footnote{A Ricci tensor $\mathcal{R}$ is a symmetric tensor of type $\left(0,2\right)$
that can be defined with the use of Levi-Civita connection on $\mathcal{M}$.
For the precise definition of $\mathcal{R}$ see \citep{DiffGeomPhys}.%
} $\mathcal{R}$ compatible with the unique Levi-Civita connection
on $\mathcal{M}$. Assume that $\mathcal{M}$ has a positive Ricci
curvature  and let
\begin{equation}
c=\min_{p\in\mathcal{M}}\min_{\begin{array}[t]{c}
v\in T_{p}\mathcal{M}\\
g\left(v,v\right)=1
\end{array}}\mathcal{R}\left(v,v\right)\,.\label{eq:lower bound ricc}
\end{equation}
The following fact guarantees that for a Riemannian manifold $\mathcal{M}$
with positive Ricci curvature we automatically have the Gaussian concentration.
\begin{fact}
(Ledoux \citep{Ledoux2005}) Let $f\in\mathcal{F}\left(\mathcal{M}\right)$
be a smooth function on $\mathcal{M}$ with a mean $\mathbb{E}_{\mu}f$.
Let $L$ be a Lipschitz constant%
\footnote{For smooth functions on $\mathcal{M}$ the Lipschitz constant is given
by $L=\sqrt{\max_{p\in\mathcal{M}}g\left(\nabla f,\nabla f\right)}$. %
}. of $f$.  Assume that $c$ from Eq.\eqref{eq:lower bound ricc}
is positive. Then, the following inequalities hold, 
\begin{gather}
\mu\left(\left\{ x\in\mathcal{M}\left|\, f\left(x\right)-\mathbb{E}_{\mu}f\geq\epsilon\right.\right\} \right)\leq\mathrm{exp}\left(-\frac{c\epsilon^{2}}{2L^{2}}\right)\,,\\
\mu\left(\left\{ x\in\mathcal{M}\left|\, f\left(x\right)-\mathbb{E}_{\mu}f\leq-\epsilon\right.\right\} \right)\leq\mathrm{exp}\left(-\frac{c\epsilon^{2}}{2L^{2}}\right)\,,\label{eq:riemmann concentration}
\end{gather}
where $c>0$ is given by Eq.\eqref{eq:lower bound ricc}.
\end{fact}
In what follows we will make an extensive use of the concentration
of measure phenomenon on the special unitary group $\mathrm{SU}\left(\mathcal{H}\right)$as
we will be interested in large deviation bounds for functions defined
on this group (see the next Section). It will be convention for us
to use a metric tensor $g_{\mathrm{HS}}$ induced on $\mathrm{SU}\left(\mathcal{H}\right)$
from the embedding of $SU\left(\mathcal{H}\right)$ in the set of
all linear operators, $End\left(\mathcal{H}\right)$, equipped with
the Hilbert-Schmidt inner product $\left\langle A,\, B\right\rangle =\mathrm{tr}\left(A^{\dagger}B\right)$.
Let us write the formula for $g_{\mathrm{HS}}$ explicitly. Special
unitary group is a Lie group and thus for every $U\in\mathrm{SU}\left(\mathcal{H}\right)$we
have an isomorphism $T_{U}\mathrm{SU}\left(\mathcal{H}\right)\approx\mathfrak{su}\left(\mathcal{H}\right)$,
where $\mathfrak{su}\left(\mathcal{H}\right)$ is the Lie algebra
of the group consisting of anti-Hermitian traceless operators on $\mathcal{H}$
(see Subsection \ref{sub:Lie-groups-and}). The linear isomorphism
is given by the following mapping
\begin{equation}
\mathfrak{su}\left(\mathcal{H}\right)\ni X\rightarrow\hat{X}=\left.\frac{d}{dt}\right|_{t=0}\exp\left(tX\right)U\in T_{U}\mathrm{SU}\left(\mathcal{H}\right)\,.\label{eq:identification special unitary}
\end{equation}
Using the identification \eqref{eq:identification special unitary}
and treating an operator $\hat{X}=\left.\frac{d}{dt}\right|_{t=0}\exp\left(tX\right)U=XU$
as an element of $End\left(\mathcal{H}\right)$ we get
\begin{equation}
g_{\mathrm{HS}}\left(\hat{X},\hat{Y}\right)=\left\langle XU,\, YU\right\rangle =\mathrm{tr}\left(\left[XU\right]^{\dagger}YU\right)=-\mathrm{tr}\left(XY\right)\,,\label{eq:explicit formula metric unitary}
\end{equation}
where we have used $X^{\dagger}=-X$, identity $UU^{\dagger}=\mathbb{I}$
and the cyclicality of the trace.
\begin{fact}
(Gromov, \citep{Anderson2010}) Consider a special unitary group $\mathrm{SU}\left(\mathcal{H}\right)$
equipped with the  metric $g_{\mathrm{HS}}$. The following equality
holds
\begin{equation}
\min_{U\in\mathrm{SU}\left(\mathcal{H}\right)}\min_{\begin{array}[t]{c}
v\in T_{U}\mathcal{\mathrm{SU}\left(\mathcal{H}\right)}\\
g_{\mathrm{HS}}\left(v,v\right)=1
\end{array}}\mathcal{R}\left(v,v\right)=\frac{N}{2}\,,\label{eq:ricci special unitary}
\end{equation}
where $N=\mathrm{dim}\mathcal{H}$. 
\end{fact}
As the inner product $g_{\mathrm{HS}}$ induces on $\mathrm{SU}\left(\mathcal{H}\right)$
the usual Haar measure $\mu$ we get the following concentration of
measure inequality for smooth functions defined on the special unitary
group.
\begin{fact}
(\citep{Anderson2010}) \label{concentration of measure unitary group}Consider
a special unitary group $\mathrm{SU}\left(\mathcal{H}\right)$ equipped
with the Haar measure $\mu$ and the metric $g_{\mathrm{HS}}$. Let
$f:\mathrm{SU}\left(\mathcal{H}\right)\rightarrow\mathbb{R}$ be a
smooth function on $\mathrm{SU}\left(\mathcal{H}\right)$ with the
mean $\mathbb{E}_{\mu}f$ and let%
\footnote{The gradient $\nabla f$ is also computed with respect to the metric
$g_{\mathrm{HS}}$.%
} 
\begin{equation}
L=\sqrt{\max_{p\in\mathcal{\mathrm{SU}\left(\mathcal{H}\right)}}g_{\mathrm{HS}}\left(\nabla f,\nabla f\right)}\label{eq:lipschitz constant formula}
\end{equation}
be a Lipschitz constant of $f$. Then, the following concentration
inequalities hold
\begin{gather}
\mu\left(\left\{ U\in\mathrm{SU}\left(\mathcal{H}\right)\left|\, f\left(U\right)-\mathbb{E}_{\mu}f\geq\epsilon\right.\right\} \right)\leq\mathrm{exp}\left(-\frac{N\epsilon^{2}}{4L^{2}}\right)\,,\\
\mu\left(\left\{ U\in\mathrm{SU}\left(\mathcal{H}\right)\left|\, f\left(U\right)-\mathbb{E}_{\mu}f\leq-\epsilon\right.\right\} \right)\leq\mathrm{exp}\left(-\frac{N\epsilon^{2}}{4L^{2}}\right)\,,\label{eq:su concentration}
\end{gather}
where $N=\mathrm{dim}\mathcal{H}$.
\end{fact}

\section{Typical properties of correlations on sets of isospectral density
matrices\label{sec:whole spase ensambles}}

In this section we apply criteria for detection of correlations derived
in Chapter \ref{sec:Proofs-concerning-Chapter multilinear witnesses}
to study the fraction of correlated states, $\eta_{\Omega}^{\mathrm{corr}}$,
on the manifolds of isospectral density matrices $\Omega\subset\mathcal{D}\left(\mathcal{H}\right)$.
We will study (just like in the rest of this thesis) the notion of
correlations defined via the choice of the class of non-correlated
pure states $\mathcal{M}\subset\mathcal{D}_{1}\left(\mathcal{M}\right)$:
the set correlated states (with respect to the choice of $\mathcal{M}$)
consists of states that do not belong to the convex hull $\mathcal{M}^{c}$
of the class of non-correlated pure states (see Figures \ref{fig:Pictorial-representation-of}
and \ref{fig:typicality question}). Moreover we will focus on the
situation when $\mathcal{M}$ can be defined via a polynomial condition
(see c.f. Chapters \ref{chap:Multilinear-criteria-for-pure-states}-\ref{chap:Polynomial-mixed states})
\begin{equation}
\kb{\psi}{\psi}\in\mathcal{M}\,\Longleftrightarrow\mathrm{tr}\left(A\left[\kb{\psi}{\psi}^{\otimes k}\right]\right)=0\,,\label{eq:polynomial characterisation again}
\end{equation}
where $A\in\mathrm{Herm_{+}}\left(\mathrm{Sym}^{k}\left(\mathcal{H}\right)\right)$
and $A\leq\mathbb{P}^{\mathrm{k,sym}}$. The section is structured
as follows. First we discuss in more detail the structure of the manifold
of isospectral density matrices $\Omega\subset\mathcal{D}\left(\mathcal{H}\right)$
and sketch the idea of our reasoning. Then, we give, with the use
of bilinear criteria derived in Chapter \ref{chap:Polynomial-mixed states},
lower bounds for $\eta_{\Omega}^{\mathrm{corr}}$ for four topical
classes of correlations considered in this thesis: entanglement, entanglement
of bosons, ``entanglement'' of fermions, and ``non-Gaussian''
correlations in fermionic systems. In the last part of the section,
using criteria derived in Section \ref{sec:Multilinear-correlation-witness},
we obtain lower bound for $\eta_{\Omega}^{\mathrm{corr}}$ for the
general case when the class $\mathcal{M}$ is given by a general polynomial
condition \eqref{eq:polynomial characterisation again}. We apply
these bound to two classes of correlations: Genuine Multiparty Entanglement
(see Subsections \ref{sub:GME} and \ref{sub:GME multiparty witness})
and Schmidt number of bipartite mixed states (see Subsections \ref{sub: Schmid rank}
and \ref{sub:multilin states with schmidt rank}).

\subsection{Geometry of $\Omega$ and the reduction to the concentration of measure
on $\mathrm{SU}\left(\mathcal{H}\right)$\label{sub:Geometry-of-omega}}

Let us start with an observation that for each $\Omega_{\left(p_{1},\ldots,p_{N}\right)}$
is actually an orbit%
\footnote{In the language of differential geometry orbits of a Lie group are
usually called homogenous spaces. For an introduction to the beautiful
and rich theory of homogenous spaces see \citep{Arvanitogeorgos2003}.%
} of the action of $\mathrm{SU}\left(\mathcal{H}\right)$ in $\mathcal{D}\left(\mathcal{H}\right)$.
The group of special unitary matrices $\mathrm{SU}\left(\mathcal{H}\right)$
acts on $\mathcal{D}\left(\mathcal{H}\right)$ ``from the left''
via conjugation, i.e. for each $U\in\mathrm{SU}\left(\mathcal{H}\right)$
we have a mapping 
\begin{equation}
\mathcal{D}\left(\mathcal{H}\right)\ni\rho\longrightarrow U.\rho=U\rho U^{\dagger}\in\mathcal{D}\left(\mathcal{H}\right)\,.\label{eq:action unitary}
\end{equation}
which satisfies $U_{1}.\left(U_{2}.\rho\right)=U_{1}U_{2}.\rho$ for
$\rho\in\mathcal{D}\left(\mathcal{H}\right)$ and $U_{1},U_{2}\in\mathrm{SU}\left(\mathcal{H}\right)$.
The action \eqref{eq:action unitary} preserves each set of isospectral
density matrices $\Omega_{\left(p_{1},\ldots,p_{N}\right)}$. Moreover,
the action of $\mathrm{SU}\left(\mathcal{H}\right)$ on $\Omega_{\left(p_{1},\ldots,p_{N}\right)}$
is transitive%
\footnote{This follows straightforwardly from the fact that every Hermitian
matrix can be diagonalized via conjugation by a element of $\mathrm{SU}\left(\mathcal{H}\right)$.
Permutation of the spectrum of the diagonal matrix can be also performed
via conjugation with an element of $\mathrm{SU}\left(\mathcal{H}\right)$.%
}, i.e. given two states $\rho_{1},\rho_{2}\in\Omega_{\left(p_{1},\ldots,p_{N}\right)}$
there exist a unitary operator $\in\mathrm{SU}\left(\mathcal{H}\right)$
such that $\rho_{2}=U_{0}.\rho_{1}$. Let us assume that we have fixed
the a particular $\Omega_{\left(p_{1},\ldots,p_{N}\right)}\equiv\Omega$.
Let us take an arbitrary $\rho_{0}\in\Omega$ and consider a mapping
\[
\varphi_{\rho_{0}}:\mathrm{SU}\left(\mathcal{H}\right)\rightarrow\Omega_{\left(p_{1},\ldots,p_{N}\right)},\,\varphi_{\rho_{0}}\left(U\right)=U.\rho_{0}=U\rho_{0}U^{\dagger}\,.
\]
We define a (normalized) measure $\mu_{\Omega}$ on $\Omega$ via
the pullback, via the mapping $\varphi_{\rho_{0}}$, of the Haar measure
$\mu$. That is, for every Borel subset ()%
\footnote{We consider on $\Omega$ the topology induced from the inclusion $\Omega\subset\mathrm{Herm}\left(\mathcal{H}\right)$,
where the latter is treated as an Euclidean space (with respect to
the Hilbert-Schmidt distance). For the comprehensive introduction
to measure theory see \citep{Reed1972}.%
} $\mathcal{O}\subset\Omega$ we define its measure via 
\begin{equation}
\mu_{\Omega}\left(\mathcal{O}\right)=\mu\left(\varphi_{\rho_{0}}^{-1}\left(\mathcal{O}\right)\right)\,.\label{eq:definition of measure on Omega}
\end{equation}
The invariance properties of $\mu$ guarantee that the measure $\mu_{\Omega}$
does not depend on the choice of $\rho_{0}$. The measure $\mu_{\Omega}$
is invariant under the action of the unitary group, i.e.
\begin{equation}
\mu_{\Omega}\left(U.\mathcal{O}\right)=\mu_{\Omega}\left(O\right)\label{eq:invariance of induced measure}
\end{equation}
for every Borel subset%
\footnote{In \eqref{eq:invariance of induced measure} we have used the notation
$U.\mathcal{O}=\left\{ U.\rho\,\left|\,\rho\in\Omega\right)\right\} $.%
} $\mathcal{O}\subset\Omega$. Therefore, one should think of $\mu_{\Omega}$
as the ``most uniform'' measure on $\Omega$ (an analogue of a uniform
measure on $n$ dimensional sphere $\mathbb{S}^{n}$). Let us now
introduce a stabilizer group
\[
\mathrm{Stab}\left(\rho_{0}\right)=\left\{ U\in\mathrm{SU}\left(\mathcal{H}\right)\,\left|\, U\rho_{0}U^{\dagger}=\rho_{0}\right.\right\} \,.
\]
Every smooth function%
\footnote{We could have considered more general functions but in what follows
we will make use only of the smooth ones.%
} $f:\mathrm{SU}\left(\mathcal{H}\right)\rightarrow\mathbb{R}$ that
is invariant under the action of $\mathrm{Stab}\left(\rho\right)$,
i.e. 
\begin{equation}
f\left(Uh\right)=f\left(U\right)\,,\label{eq:invariance of a function}
\end{equation}
for all $h\in\mathrm{Stab}\left(\rho_{0}\right)$and $U\in\mathrm{SU}\left(\mathcal{H}\right)$
can be mapped to a smooth function on $\Omega$, defined in the following
manner
\begin{equation}
\tilde{f}:\Omega\rightarrow\mathbb{R}\,,\,\tilde{f}\left(\rho\right)=f\left(\varphi_{\rho_{0}}^{-1}\left(\rho\right)\right)\,.\label{eq:induced function on omega}
\end{equation}
The above expression is well defined by virtue of \eqref{eq:invariance of a function}.
From the definition of the measure $\mu_{\Omega}$ (see Eq.\eqref{eq:definition of measure on Omega})
it follows that for a smooth function $f:\mathrm{SU}\left(\mathcal{H}\right)\rightarrow\mathbb{R}$
satisfying \eqref{eq:invariance of a function} we have
\begin{equation}
\mathbb{E}_{\mu_{\Omega}}\tilde{f}=\mathbb{E}_{\mu}f\label{eq:equality of means}
\end{equation}
or in other words
\[
\int_{\Omega}\tilde{f}\left(\rho\right)d\mu_{\Omega}\left(\rho\right)=\int_{\mathrm{SU}\left(\mathcal{H}\right)}f\left(U\right)d\mu\left(U\right)\,.
\]

\begin{defn}
Let $\Omega$ be a manifold of isospectral density matrices (with
prescribed ordered spectrum) and let $\mu_{\Omega}$ be a measure
on $\Omega$ induced from the Haar measure on $\mathrm{SU}\left(\mathcal{H}\right)$.
For a fixed class of non-correlated pure states $\mathcal{M}\subset\mathcal{D}_{1}\left(\mathcal{H}\right)$
the fraction of correlated states in $\Omega$,$\eta_{\Omega}^{\mathrm{corr}}\in\left[0,1\right]$,
is defined by
\begin{equation}
\eta_{\Omega}^{\mathrm{corr}}=\mu_{\Omega}\left(\left\{ \rho\in\Omega\,\left|\,\rho\notin\mathcal{M}^{c}\right)\right\} \right)\,.\label{eq:definition of fraction of non corr}
\end{equation}

\end{defn}
Our strategy for finding the lower bound for $\eta_{\Omega}^{\mathrm{corr}}$
is the following.
\begin{enumerate}
\item For a given manifold of isospectral density matrices $\Omega$ pick
$\rho_{0}\in\Omega$ and construct, using criteria developed in Chapter
\ref{chap:Polynomial-mixed states}, a function $f_{\Omega}:U\in\mathrm{SU}\left(\mathcal{H}\right)\rightarrow\mathbb{R}$
such that

\begin{enumerate}
\item $f_{\Omega}$ is $\mathrm{Stab}\left(\rho_{0}\right)$-invariant (in
a sense of Eq.\eqref{eq:invariance of a function});
\item the corresponding function $\tilde{f}_{\Omega}:\Omega\rightarrow\mathbb{R}$
satisfies the following property: 
\[
\tilde{f}_{\Omega}\left(\rho\right)>0\,\Longrightarrow\,\rho\notin\mathcal{M}^{c}\,.
\]

\end{enumerate}
\item Integrate $\tilde{f}_{\Omega}$ over $\Omega$ with respect to the
measure $\mu_{\Omega}$ (which amounts to integrating $f_{\Omega}$
over $\mathrm{SU}\left(\mathcal{H}\right)$ with respect to the Haar
measure). 
\item If the obtained average $\mathbb{E}_{\mu_{\Omega}}\tilde{f}_{\Omega}$
is positive infer, using relation between $\mu_{\Omega}$ and $\mu$
and concentration of measure on $\mathrm{SU}\left(\mathcal{H}\right)$,
about typicality of $\rho\in\Omega$ that are correlated ($\rho\notin\mathcal{M}^{c}$). 
\end{enumerate}
The more detailed reasoning and specific forms of functions $f_{\Omega}$
will be given in proofs of Theorems \ref{typicality bilinear} and
\ref{typicality k-linear} which are the main results of this Chapter

\subsection{Lower bound for $\eta_{\Omega}^{\mathrm{corr}}$ from bilinear witnesses\label{sub:Lower-bound-for}}

In this part we apply the strategy outlined above to give a lower
bound for $\eta_{\Omega}^{\mathrm{corr}}$ provided there exist an
operator $V\in\mathrm{Herm}\left(\mathcal{H}\otimes\mathcal{H}\right)$
that can be used to witness correlations%
\footnote{See Lemma \ref{lem:general bilin} for the necessary and sufficient
conditions that $V$ has to fulfill in order for \eqref{eq:bilin witness again}
to be satisfied.%
}, 
\begin{equation}
\mathrm{tr}\left(\rho_{1}\otimes\rho_{2}\, V\right)>0\,\Longrightarrow\,\rho_{1}\,\text{ and }\,\rho_{2}\,\text{are correlated}\,(\rho_{1},\rho_{2}\notin\mathcal{M}^{c})\,.\label{eq:bilin witness again}
\end{equation}
Note that we can assume without the loss of generality that $\left\Vert V\right\Vert \leq`1$.

Before we state our main result we first present two auxiliary lemmas
that will be needed in its proof. 
\begin{lem}
\label{Lipschitz constant}Let $V\in\mathrm{Herm}\left(\mathcal{H}^{\otimes k}\right)$
be a Hermitian operator on $\mathcal{H}^{\otimes k}$ satisfying%
\footnote{For $A\in\mathrm{End}\left(\mathcal{H}\right)$ we define an operator
norm \citep{Reed1972} by$\left\Vert A\right\Vert =\max_{\ket{\psi}:\bk{\psi}{\psi}=1}\left\Vert A\ket{\psi}\right\Vert \,,$
where $\left\Vert \cdot\right\Vert $ in the above expression denotes
the usual norm on $\mathcal{H}$.%
} $\left\Vert V\right\Vert \leq1$ and let $\rho_{1},\rho_{2},\ldots,\rho_{k}\in\mathcal{D}\left(\mathcal{H}\right)$.
Consider a function $f:\mathrm{SU}\left(\mathcal{H}\right)\rightarrow\mathbb{R}$
given by 
\begin{equation}
f\left(U\right)=\mathrm{tr}\left(\left[U^{\otimes k}\left\{ \rho_{1}\otimes\ldots\otimes\rho_{k}\right\} \left(U^{\dagger}\right)^{\otimes k}\right]\, V\right)\,.\label{eq:multilinear function definition}
\end{equation}
Then, the Lipschitz constant $L$ of the function $f$, with respect
to the metric tensor $g_{\mathrm{HS}}$, satisfies $L\leq2\cdot k$.\end{lem}
\begin{proof}
The proof is presented in the Section \ref{sec:Proofs-concerning-Chapter typicality}
of the Appendix (see page \pageref{sub:Proof-of Lemma lipschitz}). \end{proof}
\begin{lem}
\label{bilinear integration}Let $V\in\mathrm{Herm}\left(\mathcal{H}^{\otimes2}\right)$
be a Hermitian operator on $\mathcal{H}^{\otimes2}$ and let $\rho_{1},\rho_{2}\in\mathcal{D}\left(\mathcal{H}\right)$.
Consider a function $f_{\rho_{1,}\rho_{2}}:\mathrm{SU}\left(\mathcal{H}\right)\rightarrow\mathbb{R}$
given by 
\begin{equation}
f_{\rho_{1,}\rho_{2}}\left(U\right)=\mathrm{tr}\left(\left[U^{\otimes2}\left\{ \rho_{1}\otimes\rho_{2}\right\} \left(U^{\dagger}\right)^{\otimes2}\right]\, V\right)\,.\label{eq:bilinear definition}
\end{equation}
The average of the function $f$, $\mathbb{E}_{\mu}f$, with respect
to the normalized unitary measure on $\mathrm{SU}\left(\mathcal{H}\right)$
is given by 
\begin{equation}
\mathbb{E}_{\mu}f_{\rho_{1,}\rho_{2}}=\frac{\alpha+\beta}{2}+\frac{\alpha-\beta}{2}\mathrm{tr}\left(\rho_{1}\rho_{2}\right)\,,\label{eq:integration formula}
\end{equation}
where
\[
\alpha=\frac{\mathrm{tr}\left(\mathbb{P}^{\mathrm{sym}}V\right)}{\mathrm{dim}\left(\mathbb{P}^{\mathrm{sym}}\right)}\,,\,\beta=\frac{\mathrm{tr}\left(\mathbb{P}^{a\mathrm{sym}}V\right)}{\mathrm{dim}\left(\mathbb{P}^{a\mathrm{sym}}\right)}\,.
\]
\end{lem}
\begin{proof}
The proof is presented in Section \ref{sec:Proofs-concerning-Chapter typicality}
of the Appendix (see page \pageref{sub:Proof-of-Lemma bilinear integr}). \end{proof}
\begin{cor}
\label{maximum auxiliary corollary}Let $f_{\rho_{1,}\rho_{2}}:\mathrm{SU}\left(\mathcal{H}\right)\rightarrow\mathbb{R}$
be defined as in Eq.\eqref{eq:bilinear definition}. Assume that $\alpha\geq0$
and $\beta\leq0$. Then the following equality holds
\begin{equation}
\max_{\rho_{2}\in\mathcal{D\left(\mathcal{H}\right)}}\mathbb{E}_{\mu}f_{\rho_{1,}\rho_{2}}=\frac{\alpha+\beta}{2}+\frac{\alpha-\beta}{2}p_{max}\left(\rho_{1}\right)\,,\label{eq:maximumbilin}
\end{equation}
where $p_{max}\left(\rho_{1}\right)$ is the maximal eigenvalue of
the state $\rho_{1}$.\end{cor}
\begin{proof}
By the virtue of \eqref{eq:integration formula} it suffices to compute
$\max_{\rho_{2}\in\mathcal{D}\left(\mathcal{H}\right)}\mathrm{tr}\left(\rho_{1}\rho_{2}\right)$.
Due to the linearity (and thus concavity) of a function $\rho\rightarrow\mathrm{tr}\left(\rho_{1}\rho\right)$
the maximum is achieved for pure states. The equality $\max_{\kb{\psi}{\psi\in\mathcal{D}\left(\mathcal{H}\right)}}\mathrm{tr}\left(\rho_{1}\kb{\psi}{\psi}\right)=p_{max}\left(\rho_{1}\right)$
follows from the spectral decomposition of $\rho_{1}$.
\end{proof}
Our results are summarized in the following theorem. 
\begin{thm}
\label{typicality bilinear}\emph{Let }$\mathcal{M}\subset\mathcal{D}_{1}\left(\mathcal{H}\right)$
be a class of non-correlated pure states. \emph{be defined by} \textup{\eqref{eq:definition class}}
\emph{and let $V\in\mathrm{Herm}\left(\mathcal{H}^{\otimes2}\right)$
satisfy \eqref{eq:bilin witness again} and $\left\Vert V\right\Vert \leq1$.
Define
\begin{equation}
\alpha=\frac{\mathrm{tr}\left(\mathbb{P}^{\mathrm{sym}}V\right)}{\mathrm{dim}\left(\mathbb{P}^{\mathrm{sym}}\right)}\,,\,\beta=\frac{\mathrm{tr}\left(\mathbb{P}^{a\mathrm{sym}}V\right)}{\mathrm{dim}\left(\mathbb{P}^{a\mathrm{sym}}\right)}\,\label{eq:parameters of a witness}
\end{equation}
and assume}%
\footnote{Restrictions on the parameters $\alpha,\beta$  are not incidental.
In fact, they are typically satisfied for a concrete bilinear correlation
witness (see below). %
}\emph{: $\alpha>0,$ $\beta<0$ and $\alpha+\beta<0$. Let $\Omega\subset\mathcal{D}\left(\mathcal{H}\right)$
be a fixed manifold of isospectral density matrices and let $p_{\mathrm{max}}\left(\Omega\right)$
be the maximal eigenvalue of states from $\Omega$. Let 
\begin{equation}
p_{max,\mathrm{cr}}=-\frac{\alpha+\beta}{\alpha-\beta}\,.\label{eq:p critical bilin}
\end{equation}
Assume $p_{\mathrm{max}}\left(\Omega\right)=p_{\mathrm{max,cr}}+\delta$
($\delta>0$). We have the following lower bound}%
\footnote{\textsl{\emph{The formula in \eqref{eq:final estimate bilinear} is
slightly different (in the scalar factor in the denominator of the
exponent) form the analogous one given in \citep{Oszmaniec2014}.
The reason for that is a small mistake in \citep{Oszmaniec2014}.
Here we have corrected it.}}%
}\emph{ for the fraction of correlated states $\eta_{\Omega}^{\mathrm{corr}}$
in $\Omega$, }
\begin{equation}
\eta_{\Omega}^{\mathrm{corr}}=\mu_{\Omega}\left(\left\{ \rho\in\Omega|\,\rho\notin\mathcal{M}^{c}\right\} \right)\geq1-\mathrm{exp}\left(-\frac{N\delta^{2}\left(\alpha-\beta\right)^{2}}{256}\right)\,.\label{eq:final estimate bilinear}
\end{equation}
\end{thm}
\begin{proof}
The proof is a simple application of the strategy outlined in the
previous section. We fix a manifold of isospectral density matrices
$\Omega$ and $\rho_{0}\in\Omega$. From the bilinear correlation
witness $V\in\mathrm{Herm}\left(\mathcal{H}^{\otimes2}\right)$ we
construct a function $f_{\Omega}:\mathrm{SU}\left(\mathcal{H}\right)\rightarrow\mathbb{R}$,
\begin{equation}
f_{\Omega}\left(U\right)=\mathrm{tr}\left(\left[U^{\otimes2}\left\{ \rho_{0}\otimes\kb{\psi_{0}}{\psi_{0}}\right\} \left(U^{\dagger}\right)^{\otimes2}\right]\, V\right)\,,\label{eq:proposition of a function}
\end{equation}
where $\kb{\psi_{0}}{\psi_{0}}$ is the projector%
\footnote{From the Corollary \ref{maximum auxiliary corollary} it follows that
for a given witness $V$ the proposed form of the function $f_{\Omega}$
which is used is optimal in the class of functions having the form
\eqref{eq:bilinear definition}.%
} onto some fixed eigenvector of $\rho_{0}$ with the maximal eigenvalue
$p_{\max}$. One checks that the $f_{\Omega}\left(U\right)$ is $\mathrm{Stab}\left(\rho_{0}\right)$-invariant%
\footnote{The group $\mathrm{Stab}\left(\rho_{0}\right)$ consists of all operators
in $\mathrm{SU}\left(\mathcal{H}\right)$ that commute with $\rho_{0}$ %
} and therefore can be used to define a function $\tilde{f}_{\Omega}:\Omega\rightarrow\mathbb{R}$
(see Eq.\eqref{eq:induced function on omega}). Due to the fact that
$V$ is a bilinear correlation witness the following inclusion of
sets holds
\[
\left\{ \rho\in\Omega|\,\rho\notin\mathcal{M}^{c}\right\} \supset\left\{ \rho\in\Omega|\,\tilde{f}_{\Omega}\left(\rho\right)>0\right\} \,.
\]
Consequently, on the level of measures we have 
\begin{equation}
\mu_{\Omega}\left(\left\{ \rho\in\Omega|\,\rho\notin\mathcal{M}^{c}\right\} \right)\geq\mu_{\Omega}\left(\left\{ \rho\in\Omega|\,\tilde{f}_{\Omega}\left(\rho\right)>0\right\} \right)\,.\label{eq:prof typicality first}
\end{equation}
Using the definition of the measure $\mu_{\Omega}$ we get 
\begin{align}
\mu_{\Omega}\left(\left\{ \rho\in\Omega|\,\tilde{f}\left(\rho\right)>0\right\} \right) & =\mu\left(\left\{ U\in\mathrm{SU}\left(\mathcal{H}\right)|\, f_{\Omega}\left(\rho\right)>0\right\} \right)\,,\label{eq:inequality intermediate typicality bilin 1}\\
 & =\mu\left(\left\{ U\in\mathrm{SU}\left(\mathcal{H}\right)|\, f_{\Omega}\left(\rho\right)\geq0\right\} \right)\,,
\end{align}
where the second equality follows form the fact that for nontrivial%
\footnote{By a nontrivial $V$ we mean $V$ that do not commute with all operators
of the form $U\otimes U$, where $U\in\mathrm{SU}\left(\mathcal{H}\right)$. %
} $V$ we have
\[
\mu\left(\left\{ U\in\mathrm{SU}\left(\mathcal{H}\right)|\, f_{\Omega}\left(\rho\right)=0\right\} \right)=0\,.
\]
We obtain lower bound for the right hand side of \eqref{eq:inequality intermediate typicality bilin 1}
by assuming $\mathbb{E}_{\mu}f_{\Omega}\geq0$, and making use of
the inequality \eqref{eq:su concentration} and the result of the
Lemma \ref{Lipschitz constant},
\begin{align}
\mu\left(\left\{ U\in\mathrm{SU}\left(\mathcal{H}\right)|\, f_{\Omega}\left(\rho\right)\geq0\right\} \right) & =\mu\left(\left\{ U\in\mathrm{SU}\left(\mathcal{H}\right)|\, f_{\Omega}\left(\rho\right)-\mathbb{E}_{\mu}f_{\Omega}\geq-\mathbb{E}_{\mu}f_{\Omega}\right\} \right)\,,\nonumber \\
 & \geq1-\mathrm{exp}\left(-\frac{N}{64}\left(\mathbb{E}_{\mu}f_{\Omega}\right)^{2}\right)\,.\label{eq:final estimate typicality bilin}
\end{align}
We now use the Lemma \eqref{bilinear integration} and by inserting
the formula \eqref{eq:maximumbilin} to \eqref{eq:final estimate typicality bilin}.
The formula for $p_{max,\mathrm{cr}}$ stems from the condition $\mathbb{E}_{\mu}f_{\Omega}=0$.
Combination of the resulting expression with \eqref{eq:prof typicality first}
finishes the proof.
\end{proof}
Let us discuss the results the Theorem \ref{typicality bilinear}
and its proof. 
\begin{itemize}
\item By the virtue of Eq.\eqref{eq:final estimate bilinear} we know that
whenever we have a bilinear correlation witness $V$ we can give a
lower bound for $\eta_{\Omega}^{\mathrm{corr}}$, the fraction of
correlated states on a manifold of isospectral density matrices $\Omega$,
provided $p_{\mathrm{max}}\left(\Omega\right)>p_{\mathrm{cr}}\left(\Omega\right)=-\frac{\alpha+\beta}{\alpha-\beta}$.
The graphical representation of \eqref{eq:final estimate bilinear}
is given on Figure \ref{fig:fraction isospectral}. 
\end{itemize}
\begin{figure}[h]
\noindent \begin{centering}
\includegraphics[width=10cm]{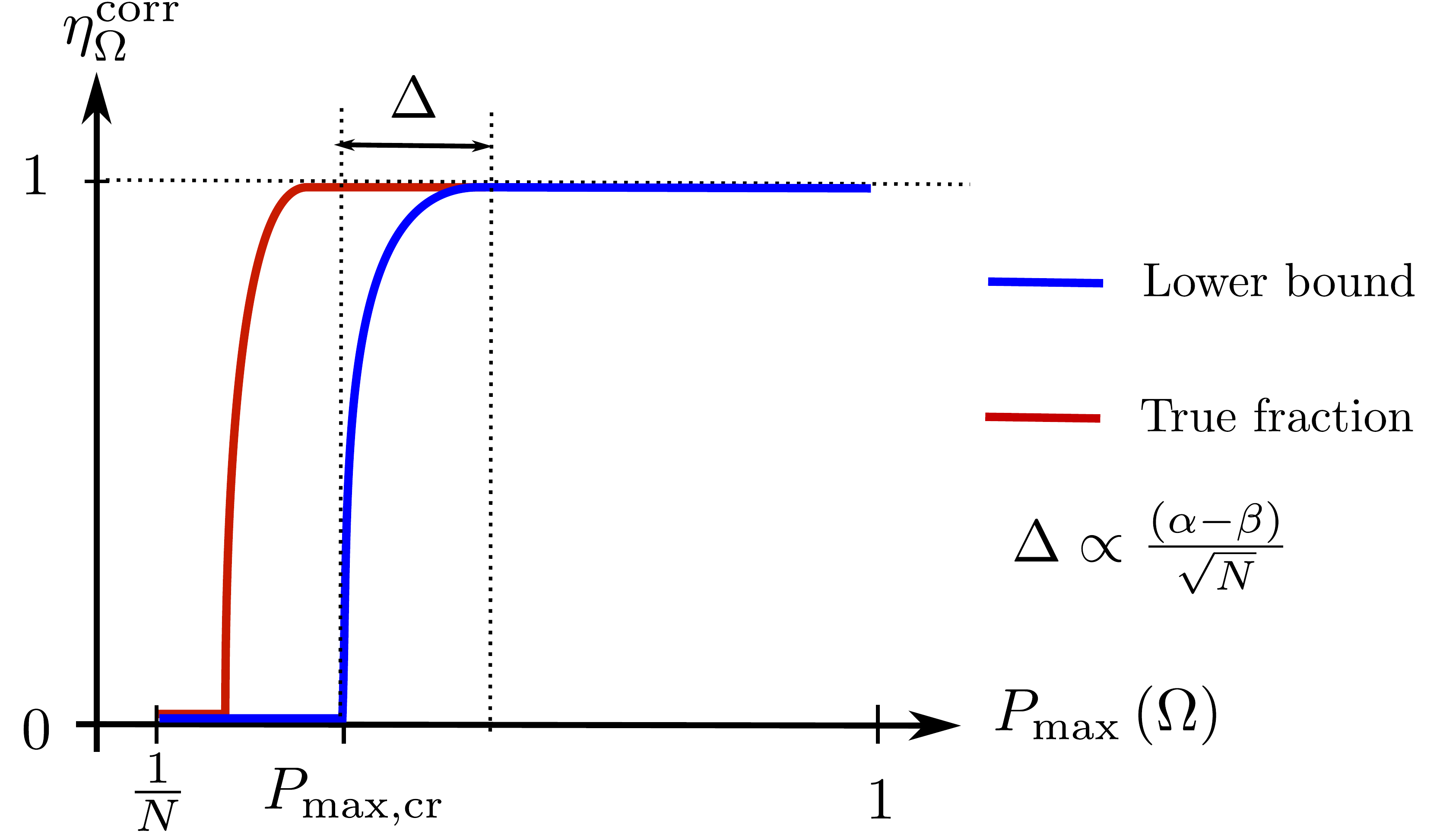}
\par\end{centering}

\noindent \protect\caption{\label{fig:fraction isospectral}Illustration of the lower bound \eqref{eq:final estimate bilinear}.
On the horizontal axis we have a maximal eigenvalue, $p_{\max}\left(\Omega\right)$,
of some one parameter family of isospectral density matrices. The
red curve depicts a hypothetical ``true'' fraction of correlated
states on $\Omega$ (as a function of $p_{\max}\left(\Omega\right)$).
The blue curve denotes the lower bound given by \eqref{eq:final estimate bilinear}.
The saturation of the blue curve is achieved on the scale $\Delta$,
which is of the order $\frac{\alpha-\beta}{\sqrt{N}}$. For $p_{\mathrm{max}}\left(\Omega\right)\leq p_{\mathrm{cr}}\left(\Omega\right)$
our methods do not allow us to give nontrivial lower bounds for $\eta_{\Omega}^{\mathrm{corr}}$. }
\end{figure}

\begin{itemize}
\item The obtained result does not tell us anything about behavior of $\eta_{\Omega}^{\mathrm{corr}}$
for $p_{\mathrm{max}}\left(\Omega\right)\leq p_{\mathrm{cr}}\left(\Omega\right)$.
In particular for many classes of non-correlated states $\mathcal{M}\subset\mathcal{D}\left(\mathcal{H}\right)$
the convex hull $\mathcal{M}^{c}$ contains a ball (with respect to,
say, the Hilbert-Schmidt metric) around the maximally mixed states.
Consequently we have $\eta_{\Omega}^{\mathrm{corr}}$ for $p_{\max}$
sufficiently close to $\frac{1}{N}$. This phenomenon happens in general
in the case of entanglement \citep{Zyczkowski1998} and, more generally,
for correlations defined via coherent states of compact simply-connected
Lie groups \citep{Grabowski2013}. Unfortunately with the use of presented
methods we can only bound the radius of the ball of non-correlated
states from above.
\item One could, in principle, use instead of $\kb{\psi_{0}}{\psi_{0}}$
some other operator in Eq.\eqref{eq:proposition of a function}. However,
by the virtue of Corollary \ref{maximum auxiliary corollary} (see
also the proof of Theorem \ref{typicality bilinear}) the choice of
$\kb{\psi_{0}}{\psi_{0}}$ gives optimal concentration inequalities%
\footnote{Instead of $\kb{\psi_{0}}{\psi_{0}}$ we could have used a any state
$\rho_{0}^{'}=\frac{f\left(\rho_{0}\right)}{\mathcal{N}}$, where
$f:\left[0,1\right]\rightarrow\mathbb{R}_{+}$ is any non-negative
function the unit interval (applied, in the usual operator-theoretic
sense, to $\rho_{0}$) and $\mathcal{N}$ is a normalization constant.%
}
\item We would like to remark that the non-linearity of the correlation
criterion \eqref{eq:bilin witness again} is essential in the proof
of Theorem \ref{typicality bilinear}. Imagine we had tried the analogous
reasoning for the criterion based on the linear criterion%
\footnote{In order to be consistent with the convention used throughout the
thesis we assume that the following implication holds $\rho\in\mathcal{M}^{c}\Longrightarrow\mathrm{tr}\left(\rho W\right)\leq0$.%
} $f\left(\rho\right)=\mathrm{tr}\left(\rho W\right)$. Then we would
have $\mathbb{E}_{\mu}f=\mathrm{\frac{1}{N}tr}\left(W\right)$. In
cases when correlations are defined via the choice of non-correlated
pure states $\mathcal{M}$, we usually%
\footnote{Strictly speaking $\frac{\mathbb{I}}{N}\in\mathcal{M}^{c}$ if and
only if $\mathcal{M}$ contains projectors onto vectors spanning $\mathcal{H}$%
} have $\frac{\mathbb{I}}{N}\in\mathcal{M}^{c}$ and from the value
expression $\mathrm{\frac{1}{N}tr}\left(W\right)=\mathrm{tr}\left(W\frac{\mathbb{I}}{N}\right)\leq0$
we cannot infer about the typical correlation properties of states
on $\Omega$.
\item Obtaining a different bound for the Lipschitz constant of $f\left(U\right)$
in Eq.\ref{eq:multilinear function definition} (for instance in terms
of the spectrum of states from $\Omega$) can lead to improved large
deviation inequalities and consequently better bounds for $\eta_{\Omega}^{\mathrm{corr}}$. 
\item We expect to obtain better concentration inequalities when we use
Ricci curvature estimates for the manifold of isospectral density
matrices $\Omega$ \citep{Arvanitogeorgos2003}.
\end{itemize}
We now apply \eqref{eq:final estimate bilinear} for four topical
types of correlations: entanglement, entanglement of bosons, ``entanglement
of fermions'' and ``non-Gaussian'' correlations in fermionic systems.
In each case we will use bilinear witnesses derived in Section \ref{sec:Bilinear-correlation-witnesses}.
All these witnesses have the following structure
\begin{equation}
V=A-c\mathbb{P}^{\mathrm{asym}}\,,\label{eq:special form witness}
\end{equation}
where $A$ is a projector onto some subspace of $\mathrm{Sym}^{2}\left(\mathcal{H}\right)$,
$c>0$ and $\mathbb{P}^{\mathrm{asym}}$ is a projector onto $\bigwedge^{2}\left(\mathcal{H}\right)$.
Consequently, for witnesses of the form \eqref{eq:special form witness}
we get
\[
\alpha=X=\frac{\mathrm{dim}\left(\mathrm{Im}\left(A\right)\right)}{\mathrm{dim}\left(\mathbb{P}^{\mathrm{sym}}\right)}\,,\,\beta=-c\,,
\]
Formulas \eqref{eq:p critical bilin} and \eqref{eq:final estimate bilinear}
take the simplified form. We get 
\begin{equation}
p_{max,\mathrm{cr}}=\frac{c-X}{c+X}\,,\label{eq:simplified pcrit}
\end{equation}
and
\begin{equation}
\eta_{\Omega}^{\mathrm{corr}}=\mu_{\Omega}\left(\left\{ \rho\in\Omega|\,\rho\notin\mathcal{M}^{c}\right\} \right)\geq1-\mathrm{exp}\left(-\frac{N\delta^{2}\left(X+c\right)^{2}}{256}\right)\,,\label{eq:simplified esitmate}
\end{equation}
for $p_{\mathrm{max}}\left(\Omega\right)=p_{\mathrm{max,cr}}+\delta(\delta>0)$.
We now give the values of the relevant parameters that describe each
of four scenarios mentioned above.
\begin{prop}
\label{values of parametres}Values of the parameters relevant for
the estimate \eqref{eq:simplified esitmate} for four classes of coherent
states: separable states, separable bosonic states, Slater determinants
and Fermionic Gaussian states are given in Table \eqref{tab:Parameters-characteising-typical}.

\begin{table}[h]
\centering{}%
\begin{tabular}{|c|c|c|c|}
\hline 
Class of states $\mathcal{M}$ & $N$ & $X$ & $c$\tabularnewline
\hline 
Separable states & $d^{L}$ & $1-2^{1-L}\left(\frac{\left(d+1\right)^{L}}{d^{L}+1}\right)$ & $\left(1-2^{1-L}\right)$\tabularnewline
\hline 
Separable bosonic states & $\binom{d+L-1}{L}$ & $1-\frac{2\binom{d+2L-1}{2L}}{\binom{d+L-1}{L}\left(\binom{d+L-1}{L}+1\right)}$ & $\left(1-2^{1-L}\right)$\tabularnewline
\hline 
Slater determinants & $\binom{d}{L}$ & $1-\frac{2\binom{d}{L}}{\binom{d}{L}+1}\cdot\frac{d+1}{\left(L+1\right)\left(d+1-L\right)}$ & $1-\frac{2}{L+1-\max\left\{ 0,2L-d\right\} }$\tabularnewline
\hline 
Fermionic Gaussian states & $2^{d-1}$ & $1-\frac{\binom{2d}{d}}{\left(2^{d-1}+1\right)2^{d-1}}$ & $1$\tabularnewline
\hline 
\end{tabular}\protect\caption{\label{tab:Parameters-characteising-typical}Parameters that appear
in the estimate \eqref{eq:simplified esitmate} for four types of
correlations.}
\end{table}
The class of ``Fermionic Gaussian states'' from the above table
consists%
\footnote{From the discussion carried out in Section \ref{sec:Classical-simulation-of-FLO}
and Subsection \ref{sub:Not-convex-Gaussian-correlation-bilin} it
follows that all convex-Gaussian states in $\mathcal{H}_{\mathrm{Fock}}\left(\mathbb{C}^{d}\right)$
have a block structure with respect to the decomposition $\mathcal{H}_{\mathrm{Fock}}\left(\mathbb{C}^{d}\right)=\mathcal{H}_{\mathrm{Fock}}^{+}\left(\mathbb{C}^{d}\right)\oplus\mathcal{H}_{\mathrm{Fock}}^{-}\left(\mathbb{C}^{d}\right)$
where $\mathcal{H}_{\mathrm{Fock}}^{\pm}\left(\mathbb{C}^{d}\right)$
denote the subspaces of the Fock space spanned by an even, and respectively,
odd number of excitations. For this reason a fraction of convex-Gaussian
states in any nontrivial $\Omega$ equals $0$. Therefore, in order
not to consider a trivial situation we consider states defined solely
on $\mathcal{H}_{\mathrm{Fock}}^{+}\left(\mathbb{C}^{d}\right)$.%
} of pure Gaussian states of parity one, $\mathcal{M}_{g}^{+}$ and
the relevant Hilbert space is $\mathcal{H}_{\mathrm{Fock}}^{+}\left(\mathbb{C}^{d}\right)$. \end{prop}
\begin{proof}
The proof is given in Section \ref{sec:Proofs-concerning-Chapter typicality}
of the Appendix (see page \pageref{sub:sketch of Proof klin}).
\end{proof}
The important thing to notice is that for the above described types
of correlations $p_{\mathrm{max},cr}\rightarrow0$ with the increase
of the relevant parameters characterizing the system (number of particles
$L$, number of modes $d$). We now give the asymptotic of $N\cdot p_{\mathrm{max},cr}$
(the quantity $\frac{1}{N}$ defines a natural ``length scale''
for the variable $p_{\max}$) for these types of correlations. We
study two asymptotic limits. In the first one (see Table \ref{tab:firts assympt})
we fix the number of modes $d$ and go to infinity with the number
of particles $L$. In the second limit%
\footnote{In Table \ref{tab:firts assympt 2} we use the following notation.
By $H:\left[0,1\right]\rightarrow\mathbb{R}_{+}$ we denote the binary
entropy:
\[
H\left(x\right)=-x\mathrm{log}\left(x\right)-\left(1-x\right)\mathrm{log}\left(1-x\right)\,.
\]
The function $f:\left(0,\infty\right)\rightarrow\mathbb{R}_{+}$ is
given by
\[
f(a)=\mathrm{log}\left(2\right)a+a\mathrm{log}\left(a\right)-2a\mathrm{log}\left(2a\right)-(1+a)\mathrm{log}\left(1+a\right)+(1+2a)\mathrm{log}(1+2a)\,.
\]
} (see Table \ref{tab:firts assympt 2}) we set the ratio $\frac{d}{L}=a$
and increase the number of modes $d$. The asymptotics presented in
Tables \ref{tab:firts assympt} and \ref{tab:firts assympt 2} give
information about the behavior of the ratio of $p_{\mathrm{max},cr}$
to the ``typical'' value of $p_{\max}$ which is of the order of
$\frac{1}{N}$. 

\begin{table}[h]
\centering{}%
\begin{tabular}{|c|c|c|}
\hline 
Class of states $\mathcal{M}$ & $N$ & $N\cdot p_{\mathrm{max},cr}$ \tabularnewline
\hline 
Separable states & $d^{L}$ & $\left(\frac{1}{2}+\frac{d}{2}\right)^{L}$\tabularnewline
\hline 
Separable bosonic states & $\frac{L^{d-1}}{\left(d-1\right)!}$ & $\frac{2^{L+d}}{L^{d}}$\tabularnewline
\hline 
\end{tabular}\protect\caption{\label{tab:firts assympt}Asymptotic behavior of $N$ and $N\cdot p_{\mathrm{max},cr}$
in the limit $d=const,\, L\rightarrow\infty$.}
\end{table}

\begin{table}[h]
\centering{}%
\begin{tabular}{|c|c|c|}
\hline 
Class of states $\mathcal{M}$ & $N$ & $N\cdot p_{\mathrm{max},cr}$ \tabularnewline
\hline 
Separable states & $d^{a\cdot d}$ & $\mathrm{exp}\left(a\right)\cdot\left(\frac{d}{2}\right)^{a\cdot d}$\tabularnewline
\hline 
Separable bosonic states & $\mathrm{exp}\left(d\left(1+a\right)\cdot H\left(\frac{a}{1+a}\right)\right)$ & $\sqrt{\frac{1+a}{2\left(1+2a\right)}}\,\mathrm{exp}\left(f\left(a\right)d\right)$\tabularnewline
\hline 
Slater determinants for $a\leq\frac{1}{2}$ & $\mathrm{exp}\left(d\cdot H\left(\frac{1}{a}\right)\right)$ & $\frac{\mathrm{\mathrm{exp}\left(d\cdot H\left(\frac{1}{a}\right)\right)}}{d}$\tabularnewline
\hline 
Slater determinants for $\frac{1}{2}\leq a\leq1$ & $\mathrm{exp}\left(d\cdot H\left(\frac{1}{a}\right)\right)$ & $\frac{\mathrm{\mathrm{exp}\left(d\cdot H\left(\frac{1}{a}\right)\right)}}{a\cdot d}$\tabularnewline
\hline 
Fermionic Gaussian states & $2^{d-1}$ & $2\frac{2^{d}}{\sqrt{\pi d}}$\tabularnewline
\hline 
\end{tabular}\protect\caption{\label{tab:firts assympt 2}Asymptotic behavior of $N$ and $N\cdot p_{\mathrm{max},cr}$
in the limit: $\frac{L}{d}=a>0$, $d\rightarrow\infty$. }
\end{table}

The asymptotics presented in Tables \ref{tab:firts assympt} and \ref{tab:firts assympt 2}
were obtained with the usage of the standard asymptotic properties
of binomial coefficients and factorials \citep{Knuth1989}. In the
specific example of $L$ qbit system the presented approach allows
to give the upper bound $R_{up}$ for the radius $R$ of the ball
of separable states around the maximally mixed state. Using the results
presented in Table \ref{tab:firts assympt} and the inequality $p_{\max}\leq\mathrm{tr}\left(\rho^{2}\right)$
we get $R_{up}\approx\left(\frac{3}{4}\right)^{L}$. This result is
much weaker than the known \citep{Gurvits2005,Hildebrand2007} scaling
$R\approx\alpha\left(\frac{1}{\sqrt{6}}\right)^{L}$, for $\alpha=O\left(1\right)$,
which is known in this case. We believe that it would be possible
to obtain much smaller values of $p_{\mathrm{max},cr}$ than the ones
presented above if, instead of criteria \eqref{eq:special form witness}
we use the general bilinear group-invariant criteria derived in Section
\ref{sec:Optimal--bilinear}.

\subsection{Lower bound for $\eta_{\Omega}^{\mathrm{corr}}$ from $k$-linear
witnesses\label{sub:Lower-bound-from multilin}}

In this part we apply the strategy presented at the end of Subsection
\ref{sub:Geometry-of-omega} to give a lower bound for $\eta_{\Omega}^{\mathrm{corr}}$
provided there exist an operator $A\in\mathrm{Herm}_{+}\left(\mathrm{Sym}^{k}\left(\mathcal{H}\right)\right)$
defining the class of non-correlated pure states $\mathcal{M}$. After
presenting our general result we apply it to estimate $\eta_{\Omega}^{\mathrm{corr}}$
for the cases when the set of correlated states consist of Genuine
multiparty entanglement in tripartite system and Schmidt number of
bipartite states. 

We first present the auxiliary results.
\begin{lem}
\label{multilinear integration}Let $V\in\mathrm{Herm}\left(\mathcal{H}^{\otimes k}\right)$
be a Hermitian operator on $\mathcal{H}^{\otimes k}$ given by
\begin{equation}
V=A-\left(k-1\right)\left(\mathbb{I}^{\otimes k}-\mathbb{P}^{\mathrm{sym,k}}\right)\,,\label{eq:quasi criterion}
\end{equation}
where $A\in\mathrm{Herm}\left(\mathrm{Sym}^{k}\left(\mathcal{H}\right)\right)$.
Consider a function $f:\mathrm{SU}\left(\mathcal{H}\right)\rightarrow\mathbb{R}$
given by 
\begin{equation}
f\left(U\right)=\mathrm{tr}\left(\left[U^{\otimes k}\left\{ \rho\otimes\left(\kb{\psi}{\psi}\right)^{\otimes\left(k-1\right)}\right\} \left(U^{\dagger}\right)^{\otimes k}\right]\, V\right)\,.\label{eq:multilinear function definition second}
\end{equation}
where $\rho\in\mathcal{D}\left(\mathcal{H}\right)$ and $\kb{\psi}{\psi}\in\mathcal{D}_{1}\left(\mathcal{H}\right)$.
The average of the function $f$, $\mathbb{E}_{\mu}f$, with respect
to the normalized unitary measure on $\mathrm{SU}\left(\mathcal{H}\right)$
is given by 
\begin{equation}
\mathbb{E}_{\mu}f=-\left(k-1\right)+\left(\left(k-1\right)+\frac{\mathrm{tr}\left(A\right)}{\mathrm{dim}\left(\mathrm{Sym}^{k}\left(\mathcal{H}\right)\right)}\right)\mathrm{tr}\left(\left[\rho\otimes\left(\kb{\psi}{\psi}\right)^{\otimes\left(k-1\right)}\right]\mathbb{P}^{\mathrm{sym,k}}\right)\,.\label{eq:multilinear average}
\end{equation}
\end{lem}
\begin{proof}
The proof is analogous to the proof of Lemma \ref{bilinear integration}
and is given in Section \ref{sec:Proofs-concerning-Chapter typicality}
of the Appendix (see page \pageref{sub:Proof-of-Lemma k ilinear integr}).\end{proof}
\begin{cor}
\label{multilin average}Let $\rho\in\mathcal{D}\left(\mathcal{H}\right)$
and $\kb{\psi}{\psi}$ be a projector onto a eigenvector of $\rho$
corresponding to its maximal eigenvalue, denoted by $p_{\mathrm{max}}$.
The expectation value of a function \eqref{eq:multilinear function definition second}
is given by
\begin{equation}
\mathbb{E}_{\mu}f=-\left(k-1\right)+\left(\left(k-1\right)+\frac{\mathrm{tr}\left(A\right)}{\mathrm{dim}\left(\mathrm{Sym}^{k}\left(\mathcal{H}\right)\right)}\right)\left(\frac{k-1}{k}p_{\mathrm{max}}-\frac{1}{k}\right)\,.\label{eq:optimal multilin average}
\end{equation}

\end{cor}
Our main results are summarized in the following theorem.
\begin{thm}
\label{typicality k-linear}\emph{Let }$\mathcal{M}\subset\mathcal{D}_{1}\left(\mathcal{H}\right)$
be a class of non-correlated pure states specified by a $k$-linear
condition\emph{
\[
\kb{\psi}{\psi}\in\mathcal{M}\,\Longleftrightarrow\,\mathrm{tr}\left(A\left[\kb{\psi}{\psi}^{\otimes k}\right]\right)=0\,,
\]
where $A\in\mathrm{Herm}_{+}\left(\mathrm{Sym}^{k}\left(\mathcal{H}\right)\right)$,
$\left\Vert A\right\Vert \leq1$. Let us set
\begin{equation}
X=\frac{\mathrm{tr}\left(A\right)}{\mathrm{tr}\left(\mathbb{P}^{\mathrm{k,sym}}\right)},\,\tilde{X}=\frac{1}{k-1}\frac{\mathrm{tr}\left(A\right)}{\mathrm{tr}\left(\mathbb{P}^{\mathrm{k,sym}}\right)}\label{eq:parameters klin}
\end{equation}
Let $\Omega\subset\mathcal{D}\left(\mathcal{H}\right)$ be a fixed
manifold of isospectral density matrices and let $p_{\mathrm{max}}\left(\Omega\right)$
be the maximal eigenvalue of states from $\Omega$. Let 
\begin{equation}
p_{max,\mathrm{cr}}=\frac{k-1-\tilde{X}}{k-1+X}\,.\label{eq:k linear crit}
\end{equation}
Assume $p_{\mathrm{max}}=p_{\mathrm{max,cr}}+\delta$ ($\delta>0$).
We have the following lower bound for the fraction of correlated states
$\eta_{\Omega}^{\mathrm{corr}}$ in $\Omega$ }
\begin{equation}
\eta_{\Omega}^{\mathrm{corr}}=\mu_{\Omega}\left(\left\{ \rho\in\Omega|\,\rho\notin\mathcal{M}^{c}\right\} \right)\geq1-\mathrm{exp}\left(-\frac{N\delta^{2}C_{k}\left(1+X\right)^{2}}{16}\right)\,,\label{eq:final estimate klin}
\end{equation}
where $C_{k}=\left(\frac{k-1}{k}\right)^{4}\,.$\end{thm}
\begin{proof}
The proof is completely analogous to the proof of Theorem \eqref{typicality bilinear}.
The only difference is that instead of the function $f_{\Omega}$
given by \eqref{eq:proposition of a function} we use now 
\[
f_{\Omega}\left(U\right)=\mathrm{tr}\left(\left[U^{\otimes k}\left\{ \rho\otimes\left(\kb{\psi_{0}}{\psi_{0}}\right)^{\otimes\left(k-1\right)}\right\} \left(U^{\dagger}\right)^{\otimes k}\right]\, V\right)\,,
\]
where $\kb{\psi_{0}}{\psi_{0}}$ is the projector onto some fixed
eigenvector of $\rho_{0}$ with the maximal eigenvalue $p_{\max}$.
In order to compute $\mathbb{E}_{\mu}f_{\Omega}$ we use Lemma \ref{multilinear integration}
and Corollary \ref{multilin average}. the precise form of $p_{max,\mathrm{cr}}$
in Eq.\eqref{eq:k linear crit} follows directly from application
of \eqref{eq:optimal multilin average}.
\end{proof}
The discussion of Theorem \ref{typicality bilinear} applies also
to this result. We would like to remark that the criterion given by
operator $V$ of the form \eqref{eq:quasi criterion} is not particularly
strong. Consequently, the obtained value of $p_{max,\mathrm{cr}}$
is certainly much overestimated. In the future we would like to develop
better correlation witnesses that in turn would lead to better estimates
for $\eta_{\Omega}^{\mathrm{corr}}$.

We now apply \eqref{eq:final estimate klin} for for correlations
describing a ``refined'' entanglement: 
\begin{itemize}
\item Tripartite states exhibiting genuine multiparty entanglement (GME)
in $\mathcal{D}\left(\mathbb{C}^{d}\otimes\mathbb{C}^{d}\otimes\mathbb{C}^{d}\right)$;
\item States in $\mathcal{D}\left(\mathbb{C}^{d}\otimes\mathbb{C}^{d}\right)$
characterized by Schmidt number greater than some natural number $n$
($n\leq d-1$).
\end{itemize}
In the case of tripartite GME the relevant class of non-correlated-states
consists of 2-separable states $\mathcal{M}_{d}^{2}\subset\mathcal{D}\left(\mathbb{C}^{d}\otimes\mathbb{C}^{d}\otimes\mathbb{C}^{d}\right)$
(see Definition \ref{def L-party sep}). In the case of correlations
defined via Schmidt number  the class of non-correlated states consists
of bipartite states $\mathcal{M}_{n}\subset\mathcal{D}\left(\mathbb{C}^{d}\otimes\mathbb{C}^{d}\right)$
with Schmidt rank bounded by $n\leq d-1$ (see Definition \ref{def:Schmidt number}).
The operators $A$ defining (via Eq.\eqref{eq:polynomial characterisation again})
classes $\mathcal{M}_{d}^{2}$ and $\mathcal{M}_{n}$ are given by
Eq.\eqref{eq:characterisation GME} and \eqref{eq:characterisation of schmidt}
respectively. We now give the values of the relevant parameters appearing
in \eqref{eq:parameters klin} and \eqref{eq:final estimate klin}
for both scenarios mentioned above.
\begin{prop}
\label{values of parametres klin}Values of the parameters relevant
for the estimate \eqref{eq:simplified esitmate} for classes of non-correlated
pure states $\mathcal{M}_{d}^{3}\subset\mathcal{D}\left(\mathbb{C}^{d}\otimes\mathbb{C}^{d}\otimes\mathbb{C}^{d}\right)$
and $\mathcal{M}_{n}\subset\mathcal{D}\left(\mathbb{C}^{d}\otimes\mathbb{C}^{d}\right)$
are given in Table \eqref{tab:Parameters-characteising-typical-klin}.

\begin{table}[h]
\noindent \centering{}%
\begin{tabular}{|c|c|c|c|}
\hline 
Class of states $\mathcal{M}$ & $N$ & $k$ & $X$\tabularnewline
\hline 
$\mathcal{M}_{d}^{2}$ & $d^{3}$ & $6$ & $X=\frac{1}{8}+O\left(\frac{1}{d}\right)$\tabularnewline
\hline 
$\mathcal{M}_{n}$ & $d^{2}$ & $n+1$ & $\frac{\binom{d}{n+1}^{2}}{\binom{d^{2}+n}{n+1}}$\tabularnewline
\hline 
\end{tabular}\protect\caption{\label{tab:Parameters-characteising-typical-klin}Parameters that
appear in the estimate \ref{eq:final estimate klin} for cases when
correlated states consist of (i) tripartite states that exhibit GME
(the relevant class of pure states is $\mathcal{M}_{d}^{2}$) and
(ii) bipartite states that have Schmidt number greater than $n$,
$n\leq d-1$ (the relevant class of pure states is $\mathcal{M}_{n}$). }
\end{table}
In the case of $\mathcal{M}_{d}^{3}$ we were not able to compute
$X=\frac{\mathrm{Tr}\left(A\right)}{\mathrm{tr}\left(\mathbb{P}^{\mathrm{k,sym}}\right)}$
explicitly. Instead we gave the asymptotic expression%
\footnote{We used the standard ``big O'' notation: for a function of a real
parameter $d$ we write $f\left(d\right)=O\left(\frac{1}{d}\right)$
if and only if $\lim_{d\rightarrow\infty}\left|f\left(d\right)\right|\cdot d<\infty$%
}. \end{prop}
\begin{proof}
The sketch of proof is given in Section \ref{sec:Proofs-concerning-Chapter typicality}
of the Appendix (see page \pageref{sub:sketch of Proof klin}).
\end{proof}

\section{Summary and open problems\label{sec:Discussion-typicality}}

In this Chapter we studied the fraction of correlated states, $\eta_{\Omega}^{\mathrm{corr}}$,
on manifolds of isospectral density matrices $\Omega\subset\mathcal{D}\left(\mathcal{H}\right)$.
The results presented in this chapter are generalization of the results
from one of the article of the author \citep{Oszmaniec2014}. We obtained
the connection between $\eta_{\Omega}^{\mathrm{corr}}$ and the polynomial
description of the class of pure non-correlated states $\mathcal{M}\subset\mathcal{D}_{1}\left(\mathcal{H}\right)$.
More precisely, we derived lower bounds to $\eta_{\Omega}^{\mathrm{corr}}$
in terms of the quantities characterizing operator $A\in\mathrm{Herm}_{+}\left(\mathrm{Sym}^{k}\left(\mathcal{H}\right)\right)$
used to define the class $\mathcal{M}$ (via Eq.\eqref{eq:polynomial characterisation again}).
In short we showed that whenever states belonging to $\Omega$ are
``pure enough'' and if the dimension of the relevant Hilbert space
is ``large'' (for more quantitative statements see Theorems \ref{typicality bilinear}
and \ref{typicality k-linear}) the fraction of correlated states
$\eta_{\Omega}^{\mathrm{corr}}$ is very close to unity. In order
to obtain these results we used the technique of measure concentration
(introduced in Section \ref{sec:Introduction-to-concentration}) and
multilinear criteria for detection of correlations derived in Chapter
\ref{chap:Polynomial-mixed states}. We applied our general results
to particular types of correlations: 
\begin{itemize}
\item Entanglement, particle entanglement of bosons, ``entanglement of
fermions'', ``non-Gaussianity'' fermionic states (see Proposition
\ref{values of parametres} and Tables \ref{tab:firts assympt} and
\ref{tab:firts assympt 2});
\item Genuine multipartite entanglement of tripartite system and the notion
of correlations based on Schmidt rank (see Proposition \ref{values of parametres klin}
and Table \ref{tab:Parameters-characteising-typical-klin}).
\end{itemize}
For each type of correlations we expressed $\eta_{\Omega}^{\mathrm{corr}}$
via the relevant parameters characterizing a given scenario (number
of particles, number of modes etc.). To our knowledge studies of the
fraction of correlated states on the manifolds of isospectral density
matrices have not been carried out before the work presented here
(and in \citep{Oszmaniec2014}).

\subsection*{Open problems}

Below we give a list of interesting open problems related to the work
presented in this Chapter.
\begin{itemize}
\item Apply the obtained results for $\eta_{\Omega}^{\mathrm{corr}}$ to
study the fraction of correlated states, $\eta^{\mathrm{corr}}$on
the set off all states $\mathcal{D}\left(\mathcal{H}\right)$ but
equipped with some $\mathrm{SU}\left(\mathcal{H}\right)$-invariant
measure (e.g., Hilbert-Schmidt measure \citep{Zyczkowski2003} or
Bures measure \citep{Sommers2003}); 
\item Find the maximal radius of a ball (with respect to Hilbert-Schmidt
distance) that is contained in $\mathcal{M}^{c}\subset\mathcal{D}\left(\mathcal{H}\right)$,
where the class of pure states $\mathcal{M}\subset\mathcal{D}_{1}\left(\mathcal{H}\right)$
is defined by the condition $\mathrm{tr}\left(A\left[\kb{\psi}{\psi}^{\otimes k}\right]\right)=0$;
\item For a given class $\mathcal{M}$ defined as above give ranges of (ordered)
spectra $\left(p_{1},\ldots,p_{N}\right)$ such that $\eta_{\Omega}^{\mathrm{corr}}\approx1$
and $\eta_{\Omega}^{\mathrm{corr}}\approx0$ respectively. 
\item Assume that $P\left(\Omega\right)=\sum_{i=1}^{N}p_{i}^{2}>1-\epsilon$,
where $\epsilon>0$ can be chosen to be arbitrary small. Do non-correlated
states in $\Omega$ have the full measure? In other words do we have
$\eta_{\Omega}^{\mathrm{corr}}<1$ whenever $\epsilon$ is nonzero?
\item Let $A\in\mathrm{Herm}\left(\mathrm{Sym}^{k}\left(\mathcal{H}\right)\right)$
be the operator defining the class of non-correlated states $\mathcal{M}$.
The problem of estimating $\eta_{\Omega}^{\mathrm{corr}}$ on the
manifold of isospectral density matrices $\Omega$ is invariant under
the conjugation of $A$ under ``global unitaries'',
\begin{equation}
A\rightarrow U.A=A'=U^{\otimes k}A\left(U^{\otimes k}\right)^{\dagger}\,,\label{eq:extra action}
\end{equation}
where $U\in\mathrm{SU}\left(\mathcal{H}\right)$. Therefore $\eta_{\Omega}^{\mathrm{corr}}$
should depend only on the orbit the group $\mathrm{SU}\left(\mathcal{H}\right)$
through the operator $A\in\mathrm{Herm}\left(\mathrm{Sym}^{k}\left(\mathcal{H}\right)\right)$.
This orbit is uniquely characterized by the polynomial invariants
\citep{Goodman1998} of the action \eqref{eq:extra action} of $\mathrm{SU}\left(\mathcal{H}\right)$
in $\mathrm{Herm}\left(\mathrm{Sym}^{k}\left(\mathcal{H}\right)\right)$.
It would be interesting to know which polynomial invariants determine
the behavior of $\eta_{\Omega}^{\mathrm{corr}}$ (for a fixed $\Omega$).
Note that the reasoning presented in this Chapter estimates $\eta_{\Omega}^{\mathrm{corr}}$
based solely on the perhaps simplest invariant of the action \eqref{eq:extra action},
the trace of the operator $A$.
\end{itemize}

\chapter{Summary and outlook \textmd{\label{chap:Conclusions-and-outlook}}}

In this thesis we analyzed various properties of correlations which
are defined in the analogous manner to entanglement. The starting
point was always the class of non-correlated pure states $\mathcal{M}\subset\mathcal{D}_{1}\left(\mathcal{H}\right)$
of the Hilbert space of interest. The set of non-correlated mixed
states consisted of states belonging to the convex hull, $\mathcal{M}^{c}\subset\mathcal{D}\left(\mathcal{H}\right)$,
of the set $\mathcal{M}$. The correlated states (with respect to
the choice of the class $\mathcal{M}$) were then defined as states
that do not belong to $\mathcal{M}^{c}$. Throughout the thesis we
analyzed cases when non-correlated pure states are defined as a zero
set of a homogenous polynomial in the density matrix, 

\begin{equation}
\kb{\psi}{\psi}\in\mathcal{M}\,\Longleftrightarrow\mathrm{tr}\left(\left[\kb{\psi}{\psi}^{\otimes k}\right]A\right)=0\,,\label{eq:final characterization}
\end{equation}
where $A$ is the non-negative operator on the $k$-fold symmetrization
of the Hilbert space $\mathcal{H}$. Bellow we list the main results
presented in this thesis.

\textbf{Chapter \ref{chap:Multilinear-criteria-for-pure-states}}

In this chapter we presented a polynomial characterization of many
classes of pure states that define a physically-relevant types of
correlations. The main results contained in the chapter are:
\begin{itemize}
\item Characterization of coherent states of compact simply-connected Lie
groups as zero sets of a polynomial (in state's density matrix) of
degree 2 (Proposition \ref{prop:proj two}).
\item Explicit forms of the polynomial of degree two characterizing the
following classes of pure states in finite dimensional Hilbert spaces:

\begin{itemize}
\item Pure separable states (Lemma \ref{lemma crit dist part});
\item Separable bosonic states (Lemma \ref{lema bosons p2}); 
\item Slater determinants (Lemma \ref{lema crit fermions});
\item Fermionic Gaussian states (Proposition \ref{algebraic form operator}). 
\end{itemize}

These are examples of classes of states that can be interpreted as
the coherent states of compact simply-connected Lie groups.

\item Explicit forms of the polynomial of degree two characterizing pure
separable states (Lemma \ref{lemma nfinite dim prod criterion}),
bosonic separable states (Lemma \ref{lemma nfinite dim bos criterion})
and Slater determinants (Lemma \ref{lemma: ferm inf dimensional})
in situations where single particle Hilbert spaces are general separable
Hilbert spaces.
\item A polynomial characterization of the multiparty pure states that do
not exhibit Genuine Multiparty Entanglement (Lemma \ref{GME pure}). 
\item A polynomial characterization of the bipartite pure states with bounded
Schmidt rank (Lemma \ref{characterization schmidt rank}).
\item Characterization of coherent states of compact simply-connected Lie
groups as a zero set of a polynomial of degree $k$ (Lemma \ref{k lin characterization coheren}).
\end{itemize}
\textbf{Chapter \ref{chap:Complete-characterisation}}

In this part of the thesis we investigated the cases when the polynomial
characterization of pure non-correlated states \eqref{eq:final characterization}
allows to analytically characterize the set of mixed correlated states
$\mathcal{M}^{c}$. The main results of the chapter are: 
\begin{itemize}
\item Group-theoretic characterization of generalized coherent states that
can be defined via the anti-unitary conjugation (Theorems \ref{theta representation}
and \ref{theta epimorphism}). For the correlations defined via this
kind of coherent states it is possible to characterize the class of
correlated mixed states via a simple analytical condition based on
Uhlmann-Wotters construction (Eq.\ref{eq:explict formula}).
\item Analytical characterization of fermionic convex-Gaussian states in
four mode fermionic Fock space (Theorem \ref{analitical characterization gaussian}).
\item A generalized Schmidt decomposition for pure states in even/odd subspace
of Four mode Fock space (Lemma \ref{eneralised Schmidt decomposition}).
\end{itemize}
\textbf{Chapter \ref{chap:Polynomial-mixed states}}

I this chapter we used the characterization \eqref{eq:final characterization}
of the set $\mathcal{M}$ to derive a nonlinear criteria for detection
of correlations in general mixed states. The most important results
of the chapter are:
\begin{itemize}
\item Derivation of the polynomial criterion detecting correlations in mixed
states based solely on the polynomial characterization \eqref{eq:final characterization}
of pure uncorrelated states for $k=2$ (Theorem \ref{thm:Main result bilin})
and for arbitrary number $k$ (Theorem \ref{main result multilinear witness}).
We derived, using results from Chapter \ref{chap:Multilinear-criteria-for-pure-states}
correlation criteria for seemingly unrelated types of correlations:

\begin{itemize}
\item Entanglement of distinguishable particles;
\item Particle entanglement of bosons;
\item ``Entanglement'' of fermions;
\item Non-convex Gaussian correlations in fermionic systems;
\item Genuine multipartite entanglement;
\item Correlations based on the notion of Schmidt number.
\end{itemize}
\item Complete description of group-invariant bilinear correlation witnesses
for correlations based on coherent states of compact simply-connected
Lie groups%
\footnote{In order for the description to be effective we need to add some additional
assumptions (see Subsection \ref{sub:general structure bilin}). %
} (Theorem \ref{finitelly generated cone}). The general method was
applied to four classes of coherent sates: separable states (Lemma
\ref{LU invariant bilin theorem}), separable bosonic states (Lemma
\ref{LUb invariant bilin theorem}), Slater determinants (Lemma \ref{LUf invariant bilin theorem})
and fermionic Gaussian states (Lemma \ref{FLO invariant bilin theorem}).
\end{itemize}
\textbf{Chapter \ref{chap:Typical-properties-of}}

In this chapter we studied the fraction of correlated states $\eta_{\Omega}^{\mathrm{corr}}$
on manifolds of isospectral density matrices $\Omega\subset\mathcal{D}_{1}\left(\mathcal{H}\right)$
of a relevant Hilbert space. We used the nonlinear criteria for detection
of correlations to find lower bounds for $\eta_{\Omega}^{\mathrm{corr}}$
for correlations defined via the Eq.\eqref{eq:final characterization}.
The main original contributions contained in the chapter are: 
\begin{itemize}
\item Estimation from bellow of the fraction of correlated states $\eta_{\Omega}^{\mathrm{corr}}$
on the manifold of isospectral density matrices. The estimate is given
in terms of the spectrum of states from $\Omega$ and the trace of
the operator $A$ from \eqref{eq:final characterization} (see Theorem
\ref{typicality bilinear} for $k=2$ and Theorem \ref{typicality k-linear}
for arbitrary $k$).
\item Application of the above estimate for concrete types of correlations:

\begin{itemize}
\item Entanglement of distinguishable particles, particle entanglement of
Bosons, ``entanglement'' of fermions, not convex-Gaussian correlations
in fermionic systems (Proposition \ref{values of parametres});
\item Genuine multiparty entanglement, States with Schmidt number greater
than $n$ (Proposition \ref{values of parametres klin}).
\end{itemize}
\end{itemize}

\subsection*{Open problems}

The lists of open problems related to results obtained in the thesis
were given at the end of each of the Chapters \ref{chap:Multilinear-criteria-for-pure-states}-
\ref{chap:Typical-properties-of}. Bellow we list the problems which
are in our opinion most important.
\begin{itemize}
\item Is it possible to find a polynomial characterization for optical coherent
states, squeezed states, and bosonic Gaussian and optical coherent
states? These classes of states and the corresponding types of correlations
are relevant in the field of quantum optics \citep{Puri2001}. These
families of states can be also interpreted as coherent states of suitably-chosen
groups represented on the bosonic Fock space. These groups however
are not compact and the representations in question are infinite-dimensional.
For this reason the methods used in this thesis cannot be directly
applied. It would be also interesting to derive for these classes
the ``bilinear correlation witnesses'' in a way analogous to the
considerations presented in Section \ref{sec:Optimal--bilinear}.
\item Derive, using ideas of quantum de-Finetti theorem \citep{Harrow2013}
a complete hierarchy of criteria characterizing convex hulls of classes
of states $\mathcal{M}$ given by Eq.\eqref{eq:final characterization}. 
\item Derive a complete characterization of mixed states (for any number
of particles and for arbitrary dimensions of single-particle Hilbert
spaces) via invariant polynomials of the group of local unitary operations.
\item Develop a full resource theory (analogous to the one existing in the
context of entanglement \citep{EntantHoro} or ancilla-assisted Clifford
computation \citep{Veitch2014}) for the ancilla-assisted FLO or TQC. 
\item The operator $A\in\mathrm{Herm}\left(\mathrm{Sym}^{k}\left(\mathcal{H}\right)\right)$
from Eq.\eqref{eq:final characterization}. describes the properties
of the convex hull $\mathcal{M}^{c}$ completely. However, the value
of $\eta_{\Omega}^{\mathrm{corr}}$ is invariant under the conjugation
of $A$ under ``global unitaries'',
\begin{equation}
A\rightarrow U.A=A'=U^{\otimes k}A\left(U^{\otimes k}\right)^{\dagger}\,,\label{eq:extra action2}
\end{equation}
where $U\in\mathrm{SU}\left(\mathcal{H}\right)$. Therefore, $\eta_{\Omega}^{\mathrm{corr}}$
should depend only on the orbit of $\mathrm{SU}\left(\mathcal{H}\right)$
through the operator $A\in\mathrm{Herm}\left(\mathrm{Sym}^{k}\left(\mathcal{H}\right)\right)$.
It is now natural to ask which polynomial invariants (of the action
\ref{eq:extra action2}) determine the behavior of $\eta_{\Omega}^{\mathrm{corr}}$
(for a fixed $\Omega$).
\end{itemize}

\chapter{Appendix\label{chap:Appendices}}

A number of results presented in Chapters \ref{chap:Multilinear-criteria-for-pure-states}-\ref{chap:Typical-properties-of}
were stated without proofs. In this chapter we give full or sketched
proofs of these results. The chapter is decided into four sections,
each corresponding to one of the chapters listed above. Section are
divided into parts, each containing a proof of a single result. We
did not restate the results whose proofs we are presenting. The same
concerns the notation used in the proofs - it is assumed that the
reader is familiar with the notation used to state a given result.

\section{Proofs of results stated in Chapter \ref{chap:Multilinear-criteria-for-pure-states}\label{sec:Proofs-chapter polyn charact}}

\subsection*{Proof of Lemma \ref{lema crit fermions}\label{sub:Proof-of-Lemma alg form ferm}}
\begin{proof}[Proof of Lemma \ref{lema crit fermions}]
We prove that $\mathbb{P}_{f}:\bigotimes^{2L}\mathcal{H}\rightarrow\bigotimes^{2L}\mathcal{H}$
defined by 
\[
\mathbb{P}_{f}=\mbox{\ensuremath{\alpha}}\left(\mathbb{P}_{11'}^{+}\circ\mathbb{P}_{22'}^{+}\circ\ldots\circ\mathbb{P}_{LL'}^{+}\right)\left(\mathbb{P}_{\left\{ 1,\ldots,L\right\} }^{\mathrm{asym}}\circ\mathbb{P}_{\left\{ 1',\ldots,L'\right\} }^{\mathrm{asym}}\right)
\]
is precisely $\mathbb{P}^{2\lambda_{0}}$. One possible proof relies
on the representation theory of $\mathrm{SU}(N)$. Main technical
tools involved are Young diagrams, Schur-Weyl duality and the theory
of plethysms \citep{Cvitanovic,Fulton1997}. We present a simpler
reasoning based on two simple facts
\begin{itemize}
\item \textit{Fact 1.} The operator $\mathbb{P}_{f}$ is the projector onto
some irreducible representation of $\mathrm{SU}(N)$ in $\bigotimes^{2L}\mathcal{H}$.
\item \textit{Fact 2.} $\mathbb{P}_{f}\left(\ket{\psi_{\lambda_{0}}}\otimes\ket{\psi_{\lambda_{0}}}\right)=\ket{\psi_{\lambda_{0}}}\otimes\ket{\psi_{\lambda_{0}}}$,
where $\ket{\psi_{\lambda_{0}}}=\ket{\psi_{1}}\wedge\ket{\psi_{2}}\wedge\ldots\wedge\ket{\psi_{L}}$
is the highest weight vector of the representation $\mathcal{H}^{\lambda_{0}}$.
\end{itemize}
\noindent Fact 1 follows from the structure of irreducible representations
of $\mathrm{SU}\left(N\right)$ in $\bigotimes^{2L}\mathcal{H}$ (See
Subsection \ref{sub:Representation-theory-of}). Before we prove Fact
2 let us assume for the moment that above two facts are true. Because
$\mathbb{P}_{f}$ preserves $\ket{\psi_{\lambda_{0}}}\otimes\ket{\psi_{\lambda_{0}}}$
and from the vector $\ket{\psi_{\lambda_{0}}}\otimes\ket{\psi_{\lambda_{0}}}$
it is possible to generate (via the action of $\mathrm{SU}(N)$ )
the whole $\mathcal{H}^{2\lambda_{0}}\subset\bigwedge^{L}\left(\mathcal{H}\right)\otimes\bigwedge^{L}\left(\mathcal{H}\right)\subset\bigotimes^{2L}\mathcal{H}$,
one concludes that $\mathbb{P}_{f}=\mathbb{P}^{2\lambda_{0}}$. Let
us now prove the second fact. We fix the basis $\left\{ \ket{\psi_{i}}\right\} _{i=1}^{i=N}$
of $\mathcal{H}$ and we set let $\ket{\psi_{\lambda_{0}}}=\ket{\psi_{1}}\wedge\ket{\psi_{2}}\wedge\ldots\wedge\ket{\psi_{L}}$
be the normalized highest weight vector of the representation $\bigwedge^{L}\left(\mathcal{H}\right)$.
From the definition of the wedge product we have {\scriptsize{}
\begin{gather}
\mathbb{P}_{f}\left(\ket{\psi_{\lambda_{0}}}\otimes\ket{\psi_{\lambda_{0}}}\right)=\mathbb{P}_{f}\left(\ket{\psi_{1}}\wedge\ket{\psi_{2}}\wedge\ldots\wedge\ket{\psi_{L}}\otimes\ket{\psi_{1}}\wedge\ket{\psi_{2}}\wedge\ldots\wedge\ket{\psi_{L}}\right)\,\nonumber \\
=\mathbb{P}_{f}\left(\sum_{\sigma\in S_{L}}\sum_{\tau\in S_{L}}\mathrm{sgn\left(\sigma\right)}\mathrm{sgn}\left(\tau\right)\ket{\psi_{\sigma(1)}}\otimes\ket{\psi_{\sigma(2)}}\otimes\ldots\otimes\ket{\psi_{\sigma(L)}}\otimes\ket{\psi_{\tau(1)}}\otimes\ket{\psi_{\tau(2)}}\otimes\ldots\otimes\ket{\psi_{\tau(L)}}\right)\,\nonumber \\
=\frac{1}{L+1}\sum_{\sigma\in S_{L}}\sum_{\tau\in S_{L}}\mathrm{sgn\left(\sigma\tau\right)\left(\ket{\psi_{\sigma(1)}}\otimes\ket{\psi_{\tau(1)}}+\ket{\psi_{\tau(1)}}\otimes\ket{\psi_{\sigma(1)}}\right)}\otimes\ldots\otimes\left(\ket{\psi_{\sigma(L)}}\otimes\ket{\psi_{\tau(L)}}+\ket{\psi_{\tau(L)}}\otimes\ket{\psi_{\sigma(L)}}\right)\,.\label{eq:koszmar1}
\end{gather}
}In the above expressions, $S_{L}$ denotes permutation group of $L$
elements and $\mathrm{sgn}\left(\cdot\right)$ denotes the sign of
a permutation. In order to simply the notation, we swapped order of
terms in the full tensor product $\bigotimes^{2L}\mathcal{H}$, i.e.
we used the isomorphism:
\[
\bigotimes^{2L}\mathcal{H}=\left(\bigotimes_{i=1}^{i=L}\mathcal{H}_{i}\right)\otimes\left(\bigotimes_{i=1'}^{i=L'}\mathcal{H}_{i}\right)\approx\left(\mathcal{H}_{1}\otimes\mathcal{H}_{1'}\right)\otimes\left(\mathcal{H}_{2}\otimes\mathcal{H}_{2'}\right)\otimes\ldots\otimes\left(\mathcal{H}_{L}\otimes\mathcal{H}_{L'}\right)\,,
\]
for $\mathcal{H}_{i}\approx\mathcal{H}$. Let us introduce the notation{\scriptsize{}
\begin{gather*}
\ket{\Phi_{k,\sigma,\theta}}=\left(\ket{\psi_{\tau(1)}}\otimes\ket{\psi_{\sigma(1)}}\right)\otimes\ldots\otimes\left(\ket{\psi_{\tau(k)}}\otimes\ket{\psi_{\sigma(k)}}\right)\otimes\left(\ket{\psi_{\sigma(k+1)}}\otimes\ket{\psi_{\tau(k+1)}}\right)\otimes\ldots\otimes\left(\ket{\psi_{\sigma(L)}}\otimes\ket{\psi_{\tau(L)}}\right)+\,\\
+\left(\ket{\psi_{\sigma(1)}}\otimes\ket{\psi_{\tau(1)}}\right)\otimes\left(\ket{\psi_{\tau(2)}}\otimes\ket{\psi_{\sigma(2)}}\right)\otimes\ldots\otimes\left(\ket{\psi_{\tau(k+1)}}\otimes\ket{\psi_{\sigma(k+1)}}\right)\otimes\left(\ket{\psi_{\sigma(k+2)}}\otimes\ket{\psi_{\tau(k+2)}}\right)\otimes\ldots+\ldots\,,
\end{gather*}
}where $\ldots$ denotes the summation over the remaining $\binom{L}{k}-2$
terms. One obtains by the different choice of $k$ element combinations
from $\left\{ 1,\ldots,L\right\} $. Reordering of terms in \eqref{eq:koszmar1}
gives
\begin{equation}
\frac{1}{L+1}\sum_{k=0}^{k=L}\left(\sum_{\sigma\in S_{L}}\sum_{\tau\in S_{L}}\mathrm{sgn}\left(\sigma\tau\right)\ket{\Phi_{k,\sigma,\theta}}\right)\,.\label{eq:koszmar2}
\end{equation}
The operator $\mathbb{P}_{f}$ preserves $\bigwedge^{L}\left(\mathcal{H}\right)\otimes\bigwedge^{L}\left(\mathcal{H}\right)$
and therefore,
\[
\mathbb{P}_{f}\left(\ket{\psi_{\lambda_{0}}}\otimes\ket{\psi_{\lambda_{0}}}\right)=\left(\mathbb{P}_{\left\{ 1,\ldots,L\right\} }^{\mathrm{asym}}\circ\mathbb{P}_{\left\{ 1',\ldots,L'\right\} }^{\mathrm{asym}}\right)\circ\mathbb{P}_{f}\left(\ket{\psi_{\lambda_{0}}}\otimes\ket{\psi_{\lambda_{0}}}\right)\,.
\]
As a result from \eqref{eq:koszmar2} we have
\begin{equation}
\frac{1}{L+1}\sum_{k=0}^{k=L}\left(\sum_{\sigma\in S_{L}}\sum_{\tau\in S_{L}}\mathrm{sgn}\left(\sigma\tau\right)\left(\mathbb{P}_{\left\{ 1,\ldots,L\right\} }^{\mathrm{asym}}\circ\mathbb{P}_{\left\{ 1',\ldots,L'\right\} }^{\mathrm{asym}}\right)\ket{\Phi_{k,\sigma,\theta}}\right)\,.\label{eq:koszmar3}
\end{equation}
We claim that for each $k=0,\ldots L$ we have
\begin{equation}
\sum_{\sigma\in S_{L}}\sum_{\tau\in S_{L}}\mathrm{sgn}\left(\sigma\tau\right)\left(\mathbb{P}_{\left\{ 1,\ldots,L\right\} }^{\mathrm{asym}}\circ\mathbb{P}_{\left\{ 1',\ldots,L'\right\} }^{\mathrm{asym}}\right)\left(\ket{\Phi_{k,\sigma,\theta}}\right)=\ket{\psi_{\lambda_{0}}}\otimes\ket{\psi_{\lambda_{0}}}\,.\label{eq:koszmar main}
\end{equation}
Indeed, application of $\mathbb{P}_{\left\{ 1,\ldots,L\right\} }^{\mathrm{asym}}\circ\mathbb{P}_{\left\{ 1',\ldots,L'\right\} }^{\mathrm{asym}}$
gives
\begin{align}
\frac{1}{\left(L!\right)^{2}}\sum_{\sigma\in S_{L}}\sum_{\tau\in S_{L}}\mathrm{sgn}\left(\sigma\tau\right)\left(\left(\ket{\psi_{\tau(1)}}\wedge\ket{\psi_{\tau(2)}}\wedge\ldots\wedge\ket{\psi_{\tau(k)}}\wedge\ket{\psi_{\sigma(k+1)}}\wedge\ldots\right)\right.\otimes\,\label{eq:koszmar4}\\
\otimes\ldots\left.\otimes\left(\ket{\psi_{\sigma(1)}}\wedge\ket{\psi_{\sigma(2)}}\wedge\ldots\wedge\ket{\psi_{\sigma(k)}}\wedge\ket{\psi_{\tau(k+1)}}\wedge\ldots\right)+\ldots\right)\,,\nonumber 
\end{align}
where $\ldots$ denotes the summation over remaining $\binom{L}{k}-1$
terms. Let $S_{L}\left(\sigma,\, k\right)$ denote the subgroup of
$S_{L}$ consisting of permutations that do not mix sets 
\[
\left\{ \sigma(1),\ldots,\sigma(k)\right\} \,\text{and }\left\{ \sigma(k+1),\ldots,\sigma(L)\right\} \,.
\]
We have $S_{L}\left(\sigma,\, k\right)\approx S_{k}\times S_{L-k}$.
As a result, for the fixed $\sigma\in S_{k}$ we have\textbf{\scriptsize{}
\[
\sum_{\tau\in S_{L}}\mathrm{sgn}\left(\sigma\tau\right)\left(\ket{\psi_{\tau(1)}}\wedge\ket{\psi_{\tau(2)}}\wedge\ldots\wedge\ket{\psi_{\tau(k)}}\wedge\ket{\psi_{\sigma(k+1)}}\wedge\ldots\right)\otimes\left(\ket{\psi_{\sigma(1)}}\wedge\ket{\psi_{\sigma(2)}}\wedge\ldots\wedge\ket{\psi_{\sigma(k)}}\wedge\ket{\psi_{\tau(k+1)}}\wedge\ldots\right)
\]
}{\scriptsize \par}

\textbf{\scriptsize{}
\[
=\sum_{\tau\in S_{L}\left(\sigma,k\right)}\mathrm{sgn}\left(\sigma\tau\right)\mathrm{sgn\left(\tau\sigma^{-1}\right)}\left(\ket{\psi_{\sigma(1)}}\wedge\ket{\psi_{\sigma(2)}}\wedge\ldots\wedge\ket{\psi_{\sigma(L)}}\right)\otimes\left(\ket{\psi_{\sigma(1)}}\wedge\ket{\psi_{\sigma(2)}}\wedge\ldots\wedge\ket{\psi_{\sigma(L)}}\right)=\left(L-k\right)!\cdot k!\ket{\psi_{\lambda_{0}}}\otimes\ket{\psi_{\lambda_{0}}}\,.
\]
}Treating all other terms in the outer bracket of \eqref{eq:koszmar4}
in the similar fashion gives
\[
\frac{1}{\left(L!\right)^{2}}\left(\sum_{\sigma\in S_{L}}\binom{L}{k}\left(L-k\right)!\cdot k!\right)\ket{\psi_{\lambda_{0}}}\otimes\ket{\psi_{\lambda_{0}}}=\ket{\psi_{\lambda_{0}}}\otimes\ket{\psi_{\lambda_{0}}}\,,
\]
which proves \eqref{eq:koszmar main}. From \eqref{eq:koszmar main}
and \eqref{eq:koszmar2} we conclude the proof of the second Fact
and therefore prove that $\mathbb{P}_{f}=\mathbb{P}^{2\lambda_{0}}$. 
\end{proof}

\subsection*{Proof of Proposition \ref{geometric interpretation of fermionic Gaussian}
\label{sub:proof of geometric interpretation}}
\begin{proof}
The proposition follows from the known characterization of fermionic
Gaussian states \citep{powernoisy2013}. Every Gaussian state $\rho\in\mathrm{Gauss}$
can be put by the conjugation of $U\in\mathcal{B}$ to the following
form.
\[
\tilde{\rho}=U\rho U^{\dagger}=\frac{1}{2^{d}}\prod_{k=1}^{d}\left(\mathbb{I}+i\lambda_{k}c_{2k-1}c_{2k}\right)\,,
\]
where coefficients $\lambda_{k}$ are eigenvalues of the matrix $\left[h_{kl}\right]$
from Eq.\eqref{eq:mixed gaussian def}. Pure Gaussian states $\kb{\psi}{\psi}$
are characterized by the condition $\lambda_{i}\in\left\{ 1,-1\right\} $.
We now use the fact that the group $\mathcal{B}$ acts on the matrix
$\left[h_{kl}\right]$ via the orthogonal conjugation by arbitrary
element of $\mathrm{SO}\left(2d\right)$. Consequently, for $\mathrm{det}\left(\left[h_{kl}\right]\right)=1$
we get that $\kb{\psi}{\psi}$ can be transformed to the vacuum state
$\kb 00$ and therefore $\kb{\psi}{\psi}\in\mathcal{M}_{g}^{+}$.
Using the analogous reasoning for the pure Gaussian state $\kb{\psi}{\psi}$
for which $\mathrm{det}\left(\left[h_{kl}\right]\right)=-1$ we get
$\kb{\psi}{\psi}\in\mathcal{M}_{g}^{-}$.
\end{proof}

\subsection*{Proof of Proposition \ref{algebraic form operator}\label{sub:Proof-of-Proposition-alg form gauss}}
\begin{proof}
Let us first note that $\Lambda$ is a sum of $2d$ Hermitian and
manually commuting operators $c_{i}\otimes c_{i}$ ($i=1,\ldots,2d$)
having eigenvalues $\pm1$. The projector onto the (one-dimensional)
common eigenspace of operators $\left\{ c_{i}\otimes c_{i}\right\} _{i=1}^{i=2d}$
is given by 
\begin{equation}
\mathbb{P}_{\vec{\mu}}=\frac{1}{2^{2d}}\prod_{i=1}^{2d}\left(\mathbb{I}\otimes\mathbb{I}+\mu_{i}c_{i}\otimes c_{i}\right)=\sum_{k=0}^{2d}\sum_{\begin{array}[t]{c}
X\subset\left\{ 1,\ldots,2d\right\} \\
\left|X\right|=k
\end{array}}\prod_{i\in X}\mu_{i}c_{i}\otimes c_{i}\,,\label{eq:projector gauss one dim}
\end{equation}
where integer vector $\vec{\mu}$ labels all possible collections
of joint spectra of operators $\left\{ c_{i}\otimes c_{i}\right\} _{i=1}^{i=2d}$,
i.e.
\[
\vec{\mu}=\left(\mu_{1},\ldots,\mu_{2d}\right)\,,
\]
with $\mu_{i}=\pm1$. Consequently we have 
\begin{equation}
\mathbb{P}_{0}=\sum_{\vec{\mu},\,\sum_{i=1}^{2d}\mu_{i}=0}\mathbb{P}_{\vec{\mu}}\,,\label{eq:combinatorial sum gauss}
\end{equation}
where the summation is over all integer $\vec{\mu}$ having entries
$\mu_{i}=\pm1$ and satisfying $\sum_{i=1}^{2d}\mu_{i}=0$. The combinatorial
sum \eqref{eq:combinatorial sum gauss} is computed via the elegant
method that uses integration over complex variables%
\footnote{We are grateful to Maciek Lisicki for recalling to the author this
method of integration while discussing the solution to the problem
of the ``Kac ring''\citep{Kac1957}. %
} which we now present. Recall the following elementary result from
complex analysis, valid for arbitrary natural number $n_{1}$, 
\begin{equation}
\frac{1}{2\pi i}\int_{S_{1}}\frac{dz}{z^{n_{1}+1}}=\begin{cases}
1 & \text{for }n_{1}=0\\
0 & \text{otherwise}
\end{cases}\,,\label{eq:residue formula}
\end{equation}
where $\int_{S_{1}}dz$ denotes the integration over the unit circle
oriented anticlockwise. We have the following chain of equalities
\begin{align}
\mathbb{P}_{0} & =\sum_{\vec{\mu},\,\sum_{i=1}^{2d}\mu_{i}=0}\mathbb{P}_{\vec{\mu}}\,,\label{eq:comb0}\\
 & =\sum_{\vec{\mu}}\mathbb{P}_{\vec{\mu}}\frac{1}{2\pi i}\int_{S_{1}}\frac{dz}{z^{\sum_{i}\mu_{1}+1}}\,,\label{eq:comb1}\\
 & =\frac{1}{2\pi i}\int_{S_{1}}dz\sum_{\vec{\mu}}\frac{1}{z^{\sum_{i}\mu_{1}+1}}\mathbb{P}_{\vec{\mu}}\,,\label{eq:comb2}\\
 & =\frac{1}{2\pi i}\int_{S_{1}}dz\sum_{\vec{\mu}}\frac{1}{z^{\sum_{i}\mu_{1}+1}}\left(\frac{1}{2^{2d}}\sum_{k=0}^{2d}\sum_{\begin{array}[t]{c}
X\subset\left\{ 1,\ldots,2d\right\} \\
\left|X\right|=k
\end{array}}\prod_{i\in X}\mu_{i}c_{i}\otimes c_{i}\right)\,,\label{eq:comb 3}\\
 & =\frac{1}{2^{2d}}\sum_{k=0}^{2d}\sum_{\begin{array}[t]{c}
X\subset\left\{ 1,\ldots,2d\right\} \\
\left|X\right|=k
\end{array}}\left(\frac{1}{2\pi i}\int_{S_{1}}dz\sum_{\vec{\mu}}\frac{\prod_{i\in X}\mu_{i}}{z^{\sum_{i}\mu_{1}+1}}\right)\prod_{i\in X}c_{i}\otimes c_{i}\,.\label{eq:comb4}
\end{align}
In \eqref{eq:comb1} we have applied to \eqref{eq:comb0} the identity
\eqref{eq:residue formula}. In the latter manipulations we changed
a number of times the order of summation and respectively integration.
In \eqref{eq:comb 3} we simply applied to \eqref{eq:comb2} the binomial
expansion of $\mathbb{P}_{\vec{\mu}}$ (see \eqref{eq:projector gauss one dim}).
For a fixed $X\subset\left\{ 1,\ldots,2d\right\} $ the integral 
\[
\frac{1}{2\pi i}\int_{S_{1}}dz\sum_{\vec{\mu}}\frac{\prod_{i\in X}\mu_{i}}{z^{\sum_{i}\mu_{1}+1}}
\]
depends only on the cardinality $\left|X\right|=k$ of the set $X$,
\begin{align*}
\frac{1}{2\pi i}\int_{S_{1}}dz\sum_{\vec{\mu}}\frac{\prod_{i\in X}\mu_{i}}{z^{\sum_{i}\mu_{1}+1}} & =\frac{1}{2\pi i}\int_{S_{1}}\frac{dz}{z^{2d+1}}\left(1-z^{2}\right)^{k}\left(1+z^{2}\right)^{2d-k}\,.
\end{align*}
By expanding the polynomial in the numerator of the above expression
and using \eqref{eq:residue formula} we get
\begin{equation}
\frac{1}{2\pi i}\int_{S_{1}}dz\sum_{\vec{\mu}}\frac{\prod_{i\in X}\mu_{i}}{z^{\sum_{i}\mu_{1}+1}}=\begin{cases}
0 & \text{for }\left|X\right|=k\,\text{odd}\\
\left(-1\right)^{\frac{k}{2}}\frac{k!\left(2d-k\right)!}{d!\left(\frac{k}{2}\right)!\left(d-\frac{k}{2}\right)!} & \text{for }\left|X\right|=k\,\text{even}
\end{cases}.\label{eq:integral closed form}
\end{equation}
Inserting \eqref{eq:integral closed form} into \eqref{eq:comb4}
and introduction of a new variable $\tilde{k}=2k$ gives precisely
\eqref{eq:closed form expression ferm}.
\end{proof}

\subsection*{Proof of Proposition \ref{schmidt number explicit}\label{sub:Proof-of-Proposition schmidt number}}
\begin{proof}
We prove Eq.\eqref{eq:explicit form Schmidt} by direct computation.
In what follows we will consequently use the isomorphism \eqref{eq:identification of tensor power}.
Before we proceed, let us introduce an auxiliary notation $E_{i,j}=\ket i\ket i\bra j\bra j$.
Using this notation we get
\[
\kb{\psi}{\psi}=\sum_{i,j=1}^{d}\lambda_{i}\lambda_{j}E_{ij}\,.
\]
Inserting the above formula to $\bra{\psi{}^{\otimes\left(n+1\right)}}A_{n}\ket{\psi{}^{\otimes\left(n+1\right)}}$
we get
\begin{equation}
\sum_{i_{1},j_{1}=1}^{d}\ldots\sum_{i_{n+1},j_{n+1}=1}^{d}\left(\prod_{k=1}^{n+1}\lambda_{i_{k}}\right)\left(\prod_{l=1}^{n+1}\lambda_{j_{l}}\right)\mathrm{tr}\left(\left[E_{i_{1},,j_{i}}\otimes\ldots\otimes E_{i_{n+1},,j_{n+1}}\right]\mathbb{P}_{A}^{\mathrm{asym,n+1}}\otimes\mathbb{P}_{B}^{\mathrm{asym,n+1}}\right)\,.\label{eq:koszmar schmidt1}
\end{equation}
From the definitions of $E_{ij}$ and operator $\mathbb{P}_{A}^{\mathrm{asym,n+1}}\otimes\mathbb{P}_{B}^{\mathrm{asym,n+1}}$
we get{\footnotesize{}
\begin{equation}
\mathrm{tr}\left(\left[E_{i_{1},j_{i}}\otimes\ldots\otimes E_{i_{n+1},j_{n+1}}\right]\mathbb{P}_{A}^{\mathrm{asym,n+1}}\otimes\mathbb{P}_{B}^{\mathrm{asym,n+1}}\right)=\frac{1}{\left[\left(n+1\right)!\right]^{2}}\delta_{\left\{ i_{1},\ldots,i_{n+1}\right\} ,\left\{ j_{1},\ldots,j_{n+1}\right\} }\delta_{\left|\left\{ i_{1},\ldots,i_{n+1}\right\} \right|,n+1}\,,\label{eq:koszmar Schmidt 2}
\end{equation}
}where
\[
\delta_{\left\{ i_{1},\ldots,i_{n+1}\right\} ,\left\{ j_{1},\ldots,j_{n+1}\right\} }
\]
is the Kronecker delta with arguments being subsets of the set $\left\{ 1,\ldots,d\right\} $.
Applying \eqref{eq:koszmar Schmidt 2} to \eqref{eq:koszmar schmidt1}
we get
\begin{equation}
\bra{\psi{}^{\otimes\left(n+1\right)}}A_{n}\ket{\psi{}^{\otimes\left(n+1\right)}}=\frac{1}{\left[\left(n+1\right)!\right]^{2}}\sum_{\begin{array}[t]{c}
X\subset\left\{ 1,\ldots,d\right\} \\
\left|X\right|=n+1
\end{array}}\left(\prod_{i\in X}\lambda_{i}^{2}\right)\cdot\left(n+1\right)!\,,\label{eq:koszmar Schmid 3}
\end{equation}
where the factor $\left(n+1\right)!$ on the right hand side of \eqref{eq:koszmar Schmid 3}
comes from the fact that the conditions enforced by \eqref{eq:koszmar Schmidt 2}
are the following
\begin{itemize}
\item There is the equality of sets $\left\{ i_{1},\ldots,i_{n+1}\right\} ,\left\{ j,\ldots,j_{n+1}\right\} $
(for fixed set $\left\{ i_{1},\ldots,i_{n+1}\right\} $ there is a
freedom in permutation of elements of the sequence $\left(j_{1},\ldots,j_{n+1}\right)$.
\item Both sets $\left\{ i_{1},\ldots,i_{n+1}\right\} ,\left\{ j,\ldots,j_{n+1}\right\} $
must contain exactly $n+1$ elements.
\end{itemize}
Equation \eqref{eq:koszmar Schmid 3} is equivalent to \eqref{eq:explicit form Schmidt}. 
\end{proof}

\subsection*{Proof of Lemma \ref{k lin characterization coheren}\label{sub:Proof-of-Lemma k casimir}}
\begin{proof}
Let us first notice that, by the virtue of Proposition \eqref{prop:K copies}
and Eq.\eqref{eq:casimir irrep value} the we have the equivalence
\begin{equation}
\bra{\psi^{\otimes k}}\left(\mathbb{P}^{\mathrm{sym,k}}-\mathbb{P}^{k\lambda_{0}}\right)\ket{\psi^{\otimes k}}=0\,\Longleftrightarrow\,\bra{\psi^{\otimes k}}L_{k}\ket{\psi^{\otimes k}}=\left(k\lambda_{0},\, k\lambda_{0}+2\delta\right)\,,\label{eq:auxiliary k lin equiv}
\end{equation}
where $L_{k}:\left(\mathcal{H}^{\lambda_{0}}\right)^{\otimes k}\rightarrow\left(\mathcal{H}^{\lambda_{0}}\right)^{\otimes k}$
is the representation of the second order Casimir in $\left(\mathcal{H}^{\lambda_{0}}\right)^{\otimes k}$,
\begin{equation}
L_{k}=-\sum_{i=1}^{M}\left(\pi\left(X_{i}\right)\otimes\mathbb{I}^{\otimes\left(k-1\right)}+\mathbb{I}\otimes\pi\left(X_{i}\right)\otimes\mathbb{I}^{\otimes\left(k-2\right)}+\ldots+\mathbb{I}^{\otimes\left(k-1\right)}\otimes\pi\left(X_{i}\right)\right)^{2}\,,\label{eq:k lin representation}
\end{equation}
where we used the same notation as in Eq.\eqref{eq:second order cas rep}.
Using the definition \eqref{eq:k lin representation} we get{\footnotesize{}
\begin{align*}
\bra{\psi^{\otimes k}}L_{k}\ket{\psi^{\otimes k}} & =\bra{\psi^{\otimes k-1}}L_{k-1}\ket{\psi^{\otimes k-1}}+2\left(k-1\right)\sum_{i=1}^{M}\bra{\psi}\pi\left(X_{i}\right)\ket{\psi}^{2}+\bra{\psi}\sum_{i=1}^{M}\pi\left(X_{i}\right)^{2}\ket{\psi}\,.
\end{align*}
}Using now induction over the natural number $k$ we get
\begin{align}
\bra{\psi^{\otimes k}}L_{k}\ket{\psi^{\otimes k}} & =-k\left(k-1\right)\sum_{i=1}^{M}\bra{\psi}\pi\left(X_{i}\right)\ket{\psi}^{2}-k\bra{\psi}\sum_{i=1}^{M}\pi\left(X_{i}\right)^{2}\ket{\psi}\\
 & =-k\left(k-1\right)\sum_{i=1}^{M}\bra{\psi}\pi\left(X_{i}\right)\ket{\psi}^{2}+k\left(\lambda_{0},\,\lambda_{0}+2\delta\right)\,,\label{eq:trivial casimir usage}\\
 & =\frac{k\left(k-1\right)}{2}\bra{\psi^{\otimes2}}L_{2}\ket{\psi^{\otimes2}}+k\left(k-2\right)\left(\lambda_{0},\,\lambda_{0}+2\delta\right)\,.\label{eq:final step}
\end{align}
where in \eqref{eq:trivial casimir usage} we used the fact that $\mathcal{H}^{\lambda_{0}}$
is an irreducible representation of a Lie group $K$ (or equivalently
$\mathfrak{k}=\mathrm{Lie}\left(K\right)$ or $\mathfrak{g}=\mathfrak{k}^{\mathbb{C}}$).
In \eqref{eq:final step} we used \eqref{eq:trivial casimir usage}
for $k=2$. We conclude the proof by noting that by the virtue of
Eq.\eqref{eq:maximum casimir} and \eqref{eq:final step} the maximum
value of $\bra{\psi^{\otimes k}}L_{k}\ket{\psi^{\otimes k}}$ (equal
$\left(k\lambda_{0},\, k\lambda_{0}+2\delta\right)$) is achieved
for $\kb{\psi}{\psi}\in\mathcal{M}_{\lambda_{0}}$.
\end{proof}

\section{Proofs of results stated in Chapter \ref{chap:Complete-characterisation}\label{sec:Proofs-concerning-Chapter rigorous}}

\subsection*{{\normalsize{}Proof of Lemma}\normalsize{} \ref{eneralised Schmidt decomposition}
\label{sub:Proof-ofgeneral schmidt} }

Let us note that in the considered setting we have an antiunitary
conjugation $\theta_{+}$ detecting correlations, as in Theorem \ref{theta epimorphism}.
From the proof of this theorem (see Eq.\eqref{eq:real decomposition})
we get, that every normalized vector $\ket{\psi}\in\mathcal{H}_{\mathrm{Fock}}^{+}\left(\mathbb{C}^{4}\right)$
can be written, up to phase, as a combination
\begin{equation}
\ket{\psi}=\sqrt{1-a^{2}}\ket{\psi_{1}}+ia\ket{\psi_{2}}\,,\label{eq: real decomp}
\end{equation}
 where $0\le a\le\frac{1}{\sqrt{2}}$, and $\ket{\psi_{1}},\ket{\psi_{2}}$
are orthogonal states satisfying $\theta_{+}\ket{\psi_{\alpha}}=\ket{\psi_{\alpha}}$,
$\alpha=1,2$ (the latter condition corresponds to $\kb{\psi_{\alpha}}{\psi_{\alpha}}$
having real coefficients in the decomposition \eqref{eq:even decomposition2}.
Furthermore, from the proof of Theorem \ref{theta epimorphism} we
have that $\mathrm{Spin}(8)$ (or equivalently the group of Bogolyubov
transformations $\mathcal{B}$ ) acts transitively on orthogonal pairs
$\ket{\psi_{1}},\ket{\psi_{2}}$. We conclude that the set of all
states $\ket{\psi}$ corresponding to a given value of $a$ is an
orbit of $\mathcal{B}$. In particular, $a=\frac{1}{\sqrt{2}}$ describes
Gaussian states, while $a=0$ is the orbit of $\ket{a_{8}}$. We conclude
the proof by re-expressing the state vector $\ket{\psi}$, up to a
phase, as a combination of orthogonal Gaussian states. Namely, setting
$\ket{\psi_{G}}=\frac{1}{\sqrt{2}}\left(\ket{\psi_{1}}-i\ket{\psi_{2}}\right)$,
we have
\begin{equation}
\ket{\psi}=\sqrt{1-p^{2}}\ket{\psi_{G}}+p\theta_{+}\ket{\psi_{G}}\,,\label{eq:generalised Schmidt2}
\end{equation}

where $p=\frac{1}{\sqrt{2}}\left(\sqrt{1-a^{2}}-a\right)$. Gaussianity
and orthogonality of $\ket{\psi_{G}},\,\theta_{+}\ket{\psi_{G}}$
are immediately verified.

\subsection*{{\normalsize{}Proof of Lemma} \normalsize{}\ref{eq:Fidelity}
\label{sub:Proof-of-Lemma fidelity}}

The proof of \eqref{eq:Fidelity} relies on Theorem 2 from \citep{Streltsov2010}
which, for our purposes, states that $F_{\mathrm{Gauss}}\left(\rho\right)$
can be described as a convex hull extension of a function defined
on pure states:
\begin{equation}
F_{\mathrm{Gauss}}\left(\rho\right)=\mathrm{sup}_{\sum p_{i}\kb{\psi_{i}}{\psi_{i}}=\rho}\left(\sum_{i}p_{i}F_{\mathrm{Gauss}}\left(\ket{\psi_{i}}\right)\right)\,,\label{eq:fidedity decomp}
\end{equation}
where
\[
F_{\mathrm{Gauss}}\left(\ket{\psi}\right)=\mathrm{max}_{\kb{\phi}{\phi}\in\mathcal{G}}\, F\left(\kb{\psi}{\psi},\,\kb{\phi}{\phi}\right)\,.
\]
Using the fact that for $\ket{\psi}\in\mathrm{Fock}_{+}\left(\mathbb{C}^{4}\right)$
we have the decomposition \eqref{eq:generalised Schmidt} we find
that
\[
F_{\mathrm{Gauss}}\left(\ket{\psi}\right)=\tilde{F}\left(C_{+}\left(\ket{\psi}\right)\right),
\]
 where $\tilde{F}(x)=\frac{1}{2}+\frac{1}{2}\sqrt{1-x^{2}}$ is a
strictly concave decreasing function on the interval $\left[0,1\right]$.
Let $\rho=\sum_{i}p_{i}\kb{\psi_{i}}{\psi_{i}}$ be the optimal decomposition
of $\rho$ leading to \eqref{eq:gen concurrence}. From \citep{Uhlmann2000}
it follows the all pure states in this decomposition have the same
value of the generalized concurrence, i.e. $C_{+}\left(\ket{\psi_{i}}\right)=C_{+}\left(\rho\right)$.
Using this fact and \eqref{eq:fidedity decomp} we have $F_{\mathrm{Gauss}}\left(\rho\right)\geq\tilde{F}\left(C_{+}\left(\rho\right)\right)$.
On the other hand, by the concavity of $\tilde{F}$ we have 
\[
\sum_{i}p_{i}F_{\mathrm{Gauss}}\left(\ket{\psi_{i}}\right)\leq\tilde{F}\left(\sum_{i}p_{i}C_{+}\left(\ket{\psi_{i}}\right)\right)=\tilde{F}\left(C_{+}\left(\rho\right)\right)\,,
\]
which concludes the proof of \eqref{eq:Fidelity}.

\section{Proofs of results stated in Chapter \ref{chap:Polynomial-mixed states}\label{sec:Proofs-concerning-Chapter multilinear witnesses}}

\subsection*{Proof of Proposition \ref{simple optimality distinguishable particles}\label{sub:Proof-of-Proposition-optimaldist}}
\begin{proof}
Let us fix some orthonormal bases
\[
\left\{ \ket{e_{j}^{(i)}}\right\} _{j=1}^{j=N_{j}}\,,\, i=1,\ldots,L
\]
of the spaces $\mathcal{H}_{i}$ composing the full tensor product
$\mathcal{H}_{d}=\bigotimes_{i=1}^{L}\mathcal{H}_{i}$. Recall that
in the considered case the operator $A$ is given by 
\[
A=\mathbb{P}^{\mathrm{sym}}-\mathbb{P}_{11'}^{+}\otimes\mathbb{P}_{22'}^{+}\otimes\ldots\otimes\mathbb{P}_{LL'}^{+}\,.
\]
Note also that the operator $A$ is $\mathrm{LU}$ invariant, i.e.
$U\otimes UAU^{\dagger}\otimes U^{\dagger}=A$ for $U\in\mathrm{LU}$
and that $\mathcal{M}_{d}$ is an orbit of $\mathrm{LU}$ in $\mathcal{D}_{1}\left(\mathcal{H}_{d}\right)$.
Due to the definition of the constant $c$ 
\[
c=2\underset{\kb vv\in\mathcal{M}_{d}}{\mathrm{max}}\,\underset{\ket{v_{\perp}}\neq0}{\mathrm{max}}\frac{\bra v\bra{v_{\perp}}A\ket v\ket{v_{\perp}}}{\bk{v_{\perp}}{v_{\perp}}}=\underset{\ket{v_{\perp}}\neq0}{\mathrm{max}}\frac{\bra{\psi_{0}}\bra{v_{\perp}}A\ket{\psi_{0}}\ket{v_{\perp}}}{\bk{v_{\perp}}{v_{\perp}}}\,,
\]
where $\ket{\psi_{0}}=\otimes_{i=1}^{L}\ket{e_{1}^{\left(i\right)}}$and
$\ket{v_{\perp}}$ is an arbitrary vector perpendicular to $\ket{\psi_{0}}$.
Straightforward computation shows that%
\footnote{See the proof of Theorem \ref{fig:finitelly generated cone} for a
deeper justification for this phenomenon.%
}
\[
A\ket{\psi_{0}}\left(\ket{e_{j_{1}}^{\left(1\right)}}\otimes\ket{e_{j_{2}}^{\left(2\right)}}\otimes\ldots\otimes\ket{e_{j_{L}}^{\left(L\right)}}\right)\perp\ket{\psi_{0}}\left(\ket{e_{j'_{1}}^{\left(1\right)}}\otimes\ket{e_{j'_{2}}^{\left(2\right)}}\otimes\ldots\otimes\ket{e_{j'_{L}}^{\left(L\right)}}\right)
\]
whenever $\left(j_{1},\ldots,j_{L}\right)\neq\left(j'_{1},\ldots,j'_{L}\right)$.
Consequently, we have
\[
c=2\underset{\left(j_{1},\ldots,j_{L}\right)}{\mathrm{max}}\bra{\psi_{0}}\left(\bra{e_{j{}_{1}}^{\left(1\right)}}\otimes\bra{e_{j{}_{2}}^{\left(2\right)}}\otimes\ldots\otimes\bra{e_{j{}_{L}}^{\left(L\right)}}\right)A\ket{\psi_{0}}\left(\ket{e_{j{}_{1}}^{\left(1\right)}}\otimes\ket{e_{j{}_{2}}^{\left(2\right)}}\otimes\ldots\otimes\ket{e_{j{}_{L}}^{\left(L\right)}}\right)\,.
\]
The minimum $c=1-2^{1-L}$ is obtained for every array $\left(j_{1},\ldots,j_{L}\right)$
that differs form $\left(1,\ldots,1\right)$ in every component.
\end{proof}

\subsection*{Proof of Proposition \ref{simple optimal ferm}\label{sub:Proof-of-Proposition-optimal ferm}}
\begin{proof}
The proof is essentially analogous to the proof of Proposition \ref{sub:Entanglement-of-distinguishable-bilin}.
Let us chose the following basis of $\mathcal{H}_{f}=\bigwedge^{L}\left(\mathbb{C}^{d}\right)$,
\begin{equation}
\ket{\psi_{I}}=\ket{\phi_{i_{1}}}\wedge\ket{\phi_{i_{2}}}\wedge\ldots\wedge\ket{\phi_{i_{L}}}\,,\label{eq:slater basis again}
\end{equation}
where $\ket{\phi_{i}}\in\mathbb{C}^{d}$, $I=\left\{ i_{1},i_{2}\ldots,i_{L}\right\} \subset\left\{ 1,\ldots,d\right\} $
and $\left|I\right|=L$. Recall that in the considered case the operator
$A$ is given by 
\[
A=\mathbb{P}^{\mathrm{sym}}-\frac{2^{L}}{L+1}\mathbb{P}_{11'}^{+}\otimes\mathbb{P}_{22'}^{+}\otimes\ldots\otimes\mathbb{P}_{LL'}^{+}\left(\mathbb{P}_{\left\{ 1,\ldots,L\right\} }^{\mathrm{asym}}\otimes\mathbb{P}_{\left\{ 1',\ldots,L'\right\} }^{\mathrm{asym}}\right)\,.
\]
Note also that the operator $A$ is $\mathrm{LU}_{f}$ -invariant,
i.e. $U\otimes UAU^{\dagger}\otimes U^{\dagger}=A$ for $U\in\mathrm{LU}_{f}$
and that $\mathcal{M}_{f}$ is an orbit of $\mathrm{LU}_{f}$ in $\mathcal{D}_{1}\left(\mathcal{H}_{f}\right)$.
Due to the definition of the constant $c$ 
\[
c=2\underset{\kb vv\in\mathcal{M}_{f}}{\mathrm{max}}\,\underset{\ket{v_{\perp}}\neq0}{\mathrm{max}}\frac{\bra v\bra{v_{\perp}}A\ket v\ket{v_{\perp}}}{\bk{v_{\perp}}{v_{\perp}}}=\underset{\ket{v_{\perp}}\neq0}{\mathrm{max}}\frac{\bra{\psi_{0}}\bra{v_{\perp}}A\ket{\psi_{0}}\ket{v_{\perp}}}{\bk{v_{\perp}}{v_{\perp}}}\,,
\]
where $\ket{\psi_{0}}=\ket{\phi_{1}}\wedge\ket{\phi_{2}}\wedge\ldots\wedge\ket{\phi_{L}}$and
$\ket{v_{\perp}}$ is an arbitrary vector perpendicular to $\ket{\psi_{0}}$.
Straightforward computation shows that%
\footnote{See the proof of Theorem \ref{fig:finitelly generated cone} for a
deeper justification for this phenomenon.%
}
\[
A\ket{\psi_{0}}\ket{\psi_{I}}\perp\ket{\psi_{0}}\ket{\psi_{I'}}\,,
\]
for different $L$-element subsets subsets, $I,I'\subset\left\{ 1,\ldots d\right\} $.
Consequently, we have
\begin{equation}
c=2\underset{I\subset\left\{ 1,\ldots d\right\} }{\mathrm{max}}\bra{\psi_{0}}\bra{\psi_{I}}A\ket{\psi_{0}}\ket{\psi_{I}}\,.\label{eq:minimal fermions step}
\end{equation}
Computations analogous to those performed in the proof of Lemma \ref{lema crit fermions}
give the following formula,
\begin{equation}
\bra{\psi_{0}}\bra{\psi_{I}}A\ket{\psi_{0}}\ket{\psi_{I}}=\frac{1}{2}-\frac{1}{L+1-k}\,,\label{eq:expectation ferm}
\end{equation}
where $k=\left|I\cap I_{0}\right|$. Consequently, the maximal value
in \ref{eq:minimal fermions step} is obtained for the minimal possible
$k$ which equals (this follows from the definition of the wedge product)
$\max\left\{ 0,2L-d\right\} $. As a result we have
\[
c=1-\frac{2}{L+1-\max\left\{ 0,2L-d\right\} }\,.
\]

\end{proof}

\subsection*{Proof of Proposition \ref{simple optiml gauss}\label{sub:Proof-of-Proposition-optimapl gauss}}
\begin{proof}[Sketch of the proof]
The proof is essentially analogous to the proof of Proposition \ref{sub:Entanglement-of-distinguishable-bilin}
presented above. Let us chose the following basis of $\mathcal{H}_{\mathrm{Fock}}^{+}\left(\mathbb{C}^{d}\right)$,
\begin{equation}
\ket{\psi_{I}}=\prod_{j\in J}a_{j}^{\dagger}\ket 0\,,\label{eq:Fock states gauss parity}
\end{equation}
where $J\subset\left\{ 1,\ldots,d\right\} $ are subsets of the set
$\left\{ 1,\ldots,d\right\} $ that have even number of elements 
\[
\left|J\right|=0,2,\ldots,2\left\lfloor \frac{d}{2}\right\rfloor \,.
\]
Recall that the relevant operator $A:\mathrm{Sym}^{2}\left(\mathcal{H}_{\mathrm{Fock}}^{+}\left(\mathbb{C}^{d}\right)\right)\rightarrow\mathrm{Sym}^{2}\left(\mathcal{H}_{\mathrm{Fock}}^{+}\left(\mathbb{C}^{d}\right)\right)$
is given by \eqref{eq:appropiate A}. Note also that the operator
$A$ is $\mathrm{Spin}\left(2d\right)$ invariant, i.e. $U\otimes UAU^{\dagger}\otimes U^{\dagger}=A$
for $U\in\mathrm{Spin}\left(2d\right)$ and $\mathcal{M}_{g}^{+}$
is an orbit of $\mathrm{\mathrm{Spin}\left(2d\right)}$ in $\mathcal{D}_{1}\left(\mathcal{H}_{\mathrm{Fock}}^{+}\left(\mathbb{C}^{d}\right)\right)$.
Therefore we have 
\[
c=2\underset{\kb vv\in\mathcal{M}_{g}^{+}}{\mathrm{max}}\,\underset{\ket{v_{\perp}}\neq0}{\mathrm{max}}\frac{\bra v\bra{v_{\perp}}A\ket v\ket{v_{\perp}}}{\bk{v_{\perp}}{v_{\perp}}}=\underset{\ket{v_{\perp}}\neq0}{\mathrm{max}}\frac{\bra 0\bra{v_{\perp}}A\ket 0\ket{v_{\perp}}}{\bk{v_{\perp}}{v_{\perp}}}\,,
\]
where $\ket{v_{\perp}}$ is an arbitrary vector perpendicular to $\ket 0$.
A Straightforward computation shows that%
\footnote{See the proof of Theorem \ref{fig:finitelly generated cone} for a
deeper justification for this phenomenon.%
}
\[
A\ket 0\ket{\psi_{I}}\perp\ket 0\ket{\psi_{I}}\,,
\]
where $I\neq\emptyset$. Consequently we have 
\begin{equation}
c=2\underset{\begin{array}[t]{c}
I\subset\left\{ 1,\ldots,d\right\} \\
\left|I\right|\text{-even}
\end{array}}{\mathrm{max}}\bra{\psi_{0}}\bra{\psi_{I}}A\ket{\psi_{0}}\ket{\psi_{I}}\,.\label{eq:almost final expression}
\end{equation}
Equation \eqref{eq:constant gauss} follows from the above expression
when we implement the precise form of the operator $A$ and compare
numerically (for $d\leq1000$) appropriate matrix elements (see the
proof of Lemma \ref{FLO invariant bilin theorem} for details of the
computations).
\end{proof}

\subsection*{Proof of Lemmas \ref{lem:commutant bosons-} and \ref{commutant ferm}\label{sub:Proof-of-Lemmas-commutants}}

Before we prove Lemmas \ref{lem:commutant bosons-} and \ref{commutant ferm}
we first state and prove an auxiliary technical result that will be
also used also in proof of Lemma \ref{commutant fer gauss}.
\begin{lem}
\label{technical lemma}Let $\mathcal{A}\subset\mathrm{End}\left(\mathcal{H}\right)$
be a matrix $\mathbb{C}^{\ast}$-algebra%
\footnote{We use the term matrix $\mathbb{C}^{\ast}$-algebra in following sense:
$\mathcal{A}$ is an matrix subalgebra of $\mathrm{End}\left(\mathcal{H}\right)$
which is closed under the operation of Hermitian conjugation, $A\rightarrow A^{\dagger}$.%
} and let $\mathcal{B}=\mathrm{Comm}\left(\mathcal{A}\right)\subset\mathrm{End}\left(\mathcal{H}\right)$
be the commutant of $\mathcal{A}$ in $\mathrm{End}\left(\mathcal{H}\right)$.
Let $\mathcal{H}'\subset\mathcal{H}$ be a subspace of $\mathcal{H}$
and let $\mathbb{P}_{\mathcal{H}'}:\mathcal{H}\rightarrow\mathcal{H}$
denotes the orthonormal projector onto $\mathcal{H}'$. Assume $\mathbb{P}_{\mathcal{H}'}\in\mathcal{B}$.
Let $\mathrm{Comm'}\left(\mathcal{C}\right)$ denotes the commutant
of a $\mathbb{C}^{\ast}$-algebra $\mathcal{C}\subset\mathrm{End}\left(\mathcal{H}'\right)$
in $\mathrm{End}\left(\mathcal{H}'\right)$ ($\mathrm{Comm'}\left(\mathcal{C}\right)\subset\mathrm{End}\left(\mathcal{H}'\right))$.
The following equality holds
\begin{equation}
\mathrm{Comm'}\left(\mathcal{A}'\right)=\mathcal{B}'\,,\label{eq:comutant equation}
\end{equation}
where%
\footnote{We use the notation $\mathbb{P}_{\mathcal{H}'}\mathcal{A}\mathbb{P}_{\mathcal{H}'}$
to denote a $\mathbb{C}^{\ast}$-algebra in $\mathrm{End}\left(\mathcal{H}'\right)$
which is generated by operators of the form $\mathbb{P}_{\mathcal{H}'}X\mathbb{P}_{\mathcal{H}'}$,
where $X\in\mathcal{A}$. %
} $\mathcal{A}'=\mathbb{P}_{\mathcal{H}'}\mathcal{A}\mathbb{P}_{\mathcal{H}'}\subset\mathrm{End}\left(\mathcal{H}'\right)$
is a restriction of $\mathcal{A}$ to $\mathcal{H}'$ and $\mathcal{B}'=\mathbb{P}_{\mathcal{H}'}\mathcal{B}\mathbb{P}_{\mathcal{H}'}$
is a restriction of $\mathcal{B}$ to $\mathcal{H}'$. \end{lem}
\begin{proof}
Due to the fact that $\mathbb{P}_{\mathcal{H}'}\in\mathrm{Comm}\left(\mathcal{A}\right)=\mathcal{B}$
every element $X\in\mathbb{P}_{\mathcal{H}'}\mathcal{A}\mathbb{P}_{\mathcal{H}'}=\mathcal{A}'$
has a form
\begin{equation}
X=\mathbb{P}_{\mathcal{H}'}\tilde{X}\mathbb{P}_{\mathcal{H}'}=\mathbb{P}_{\mathcal{H}'}\tilde{X}=\tilde{X}\mathbb{P}_{\mathcal{H}'}\,,\label{eq:X presentation}
\end{equation}
for $\tilde{X}\in\mathcal{A}$. Moreover, for the same reason every$Y\in\mathbb{P}_{\mathcal{H}'}\mathcal{B}\mathbb{P}_{\mathcal{H}'}$
belongs to the algebra $\mathcal{B}$. Consequently, we have $\left[X,Y\right]=0$
and thus $\mathbb{P}_{\mathcal{H}'}\mathcal{B}\mathbb{P}_{\mathcal{H}'}\subset\mathrm{Comm'}\left(\mathcal{A}'\right)$.
Conversely, let $Y\in\mathrm{Comm'}\left(\mathcal{A}'\right)$ using
\eqref{eq:X presentation} and the fact that $Y\mathbb{P}_{\mathcal{H}'}=\mathbb{P}_{\mathcal{H}'}Y=Y$
we get the following chain of equalities
\begin{align}
0 & =\left[Y,X\right]=\left[Y,\mathbb{P}_{\mathcal{H}'}\tilde{X}\mathbb{P}_{\mathcal{H}'}\right]\,,\nonumber \\
 & =Y\mathbb{P}_{\mathcal{H}'}\tilde{X}\mathbb{P}_{\mathcal{H}'}-\mathbb{P}_{\mathcal{H}'}\tilde{X}\mathbb{P}_{\mathcal{H}'}Y\,,\nonumber \\
 & =\left[Y,\tilde{X}\right]\,.\label{eq:conclusion X}
\end{align}
Equation \eqref{eq:conclusion X} has to be satisfied by all $\tilde{X}\in\mathcal{A}$.
Consequently $Y\in\mathbb{P}_{\mathcal{H}'}\mathcal{B}\mathbb{P}_{\mathcal{H}'}$
and therefore $\mathrm{Comm'}\left(\mathcal{A}'\right)\subset\mathbb{P}_{\mathcal{H}'}\mathcal{B}\mathbb{P}_{\mathcal{H}'}$.
\end{proof}
We are now ready to prove Lemma \ref{lem:commutant bosons-}.
\begin{proof}[Proof of Lemmas \ref{lem:commutant bosons-} and \ref{commutant ferm}]

\textbf{\textit{Bosonic case.}} We first focus on the case of bosons
(Lemma \ref{lem:commutant bosons-}). The proof of Lemma \ref{commutant ferm}
is completely analogous and we will discuss it briefly latter. In
the first step of the proof we show that $\mathbb{S}_{b}^{k}\in\mathrm{Comm}\left(\mathbb{C}\left[\Pi_{b}\otimes\Pi_{b}\left(\mathrm{LU}_{b}\right)\right]\right)$,
where operators $\mathbb{S}_{b}^{k}$ defined by Eq.\ref{eq:complicated swap operator bosons}.
In the second step we show that $\mathrm{Comm}\left(\mathbb{C}\left[\Pi_{b}\otimes\Pi_{b}\left(\mathrm{LU}_{b}\right)\right]\right)$
is commutative and that $\mathbb{S}_{b}^{k}$ ($k=0,\ldots,L$) span
$\mathrm{Comm}\left(\mathbb{C}\left[\Pi_{b}\otimes\Pi_{b}\left(\mathrm{LU}_{b}\right)\right]\right)$.

\textit{Step 1.} We start with proving that the commutant We first
employ Lemma \ref{technical lemma} to the situation in question.
Let us first introduce some notation
\begin{equation}
\mathcal{H}_{tot}=\mathcal{H}^{\otimes2L}=\left(\mathcal{H}_{1}\otimes\ldots\otimes\mathcal{H}_{L}\right)\otimes\left(\mathcal{H}_{1'}\otimes\ldots\otimes\mathcal{H}_{L'}\right)\,,\label{eq:spliting 2L copies}
\end{equation}
\begin{equation}
\mathcal{H}'=\mathrm{Sym}^{L}\left(\mathcal{H}\right)\otimes\mathrm{Sym}^{L}\left(\mathcal{H}\right)\,,\mathcal{\, A}=\mathbb{C}\left[\left\{ U^{\otimes2L}|\, U\in\mathrm{SU\left(\mathcal{H}\right)}\right\} \right]\,.\label{eq:bosons two subspace}
\end{equation}
Due to the Schur-Weyl duality (c.f. Subsection \ref{sub:Representation-theory-of})
we have
\[
\mathcal{B}=\mathrm{Comm}\left(\mathcal{A}\right)=\mathbb{C}\left[\left\{ \rho\left(\sigma\right)\,|\,\sigma\in\mathfrak{S}_{2L}\right\} \right]\,,
\]
where $\rho:\mathfrak{S}_{2L}\rightarrow\mathrm{U}\left(\mathcal{H}^{\otimes2L}\right)$
is the standard representation of the permutation group of $2L$ element
set%
\footnote{We have employed the convention of labeling of elements of the $2L$
element set in agreement with the notation from Eq.\ref{eq:spliting 2L copies}%
} $\left\{ 1,\ldots,L,1',\ldots,L'\right\} $ (for the definition of
this representation see Subsection \ref{sub:Representation-theory-of}).
We want to describe the commutant of $\mathbb{C}^{\ast}$-algebra
\begin{equation}
\mathcal{A}'=\mathrm{Comm}\left(\mathbb{C}\left[\Pi_{b}\otimes\Pi_{b}\left(\mathrm{LU}_{b}\right)\right]\right)\label{eq:bosons two comm}
\end{equation}
 which the algebra $\mathcal{A}$ restricted to $\mathcal{H}'$. Notice
that $\mathbb{P}_{\mathcal{H}'}=\mathbb{P}_{\left\{ 1,\ldots,L\right\} }^{\mathrm{sym}}\circ\mathbb{P}_{\left\{ 1',\ldots,L'\right\} }^{\mathrm{sym}}\in\mathcal{B}$
and consequently the assumptions of Lemma \eqref{technical lemma}
are satisfied. Consequently we get
\[
\mathrm{Comm}'\left(\mathbb{P}_{\mathcal{H}'}\mathcal{A}\mathbb{P}_{\mathcal{H}'}\right)=\mbox{\ensuremath{\mathbb{P}}}_{\mathcal{H}'}\mathcal{B}\mbox{\ensuremath{\mathbb{P}}}_{\mathcal{H}'}\,.
\]
The operators $\mathbb{S}_{b}^{k}$ by definition belong to $\mbox{\ensuremath{\mathbb{P}}}_{\mathcal{H}'}\mathcal{B}\mbox{\ensuremath{\mathbb{P}}}_{\mathcal{H}'}$.

\textit{Step 2.} The algebra $\mbox{\ensuremath{\mathbb{P}}}_{\mathcal{H}'}\mathcal{B}\mbox{\ensuremath{\mathbb{P}}}_{\mathcal{H}'}$
is now generated by operators
\begin{equation}
\left.\rho\left(\tau_{ij}\right)\right|_{\mathcal{H}'}=\mathbb{P}_{\left\{ 1,\ldots,L\right\} }^{\mathrm{sym}}\circ\mathbb{P}_{\left\{ 1',\ldots,L'\right\} }^{\mathrm{sym}}\rho\left(\tau_{ij}\right)\mathbb{P}_{\left\{ 1,\ldots,L\right\} }^{\mathrm{sym}}\circ\mathbb{P}_{\left\{ 1',\ldots,L'\right\} }^{\mathrm{sym}}\,,\label{eq:comm two bos gen}
\end{equation}
where $\tau_{ij}\in\mathfrak{S}_{2L}$ is a transposition between
elements $i,j\in\left\{ 1,\ldots,L,1',\ldots,L'\right\} $. Note that
if $i,j\in\left\{ 1,\ldots,L\right\} $ or $i,j\in\left\{ 1',\ldots,L'\right\} $
we have $\left.\rho\left(\tau_{ij}\right)\right|_{\mathcal{H}'}=\left.\mathbb{I}\right|_{\mathcal{H}'}$.
Moreover, if $i\in\left\{ 1,\ldots,L\right\} $ and $j\in\left\{ 1,\ldots,L\right\} $
we have $\left.\rho\left(\tau_{ij}\right)\right|_{\mathcal{H}'}=\left.\rho\left(\tau_{11'}\right)\right|_{\mathcal{H}'}$.
Consequently, algebra $\mbox{\ensuremath{\mathbb{P}}}_{\mathcal{H}'}\mathcal{B}\mbox{\ensuremath{\mathbb{P}}}_{\mathcal{H}'}$
is generated solely by operator $\mbox{\ensuremath{\mathbb{P}}}_{\mathcal{H}'}\mathcal{B}\mbox{\ensuremath{\mathbb{P}}}_{\mathcal{H}'}$
and thus must be commutative. We now show that powers of $\left.\rho\left(\tau_{11'}\right)\right|_{\mathcal{H}'}$
can be expressed by $\mathbb{S}_{b}^{k}$. The explicit computation
shows that the operator. 
\[
\tau=\sum_{\begin{array}[t]{c}
i\in\left\{ 1,\ldots,L\right\} \\
j\in\left\{ 1',\ldots,L'\right\} 
\end{array}}\tau_{ij}\,
\]
commutes with $\mbox{\ensuremath{\mathbb{P}}}_{\mathcal{H}'}$. We
have
\[
\mbox{\ensuremath{\mathbb{P}}}_{\mathcal{H}'}\tau\mbox{\ensuremath{\mathbb{P}}}_{\mathcal{H}'}=L^{2}\left.\rho\left(\tau_{11'}\right)\right|_{\mathcal{H}'}.
\]
Moreover $\mbox{\ensuremath{\mathbb{P}}}_{\mathcal{H}'}\tau^{k}\mbox{\ensuremath{\mathbb{P}}}_{\mathcal{H}'}=\left(\mbox{\ensuremath{\mathbb{P}}}_{\mathcal{H}'}\tau\mbox{\ensuremath{\mathbb{P}}}_{\mathcal{H}'}\right)^{k}$
can be expressed via $\mathbb{S}_{b}^{k}$ ($k=0,\ldots,L$). This
concludes the proof of the fact that operators $\mathbb{S}_{b}^{k}$
generate $\mathrm{Comm}\left(\mathbb{C}\left[\Pi_{b}\otimes\Pi_{b}\left(\mathrm{LU}_{b}\right)\right]\right)$.

\textbf{\textit{Fermionic case. }}Proof of the fact that $\mathrm{Comm}\left(\mathbb{C}\left[\Pi_{f}\otimes\Pi_{f}\left(\mathrm{LU}_{f}\right)\right]\right)$
is commutative and that operators $\mathbb{S}_{f}^{k}$ ($k=0,\ldots,L$)
span $\mathrm{Comm}\left(\mathbb{C}\left[\Pi_{f}\otimes\Pi_{f}\left(\mathrm{LU}_{f}\right)\right]\right)$
is completely analogous. One can repeat all the steps of the proof
given above.
\end{proof}

\subsection*{Proof of Lemma \ref{LUb invariant bilin theorem}\label{sub:Proof-of-Lemma-Lub ineq}}
\begin{proof}
We will show that inequalities \eqref{eq:inequalities entanglement bosons}
and \eqref{eq:inequalities structure general} are equivalent in the
considered setting. We fix the orthonormal basis $\left\{ \ket{e_{i}}\right\} _{i=1}^{N}$of
a single particle Hilbert space $\mathcal{H}\approx\mathbb{C}^{N}$.
As a highest weight state $\ket{\psi_{0}}$ of $\mathrm{LU_{f}}$
we take
\begin{equation}
\ket{\psi_{0}}=\ket{e_{1}^{\otimes L}}\,.\label{eq:high weight vector ferm-1}
\end{equation}
As normalized weight vectors we take (see Subsection \ref{sub:Structural-theory-of}
for details) the generalized Dicke states \citep{Parashar2009-dicke},
$\ket{\psi_{\lambda}}=\ket{\psi_{\vec{n}}}$. Generalized Dicke states
$\ket{\psi_{\vec{n}}}$ (and therefore the corresponding weights $\lambda$)
are labeled by integer-valued sequences $\vec{n}=\left(n_{1},\ldots,n_{N}\right)$
of length $N$, which satisfy $n_{1}+\ldots+n_{L}=L$. The explicit
expression reads
\begin{equation}
\ket{\psi_{\vec{n}}}=\,\sqrt{\frac{L!}{n_{1}!\cdot\ldots\cdot n_{N}!}}\mathbb{P^{\mathrm{sym}}\left(\ket{e_{1}^{\otimes n_{1}}}\ket{e_{2}^{\otimes n_{2}}}\ldots\ket{e_{N}^{\otimes n_{L}}}\right)},\label{eq:weight vector bos}
\end{equation}
where $\mathbb{P}^{\mathrm{sym}}:\mathcal{H}^{\otimes L}\rightarrow\mathcal{H}^{\otimes L}$
is the orthonormal projector onto $\mathrm{Sym}^{L}\left(\mathcal{H}\right)\subset\mathcal{H}^{\otimes L}$.

In what follows we apply \eqref{eq:inequalities entanglement fermions}
for the weight vectors specified above. Using Eq. \eqref{eq:fromula swap bosons}
and making use of the introduced notation we get 
\begin{equation}
\bra{\psi_{0}}\bra{\psi_{\lambda}}\mathbb{S}_{b}^{k}\ket{\psi_{0}}\ket{\psi_{\lambda}}=\bra{\psi_{0}}\bra{\psi_{\vec{n}}}\mathbb{S}_{b}^{k}\ket{\psi_{0}}\ket{\psi_{\vec{n}}}=\mathrm{tr}\left(\left[\kb{e_{1}^{\otimes L}}{e_{1}^{\otimes L}}\right]_{\left(k\right)}\left[\kb{\psi_{\vec{n}}}{\psi_{\vec{n}}}\right]_{\left(k\right)}\right)\,,\label{eq:matrix element ferm swap-1}
\end{equation}
where the trace in the right hand side is over the space $\mathrm{Sym}^{k}\mathcal{H}$.
We now compute the above expression directly,
\begin{equation}
\mathrm{tr}\left(\left[\kb{e_{1}^{\otimes L}}{e_{1}^{\otimes L}}\right]_{\left(k\right)}\left[\kb{\psi_{\vec{n}}}{\psi_{\vec{n}}}\right]_{\left(k\right)}\right)=\mathrm{tr}\left(\left[\kb{e_{1}^{\otimes k}}{e_{1}^{\otimes k}}\otimes\mathbb{I}^{\otimes\left(L-k\right)}\right]\kb{\psi_{\vec{n}}}{\psi_{\vec{n}}}\right)=\bk{\Psi_{\vec{n}}}{\Psi_{\vec{n}}}\,,\label{eq:identity fermions-1}
\end{equation}
where vector $\ket{\Psi_{\vec{n}}}$ is given by
\begin{equation}
\ket{\Psi_{\vec{n}}}=\frac{1}{\sqrt{n_{1}!\cdot\ldots\cdot n_{N}!\cdot L!}}\left[\kb{e_{1}^{\otimes k}}{e_{1}^{\otimes k}}\otimes\mathbb{I}^{\otimes\left(L-k\right)}\right]\sum_{\sigma\in S_{L}}\sigma\left(\ket{e_{1}^{\otimes n_{1}}}\ket{e_{2}^{\otimes n_{2}}}\ldots\ket{e_{N}^{\otimes n_{L}}}\right)\,,\label{eq:aux state bosons}
\end{equation}
where $S_{L}$ denotes the permutation group of $L$ elements and
permutations and $\sigma\in S_{L}$ act on $\mathcal{H}^{\otimes L}$
in a natural manner. We notice that $\bk{\Psi_{\vec{n}}}{\Psi_{\vec{n}}}=0$
for $n_{1}<k$ and therefore in what follows we assume $n_{1}\geq k$.
Keeping the track of all non-vanishing terms in \eqref{eq:aux state bosons}
we get
\begin{equation}
\ket{\Psi_{\vec{n}}}=\frac{N(n_{1},k)}{\sqrt{n_{1}!\cdot\ldots\cdot n_{N}!\cdot L!}}\ket{e_{1}^{\otimes k}}\otimes\sum_{\sigma\in S_{L-k}}\sigma\left(\ket{e_{1}^{\otimes n_{1}-k}}\ket{e_{2}^{\otimes n_{2}}}\ldots\ket{e_{N}^{\otimes n_{L}}}\right)\,,\label{eqbox aux simpl}
\end{equation}
where $N\left(n_{1},k\right)=\frac{n_{1}!}{\left(n_{1}-k\right)!}$
denotes the number of permutations in $\pi\in S_{n_{!}}$ preserving
the set $\left\{ 1,\ldots,k\right\} \subset\left\{ 1,\ldots,n_{1}\right\} $.
The inner product $\bk{\Psi_{\vec{n}}}{\Psi_{\vec{n}}}$ can be now
explicitly computed as in the right hand side of \eqref{eqbox aux simpl}
we have a non-normalized Dicke state corresponding to $\vec{\tilde{n}}=\left(n_{1}-k,n_{2},\ldots,n_{N}\right)$
and the final result is the following,
\begin{equation}
\bra{\psi_{0}}\bra{\psi_{\vec{n}}}\mathbb{S}_{b}^{k}\ket{\psi_{0}}\ket{\psi_{\vec{n}}}=\begin{cases}
\frac{\binom{n_{1}}{k}}{\binom{L}{k}} & \text{if}\, k\leq n_{1}\\
0 & \text{otherwise}
\end{cases}\,.\label{eq:matrix element swap bos}
\end{equation}
Thanks to the Lemma \ref{lem:commutant bosons-}, every operator $V\in\mathrm{Comm}\left(\mathbb{C}\left[\Pi_{b}\otimes\Pi_{b}\left(\mathrm{LU}\right)\right]\right)\cap\mathrm{Herm}\left(\mathcal{H}_{b}\otimes\mathcal{H}_{b}\right)$
can be written as
\[
V=\sum_{k=0}^{L}a_{k}\mathbb{S}_{b}^{k}\,,
\]
where $a_{m}\in\mathbb{R}$ are real parameters. By choosing all possible
$\ket{\psi_{\lambda}}=\ket{\psi_{\vec{n}}}$ and applying \eqref{eq:matrix element swap bos}
to \eqref{eq:inequalities entanglement fermions} we obtain inequalities
\eqref{eq:inequalities entanglement bosons}.
\end{proof}

\subsection*{Proof of Lemma \ref{lem: extrem rays entanglement bos} \label{sub:Proof-of-Lemma-extr rays bos}}
\begin{proof}
The proof of Lemma \ref{lem: extrem rays entanglement bos} is analogous
to the proof of Lemma \ref{lem: extrem rays entanglement}. The role
of operators $V^{Y}$ ($Y\subset\left\{ 1,\ldots,L\right\} $) is
played by operators $V_{b}^{m}$ ($m=0,\ldots,L$) wheres inequalities
\eqref{eq:inequalities entanglement} are replaced by inequalities
\eqref{eq:inequalities entanglement bosons}. The only difference
lays in the fact that instead of transformation \eqref{eq:subsets relations}
we have the following linear mapping,
\begin{equation}
b_{n}=\sum_{k=0}^{n}\frac{\binom{n}{k}}{\binom{L}{k}}a_{k}\,,\label{eq:relations bosons}
\end{equation}
where $n,k=0,\ldots,L$. The inverse mapping of the above can be computed
explicitly by noticing that \eqref{eq:relations bosons} is a binomial
transform of the sequence $\left\{ \tilde{a}_{k}\right\} _{k=1}^{L}$,
where $\tilde{a}_{k}=\frac{a_{k}}{\binom{L}{k}}$. Recall that the
binomial transform \citep{Knuth1989} of a sequence $\left\{ x_{k}\right\} _{k=0}^{L}$
is a sequence $\left\{ y_{n}\right\} _{n=0}^{L}$, defined by 
\begin{equation}
y_{n}=\sum_{k=0}^{n}\binom{n}{k}x_{k}\,.\label{eq:binomial transform}
\end{equation}
The inverse formula of the above reads
\begin{equation}
x_{k}=\sum_{n=0}^{m}\binom{m}{k}\left(-1\right)^{n+k}y_{n}\,.\label{eq:inverse binomial}
\end{equation}
Therefore the inverse of \eqref{eq:relations bosons} is given by
\[
a_{k}=\sum_{n=0}^{m}\binom{m}{k}\binom{L}{k}\left(-1\right)^{n+k}b_{n}\,.
\]
Application of the above formula to \eqref{eq:bilin bos entangl structure}
in the analogous way as it was done with \eqref{eq:inclusion-exclusion}
in the proof of Lemma \ref{lem: extrem rays entanglement bos} finishes
the proof.
\end{proof}

\subsection*{Proof of Proposition \ref{Lem:equivalence bos disting}\label{sub:Proof-of-Proposition-bosons oper compar}}
\begin{proof}
We show by a direct computation that
\begin{equation}
\mbox{\ensuremath{\mathbb{P}}}^{\mathrm{sym}}\otimes\mbox{\ensuremath{\mathbb{P}}}^{\mathrm{sym}}V^{Y}\mbox{\ensuremath{\mathbb{P}}}^{\mathrm{sym}}\otimes\mbox{\ensuremath{\mathbb{P}}}^{\mathrm{sym}}=\frac{1}{\binom{L}{m}}V_{b}^{m}\,,\label{eq:etremal rays bos}
\end{equation}
where $\left|Y\right|=m$, $V^{Y}$ is defined in Eq.\eqref{eq:formula extremal entanglement}
and $V_{b}^{m}$ is defined in Eq.\eqref{eq:formula extremal bos entanglement}. 

First note that for all $X\subset\left\{ 1,\ldots,L\right\} $ with
$\left|X\right|=k$ we have (via Eq.\eqref{eq:complicated swap operator bosons})
\[
\mbox{\ensuremath{\mathbb{P}}}^{\mathrm{sym}}\otimes\mbox{\ensuremath{\mathbb{P}}}^{\mathrm{sym}}\mathbb{S}^{X}\mbox{\ensuremath{\mathbb{P}}}^{\mathrm{sym}}\otimes\mbox{\ensuremath{\mathbb{P}}}^{\mathrm{sym}}=\mathbb{S}_{b}^{k}\,.
\]
Therefore we have 
\begin{equation}
\mbox{\ensuremath{\mathbb{P}}}^{\mathrm{sym}}\otimes\mbox{\ensuremath{\mathbb{P}}}^{\mathrm{sym}}V^{Y}\mbox{\ensuremath{\mathbb{P}}}^{\mathrm{sym}}\otimes\mbox{\ensuremath{\mathbb{P}}}^{\mathrm{sym}}=-\left(-1\right)^{m}\sum_{k:k\geq m}N_{L}(k,m)\,\left(-1\right)^{k}\mathbb{S}_{b}^{k}\,,\label{eq:restriction bos}
\end{equation}
where $N_{L}(k,m)$ denotes the number of $k$-element subsets of
the $L$-element set which contains some fixed set consisting of $m$
elements. Simple combinatorics shows that
\[
N_{L}(k,m)=\binom{L-m}{k-m}=\frac{\binom{L}{k}\binom{k}{m}}{\binom{L}{m}}\,.
\]
The above, when put together with \eqref{eq:restriction bos}, proves
\eqref{lem: extrem rays entanglement bos}.
\end{proof}

\subsection*{Proof of Lemma \ref{LUf invariant bilin theorem}\label{sub:Proof-of LUf invariant bilin th}}
\begin{proof}
The proof is analogous to proofs of Theorems \ref{LU invariant bilin theorem}
and \ref{LUb invariant bilin theorem} and relies on Theorem \ref{finitelly generated cone}.
We will show the equivalence between inequalities \eqref{eq:inequalities entanglement fermions}
and \eqref{eq:inequalities structure general} in this particular
setting. Let us fix the orthonormal basis $\left\{ \ket{e_{i}}\right\} _{i=1}^{N}$
of a single particle Hilbert space $\mathcal{H}\approx\mathbb{C}^{N}$.
As a highest weight state $\ket{\psi_{0}}$ of $\mathrm{LU_{f}}$
we take 
\begin{equation}
\ket{\psi_{0}}=\ket{e_{1}}\wedge\ket{e_{2}}\wedge\ldots\wedge\ket{e_{L}}\,.\label{eq:high weight vector ferm}
\end{equation}
As weight vectors we can take (see Subsection \ref{sub:Structural-theory-of}
for details) the standard Slater determinants that correspond to the
choice of the orthonormal basis in $\mathcal{H}$.
\begin{equation}
\ket{\psi_{\lambda}}=\ket{\psi_{I}}=\ket{e_{i_{1}}}\wedge\ket{e_{i_{2}}}\wedge\ldots\wedge\ket{e_{i_{L}}}\,,\label{eq:weight vector ferm}
\end{equation}
where $I=\left\{ i_{1},i_{2},\ldots,i_{L}\right\} $ is the $L$-element
subset of the set $\left\{ 1,\ldots,N\right\} $. Moreover we assume
that indices $i_{1},i_{2},\ldots,i_{L}$ that appear in \eqref{eq:high weight vector ferm}
are ordered increasingly. Therefore the weights $\lambda$, that correspond
to weight vectors $\ket{\psi_{\lambda}}$ appearing in $\bigwedge^{L}\mathcal{H}$,
can be identified with subsets $L$-element subsets $I\subset\left\{ 1,\ldots,N\right\} $.
Let us denote $I_{0}=\left\{ 1,\ldots,L\right\} $.

Our aim is to apply \eqref{eq:inequalities entanglement fermions}
for the weight vectors specified above. Using Eq. \eqref{eq:fromula swap fermions}
and making use of the introduced notation we get 
\begin{equation}
\bra{\psi_{0}}\bra{\psi_{\lambda}}\mathbb{S}_{f}^{k}\ket{\psi_{0}}\ket{\psi_{\lambda}}=\bra{\psi_{I_{0}}}\bra{\psi_{I}}\mathbb{S}_{f}^{k}\ket{\psi_{I_{0}}}\ket{\psi_{I}}=\mathrm{tr}\left(\left[\kb{\psi_{I_{0}}}{\psi_{I_{0}}}\right]_{\left(k\right)}\left[\kb{\psi_{I}}{\psi_{I}}\right]_{\left(k\right)}\right)\,,\label{eq:matrix element ferm swap}
\end{equation}
where the trace in the right hand side is over the space $\bigwedge^{k}\mathcal{H}$.
In order to compute the above expression explicitly we use the following
identity \citep{Fetter2003}, valid for all $L$-element sets $I$
\begin{equation}
\left[\kb{\psi_{I}}{\psi_{I}}\right]_{\left(k\right)}=\frac{1}{\binom{L}{k}}\sum_{\begin{array}[t]{c}
J:J\subset I\\
\left|J\right|=k
\end{array}}\kb{\psi_{J}}{\psi_{J}}\,,\label{eq:identity fermions}
\end{equation}
where the summation is over all subset of the set $I$ that have cardinality
$k$. Plugging \eqref{eq:identity fermions} to \eqref{eq:matrix element ferm swap}
we get 
\begin{equation}
\bra{\psi_{I_{0}}}\bra{\psi_{I}}\mathbb{S}_{f}^{k}\ket{\psi_{I_{0}}}\ket{\psi_{I}}=\begin{cases}
\frac{\binom{\left|I_{0}\cap I\right|}{k}}{\binom{L}{k}^{2}} & \text{if}\, k\leq\left|I_{0}\cap I\right|\\
0 & \text{otherwise}
\end{cases}\,.\label{eq:matrix element swap ferm}
\end{equation}
Due to the Lemma \ref{commutant ferm} every operator $V\in\mathrm{Comm}\left(\mathbb{C}\left[\Pi_{f}\otimes\Pi_{f}\left(\mathrm{LU}\right)\right]\right)\cap\mathrm{Herm}\left(\mathcal{H}_{f}\otimes\mathcal{H}_{f}\right)$
can be written as
\[
V=\sum_{k=0}^{L}a_{k}\mathbb{S}_{f}^{k}\,,
\]
where $a_{m}\in\mathbb{R}$ are real parameters. We note that as we
go over all $\lambda$ in \eqref{eq:inequalities structure general}
we essentially sample, by the virtue of \eqref{eq:matrix element swap ferm},
over all possible $\left|I_{0}\cap I\right|=m=0,\ldots,L$ (at this
point we are using the assumption $2d\geq L$). Inequalities \eqref{eq:inequalities entanglement fermions}
follow from applying to the above expression inequalities \eqref{eq:inequalities structure general}
together with Eq.\eqref{eq:matrix element swap ferm}.
\end{proof}

\subsection*{Proof of Lemma \ref{commutant fer gauss}\label{sub:Proof-of-Lemma comm gauss}}
\begin{proof}[Sketch of the proof]
Just like in the cases of bosons and fermions we will use Lemma \ref{technical lemma}.
In \citep{Wenzl2012} it was proven that operators 
\begin{equation}
\tilde{\mathbb{C}}_{k}=\left(\sum_{\begin{array}[t]{c}
X\subset\left\{ 1,\ldots,2d\right\} \\
\left|X\right|=k
\end{array}}\prod_{i\in X}c_{i}\otimes c_{i}\right)\,,\, k=0,1,\ldots,2d\,,\label{eq:commutant pin}
\end{equation}
form a basis of the commutant of the spinor representation $\Pi_{s}\otimes\Pi_{s}$
of the group $\mathrm{Pin}\left(2d\right)$ in $\mathcal{H}_{\mathrm{Fock}}\left(\mathbb{C}^{d}\right)\otimes\mathcal{H}_{\mathrm{Fock}}\left(\mathbb{C}^{d}\right)$.
Moreover, form \citep{Wenzl2012} it also follows that the commutant
is commutative. We do not need to describe the group $\mathrm{Pin}\left(2d\right)$
in detail%
\footnote{The group $\mathrm{Pin}\left(2d\right)$ is defined as the universal
covering group of $\mathrm{O}\left(2d\right)$.%
}. The representation $\Pi_{s}$ of the group $\mathrm{Pin}\left(2d\right)$
in $\mathcal{H}_{\mathrm{Fock}}\left(\mathbb{C}^{d}\right)$ is generated
by $\mathcal{B}=\Pi_{s}\left(\mathrm{Spin}\left(2d\right)\right)$
and an arbitrary Majorana operator, say $c_{1}$. A simple observation
shows that 
\begin{equation}
\mathbb{C}\left[\Pi_{s}^{+}\otimes\Pi_{s}^{+}\left(\mathrm{Spin}\left(2d\right)\right)\right]=\mbox{\ensuremath{\mathbb{P}}}_{+}\mathbb{C}\left[\Pi_{s}\otimes\Pi_{s}\left(\mathrm{Pin}\left(2d\right)\right)\right]\mbox{\ensuremath{\mathbb{P}}}_{+}\,,\label{eq:restriction gaussian}
\end{equation}
where
\[
\mbox{\ensuremath{\mathbb{P}}}_{+}\in\mathrm{Comm}\left(\mathbb{C}\left[\Pi_{s}\otimes\Pi_{s}\left(\mathrm{Pin}\left(2d\right)\right)\right]\right)
\]
 is the orthonormal projector onto $\mathcal{H}_{\mathrm{Fock}}^{+}\left(\mathbb{C}^{d}\right)\otimes\mathcal{H}_{\mathrm{Fock}}^{+}\left(\mathbb{C}^{d}\right)\subset\mathcal{H}_{\mathrm{Fock}}\left(\mathbb{C}^{d}\right)\otimes\mathcal{H}_{\mathrm{Fock}}\left(\mathbb{C}^{d}\right)$.
We can now use Lemma \ref{technical lemma} to infer that 
\begin{equation}
\mathrm{Comm}\left(\mathbb{C}\left[\Pi_{s}^{+}\otimes\Pi_{s}^{+}\left(\mathrm{Spin}\left(2d\right)\right)\right]\right)=\mbox{\ensuremath{\mathbb{P}}}_{+}\mathrm{Comm}\left(\mathbb{C}\left[\Pi_{s}\otimes\Pi_{s}\left(\mathrm{Pin}\left(2d\right)\right)\right]\right)\mbox{\ensuremath{\mathbb{P}}}_{+}\,.\label{eq:commutants spin pin}
\end{equation}
We can now use the formula 
\[
\mathbb{P}_{\mathcal{H}_{\mathrm{Fock}}^{+}\left(\mathbb{C}^{d}\right)\otimes\mathcal{H}_{\mathrm{Fock}}^{+}\left(\mathbb{C}^{d}\right)}=\frac{1}{4}\left(\mathbb{I}+Q\right)\otimes\left(\mathbb{I}+Q\right)
\]
 and \eqref{eq:commutant pin} to infer that $\mathrm{Comm}\left(\mathbb{C}\left[\Pi_{s}^{+}\otimes\Pi_{s}^{+}\left(\mathrm{Spin}\left(2d\right)\right)\right]\right)$
is spanned by the operators 
\[
\widetilde{\mathbb{C}}_{k}=\left(\mathbb{I}+Q\right)\otimes\left(\mathbb{I}+Q\right)\tilde{\mathbb{C}}_{k}\left(\mathbb{I}+Q\right)\otimes\left(\mathbb{I}+Q\right)\,,\, k=0,1,\ldots,2d\,.
\]
Using the anticommutation relations satisfied by the Majorana fermion
operations we obtain $\widetilde{\mathbb{C}}_{k}=0$ for odd $k$
and that for $k\leq\left\lfloor \frac{d}{2}\right\rfloor $ $\widetilde{\mathbb{C}}_{k}=\widetilde{\mathbb{C}}_{2d-2k}$.
We therefore obtain that 
\[
\mathrm{Comm}\left(\mathbb{C}\left[\Pi_{s}^{+}\otimes\Pi_{s}^{+}\left(\mathrm{Spin}\left(2d\right)\right)\right]\right)
\]
 is spanned by 
\[
\mathbb{C}_{k}=\mbox{\ensuremath{\mathbb{P}}}_{+}\left(\sum_{\begin{array}[t]{c}
X\subset\left\{ 1,\ldots,2d\right\} \\
\left|X\right|=2k
\end{array}}\prod_{i\in X}c_{i}\otimes c_{i}\right)\,\mbox{\ensuremath{\mathbb{P}}}_{+}\,,\, k=0,\ldots,\left\lfloor \frac{d}{2}\right\rfloor \,.
\]
We conclude the proof by noting that operators $\mathbb{C}_{k}$ are
formed by sums of expressions of the form 
\[
\prod_{i\in X}c_{i}\otimes c_{i}\pm\prod_{i\in\bar{X}}c_{i}\otimes c_{i}\,,
\]
 where $\bar{X}$ denotes the complement of the set $X$ in the ``global
set'' $\left\{ 1,\ldots,2d\right\} $. From this we infer that operators
$\mathbb{C}_{k}$, $k=0,\ldots,\left\lfloor \frac{d}{2}\right\rfloor $,
are linearly independent. 

The commutativity of $\mathrm{Comm}\left(\mathbb{C}\left[\Pi_{s}^{+}\otimes\Pi_{s}^{+}\left(\mathrm{Spin}\left(2d\right)\right)\right]\right)$
follows from the commutativity of 
\[
\mathrm{Comm}\left(\mathbb{C}\left[\Pi_{s}\otimes\Pi_{s}\left(\mathrm{Pin}\left(2d\right)\right)\right]\right)
\]
 and Eq.\eqref{eq:commutants spin pin}.
\end{proof}

\subsection*{Proof of Lemma \ref{FLO invariant bilin theorem}\label{sub:Proof-of-flo invariant}}
\begin{proof}
We are going to use Theorem \ref{finitelly generated cone} directly
and thus we first recall the structure of weight vectors $\ket{\psi_{\lambda}}\in\mathcal{H}_{\mathrm{Fock}}^{+}\left(\mathbb{C}^{d}\right)$
of the representation $\Pi_{s}^{+}$. As discussed in Subsection \ref{sub:Spinor-represenations-of}
as a highest weight vector we can take the Fock vacuum $\ket 0$.
As weight vectors we can take Fock states with even number of excitations,
i.e.
\begin{equation}
\ket{\psi_{\lambda}}=\ket{\psi_{I}}=\prod_{j\in J}a_{j}^{\dagger}\ket 0\,,\label{eq:Fock states gauss}
\end{equation}
where $J\subset\left\{ 1,\ldots,d\right\} $ is a subset of the set
$\left\{ 1,\ldots,d\right\} $ with even number of indices $\left|J\right|=0,2,\ldots,2\left\lfloor \frac{d}{2}\right\rfloor $.
We assume that the product $\prod_{j\in J}a_{j}^{\dagger}$ in \eqref{eq:Fock states gauss}
is ordered in a non-increasing manner. Let us now compute $\bra 0\bra{\psi_{J}}\mathbb{C}_{k}\ket 0\ket{\psi_{J}}$
for arbitrary even element set $J\subset\left\{ 1,\ldots,d\right\} $.
Using the fact that $\ket 0,\ket{\psi_{J}}\in\mathcal{H}_{\mathrm{Fock}}^{+}\left(\mathbb{C}^{d}\right)$
we get
\begin{align}
\bra 0\bra{\psi_{J}}\mathbb{C}_{k}\ket 0\ket{\psi_{J}} & =\sum_{\begin{array}[t]{c}
X\subset\left\{ 1,\ldots,2d\right\} \\
\left|X\right|=2k
\end{array}}\bra 0\bra{\psi_{J}}\left(\prod_{j\in X}c_{j}\otimes c_{j}\right)\ket 0\ket{\psi_{J}}\nonumber \\
 & =\sum_{\begin{array}[t]{c}
X\subset\left\{ 1,\ldots,2d\right\} \\
\left|X\right|=2k
\end{array}}\bra 0\prod_{j\in X}c_{i}\ket 0\bra{\psi_{j}}\prod_{j\in X}c_{i}\ket{\psi_{J}}\,.\label{eq:manipulation 1 optimal gauss}
\end{align}
Due to the definition of Majorana operators $\bra 0\prod_{j\in X}c_{j}\ket 0\neq0$
if and only if subset $X$ contains only pairs of indices associated
to the same mode in the Fock space. In other words only nonzero contributions
to the sum \eqref{eq:manipulation 1 optimal gauss} come from $X$
having the form
\[
X=\left\{ 2i_{1}-1,2i_{1},2i_{2}-1,2i_{2}-1,\ldots,2i_{k}-1,2i_{k}\right\} ,
\]
where $\left\{ i_{1},i_{2},\ldots,i_{k}\right\} \subset\left\{ 1,\ldots,d\right\} $
as an arbitrary $k$-element subset of the set $\left\{ 1,\ldots,d\right\} $.
Simple algebra gives the identity valid for $m=1,\ldots d$, 
\begin{equation}
ic_{2m-1}c_{2m}=-\left(\mathbb{I}-2\hat{n}_{m}\right)\,,\label{eq:identity gauss proof1}
\end{equation}
where $\hat{n}_{m}$ is the operator of particle number in the m'th
mode. Using \eqref{eq:identity gauss proof1} we get 
\begin{align*}
\bra 0\bra{\psi_{J}}\mathbb{C}_{k}\ket 0\ket{\psi_{J}} & =\sum_{\begin{array}[t]{c}
X\subset\left\{ 1,\ldots,d\right\} \\
\left|X\right|=k
\end{array}}\left(-1\right)^{k}\bra 0\prod_{j\in X}\left(\mathbb{I}-2\hat{n}_{j}\right)\ket 0\bra{\psi_{J}}\prod_{j\in X}\left(\mathbb{I}-2\hat{n}_{j}\right)\ket{\psi_{J}}\\
 & =\left(-1\right)^{k}\sum_{\begin{array}[t]{c}
X\subset\left\{ 1,\ldots,d\right\} \\
\left|X\right|=k
\end{array}}\bra{\psi_{J}}\prod_{j\in X}\left(\mathbb{I}-2\hat{n}_{j}\right)\ket{\psi_{J}}\,,
\end{align*}
where we have used that $\bra 0\hat{n}_{j}\ket 0=0$ for all $j=1,\ldots d$.
Binomial expansion of 
\[
\bra{\psi_{J}}\prod_{j\in X}\left(\mathbb{I}-2\hat{n}_{j}\right)\ket{\psi_{J}}
\]
and simple combinatorial manipulations give the following chain of
equalities
\begin{align}
\bra 0\bra{\psi_{J}}\mathbb{C}_{k}\ket 0\ket{\psi_{J}} & =\left(-1\right)^{k}\sum_{\begin{array}[t]{c}
X\subset\left\{ 1,\ldots,d\right\} \\
\left|X\right|=k
\end{array}}\sum_{m=0}^{k}\left(-2\right)^{m}\left(\sum_{\begin{array}[t]{c}
Y\subset X\\
\left|Y\right|=m
\end{array}}\prod_{j\in Y}\bra{\psi_{J}}\hat{n}_{j}\ket{\psi_{J}}\right)\nonumber \\
 & =\left(-1\right)^{k}\sum_{m=0}^{k}\left(-2\right)^{m}\binom{d-m}{k-m}\left(\sum_{\begin{array}[t]{c}
Y\subset\left\{ 1,\ldots,d\right\} \\
\left|Y\right|=m
\end{array}}\prod_{j\in Y}\bra{\psi_{J}}\hat{n}_{j}\ket{\psi_{J}}\right)\nonumber \\
 & =\left(-1\right)^{k}\sum_{m=0}^{\mathrm{min}\left\{ k,\left|J\right|\right\} }\left(-2\right)^{m}\binom{d-m}{k-m}\binom{\left|J\right|}{m}\,\label{eq:final expression gauss opt}
\end{align}
Recalling that $\ket{\psi_{J}}$ labeled all weight vectors of $\mathrm{Spin}\left(2d\right)$
in $\mathcal{H}_{\mathrm{Fock}}^{+}\left(\mathbb{C}^{d}\right)$ and
applying \eqref{eq:final expression gauss opt} to \eqref{eq:inequalities structure general}
for all subsets $J\subset\left\{ 1,\ldots,d\right\} $ having even
number of elements we recover inequalities \eqref{eq:inequalities gauss structure}. 
\end{proof}

\section{Proofs of results stated in Chapter \ref{chap:Typical-properties-of}\label{sec:Proofs-concerning-Chapter typicality}}

\subsection*{Proof of Lemma \ref{Lipschitz constant}\label{sub:Proof-of Lemma lipschitz}}
\begin{proof}
Due to the definition of the gradient $\nabla f$ and the structure
of the tangent space $T_{U}\mathrm{SU}\left(\mathcal{H}\right)$ for
$U\in\mathrm{SU}\left(\mathcal{H}\right)$ (see Eq.\eqref{eq:identification special unitary})
we have
\begin{equation}
g_{\mathrm{HS}}\left(\left.\nabla f\right|_{U},\, XU\right)==\left.\frac{d}{dt}\right|_{t=0}f_{\Omega}\left(\exp\left(tX\right)U\right)\,,\label{eq:gradient 1}
\end{equation}
where $X=-X^{\dagger}$ and $XU\in T_{U}SU\left(\mathcal{H}\right)$
(the tangent space treated as a (real) subspace of $\mathrm{End}\left(\mathcal{H}\right)$).
It follows that {\footnotesize{}
\begin{equation}
\left\langle \left.\nabla f\right|_{U},\, XU\right\rangle =\mathrm{tr}\left(\left[X\otimes\mathbb{I}^{\otimes\left(k-1\right)}+\mathbb{I}\otimes X\otimes\mathbb{I}^{\otimes\left(k-2\right)}+\ldots+\mathbb{I}^{\otimes\left(k-1\right)}\otimes X,\,\rho_{1}'\otimes\ldots\otimes\rho'_{k}\right]\, V\right),\label{eq: gradient 2}
\end{equation}
}where $\rho_{i}'=U\rho_{i}U^{\dagger}$ and $\mathbb{I}$ is the
identity operator on $\mathcal{H}$. We have the following chain of
inequalities {\footnotesize{}
\begin{align}
\left|\left\langle \left.\nabla f\right|_{U},\, XU\right\rangle \right| & \leq2\left|\mathrm{tr}\left(\,\rho_{1}'\otimes\ldots\otimes\rho'_{k}\, V\,\left[X\otimes\mathbb{I}^{\otimes\left(k-1\right)}+\mathbb{I}\otimes X\otimes\mathbb{I}^{\otimes\left(k-2\right)}+\ldots+\mathbb{I}^{\otimes\left(k-1\right)}\otimes X\right]\right)\right|\,.\label{eq:almost final}\\
 & \leq2\cdot\left(\left|\mathrm{tr}\left(\rho_{1}'\otimes\ldots\otimes\rho'_{k}\, V\, X\otimes\mathbb{I}^{\otimes\left(k-1\right)}\right)\right|+\ldots+\left|\mathrm{tr}\left(\rho_{1}'\otimes\ldots\otimes\rho'_{k}\, V\,\mathbb{I}^{\otimes\left(k-1\right)}\otimes X\right)\right|\right)\,.\label{eq:almost final 2}
\end{align}
}Every term on the right-hand side of \eqref{eq:almost final 2} is
bounded above by $\left\Vert X\right\Vert _{\mathrm{HS}}=\sqrt{-\mathrm{tr}\left(X^{2}\right)}$.
Let us consider $\left|\mathrm{tr}\left(\rho_{1}'\otimes\ldots\otimes\rho'_{k}\, V\, X\otimes\mathbb{I}^{\otimes\left(k-1\right)}\right)\right|$
as the proof for the remaining terms is analogous. The following inequalities
hold.
\begin{align}
\left|\mathrm{tr}\left(\left[\rho_{1}'\otimes\ldots\otimes\rho'_{k}\right]\, V\, X\otimes\mathbb{I}^{\otimes\left(k-1\right)}\right)\right| & \leq\left\Vert V\, X\otimes\mathbb{I}^{\otimes\left(k-1\right)}\right\Vert \,,\label{eq:annoy 1}\\
 & \leq\left\Vert V\right\Vert \left\Vert X\right\Vert \,,\label{eq:annoy 2}\\
 & \leq\left\Vert X\right\Vert _{\mathrm{HS}}\,.\label{eq:annoy 3}
\end{align}
In \eqref{eq:annoy 1} we used the fact that states $\rho_{i}^{'}$
are normalized so the modulus of expectation value of some operator
cannot exceed its operator norm. In \eqref{eq:annoy 2} we used the
inequality $\left\Vert AB\right\Vert \leq\left\Vert A\right\Vert \left\Vert B\right\Vert $
and the spectral decomposition of $X\otimes\mathbb{I}^{\otimes\left(k-1\right)}$.
In \eqref{eq:annoy 3} we have used the assumption $\left\Vert V\right\Vert \leq1$
and the inequality $\left\Vert X\right\Vert \leq\left\Vert X\right\Vert _{\mathrm{HS}}$
which is valid for all anti-Hermitian $X$ (and follows from their
spectral decomposition). Inserting \eqref{eq:annoy 3} to \eqref{eq:almost final}
we get 
\[
\left|\left\langle \left.\nabla f\right|_{U},\, XU\right\rangle \right|\leq2\cdot k\left\Vert X\right\Vert _{\mathrm{HS}}\,.
\]
Inserting to the above $XU=\left.\nabla f\right|_{U}$ concludes the
proof. 
\end{proof}

\subsection*{Proof of Lemma \ref{bilinear integration}\label{sub:Proof-of-Lemma bilinear integr}}
\begin{proof}
The computation of $\mathbb{E}_{\mu}f_{\rho_{1,}\rho_{2}}$ follows
easily from Weyl's ``unitary trick'' \citep{HallGroups}. Due to
the linearity of a trace we have 
\begin{equation}
\mathbb{E}_{\mu}f_{\rho_{1,}\rho_{2}}=\mathrm{\mathrm{tr}}\left(\int_{\mathrm{SU}\left(\mathcal{H}\right)}\left[U^{\otimes2}\left\{ \rho_{1}\otimes\rho_{2}\right\} \left(U^{\dagger}\right)^{\otimes2}\right]d\mu\left(U\right)\, V\right)\,.\label{eq:intermediate integration}
\end{equation}
Due to the left-invariance of the Haar measure $\mu$, for every $X\in\mathrm{Herm}\left(\mathcal{H}\otimes\mathcal{H}\right)$
the operator $\hat{X}=\int_{SU\left(\mathcal{H}\right)}dU\,\left(U^{\otimes2}\, X\,\left(U^{\dagger}\right)^{\otimes2}\right)$
is $\mathrm{SU}\left(\mathcal{H}\right)$ invariant, i.e. $\left[U^{\otimes2},\hat{X}\right]=0$.
Consequently, due to Schur lemma (see Fact \ref{Schur-Lemma})
\begin{equation}
\hat{X}=a\mathbb{P}^{\mathrm{sym}}+b\mathbb{P}^{\mathrm{asym}}\,,\label{eq:integration operator}
\end{equation}
where constants $a,b$ can be computed easily: $a=\frac{\mathrm{tr}\left(X\,\mathbb{P}^{\mathrm{sym}}\right)}{\mathrm{dim}\left(\mathbb{P}^{\mathrm{sym}}\right)}$
and $b=\frac{\mathrm{tr}\left(X\,\mathbb{P}^{\mathrm{asym}}\right)}{\mathrm{dim}\left(\mathbb{P}^{\mathrm{sym}}\right)}$.
Applying \eqref{eq:integration operator} to $X=\rho_{1}\otimes\rho_{2}$,
inserting the result to \eqref{eq:intermediate integration} and taking
into account $\mathbb{P}^{\mathrm{sym}}=\frac{\mathbb{I}\otimes\mathbb{I}+\mathbb{S}}{2},\,\mathbb{P}^{\mathrm{asym}}=\frac{\mathbb{I}\otimes\mathbb{I}-\mathbb{S}}{2}$
gives \eqref{eq:integration formula}.
\end{proof}

\subsection*{Proof of Lemma \ref{multilinear integration}. \label{sub:Proof-of-Lemma k ilinear integr}}
\begin{proof}
The proof is analogous to the proof of Lemma \ref{bilinear integration}.
For $f:\mathrm{SU}\left(\mathcal{H}\right)\rightarrow\mathbb{R}$
given by Eq \eqref{eq:multilinear function definition} we have 
\begin{align}
\mathbb{E}_{\mu}f & =\mathrm{tr}\left(\int_{\mathrm{SU}\left(\mathcal{H}\right)}\left[U^{\otimes k}\left\{ \rho\otimes\left(\kb{\psi}{\psi}\right)^{\otimes\left(k-1\right)}\right\} \left(U^{\dagger}\right)^{\otimes k}\right]d\mu\left(U\right)\, V\right)\,,\label{eq:first k inter}\\
 & =\mathrm{tr}\left(\left[\rho\otimes\left(\kb{\psi}{\psi}\right)^{\otimes\left(k-1\right)}\right]\,\int_{\mathrm{SU}\left(\mathcal{H}\right)}\left[U^{\otimes k}V\left(U^{\dagger}\right)^{\otimes k}\right]d\mu\left(U\right)\right)\,,\label{eq:second k iner}\\
 & =\mathrm{tr}\left(\left[\rho\otimes\left(\kb{\psi}{\psi}\right)^{\otimes\left(k-1\right)}\right]\,\left[\left\{ \left(k-1\right)+\frac{\mathrm{tr}\left(A\right)}{\mathrm{dim}\left(\mathrm{Sym}^{k}\left(\mathcal{H}\right)\right)}\right\} \mathbb{P}^{\mathrm{sym,k}}-\left(k-1\right)\mathbb{I}^{\otimes k}\right]\right)\,.\label{eq:fourtk k interg}
\end{align}
In Eq.\eqref{eq:second k iner} we used the invariance of the cyclicity
of the trace and invariance of the Haar measure under the transformation
$U\rightarrow U^{-1}.$ In Eq.\eqref{eq:fourtk k interg} we used
the particular form of the operator $V$ (see Eq.\eqref{eq:quasi criterion}).
The equality \eqref{eq:multilinear average} is a direct consequence
of \eqref{eq:fourtk k interg}.
\end{proof}

\subsection*{Proof of Proposition \ref{values of parametres} \label{sub:Proof-of values x bilin}}
\begin{proof}
The values of the parameter $c$ from Table \eqref{tab:Parameters-characteising-typical}
follows from the results of the Section \ref{sec:Bilinear-correlation-witnesses}.
However, the values of the parameter $X=\frac{\mathrm{dim}\left(Im\left(A\right)\right)}{\mathrm{dim}\left(\mathrm{Sym}^{2}\left(\mathcal{H}\right)\right)}$
requires justification. The value of $X$ for separable states follows
directly from Eq.\eqref{eq:crit prod states}. 

Let us now move to the case of bosonic product states and Slater determinants.
One can identify $\mathcal{H}_{b}$ and $\mathcal{H}_{f}$ with the
carrier spaces of irreducible representations of $SU\left(d\right)$
characterized by highest weights %
\footnote{See Subsection \ref{sub:Representation-theory-of} for the explanation
of the notation used in \eqref{eq:weights proof}.%
}
\begin{equation}
\lambda_{b}=\left(L,0,\ldots,0\right)\,\text{and}\,\lambda_{f}=\left(\stackrel{L}{\overbrace{1,\ldots,1}},0,\ldots,0\right)\,.\label{eq:weights proof}
\end{equation}
One can also represent $\lambda$ by a Young diagram - a collection
of boxes arranged in left-justified rows, with non-increasing lengths
when looked  from the top to the bottom. In Section \ref{sec:semisimple-quadratic-characterisation}
it was proven that the operator $A\in\mathrm{Herm}\left(\mathrm{Sym}^{2}\left(\mathcal{H}\right)\right)$
defining families of separable bosonic states and Slater determinants
is given by 
\begin{equation}
A=\mathbb{P}^{\mathrm{sym}}-\mathbb{P}^{2\lambda},\,\label{eq:structure A}
\end{equation}
where $\mathbb{P}^{2\lambda}$ is the projector onto the unique irreducible
representation of the type $2\lambda$ that appear in $\mathrm{Sym}^{2}\left(\mathcal{H}\right)$
(treated as a carrier space of a representation of $SU\left(d\right)$).
Therefore, the problem of computing $X$ for bosonic separable states
and Slater determinants reduces to computation of dimensions of representations
of $SU\left(d\right)$ described by highest weights $2\lambda_{b}$
and $2\lambda_{f}$ respectively. We perform these computations explicitly
using methods described in Subsection \ref{sub:Representation-theory-of}.

The value of $X$ for the case of pure fermionic Gaussian states $\mathcal{M}_{g}^{+}\subset\mathcal{D}_{1}\left(\mathcal{H}_{Fock}^{+}\left(\mathbb{C}^{d}\right)\right)$
follows from the discussion from Section \ref{sub:Fermionic-Gaussian-states}.
It was proven there that in this case 
\[
A=\mathbb{P}^{\mathrm{sym}}-\mathbb{P}^{+}\mathbb{P}_{0}\mathbb{P}^{+}\,,
\]
where:
\begin{itemize}
\item $\mathbb{P}_{0}$ is a projector onto eigenspace $0$ of the operator
$\Lambda\in\mathrm{Herm}\left(\mathcal{H}\left(\mathbb{C}^{d}\right)\otimes\mathcal{H}_{Fock}\left(\mathbb{C}^{d}\right)\right)$,
defined by $\Lambda=\sum_{i=1}^{2d}c_{i}\otimes c_{i}$,
\item $\mathbb{P}^{+}=\frac{1}{4}\left(\mathbb{I}+Q\right)\otimes\left(\mathbb{I}+Q\right)$, 
\item $\mathbb{P}^{\mathrm{sym}}$ - projector onto $\mathrm{Sym}^{2}\left(\mathcal{H}_{Fock}^{+}\left(\mathbb{C}^{d}\right)\right)$.
\end{itemize}
Straightforward computation shows that $\mathrm{dim}\left(\mathbb{P}^{+}\mathbb{P}_{0}\mathbb{P}^{+}\right)=\frac{1}{2}\mathrm{dim}\left(\mathrm{Im}\left(\mathbb{P}_{0}\right)\right)=\frac{1}{2}\binom{2d}{d}$.
\end{proof}

\subsection*{Proof of Proposition \ref{values of parametres klin}\label{sub:sketch of Proof klin}}
\begin{proof}[Sketch of the proof]
We show now that the values of the parameter $X=\frac{\mathrm{tr}\left(A\right)}{\mathrm{tr}\left(\mathbb{P}^{k,\mathrm{sym}}\right)}$
are indeed given by expressions from Table \ref{tab:Parameters-characteising-typical-klin}.
In the case of pure states $\mathcal{M}_{n}\subset\mathcal{D}_{1}\left(\mathbb{C}^{d}\otimes\mathbb{C}^{d}\right)$
follows directly from the form of the operator $A$ which is relevant
in this case (see Eq.\eqref{eq:operator schmidt}). 

In the case of states $\mathcal{M}_{d}^{3}\subset\mathcal{D}_{1}\left(\mathbb{C}^{d}\otimes\mathbb{C}^{d}\otimes\mathbb{C}^{d}\right)$
that do not exhibit genuine multiparty entanglement situation is more
involved and we limit ourselves only to giving the asymptotic of $X$
in the limit $d\rightarrow\infty$. In what follows will use the notation
from Subsection \ref{sub:GME}. In particular we will write $\mathcal{H}=\mathcal{H}_{1}\otimes\mathcal{H}_{2}\otimes\mathcal{H}_{3}$,
where $\mathcal{H}_{i}\approx\mathbb{C}^{d}$ and use the identification
$\mathcal{H}^{\otimes6}\approx\left(\mathcal{H}_{1}^{\otimes6}\right)\otimes\left(\mathcal{H}_{2}^{\otimes6}\right)\otimes\left(\mathcal{H}_{3}^{\otimes6}\right)$.
Let us first notice that in the considered scenario $k=6$ and $\mathrm{dim}\left(\mathcal{H}\right)=d^{3}$
and consequently 
\begin{equation}
\mathrm{tr}\left(\mathbb{P}^{k,\mathrm{sym}}\right)=\mathrm{dim}\left(\mathrm{Sym}^{6}\left(\mathcal{H}\right)\right)=\binom{d^{3}+5}{6}=\frac{1}{6!}\left(d^{3}\right)^{6}+O\left(\left(d^{3}\right)^{5}\right)=\frac{d^{18}}{6!}+O\left(d^{15}\right)\,,\label{eq:dimension GME}
\end{equation}
On the other hand the operator $A_{GME}$ is given by (see Eq.\eqref{eq:operator gme})
\[
A_{\mathrm{GME}}=\mbox{\ensuremath{\mathbb{P}}}^{\mathrm{sym,6}}\circ\left(A^{1:23}\otimes A^{2:13}\otimes A^{3:12}\right)\circ\mbox{\ensuremath{\mathbb{P}}}^{\mathrm{sym,6}}\,,
\]
where $A^{a:bc}\in\mathrm{Herm}\left(\mathcal{H}^{\otimes2}\right)$
is given by $A^{a:bc}=\mathbb{P}_{aa'}^{-}\circ\mathbb{P}_{\left\{ bc\right\} \left\{ b'c'\right\} }^{-}$.
The operator $A^{1:23}\otimes A^{2:13}\otimes A^{3:12}$ is a projector
onto some subspace $W\subset\mathcal{H}^{\otimes6}$. Let us introduce
the notation $\mathbb{P}_{W}=A^{1:23}\otimes A^{2:13}\otimes A^{3:12}$.
We have 
\begin{equation}
\mathrm{tr}\left(A_{\mathrm{GME}}\right)=\mathrm{tr}\left(\mathbb{P}_{W}\mbox{\ensuremath{\mathbb{P}}}^{\mathrm{sym,6}}\right)=\sum_{\alpha}\bra{\alpha}\mbox{\ensuremath{\mathbb{P}}}^{\mathrm{sym,6}}\ket{\alpha}\,,\label{eq:changed sum gme}
\end{equation}
where the index $\alpha$ runs over some orthonormal basis of $W$.
Let us first focus on a single operator $A^{1:23}$ which projects
on the subspace $\mathcal{\tilde{W}}$ of $\mathrm{Sym}^{2}\left(\mathcal{H}\right)$
spanned by normalized vectors of the form%
\footnote{By $\left\{ \ket i\right\} _{i=1}^{d}$ we denote the standard basis
of $\mathbb{C}^{d}$.%
}
\begin{equation}
\ket{\psi_{I_{1};I_{23}}}=\left(\ket{i_{1}}\wedge\ket{j_{1}}\right)\otimes\left(\left(\ket{i_{2}}\otimes\ket{i_{3}}\right)\wedge\left(\ket{j_{2}}\otimes\ket{j_{3}}\right)\right),\label{eq:state}
\end{equation}
where 
\begin{itemize}
\item $I_{1}=\left\{ i_{1},j_{1}\right\} $, $I_{23}=\left(\left\{ i_{2},i_{3}\right\} ,\left\{ j_{2},j_{3}\right\} \right)$ 
\item $i_{a},j_{a}\in\left\{ 1,\ldots,d\right\} $, $i_{1}>j_{1}$, $\left\{ i_{2},i_{3}\right\} \neq\left\{ j_{2},j_{3}\right\} $
\end{itemize}
In Eq. \eqref{eq:state} we have used for convenience the identification
$\mathcal{H}^{\otimes2}\approx\left(\mathcal{H}_{1}^{\otimes2}\right)\otimes\left(\mathcal{H}_{2}\otimes\mathcal{H}_{3}\right)\otimes\left(\mathcal{H}_{2}\otimes\mathcal{H}_{3}\right)$.
The vector $\ket{\psi_{I_{1};I_{23}}}$ assumes implicitly that we
have chosen the order of terms in $I_{1}$ as well as in $I_{23}$.
In what follows when we write $\ket{\psi_{I_{1};I_{23}}}$ we assume
that for each $I_{1}$and $I_{23}$ we chose exactly one ordering. 

We can now turn to estimating the asymptotic behavior of $\mathrm{tr}\left(A_{\mathrm{GME}}\right)$.
The basis of $W\subset\left(\mathcal{H}_{1}^{\otimes6}\right)\otimes\left(\mathcal{H}_{2}^{\otimes6}\right)\otimes\left(\mathcal{H}_{3}^{\otimes6}\right)$
can be now be chosen to consist of vectors of the form
\[
\ket{\alpha}=\ket{\psi_{I_{1}^{(1)};I_{23}^{\left(1\right)}}}\otimes\ket{\psi_{I_{2}^{(2)};I_{13}^{\left(2\right)}}}\otimes\ket{\psi_{I_{3}^{(3)};I_{12}^{\left(3\right)}}}\,,
\]
where vectors $\ket{\psi_{I_{1}^{(1)};I_{23}^{\left(1\right)}}}$,
$\ket{\psi_{I_{2}^{(2)};I_{13}^{\left(2\right)}}}$ and $\ket{\psi_{I_{3}^{(3)};I_{12}^{\left(2\right)}}}$
are constructed analogously as above. The leading contribution to
the sum in \eqref{eq:changed sum gme} comes from vectors $\ket{\alpha}$
such that each collection of indices
\[
i_{a}^{(1)},j_{a}^{(1)},i_{a}^{(2)},j_{a}^{(2)},i_{a}^{(3)},j_{a}^{(3)}\,,a=1,2,3\,,
\]
consists of six different numbers. The number of vectors $\ket{\alpha}$
satisfying such condition scales like $\mathcal{N}\approx\frac{d^{18}}{64}$.
Moreover, direct computation shows that for $\ket{\alpha}$ of this
form we have 
\[
\bra{\alpha}\mbox{\ensuremath{\mathbb{P}}}^{\mathrm{sym,6}}\ket{\alpha}=\frac{1}{6!}8\,.
\]
Consequently we get
\begin{equation}
\sum_{\alpha}\bra{\alpha}\mbox{\ensuremath{\mathbb{P}}}^{\mathrm{sym,6}}\ket{\alpha}\approx\mathcal{N}\frac{8}{6!}=\frac{d^{18}}{8\cdot6!}\,.\label{eq:almost final assymptotics}
\end{equation}
Combining \eqref{eq:almost final assymptotics} and \eqref{eq:dimension GME}
we get the desired result $\frac{\mathrm{tr}\left(A_{GME}\right)}{\mathrm{tr}\left(\mathbb{P}^{k,\mathrm{sym}}\right)}=\frac{1}{8}+O\left(\frac{1}{d}\right)$.
\end{proof}
\bibliographystyle{elsarticle-numalt}
\bibliography{bibliografia}

\end{document}